\documentclass[prx,twocolumn,superscriptaddress,showpacs]{revtex4-2}

\RequirePackage{mathtools}
\usepackage{tikz-cd}
\usetikzlibrary{cd}
\usepackage{mathrsfs}
\usepackage{amsfonts}
\usepackage{eucal}
\usepackage{latexsym}
\usepackage{amsthm}
\usepackage{amsmath}
\usepackage{faktor}
\usepackage{xfrac}
\usepackage{graphicx}
\usepackage{xcolor}
\usepackage[colorlinks,linkcolor=blue,citecolor=blue,urlcolor=blue]{hyperref}
\usepackage[caption=false]{subfig}
\usepackage{hyperref}
\usepackage{epsdice}
\usepackage{fontawesome}
\usepackage{dsfont}
\usepackage{bm}
\usepackage{skak}
\usepackage{physics}
\hypersetup{colorlinks,%
citecolor=blue,%
filecolor=blue,%
linkcolor=blue,%
urlcolor=blue,%
pdftex}
\usepackage{CJK}

\makeatletter
\def\l@subsubsection#1#2{}
\makeatother

\newtheorem{theorem}{Theorem}[section]
\newtheorem{definition}[theorem]{Definition}
\newtheorem{lemma}[theorem]{Lemma}
\newtheorem{proposition}[theorem]{Proposition}
\newtheorem{corollary}[theorem]{Corollary}

\newcommand{\phantomsubfloat}[1]{
    {
        \captionsetup[subfigure]{labelformat=empty}
        \subfloat[][]{#1}
    }%
}

\begin{document}
\title{\textbf{Building holographic code from the boundary}}
\begin{CJK*}{UTF8}{gbsn}
\author{Wei Wang (王巍)}
\email{wwwei\_wwwang@sjtu.edu.cn}
\affiliation{Tsung-Dao Lee Institute, Shanghai Jiao Tong University, Shanghai, 201210, China}

\begin{abstract}
Holographic quantum error-correcting code, the quantum-information structure hypothesized for the AdS/CFT correspondence, has being attracting increasing attention in new directions interrelating the studies of quantum gravity and quantum simulation. In this work, we initiate a novel approach for building holographic code that can be generally applied in potentially broad and interdisciplinary contexts. Our approach takes an ``opposite'' route to the conventional paradigm that is based on bulk tensor-networks. As illustrated in an exact model, we start from scalable descriptions of boundary qudits which can guide succinct quantum-circuit simulations, and rigorously show how the bulk qudits and the encoding structure emerge from boundary entanglement patterns. By investigating the entanglement patterns, we systematically unfold the hypothetical structure for bulk reconstruction and the details of the Ryu-Takayanagi formula in the formalism of operator-algebra quantum error correction, demonstrating desired properties that are not yet proved in the established models. Our work might offer a fresh perspective for the study of holographic code.
\end{abstract}

\maketitle

\tableofcontents


\section{Introduction}
Holographic quantum error-correcting code (HQEC)~\cite{almheiri2015,pastawski2015,dong2016,harlow2017} is proposed to interpret the emergent bulk locality and bulk reconstruction in a perturbative setting of the anti-de Sitter/conformal field theory (AdS/CFT) correspondence~\cite{witten1998,maldacena1999}. It has laid the foundation for more advanced studies of quantum gravity from the quantum-information perspective~\cite{donnelly2017,pastawski2017,almheiri2018,akers2019,cotler2019,hayden2020,osborne2020,penington2020,cao2021,harlow2021,akers2022a,akers2022,balasubramanian2023}, and has motivated the development of variant versions of the holographic correspondence, e.g., the $p$-adic AdS/CFT~\cite{gubser2017,heydeman2018a,bhattacharyya2018,hung2019,chen2021a,chen2021,yan2023,ebert2023}.

Intriguingly, HQEC has garnered increasing attention from broader areas of quantum physics. An earlier example is its special error-correction properties that have been an interesting topic in the study of quantum gates for fault-tolerant computation~\cite{faist2020,harris2020,cree2021,cao2022,farrelly2022,bao2022a}. More recently, within the vigorously evolving field of quantum simulation, the topic of ``quantum gravity in the lab'' has pointed out the significance and the feasibility to study systems featuring in properties of gravity through simulating the dual quantum systems~\cite{periwal2021,daley2022,anglesmunne2024,bluvstein2022,jafferis2022a,shapoval2022,bhattacharyya2022,brown2023,nezami2023,bluvstein2024,xu2024}. In this direction, quantum simulations of HQEC appears viable on near-term quantum devices, and are believed to capture salient aspects of the desired studies.

HQEC is essentially a hypothetical quantum-information structure. It views the bulk local degrees of freedom (logical qudits) as encoded in the CFT degrees of freedom (physical qudits) on the asymptotic boundary, and is believed to formalize the principle of holography about how bulk locality is emergent from the boundary entanglement~\cite{ryu2006,almheiri2015,harlow2017}. This hypothetical structure, though still in development, can be essentially encapsulated by a series of characteristics for many-qudit systems~\cite{almheiri2015,harlow2017,pastawski2017}, which ``qualify'' whether a quantum-error-correction construction captures the expected quantum-information interpretation of holography.


Accordingly, the study of HQEC in any specific context desires an exact model of holographic code which realizes these characteristics as a whole~\footnote{We use ``HQEC'' to refer to the hypothetical quantum-information structure for AdS/CFT, and use ``holographic code'' for models that are expected to realize the structure.}. To date, the tensor-network approach~\cite{jahn2021} plays the major role in building holographic codes, as it has brought brilliant insights into how a model can be ``assembled'' on the bulk, i.e., through the contraction of bulk local tensor structures on a network with hyperbolic geometry. The well-established models mainly include the pioneering HaPPY code~\cite{pastawski2015}, its elegant variations, and approximate versions~\cite{hayden2016,donnelly2017,evenbly2017,kim2017a,harris2018,jahn2019,jahn2019a,mcmahon2020,cao2021,jahn2021,pollack2022,steinberg2023}.

While these inspiring models have substantiated the significance of HQEC~\cite{almheiri2015,pastawski2015,dong2016,harlow2017,pastawski2017,cao2021}, the broadening focus and the deepening exploration of HQEC necessitate advancements in the construction of holographic codes (see the following subsection for details). On one hand, in the established models, the current methods still face challenges in demonstrating certain desired aspects of the hypothetical structure that correspond to properties of bulk reconstruction~\cite{jahn2021,cao2021}. On the other hand, the ``bulk-based'' tensor-network approach does not directly describe the structure of the boundary states~\cite{jahn2021}, and hence provides insufficient prescription for potential scalable quantum simulation of the code states. Importantly, it remains elusive what structures of boundary states can warrant the emergent bulk locality and determines the expected bulk reconstruction~\cite{harlow2017,jahn2021}, which is essential in the quantum-information interpretation of holography.

In this work, we present a ``boundary-based'' approach, contributing to address the above issues simultaneously. Particularly, our construction is based on ``elementary'' description of the physical qudits, and our arguments only rely on the necessary concepts in the general framework of quantum error correction, without introducing additional ingredients or methods. Hence, we anticipate that our work can initiate a novel perspective for building and studying holographic code in general and interdisciplinary contexts.

\subsection{Specific motivations}
More specifically, our work is motivated by the following detailed challenges in theoretical construction of exact holographic code and the prospect for quantum simulation studies based on scalable structures of boundary states. Given the following facts, investigating a boundary-based construction of holographic code indeed aligns naturally with the needs for advancements.

\subsubsection{Challenges in constructing exact models}
In theoretical studies of holographic code, there are challenges in demonstrating the hypothetical subregion duality within the formalism of operator-algebra quantum error correction (OAQEC)~\cite{harlow2017,jahn2021,cao2021}. Here, OAQEC is the general formalism of quantum error correction~\cite{beny2007,beny2007a}, which includes the relatively simpler subsystem-code formalism as a special case. The importance of this formalism is that the genuine OAQEC (distinct from the subsystem-code formalism) is essential to the description of the quantum correction to the order $G^0$ beyond the semi-classical-limit bulk gravity, and underlies a version of the Ryu-Takayanagi/Faulkner-Lewkowycz-Maldacena (RT/FLM) formula of entanglement entropy~\cite{ryu2006,faulkner2013,harlow2017}. Of particular note is that recently OAQEC is receiving significantly increasing attention from the community of quantum physics, and one of its most prominent advance is in the study of HQEC~\cite{dauphinais2024}.


One difficulty is to show that the genuine OAQEC and the condition of complementary recovery are satisfied for \emph{arbitrary} boundary bipartition (both connected and disconnected), which are expected in interpreting the AdS/CFT correspondence, and is the basis for systematically studying subregion duality in OAQEC formalism~\cite{harlow2017}. In the earlier models, e.g., the original HaPPY pentagon code and its variations~\cite{pastawski2015,jahn2021}, the subsystem-code formalism is present while the genuine OAQEC cannot be evidently realized. In the later improvement, e.g., hybrid tensor-network models~\cite{donnelly2017,cao2021}, the formalism of OAQEC can be only demonstrated for restricted cases of boundary bipartition.

Another difficulty is the systematical demonstration of uberholography~\cite{pastawski2017}. The properties of uberholography are featured by certain universality of the scaling behavior of bulk-operator reconstructions on the boundary, and are believed inherent and peculiar to HQEC with respect to the geometric manifestations of algebraic aspects of quantum error correction. However, in the established models, while certain evidence of uberholography can be observed~\cite{cao2021}, it remains elusive how to realize the universality of the scaling behavior.




At a more fundamental level beneath these difficulties lies the inadequacy of the currently developed method for rigorously demonstrating the complementary recovery or explicitly describing the expected structure of entanglement wedges and logical-operator subalgebras~\cite{cao2021,jahn2021}. Particularly, the tensor-network based greedy algorithm~\cite{pastawski2015} shows limitations in fully capturing the subregion duality. To address these issues from a fundamental perspective, we need deeper insights into the relationship between the description of the boundary and the reconstruction properties of the bulk.

However, studying the boundary states seems not straightforward in the tensor-network based construction. Indeed, the encoding in the tensor-network paradigm can be viewed as constructed from the bulk to the boundary, hence the geometry of the boundary (physical) qudits, as determined by the uncontracted tensor legs, is not uniquely or clearly specified, but depends on different inflation rules of growing the tensors. Moreover, the computation of the boundary states from the tensor contraction can be highly inefficient in large system size. These facts~\cite{jahn2021} necessitate novel considerations on the model construction, which can shed light on deeper understanding of HQEC from the boundary side.

\subsubsection{Potential quantum simulation studies}
Similar motivation comes from the challenges in potential quantum simulation studies of HQEC. Indeed, while some initial progress has been made towards realizing the holographic physics on analogue quantum simulator~\cite{periwal2021,daley2022}, it is difficult to find a realistic Hamiltonian that stabilizes the structure of HQEC as the low-energy physics. And a recent proposal~\cite{anglesmunne2024} for simulating the HaPPY code in digital quantum computer setups~\cite{daley2022} reveals the challenges in scalability: The scheme requires sufficient number of gates entangling bulk qubits and boundary qubits that are realized in the same quantum hardware, hence the feasibility is challenged by the increasing system size and the aim of sufficient realization of holographic properties relevant to the emergent bulk locality.


There is an exemplary case of a scalable scheme for quantum simulation of quantum error-correcting code, which might shed light on the case of holographic code. The example is the recent success in realizing the topological surface code (toric code) in digital quantum computers~\cite{kitaev2003,andersen2020,satzinger2021,marques2022,krinner2022,zhao2022,liu2022,googlequantumai2023}. The experimental schemes utilize the explicit structure of the topological code basis states, i.e., they can be defined from a product state together with a family of quantum gates as $\ket{\Psi_{\mathrm{toric}}}\sim\prod_p(\frac{\mathds{1}+B_p}{2})\ket{000\cdots}$, an equal-weight sum of product states specified by the $B_p$'s~\cite{kitaev2003,satzinger2021,liu2022}. Although the complete set of stabilizer generators of the surface code includes many other terms, the $B_p$ gates together with the product state $\ket{000\cdots}$ solely form a minimal data set on the physical qubits for a scalable description of the code structure: The $B_p$'s all have fixed form and finite support, and are hence independent on the system size. Furthermore, properties of the gates also comprise the knowledge that bridges the structures of the code states to the logical qubits and logical operators~\cite{kitaev2003}.

Despite the difference in the essence of quantum error correction, the advantage of the code-state-structure in the simulation of surface code prompts us to consider possible scalable description of the boundary code states in a holographic code. Indeed, in this line of thought, the recent development of quantum processor, e.g., the reconfigurable atom arrays~\cite{bluvstein2022,bluvstein2024,xu2024}, has unveiled a promising platform for efficiently simulating many-body entangled states of physical qubits, and attracts the attention for potential quantum simulation of holographic code.

However, as mentioned previously, for models built from tensor-network on bulk geometry, uncovering the structure of boundary code state can suffer from ambiguity of the boundary geometry and extensive computational workload in large scale. Hence, there lacks straightforward ways towards a concise and scalable description of the code-state structures~\cite{pastawski2015,jahn2021}. These facts further concretize the need for a new perspective on the construction of holographic codes, from which we can sufficiently access the boundary entanglement structures.

\subsection{Summary of main results}
With the aim to address the above two issues within one novel approach for building holographic codes, we take an ``opposite'' perspective to the tensor-network paradigm. In our approach, we do not resort to any additional structure like the tensor network, but only rely on certain ``elementary'' description of the physical qudits. Then, starting with the boundary, our construction of exact model simply exemplifies the principle of holography, i.e., how the bulk is emergent from the boundary entanglement. The reason enabling our approach is that the hypothetical HQEC structure inherently embodies the guiding principles for the model construction, i.e., the way to specify the bulk discrete geometry and the encoding isometry.

As we will show, an effective discrete hyperbolic geometry can be inherently specified in a proposal for ``demystifying'' the peculiar properties of uberholography. The key of the proposal is to identify the fractional number $h$ of the universal scaling component $1/h$ in uberholography as the fractional Hausdorff dimension of certain alternative geometry of the physical/boundary qudits, which is linked to the standard 1D geometry through a rearrangement operation. Note that the geometric rearrangement resembles the experimental dynamic reconfiguration of cold neutral atoms in the recent development of hybrid analogue-digital quantum simulation~\cite{bluvstein2022,bluvstein2024,xu2024}, and such coincidence can provide convenience for quantum simulation studies.

It will also be clear that if we conceive a boundary basis state as a sum of qudit-product-states, then hypothetical properties regarding the reconstruction of a single bulk qudit can derive detailed nonlocal constraints for those qudit-product-states, thereby specifying the entanglement patterns. Then, as a pivotal point in our approach, we can show how the logical/bulk degrees of freedom can emerge from the entanglement patterns of physical qudits, and how the encoding structure can be specified accordingly. Here, the term ``entanglement patterns'' is borrowed from the general description of many-body entanglement in the studies of various topics in condensed matter physics~\cite{zeng2019}, and appears to be applicable in broad contexts.

Then, the fundamental reason our construction can excel in comprehensively realizing the HQEC structure is as follows: (1) With explicit structures of the boundary code states we can investigate the reconstruction of the bulk local operators or local operator algebra of single bulk qudits, in different forms and on various boundary subregions. (2) Based on the knowledge of such reconstructions, for given boundary bipartition, the entanglement wedges consisting of reconstructable bulk qudits and the logical subalgebras consisting of reconstructable bulk operators can be completely specified.

A main part of this work is devoted to illustrate our approach in an exact model of holographic code. The illustrative model is built on ququarts (qudits with $d=4$) which aligns with recent trend on developing qudit-based experimental setups~\cite{low2020,chi2022,ringbauer2022,seifert2023a,liu2023,fischer2023}. The model is defined only based on a minimal data set describing the physical qudits, i.e., a family of two-qudit quantum gates $\{T_{ii'}\}$ together with a product state $\ket{000\cdots}$. And the boundary code basis state has the scalable form 
\begin{equation}\label{bcs1}
\ket*{\tilde{\varphi}_n}\propto\prod_{\{ii',jj',kk'\}}(\frac{\mathds{1}+T_{ii'}T_{jj'}T_{kk'}}{2})\ket{\alpha\alpha'\alpha''\cdots},
\end{equation}
where $\ket{\alpha\alpha'\alpha''\cdots}$ is a product state prepared from $\ket{000\cdots}$ through sequential application of the $\{T_{ii'}\}$ gates, which can be represented pictorially to depict the entanglement patterns. In terms of the scalable form, we describe example of succinct quantum-circuit simulation of the boundary states with the number of layers scaling sublinearly on the total number $N$ of the physical qudits. An interesting property of the scaling behavior is that it is upper bounded by $N^{\frac{1}{h}}$ as determined by the universal component in uberholography.

A distinguishing feature of our work is that merely based on the boundary entanglement patterns as given above, we not only rigorously specify the emergent bulk degrees of freedom and the encoding isometry, but also unfold the hypothetical HQEC structures by demonstrating the following characteristics.

For the aspects regarding the reconstruction of single bulk qudits, we show that \emph{i)} the connected code distance scales linearly on the system size and agrees with the radial direction from the center of the bulk to the boundary; \emph{ii)} the (disconnected) minimal boundary subregion for recovering individual bulk qudits scales sublinearly on the system size with exactly the form $\sim N^{\frac{1}{h}}$, realizing the universal scaling component $1/h$ and underlying systematic description of uberholography.

For the aspects regarding subregion duality, \emph{i)} we show that the condition of complementary recovery and the genuine OAQEC formalism (nontrivial center of subalgebra) are satisfied for arbitrary boundary bipartitions, including both connected and disconnected cases; \emph{ii)} we establish a general framework to exactly describe the entanglement wedges and the subalgebras, which translates between the geometric and the algebraic aspects of subregion duality and seems rarely presented in the literature; \emph{iii)} based on the framework we can exactly extract the bulk terms and the area (operator) terms of the RT formula in OAQEC which is in a version of the RT/FLM formula; \emph{iv)} the representative geometric illustration of subregion duality, as shown below, qualitatively agrees with the expected properties.

\begin{center}
\begin{figure}[h]
\centering
    \includegraphics[width=8.5cm]{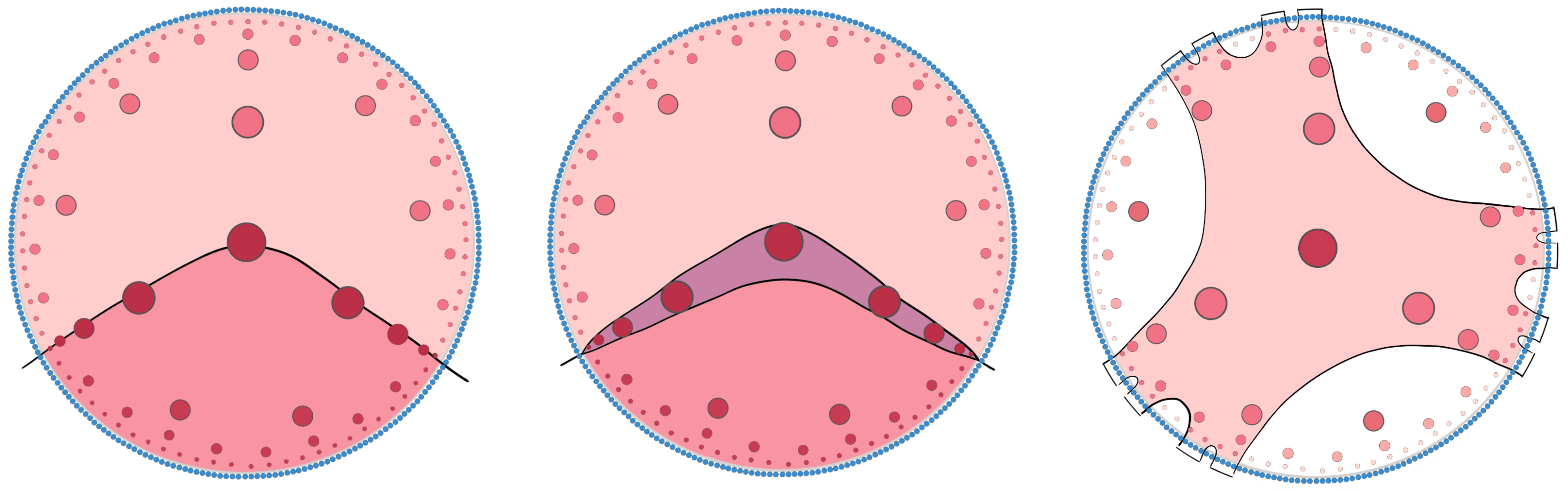}   
\end{figure}
\end{center}

The above illustration for our model show that subregion duality within genuine OAQEC is featured by the presence of an entangling surface as the overlap of two entanglement wedges, which is distinct from the subsystem-code case with geometric complementarity. The left shows a conventional depiction in which a minimal surface threads the bulk qudits ``shared'' by the two subalgebras. In our arguments, we adopt the equivalent depiction as shown in the middle which clearly specify the entangling surface and the overlapped entanglement wedges. The right is an exact example illustrating uberholography for the recovery of the central bulk qudit, where the entanglement wedge corresponds to a disconnect and measure-zero boundary subregion.

The exact holographic code might be regarded as a minimalistic model of HQEC in the following sense: On one hand, it suffices to formally realize the hypothetical structures and hence confirms the feasibility and advantage of our approach, and also shows promise in potential quantum simulation studies; on the other hand, as the initial model for our approach, it can be also further extended or developed towards more complex and hence more optimal models in studying AdS/CFT, which will be covered in our text.

\subsection{Outline}
The following text is organized as follows: In Sec.~\ref{pre} we concisely describe the hypothetical HQEC structure as a series of characteristics which will be demonstrated in our exact model. It is noticeable that in Sec.~\ref{pre}, we discuss how to systematically describe subregion duality in genuine OAQEC and present criterion for complementary recovery, which might be new since they are insufficiently addressed in the literature. In Sec.~\ref{idea} we discuss the basic idea in our approach for building holographic code, which is illustrated by the construction of our exact model in Sec~\ref{exmodel}. Then, based on the construction, we directly prove the condition of complementary recovery for our model in Sec.~\ref{crsec}, and demonstrate the HQEC characteristics regarding the reconstruction of single bulk qudit in Sec.~\ref{minisec}.

Results in Sec.~\ref{crsec} and \ref{minisec} can be viewed as a bridge between the construction of the exact model and the demonstration of the main HQEC characteristics. Based on these results, in Sec.~\ref{sdrt}, we systematically demonstrate the characteristics regarding subregional duality, including both the algebraic and the geometric aspects, and illustrate how to explicitly extract the terms in the RT formula. In the remark at the end of Sec.~\ref{sdrt}, we discuss possible way to extend the exact model and also briefly discuss how the model and the extensions are relevant to the quantum-information studies of the AdS/CFT or the $p$-adic AdS/CFT correspondence.

Note that an important part of our results will be presented in the form of theorems. In this work, all the proved theorems without specifically referring to our model apply to the general setting for studying holographic code.



\section{Preliminaries}\label{pre}
Before introducing our basic idea for the new approach, we give a compact description of the characteristics of HQEC, which will be demonstrated in our exact model. The demonstrations will be either in the form of theorems, or in detailed descriptions together with pictorial illustrations.

These characteristics were initially hypothesized and developed in Ref.~\cite{almheiri2015,harlow2017,pastawski2017}, but remains to be systematically and sufficiently addressed in the language of OAQEC. Indeed, since the proposal of OAQEC~\cite{beny2007,beny2007a}, this general formalism of quantum error correction has received limited attention until recent significantly renewed focus~\cite{dauphinais2024}. And in the studies of holographic code, realizing certain desired HQEC characteristics, including the conditions of genuine OAQEC together and complementary recovery for \emph{all} boundary bipartition, seems overly stringent for the current methods of model construction. Hence, the subsystem-code formalism has been the major framework to understand the subregion duality within the quantum-information interpretation.

Here, to add necessary ingredients, we describe the HQEC characteristics in the language of OAQEC. Since OAQEC is general and includes the subsystem-code formalism as a special case, our description is consistent to those in the literature. And when necessary we will particularly indicate that the necessity of genuine OAQEC condition for subregion duality.

There are three parts in our discussion on the HQEC characteristics, which might not be covered conventionally in the literature: We clarify important aspects of the condition of complementary recovery that were insufficiently addressed, and give two criteria for demonstrating the condition in the language of OAQEC. We show how to formally and generally capture the meaning of the entanglement wedge and the corresponding geometric presentation within the formalism of OAQEC. We also describe how to systematically capture the conditions of uberholography in an exact model of holographic code.




\subsection{Concepts and notations}\label{basic}
Here, we give a brief but comprehensive description of the mathematical concepts that are necessary for presenting the essential characteristics of HQEC within the formalism of OAQEC. For readers who are familiar with OAQEC, this part can be skipped in the first reading.

The basic perspective of HQEC is viewing local degrees of freedom in the bulk as logical qudits that are encoded in the boundary (physical) qudits. Here, a qudit $\mathbb{C}^d$ is simply a $d$-dimensional Hilbert space. The construction of a holographic code consists of two ingredients: (1) Discretization of the 1D boundary that specifies the geometry of physical qudits $\{\mathfrak{h}_1=\mathbb{C}^d,\mathfrak{h}_2=\mathbb{C}^d,\mathfrak{h}_3=\mathbb{C}^d,\ldots\}$, and discretization of the 2D hyperbolic bulk that specifies the geometry of the logical qudits $\{\mathfrak{e}_{\boldsymbol{1}}=\mathbb{C}^{d'},\mathfrak{e}_{\boldsymbol{2}}=\mathbb{C}^{d'},\mathfrak{e}_{\boldsymbol{3}}=\mathbb{C}^{d'},\ldots\}$. (2) An encoding isometry $\mathcal{E}\xrightarrow{R}\mathcal{H}$, i.e., a linear transformation which maps the states in the bulk Hilbert space $\mathcal{E}=\mathfrak{e}_{\boldsymbol{1}}\otimes\mathfrak{e}_{\boldsymbol{2}}\otimes\mathfrak{e}_{\boldsymbol{3}}\otimes\cdots=(\mathbb{C}^{d'})^{\otimes K}$ to the states in the boundary Hilbert space $\mathcal{H}=\mathfrak{h}_1\otimes\mathfrak{h}_2\otimes\mathfrak{h}_3\otimes\cdots=(\mathbb{C}^d)^{\otimes N}$, and satisfies $R^+R=\mathds{1}_{\mathcal{E}}$. Here, $\mathds{1}_{\mathcal{E}}$ is the identity operator on $\mathcal{E}$, and $R^+R=\mathds{1}_{\mathcal{E}}$ means that the linear map $R$ preserves the inner product, i.e. $\braket{R\psi}{R\psi'}_{\mathcal{H}}=\braket{\psi}{\psi'}_{\mathcal{E}}$.


As a convention of notation, we use $i,i',i_1,j,k,\ldots$ to label boundary qudits with specific locations, and use the bold $\boldsymbol{x},\boldsymbol{x}',\boldsymbol{x}_1,\ldots$ for bulk qudits. These notations can refer to specific qudits in the discrete geometry, or equivalently the corresponding sites of the discrete geometry on which the qudits live. We use $\ket{\alpha_i},\ket{\alpha'_i},\ket{\alpha''_i},\ldots\in\mathbb{C}^d$ to denote the state of a single local boundary qudit, and use $\ket{\boldsymbol{\beta}_{\boldsymbol{x}}},\ket{\boldsymbol{\beta}'_{\boldsymbol{x}}},\ket{\boldsymbol{\beta}''_{\boldsymbol{x}}},\ldots\in\mathbb{C}^{d'}$ to denote the state of a single local bulk qudit. In our notation, for a local qudit (either boundary or bulk), the labels $\alpha_i$ and $\boldsymbol{\beta}_{\boldsymbol{x}}$ are themselves the indices of local bases of states, while for the whole boundary or bulk system, we add a subscript as the indices, e.g., $\ket{\psi_m}$, $\ket{\widetilde{\varphi}_n}$ and $\ket{\boldsymbol{B}_{\boldsymbol{n}}}$.

\subsubsection{Reconstruction of logical operators}
The image of $R$ is the code subspace $\mathcal{H}_{\mathrm{code}}=P_{\mathrm{code}}\mathcal{H}\subset\mathcal{H}$ with the projection operator $P_{\mathrm{code}}=RR^+$ mapping states in $\mathcal{H}$ onto states in $\mathcal{H}_{\mathrm{code}}$. We use $\mathbf{L}(\mathcal{E})$, $\mathbf{L}(\mathcal{H})$ and $\mathbf{L}(\mathcal{H}_{\mathrm{code}})$ to denote the set of operators, which can be viewed as a vector space or an algebra~\footnote{With the operator summation and the multiplication by a scalar, $\mathbf{L}(\mathcal{H})$ can be viewed as a vector space; with the operator summation and the operator multiplication, $\mathbf{L}(\mathcal{H})$ can be viewed as a ring in which the identity operator plays the role of the identity element; the coexistence of the vector-space and ring structures defines $\mathbf{L}(\mathcal{H})$ as an algebra.}. In the following text, we use the notations $O,O',O_1,\ldots$ for general operators in $\mathbf{L}(\mathcal{H})$; use the notations $\widetilde{O}, \widetilde{O}',\widetilde{O}_1,\ldots$ for logical operators on $\mathcal{H}_{\mathrm{code}}$; and use the notations $\widetilde{\boldsymbol{O}}$ for bulk operators on $\mathcal{E}$.

It might be noticeable to the mathematically inclined readers that for simplicity, we identify the operators in $\mathbf{L}(\mathcal{H}_{\mathrm{code}})$ as operators $O\in\mathbf{L}(\mathcal{H})$ with $O=P_{\mathrm{code}}OP_{\mathrm{code}}$. The latter are simply those operators that act as the zero operator on the orthogonal complement $(\mathcal{H}_{\mathrm{code}})^{\perp}$ and leave $\mathcal{H}_{\mathrm{code}}$ invariant under the actions, which surely commute with $P_{\mathrm{code}}$. Within our arguments for quantum error correction and operator reconstruction, the difference between the two types plays no role, and it is more conventional~\cite{harlow2017} and convenient to work with the latter.




The central theme in HQEC is to investigate how bulk operators can be reconstructed, or recovered on various boundary subregion $A$ by studying the quantum-error-correction property of $\mathcal{H}_{\mathrm{code}}$ against the erasure of the complement boundary subregion $\overline{A}$. Here, a boundary subregion $A=\{i_1,i_2,\ldots\}$ can be understood as a subset of the boundary qudits, or equivalently the subset of the corresponding sites of the discrete boundary geometry. The bipartition $A\overline{A}$ of physical qudits specifies a decomposition $\mathcal{H}=\mathcal{H}_{A}\otimes\mathcal{H}_{\overline A}$, i.e., a tensor product of $\mathcal{H}_A=\otimes_{i\in A}\mathfrak{h}_i$ and $\mathcal{H}_{\overline A}=\otimes_{i\in \overline{A}}\mathfrak{h}_i$. 

For simplicity, we use the notations $\mathds{1}_A$ and $\mathds{1}_{\overline{A}}$ for the identity operators $\mathds{1}_{\mathcal{H}_A}$ and $\mathds{1}_{\mathcal{H}_{\overline{A}}}$ respectively. We use $O_A,O'_A,Q_A,\ldots$ to simplify the notation $O_A\otimes\mathds{1}_{\overline{A}},O'_A\otimes\mathds{1}_{\overline{A}},Q_A\otimes\mathds{1}_{\overline{A}},\ldots$, and use $O_{\overline{A}},O'_{\overline{A}},Q_{\overline{A}},\ldots$ in a similar way.

We say that a logical operator $\widetilde{O}$ can be reconstructed on subregion $A$ if and only if $\widetilde{O}=P_{\mathrm{code}}O_AP_{\mathrm{code}}$ for some $O_A$ with $[O_A,P_{\mathrm{code}}]=0$. Note that in some literature~\cite{harlow2017}, the reconstruction is expressed in another equivalent way: $\widetilde{O}\ket*{\tilde{\psi}}=O_A\ket*{\tilde{\psi}}$ and $\widetilde{O}^+\ket*{\tilde{\psi}}={O_A}^+\ket*{\tilde{\psi}}$ for any code state $\ket*{\tilde{\psi}}\in\mathcal{H}_{\mathrm{code}}$~\footnote{As a basic property of operators, a subspace $\mathcal{H}_{\mathrm{code}}$ is invariant under the action of both $O_A$ and ${O_A}^+$ if and only if $[O_A,P_{\mathrm{code}}]=0$.}.

\subsubsection{von Neumann algebra on $\mathcal{H}_{\mathrm{code}}$ and OAQEC}
Given a boundary bipartition $A\overline{A}$, the formalism of OAQEC describes whether a von Neumann algebra $\mathcal{M}$ of logical operators can be protected against the erasure errors on $\overline{A}$, which generalizes the subsystem code formalism and the conventional quantum erasure correction.

By a von Neumann algebra $\mathcal{M}$ on the finite-dimensional Hilbert space $\mathcal{H}_{\mathrm{code}}$, we mean a subalgebra of $\mathbf{L}(\mathcal{H}_{\mathrm{code}})$ (including the identity operator $\mathds{1}_{\mathcal{H}_{\mathrm{code}}}$) in which the adjoint operator $\widetilde{O}^+$ of each element $\widetilde{O}\in\mathcal{M}$ also belongs to $\mathcal{M}$. In the case of finite-dimensional Hilbert space, the von Neumann algebra generated by a given collection of operators $\{\mathds{1}_{\mathcal{H}_{\mathrm{code}}},O_1,O_2,\ldots\}$ (including the identity) simply consists of all the linear combinations of products of the form $O_2O^+_3\cdots$, i.e. of operators in the collection and their adjoints.

Given $\mathcal{M}$, we can define another von Neumann algebras $\mathcal{M}'$, the commutant of $\mathcal{M}$, i.e. the collection of all operators in $\mathbf{L}(\mathcal{H}_{\mathrm{code}})$ that commute with every operator in $\mathcal{M}$. We define the the center of $\mathcal{M}$ as $\mathrm{Z}(\mathcal{M})=\mathcal{M}\cap\mathcal{M}'$, i.e. the operators in $\mathcal{M}$ that commute with every operator in $\mathcal{M}$, which is also a von Neumann algebra. A fundamental result from the theory of von Neumann algebra is that $\mathcal{M}''=(\mathcal{M}')'=\mathcal{M}$, and hence $\mathrm{Z}(\mathcal{M})=\mathrm{Z}(\mathcal{M}')$. And an important property of the commutant is that $\mathcal{M}_1\subset\mathcal{M}_2$ implies $\mathcal{M}'_2\subset\mathcal{M}'_1$. When we say that $\mathcal{M}$ can be reconstructed on the boundary subregion $A$, we mean that every operator $\widetilde{O}\in\mathcal{M}$ can be reconstructed on $A$.

Based on the language of von Neumann algebra, we can establish an equivalence between the correctability against the erasure of $\overline{A}$ and the logical-operator reconstruction on $A$.
\begin{lemma}[See Ref.~\cite{harlow2017,pastawski2017}]\label{oaqec}
Consider a von Neumann algebra $\mathcal{M}$ on $\mathcal{H}_{\mathrm{code}}\subset\mathcal{H}$ and a decomposition $\mathcal{H}=\mathcal{H}_{A}\otimes\mathcal{H}_{\overline A}$. Then, the (erasure of) subregion $\overline{A}$ is correctable with respect to $\mathcal{M}$, if and only if for every operator $O_{\overline{A}}$ supported on $\overline{A}$ we have $P_{\mathrm{code}}O_{\overline{A}}P_{\mathrm{code}}\in\mathcal{M}'$; if and only if every logical operator in $\mathcal{M}$ can be reconstructed on subregion $A$.
\end{lemma}

This lemma can be viewed as a definition of the OAQEC formalism of quantum error correction for erasure errors~\cite{pastawski2017}. It borrows the term ``reconstruction'' from the context of the AdS/CFT correspondence, and underlies the general quantum-information interpretation of holography~\cite{almheiri2015,harlow2017}. Within this generality, the subsystem-code formalism for $\mathcal{M}$ is the special case when $\mathrm{Z}(\mathcal{M})=\mathbb{C}\mathds{1}_{\mathcal{H}_{\mathrm{code}}}$, and hence $\mathcal{M}$ is called a factor. Here $\mathbb{C}\mathds{1}_{\mathcal{H}_{\mathrm{code}}}$ denotes the trivial von Neumann algebra consisting of constant operators. If we further have trivial $\mathcal{M}'$, i.e., $\mathcal{M}=\mathbf{L}(\mathcal{H}_{\mathrm{code}})$, then we will be in the conventional sense of quantum error correction of erasure errors. Note that by the genuine OAQEC formalism, we require that $\mathrm{Z}(\mathcal{M})\ne\mathbb{C}\mathds{1}_{\mathcal{H}_{\mathrm{code}}}$, i.e., the center is nontrivial. Sometimes we omit the subscript of $\mathds{1}_{\mathcal{H}_{\mathrm{code}}}$ as simply $\mathds{1}$ when the corresponding Hilbert space is clear in the context.

\subsubsection{Tensor product of local bulk von Neumann algebras}\label{vnaob}
In our arguments, we use the terms ``logical operator'' and ``bulk operator'' interchangeably for both operators on $\mathcal{H}_{\mathrm{code}}$ and operators on $\mathcal{E}$, which can be identified as one another. That is because when we view $R$ as the unitary map between $\mathcal{E}$ and $\mathcal{H}_{\mathrm{code}}$, $\mathbf{L}(\mathcal{E})\xrightarrow{R\boldsymbol{\cdot} R^+}\mathbf{L}(\mathcal{H}_{\mathrm{code}})$ and $\mathbf{L}(\mathcal{H}_{\mathrm{code}})\xrightarrow{R^+\boldsymbol{\cdot} R}\mathbf{L}(\mathcal{E})$ are isomorphisms of operator algebras. In other words, the algebraic properties of $\boldsymbol{\mathcal{M}}=R^+\mathcal{M}R$ on $\mathcal{E}$ are the same as those of $\mathcal{M}$ on $\mathcal{H}_{\mathrm{code}}$, and reversely, the algebraic properties of $\mathcal{M}=R\boldsymbol{\mathcal{M}}R^+$ on $\mathcal{H}_{\mathrm{code}}$ is the same as those of $\boldsymbol{\mathcal{M}}$ on $\mathcal{E}$. For example, we have $\mathcal{M}'=(R\boldsymbol{\mathcal{M}}R^+)'=R\boldsymbol{\mathcal{M}}'R^+$ and $\mathrm{Z}(\mathcal{M})=R\mathrm{Z}(\boldsymbol{\mathcal{M}})R^+$.

To be consistent with the convention in the literature and to also follow the previous definition, when discussing the reconstruction of an operator $\widetilde{\boldsymbol{O}}$ on $\mathcal{E}$, we actually work on operators $\widetilde{O}=R\widetilde{\boldsymbol{O}}R^+$ on $\mathcal{H}_{\mathrm{code}}$. And when discussing the geometric connotation of a von Neumann algebra $\mathcal{M}$ on $\mathcal{H}_{\mathrm{code}}$ in the study of subregion duality, i.e., its structure described in terms of local bulk qudit operators, we study the von Neumann algebra $\boldsymbol{\mathcal{M}}=R^+\mathcal{M}R$ on $\mathcal{E}$ and show how it can be represented by tensor product of local bulk operator algebras. In the latter case, with respect to the tensor product structure $\mathcal{E}=\cdots\otimes\mathfrak{e}_{\boldsymbol{x}}\otimes\mathfrak{e}_{\boldsymbol{x}'}\otimes\mathfrak{e}_{\boldsymbol{x}''}\otimes\cdots$, a tensor product of local bulk von Neumann algebras $\cdots\otimes\boldsymbol{\mathcal{M}}_a(\boldsymbol{x})\otimes\boldsymbol{\mathcal{M}}_a(\boldsymbol{x}')\otimes\boldsymbol{\mathcal{M}}_a(\boldsymbol{x}'')\otimes\cdots$ is a von Neumann algebra on $\mathcal{E}$ which consists of all linear combinations of operators of the form $\cdots\otimes\widetilde{\boldsymbol{O}}_{\boldsymbol{x}}\otimes\widetilde{\boldsymbol{O}}_{\boldsymbol{x}'}\otimes\widetilde{\boldsymbol{O}}_{\boldsymbol{x}''}\otimes\cdots$ with $\widetilde{\boldsymbol{O}}_{\boldsymbol{x}}\in\boldsymbol{\mathcal{M}}_a(\boldsymbol{x})\subset\mathbf{L}(\mathfrak{e}_{\boldsymbol{x}}),\widetilde{\boldsymbol{O}}_{\boldsymbol{x}'}\in\boldsymbol{\mathcal{M}}_a(\boldsymbol{x}')\subset\mathbf{L}(\mathfrak{e}_{\boldsymbol{x}'}),\widetilde{\boldsymbol{O}}_{\boldsymbol{x}''}\in\boldsymbol{\mathcal{M}}_a(\boldsymbol{x}'')\subset\mathbf{L}(\mathfrak{e}_{\boldsymbol{x}''}),\dots$.

Note that when it is necessary to distinguish a single-bulk-qudit operator $\widetilde{\boldsymbol{O}}_{\boldsymbol{x}}\in\mathbf{L}(\mathfrak{e}_{\boldsymbol{x}})$ and the corresponding bulk operator $\cdots\otimes\mathds{1}_{\mathfrak{e}_{\boldsymbol{x}'}}\otimes\widetilde{\boldsymbol{O}}_{\boldsymbol{x}}\otimes\mathds{1}_{\mathfrak{e}_{\boldsymbol{x}''}}\otimes\cdots\in\mathbf{L}(\mathcal{E})$ with effect only on the qudits $\boldsymbol{x}$, we use the abbreviated notation $\cdots\otimes\widetilde{\boldsymbol{O}}_{\boldsymbol{x}}\otimes\cdots$ for the latter.

In the discussion on subregion duality, we use the notation $\mathrm{W}[A]$ for certain subset of bulk qudits associated to a boundary subregion $A$. For simplicity, we sometimes use the notation $\otimes_{\boldsymbol{x}\in\mathrm{W}[A]}\mathbf{L}(\mathfrak{e}_{\boldsymbol{x}})$ for the complete form $(\otimes_{\boldsymbol{x}\in\mathrm{W}[A]}\mathbf{L}(\mathfrak{e}_{\boldsymbol{x}}))\otimes(\otimes_{\boldsymbol{x}\notin\mathrm{W}[A]}\mathbb{C}\mathds{1}_{\mathfrak{e}_{\boldsymbol{x}}})$ of all bulk operators supported on bulk qudits in $\mathrm{W}[A]$. In the following texts, we always adopt this type of convention to omit the identity operators or the trivial parts in a tensor product when the meaning is clear in the context. 

Note that we say that a bulk qudit $\boldsymbol{x}$ lies outside the support of a bulk operator $\widetilde{\boldsymbol{O}}$ if the operator acts trivially on $\boldsymbol{x}$, i.e., the operator can be written in a product form with $\mathds{1}_{\mathfrak{e}_{\boldsymbol{x}}}$. Then, bulk qudits not of this type form the support of $\widetilde{\boldsymbol{O}}$. We say a bulk qudit $\boldsymbol{x}$ lies outside the support of $\boldsymbol{\mathcal{M}}$, if every operator in $\boldsymbol{\mathcal{M}}$ acts trivially on the qudit. Then, the support of $\boldsymbol{\mathcal{M}}$ can be specified.



\subsection{HQEC characteristics}\label{hqecc}

Now, we consider an encoding isometry $(\mathbb{C}^{d'})^{\otimes K}=\mathcal{E}\xrightarrow{R}\mathcal{H}=(\mathbb{C}^d)^{\otimes N}$, and view the physical qudits and the logical qudits as arranged in the discrete boundary geometry and the bulk geometry respectively (see Fig.~\ref{fig1}). The expected characteristics of a holographic code mainly describe how the algebraic properties of bulk operator reconstruction or quantum error correction manifest in geometric properties of the bulk~\cite{almheiri2015,harlow2017,pastawski2017,cree2021,cao2021}, which formalize the intuitive picture that the bulk is emergent from the boundary entanglement. In the following brief discussion of the HQEC characteristics, a main part will be devoted to the subregion duality in the genuine OAQEC formalism, since it has certain distinct meaning from the commonly studied cases in the subsystem-code formalism, and is insufficiently addressed in the literature.

\subsubsection{A word on the discrete bulk geometry}
The discretization of the bulk geometry is conventionally regarded as a regular tessellation of the Poincar\'e disk in which each bulk qudit lives on a tile~\cite{pastawski2015}. There are also non-regular tessellations~\cite{jahn2021,cao2021} in which the bulk qudits only live on a part of the tiles and are hence more sparse. In either way, even in the large-system-size limit, the discrete description of the bulk cannot fully capture the continuous geometry in the original AdS/CFT correspondence, e.g., properties of the symmetries. In addressing this gap, an interesting proposal is to alternatively consider the $p$-adic AdS/CFT~\cite{gubser2017,heydeman2018a,bhattacharyya2018,hung2019} which replaces the reals $\mathbb{R}$ for the boundary by the $p$-adic numbers $\mathbb{Q}_p$~\footnote{$\mathbb{Q}_p$ is simply an alternative geometric completion of the rationals $\mathbb{Q}$ with respect to the alternatively defined norm associated with the prime $p$.} and hence possesses intrinsically discrete bulk geometry. In this case, the bulk could be viewed as either an abstract hyperbolic tree graph without any embedding on the AdS continuous space, or an tessellation of the Poincar\'e disk with ideal polygons (with vertex points live on the asymptotic boundary). It remains a challenge to study the $p$-adic AdS/CFT in a concrete holographic code.

Considering the various types of bulk geometry discretization and the gap in capturing the symmetry, we may view the structure of HQEC as more focused on how the boundary entanglement can qualitatively guarantee the expected bulk reconstruction. Hence, it is inclusive to only require the basic features for reading the data of the bulk geometry from the boundary entanglement. That are simply the necessary characteristics needed to describe the following HQEC characteristics on the bulk reconstruction, i.e., a hyperbolic tessellation which can represent the radial direction and properties relevant to the subregion duality. In the tensor network paradigm, the necessary characteristics is equivalent to the dual graph of the tessellation, e.g., the bulk network for the tensor contractions.  

\subsubsection{Reconstruction of bulk local operators}
Foundationally, HQEC interprets the counter-intuitive properties from the AdS/Rindler construction as inherent quantum-error-correction characteristics of $\mathcal{H}_{\mathrm{code}}\subset\mathcal{H}$. For instance, it interprets the overlapping causal wedges for the recovery of local bulk operators as the multiple reconstruction of local bulk operators on different boundary subregions. Here, the local bulk operator algebra on a bulk qudit (located at) $\boldsymbol{x}$, i.e. $\mathbf{L}(\mathfrak{e}_{\boldsymbol{x}})$, can be identified with the von Neumann algebra $\mathcal{M}(\boldsymbol{x})=R(\cdots\otimes\mathbb{C}\mathds{1}_{\mathfrak{e}_{\boldsymbol{x}'}}\otimes\mathbf{L}(\mathfrak{e}_{\boldsymbol{x}})\otimes\mathbb{C}\mathds{1}_{\mathfrak{e}_{\boldsymbol{x}''}}\otimes\cdots)R^+$ on $\mathcal{H}_{\mathrm{code}}$. Note that the von Neumann algebra $\cdots\otimes\mathbb{C}\mathds{1}_{\mathfrak{e}_{\boldsymbol{x}'}}\otimes\mathbf{L}(\mathfrak{e}_{\boldsymbol{x}})\otimes\mathbb{C}\mathds{1}_{\mathfrak{e}_{\boldsymbol{x}''}}\otimes\cdots$ on $\mathcal{E}$ consists of all the bulk operators $\cdots\otimes\mathds{1}_{\mathfrak{e}_{\boldsymbol{x}'}}\otimes\widetilde{\boldsymbol{O}}_{\boldsymbol{x}}\otimes\mathds{1}_{\mathfrak{e}_{\boldsymbol{x}''}}\otimes\cdots$ which act trivially on the bulk qudits other than $\boldsymbol{x}$, and is isomorphic to $\mathbf{L}(\mathbb{C}^{d'})$.

As a primary characteristic, a holographic code needs to reconcile the contradiction between the radial commutativity in AdS/CFT and the basic property of local quantum field theory. That is, nontrivial bulk operators can commute with all local boundary operators just as logical operators are protected against (and hence commute with) all local erasure errors of certain size~\cite{almheiri2015,harlow2017}. Furthermore, certain measures of how well the bulk local information is protected should be consistent with the bulk radial direction: the closer to the center of bulk, the better the information is protected. Such measure can be captured by the \emph{connected distance} $\mathrm{d_c}({\boldsymbol{x}})$, i.e. the smallest size of connected boundary subregion whose erasure is not correctable with respect to $\mathcal{M}(\boldsymbol{x})$~\cite{pastawski2017,cree2021}. Formally the measure is defined as
\begin{align*}
\begin{split}
&\mathrm{d_c}({\boldsymbol{x}})=\min_{\overline{A}}\{|\overline{A}|: \mathcal{M}(\boldsymbol{x})~\mathrm{cannot~be~reconstructed~on}~A\}\\
&=\min_{\overline{A}}\{|\overline{A}|: \overline{A}~\mathrm{is~not~correctable~with~respect~to}~\mathcal{M}(\boldsymbol{x})\}
\end{split}
\end{align*}
where $\overline{A}$ runs through only \emph{connected} boundary subregions and $|\overline{A}|$ is the size of $\overline{A}$. The connectedness respects the boundary locality and geometry. We can also define the connected price $\mathrm{p_c}({\boldsymbol{x}})$ as the smallest \emph{connected} boundary subregion that recover all the local bulk operator on the bulk qudit $\boldsymbol{x}$. That is 
\begin{equation*}
\mathrm{p_c}({\boldsymbol{x}})=\min_{A}\{|A|: \mathcal{M}(\boldsymbol{x})~\mathrm{can~be~reconstructed~on}~A\}.
\end{equation*}

Then, we can specify the following characteristics.



\paragraph*{\textbf{Characteristic 1}}\label{chac1} To demonstrate the radial commutativity, $\mathrm{d_c}({\boldsymbol{x}})$ should scales linearly on the total number $N$ of boundary qudits. This way, by definition, the bulk qudit $\boldsymbol{x}$ can be protected against arbitrary boundary local erasure errors with size smaller than $\mathrm{d_c}({\boldsymbol{x}})$. And according to Lemma~\ref{oaqec}, bulk operators on $\boldsymbol{x}$ commute with any local boundary operator with size smaller than $\mathrm{d_c}({\boldsymbol{x}})$. It is also required that moving from the boundary to the center of the bulk, $\mathrm{d_c}({\boldsymbol{x}})$ increase along the radial direction.

Note that here we emphasize the importance of the connected code distance $\mathrm{d_c}({\boldsymbol{x}})$ over that of the ``standard'' code distance $\mathrm{d}({\boldsymbol{x}})$. The emphasis is due to the following three reasons: (1) $\mathrm{d_c}({\boldsymbol{x}})$ respects the locality and geometry of the boundary qudits, and is hence relevant to the essence of radial commutativity~\cite{almheiri2015,cree2021}. (2) $\mathrm{d}({\boldsymbol{x}})$ is not guaranteed as increasing with system size in many models of holographic code, but this seems not to affect the demonstration of important bulk reconstruction properties. For instance, in the premier toy model for studying HQEC, the HaPPY pentagon code, the code distance $\mathrm{d}({\boldsymbol{x}})$ for the central bulk qubit is a constant, i.e., two pairs of apart boundary qubits can recover certain nontrivial bulk operators on the central bulk qubit~\cite{pastawski2015}. (3) With different geometry or connectedness, the size of boundary recovery for bulk qudits can be quite different~\cite{pastawski2017}. Indeed, according to uberholography, we have $\lim_{N\to\infty}\frac{\mathrm{d}({\boldsymbol{x}})}{\mathrm{d_c}({\boldsymbol{x}})}=0$, which means that the distance is negligible relative to the connected distance in large system size. This will be further discussed later.

\subsubsection{Subregion duality: complementary recovery and OAQEC}
A more general perspective on the bulk reconstruction considers what bulk operators can be reconstructed on a given boundary subregion $A$ or its complement $\overline{A}$. In AdS/CFT (with semi-classical bulk gravity), this is captured by subregion duality for a boundary bipartition $A\overline{A}$, i.e., the assignment of two entanglement wedges, as two bulk subregions, to $A$ and $\overline{A}$ respectively. The entanglement wedges are expected to underlie (as the support of) the bulk operators that can be reconstructed on $A$ and $\overline{A}$ respectively. And in the continuum limit, we can define the entanglement wedges as the bulk domain of dependence in an achronal surface that is separated by a minimum-area co-dimension-one surface and enclosed by $A$ ($\overline{A}$) together with the minimal surface~\cite{almheiri2015,harlow2017,pastawski2017,cree2021}.

In a holographic code (finite system), the characteristics of subregion region  duality include both algebraic and geometric aspects, which are dependent on the conditions of complementary recovery~\footnote{In this work we only consider the state-independent or exact complementary recovery~\cite{cao2021}.} and the formalism of quantum error correction. Generally, within the formalism of OAQEC (including certain sub-leading quantum corrections to the semi-classical bulk gravity), the formal description of subregion duality relies on the logical subalgebras in stead of merely the bulk degrees of freedom. In this part, we describe the necessary algebraic aspects for the OAQEC formalism, and then in the next part we describe how these algebraic aspects manifest in the bulk geometry.

We firstly clarify the condition of complementary recovery for $\mathcal{H}_{\mathrm{code}}$ as defined in Ref.~\cite{harlow2017}: the simultaneous recoveries of certain von Neuamnn algebra $\mathcal{M}$ (of logical operators in $\mathcal{H}_{\mathrm{code}}$) and its commutant $\mathcal{M}'$ on boundary subregions $A$ and $\overline{A}$ respectively. Indeed, the meaning of this condition has not yet been sufficiently elucidated. For example, it is not clear from the original definition which $\mathcal{M}$ ($\mathcal{M}'$) is to be reconstructed.

In the established models, since the greedy wedge based on the tensor-network structure can play a substitute role of the entanglement wedge for certain cases of boundary bipartition, thorough comprehension on the condition of complementary recovery might be avoided. However, if only relying on the greedy wedge technics, for many cases of boundary bipartition, even connected cases, the complementarity is ambiguous~\cite{jahn2021}.

Hence, to study the entanglement wedge without tensor network, we should demonstrate the subregion duality in a more general and rigorous way. In the following, we clarify the uniqueness of $\mathcal{M}$ ($\mathcal{M}'$) to be reconstructed, and then develop criteria for checking the condition of complementary recovery.

We define $\mathcal{M}_A\subset\mathbf{L}(\mathcal{H}_{\mathrm{code}})$ to be the von Neumann algebra consisting of all logical operators that can be reconstructed on $A$, and define $\mathcal{M}_{\overline{A}}\subset\mathbf{L}(\mathcal{H}_{\mathrm{code}})$ in a similar way. That is
\begin{align}\label{mamba}
\begin{split}
&\mathcal{M}_A=\{P_{\mathrm{code}}O_{A}P_{\mathrm{code}}:[O_A,P_{\mathrm{code}}]=0\},\\
&\mathcal{M}_{\overline{A}}=\{P_{\mathrm{code}}O_{\overline{A}}P_{\mathrm{code}}:[O_{\overline{A}},P_{\mathrm{code}}]=0\}.
\end{split}
\end{align} 
Note that the commutativity with $P_{\mathrm{code}}$ is necessary, without which we cannot guarantee that the two defined are von Neumann algebras.

We can prove the following criterion which states that $\mathcal{M}_A$ and $\mathcal{M}'_A$ are the only von Neumann algebras to satisfy the condition of complementary recovery. 
\begin{proposition}[Criterion 1 of complementary recovery]\label{cr1}
For a given decomposition $\mathcal{H}=\mathcal{H}_{A}\otimes\mathcal{H}_{\overline{A}}$, the code $\mathcal{H}_{\mathrm{code}}$ exhibits complementary recovery if and only if $\mathcal{M}_A$ and $\mathcal{M}'_A$ can be reconstructed on $A$ and on $\overline{A}$ respectively, or equivalently, $\mathcal{M}'_A=\mathcal{M}_{\overline{A}}$ ($\mathcal{M}'_{\overline{A}}=\mathcal{M}_A$).
\end{proposition}
Note that the last two equations are equivalent since $\mathcal{M}''_A=\mathcal{M}_A$. Though this proposition is not explicitly given in the literature about HQEC, it should be an elementary property of von Neumann algebra already implicitly included in related work. For clarity and completeness, we still give our proof for this proposition in Appendix~\ref{crcr}.

What might be really useful in studying concrete model of holographic code, especially when the structure of boundary code states is known, might be the following criterion. It is indeed what we use later to the illustrative model.
\begin{proposition}[Criterion 2 of complementary recovery]\label{cr2}
For a given decomposition $\mathcal{H}=\mathcal{H}_{A}\otimes\mathcal{H}_{\overline{A}}$, the code $\mathcal{H}_{\mathrm{code}}$ exhibits complementary recovery if and only if every $P_{\mathrm{code}}O_{\overline{A}}P_{\mathrm{code}}$, with $O_{\overline{A}}$ not necessarily commuting with $P_{\mathrm{code}}$, equals $P_{\mathrm{code}}Q_{\overline{A}}P_{\mathrm{code}}$ for some $Q_{\overline{A}}$ commuting with $P_{\mathrm{code}}$.
\end{proposition}

The proof is given in Appendix~\ref{crcr}. To apply this criterion to a code, we do not have to generally study the $P_{\mathrm{code}}O_{\overline{A}}P_{\mathrm{code}}$ operators. Instead, we consider a basis of operators spanning $\mathbf{L}(\mathcal{H}_{\overline{A}})$. If the projection of every basis operator can be replaced by some $P_{\mathrm{code}}Q_{\overline{A}}P_{\mathrm{code}}$ with $[Q_{\overline{A}},P_{\mathrm{code}}]=0$, then the projection of arbitrary operator supported on $\overline{A}$, as a linear combination of the basis operators, will possess the same property.

Note that based on the above propositions we can derive the equivalent result in Ref.~\cite{pollack2022} on the uniqueness for the recovered subalgebra. In that work, the corresponding theorem concerns more about when the set $R^+(\mathbf{L}(\mathcal{H}_A)\otimes\mathbb{C}\mathds{1}_{\overline{A}})R$ is a subalgebra.

Now, we can list the characteristics as the basis and the algebraic aspects for describing the subregion duality in the formalism of OAQEC. And we emphasize the \emph{arbitrariness} of boundary bipartition $A\overline{A}$, which is required for the comprehensiveness in the description but has not been demonstrated in the tensor-network models.  

\paragraph*{\textbf{Characteristic 2}\label{cplmrc} (complementary recovery)} For \emph{arbitrary} boundary bipartition $A\overline{A}$, including both connected and disconnected cases, we expect that Prop.~\ref{cr1}, \ref{cr2} or other equivalent conditions are satisfied.

According to Prop.~\ref{cr1}, once Characteristic 2 is satisfied, $\mathcal{M}_A$ and $\mathcal{M}_{\overline{A}}$ will be the unique players in the formal description of the subregion duality. Hence, if we particularly require the genuine OAQEC formalism, the requirement will be on $\mathcal{M}_A$ and $\mathcal{M}_{\overline{A}}$. In the following, we use $\mathds{1}$ instead of $\mathds{1}_{\mathcal{H}_{\mathrm{code}}}$, since the Hilbert space is clear in the context.

\paragraph*{\textbf{Characteristic 3}\label{chac3} (genuine OAQEC)} In this work, we particularly require that for \emph{arbitrary} boundary bipartition $A\overline{A}$, upon the condition of complementary recovery, the code has nontrivial center $\mathrm{Z}(\mathcal{M}_A)\ne\mathbb{C}\mathds{1}$, i.e., there exists at least one nontrivial logical operator (not a constant operator) in the center. If the center is trivial, the description will be simply in the subsystem-code formalism.

Note that if a holographic code satisfies Characteristic 3, then it will have intrinsic difference from the conventional sense of quantum error correction. For example, while the three-qutrit code $[[3,1,2]]_3$~\cite{cleve1999,almheiri2015,harlow2017} is the simplest illustration of the idea of HQEC in the conventional sense, there \emph{cannot} be a three-qudit code, in which OAQEC is satisfied for all bipartition and the single-qudit erasure error can also be detected. That is because for any bipartition, i.e., separating one qudit from the other two, the encoded information can be only partially protected in the form of subalgebra.

We can understand the above fact as showing the restriction imposed by Characteristic 3 on the correctability of holographic code. It reveals a trait of the quantum-error-correction interpretation of the AdS/CFT correspondence: the better the logical information is protected, the more ``trivial'' the bulk quantum gravity is described. 

\subsubsection{Subregion duality: entanglement wedge and RT/FLM formula}\label{sdew}
According to the above characteristics, within the formalism of genuine OAQEC (Characteristic 3), the logical operators that can be reconstructed on the boundary subregions are described by $\mathcal{M}_A$ and $\mathcal{M}_{\overline{A}}$. In other words, the players in the subregion duality are von Neumann algebras instead of bulk subregions, which implies important difference in representing the subregion duality in contrast to the conventional cases. We start the following discussion with elucidating this difference by contrasting three possibilities of subregion duality. The first and the second possibilities are typical in the literature, and the third one is general for the genuine OAQEC and encompasses the cases in our model.

\paragraph*{\textbf{Residual bulk qudit}} The first possible case of subregion duality is generic in the situation where the condition of complementary recovery fails. In this case, there exists at least one ``residual'' bulk qudit $\boldsymbol{x}$ such that it can be recovered neither on $A$ nor on $\overline{A}$, and additionally, $\mathcal{M}(\boldsymbol{x})$ commutes with both $\mathcal{M}_A$ and $\mathcal{M}_{\overline{A}}$. Conventionally, we use $\mathrm{W}[A]$ ($\mathrm{W}[\overline{A}]$) to denote the collection of all bulk qudits that can be reconstructed on $A$ ($\overline{A}$), i.e.,
\begin{align}\label{ew1}
\begin{split}
\mathrm{W}[A]=\{\boldsymbol{x}:\mathcal{M}(\boldsymbol{x})\subset\mathcal{M}_A\},\\
\mathrm{W}[\overline{A}]=\{\boldsymbol{x}:\mathcal{M}(\boldsymbol{x})\subset\mathcal{M}_{\overline{A}}\}.
\end{split}
\end{align}
Then, if we regard $\mathrm{W}[A]$ and $\mathrm{W}[\overline{A}]$ as the entanglement wedges, as shown in Fig.~\ref{1a}, there will be a ``residual'' bulk subregion consisting of the unrecoverable bulk qudits (the white subregion in Fig.~\ref{1a}), which insulates the two entanglement wedges.



\paragraph*{\textbf{Subsystem-code formalism}} In the second case, we have a geometric version of complementarity~\cite{pastawski2017}: For any bulk qudit at $\boldsymbol{x}$, we have either $\mathcal{M}(\boldsymbol{x})\subset\mathcal{M}_A$ or $\mathcal{M}(\boldsymbol{x})\subset\mathcal{M}_{\overline{A}}$. Accordingly, it is easy to prove the condition of complementary recovery (see the criterion in Prop.~\ref{cr1}). And it can be also proved that the equivalent von Neumann algebras $R^+\mathcal{M}_A R$ and $R^+\mathcal{M}_{\overline{A}} R$ on $\mathcal{E}$ (Sec.~\ref{vnaob}) can be represented in terms of the local bulk operators as (see App.~\ref{possc} for proof)
\begin{align}\label{ssc}
\begin{split}
&R^+\mathcal{M}_AR=\otimes_{\boldsymbol{x}\in\mathrm{W}[A]}\mathbf{L}(\mathfrak{e}_{\boldsymbol{x}})\\
&R^+\mathcal{M}_{\overline{A}}R=\otimes_{\boldsymbol{x}\in\mathrm{W}[\overline{A}]}\mathbf{L}(\mathfrak{e}_{\boldsymbol{x}}).
\end{split}
\end{align}
Here we keep adopting the definition of $\mathrm{W}[A]$ and $\mathrm{W}[\overline{A}]$ as in Eq.~\ref{ew1}. Based on these equalities, it is easy to see that $\mathrm{Z}(\mathcal{M}_A)=\mathcal{M}_A\cap\mathcal{M}_{\overline{A}}=\mathbb{C}\mathds{1}$ and $\mathrm{Z}(R^+\mathcal{M}_A R)=R^+\mathrm{Z}(\mathcal{M}_A)R=\mathbb{C}\mathds{1}$, i.e., the case is in the subsystem-code formalism.

In this case, Eq.~\ref{ssc} shows a ``perfect match'' of operators algebras and the bulk degrees of freedom, i.e., $R^+\mathcal{M}_A R$ ($R^+\mathcal{M}_{\overline{A}} R$) exactly consists of all bulk operators supported on the bulk subregion $\mathrm{W}[A]$ ($\mathrm{W}[\overline{A}]$). It follows that the subregion duality can be completely described in terms of the geometric representation of $\mathrm{W}[A]$ and $\mathrm{W}[\overline{A}]$. As illustrated in Fig.~\ref{1b}, the expected geometric representation of subregion duality is that the bulk minimal surface assigned to the boundary bipartition $A\overline{A}$ simply separates the the two wedges~\cite{almheiri2015,pastawski2015,harlow2017,pastawski2017}. According to this ``perfect match'', describing subregion duality usually only focuses on the bulk qudits, and hence the greedy algorithm in the tensor network paradigm can work efficiently in specifying the greedy wedge~\cite{pastawski2015,jahn2021}.

\paragraph*{\textbf{Genuine OAQEC: failure of geometric complementarity}} The third possible case is general when the genuine OAQEC formalism (nontrivial center $\mathrm{Z}(\mathcal{M}_A)$) together with complementary recovery (Characteristic 2 and 3) are satisfied. In this case, the nontrivial center guarantees the existence of at least one bulk qudit $\boldsymbol{x}$ such that $\mathcal{M}(\boldsymbol{x})$ belongs to neither $\mathcal{M}_A$ nor $\mathcal{M}_{\overline{A}}$ (failure of the geometric complementarity), otherwise, through the arguments in App.~\ref{possc}, we will get back to the subsystem-code formalism. However, because of the condition of complementary recover, the full local algebra $\mathcal{M}(\boldsymbol{x})$ of such a bulk qudit $\boldsymbol{x}$ does not commute with $\mathcal{M}_A$ or $\mathcal{M}_{\overline{A}}$, and hence this case is distinct from the first possibility above with the residual region.

\begin{center}
\begin{figure}[ht]
\phantomsubfloat{\label{1a}}\phantomsubfloat{\label{1b}}
\phantomsubfloat{\label{1c}}\phantomsubfloat{\label{1d}}\phantomsubfloat{\label{1e}}\phantomsubfloat{\label{1f}}
\centering
    \includegraphics[width=8.5cm]{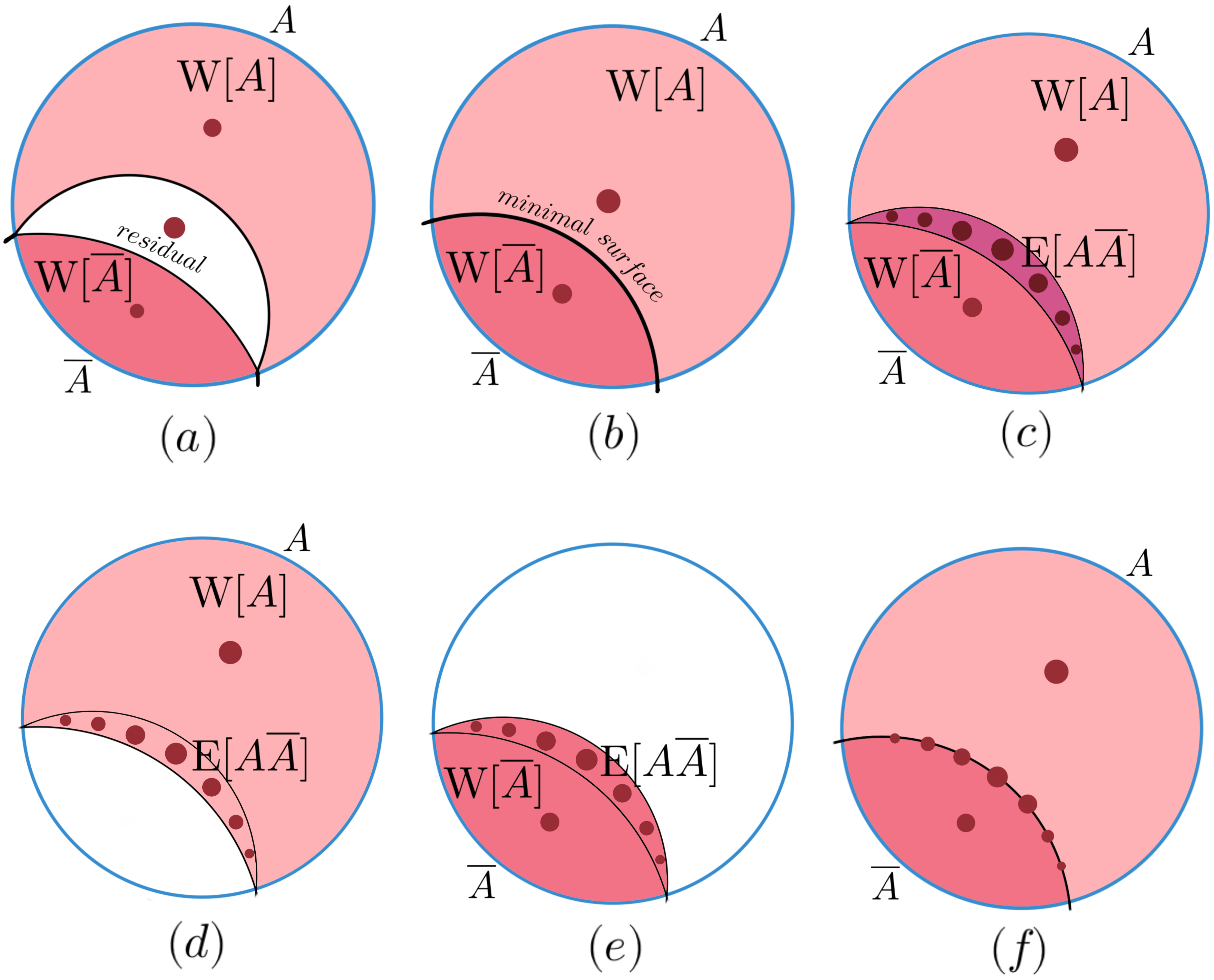}   
\caption{General illustrations of subregion duality for a connected boundary bipartition $A\overline{A}$. The blue boundary stands for the boundary/physical qudits. For conciseness, we draw single red dots to represent all the bulk/logical qudits lying within bulk subregions $\mathrm{W}[A]$ and $\mathrm{W}[\overline{A}]$. (a) Complementary recovery fails and a residual bulk subregion (the white) insulates recoverable subregions (the pink and the red). (b) Geometric complementarity, i.e., complementary recovery with subsystem code correctability, where the entanglement wedges are separated by the minimal surface. (c) Complementary recovery with genuine OAQEC. The entangling surface $\mathrm{E}[A\overline{A}]$ (the purple) that consists of not fully recoverable bulk qudits separates the two bulk subregions $\mathrm{W}[A]$ and $\mathrm{W}[\overline{A}]$ of fully recoverable bulk qudits. (d) The entanglement wedge associated to boundary subregion $A$. (e) The entanglement wedge associated to boundary subregion $\overline{A}$. (f) Alternative representation of subregion duality for genuine OAQEC: the not-fully-recoverable bulk qudits are threaded through a extremal surface.}
\label{fig1}
\end{figure}
\end{center}

\paragraph*{\textbf{Discrepancy between subalgebra and subregion}} These facts imply that there is no simple match of the von Neumann algebras on the code subspace and the subregions of bulk qudits. Instead, some bulk qudits can be included in both the support of $R^+\mathcal{M}_AR$ and that of $R^+\mathcal{M}_{\overline{A}}R$; additionally, a part of operators on these bulk qudits is included in $R^+\mathcal{M}_AR$ while some other part is included in $R^+\mathcal{M}_{\overline{A}}R$. Then, in the genuine OAQEC formalism, if we view the entanglement wedges as bulk subregions, the geometric representation of the entanglement wedges has limitations in describing the subregion duality. It was hence suggested that in this case the entanglement wedges might only possess the meaning of $R^+\mathcal{M}_AR$ and $R^+\mathcal{M}_{\overline{A}}R$ themselves~\cite{harlow2017,cree2021}. However, to elucidate the geometric meaning of the subregion duality, the structures of $R^+\mathcal{M}_AR$ and $R^+\mathcal{M}_{\overline{A}}R$ need to be represented in terms of the bulk degrees of freedom, in a similar way as in Eq.~\ref{ssc}. This is because only through the bulk qudits that represent the bulk locality can we concretize the meaning of the bulk geometry in a model.

\paragraph*{\textbf{Needs for a general framework}} In analogy to the subsystem-code case in which through the geometric representation in Fig.~\ref{1b} we can access the complete information of the subregion duality through Eq.~\ref{ssc}, we need a similar framework to rigorously ``translate'' between the algebraic and geometric aspects of subregion duality in the genuine OAQEC. Note that in the study of an exact model, we not only need to demonstrate the expected subregion duality, but also need to investigate properties of holography based on the subregion duality. Hence, analogous to the role of greedy algorithm in the tensor-network paradigm, the framework is expected to address the following questions: For a given boundary bipartition $A\overline{A}$, how to access the complete description of subregion duality, including the structures of $\mathcal{M}_A$, $\mathcal{M}_{\overline{A}}$ and $\mathrm{Z}(\mathcal{M}_A)$, and how they manifest in the bulk geometry as entanglement wedges?

\paragraph*{\textbf{Entangling surface}} Such expected framework might be specific to different models. Hence, to specify the necessary characteristics regarding the framework, we firstly consider general and rigorous but ``rough'' connections between logical subalgebras and bulk subregions in genuine OAQEC. We start with dividing all the bulk qudits into three disjoint sub-collections $\mathrm{W}[A]$, $\mathrm{W}[\overline{A}]$ and $\mathrm{E}[A\overline{A}]$ according to how the local operator algebra of each bulk qudit can be reconstructed. The three sub-collections are illustrated in Fig.~\ref{1c}. Here, we keep using the the definition of $\mathrm{W}[A]$ and $\mathrm{W}[\overline{A}]$ as in the above two cases (see Eq.~\ref{ew1}), i.e., the collection of bulk qudits whose full local algebra can be reconstructed on $A$ or on $\overline{A}$. And we define $\mathrm{E}[A\overline{A}]$ as the entangling surface in consistent with the terminology in Ref.~\cite{donnelly2017,cree2021}:
\begin{definition}[Entangling surface]\label{esdef}
Upon the condition in Characteristic 2 and 3, we define $\mathrm{E}[A\overline{A}]$ as the collection of those bulk qudits whose full local algebra is recoverable neither on $A$ nor on $\overline{A}$,
\begin{equation}
\mathrm{E}[A\overline{A}]=\{\boldsymbol{x}:\mathcal{M}(\boldsymbol{x})\not\subset\mathcal{M}_A~and~\mathcal{M}(\boldsymbol{x})\not\subset\mathcal{M}_{\overline{A}}\}.
\end{equation}
\end{definition}

\paragraph*{\textbf{A rough description of the subalgebras}} Now, the bulk Hilbert space can be correspondingly decomposed as $\mathcal{E}=(\otimes_{\boldsymbol{x}\in\mathrm{W}[A]}\mathfrak{e}_{\boldsymbol{x}})\otimes(\otimes_{\boldsymbol{x}\in\mathrm{E}[A\overline{A}]}\mathfrak{e}_{\boldsymbol{x}})\otimes(\otimes_{\boldsymbol{x}\in\mathrm{W}[\overline{A}]}\mathfrak{e}_{\boldsymbol{x}})$. And it is easy to show that the supports of the von Neumann algebras are~\footnote{By definition and according to similar arguments in App.~\ref{possc}, $\otimes_{\boldsymbol{x}\in\mathrm{W}[A]}\mathbf{L}(\mathfrak{e}_{\boldsymbol{x}})\subset R^+\mathcal{M}_A R$, hence $\mathrm{W}[A]$ must lie within the support of $R^+\mathcal{M}_A R$. It is also easy to show that $\mathrm{E}[A\overline{A}]$ lies within the support. Indeed, suppose that there exists a bulk qudit $\boldsymbol{x}\in\mathrm{E}[A\overline{A}]$ which is excluded from the support, then all operators in $R^+\mathcal{M}_A R$ must act trivially on this qudit. Or equivalently, $\mathcal{M}(\boldsymbol{x})\subset\mathcal{M}'_A=\mathcal{M}_{\overline{A}}$, i.e. $\boldsymbol{x}\in\mathrm{W}[\overline{A}]$, which contradicts the fact that $\boldsymbol{x}\in\mathrm{E}[A\overline{A}]$. On the other hand, it is easy to see that $\mathrm{W}[\overline{A}]$ is excluded from the support of $R^+\mathcal{M}_A R$, because it lies within the support of $R^+\mathcal{M}_{\overline{A}} R$ which commute with $R^+\mathcal{M}_A R$. Hence, we can conclude that the support of $R^+\mathcal{M}_A R$ is $\mathrm{W}[A]\cup\mathrm{E}[A\overline{A}]$. Similar arguments apply to the support of $R^+\mathcal{M}_{\overline{A}} R$.}
\begin{align}\label{suppen}
\begin{split}
&\mathrm{supp}(R^+\mathcal{M}_A R)=\mathrm{W}[A]\cup\mathrm{E}[A\overline{A}],\\
&\mathrm{supp}(R^+\mathcal{M}_{\overline{A}}R)=\mathrm{W}[\overline{A}]\cup\mathrm{E}[A\overline{A}]. 
\end{split}
\end{align}
Furthermore, Eq.~\ref{ssc} is replaced by
\begin{align}\label{oaqecsd}
\begin{split}
\otimes_{\boldsymbol{x}\in\mathrm{W}[A]}\mathbf{L}(\mathfrak{e}_{\boldsymbol{x}})&\subset R^+\mathcal{M}_AR,\\
\otimes_{\boldsymbol{x}\in\mathrm{W}[\overline{A}]}\mathbf{L}(\mathfrak{e}_{\boldsymbol{x}})&\subset R^+\mathcal{M}_{\overline{A}}R,\\
\mathrm{Z}(R^+\mathcal{M}_AR)=R^+\mathrm{Z}(\mathcal{M}_A)R&\subset\otimes_{\boldsymbol{x}\in\mathrm{E}[A\overline{A}]}\mathbf{L}(\mathfrak{e}_{\boldsymbol{x}}).
\end{split}
\end{align}
Here, the first two relations directly follow the above discussion and the definition of $\mathrm{W}[A]$ and $\mathrm{W}[\overline{A}]$. And the last relation says that the support of the nontrivial center lies within the entangling surface. To show the last relation, we notice the fact that for each operator $\widetilde{O}\in\mathrm{Z}(\mathcal{M}_A)$, since it commutes with both $\mathcal{M}_A$ and $\mathcal{M}_{\overline{A}}$, $R^+\widetilde{O}R$ must commute with both $\otimes_{\boldsymbol{x}\in\mathrm{W}[A]}\mathbf{L}(\mathfrak{e}_{\boldsymbol{x}})$ and $\otimes_{\boldsymbol{x}\in\mathrm{W}[\overline{A}]}\mathbf{L}(\mathfrak{e}_{\boldsymbol{x}})$. Hence, $R^+\widetilde{O}R$ must lie within $\otimes_{\boldsymbol{x}\in\mathrm{E}[A\overline{A}]}\mathbf{L}(\mathfrak{e}_{\boldsymbol{x}})$.

\paragraph*{\textbf{Entanglement wedge and geometric representation}}
If we insist in viewing the entanglement wedges as subregions of bulk qudits, then according to Eq.~\ref{suppen}, we can define the entanglement wedges as $\mathrm{W}[A]\cup\mathrm{E}[A\overline{A}]$ (see Fig.~\ref{1d}) and $\mathrm{W}[\overline{A}]\cup\mathrm{E}[A\overline{A}]$ (see Fig.~\ref{1e}), which partially capture the subregion duality but are sufficient to show the basic geometric illustration of the subregion duality. According to the above definition, the entangling surface should be understood as the overlap of the two entanglement wedges (see Fig.~\ref{1c}). In the geometric representation, since the genuine OAQEC only capture certain sub-leading quantum correction to the semi-classical bulk gravity, the entangling surface is expected to replace the minimal surface with a slight change, as shown in Fig.\ref{1c}.

It is noticeable that in Fig.\ref{1c} the entangling surface appears as a bulk subregion with dimension the same as the entanglement wedges. Indeed, this is simply due to the way we label qudits in the finite discretization of the bulk, and in the large-system-size limit it should appear as of the same dimension as the minimal surface. It is also worthwhile to emphasize again that the entangling surface characterizes the complementary recovery for genuine OAQEC, and should be distinguished from the residual bulk subregion in Fig.~\ref{1a} which insulates two noncomplementary wedges.

The readers might also notice the discrepancy between the entanglement wedges and the notations $\mathrm{W}[A]$ and $\mathrm{W}[\overline{A}]$ in our description, which is completely due to our convention for convenience in description, i.e., we define $\mathrm{W}[A]$ and $\mathrm{W}[\overline{A}]$ as concrete as in Eq.~\ref{ew1}. This discrepancy can be removed by defining $\mathrm{W}[A]$ and $\mathrm{W}[\overline{A}]$ as the supports of the von Neumann algebras, then the illustration of the subregion duality might look like Fig.~\ref{1f}, where an extremal surface threads those bulk qudits not fully recoverable on $A$ or $\overline{A}$. This convention is adopted by Ref.~\cite{cao2021} in which the genuine OAQEC and the condition of complementary recovery are shown for certain connected boundary bipartions. However, in that way, it seems not convenient to explicitly specify the structure of the von Neumann algebras, and the pictorial illustration of subregion duality in OAQEC might appear inconsistent with those in the subsystem-code formalism, which will be clear later. Hence, bearing in mind the equivalence, we choose the former with defining the entangling surface.

\paragraph*{\textbf{Splits in the entangling surface}} Though the geometric representation illustrated in Fig.~\ref{1c} capture the supports of $R^+\mathcal{M}_A R$, $R^+\mathcal{M}_{\overline{A}}R$, there remains a gap to access the complete information of the subregion duality, i.e. how bulk qudits in the supports contribute to structures of $R^+\mathcal{M}_A R$ and $R^+\mathcal{M}_{\overline{A}}R$. Indeed, Eq.~\ref{suppen} together with Eq.~\ref{oaqecsd} imply that for each bulk qudit or each of disjoint sub-collections of bulk qudits in $\mathrm{E}[A\overline{A}]$, part of its operators belongs to $R^+\mathcal{M}_AR$ while another part belongs to $R^+\mathcal{M}_{\overline{A}}R$, and there can be nontrivial overlap between the two parts which contributes to the center $\mathrm{Z}(R^+\mathcal{M}_AR)$. In other words, the corresponding (local) operator algebra for bulk qudits in the entangling surface ``splits''. Hence, $R^+\mathcal{M}_A R$ ($R^+\mathcal{M}_{\overline{A}}R$) should be essentially a von Neumann algebra tensor product of factors going through certain (local) split parts in the entangling surface and also through the full local bulk operator algebras for qudits in $\mathrm{W}[A]$ ($\mathrm{W}[\overline{A}]$), and $\mathrm{Z}(R^+\mathcal{M}_A R)$ should be simply a von Neumann algebra tensor product of certain shared split parts in the entangling surface.



\paragraph*{\textbf{Characteristic 4.1} (Complete description of subregion duality)}\label{chac41} According to the above analysis, the construction of a model need to present a framework such that for a given boundary bipartition $A\overline{A}$, the structures of the von Neumann algebras $R^+\mathcal{M}_A R$, $R^+\mathcal{M}_{\overline{A}}R$ and $\mathrm{Z}(R^+\mathcal{M}_A R)$ can be explicitly presented in terms of the bulk local degrees of freedom. And the following aspects are expected to be achieved through the framework: (1) Upon Characteristic 2 and 3, given a boundary bipartition $A\overline{A}$, the framework can specify the entangling surface and the entanglement wedges, as describing the supports of $R^+\mathcal{M}_A R$ and $R^+\mathcal{M}_{\overline{A}}R$; (2) the splits of the bulk qudits in the entangling surface, including their locality together with their algebraic characterizations, can be formally described; (3) the structure of the von Neumann algebras in terms of the (local) splits, i.e., how each split contributes to the structure with respect to its bulk locality, can be explicitly described. Then, for a connected boundary bipartition, the geometric representation of the subregion duality in terms of entanglement wedges, as can be specified through the framework, is expected to be in a way consistent with Fig.~\ref{1c}.



In addition to the above basic characteristic of subregion duality, it is also expected to represent the basic reconstruction properties of a single bulk qudit in terms of the entanglement wedges. In such illustrations, unlike the case in Fig.~\ref{1c}, there can be multiple boundary subregions. And due to the particular geometric representation of the entanglement wedge in genuine OAQEC, the illustrations cannot be exactly the same as the conventional ones in the literature~\cite{almheiri2015,harlow2017,pastawski2017}. Hence, we make the following convention to keep the consistency.



\paragraph*{\textbf{Convention for pictorial illustration in genuine OAQEC}} (1) When illustrating the case for a connected boundary bipartition (only one bipartition), we shade all the three parts for $\mathrm{W}[A]$, $\mathrm{W}[\overline{A}]$ and $\mathrm{E}[A\overline{A}]$, as in Fig.~\ref{1c}. (2) When pictorially representing the case for multiple recoveries, i.e., with multiple connected boundary subregion $A_1,A_2,\cdots$ (see Fig.~\ref{2a} and \ref{2b}), we only shade $\mathrm{W}[A_1],\mathrm{W}[A_2],\cdots$ and add annotation for the entangling surfaces when necessary. (3) When representing disconnected bipartition for uberholography as hollowed out from a connected reconstruction $A$ by removing subregions $A_1,A_2,\cdots$, we simply shade the remaining while leaving $\mathrm{W}[A_1],\mathrm{W}[A_2],\cdots$ blank as shown in Fig.~\ref{2c} and \ref{2d} so that the shaded represent the entanglement wedge associated to disconnected boundary subregion.

\paragraph*{\textbf{Characteristic 4.2}}\label{chac42} As derived characteristics of the basic properties of subregion duality, the geometric representations of the entanglement wedges are expected to demonstrate the multiple reconstruction of local bulk operators on connected boundary subregions. For example, a bulk local operator $\widetilde{O}_{\boldsymbol{x}}$ can be reconstructed on two subregions $A_1$ and $A_2$ but not on $A_1\cap A_2$ (see Fig.~\ref{2a}). This appears contradictary to the field-theory picture. As another example, for certain tripartition $ABC$ of the boundary (see Fig.~\ref{2b}), some bulk local operator $\widetilde{O}_{\boldsymbol{x}}$ on the central bulk qudit $\boldsymbol{x}$ can be reconstructed on $A\cup B$, $A\cup C$ and $B\cup C$, but on none of $A$, $B$ or $C$.

\paragraph*{\textbf{Entanglement entropy}} A direct but simple characterization of how the bulk geometry is emergent from the boundary entanglement is through the entanglement entropy $\mathrm{S}(\widetilde{\rho}_A)$ with $\widetilde{\rho}_A=\mathrm{Tr}_{\overline{A}}[\widetilde{\rho}]$ for an arbitrary state $\widetilde{\rho}$ on $\mathcal{H}_{\mathrm{code}}$. As shown in Ref.~\cite{harlow2017}, Characteristic 2 and 3 ensure the following form of the entanglement entropy
\begin{align}
\begin{split}
\mathrm{S}(\widetilde{\rho}_A)&=\mathrm{Tr}[\widetilde{\rho}\mathcal{L}_A]+\mathrm{S}(\widetilde{\rho},\mathcal{M}_A)\\
\mathrm{S}(\widetilde{\rho}_{\overline{A}})&=\mathrm{Tr}[\widetilde{\rho}\mathcal{L}_A]+\mathrm{S}(\widetilde{\rho},\mathcal{M}_{\overline{A}}),
\end{split}
\end{align}
which, in the above description of the subregion duality, can be interpreted as a version of the Ryu-Takayanagi/Faulkner-Lewkowycz-Maldacena (RT/FLM) formula that includes the original Ryu-Takayanagi formula (the leading term in the approximation of the bulk gravity)~\cite{ryu2006} together with certain sub-leading quantum corrections to the order of $G^0$. Here, the area operator $\mathcal{L}_A$ can be viewed as an operator in the nontrivial center $\mathrm{Z}(\mathcal{M}_A)$, which is interpreted as certain bulk operator integrated along the minimal surface. The second term refers to the entanglement in the bulk, and $\mathrm{S}(\widetilde{\rho},\mathcal{M}_A)$ ($\mathrm{S}(\widetilde{\rho},\mathcal{M}_{\overline{A}})$) is the algebraic entropy of $\widetilde{\rho}$ which is uniquely defined with respect to the von Neumann algebra $\mathcal{M}_A$ ($\mathcal{M}_{\overline{A}}$)~\cite{harlow2017}.

Note that due to the nontrivial center in the genuine OAQEC formalism, this bulk term shows intrinsic difference from that in the special case of subsystem-code formalism. For example, in the subsystem-code formalism, for any pure code state as a product state across the two entanglement wedges, e.g., bulk-qudit-product state or computational basis state, its bulk term is zero. But it is in general not the case in the genuine OAQEC formalism. And this difference concretizes the correction from the bulk entanglement~\cite{faulkner2013,harlow2017}.

\paragraph*{\textbf{Characteristic 5}}\label{chac5} In an exact model of holographic code, for a connected boundary bipartition $A\overline{A}$, both the bulk term of an arbitrary $\widetilde{\rho}_{\mathrm{code}}$ and the area operator $\mathcal{L}_A$ should be rigorously derivable. And the area term in the entropy, i.e. $\mathrm{Tr}[\widetilde{\rho}\mathcal{L}_A]$ is expected to be consistent with certain meaning similar to the minimal surface. Indeed, since the meaning of minimal surface in a model based on a given bulk discretization or tessellation can only be concretized through the discrete bulk qudits or tiles, this requirement can only have the following sense: The area term should be consistent with the number of tiles that certain geodesic crosses or passes by.

Demonstrations of the above characteristics regarding the subregion duality are restricted in the established models, since the greedy-wedge method does not directly show the details of how the local bulk operator algebra in the entangling surface splits and is hence insufficient to specify the explicit structures of the von Neumann algebras. In recent advancement, this difficulty could be addressed by investigating the structure of the tensors~\cite{cao2021}, though it is limited to certain connected boundary bipartitions. As we will show, in our exact model, the above characteristics can be rigorously demonstrated for arbitrary connected boundary bipartitions, and also for all disconnected cases relevant to uberholography.


\subsubsection{Uberholography}\label{uberintro}
Uberholography characterizes the bulk reconstruction on disconnected boundary subregions, which is beyond the description of the AdS/Rindler reconstruction~\cite{almheiri2015,pastawski2017}. As illustrated in Fig.~\ref{2c} and \ref{2d}, it says that the reconstruction of $\mathcal{M}(\boldsymbol{x})$ for a local bulk qudit $\boldsymbol{x}$ on connected boundary subregion is too ``consuming''; and sequential removals of sub-connected parts from the connected boundary subregion does not affect the recovery of $\mathcal{M}(\boldsymbol{x})$. Furthermore, if one keeps the process until the recovery is no longer protected against further removal, the eventual result is a disconnected ``economical'' subregion. Importantly, all such disconnected boundary subregion for recovery scales sublinearly on the total number of physical qudits $N$ (hence measure zero), and the scaling is tightly captured by a universal component $1/h$, i.e. $\le N^{1/h}$. Note that we express the component in the reciprocal form to emphasize the importance of $h$ as will be discussed later.

Originally, the properties of uberholography, including the universal $1/h$, are derived based on assuming both the algebraic and geometric complementarity (see Fig.~\ref{1b}) and rely on the continuum-limit treatment of the minimal surface~\cite{pastawski2017}. However, in a finite-size holographic code of OAQEC, the entanglement wedges is not simply separated by the minimal surface (see previous discussion). Hence the original way, i.e., utilizing the minimal surface to hollow out a connected boundary subregion for reconstruction, cannot be replicated, and we should not anticipate to obtain the exactly same results or the same value of the universal component $1/h$. In other words, certain difference in the geometric illustration of uberholography might be inevitable.

\begin{center}
\begin{figure}[ht]
\phantomsubfloat{\label{2a}}\phantomsubfloat{\label{2b}}
\phantomsubfloat{\label{2c}}\phantomsubfloat{\label{2d}}
\centering
    \includegraphics[width=8.5cm]{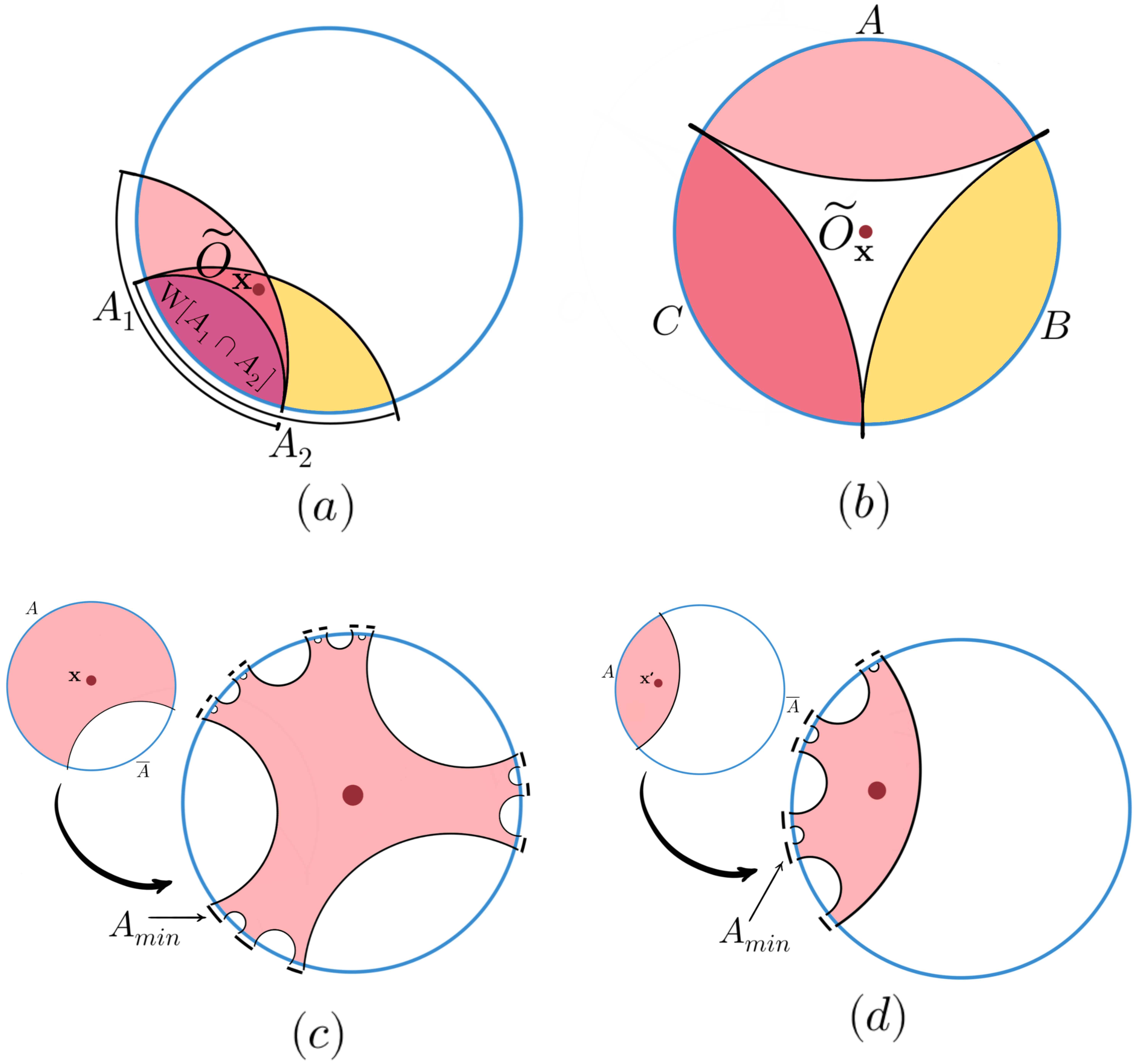} 
\caption{Properties of reconstruction of single-bulk-qudit operators represented in terms of entanglement wedges. (a) A bulk qudit operator $\widetilde{O}_{\boldsymbol{x}}$ can be reconstructed on both boundary subregions $A_1$ and $A_2$, but not on $A_1\cap A_2$. (b) A bulk qudit operator $\widetilde{O}_{\boldsymbol{x}}$ on the central bulk qudit $\boldsymbol{x}$ can be reconstructed on $A\cup B$, $A\cup C$ and $B\cup C$, but on none of $A$, $B$ or $C$. (c) (d) Illustrations of a minimal reconstruction $A_{\mathrm{min}}$ for bulk qudit $\boldsymbol{x}$ as hollowed out by removing subregions $A_1,A_2,\cdots$ from a connected reconstruction $A$. Starting from the shaded bulk subregion for $\mathrm{W}[A]$, we blank the bulk subregions for $\mathrm{W}[A_1],\mathrm{W}[A_2],\cdots$. The bulk qudit in (c) is closer to the center of bulk than the bulk qudit in (d).}
\label{fig2} 
\end{figure}
\end{center}

In the established models of holographic code, the signature of uberholography can be probed~\cite{cao2021}, but a systematic demonstration, including the universality of the scaling component, is hard to be reached. The main reason for this difficulty is that the boundary geometry in the tensor-network model is not uniquely specified, and the greedy-wedge method does not generally work for the disconnected boundary bipartition~\cite{jahn2021}.

While the value of $1/h$ cannot be exactly the same as originally predicted in the continuum case, the universality of the scaling behavior should be important in a holographic code. To formalize the characteristics of uberholography in a concrete model, we introduce the definition of minimal boundary subregion for reconstruction to generally capture these discrete and measure-zero boundary recoveries.

\begin{definition}[Minimal recovery]\label{mini}
A boundary subregion $A$ that reconstructs $\mathcal{M}(\boldsymbol{x})$ is called minimal if further removal of any of its subset of physical qudits will disable the recovery.
\end{definition}

Note that the minimal subregion is not unique and not necessarily the smallest subregion. For example as in the original arguments~\cite{pastawski2017}, starting from different connected reconstruction of $\mathcal{M}(\boldsymbol{x})$ on the boundary, one can reach different disconnected boundary subregion, but there is no guarantee that such subregions all have the same size. Based on the meaning of minimal recoveries, we can view following properties as the characteristic of uberholography in a holographic code.

\paragraph*{\textbf{Characteristic 6} (Uberholography)}\label{chac6} There exists $0<1/h<1$, such that any minimal boundary subregion reconstructing $\mathcal{M}(\boldsymbol{x})$ for a bulk qudit $\boldsymbol{x}$ is disconnected and scales as $\sim N^{1/h}$. A necessary condition is that the \emph{price} $\mathrm{p}(\boldsymbol{x})$ (and hence the distance which is always smaller) scales sublinearly on $N$.



\section{Basic idea for the new approach}\label{idea}
To describe the basic idea for a new possible way to construct holographic code, we start with a fundamental perspective on the relationship between the two ingredients of a model, i.e., the discretization of the bulk and the boundary geometries, and the encoding isometry $\mathfrak{e}_{\boldsymbol{1}}\otimes\mathfrak{e}_{\boldsymbol{2}}\otimes\mathfrak{e}_{\boldsymbol{3}}\otimes\cdots=\mathcal{E}\xrightarrow{R}\mathcal{H}=\mathfrak{h}_1\otimes\mathfrak{h}_2\otimes\mathfrak{h}_3\otimes\cdots$. We try to emphasize that in the general scenario of conceiving an exact model it is crucial to ensure the compatibility of the two ingredients, prior to investigating the holographic and quantum-error-correction properties of the encoding.

The key is to note that the two ingredients are individual as different types of structure. That is, while the geometric discretization can assign the qudits specific locations in the bulk and boundary geometries, the encoding $\mathcal{E}\xrightarrow{R}\mathcal{H}$ is itself a purely algebraic structure whose characterization is independent on the geometry. However, in realizing the hypothetical HQEC characteristics, the two ingredients are necessarily present as an organic whole. In other words, the geometry of qudits need to be compatible with the encoding isometry in order to underlie the expected geometric manifestation of the operator-algebra properties derived from the encoding.

The simplest example of this compatibility is the boundary reconstruction of a single bulk qudit (see Sec.~\ref{hqecc} or Ref.~\cite{almheiri2015}). For a bulk qudit not in the center, its boundary reconstruction has smaller size on one side of the boundary that is closer to the bulk qudit, and has greater size on the opposite side; but for the central bulk qudit, this difference in size disappears. Presence of this property not only requires the geometry on the relative location of the bulk qudits and boundary qudits to coincide with the algebraic operator reconstruction properties, but also requires their bulk radial distance from the center to coincide with their code distance.



In the conventional paradigm~\cite{jahn2021}, the tensor network might be viewed as an additional structure that simultaneously specifies the two ingredients and guarantees their compatibility: The inherent discrete network holds the geometries of the bulk qudits and the boundary qudits; the tensors, as local algebraic structures living on nodes of the network, are contracted into the integrated encoding isometry. Thus, demonstrating the algebraic aspects of the operator reconstruction properties of the encoding is essentially investigating properties of the tensor-network structure themselves. And once these properties are demonstrated, they are inherently manifested in the geometries as expected.

While the tensor-network approach provides convenience in defining models, certain HQEC characteristics, e.g., the conditions of uberholography, are proposed afterwards, and hence might not be straightforwardly demonstrated with the developed methods~\cite{pastawski2017}. Therefore, instead of leveraging the convenience of tensor-network construction of models, we try to take a fundamental perspective on how to organically define the two ingredients for a model without resorting to additional structures. In the following, we show that a possibly deeper insight into the demonstration of uberholography inherently embodies a guideline to specify the bulk geometry discretization, and also offers clues for defining a compatible encoding map.

\subsection{Rearrangement of physical qudits}
Before proceeding, it is necessary to discuss the meaning of an important geometric operation on the physical qudits, i.e., the rearrangement or reconfiguration which is simply a one-to-one correspondence between two different geometric arrangements or presentations of the same qudits $\mathfrak{h}_1,\mathfrak{h}_2,\mathfrak{h}_3,\ldots$ (see Fig.~\ref{fig3}). It will be adopted in our following arguments that the rearrangement of physical qudits leaves the algebraic properties of the encoding isometry invariant. For example, the essences of quantum-error-correction properties including Characteristics 2 and 3 (see Par.~\ref{chac3}) do not rely on any geometric arrangement of the qudits. In other words, assigning the qudits any different geometry does not affect whether the characteristics are satisfied or not. In practice, it can happen that the proof of these characteristics might be more convenient in certain special geometries different from the standard 1D geometry. In that case, rearranging the physical qudits back into the standard 1D geometry, the proved results for Characteristics 2 and 3 still hold.



\begin{center}
\begin{figure}[ht]
\centering
    \includegraphics[width=8.5cm]{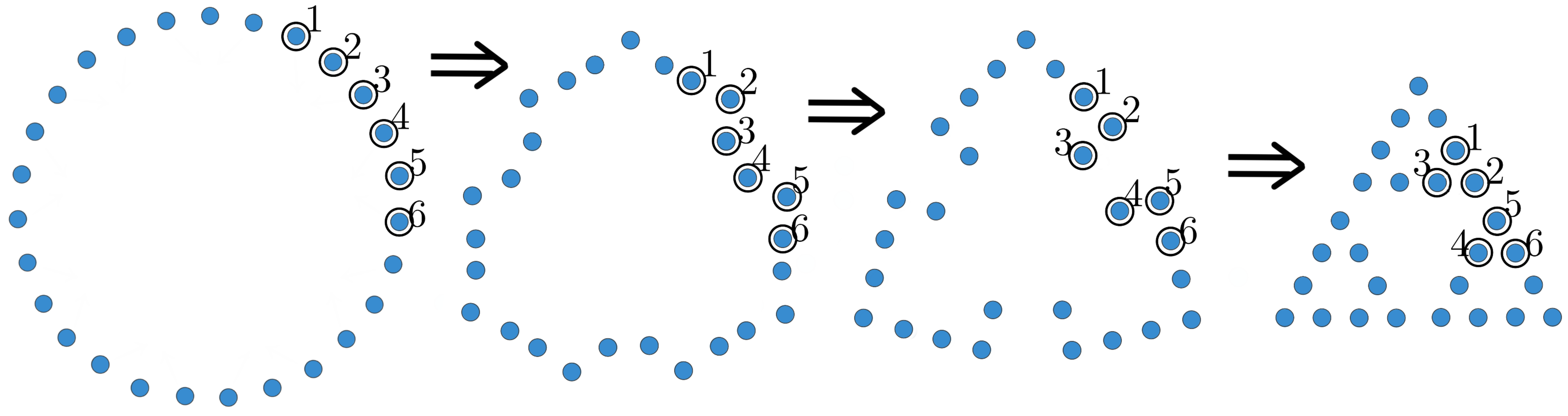}   
\caption{Rearrangement of the qudits from one lattice geometry (1D ring) to another (Sierpi\'nski triangle). The rearrangement can be described in a self-similar way: three qudtis are retracted into a triangle, while three such triangles are retracted into a larger triangular blocks. In larger system size, three of such triangular blocks are further retracted and such process continues. According to the rearrangement, the locality is changed: The nearest neighboring qudits 3 and 4 are rearranged apart, while the next-next neighboring qudits 2 and 5 are rearranged into nearest neighbors. This change will amplify drastically in larger system size.}
\label{fig3}
\end{figure}
\end{center}

What depend on the geometry is how the algebraic properties are manifest or interpreted through the locality, the connectedness and other features that are associated to a concrete (lattice) geometry. For example, Characteristics 1, 4.1, 4.2, 5 and 6 can be viewed as how the logical operator reconstruction is interpreted through the 1D boundary geometry and the hyperbolic bulk geometry. Since rearrangement can (drastically) change the locality of qudits (see Fig.~\ref{fig3}), Characteristics 1, 4.1, 4.2, 5 and 6 can have different geometric presentations in distinct geometries, but the rearrangement establishes their equivalence.

\subsection{Scaling component $\frac{1}{h}$ vs. Hausdorff dimension $h$}\label{h1h}
In building a model of holographic code, to ensure the characteristics of uberholography with a specific scaling component $1/h$, we need to gain insight into the following questions.

First, how to manifest the given value of $1/h$ in a model? Or equivalently, what feature of the construction can guarantee the scalings of minimal boundary subregions all subject to the given $1/h$? Second, as been previously discussed for Characteristic 6 (see Sec.~\ref{uberintro}), the value of $1/h$ might be different from the originally derived one~\cite{pastawski2017}. Then, what determines the possible value of $1/h$ prior to building a holographic code? 


Moreover, since the minimality of the boundary recovery implies the bare sufficiency to read the information of a local bulk qudit, the hierarchical properties of the bulk qudits associated to the bulk radial distance should manifest on the sizes of the minimal subregions. However, due to the zero measure and the disconnectedness on the 1D geometry (see Fig.~\ref{2c} and \ref{2d}), such manifestation cannot be readily apparent. And there is a lack of clarity regarding what feature of the construction is responsible for this manifestation. In other words, the presence of uberholography in the standard 1D geometry, regarding the minimality of a subregion for bulk operator reconstruction, seems deviate from a more conventional picture in experience in which the minimality is expected to appear as connected and compact in geometry.


Intriguingly, the above issues can be illuminated if we view $h$ as the Hausdorff dimension of certain alternative geometry for the physical qudits, which can be linked to the standard 1D geometry through a geometric rearrangement. Here for a specific $1/h$ ($0.5<1/h<1$), the dimension $h$ must be fractional, i.e. $1<h<2$, like that of fractal geometry embedded in 2D but not necessarily referring to fractality. Importantly, while in the 1D case the total number $N$ of the physical qudits and the linear size $N_0$ of the geometry are the same, in such an alternative geometry with fractional Hausdorff dimension $h$, the two numbers are splitted into two, i.e. $N\sim{N_0}^h$. This implies that any subregion that scales sublinearly as $\sim N^{1/h}$, e.g., the possible subregions to be demonstrated as minimal recoveries in a model, now scales linearly as $\sim N_0$ and with nonzero measure with respect to $N_0$. These relations are listed in Tab.~\ref{table1}.

\begin{table}[ht]
\begin{center}
\begin{tabular}{ |c|c|c| } 
\hline
   Hausdorff dimension  & 1 & $h$ \\ 
 \hline
 linear size & $N$ & $N_0$ \\ 
 \hline
 total number & $N$ & $N={N_0}^h$ \\ 
 \hline
 $A_{\mathrm{min}}$ & $\le N^{1/h}$, & linear on $N_0$, \\
    & disconnected, & (almost) connected, \\
    & measure zero & measure nonzero\\
\hline
$A$ (connected in 1D) & linear on $N$ & superlinear on $N_0$ \\
\hline
\end{tabular}.
\end{center} 
\caption{Comparison table of the sizes between the standard 1D geometry (left column) and the alternative geometry (right column).}
\label{table1}
\end{table}

According to this observation, the alternative geometry together with the rearrangement might provide a more natural presentation of uberholography which can ``demystify'' the meaning of the universal $1/h$. In the following, we elaborate on this point with the simplest example of alternative geometry: The value $1/h=\mathrm{log}(2)/\mathrm{log}(3)\approx0.631$ corresponding to the Hausdorff dimension $h=\mathrm{log}(3)/\mathrm{log}(2)\approx 1.585$ of the Sierpi\'nski (triangle) fractal geometry where the linear size $N_0$ can be viewed as the size of any lateral side (see Fig.~\ref{fig3} and \ref{fig4}). And there exists a rearrangement bridging the alternative geometry and the standard 1D geometry as illustrated in Fig.~\ref{fig3} and \ref{fig4}. The rearrangement can be described as in both directions: from 1D to the alternative geometry, the rule for the rearrangement is described in Fig.~\ref{fig3}, and from the alternative back to 1D, the reverse of the rule is illustrated in Fig.~\ref{4a}. Importantly, the description is scalable.

Note that our application of the fractal geometry in the study of uberholography is distinct from that in Ref.~\cite{bao2022}. In the later, the fractal geometry is a subset of the boundary qubits considered as noise. In our case, the fractal geometry is simply an alternative geometry of the whole of the physical qudits. It will be clear in the following texts that our study uncovers a new and possibly essential connection between the growing focuses on the fractal geometries from the many-body physics perspective~\cite{yoshida2013,kempkes2019,manna2020,xu2021,biesenthal2022,zhu2022} and the research of HQEC.


\subsection{``Demystifying'' properties of uberholography}\label{demys}
First, in the alternative geometry, typical subregions that scale sublinearly on $N$ can be (almost) connected~\footnote{Here, by a almost connected subregion, we mean a subregion which can be viewed as the union of a small number of connected parts.} and compact. In the illustrative example as shown in Fig.~\ref{fig4}, such subregions can be identified as those open paths surrounding the holes. It is noticeable that after rearranging back to the 1D geometry, these paths manifest in disconnectedness and zero measure (see Fig.~\ref{4a}, \ref{4b} and \ref{4c}). Especially, their geometries in 1D appears similar to the minimal subregions for recovering single bulk qudits in uberholography as predicted originally in Ref.~\cite{pastawski2017}. This observation suggests that the ``fractal-like'' geometry of the minimal boundary recoveries as embedded in 1D and with scaling component $1/h$~\cite{pastawski2017} might result from the alternative geometry with the fractional Hausdorff dimension $h$ and the rearrangement between the alternative geometry and the standard 1D geometry.

\onecolumngrid
\begin{center}
\begin{figure}[ht]
\centering
    \includegraphics[width=17cm]{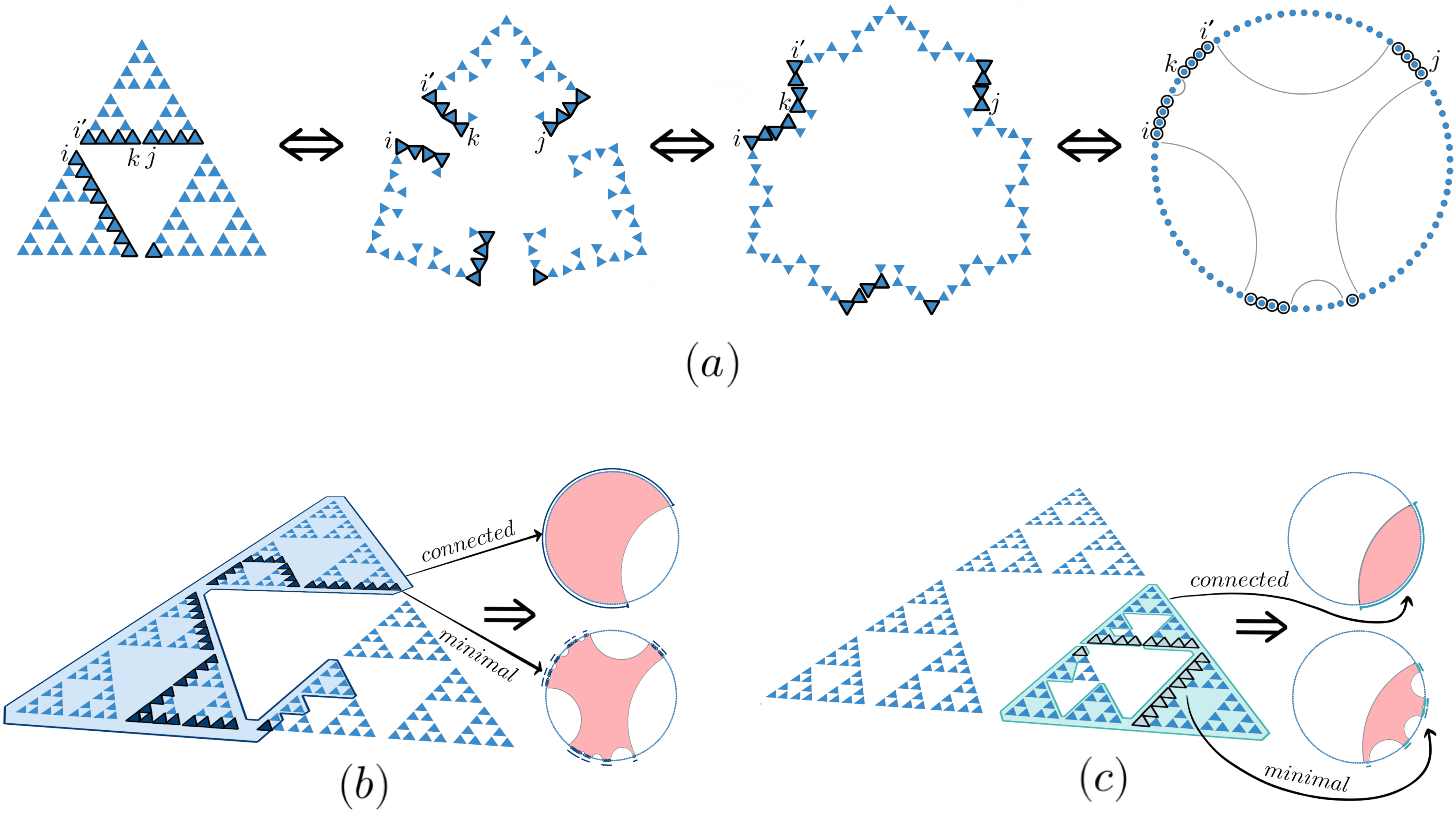}   
\phantomsubfloat{\label{4a}}\phantomsubfloat{\label{4b}}
\phantomsubfloat{\label{4c}}
\caption{Rearrangement between the alternative Sierpi\'nski geometry and the standard 1D geometry of the physical qudits. The singly highlighted represent the open paths, while the shaded (in blue or green) represent subregions in the alternative geometry, which are rearranged into a connected boundary subregion in 1D. The bulk geodesics and the shaded bulk subregions (in pink) in the standard geometry are merely used for the analogy to the predictions of uberholography in the literature, they do not acquire specific meanings within a holographic code until our later demonstrations in an exact model. (a) Illustration of the scalable rule that guides the rearrangement between the Sierpi\'nski geometry and the 1D geometry. (b) The singly highlighted open path surrounds the central hole. According to the rule of the rearrangement, the blue-shaded region is the smallest region which includes the path as a subset and can be rearranged into a connected subregion in 1D. (c) The singly highlighted open path surrounds a non-central hole. Similarly, the green-shaded region is the smallest region that includes the path and is rearranged into a connected subregion in 1D.}
\label{fig4}
\end{figure}
\end{center}
\twocolumngrid

Second, due to the fractional Hausdorff dimension, there appear ``holes'' in the alternative geometry as hollowed out from the 2D plane (see Fig.~\ref{4b} and \ref{4c}), which inherently possess certain hierarchical structure analogous to that of the bulk qudits consistent with the bulk radial direction. Then, after rearranging back to the 1D geometry, the paths surrounding holes of different size in the alternative geometry resemble the picture of different disconnected minimal subregions recovering bulk qudits of different radial distances (comparing Fig.~\ref{4b} and \ref{4c} with reference to Fig.~\ref{2c}, \ref{2d} and Ref.~\cite{pastawski2017}). This observation suggests that the geometric structure of the holes in the alternative geometry effectively captures the bulk discrete geometry for the bulk qudits in the standard picture. In this way, the paths surrounding the holes seems to naturally present how the hierarchical properties of the bulk qudits associated to the bulk radial distance manifest on the sizes of the minimal subregions, as questioned previously.

Third, the rule of the rearrangement seems to underlie the peculiar property of uberholography: the size of the smallest connected boundary recoveries of a bulk qudit overwhelms the size of any disconnected minimal boundary recoveries (see Sec.~\ref{uberintro}). Indeed, according to the discussion in Sec.~\ref{hqecc}, any connected boundary subregion $A$ supporting the reconstruction of $\mathcal{M}(\boldsymbol{x})$ must include a subregion $A_{\mathrm{min}}$ supporting a minimal recovery of $\mathcal{M}(\boldsymbol{x})$, i.e., $A_{\mathrm{min}}\subset A$. Hence, if the open paths can be demonstrated to be the minimal recoveries in certain model, according the way the open paths are rearranged back to 1D, the subregion (shaded in blue in Fig.~\ref{4b} or green in Fig.~\ref{4c}) that is rearranged into a connected boundary subregion in 1D and is the smallest that includes the open path (highlighted singly) as a subset has to be much larger and scale superlinearly on the open path.

Fourth, the above relationship (in terms of the rearrangement) between the disconnected minimal subregions and the connected subregions for recovery might unveil how the minimal recoveries results from hollowing out certain qudits from a connected recovery, which is the original arguments in the introduction of uberholography~\cite{pastawski2017}. As illustrated in Fig.~\ref{4b} and \ref{4c}, the qudits in the shaded subregion that are rearranged to the connected 1D boundary subregion but not in overlap with the paths can all be hollowed out, and the imagined entanglement wedges as if resulted form the hollowing (by virtually adding the minimal surfaces as the geodesics) resembles those predicted in Ref.~\cite{pastawski2017}.


Furthermore, we note that each open path has clear association to certain hole in the alternative geometry (comparing the paths in Fig.~\ref{4b} and \ref{4c}), while the paths can surround a hole in different ways (comparing the paths in Fig.~\ref{4a} and \ref{4b}). It will be clear in later arguments that this property of the open paths underlies the demonstration for the connected distance of a single bulk qudit. That is, by suitably choosing such paths surround one hole in the alternative geometry, after rearranged back to the standard 1D geometry, we can avoid all connected boundary subregions of certain size. This resembles the argument: operators on a single bulk qudit can be reconstructed on different minimal boundary subregions to avoid all erasure errors of certain size on the boundary.

The above analogy might has brought light to the questions raised previously on the realization of uberholography in a model of holographic code, and also provided insight into the peculiar properties of uberholography. In reality, it could be very difficult to check this perspective with established tensor-network models, since in these model the boundary is not uniquely specified and hence it is elusive how to systematically demonstrate uberholography~\cite{jahn2021}. However, the analogy presents an avenue for initiating a model construction, which will be detailed in the following subsections. And this avenue will be validated in our construction of an exact model of holographic code which satisfies the HQEC characteristics discussed in the previous section. 

\subsection{Hyperbolic tessellation}
The essence of the above scenario is more than a conjecture, since the process of rearranging the physical qudits from the standard 1D geometry towards a geometry with higher Hausdorff dimension is always viable. Indeed, if a model has been demonstrated for certain characteristics of HQEC with a specific $1/h$, or more loosely with multiple values of $1/h$ for recovering different bulk qudits, then, one can always rearrange the physical qudits towards ``reducing'' the scatteredness and ``increasing'' the compactness of the minimal subregions for reconstruction.  

Conversely, the rearrangement of physical qudits from the geometry of Hausdorff dimension $h$ to the standard 1D geometry can inherently specify a structure that effectively captures a possible hyperbolic tessellation in finite system size, and hence gives rise to the bulk geometry discretization. In the case of the Sierpi\'nski fractal as illustrated in Fig.~\ref{5a}, the effect of the rearrangement can be viewed as deforming the sides of the holes into arcs of circles lying within a disk. Note that for the hyperbolic bulk as the Poincar\'e disk, a geodesic is either a straight line as the diameter or an arc of a circle with the two ends perpendicular to the 1D asymptotic boundary. Hence, as shown in Fig.~\ref{5a}, the holes are eventually deformed into ideal polygons in the hyperbolic geometry with curvature $-1$~\footnote{Ideal polygons on the Poincar\'e disk have their endpoints on the asymptotic boundary, and have their sides the geodesics of the hyperbolic geometry. An ideal polygons can be simply specified by the geodesics, and an ideal hyperbolic tessellation can be specified in a similar way with a family of geodesics.}. Note that every ideal triangle has area equal to $\pi$ and every ideal quadrilateral has area equal to $2\pi$.

It is noticeable that as been indicated in the previous discussion, the hierarchy of the holes in the fractional-Hausdorff-dimension geometry resembles that of the expected bulk qudits with different radial distance. Hence, according to Fig.~\ref{5a}, we can view the ideal polygons as the bulk discretization for the bulk qudits, and use the index $\boldsymbol{x}$ for the holes as the index for the discrete bulk locations. Then, we can reach a complete specification of the discrete bulk and boundary geometry for the qudits as shown in Fig.~\ref{5b}.

In general, for a system of finite size, if we roughly view a tessellation of the hyperbolic bulk as specified by a family of geodesics, then the rearrangement from the alternative geometry to the standard 1D geometry can specify such a family in the following way: As illustrated in Fig.~\ref{5c}, in the rearrangement, pairs of physical qudits (linked by the blue or green dashed lines) in a fractional-Hausdorff-dimension geometry drift away from each other towards apart positions in the 1D ring, stretching out their separation into a geodesic (the blue and green dashed lines on the left, the black solid line on the right) in the hyperbolic space. In this way, geodesics crossing in the hyperbolic disk (see the right of Fig.~\ref{5c}) can specify regular polygons, while geodesics meeting at ideal points on the asymptotic boundary (see Fig.~\ref{5a}) can specify ideal polygons. Due to the wealth of the fractional-Hausdorff-dimension geometries and their associated rearrangements to 1D, there should be a wide class of tessellations to be explored.

\onecolumngrid
\begin{center}
\begin{figure}[ht]
\centering
    \includegraphics[width=17cm]{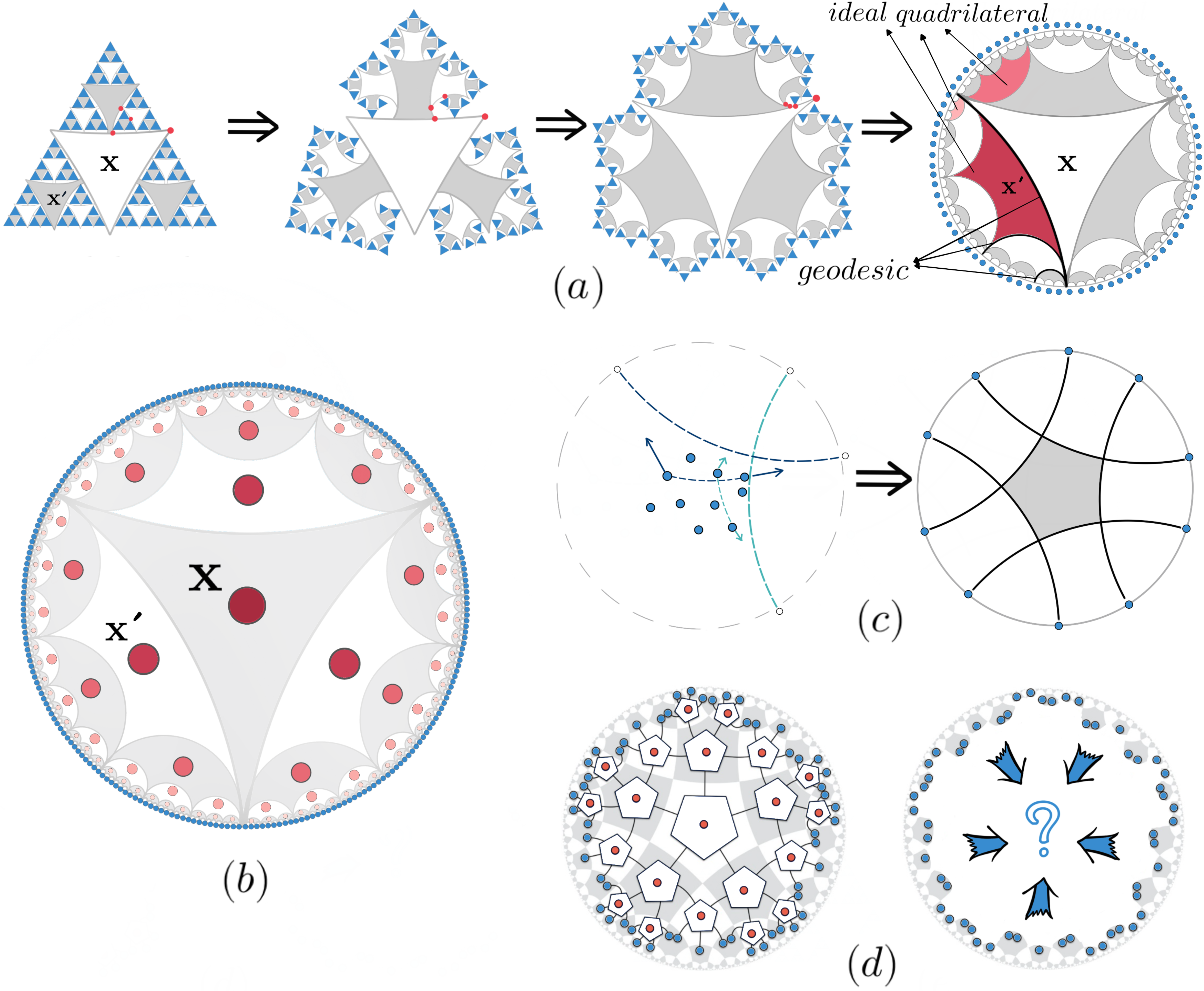}   
\phantomsubfloat{\label{5a}}\phantomsubfloat{\label{5b}}
\phantomsubfloat{\label{5c}}\phantomsubfloat{\label{5d}}
\caption{(a) The rearrangement of the physical qudits from the alternative geometry (Sierpi\'nski fractal) to the standard 1D geometry can automatically specify a geometric structure that underlies an ideal hyperbolic tessellation of the Poincar\'e disk. During the rearrangement, the holes in the alternative geometry are stretched into ideal polygons on the hyperbolic plane whose sides are the geodesics. The red dots indicate the stretch. (b) The bulk geometry discretization, where we identify the discrete location for the bulk qudits as originated from the holes that evolve during the rearrangement. (c) Illustration of how a family of hyperbolic geodesics can be specified  through rearranging pairs of physical qudits from a more ``compact'' geometry to the 1D geometry. The blue and green dashed lines indicate the pairs and are stretched into geodesics in the hyperbolic plane. (d) The illustration of the tensor-network construction of the HaPPY Pentagon code. The two figures of (d) are adapted from figures in Ref.~\cite{pastawski2015}. The geometry of the physical qubits in the HaPPY Pentagon code might be viewed as rearranged from certain unknown alternative geometry in a similar way as described in (a).}
\label{fig5}
\end{figure}
\end{center}
\twocolumngrid

\subsection{Remarks on the bulk discretization}\label{remark1}
The above arguments describes a general mechanism or scenario that appears to naturally incorporate both the specification of the bulk-geometry discretization and a seemingly essential interpretation of uberholography. Before we discuss how this scenario can underlie the construction of exact models of holographic code, it might be important to consider whether and how the established tensor-network models can fit into the above description.

For example, as illustrated in Fig.~\ref{5d}, there might be certain fractional-Hausdorff-dimension alternative geometry for the HaPPY pentagon code, in which the minimal recoveries of bulk local qudits have compact and almost connected geometry. Indeed, following the regular tessellation structures for the pentagon code, i.e., the geodesics (see Fig.~\ref{5d}), we can possibly apply the reverse of the above general scenario to ``shrink'' the uncontracted tensor legs representing the boundary (the blue dots in Fig.~\ref{5d}) to some more compact geometry of fractional Hausdorff dimension. And as we will indicate later in Sec.~\ref{remark2}, there is signature for such possibility from the perspective of the reconstruction of single bulk qudit. Nevertheless, such validation is limited to the insufficient demonstration of uberholography in these established models at their current stage of exploration, and there lacks concrete example of rigorously proved minimal recoveries~\cite{jahn2021}.

In the present work, we adopt a ``constructive'' perspective to directly build exact models of holographic code based on the above general scenario. As we will show, our construction based on the scenario inherently guarantees more systematical demonstrations of uberholography than in the established models. To explicitly show this advantage and to rigorously realize the HQEC characteristics (discussed in Sec.~\ref{hqecc}) based on the above scenario, instead of exploring the wealth of the possible alternative geometries and the rearrangements, we consider the most concise case as shown in Fig.~\ref{5a} and \ref{5b}.

It might be noted that the bulk geometry shown in Fig.~\ref{5b} does not possess the symmetry properties of the regular tessellation shown in Fig.~\ref{5d}. Indeed, in the studies of qudit systems with a finite size where the quantum-information structure can be well-defined, e.g., in quantum simulation studies, the expected symmetry properties at asymptotic limit cannot be formally established. And even if the system can be described in the large-system-size limit, the symmetry remains intrinsically different from that of the continuous AdS/CFT geometry. Hence, in a formal study of holographic codes, it might be a prioritized to focus on how to qualitatively capture the the emergent bulk locality and bulk reconstruction from the boundary entanglement.

It follows that the bulk discretization shown in Fig.~\ref{5b} might already suffice to capture the bulk hierarchical structure needed for the HQEC characteristics. And such bulk discretization appears to have advantage in accompanied by unambiguously specified boundary geometry for explicitly investigating and presenting the entanglement structures of boundary states (comparing Fig.~\ref{5a}, \ref{5d} and Ref.~\cite{anglesmunne2024}).

Indeed, as we will show in Sec.~\ref{remark3}, the exact model built based on the description of Fig.~\ref{5b} can serve as a minimalistic yet comprehensive example of holographic code, upon which models with more complex bulk geometry and more idealized symmetry properties can be straightforwardly extended. In Sec.~\ref{remark3}, we will also discuss the relationship between our model and its possible extensions to the study of AdS/CFT or p-adic AdS/CFT.

\subsection{Basics for defining the encoding}\label{iencoding}
Following the above discussion on the geometric setting, we consider how to define the encoding isometry $\mathfrak{e}_{\boldsymbol{1}}\otimes\mathfrak{e}_{\boldsymbol{2}}\otimes\mathfrak{e}_{\boldsymbol{3}}\otimes\cdots=\mathcal{E}\xrightarrow{R}\mathcal{H}=\mathfrak{h}_1\otimes\mathfrak{h}_2\otimes\mathfrak{h}_3\otimes\cdots$. The key is to ensure that the defined encoding is compatible with the discretization in the standard geometry (see the discussion in the commencement of Sec.~\ref{idea}). And the clue is that the rearrangement not only specify the discretization, but also inherently sketches the subregions of physical qudits in the alternative geometry which are rearranged to the disconnected subregions in the standard geometry that could be potentially realized as minimal recoveries in a model (see Sec.~\ref{demys}).

The clue implies that in conceiving a holographic code with a specific universal scaling component $1/h$, we can work in the alternative geometry: (1) We conceive the physical qudits $\mathfrak{h}_1\otimes\cdots\otimes\mathfrak{h}_i\otimes\cdots\otimes\mathfrak{h}_N$ as on the alternative geometry with fractional-Hausdorff-dimension $h$. (2) We use the holes to index the logical qudits $\mathfrak{e}_{\boldsymbol{1}}\otimes\cdots\otimes\mathfrak{e}_{\boldsymbol{x}}\otimes\cdots\otimes\mathfrak{e}_{K}$ (see Fig.~\ref{fig4} and \ref{5a}). Then, the sketch of minimal recoveries in the alternative geometry with respect to the holes (see Fig.~\ref{fig4}) sets up a criterion for the compatibility: the minimal boundary recoveries as derived from the algebraic properties of $\mathcal{E}\xrightarrow{R}\mathcal{H}$ need to be consistent with the sketch. And if the above criterion for compatibility are satisfied, we can rearrange back to the standard 1D geometry so that in demonstrating the HQEC characteristics the physical qudits truly have the meaning of boundary qudits and the logical qudits truly have the meaning of bulk qudits (see Fig.~\ref{5a} and \ref{5b}).

Crucially, we note that the minimal recoveries encapsulate the bare essentials of the information about the bulk qudits and their hierarchical structure consistent with the bulk radial direction (see Par.~\ref{chac6}). Hence, the encoding is not only expected to compatibly satisfy the HQEC characteristics regarding the reconstruction of single bulk qudits as predicted by the illustrations in Fig.~\ref{4b} and \ref{4c}, it is also expected to derive other HQEC characteristics. This is because given the knowledge of the recovery of single bulk qudits (the local operators or local algebra), for a boundary bipartition, it will be clear which bulk qudits or which part of its local operator algebra can be recovered on one or both sides of the bipartition in specifying the entanglement wedges. This consideration shapes our basic idea, and will be illustrated in our exact model.

To see how to define such a compatible encoding with the alternative geometry, we note that in general, defining an isometry $\mathcal{E}\xrightarrow{R}\mathcal{H}$ is equivalent to specifying (the structure of) the basis states $\{\ket*{\tilde{\varphi}_n}=R\ket{\boldsymbol{\beta}_{\mathbf{1}}\cdots\boldsymbol{\beta}_{\boldsymbol{x}}\cdots\boldsymbol{\beta}_{K}}\}_{1\le n\le d'^{K}}$ that span the code subspace $\mathcal{H}_{\mathrm{code}}$. Here, $\ket{\boldsymbol{\beta}_{\boldsymbol{1}}\cdots\boldsymbol{\beta}_{\boldsymbol{x}}\cdots\boldsymbol{\beta}_{K}}$ is the abbreviation for the bulk qudit-product basis states $\ket{\boldsymbol{\beta}_{\boldsymbol{1}}}\otimes\ket{\boldsymbol{\beta}_{\boldsymbol{2}}}\otimes\ket{\boldsymbol{\beta}_{\boldsymbol{3}}}\otimes\cdots\in\mathfrak{e}_{\boldsymbol{1}}\otimes\mathfrak{e}_{\boldsymbol{2}}\otimes\mathfrak{e}_{\boldsymbol{3}}\otimes\cdots=\mathcal{E}$. In other words, we simply need to establish a one-to-one correspondence between the orthonormal bulk states $\ket{\boldsymbol{\beta}_{\mathbf{1}}\cdots\boldsymbol{\beta}_{\boldsymbol{x}}\cdots\boldsymbol{\beta}_{K}}$s that span $\mathcal{E}$ and a family of orthonormal states $\{\ket*{\tilde{\varphi}_n}\}$ in $\mathcal{H}$.

In terms of this fact, we can view the expected error-correction/operator-reconstruction properties of an encoding $R$ as determined by the structures of the $\ket*{\tilde{\varphi}_n}$ states and their correspondence to the bulk qudit-product-states. In a step further, we might even view the corresponding bulk degrees of freedom $\boldsymbol{\beta}_{\mathbf{1}}\cdots\boldsymbol{\beta}_{\boldsymbol{x}}\cdots\boldsymbol{\beta}_{K}$ (and hence the correspondence) as ``emergent'' from the structures of the $\ket*{\tilde{\varphi}_n}$ states, and hence the desired (algebraic) HQEC characteristics based on the correspondence can all be viewed as rooted in the structures of the family of $\ket*{\tilde{\varphi}_n}$ states.

In the rest context of this section, we will discuss how the ``emergence'' viewpoint of the bulk local degrees of freedom is possible. Then, in the construction following such a viewpoint, we can start with a family of physical-qudits states with particularly defined structure, i.e., the potential boundary code basis $\ket*{\tilde{\varphi}_n}$s, based on which we can specify the ``emergent'' degrees of freedom $\boldsymbol{\beta}_{\mathbf{1}},\ldots,\boldsymbol{\beta}_{\boldsymbol{x}},\ldots,\boldsymbol{\beta}_{K}$s and define the encoding isometry. The following discussion will underlie a new possibility and a new scenario for building the encoding isometry for a holographic code, which are featured by ``emergence''.

\subsection{Entanglement patterns as structures of states}\label{pattern0}
In the following arguments, we employ a backward induction approach: We start with \emph{imagining} an encoding $R$ structure with a code basis $\{\ket*{\tilde{\varphi}_n}=R\ket{\boldsymbol{\beta}_{\mathbf{1}}\cdots\boldsymbol{\beta}_{\boldsymbol{x}}\cdots\boldsymbol{\beta}_{K}}\}_{1\le n\le d'^{K}}$, which satisfies the HQEC characteristics and is compatible with the alternative geometry setting as described above. Then, we work \emph{backward} on deriving necessary conditions on the structure of the $\ket*{\tilde{\varphi}_n}$s by showing that certain essential operator-reconstruction properties on the minimal recoveries can manifest on the structure of the code basis states. Thus, the derived conditions can serve as a guideline for incorporating the sketch of minimal recoveries in the alternative geometry, and hence for specifying the desired state structures. Indeed, the derived conditions can also sharpen the sketch.


\paragraph*{\textbf{Structure of a state}} Before proceeding to elaborate the above basic idea, we annotate on the ``structure'' of a state. While it is intricate to generally characterize the structure of a many-body state, we can loosely regard the structure of $\ket*{\tilde{\varphi}_n}$ as the whole of two aspects: (1) What qudit-product-states $\{\ket{\alpha_1\cdots\alpha_i\cdots\alpha_N}\}$s expand $\ket*{\tilde{\varphi}_n}$; (2) and how these qudit-product-states expand $\ket*{\tilde{\varphi}_n}$, i.e., the expansion coefficients. The two aspects are expected to effectively distinguish different structures among states within certain context.

\paragraph*{\textbf{Assumption on the expansion}} To simplify the formal description of the structure of states and to better elucidate the meaning of ``emergence'', we make the following assumption. We assume that in the encoding $R$ which is to be considered in the following arguments or to be construct in our model, each code basis state $\ket*{\tilde{\varphi}_n}$ (as mentioned above) is an equal-weight sum of certain qudit-product-states from the basis $\{\ket{\alpha_1\cdots\alpha_i\cdots\alpha_N}\}$. In general, the expansions of the $\ket*{\tilde{\varphi}_n}$s do not exhaust all the basis states, hence only a part of such states participate in the expansions. For convenience, we denote such qudit-product-states by $\{\ket{\psi_m}=\ket{\alpha_1\cdots\alpha_i\cdots\alpha_N}\}$.

\paragraph*{\textbf{Entanglement pattern represents structure}}\label{pattern} Following the above assumption, the coefficients in the expansion of a $\ket*{\tilde{\varphi}_n}$ state are constant. Hence, the structure of a $\ket*{\tilde{\varphi}_n}$ state can be solely captured by those qudit-product-states $\ket{\psi_m}$s that expand $\ket*{\tilde{\varphi}_n}$. Note that one qudit-product-state $\ket{\psi_m}$ is completely determined by the configuration $\alpha_1,\ldots,\alpha_i,\ldots,\alpha_N$ of the local physical-qudit degrees of freedom. Hence, a sub-collection of the $\ket{\psi_m}$ states, e.g., those expand a basis state $\ket*{\tilde{\varphi}_n}$, is specified by certain constraint on the configurations. For example, the constraint can be understood as the pattern of how up and down spins are arranged as manifesting different types of quantum ordering, which are intensively studied in the context of condensed matter physics~\cite{levin2005,chen2010,wang2017,zeng2019}. Here, we borrow the term ``\emph{entanglement pattern}'' from the context of condensed matter physics to represent the constraint that specifies a sub-collection of the $\ket{\psi_m}$ states in the expansion of $\ket*{\tilde{\varphi}_n}$.

Note that similar to the usage in condensed matter physics, here the ``entanglement pattern'' is used to represent the structure of states. The term ``entanglement'' is included because the constraint of configurations are regarded as an \emph{intuitive} nonlocal characterization of the many-body entanglement of states~\cite{zeng2019}, though, it is not clear how the many-body entanglement can be formally captured. Indeed, it will be feasible to explore the desired structure of states if the attempt focuses on the constraints of configurations $\{\alpha_1,\ldots,\alpha_i,\ldots,\alpha_N\}$. This point will be clear in terms of the following discussion and the illustrative model. In the following texts, we will use ``entanglement pattern'' and ``pattern'' interchangeably. 

The following arguments only rely on the equal-weight-sum assumption and the algebraic properties from the HQEC characteristics. The following two corollaries is even independent on the assumption, but only follow the HQEC characteristics.





\subsection{Conditions and guidelines}\label{cencoding}

\subsubsection{Correspondence between entanglement patterns and bulk degrees of freedom}
\paragraph*{\textbf{Condition 1}} Following the assumption of equal-weight summation, it is easy to show that in the \emph{imagined} encoding $R$, any two orthogonal basis states $\ket*{\tilde{\varphi}_n}$ and $\ket{\tilde{\varphi}_{n'}}$ must be expanded by disjoint sub-collections of the $\ket{\psi_m}$ states. In other words, each $\ket{\psi_m}$ only contributes to the expansion of one $\ket*{\tilde{\varphi}_n}$, and the structures of the $\ket*{\tilde{\varphi}_n}$ states are characterized by distinct patterns on the configurations $\alpha_1\cdots\alpha_i\cdots\alpha_N$. It implies that the correspondence between the bulk qudit-produt-state $\ket{\boldsymbol{\beta}_{\mathbf{1}}\cdots\boldsymbol{\beta}_{\boldsymbol{x}}\cdots\boldsymbol{\beta}_{K}}$ and the code basis state $\ket*{\tilde{\varphi}_n}$ is essentially a correspondence between the entanglement patterns and the bulk local degrees of freedom. It further means that distinguishing different configurations of bulk local degrees of freedom is equivalent to distinguishing different patterns of boundary local degrees of freedom.




The above condition shows a feasible framework to establish the correspondence in building the structure of states, i.e., to explore the distinguishability among the entanglement patterns. This framework underlies the guideline, and will be the central thread running through our construction of the illustrative model.


\subsubsection{Distinguishability of entanglement patterns}

To see the difference of the entanglement patterns corresponding to the difference in bulk local degrees of freedom, we can consider
\begin{align}\label{eg1}
\begin{split}
\ket*{\tilde{\varphi}_n}&=R\ket{\boldsymbol{\beta}_{\mathbf{1}}\cdots\boldsymbol{\beta}_{\boldsymbol{x}}\cdots\boldsymbol{\beta}_{K}}\\
\ket{\tilde{\varphi}_{n'}}&=R\ket{\boldsymbol{\beta}_{\mathbf{1}}\cdots\boldsymbol{\beta}'_{\boldsymbol{x}}\cdots\boldsymbol{\beta}_{K}}
\end{split}
\end{align}
to derive more specific conditions. Here the only difference is $\boldsymbol{\beta}'_{\boldsymbol{x}}\ne\boldsymbol{\beta}_{\boldsymbol{x}}$, other $\boldsymbol{\beta}_{\boldsymbol{x}'}$s with $\boldsymbol{x}'\ne \boldsymbol{x}$ are all the same. Indeed, for any two bulk qudit-product-states, we can interpolate a sequence of bulk qudit-product-states such that every neighboring two states differ only in one single-qudit-state. Hence, studying the entanglement patterns for $\ket*{\tilde{\varphi}_n}$ and $\ket{\tilde{\varphi}_{n'}}$ as defined in Eq.~\ref{eg1} can reveal general properties of the entanglement patterns.

Now, we consider a minimal subregion $A$ that recover $\mathcal{M}(\boldsymbol{x})$. Then, according to Lemma.~\ref{oaqec} and Prop.~\ref{cr2}, we can straightforwardly prove the following corollary~\footnote{For the proof, we notice that there is some operator $Q_A$ commuting with $P_{\mathrm{code}}$ and satisfying $P_{\mathrm{code}}Q_AP_{\mathrm{code}}=R\dyad{\boldsymbol{\beta}'_{\boldsymbol{x}}}{\boldsymbol{\beta}_{\boldsymbol{x}}}R^+$. Then, we have $\ket{\tilde{\varphi}_{n'}}=P_{\mathrm{code}}Q_AP_{\mathrm{code}}\ket*{\tilde{\varphi}_n}$. With this equality, we can easily show that $\ev{O_A}{\varphi_n}=1$ while $\ev{O_A}{\varphi_{n'}}=\ev{O_A}{P_{\mathrm{code}}Q_AP_{\mathrm{code}}\varphi_n}=0$. In a similar way, using the commutativity between operators recoverable on $A$ and operators supported on $\overline{A}$, we can also easily show that $\ev{O_{\overline{A}}}{\varphi_n}=\ev{O_{\overline{A}}}{\varphi_{n'}}$ and $\mel{\varphi_n}{O'_{\overline{A}}}{\varphi_{n'}}=0$ for any $O_{\overline{A}}$ and $O'_{\overline{A}}$.}.
\begin{corollary}\label{idea1}
Assume a well-defined encoding isometry $R$ satisfying Characteristic 2 (complementary recovery). Consider the states $\ket*{\tilde{\varphi}_n}$ and $\ket{\tilde{\varphi}_{n'}}$ as specified above, and a minimal boundary subregion $A$ recovering $\mathcal{M}(\boldsymbol{x})$. Then, (1) any boundary operators $O_A$ and $O'_A$ that reconstruct $\dyad{\boldsymbol{\beta}_{\boldsymbol{x}}}$ and $\dyad{\boldsymbol{\beta}'_{\boldsymbol{x}}}$ respectively can distinguish $\ket*{\tilde{\varphi}_n}$ and $\ket{\tilde{\varphi}_{n'}}$; (2) $\ket*{\tilde{\varphi}_n}$ and $\ket{\tilde{\varphi}_{n'}}$ are indistinguishable for any boundary operators (not necessarily commuting with $P_{\mathrm{code}}$) on the complement subregion $\overline{A}$.
\end{corollary}

\paragraph*{\textbf{Condition 2}} This corollary brings to our mind the importance of the sub-configuration $\alpha_{i_1}\alpha_{i_2}\cdots$ of physical qudits in the minimal subregion $A$ and the sub-configuration $\alpha_{j_1}\alpha_{j_2}\cdots$ of physical qudits in the complement $\overline{A}$. In terms of the correspondence between the entanglement patterns and $\{\ket{\boldsymbol{\beta}_{\mathbf{1}}\cdots\boldsymbol{\beta}_{\boldsymbol{x}}\cdots\boldsymbol{\beta}_{K}}\}$, Cor.~\ref{idea1} implies that contrasting some difference on the sub-configuration $\alpha_{i_1}\alpha_{i_2}\cdots$ suffices to distinguish the patterns corresponding to the difference in $\boldsymbol{\beta}_{\boldsymbol{x}}\ne\boldsymbol{\beta}'_{\boldsymbol{x}}$. On the contrary, no difference on complement sub-configuration $\alpha_{j_1}\alpha_{j_2}\cdots$ can distinguish the patterns. This indeed manifests the basic properties of quantum error-correction as stated in Lemma~\ref{oaqec}.


This condition combines the expectation on the entanglement patterns from three aspects. First, it describes how the operator-reconstruction properties can manifest on the entanglement patterns and the correspondence. Second, it indicates that the conditions regarding the minimal recoveries suffices to guide specifying the ``emergent'' degrees of freedom in building the structure of states. Third, it is exactly the entry point to incorporate the sketch of subregions in the alternative geometry that are rearranged to potential minimal recoveries in the standard geometry.

If we incorporate the the sketch in Fig.~\ref{fig4}, i.e., viewing the open paths surrounding the hole $\boldsymbol{x}$ as potential sub-configurations in the above condition for distinguishing $\boldsymbol{\beta}_{\boldsymbol{x}}\ne\boldsymbol{\beta}'_{\boldsymbol{x}}$, then the condition further outlines the search for the entanglement patterns. That is, in the candidate patterns, the emergent $\{\boldsymbol{\beta}_{\boldsymbol{x}}\}$, i.e., certain distinguishable features, needs to be barely read off the sub-configuration $\alpha_{i_1}\alpha_{i_2}\cdots$ supported on a sketched open path while it can never be read off the complement sub-configuration $\alpha_{j_1}\alpha_{j_2}\cdots$.


\paragraph*{\textbf{Condition 3}} Now, if we consider the indistinguishability on $\overline{A}$, Cor.~\ref{idea1} also implies properties on the expansion of the code basis states. Indeed, if some $\ket{\psi_m}=\ket{\alpha_1\cdots\alpha_i\cdots\alpha_N}$ in the expansion of $\ket*{\tilde{\varphi}_n}$ meets the sub-configuration $\alpha_{j_1}\alpha_{j_2}\cdots$ on $\overline{A}$, then there must be some $\ket{\psi_{m'}}=\ket{\alpha'_1\cdots\alpha'_i\cdots\alpha'_N}$ in the expansion of $\ket{\tilde{\varphi}_{n'}}$, which also meets the same sub-configuration on $\overline{A}$, and vice versa. It follows that the appearance frequency of such sub-configuration in the expansion of $\ket*{\tilde{\varphi}_n}$ and $\ket{\tilde{\varphi}_{n'}}$ must be equal, otherwise some operator supported on $\overline{A}$ can distinguish $\ket*{\tilde{\varphi}_n}$ and $\ket{\tilde{\varphi}_{n'}}$ through the sub-configuration. Consequently, the states $\ket*{\tilde{\varphi}_n}$, $\ket{\tilde{\varphi}_{n'}}$, and all other basis states spanning the code subspace are expanded by the same number of qudit-product-states~\footnote{The argument of sub-configuration $\alpha_{j_1}\alpha_{j_2}\cdots$ in the complement $\overline{A}$ applies to any $\ket*{\tilde{\varphi}_n}=R\ket{\boldsymbol{\beta}_{\mathbf{1}}\cdots\boldsymbol{\beta}_{\boldsymbol{x}}\cdots\boldsymbol{\beta}_{K}}$ and $\ket{\tilde{\varphi}_{n'}}=R\ket{\boldsymbol{\beta}_{\mathbf{1}}\cdots\boldsymbol{\beta}'_{\boldsymbol{x}}\cdots\boldsymbol{\beta}_{K}}$. And it simply means that the numbers of qudit-product-states in their expansion are the same. Then, since any two basis states are indeed the two ends of a sequence of such states appearing pairwise, we conclude that all such basis states are expanded by the same number of qudit-product-states.}. In other words, different entanglement patterns specify the same number of the $\ket{\psi_m}$ states.



\subsubsection{Manifestation of further operator-reconstruction property}

Further investigation on the properties of minimal boundary reconstructions reveals more restrictive conditions, and hence will provide more concrete clues. For example, the following corollary describes the basic properties of the sub-configurations as mentioned above, and can be easily proved~\footnote{The minimal boundary subregion for recovering $\mathcal{M}(\boldsymbol{x})$ cannot be unique, since Characteristic 1 guarantees that any small enough part of a minimal recovery is correctable, i.e., there exists another minimal recovery that excludes this part. However, two minimal recovery cannot be disjoint, since if so then one of the two should be correctable, and we get a contradiction to the nature of a minimal recovery.}.
\begin{corollary}
Assume a well-defined encoding $R$ that satisfies Characteristic 1. Then the minimal boundary subregion for recovering $\mathcal{M}(\boldsymbol{x})$ is not unique. And for any two different such minimal subregions $A_1$ and $A_2$, the overlap $A_1\cap A_2$ is nonempty.
\end{corollary}

\paragraph*{\textbf{Condition 4}} This corollary implies that the sub-configurations (and its underlying subregions), upon which we can distinguish the entanglement patterns corresponding to the difference $\boldsymbol{\beta}_{\boldsymbol{x}}\ne\boldsymbol{\beta}'_{\boldsymbol{x}}$, is not unique. Here, the clue is that in specifying an emergent local degrees of freedom $\{\boldsymbol{\beta}_{\boldsymbol{x}}\}$ by suitably defining the entanglement patterns, the distinguishable feature corresponding to $\boldsymbol{\beta}_{\boldsymbol{x}}\ne\boldsymbol{\beta}'_{\boldsymbol{x}}$ should apply to a large enough subregion of physical qudits, which encompasses different potential minimal subregions for recovery. However, it also says that the distinguishable feature corresponding to $\boldsymbol{\beta}_{\boldsymbol{x}}\ne\boldsymbol{\beta}'_{\boldsymbol{x}}$ can be simply read off one of such subregions. This seemingly paradoxical condition further narrow down the exploration of the satisfactory entanglement patterns, which will be illustrated in our model.

It is noticeable that this clue is consistent with the sketch for minimal recoveries, e.g., open paths surrounding a hole can be multiple as shown in Fig.~\ref{4a} and \ref{4b}. It also sharpens the sketch: those open paths surrounding a hole and to be rearranged to potential minimal recoveries in the standard geometry should overlap with each other. Moreover, it says that in candidate entanglement patterns in our search, the sub-configurations on different open paths surrounding the same hole should be correlated, and deriving more restrictive conditions can clarify such correlation.

\paragraph*{\textbf{Condition 5}} The above condition seems also paradoxical to Cor.~\ref{idea1}: sub-configuration on $\overline{A_1}$ appears unrelated to the bulk degrees of freedom at $\boldsymbol{x}$, but it has nonzero overlap with that in $A_2$, and hence partially participates in the emergence of the bulk degrees of freedom. Indeed, this observation points to further restricted conditions on the entanglement patterns: Sub-configurations on the minimal subregions $A_1,A_2,\ldots$ are also minimal in distinguishing the patterns corresponding to $\ket*{\tilde{\varphi}_n}=R\ket{\boldsymbol{\beta}_{\mathbf{1}}\cdots\boldsymbol{\beta}_{\boldsymbol{x}}\cdots\boldsymbol{\beta}_{K}}$ and $\ket{\tilde{\varphi}_{n'}}=R\ket{\boldsymbol{\beta}_{\mathbf{1}}\cdots\boldsymbol{\beta}'_{\boldsymbol{x}}\cdots\boldsymbol{\beta}_{K}}$. It follows that sub-configurations on the intersection of two minimal subregions cannot distinguish the entanglement patterns. This condition indeed manifests Characteristic 4.2 of HQEC: a bulk local qudit can be recovered in two different boundary subregions but not in their intersection.

Based on how Condition 4 sharpens the sketch, ``solving'' the above seemingly paradoxical condition by investigating how sub-configurations on different paths (see Fig.~\ref{fig5}) can be correlated might unveil form of the desired entanglement patterns.

\begin{center}
\begin{figure}[ht]
\centering
    \includegraphics[width=8.5cm]{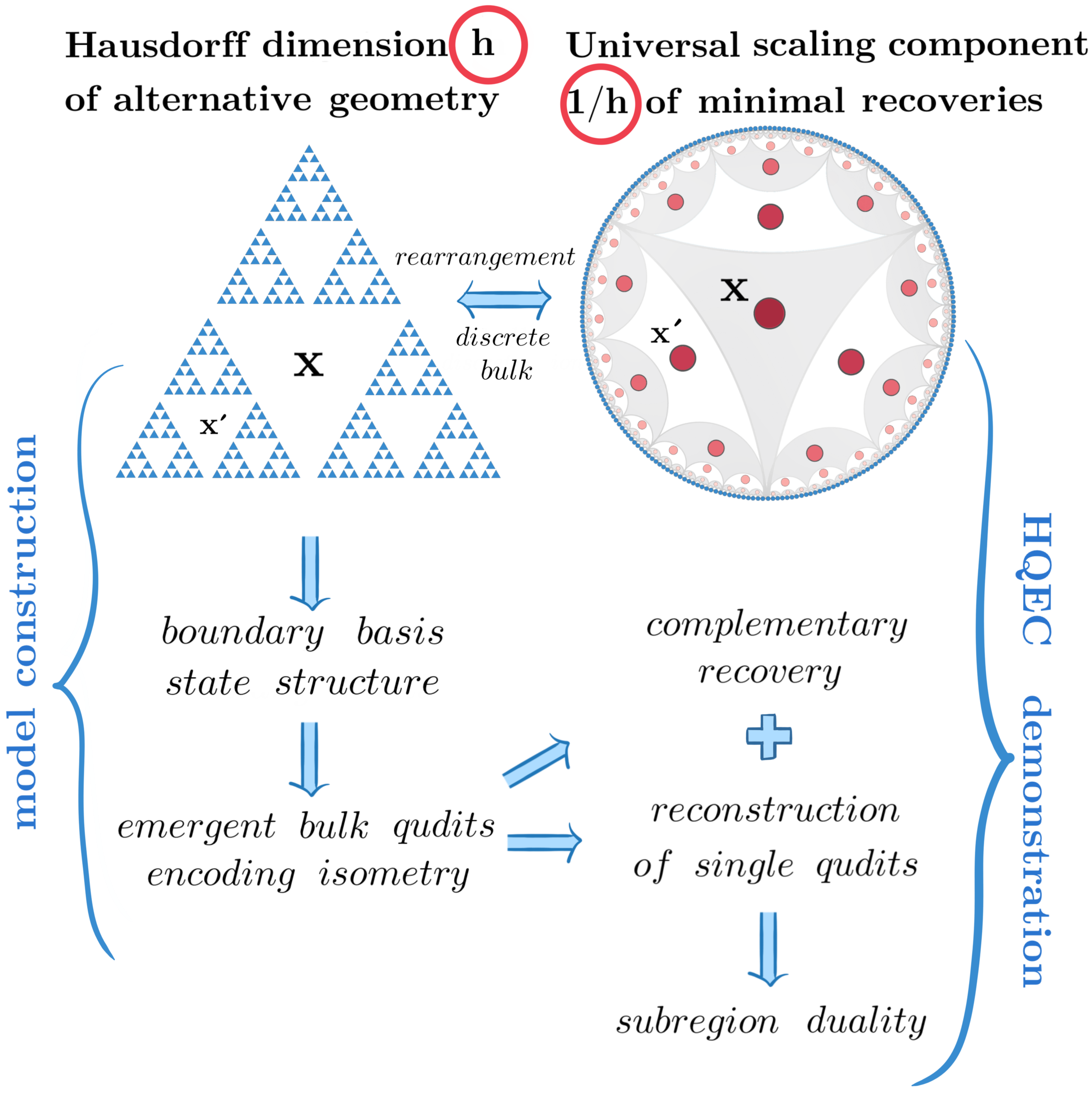}   
\caption{An outline of the steps to construct a holographic code and demonstrate the HQEC characteristics.}
\label{fig6}
\end{figure}
\end{center}

\subsubsection{Outline of construction and demonstration}\label{outline}
In principle, we can keep exploring the conditions and further concretize the guideline. However, a concrete illustration as an exact model of holographic code might be more needed to elucidate our approach as already detailed above.

As shown in Fig.~\ref{fig6}, we summarize the above discussion into basic steps for building a model and demonstrating the HQEC characteristics. The underlying idea is to identify the $h$ in the universal scaling component $1/h$ (to be demonstrated for uberholography) as the fractional Hausdorff dimension of certain alternative geometry of the physical qudits, and to specify a rearrangement between the alternative geometry and the standard 1D geometry which inherently gives rise to the bulk discretization. The geometric rearrangement translates between two different sets of geometries for both the physical/boundary and logical/bulk qudits (see Fig.~\ref{5a}, \ref{5b} and \ref{fig6}).

We firstly build the model in the alternative geometric setting, i.e., indexing the physical qudits with the vertices of the fractional-Hausdorff-dimension lattice and index the logical qudits with the holes of the lattice, and specify the structure of the encoding according to the above guidelines. Then, rearranging the system, i.e., the physical/boundary and logical/bulk qudits, to the standard geometric setting (see Fig.~\ref{5b} and \ref{fig6}), the qudits are automatically indexed in the standard way so that we can demonstrate the HQEC characteristics. 



This outline also describes the following construction of an exact illustrative model which will explicitly confirm our basic idea for the new approach of building holographic codes. And as discussed in Sec.~\ref{remark1}, based on the concise yet comprehensive model, more advanced and complex models should be straightforwardly developed by exploring the wealth of the fractaional-Hausdorff-demension geometries and associated rearrangements.

\section{Exact model}\label{exmodel}
In this section, we continue elaborating on the new approach. We concretize the meaning of ``emergence'' in our approach, and illustrate the basic idea in Sec.~\ref{iencoding} through the construction of an explicit encoding isometry with the physical/boundary qudits living on the alternative geometry and the logical/bulk qudits indexed by the holes. Especially, we show that there are explicit state structures, as described by entanglement patterns, which satisfy the conditions derived in Sec.~\ref{cencoding}. Our construction starts with specifying these states, i.e., those to be proved as the code basis states $\ket*{\tilde{\varphi}_n}$s, then, we formally establish the encoding isometry from the these states. Based on the formulation in this section, after rearranging the alternative geometric setting to the standard setting (see Fig.~\ref{fig4}, \ref{fig5} and \ref{fig6}), i.e., when viewing the qudits as living on the 1D boundary and the hyperbolic bulk respectively, we will rigorously prove the HQEC characteristics in the following sections.

\onecolumngrid
\begin{center}
\begin{figure}[ht]
\centering
    \includegraphics[width=17cm]{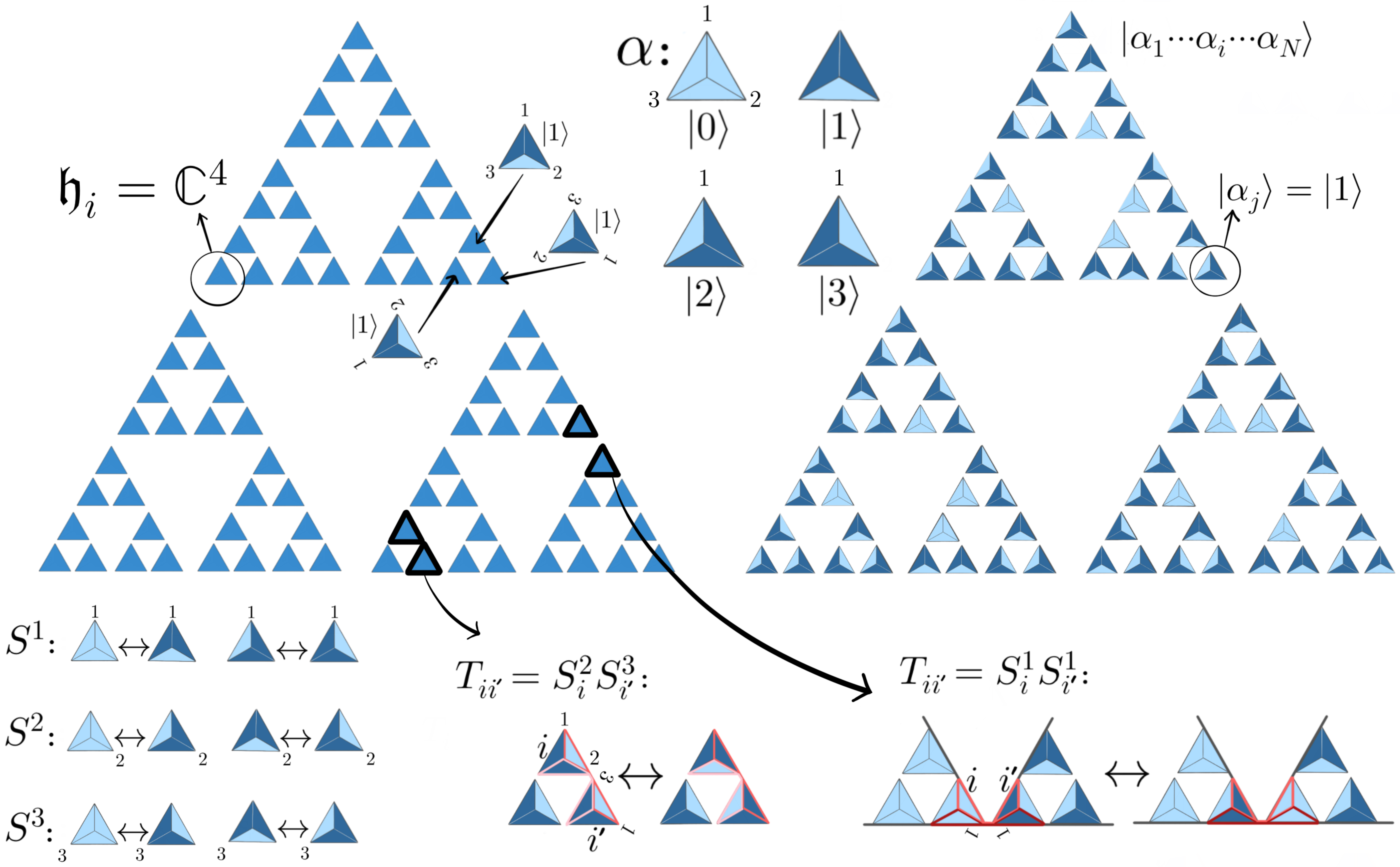}   
\caption{Pictorial representation of a physical qudit-product-state $\ket{\alpha_1\cdots\alpha_i\cdots\alpha_N}$, the convention of orientation, and illustrations of operators $S^1,S^2,S^3,T_{ii'}$. Each colored triangle has three sides with either two dark blue ones and a light blue one, or with three sides all light blue. In fitting the colored triangles to the lattice, we specify the convention of orientation by indicators $1,2,3$ for the three angles, or equivalently, the three axis. We view $1$ as the top, and the three indicator always appear clockwise from $1$ to $2$ and to $3$. The convention in fitting is that the top is always outward relative to each smallest triangular three-qudit block. The support as neighboring pairs of vertices for $T_{ii'}$ is highlighted by black solid lines. Note that the illustrations for $S^1,S^2,S^3$ include all possibilities of $\ket{\alpha}$ and hence completely describe the actions, while the illustration for a $T_{ii'}$ only include one concrete example of the $\ket{\alpha}$s.}
\label{fig7}
\end{figure}
\end{center}
\twocolumngrid

\subsection{Minimal data set}
Specifically, we arrange $N$ physical (boundary) qudits in the Sierpi\'nski geometry, which can be rearranged to the standard 1D geometry as on the boundary (see Fig.~\ref{fig4}, \ref{fig5} and \ref{fig7}). Each physical qudit is a ququart ($d=4$) with a local basis denoted by $\ket{\alpha}=\ket{0},\ket{1},\ket{2},\ket{3}\in\mathbb{C}^4$, and can be understood as a $3/2$-spin. Our choice of the ququarts coincides with recent experimental development of quantum simulator where the local degrees of freedom in the quantum hardware has inherently four levels or other number of levels greater than two~\cite{low2020,chi2022,ringbauer2022,seifert2023a,liu2023,fischer2023}.

The guiding principle for our construction of the encoding isometry is the establishment of a correspondence between the boundary state structures that describe the physical qudits and certain degrees of freedom ``emergent'' from the entanglement patterns (see Sec.~\ref{pattern0}). The emergent degrees of freedom will be proved to form the bulk, while the correspondence will derive the formulation of the encoding isometry.

In establishing the framework described by the guideline, all the components, including the qudit-product-states $\{\ket{\psi_m}\}$ that expand the boundary code basis states $\{\ket*{\tilde{\varphi}_n}\}$ (see Sec.~\ref{pattern0}), can be viewed as originated from a minimal data set on the physical qudits, i.e., 
\begin{equation}\label{data}
\ket{0\cdots0\cdots0},\quad \{T_{ii'}\},
\end{equation}
a qudit-product-state together with a family of two-qudit quantum gates. For example, while the qudit-product-state $\ket{\psi_m}=\ket{\alpha_1\cdots\alpha_i\cdots\alpha_N}$ can be characterized by certain properties of the configuration $\alpha_1,\ldots,\alpha_i,\ldots,\alpha_N$, it can also be equivalently defined as
\begin{equation}\label{dpsi}
\ket{\psi_m}=(T_{ii'}T_{jj'}\cdots)\ket{0\cdots0\cdots0},
\end{equation} 
where the product goes through certain gates within the family. The gates and the $\ket{\psi_m}$ states are inspired by our recent work~\cite{wang2023} on the long-range entanglement in many-body physics. While they are not necessary to define or to describe the model, they form a bridge to potential quantum simulation studies of the model.

In describing the $\ket{\psi_m}$s and the $\ket*{\tilde{\varphi}_n}$s, the gates also provide convenience in formally representing essential properties of the structures of the states. Hence, the gates will also be a powerful tool for demonstrating the HQEC characteristics in terms of the structures. In the following, our arguments will be based on both the properties of the gates and pictorial representations of the configuration $\alpha_1,\ldots,\alpha_i,\ldots,\alpha_N$ that specify the $\ket{\psi_m}$ states.

\subsection{Qudit-product-state}
To describe the operators and the states in a straightforward and concise way, we utilize a pictorial representation of the local qudit basis state $\ket{\alpha}$ so that we can illustrate the configuration $\alpha_1,\ldots,\alpha_i,\ldots,\alpha_N$ explicitly in pictures. In this way, the patterns of the configurations (see Par.~\ref{pattern}) for defining the boundary code states $\{\ket*{\tilde{\varphi}_n}\}$ will be clear.

\paragraph*{\textbf{Pictorial representation}}\label{picrep} We represent the basis $\ket{\alpha}=\ket{0},\ket{1},\ket{2},\ket{3}$ by the colored triangles as shown in Fig.~\ref{fig7}. Each of such triangle has three sides with either two dark blue ones and a light blue one, or with three sides all light blue---four possible cases in total. In fitting these triangles into the lattice and representing a product state $\ket{\alpha_1\cdots\alpha_i\cdots\alpha_N}$, we make the convention on the orientation by indicating the three angles of the triangle with $1,2,3$ as illustrated in Fig.~\ref{fig7}. In the illustration, $1$ labels the top of the triangle in our convention of orientation, and $2$ and $3$ follows clockwise (see Fig.~\ref{fig7}). Then, the convention of fitting is that the top of the triangle, as indicated by $1$, is outward relative to the smallest triangular three-qudit block containing the triangle. Obviously, based on the convention, all possible fittings of the colored triangles to the lattice can be mapped one-to-one onto all the qudit-product-states $\ket{\alpha_1\cdots\alpha_i\cdots\alpha_N}$. And in our following study of the physical qudits, we can always refer to this pictorial representation.


A $T_{ii'}$ gate can be viewed as a product of two single-qudit operators defined on the two physical qudits $i$ and $i'$ respectively. The index $\{ii'\}$ exactly goes through all neighbors in the Sierpi\'nski geometry. For each pair of neighboring qudits, the choice of such two single-qudit operators depends on the location of the pair in the geometry. There are three choices for such single-qudit operator on each qudit, 
\begin{align}\label{ds}
\begin{split} 
S^1=\dyad{0}{1}+\dyad{1}{0}+\dyad{2}{3}+\dyad{3}{2},\\
S^2=\dyad{0}{2}+\dyad{2}{0}+\dyad{1}{3}+\dyad{3}{1},\\
S^3=\dyad{0}{3}+\dyad{3}{0}+\dyad{1}{2}+\dyad{2}{1}.
\end{split}
\end{align}
Here, the superscripts $1,2,3$ coincide with the three indicators of the angles for the convention of orientation (see Fig.~\ref{fig7}). Or equivalently, they can be viewed as the three directions of axis of the colored-triangle-representation of the qudit basis state $\ket{\alpha}$. Then as illustrated in Fig.~\ref{fig7}, the action of the three operators on the local basis state simply shifts the dark and light sides or annihilate/create dark sides on the two sides of the corresponding axis.

For convenience of discussion, we use $\sigma=0,1,2,3$ for the superscript and define $S^0=\mathds{1}$. Obviously, $S^\sigma$ is both unitary and self-adjoint, and hence, according to the definition, we can show that
\begin{align}\label{ss}
\begin{split}
&S^\sigma S^\sigma=\mathds{1},\quad S^\sigma S^{\sigma'}=S^{\sigma'}S^\sigma \\
&S^1 S^2=S^3,\quad S^1 S^3=S^2,\quad S^2 S^3=S^1\\
&S^1 S^2S^3=\mathds{1}.
\end{split}
\end{align}
It is also easy to check that the multiplications of the $S^\sigma$ operators and projections of the form $\dyad{\alpha}$ for $\ket{\alpha}=\ket{0},\ket{1},\ket{2},\ket{3}$ result in a basis of operators, $\{\dyad{\alpha'}{\alpha}\}$, which spans $\mathbf{L}(\mathbb{C}^4)$. Hence, the $S^\sigma$ operators together with the $\dyad{\alpha}$ projections form the generators of all operators on the single qudit.

\paragraph*{\textbf{The two-qudit gates}} Now, we define that when the neighboring qudits $i$ and $i'$ belong to different smallest triangular three-qudit blocks (see Fig.~\ref{fig7}), we have
\begin{equation}\label{dt1}
T_{ii'}=S^1_iS^1_{i'};
\end{equation}
and when $i$ and $i'$ belong to the same smallest triangular three-qudit block, and it goes clockwise from $i$ to $i'$ (see Fig.~\ref{fig7}), we have
\begin{equation}\label{dt2}
T_{ii'}=S^2_iS^3_{i'}.
\end{equation}
By definition, we can easily show the unitarity of $T_{ii'}$, its self-adjointness and the commutativity of any two $T_{ii'}$s,
\begin{align}\label{tc1}
\begin{split}
T_{ii'}=T_{ii'}^+=T_{ii'}^{-1},\quad [T_{ii'},T_{i''i'''}]=0.
\end{split}
\end{align}

Formally, we can now simply define $\{\ket{\psi_m}\}$ as all the states satisfying Eq.~\ref{dpsi}, i.e. any qudit-product-state that can be obtained by applying a finite number of the $T_{ii'}$ gates to $\ket{0\cdots0\cdots0}$. However, to explicitly describe the $\ket{\psi_m}$ states in terms of the configuration $\alpha_1,\ldots,\alpha_i,\ldots,\alpha_N$, we illustrate the action of the $T_{ii'}$ gates on a product state $\ket{\alpha_1\cdots\alpha_i\cdots\alpha_N}$.

We take advantage of the pictorial representation of a qudit-product-state as shown in Fig.~\ref{fig7} and \ref{8a}, which can be viewed as the structure of interlocking loops encircling the holes of the lattice. As shown in Fig.~\ref{8a}, such a loop can be viewed as a strand of the dark/light blue sides of the colored triangles tightly surrounding a hole, and the loops are interlocked through the colored triangles, i.e., each such triangle is engaged in three loops (see Fig.~\ref{8a}). An important property of such a loop is the parity of the number of dark sides (see Fig.~\ref{8a}). For instance, in $\ket{0\cdots 0\cdots 0}$, the parity for each loop is even (see Fig.~\ref{8b}).


As illustrated in Fig.~\ref{fig7}, the action of $T_{ii'}$ on $\ket{\alpha_1\cdots\alpha_i\cdots\alpha_N}$ affects the two loops that qudits $i$ and $i'$ are both engaged: $T_{ii'}$ either annihilates/create a pair of dark sides in a loop (see the left illustration of $T_{ii'}$ in Fig.~\ref{fig7}), or switches between the configurations (dark, light) and (light, dark) for pair of sides in a loop (see the right illustration of $T_{ii'}$ in Fig.~\ref{fig7}). Then, as illustrated in Fig.~\ref{8b}, such action on $\ket{0\cdots 0\cdots 0}$ results in a product state $\ket{\psi_{m_1}}=\ket{\alpha_1\cdots\alpha_i\cdots\alpha_N}$ which preserves the parity of the number of dark sides on each loop (also on each of the three laterals of the lattice) to be \emph{even}. Obviously, further acting arbitrarily finite number of the $T_{ii'}$ operators (supported on distinct neighboring qudits) on $\ket{\psi_{m_1}}$ also results in new product states, e.g., $\ket{\psi_{m_2}}$ and $\ket{\psi_{m_3}}$ as shown in Fig.~\ref{8b}, which still preserves the even parity (see Fig.~\ref{8b}).

\paragraph*{\textbf{The physical qudit-product-states}}\label{psim} Hence, according to Eq.~\ref{dpsi}, each $\ket{\psi_m}$ is simply a product state that has even parity of the number of dark sides on each loop (also on the laterals). Furthermore, it can be proved~\cite{wang2023} that any such product state of even-parity property can be defined by Eq.~\ref{dpsi}. Therefore, the even-parity property on the configuration $\alpha_1,\ldots,\alpha_i,\ldots,\alpha_N$ exactly characterizes the collection $\{\ket{\psi_m}\}$ of qudit-product-states. Note that by definition, $\ket{0\cdots0\cdots0}$ belongs to the collection.

Indeed, the characterization of the $\ket{\psi_m}$ states in Eq.~\ref{dpsi} can be refined. That is, there is an accurate correspondence relation between different sub-collections of the gates and distinct $\ket{\psi_m}$ states. This relation is summarized in the following proposition, which is important in defining the encoding isometry and in presenting the structural properties of the boundary code states in terms of the gates.

\begin{center}
\begin{figure}[ht]
\centering
    \includegraphics[width=8.5cm]{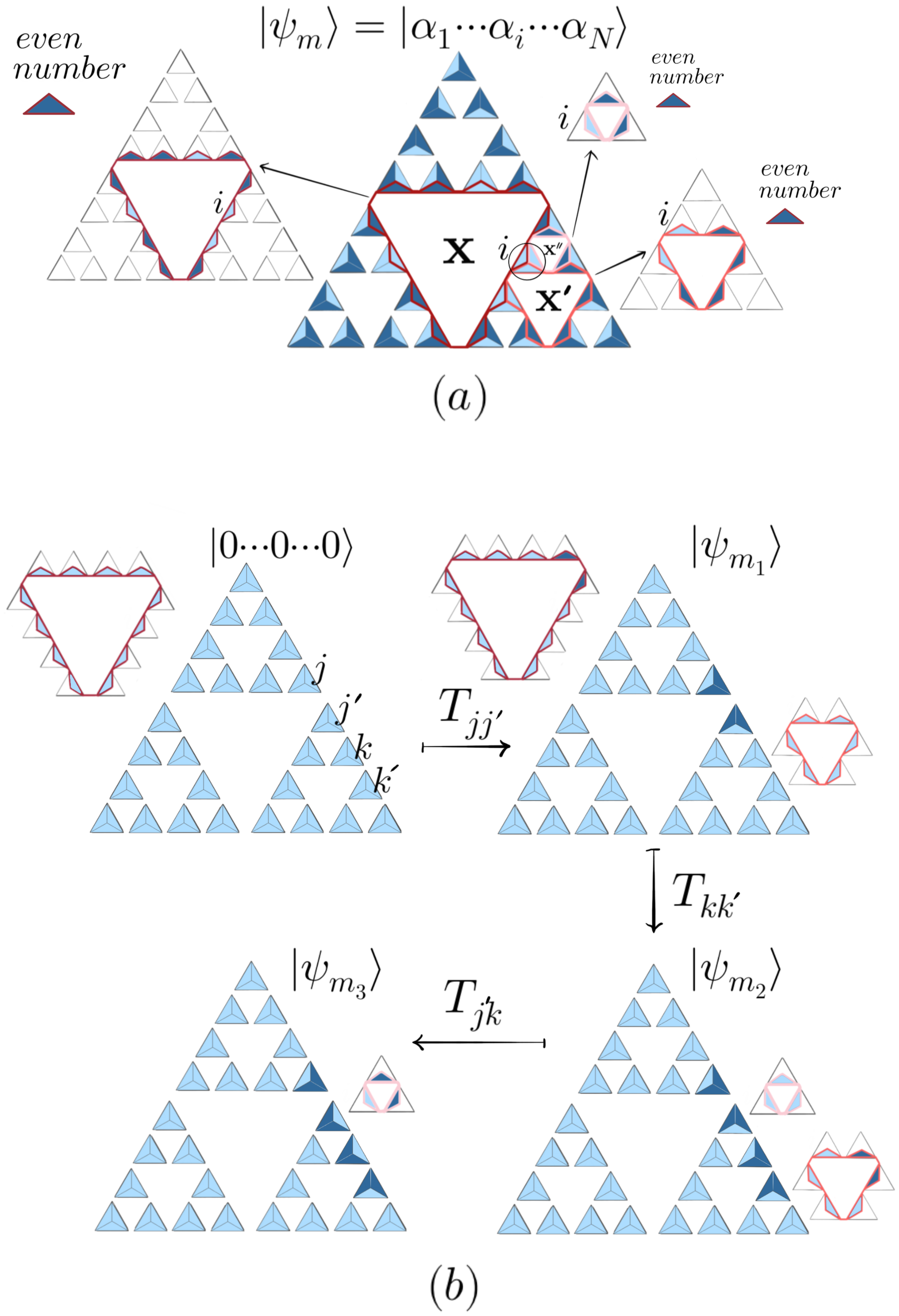}   
\phantomsubfloat{\label{8a}}\phantomsubfloat{\label{8b}}
\caption{Pictorial illustration of the physical qudit-product-state $\ket{\psi_m}$. (a) views the pictorial illustration given by Fig.~\ref{fig7} as interlocking loops surround the holes, in which each loop is simply a strand of the dark/light blue sides of the colored triangles. The qudit $i$, circulated by black solid line, is engaged in three such loops. In $\ket{\psi_m}$, the parity of dark sides in each loop (and also in each of the three laterals of the lattice) is even. (b) $T_{jj'}$ maps $\ket{0\cdots0\cdots0}$ to $\ket{\psi_{m_1}}$, $T_{kk'}$ maps $\ket{\psi_{m_1}}$ to $\ket{\psi_{m_2}}$, and $T_{j'k}$ maps $\ket{\psi_{m_2}}$ to $\ket{\psi_{m_3}}$. The even-parity for all the loops is kept through the mappings.}
\label{fig8}
\end{figure}
\end{center}

\begin{proposition}\label{gates1}
Fix an arbitrary $\ket{\psi_{m_0}}$, the states of the form $(T_{ii'}T_{jj'}\cdots)\ket{\psi_{m_0}}$ go through all the $\ket{\psi_m}$ states. And any two distinct products $(T_{i_1i'_1}T_{i_2i'_2}\cdots)$ and $(T_{j_1j'_1}T_{j_2j'_2}\cdots)$ always map $\ket{\psi_{m_0}}$ to distinct qudit-product-states.
\end{proposition}

The proof of the proposition is given the in App.\ref{pogates1}. For simplicity in the following arguments, we define $\mathcal{H}_0\subset\mathcal{H}$ as the subspace spanned by $\{\ket{\psi_m}\}$, and denote the corresponding projection operator by $P_0$ ($P_0\mathcal{H}=\mathcal{H}_0$). Since each $T_{ii'}$ (and its adjoint) maps one $\ket{\psi_m}$ to another $\ket{\psi_{m'}}$ still within $\mathcal{H}_0$, we have
\begin{equation}\label{tc2}
[P_0,T_{ii'}]=0.
\end{equation}

\onecolumngrid
\begin{center}
\begin{figure}[ht]
\centering
    \includegraphics[width=17cm]{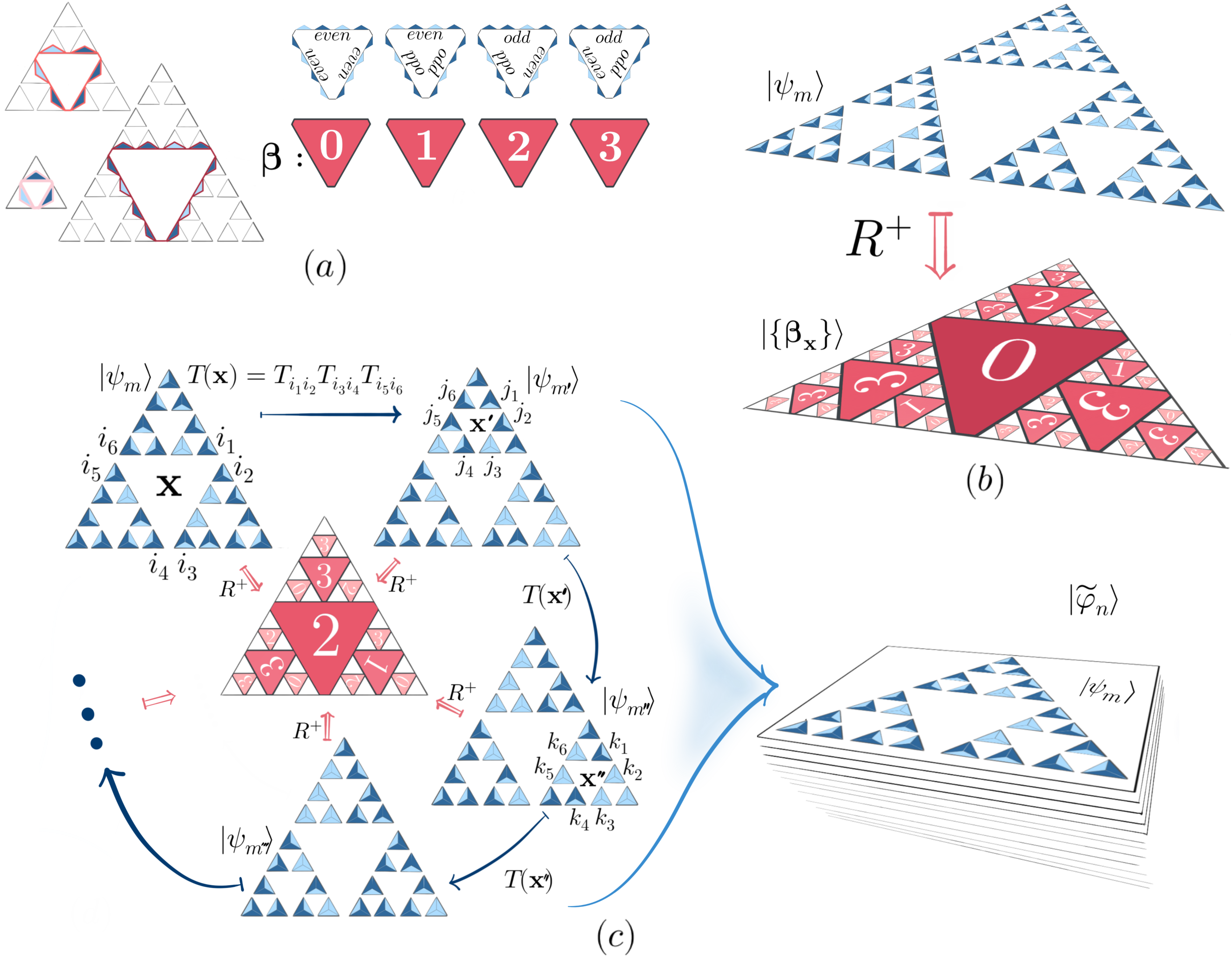}   
\phantomsubfloat{\label{9a}}\phantomsubfloat{\label{9b}}\phantomsubfloat{\label{9c}}
\caption{The emergent degrees of freedom, the assembled gates and the entanglement patterns for the code basis state $\ket{\tilde{\varphi}_n}$. (a) The degrees of freedom for an emergent qudit can be read off the four possibilities of parities on the laterals of a loop. (b) According to the way the emergent degrees of freedom are specified, each physical qudit-product-state $\ket{\psi_m}$ is associated to a unique configurations of the emergent degrees of freedom $\{\boldsymbol{\beta}_{\boldsymbol{x}}\}=\boldsymbol{\beta}_{\boldsymbol{1}},\ldots,\boldsymbol{\beta}_{\boldsymbol{x}},\ldots,\boldsymbol{\beta}_{K}$, or equivalently characterized by a unique entanglement pattern. (c) The assembled gates and the illustration of their actions on a physical qudit-product-state. The action of any product of assembled gates on $\ket{\psi_m}$ does not change the associated $\ket{\{\boldsymbol{\beta}_{\boldsymbol{x}}\}_n}$ or the characterizing entanglement pattern. An equal-weight sum of all the $\ket{\psi_m}$s associated to the same $\ket{\{\boldsymbol{\beta}_{\boldsymbol{x}}\}_n}$ gives rise to $\ket{\tilde{\varphi}_n}$.}
\end{figure}
\label{fig9}
\end{center}
\twocolumngrid

\subsection{Entanglement patterns and emergent bulk qudits}\label{pattern1}
Based on the qudit-product-states $\{\ket{\psi_m}=\ket{\alpha_1\cdots\alpha_i\cdots\alpha_N}\}$, we can define the code states and the encoding isometry. According to the basic idea discussed in Sec.~\ref{iencoding}, the key is to identify certain patterns of the configurations $\alpha_1,\cdots,\alpha_i,\cdots,\alpha_N$, which are expected to give rise to the emergent bulk degrees of freedom and demonstrate the conditions derived in Sec.~\ref{iencoding}.

Recall that the patterns of the configurations can be viewed as constraints on those $\alpha_i$s in a qudit-product-state. While a specific configuration $\alpha_1,\cdots,\alpha_i,\cdots,\alpha_N$ determines a $\ket{\psi_m}$ state, a specific pattern on the configurations characterizes a sub-collection of the $\ket{\psi_m}$ states. As we will show, certain nonlocal features of the configurations in the $\ket{\psi_m}$ states give rise to the entanglement patterns and also underlie the emergent degrees of freedom.

In terms of the pictorial representation of the $\ket{\psi_m}$ states (see Fig.~\ref{fig7} and \ref{8a}), we focus on the loop structure surrounding each hole $\boldsymbol{x}$. We can identify four possibilities of the sub-configuration $\alpha_{i_1},\alpha_{i_2},\ldots$ on the loop (see Fig.~\ref{8a} and \ref{9a}). The four possibilities can be viewed as degrees of freedom $\boldsymbol{\beta}=\boldsymbol{0},\boldsymbol{1},\boldsymbol{2},\boldsymbol{3}$: $\boldsymbol{\beta}=\boldsymbol{0}$ corresponds to the possibility where the three loop laterals all have even number of dark sides; $\boldsymbol{\beta}=\boldsymbol{1},\boldsymbol{2}$ or $\boldsymbol{3}$ correspond to those in which two out of the three loop laterals have odd number of dark sides (see Fig.~\ref{9a}). Note that the convention of orientation is indicated in Fig.~\ref{9b}.

\paragraph*{\textbf{Entanglement patterns}}\label{pattern11} Clearly, specifying one out of the four possibilities for each hole imposes nonlocal constraint on the full configuration $\alpha_1,\cdots,\alpha_i,\cdots,\alpha_N$ in the $\ket{\psi_m}$ states. Then, by a entanglement pattern, we simply refer to such a specification. Obviously, all the entanglement patterns (specifications) classify the the $\ket{\psi_m}$s into disjoint sub-collections, i.e., each $\ket{\psi_m}$ is associated to a unique entanglement pattern (see Fig.~\ref{9a} and \ref{9b}). Following the basis idea in Sec.~\ref{cencoding}, later we will formally define a code basis state $\ket*{\tilde{\varphi}_n}$ as expanded by the $\ket{\psi_m}$s within such a sub-collection. As will be clear in the next section, the loop of qudits include different subregions of minimal recoveries for the same bulk qudit. Indeed, specifying the entanglement patterns in terms of the loops follows the guideline from Condition 4 and 5 in Sec.~\ref{cencoding}.

\paragraph*{\textbf{Emergent bulk qudits}} More formally, for each hole $\boldsymbol{x}$, we can associate a ququart $\mathfrak{e}_{\boldsymbol{x}}=\mathbb{C}^4$ (also labeled by $\boldsymbol{x}$) with a local basis $\ket{\boldsymbol{\beta}_{\boldsymbol{x}}}=\ket{\boldsymbol{0}},\ket{\boldsymbol{1}},\ket{\boldsymbol{2}},\ket{\boldsymbol{3}}$ according to the four possibilities. It means that specifying one out of the four possibilities for each hole indeed specifies a product state $\ket{\boldsymbol{\beta}_{\boldsymbol{1}}\cdots\boldsymbol{\beta}_{\boldsymbol{x}}\cdots\boldsymbol{\beta}_{K}}\in\mathfrak{e}_{\boldsymbol{1}}\otimes\mathfrak{e}_{\boldsymbol{2}}\otimes\mathfrak{e}_{\boldsymbol{3}}\otimes\cdots$.

In the following, we call the emergent ququarts $\mathcal{E}=\mathfrak{e}_{\boldsymbol{1}}\otimes\mathfrak{e}_{\boldsymbol{2}}\otimes\mathfrak{e}_{\boldsymbol{3}}\otimes\cdots=(\mathbb{C}^{4})^{\otimes K}$ as the logical/bulk qudits in our model, since according to the rearrangement linking the alternative geometry to the standard geometry, the index $\boldsymbol{x}$ for the holes is the index of the discrete hyperbolic geometry (see Fig.~\ref{fig5}), and the emergent qudits are indeed assigned with specific bulk discrete geometry in the standard geometric setting as shown in Fig.~\ref{5b}. Then, based on the nonlocal constraints, we have established a correspondence between the entanglement patterns and the bulk qudit-product-states $\ket{\boldsymbol{\beta}_{\boldsymbol{1}}\cdots\boldsymbol{\beta}_{\boldsymbol{x}}\cdots\boldsymbol{\beta}_{K}}\in\mathcal{E}$. This correspondence is the essence of the desired encoding isometry and will be formalized below.

\subsection{Encoding isometry and code basis states}
To show that the above correspondence between the entanglement patterns and the $\ket{\boldsymbol{\beta}_{\boldsymbol{1}}\cdots\boldsymbol{\beta}_{\boldsymbol{x}}\cdots\boldsymbol{\beta}_{K}}$ states can be formalized into to a well-defined encoding isometry, and to explicitly define the code basis states, we take advantage of both the properties of the gates and the pictorial representation of the entanglement patterns. For simplicity in notation, we use $\ket{\{\boldsymbol{\beta}_{\boldsymbol{x}}\}_n}$ to denote all the $\ket{\boldsymbol{\beta}_{\boldsymbol{1}}\cdots\boldsymbol{\beta}_{\boldsymbol{x}}\cdots\boldsymbol{\beta}_{K}}$ states that are engaged in the above correspondence.

It is convenient to define an operator that will play the role of $\mathcal{H}\xrightarrow{R^+}\mathcal{E}$, which will naturally capture how we have established the above correspondence. Then, we just need to prove that $R=(R^+)^+$ is an isometry so that $R$ will be the desired encoding. Indeed, according to the above discussion, each $\ket{\psi_m}$ belongs to a unique sub-collection characterized by a entanglement pattern, and is hence assigned to a unique $\ket{\{\boldsymbol{\beta}_{\boldsymbol{x}}\}_n}$ state. Since the $\ket{\psi_m}$s are orthonormal, this assignment enables us to define an operator~\footnote{For states orthogonal to $\mathcal{H}_0$, they are assigned to $0$.} such that
\begin{equation}\label{dr+}
R^+\ket{\psi_m}=(1/\sqrt{2^K})\ket{\{\boldsymbol{\beta}_{\boldsymbol{x}}\}_n}, 
\end{equation}
where the meaning of the constant $1/\sqrt{2^K}$ will be clear later.

By definition, $R^+$ accurately capture the correspondence. That is, all the $\ket{\psi_m}$s with $R^+\ket{\psi_m}$ equal to a specific $\ket{\{\boldsymbol{\beta}_{\boldsymbol{x}}\}_n}$ state exactly form the sub-collection characterized by the entanglement pattern that corresponds to $\ket{\{\boldsymbol{\beta}_{\boldsymbol{x}}\}_n}$. Then, if we define $\ket*{\tilde{\varphi}_n}$ as the equal-weight summation of all these $\ket{\psi_m}$ states assigned to $\ket{\{\boldsymbol{\beta}_{\boldsymbol{x}}\}_n}$, and times the constant $1/\sqrt{2^K}$, we can write $R^+$ as
\begin{equation*}
R^+=\sum_n\dyad{\{\boldsymbol{\beta}_{\boldsymbol{x}}\}_n}{\varphi_n}.
\end{equation*}
Thus, $R=(R^+)^+$ has the form
\begin{equation}\label{encoding1}
R=\sum_n\dyad{\varphi_n}{\{\boldsymbol{\beta}_{\boldsymbol{x}}\}_n}.
\end{equation}

Now, the correspondence between the entanglement patterns and the $\ket{\{\boldsymbol{\beta}_{\boldsymbol{x}}\}_n}$ states is clearly presented in $R$. To show that $R$ is the desired encoding isometry ($R^+R=\mathds{1}$) and $\{\ket*{\tilde{\varphi}_n}\}$ forms the desired code basis, we simply need to prove that the $\ket{\{\boldsymbol{\beta}_{\boldsymbol{x}}\}_n}$ states in Eq.~\ref{encoding1} cover all the possible product states $\ket{\boldsymbol{\beta}_{\boldsymbol{1}}\cdots\boldsymbol{\beta}_{\boldsymbol{x}}\cdots\boldsymbol{\beta}_{K}}$s that span $\mathcal{E}$, and the $\ket*{\tilde{\varphi}_n}$s are orthonormal. With this in mind, we further investigate the properties of the entanglement patterns by studying how the the $\{T_{ii'}\}$ gates represent the structural properties of the $\ket*{\tilde{\varphi}_n}$ states and the $R$ operator.

\paragraph*{\textbf{Assembled gates}} For convenience in the following arguments, we assemble the $T_{ii'}$ gates into more complex $T(\boldsymbol{x})$ gates in the following way: We define each $T(\boldsymbol{x})$ associated to a hole (bulk qudit) $\boldsymbol{x}$ as $T(\boldsymbol{x})=T_{i_1i_2}T_{i_3i_4}T_{i_5i_6}$, i.e. the product of the $T_{ii'}$ operators on the three corners of the loop surrounding the hole (see Fig.~\ref{9c}). For the smallest loop formed by three qudits $i,j,k$, we define $T(\boldsymbol{x})=T_{ij}T_{jk}T_{ki}$. An important property of the way we assemble the gates is that each gate $T_{ii'}$ belongs to a unique one of such assembling, which can be easy checked with Fig.~\ref{fig7} and \ref{9c}. It follows that the way we assemble the gates indeed partitions the gates according to the holes. Due to the unitarity and the self-adjointness of the $T_{ii'}$ operators, $T(\boldsymbol{x})$ is also unitary and self-adjoint. Furthermore, since the $T_{ii'}$s commute with each other and commute with $P_0$, the $T(\boldsymbol{x})$s also commute with each other and commute with $P_0$. 



Our following arguments are based on the fact that any two $\ket{\psi_m}$ and $\ket{\psi_{m'}}$ can be connected by applying suitably chosen $T_{ii'}$ gates (see Prop.~\ref{gates1}), according to which we can prove a proposition for how the the associated emergent degrees of freedom can change if applying the gates (both the $T_{ii'}$s and the $T(\boldsymbol{x})$s) to a $\ket{\psi_m}$ state.
\begin{proposition}\label{pp1}
We follow the above setting for constructing the code. (1) Consider an arbitrary $\ket{\psi_m}$ and an arbitrary $T(\boldsymbol{x})=T_{i_1i_2}T_{i_3i_4}T_{i_5i_6}$. The three bulk qudit-product-states 
\begin{equation*}
\begin{split}
&\ket{\boldsymbol{\beta}_{\boldsymbol{1}}\cdots\boldsymbol{\beta}_{\boldsymbol{x}}\cdots\boldsymbol{\beta}_{K}}=\sqrt{2^K}R^+\ket{\psi_m}\\
&\ket{\boldsymbol{\beta}_{\boldsymbol{1}}\cdots\bar{\boldsymbol{\beta}}_{\boldsymbol{x}}\cdots\boldsymbol{\beta}_{K}}=\sqrt{2^K}R^+T_{i_1i_2}\ket{\psi_m}\\&\ket{\boldsymbol{\beta}_{\boldsymbol{1}}\cdots\tilde{\boldsymbol{\beta}}_{\boldsymbol{x}}\cdots\boldsymbol{\beta}_{K}}=\sqrt{2^K}R^+T_{i_3i_4}T_{i_1i_2}\ket{\psi_m}
\end{split}
\end{equation*}
differ from each other exactly in $\boldsymbol{\beta}_{\boldsymbol{x}}\ne\bar{\boldsymbol{\beta}}_{\boldsymbol{x}}\ne\tilde{\boldsymbol{\beta}}_{\boldsymbol{x}}\ne\boldsymbol{\beta}_{\boldsymbol{x}}$ (other $\boldsymbol{\beta}_{\boldsymbol{x}'}$s are the same). Here, $T_{i_1i_2}$ can be replaced by any other in $\{T_{i_1i_2},T_{i_3i_4},T_{i_5i_6}\}$, and $T_{i_3i_4}T_{i_1i_2}$ can be replaced by the product of any two out of the three. Furthermore, by suitably choosing operators in the product, $\bar{\boldsymbol{\beta}}_{\boldsymbol{x}}$ (or $\tilde{\boldsymbol{\beta}}_{\boldsymbol{x}}$) can be arbitrary in $\boldsymbol{0},\boldsymbol{1},\boldsymbol{2},\boldsymbol{3}$ (but different from $\boldsymbol{\beta}_{\boldsymbol{x}}$).

 
(2) For arbitrary $\ket{\psi_m}$ and arbitrary nontrivial product $(T_{ii'}T_{jj'}\cdots)$, the equality
\begin{equation*}
R^+(T_{ii'}T_{jj'}\cdots)\ket{\psi_m}=R^+\ket{\psi_m}.
\end{equation*} 
holds if and only if $(T_{ii'}T_{jj'}\cdots)$ can be written as a product of the $T(\boldsymbol{x})$ gates, i.e. $(T(\boldsymbol{x})T(\boldsymbol{x}')\cdots)$. 
%
\end{proposition}

The proof for this proposition is given in App.~\ref{popp1}. This proposition simply says that applying suitably chosen product of the $T_{ii'}$ gates to $\ket{\psi_m}$, we can reach a $\ket{\psi_{m'}}$ which is mapped through $R^+$ to an arbitrarily targeted $\ket{\boldsymbol{\beta}'_{\boldsymbol{1}}\cdots\boldsymbol{\beta}'_{\boldsymbol{x}}\cdots\boldsymbol{\beta}'_{K}}$ in $\mathcal{E}$. In other words, the $\ket{\{\boldsymbol{\beta}_{\boldsymbol{x}}\}_n}$ states appeared in the range of $R^+$ traverse the whole basis spanning $\mathcal{E}$, which is desired.

Prop.~\ref{pp1} also says that if the applied gates consists of the $T(\boldsymbol{x})$s, the resulted physical qudit-product-state will be confined within the characterization of the same entanglement pattern and also the same corresponding $\ket{\{\boldsymbol{\beta}_{\boldsymbol{x}}\}_n}$. Combining Prop.~\ref{pp1} and Prop.~\ref{gates1}, we can conclude that for an arbitrary entanglement pattern, all the $\ket{\psi_m}$ states lying within the sub-collection characterized by this entanglement pattern can be specified in the following way: We fix an arbitrary $\ket{\psi_m}$ in the sub-colledtion, then any other $\ket{\psi_{m'}}$ can be obtained by applying certain product of the $T(\boldsymbol{x})$ gates to $\ket{\psi_m}$, i.e. $\ket{\psi_{m'}}=(T(\boldsymbol{x})T(\boldsymbol{x}')\cdots)\ket{\psi_m}$~\footnote{It should be noted that by fixing different $\ket{\psi_m}$ and $\ket{\psi_{m'}}$ lying within the same sub-collection, by applying the $T(\boldsymbol{x})$ gates we can traverse the same sub-collection of states. There is no contradiction. Indeed, because of the equality $T(\boldsymbol{x})T(\boldsymbol{x})=\mathds{1}$, if we have $\ket{\psi_{m'}}=(T(\boldsymbol{x})T(\boldsymbol{x}')\cdots)\ket{\psi_m}$ then we also have $(T(\boldsymbol{x})T(\boldsymbol{x}')\cdots)\ket{\psi_{m'}}=(T(\boldsymbol{x})T(\boldsymbol{x}')\cdots)(T(\boldsymbol{x})T(\boldsymbol{x}')\cdots)\ket{\psi_m}=\ket{\psi_m}$. That is, $\ket{\psi_m}$ and $\ket{\psi_{m'}}$ can be expressed in terms of each other with the same product $(T(\boldsymbol{x})T(\boldsymbol{x}')\cdots)$.}. Moreover, according to Prop.~\ref{gates1}, it is clear that this way of specification establishes a one-to-one correspondence between all the $\ket{\psi_{m'}}$ states within the same sub-collection and all the possible products of the $T(\boldsymbol{x})$ gates.

Note that all the possible products of the $T(\boldsymbol{x})$ gates are exactly the terms in the expansion of the following projection operator~\footnote{Accroding to $[T(\boldsymbol{x}),T(\boldsymbol{x}')]=0$, $T(\boldsymbol{x})T(\boldsymbol{x})=\mathds{1}$ and $(\frac{\mathds{1}+T(\boldsymbol{x})}{2})(\frac{\mathds{1}+T(\boldsymbol{x})}{2})=(\frac{\mathds{1}+T(\boldsymbol{x})}{2})$, it can be easily checked that this operator is a projection.}
\begin{align}\label{exp1}
\begin{split}
\prod_{\boldsymbol{x}=1}^{K}(\frac{\mathds{1}+T(\boldsymbol{x})}{2})=&\frac{1}{2^{K}}(\mathds{1}+\sum_{\boldsymbol{x}=1}^{K}T(\boldsymbol{x})+\sum_{\boldsymbol{x}<\boldsymbol{x}'}^{K}T(\boldsymbol{x})T(\boldsymbol{x}')\\
&+\sum_{\boldsymbol{x}<\boldsymbol{x}'<\boldsymbol{x}''}^{K}T(\boldsymbol{x})T(\boldsymbol{x}')T(\boldsymbol{x}'')\\
&+\cdots+\prod_{\boldsymbol{x}=1}^{K}T(\boldsymbol{x})),
\end{split}
\end{align} 
where there are exactly $2^K$ terms.

\paragraph*{\textbf{Code basis state}}\label{cbs0} Then, for a given entanglement pattern (or a corresponding $\ket{\{\boldsymbol{\beta}_{\boldsymbol{x}}\}_n}$ state) and an arbitrary $\ket{\psi_m}$ that lies within the sub-collection characterized by the entanglement pattern, the accordingly defined $\ket*{\tilde{\varphi}_n}$ state has the following form
\begin{align}\label{cbs1}
\begin{split}
&\ket*{\tilde{\varphi}_n}\\
&=\sqrt{2^{K}}\prod_{\boldsymbol{x}=1}^{K}(\frac{\mathds{1}+T(\boldsymbol{x})}{2})\ket{\psi_m}\\
&=\frac{1}{\sqrt{2^{K}}}(\ket{\psi_m}+\sum_{\boldsymbol{x}=1}^{K}T(\boldsymbol{x})\ket{\psi_m}\\
&+\sum_{\boldsymbol{x}<\boldsymbol{x}'}^{K}T(\boldsymbol{x})T(\boldsymbol{x}')\ket{\psi_m}+\cdots+\prod_{\boldsymbol{x}=1}^{K}T(\boldsymbol{x})\ket{\psi_m}).
\end{split}
\end{align} 
This form shows that $\ket*{\tilde{\varphi}_n}$ is normalized. Then, since different $\ket*{\tilde{\varphi}_n}$ and $\ket{\tilde{\varphi}_{n'}}$ are by definition orthogonal, we have $\braket{\tilde{\varphi}_n}{\tilde{\varphi}_{n'}}=\delta_{nn'}$, i.e., the $\{\ket*{\tilde{\varphi}_n}\}$ is orthonormal as desired.

The above arguments and Eq.~\ref{cbs1} reveal how the assembled gates represent the structural properties of the $\ket*{\tilde{\varphi}_n}$ states. With the pictorial illustration in Fig~\ref{9c}, they also show that the structure of $\ket*{\tilde{\varphi}_n}$ can be simply identified as an entanglement patterns, i.e., it is an equal-weight sum of all the $\ket{\psi_m}$ states characterized by such an entanglement pattern. In this way, upon the demonstration of HQEC characteristics in the following sections, our construction will explicitly illustrate how the logical/bulk degrees of freedom $\{\boldsymbol{\beta}_{\boldsymbol{x}}\}_n$ are emergent from the entanglement patterns of physical/boundary qudits, as described in Sec.~\ref{cencoding}. Now, we can conclude these argument into the following theorem on the encoding isometry.
\begin{theorem}[\textbf{Encoding isometry}]\label{cp1}
$\mathcal{E}\xrightarrow{R}\mathcal{H}$, as given in Eq.~\ref{encoding1}, is an isometry, i.e. $R^+R=\mathds{1}$. And the code subspace $\mathcal{H}_{\mathrm{code}}=P_{\mathrm{code}}\mathcal{H}$ with $P_{\mathrm{code}}=RR^+$ is spanned by the orthonormal states $\{\ket*{\tilde{\varphi}_n}\}$ as described in Eq.~\ref{cbs1}. Furthermore, the projection operator $P_{\mathrm{code}}$ of the code subspace has the form
\begin{equation}\label{p1}
P_{\mathrm{code}}=\big[\prod_{\boldsymbol{x}=\boldsymbol{1}}^{K}(\frac{\mathds{1}+T(\boldsymbol{x})}{2})\big]P_0.
\end{equation} 
\end{theorem}

Note that the last equation can be easily derived according to Eq.~\ref{cbs1} and the obvious fact that $P_0$ commutes with $\prod_{\boldsymbol{x}=\boldsymbol{1}}^{K}(\frac{\mathds{1}+T(\boldsymbol{x})}{2})$~\footnote{The commutativity guarantees that $P_{\mathrm{code}}$ is a projection operator. According to Eq.~\ref{cbs1}, we have $P_{\mathrm{code}}\ket{\psi_m}=\frac{1}{\sqrt{2^{K}}}\ket*{\tilde{\varphi}_n}\in\mathcal{H}_{\mathrm{code}}$ and $P_{\mathrm{code}}\ket*{\tilde{\varphi}_n}=\sqrt{2^{K}}P_{\mathrm{code}}P_{\mathrm{code}}\ket{\psi_m}=\sqrt{2^{K}}P_{\mathrm{code}}\ket{\psi_m}=\ket*{\tilde{\varphi}_n}$.}. It is also easy to see that the gates commute with $P_{\mathrm{code}}$, i.e.
\begin{equation}
[P_{\mathrm{code}},T_{ii'}]=0,\quad [P_{\mathrm{code}},T(\boldsymbol{x})]=0.
\end{equation}

\begin{center}
\begin{figure}[ht]
\centering
    \includegraphics[width=8.5cm]{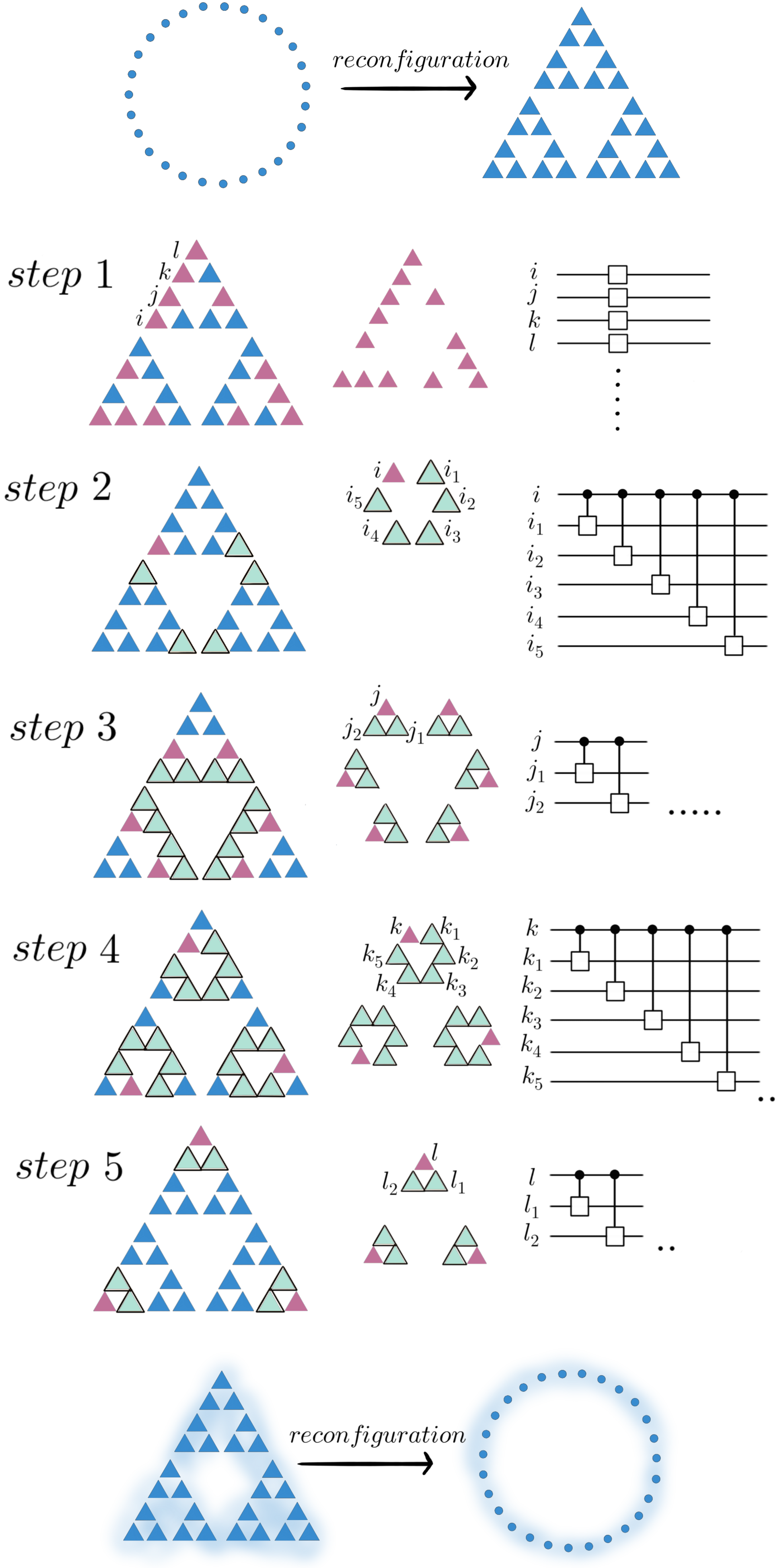}   
\caption{The purple-colored play the role of control qudits while the highlighted play the role of target qudits. The specification of each gate, i.e., the specific operator, is ``incomplete'', since in ququart system there are much more choices than in the qubit case~\cite{wang2020a}. However, as pointed out in App.~\ref{poscalable}, in the controlled gates, the operators for the target qudits is the $S^{\sigma}$ operators determined by the assembled $T(\boldsymbol{x})$ gates.}
\label{fig10}
\end{figure}
\end{center}

\subsection{Scalability and quantum circuit simulation}\label{scalable}
Before we unfold the HQEC characteristics in our model, we point out that Eq.~\ref{cbs1} connects two equivalent characterizations of the $\ket{\tilde{\varphi}_n}$ states, i.e, the algebraic form in terms of the assemble gates and the pictorial representation of the entanglement patterns. Both the two characterizations underlie scalability of the structures of the $\ket{\tilde{\varphi}_n}$s and hence the encoding, which manifests in two aspects.

First, the assembled $T(\boldsymbol{x})$ gates, in either the standard geometry or the alternative geometry, have fixed forms and sizes so that increasing system size simply adds new $(1/\sqrt{2})(\mathds{1}+T(\boldsymbol{x}))$ terms in Eq.~\ref{cbs1}. And the entanglement patterns, i.e. the equal-weight-sum form of $\ket{\tilde{\varphi}_n}$ in terms of the qudit-product-states with pictorial representation, provide convenience to analytically study the model in arbitrary system size. In terms of these properties, studies of the model for different purposes with respect to the boundary states, e.g., studying the dynamic properties, investigating the possible parent (nonlocal) Hamiltonians, developing approximate holographic code and exploring advanced topics in quantum gravity, should be direct and transparent.

Second, based on Eq.~\ref{cbs1}, one can develop scalable and succinct schemes for potential quantum circuit simulation of the states on near-term quantum devices. Here, instead of detailed discussion on comprehensive simulations of the model which requires a considerable amount of space of the texts, we illustrate a possible scheme by sketching the steps for quantum-circuit simulation of a concrete state $\ket{\tilde{\varphi}_{n_0}}=\sqrt{2^{K}}\prod_{\boldsymbol{x}=1}^{K}(\frac{\mathds{1}+T(\boldsymbol{x})}{2})\ket{0\cdots0\cdots0}$. Our description is inspired by the recently developing technique in quantum processor and hybrid analogue-digital quantum simulation, i.e., dynamically reconfigurable arrays of cold neutral atoms which can be ``shuttled'' between different geometric arrangements with optical tweezers for entangling distant qubits without complex quantum gates on neighboring qubits~\cite{bluvstein2022,bluvstein2024,xu2024}.

Intriguingly, this technique of reconfiguration which coherently transports entangled cold atoms in experiments aligns with the geometric concept utilized in our approach, i.e., the rearrangement which mathematically maps qudits in different geometric arrangements. According to this alignment, when initializing the physical qudits in the state $\ket{0\cdots0\cdots0}$, it is not necessary to arrange the qudits in the standard geometry or the alternative geometry. But in our sketch as illustrated in Fig.~\ref{fig10} we imagine the alternative geometric setting in order to conveniently describe the circuit structure associated to each assembled gate (see Fig.~\ref{9c}).

Fig.~\ref{fig10} describes the quantum-circuit in steps, where the number of steps agrees with the number of layers in the circuit which scales linearly on $N_0=N^{1/h}$. Recall that $1/h$ is exactly the universal scaling component in uberholography and $N_0$ is the linear size of the alternative geometry which scales sublinearly on the total number $N$ of physical qudits. As shown in Fig.~\ref{fig10}, the steps correspond to a way we group the $T(\boldsymbol{x})$ gates. The underlying fact is that the $(1/\sqrt{2})(\mathds{1}+T(\boldsymbol{x}))$ term in Eq.~\ref{cbs1} simply maps one qudit-product-state to an equal-weight sum of two qudit-product-states, and the effect can be simulated with two-qudit controlled gates. The $T(\boldsymbol{x})$ gates in each step are disjoint but have overlap with the gates in the preceding and the succeeding steps. For each step, we label certain qudits (the purple in Fig.~\ref{fig10}) in the supports of the assembled gates so that they are only in the overlap with the succeeding step.

Then, Step 1 simply applies single-qudit gates to these control qudits, and the rest steps apply the controlled gates sequentially to realize the effect of the $(1/\sqrt{2})(\mathds{1}+T(\boldsymbol{x}))$ term in Eq.~\ref{cbs1}. And finally, we imagine that the reconfiguration technique can transport the qudits back to the standard geometric setting. The detail of the circuit is discussed in App.~\ref{poscalable}. Importantly, the illustration in Fig.~\ref{fig10} can be straightforwardly scaled up to arbitrarily large size.

\section{Complementary recovery}\label{crsec}

In the preceding section, we have explicitly established an encoding isometry $\mathcal{E}\xrightarrow{R}\mathcal{H}$ with $\mathcal{E}=\mathfrak{e}_{\boldsymbol{1}}\otimes\mathfrak{e}_{\boldsymbol{2}}\otimes\mathfrak{e}_{\boldsymbol{3}}\otimes\cdots=(\mathbb{C}^4)^{\otimes K}$, $\mathcal{H}=\mathfrak{h}_1\otimes\mathfrak{h}_2\otimes\mathfrak{h}_3\otimes\cdots=(\mathbb{C}^4)^{\otimes N}$. Before we unfold the properties of bulk operator reconstruction and show how they can manifest in geometry, we demonstrate the condition of complementary recovery (see Par.~\ref{cplmrc}) which is fundamental to other characteristics and is independent on the geometry. As will be shown below, this condition can be derived from the structural properties of the code states $\ket*{\tilde{\varphi}_n}$s as represented by properties of the gates. In our arguments, we will also describe the basic form of operators supported on specific subregions of the boundary, which recover logical operators. These arguments will lay a basis for the proof in the next section.

We consider an arbitrary nontrivial bipartition of the physical qudits $A\overline{A}$, and the corresponding decomposition $\mathcal{H}=\mathcal{H}_A\otimes\mathcal{H}_{\overline{A}}$. According to the criterion stated in Prop.~\ref{cr2}, we need to prove that for any operator $O_A$ supported on the physical qudits in $A$, there exists an operator $Q_A$ that is commutative with $P_{\mathrm{code}}$ and also supported on $A$ so that we have $P_{\mathrm{code}}Q_AP_{\mathrm{code}}=P_{\mathrm{code}}O_AP_{\mathrm{code}}$.

Indeed, in the proof, it suffices to consider only a basis of operators spanning $\mathbf{L}(\mathcal{H}_A)$. That is, if the projection of every basis operator can be replaced by some $P_{\mathrm{code}}Q_AP_{\mathrm{code}}$ with $[Q_A,P_{\mathrm{code}}]=0$, then the projection of arbitrary operator supported on $A$, as a linear combination of the basis operators, will possess the same property. Recall that the operator basis for a single qudit can be represented as $S_i^\sigma\dyad{\alpha_i}$ for $\sigma=0,1,2,3$ and $\ket{\alpha_i}=\ket{0},\ket{1},\ket{2},\ket{3}$. Then, an operator of the basis spanning $\mathbf{L}(\mathcal{H}_A)$ has the form
\begin{align}
\begin{split}
&(S_{i_1}^{\sigma_{i_1}}\dyad{\alpha_{i_1}})\otimes(S_{i_2}^{\sigma_{i_2}}\dyad{\alpha_{i_2}})\otimes\cdots\\
=&(S_{i_1}^{\sigma_{i_1}}\otimes S_{i_2}^{\sigma_{i_2}}\otimes\cdots)(\dyad{\alpha_{i_1}}\otimes\dyad{\alpha_{i_2}}\otimes\cdots)\\
=&(\otimes_{i\in A}S_i^{\sigma_i})(\otimes_{i\in A}\dyad{\alpha_i}).
\end{split}
\end{align}

It is easy to show that this operator does not commute with $P_{\mathrm{code}}$. That is because the action of $(\dyad{\alpha_{i_1}}\otimes\dyad{\alpha_{i_2}}\otimes\cdots)$ on $\ket*{\tilde{\varphi}_n}$ only ``picks up'' the expanding $\ket{\psi_{m}}$ states that matches the sub-configuration $\alpha_{i_1},\alpha_{i_2},\ldots$. And there are some other $\ket{\psi_{m'}}$ states which are also in the expansion but does not match the sub-configuration. Hence, the action results in a state not in $\mathcal{H}_{\mathrm{code}}$.

Before we show the existence of the desired $Q_{A}$ for the projection of this basis operator on $\mathcal{H}_{\mathrm{code}}$, i.e. $P_{\mathrm{code}}[(\otimes_{i\in A}S_i^{\sigma_i})(\otimes_{i\in A}\dyad{\alpha_i})\otimes\mathds{1}_{\overline{A}}]P_{\mathrm{code}}$, or abbreviated as $P_{\mathrm{code}}(\otimes_{i\in A}S_i^{\sigma_i})(\otimes_{i\in A}\dyad{\alpha_i})P_{\mathrm{code}}$, we simplify the problem and show that it suffices to consider the desired $Q_A$ for $P_{\mathrm{code}}(\otimes_{i\in A}\dyad{\alpha_i})P_{\mathrm{code}}$.



Indeed, we recall that $P_{\mathrm{code}}=\big[\prod_{\boldsymbol{x}=1}^{K}(\frac{\mathds{1}+T(\boldsymbol{x})}{2})\big]P_0$ in which the bracket commutes with $P_0$, and note that $\otimes_{i\in A}\dyad{\alpha_i}$ also commutes with $P_0$~\footnote{We write $\otimes_i\dyad{\alpha_i}$ as $\dyad{\alpha_{i_1}\alpha_{i_2}\cdots}$. Then the state $\ket{\alpha_{i_1}\alpha_{i_2}\cdots}$ on $A$ either matches some $\ket{\psi_m}$ on which $\otimes_i\dyad{\alpha_i}$ acts as the identity operator, or it matches no states in $\mathcal{H}_0$ and hence acts as $0$. It means that the action of $\otimes_i\dyad{\alpha_i}$ (and also $(\otimes_i\dyad{\alpha_i})^+=\otimes_i\dyad{\alpha_i}$) leaves $\mathcal{H}_0$ invariant. Hence $\otimes_i\dyad{\alpha_i}$ commutes with $P_0$.}. Then, we can write
\begin{align*}
\begin{split}
&P_{\mathrm{code}}(\otimes_{i\in A}S_i^{\sigma_i})(\otimes_{i\in A}\dyad{\alpha_i})P_{\mathrm{code}}\\
&=\prod_{\boldsymbol{x}=1}^{K}(\frac{\mathds{1}+T(\boldsymbol{x})}{2})P_0(\otimes_{i\in A}S_i^{\sigma_i})P_0(\otimes_{i\in A}\dyad{\alpha_i})\\
&\times\prod_{\boldsymbol{x}=1}^{K}(\frac{\mathds{1}+T(\boldsymbol{x})}{2}).
\end{split}
\end{align*}
This express can simplified according to the following lemma on the term $P_0(\otimes_{i\in A}S_i^{\sigma_i})P_0$. 
\begin{lemma}\label{stot1}
In the setting of our code, $P_0(\otimes_{i\in A}S_i^{\sigma_i})P_0\ne 0$ if and only if $\otimes_{i\in A}S_i^{\sigma_i}$ equals the identity operator $\mathds{1}$ or a product of the $T_{ii'}$ gates, i.e., $\otimes_{i\in A}S_i^{\sigma_i}=(T_{ii'}T_{jj'}\cdots)$.
\end{lemma}

The proof is given in App.~\ref{postot1}. This lemma says that only the basis operators of the form $(T_{ii'}T_{jj'}\cdots)(\otimes_{i\in A}\dyad{\alpha_i})$ gives rise to nonzero projections on $\mathcal{H}_{\mathrm{code}}$. Furthermore, since the gates commute with $P_{\mathrm{code}}$, the nonzero projection has the form
\begin{multline*}
P_{\mathrm{code}}(T_{ii'}T_{jj'}\cdots)(\otimes_{i\in A}\dyad{\alpha_i})P_{\mathrm{code}}\\
=(T_{ii'}T_{jj'}\cdots)P_{\mathrm{code}}(\otimes_{i\in A}\dyad{\alpha_i})P_{\mathrm{code}}.
\end{multline*}
It follows that we only need to find the $Q_A$ for $P_{\mathrm{code}}(\otimes_{i\in A}\dyad{\alpha_i})P_{\mathrm{code}}$, and then, $(T_{ii'}T_{jj'}\cdots)Q_A$ is the desired operator for the general basis operator.


The $Q_A$ operator for $P_{\mathrm{code}}(\otimes_{i\in A}\dyad{\alpha_i})P_{\mathrm{code}}$ can be constructed as follows.
\begin{align}\label{qoa1}
\begin{split}
Q_A&=\frac{1}{2^{K}}\bigg[\otimes_i\dyad{\alpha_i}+\sum_{\boldsymbol{x}=1}^{K}T(\boldsymbol{x})(\otimes_i\dyad{\alpha_i})T(\boldsymbol{x})\\
&+\sum_{\boldsymbol{x}<\boldsymbol{x}'}^{K}T(\boldsymbol{x})T(\boldsymbol{x}')(\otimes_i\dyad{\alpha_i})T(\boldsymbol{x})T(\boldsymbol{x}')\\
&+\cdots+(\prod_{\boldsymbol{x}=1}^{K}T(\boldsymbol{x}))(\otimes_i\dyad{\alpha_i})(\prod_{\boldsymbol{x}=1}^{K}T(\boldsymbol{x}))\bigg].
\end{split}
\end{align}
In comparison with Eq.~\ref{cbs1} for the expansion of $\ket*{\tilde{\varphi}_n}$, the above expression of $Q_A$ simply sums up all the $(T(\boldsymbol{x})T(\boldsymbol{x}')\cdots)(\otimes_i\dyad{\alpha_i})(T(\boldsymbol{x})T(\boldsymbol{x}')\cdots)$ operators, and the summation is exactly indexed by the terms in the expansion given by Eq.~\ref{exp1}. Through the following proposition, we show that this operator exactly satisfies the desired condition for $Q_A$.

\begin{proposition}\label{qa}
In the setting for our code, for any $\otimes_{i\in A}\dyad{\alpha_i}=\dyad{\alpha_{i_1}\alpha_{i_2}\cdots}$ supported on $A$, the operator $Q_A$ given by Eq.~\ref{qoa1} satisfies following conditions:

(1) $Q_A$ is also supported on $A$. 

(2) $[P_{\mathrm{code}},Q_A]=0$. For any $\ket*{\tilde{\varphi}_n}$, $Q_A\ket*{\tilde{\varphi}_n}\ne0$ if and only if some $\ket{\psi_m}$ in the expansion of $\ket*{\tilde{\varphi}_n}$ matches the sub-configuration $\alpha_{i_1},\alpha_{i_2},\ldots$ on $A$. 

(3) In the nonzero case, $Q_A\ket*{\tilde{\varphi}_n}=c_A\ket*{\tilde{\varphi}_n}$ for each $\ket*{\tilde{\varphi}_n}$, where the positive constnat $c_A$ is only determined by the geometry of the support $A$ (independent on $\ket*{\tilde{\varphi}_n}$ and the content of the configuration $\alpha_{i_1},\alpha_{i_2},\ldots$).

(4) $P_{\mathrm{code}}Q_AP_{\mathrm{code}}=P_{\mathrm{code}}(\otimes_{i\in A}\dyad{\alpha_i})P_{\mathrm{code}}$.
\end{proposition}

The proof of the this proposition is given in App.~\ref{poqa}. According to this proposition and the above discussion, we can conclude as follows. 
\begin{theorem}[Complementary recovery]\label{crthm}
Our model satisifies the condition of complementary recovery as described in Par.~\ref{cplmrc}, in which the boundary bipartition $A\overline{A}$ is arbitrary. 
\end{theorem}

According to Prop.~\ref{qa}, we can also prove the following corollary which describes the conditions under which two $Q_A$ and $Q'_A$ operators, as defined in Eq.~\ref{qoa1}, give the same projections within the code subspace $\mathcal{H}_{\mathrm{code}}$. The result will be important to the arguments in the next section.
\begin{corollary}\label{qaa}
In the setting of our code, we consider two sub-configurations $\alpha_{i_1},\alpha_{i_2},\ldots$ and $\alpha'_{i_1},\alpha'_{i_2},\ldots$ both supported on a boundary subregion $A$. Let $Q_A$ and $Q'_A$ be the operators defined as in Eq.~\ref{qoa1} for $\otimes_{i\in A}\dyad{\alpha_i}$ and $\otimes_{i\in A}\dyad{\alpha'_i}$ respectively, and assume that $P_{\mathrm{code}}Q_AP_{\mathrm{code}}\ne0$, $P_{\mathrm{code}}Q'_AP_{\mathrm{code}}\ne0$. Then the following three conditions are equivalent:

(1) $Q_A=Q'_A$.

(2) Some $\ket{\psi_m}$ and $\ket{\psi_{m'}}$ that both participate the expansion of the same $\ket*{\tilde{\varphi}_n}$ match the two sub-configurations respectively.

(3) The operator $\otimes_{i\in A}\dyad{\alpha'_i}$ is equal to $(T(\boldsymbol{x})T(\boldsymbol{x}')\cdots)(\otimes_{i\in A}\dyad{\alpha_i})(T(\boldsymbol{x})T(\boldsymbol{x}')\cdots)$ for some product of the assembled gates.
\end{corollary}
The proof is given in App.~\ref{poqaa}.

\section{Reconstruction of a single bulk qudit}\label{minisec}
In this section, we demonstrate the HQEC characteristics regarding the reconstruction of a single bulk qudit in our code. The central theme is to prove Characteristic 1 on the connected code distance, and to describe all the possible minimal boundary subregions for reconstructing the subalgebra $\mathcal{M}(\boldsymbol{x})=R(\cdots\otimes\mathbb{C}\mathds{1}_{\mathfrak{e}_{\boldsymbol{x}'}}\otimes\mathbf{L}(\mathfrak{e}_{\boldsymbol{x}})\otimes\mathbb{C}\mathds{1}_{\mathfrak{e}_{\boldsymbol{x}''}}\otimes\cdots)R^+$ in terms of the distinguishability properties of the entanglement patterns of the code states (see Sec.~\ref{pattern1}). This section serves as a bridge between the construction of the code in the preceding section and the demonstration of the HQEC characteristics regarding subregion duality in the next section. The results on reconstructing bulk local operators will play a crucial role in describing subregion duality in the OAQEC formalism.

We will keep employing the rearrangement of physical qudits (see Fig.~\ref{fig4}, \ref{fig5} and \ref{fig6}) between the standard geometric setting and the alternative geometric setting. In the alternative geometry we follow the construction of the code and systematically describe the minimal recoveries. It will be clear that the entanglement patterns described in the preceding section satisfy the conditions predicted in Sec.~\ref{cencoding}. In rearranging to the standard geometric setting of the code, we show how the geometric properties of the disconnected minimal recoveries give rise to the proof of Characteristic 1. Our arguments will elucidate how the alternative geometry and the rearrangement underlie the expected HQEC characteristics regarding the reconstruction of a single bulk qudit, as proposed in Sec.~\ref{demys}.

Note that according to how uberholography was firstly proposed in Ref.~\cite{pastawski2017}, one might expect the minimal recoveries to be specified based on the entanglement wedges. Indeed, thanks to the explicit structure of the boundary code states shown in the previous section, we can directly investigate and accurately describe the minimal recoveries, which not only gives an exact demonstration of the universal scaling component $1/h$, but also underlies the construction of entanglement wedges for disconnected boundary bipartition. In this section, the geometric manifestation of the minimal recoveries will be only on the boundary. The full geometric manifestation including the structure of entanglement wedges in the bulk will be presented in the next section when studying the subregion duality.


\subsection{Connected code distance and reconstruction of $\mathcal{M}(\pmb{x})$ generators}\label{rcgmx}
As of primary importance in demonstrating the HQEC characterisitcs, we need to prove that the connected distance $\mathrm{d_c}({\boldsymbol{x}})$ of a bulk qudit $\boldsymbol{x}$ scales linearly on the total number $N$ of physical qudits, and the values for different bulk qudits are consistent with their bulk radial distance to the center (see Par.~\ref{chac1}). In principle, for each bulk qudit $\boldsymbol{x}$, we might need to explore the subregion $\overline{A}$ of correctability to ensure the the maximal size. However, according to Lemma~\ref{oaqec}, for given bipartition $A\overline{A}$ of the boundary, $\overline{A}$ is correctable with respect to $\mathcal{M}(\boldsymbol{x})$ if and only if $\mathcal{M}(\boldsymbol{x})$ can be reconstructed on $A$. Hence, it is sufficient to explore the subregion $A$ of reconstructability to confirm the minimal size. Indeed, we show that based on certain typical examples of boundary recoveries of $\mathcal{M}(\boldsymbol{x})$, it suffices for us to prove the following theorem.

\begin{theorem}\label{cdistance}
In our code, the connected code distance $\mathrm{d_c}({\boldsymbol{x}})$ for the bulk qudits takes the values $\frac{N}{9}+1,\frac{N}{27}+1,\frac{N}{81}+1,\ldots$, along the radial direct from the center towards the boundary.
\end{theorem}

The proof takes advantage of the rearrangement between the alternative geometry and the standard 1D geometry of the physical qudits. As shown in Fig.~\ref{fig11}, with respect to the central bulk qudit $\boldsymbol{x}$, through the rearrangement the alternative geometry can classify all possible (to-be-erased) connected boundary subregion $\overline{A}$ of the size $\mathrm{d_c}({\boldsymbol{x}})-1$ (see above theorem) into two types: (1) $\overline{A}$ lies within one big triangular block (see Fig.~\ref{11a} and \ref{11b}); (2) $\overline{A}$ crosses two big triangular blocks (see Fig.~\ref{11c} and \ref{11d}). An important feature that will be utilized in our proof is that in either of the two types, $\overline{A}$ only covers at most one out of the six corner qudits $i_1,i_2,i_3,i_4,i_5,i_6$ that supports the $T(\boldsymbol{x})$ gate (see Fig.~\ref{9c}), and the rest five qudits can be covered by connected open paths in the alternative geometry that avoids $\overline{A}$, which will be illustrated later.

\onecolumngrid
\begin{center}
\begin{figure}[ht]
\centering
    \includegraphics[width=14cm]{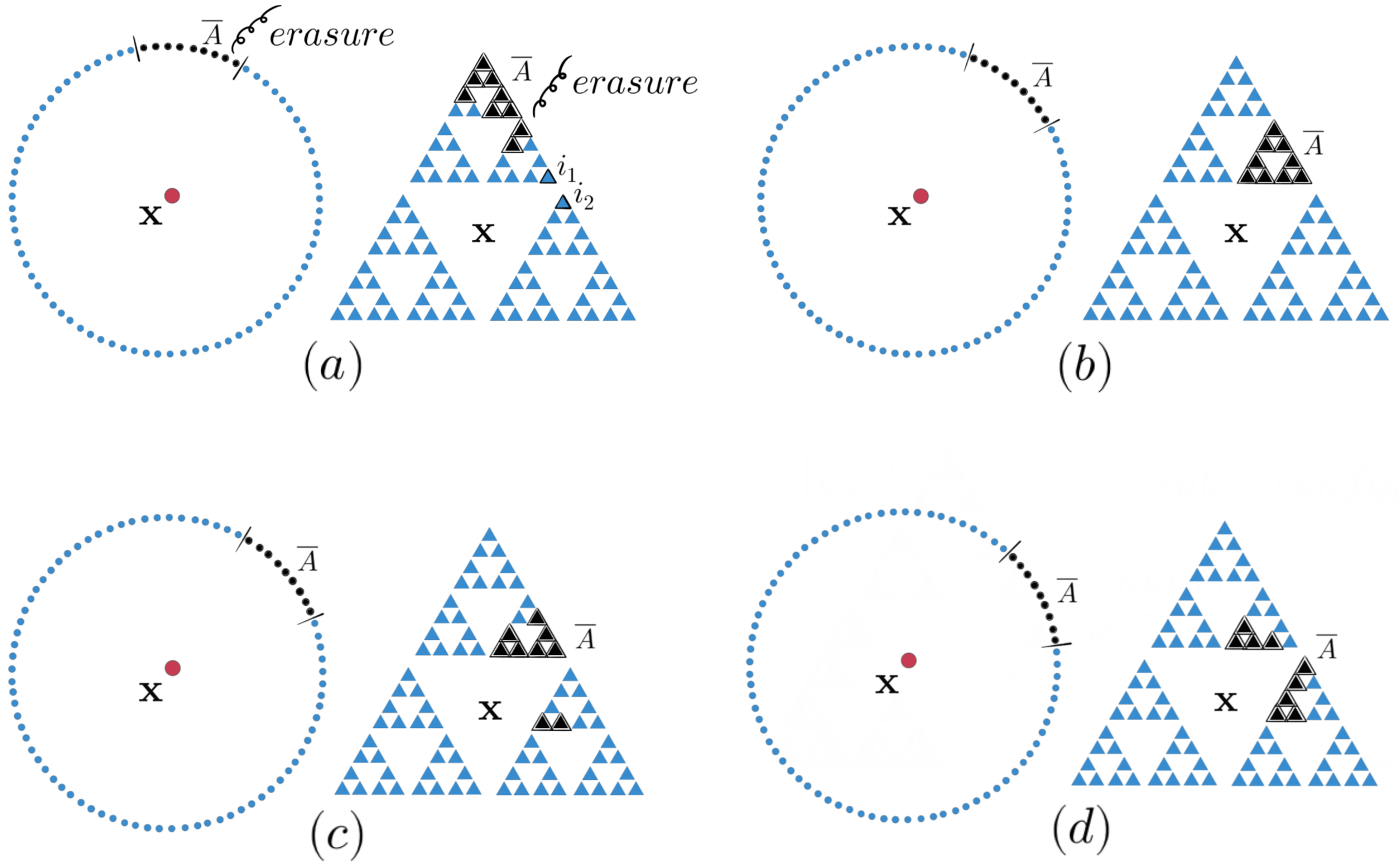}   
\phantomsubfloat{\label{11a}}\phantomsubfloat{\label{11b}}
\phantomsubfloat{\label{11c}}\phantomsubfloat{\label{11d}}
\caption{Through the rearrangement, the alternative geometry classifies the connected boundary subregions $\overline{A}$ with size $\mathrm{d_c}({\boldsymbol{x}})-1$ for the center bulk qudit $\boldsymbol{x}$ (see Thm.~\ref{cdistance}) into two types. (a), (b) Type 1, in which $\overline{A}$ lies within one big triangular block. (c), (d) Type 2, in which $\overline{A}$ crosses two blocks. In either of the two types, only at most one out of the six corner qudits $i_1,i_2,i_3,i_4,i_5,i_6$ supporting the $T(\boldsymbol{x})$ gate is covered by $\overline{A}$, and the rest five can be connected by open paths.}
\label{fig11}
\end{figure}
\end{center}
\twocolumngrid

Then to prove Thm.~\ref{cdistance} for the central bulk qudit, we simply need to show that \emph{i)} $\mathcal{M}(\boldsymbol{x})$ has recoveries on the subregion $A$ for any $\overline{A}$ of the two types as illustrated in Fig.~\ref{fig11}; \emph{ii)} extension by one physical qudit of such a $\overline{A}$ subregion while keeping the connectedness can disable the recovery on $A$. This will be a direct result from the following discussion on the reconstruction of the generator operators in $\mathcal{M}(\boldsymbol{x})$. Note that for other bulk qudits closer to the boundary, due to the self-similarity of the alternative geometry, the situation is similar and we can replicate the proof for the central bulk qudit.

To initialize the arguments, we recall that operators in $\mathcal{M}(\boldsymbol{x})$ simply have the form $\widetilde{O}=R(\cdots\otimes\mathds{1}_{\mathfrak{e}_{\boldsymbol{x}'}}\otimes\widetilde{\boldsymbol{O}}_{\boldsymbol{x}}\otimes\mathds{1}_{\mathfrak{e}_{\boldsymbol{x}''}}\otimes\cdots)R^+=R(\cdots\otimes\widetilde{\boldsymbol{O}}_{\boldsymbol{x}}\otimes\cdots)R^+$ with $\widetilde{\boldsymbol{O}}_{\boldsymbol{x}}\in\mathbf{L}(\mathfrak{e}_{\boldsymbol{x}})$. Since an emergent bulk qudit is also a ququart $\mathbb{C}^4$, i.e., $\mathbf{L}(\mathfrak{e}_{\boldsymbol{x}})=\mathbf{L}(\mathbb{C}^4)$, then, similar to the case of a physical ququart (see discussion on Eq.~\ref{ds}), $\mathbf{L}(\mathfrak{e}_{\boldsymbol{x}})$ is generated by two types of operators
\begin{align}\label{dts}
\begin{split}
&\dyad{\boldsymbol{\beta}},\quad \boldsymbol{\beta}=\boldsymbol{0},\boldsymbol{1},\boldsymbol{2},\boldsymbol{3};\\
&\widetilde{\boldsymbol{S}}^{\boldsymbol{0}}=\mathds{1},\\
&\widetilde{\boldsymbol{S}}^{\boldsymbol{1}}=\dyad{\boldsymbol0}{\boldsymbol1}+\dyad{\boldsymbol1}{\boldsymbol0}+\dyad{\boldsymbol2}{\boldsymbol3}+\dyad{\boldsymbol3}{\boldsymbol2},\\
&\widetilde{\boldsymbol{S}}^{\boldsymbol{2}}=\dyad{\boldsymbol0}{\boldsymbol2}+\dyad{\boldsymbol2}{\boldsymbol0}+\dyad{\boldsymbol1}{\boldsymbol3}+\dyad{\boldsymbol3}{\boldsymbol1},\\
&\widetilde{\boldsymbol{S}}^{\boldsymbol{3}}=\dyad{\boldsymbol0}{\boldsymbol3}+\dyad{\boldsymbol3}{\boldsymbol0}+\dyad{\boldsymbol1}{\boldsymbol2}+\dyad{\boldsymbol2}{\boldsymbol1}.
\end{split}
\end{align}
And it is easy to see the noncommutativity between these generators, i.e.
\begin{align}\label{cts}
\begin{split}
&[\widetilde{\boldsymbol{S}}^{\boldsymbol{\sigma}},\dyad{\boldsymbol{\beta}}]=\dyad{\boldsymbol{\beta}'}{\boldsymbol{\beta}}-\dyad{\boldsymbol{\beta}}{\boldsymbol{\beta}'}\ne0,\\
&\boldsymbol{\sigma}=\boldsymbol{1},\boldsymbol{2},\boldsymbol{3}.
\end{split}
\end{align}
The commutativity only holds between certain pairs of $(\dyad{\boldsymbol{\beta}}+\dyad{\boldsymbol{\beta}'})$ and $\widetilde{\boldsymbol{S}}^{\boldsymbol{\sigma}}$, i.e.,
\begin{align}\label{ctss}
\begin{split}
&[\widetilde{\boldsymbol{S}}^{\boldsymbol{1}},(\dyad{\boldsymbol{0}}+\dyad{\boldsymbol{1}})]=[\widetilde{\boldsymbol{S}}^{\boldsymbol{1}},(\dyad{\boldsymbol{2}}+\dyad{\boldsymbol{3}})]=0,\\
&[\widetilde{\boldsymbol{S}}^{\boldsymbol{2}},(\dyad{\boldsymbol{0}}+\dyad{\boldsymbol{2}})]=[\widetilde{\boldsymbol{S}}^{\boldsymbol{2}},(\dyad{\boldsymbol{1}}+\dyad{\boldsymbol{3}})]=0,\\
&[\widetilde{\boldsymbol{S}}^{\boldsymbol{3}},(\dyad{\boldsymbol{0}}+\dyad{\boldsymbol{3}})]=[\widetilde{\boldsymbol{S}}^{\boldsymbol{3}},(\dyad{\boldsymbol{2}}+\dyad{\boldsymbol{1}})]=0,\\
&[\widetilde{\boldsymbol{S}}^{\boldsymbol{\sigma}},(\dyad{\boldsymbol{\beta}}+\dyad{\boldsymbol{\beta}'})]\ne0~for~others.
\end{split}
\end{align}



Now, considering the algebra isomorphism $\mathbf{L}(\mathcal{E})\xrightarrow{R\boldsymbol{\cdot} R^+}\mathbf{L}(\mathcal{H}_{\mathrm{code}})$, $\mathcal{M}(\boldsymbol{x})$ is equivalently generated by operators of the form $R(\cdots\otimes\dyad{\boldsymbol{\beta}_{\boldsymbol{x}}}\otimes\cdots)R^+$ and $R(\cdots\otimes\widetilde{\boldsymbol{S}}^{\boldsymbol{\sigma}}_{\boldsymbol{x}}\otimes\cdots)R^+$, and with the same commutation relation. In the following, we study the reconstructions of these generators in $\mathcal{M} (\boldsymbol{x})$, whose combination will underlie the minimal reconstruction of the subalgebra $\mathcal{M}(\boldsymbol{x})$. As will be shown in the next section, the reconstructions of these generators will provide great convenience to explicitly describe the entanglement wedges and the von Neumann algebras in the study of subregion duality.

Note that for simplicity in notation, we use the abbreviation ``reconstruction of $\dyad{\boldsymbol{\beta}_{\boldsymbol{x}}}$'' instead of the complete version ``reconstruction of $R(\cdots\otimes\dyad{\boldsymbol{\beta}_{\boldsymbol{x}}}\otimes\cdots)R^+$''. But as discussed in Sec.~\ref{basic}, in the arguments regarding reconstruction, we formally work on operator $R(\cdots\otimes\dyad{\boldsymbol{\beta}_{\boldsymbol{x}}}\otimes\cdots)R^+$ in $\mathbf{L}(\mathcal{H}_{\mathrm{code}})$.

\subsubsection{Criterion for reconstructing $\dyad{\pmb{\beta}_{\pmb{x}}}$ on a subregion}
For convenience and clarity in the following arguments, we use the notation $\dyad{\overline{\boldsymbol{\beta}}_{\boldsymbol{x}}}$ instead of $\dyad{\boldsymbol{\beta}_{\boldsymbol{x}}}$. That is, by using $\overline{\boldsymbol{\beta}}_{\boldsymbol{x}}$, we mean a specifically identified one out of $\boldsymbol{\beta}_{\boldsymbol{x}}=\boldsymbol0,\boldsymbol1,\boldsymbol2,\boldsymbol3$ for the bulk qudit $\boldsymbol{x}$ in consideration, while we use $\boldsymbol{\beta}_{\boldsymbol{x}'},\boldsymbol{\beta}'_{\boldsymbol{x}'}$ for $\boldsymbol{x}'\ne\boldsymbol{x}$ to denote arbitrariness. The focus of our discussion will be sub-configurations $\alpha_{i_1},\alpha_{i_2},\ldots$s on a boundary subregion $A$. To circumvent verbosity in expression, by ``a sub-configuration matches a entanglement pattern (or a $\ket*{\tilde{\varphi}_n}$)'', or ``a sub-configuration is extracted from a entanglement pattern (or a $\ket*{\tilde{\varphi}_n}$)'', we simply mean that a sub-configuration $\alpha_{i_1},\alpha_{i_2},\ldots$ matches a $\ket{\psi_m}=\ket{\alpha_1\cdots\alpha_{i_1}\alpha_{i_2}\cdots\alpha_N}$ state (see Fig.~\ref{12a}) characterized by an entanglement pattern or participating in the expansion of $\ket*{\tilde{\varphi}_n}$.

Recall how $\cdots\otimes\dyad{\overline{\boldsymbol{\beta}}_{\boldsymbol{x}}}\otimes\cdots$ acts on the bulk qudit-product-states: $(\cdots\otimes\dyad{\overline{\boldsymbol{\beta}}_{\boldsymbol{x}}}\otimes\cdots)\ket{\boldsymbol{\beta}_{\boldsymbol{1}}\cdots\boldsymbol{\beta}_{\boldsymbol{x}}\cdots\boldsymbol{\beta}_{K}}\ne0$ if and only if $\boldsymbol{\beta}_{\boldsymbol{x}}=\overline{\boldsymbol{\beta}}_{\boldsymbol{x}}$. Accordingly, it is easy to show an equivalent condition as a criterion for a boundary operator supported on $A$ to reconstruct $R(\cdots\otimes\dyad{\overline{\boldsymbol{\beta}}_{\boldsymbol{x}}}\otimes\cdots)R^+$ (recall the definition of reconstruction in Sec.~\ref{basic}). Clearly, the condition consists of three aspects: \emph{i)} $[Q_A,P_{\mathrm{code}}]=0$; \emph{ii)} $Q_A\ket*{\tilde{\varphi}_n}\ne0$ if and only if $\ket*{\tilde{\varphi}_n}=R\ket{\boldsymbol{\beta}_{\boldsymbol{1}}\cdots\overline{\boldsymbol{\beta}}_{\boldsymbol{x}}\cdots\boldsymbol{\beta}_{K}}$, irrespective of $\boldsymbol{\beta}_{\boldsymbol{x}'}$ ($\boldsymbol{x}'\ne\boldsymbol{x}$); \emph{iii)} in the nonzero case, we have $Q_A\ket*{\tilde{\varphi}_n}=\ket*{\tilde{\varphi}_n}$.

The above condition says that the desired $Q_A$ exactly ``picks up'' those ``right'' $\ket*{\tilde{\varphi}_n}$ states that correspond to the $\ket{\boldsymbol{\beta}_{\boldsymbol{1}}\cdots\overline{\boldsymbol{\beta}}_{\boldsymbol{x}}\cdots\boldsymbol{\beta}_{K}}$s ($\boldsymbol{\beta}_{\boldsymbol{x}'}$ arbitrary for $\boldsymbol{x}'\ne\boldsymbol{x}$) in the encoding. Based on the meaning of ``picks up'', we can develop a sufficient condition for the reconstruction of $\dyad{\overline{\boldsymbol{\beta}}_{\boldsymbol{x}}}$ on a given boundary subregion $A$, which depends only on the distinguishability of the entanglement patterns on boundary subregions with respect to the emergent bulk degrees of freedom.


To see what operators on a boundary subregion $A$ can ``pick up'' a ``right'' state, we consider a sub-configuration $\alpha_{i_1},\alpha_{i_2},\ldots$ supported on $A$, which is extracted from the entanglement pattern (and hence the $\ket*{\tilde{\varphi}_n}$) corresponding to a specific $\ket{\boldsymbol{\beta}_{\boldsymbol{1}}\cdots\overline{\boldsymbol{\beta}}_{\boldsymbol{x}}\cdots\boldsymbol{\beta}_{K}}$ in the encoding. For $\dyad{\alpha_{i_1}\alpha_{i_2}\cdots}$, Eq.~\ref{qoa1} defines an operator $Q_A(\alpha_{i_1},\alpha_{i_2},\ldots)$ supported on $A$ and commutative with $P_{\mathrm{code}}$. Then according to Prop.~\ref{qa}, since $\alpha_{i_1},\alpha_{i_2},\ldots$ matches the code basis state $\ket*{\tilde{\varphi}_n}$ by definition, we have $Q_A(\alpha_{i_1},\alpha_{i_2},\ldots)\ket*{\tilde{\varphi}_n}=c_A\ket*{\tilde{\varphi}_n}$ with $\ket*{\tilde{\varphi}_n}=R\ket{\boldsymbol{\beta}_{\boldsymbol{1}}\cdots\overline{\boldsymbol{\beta}}_{\boldsymbol{x}}\cdots\boldsymbol{\beta}_{K}}$.

Though this example is trivial, it shows that there exists an operator $Q_A(\alpha_{i_1},\alpha_{i_2},\ldots)$ on $A$, which can ``pick up'' at least one ``right'' state. It also implies that certain specific sub-configuration on $A$ as described above can narrow down entanglement patterns to include (but not limited to) those corresponding to the $\ket{\boldsymbol{\beta}_{\boldsymbol{1}}\cdots\overline{\boldsymbol{\beta}}_{\boldsymbol{x}}\cdots\boldsymbol{\beta}_{K}}$s in the encoding. It follows that if we combine the effects of restriction from all such sub-configurations, then all the ``right'' code basis states can be ``picked up''.

Following this line of thinking, we only need to worry if this ``pick-up'' from the combined effects serves exactly. Indeed, though a sub-configuration $\alpha_{i_1},\alpha_{i_2},\ldots$ is extracted from a entanglement pattern corresponding to some $\ket{\boldsymbol{\beta}_{\boldsymbol{1}}\cdots\overline{\boldsymbol{\beta}}_{\boldsymbol{x}}\cdots\boldsymbol{\beta}_{K}}$, it might also matches some entanglement pattern corresponding to $\ket{\boldsymbol{\beta}'_{\boldsymbol{1}}\cdots\overline{\boldsymbol{\beta}'}_{\boldsymbol{x}}\cdots\boldsymbol{\beta}'_{K}}$ with $\overline{\boldsymbol{\beta}'}_{\boldsymbol{x}}\ne\overline{\boldsymbol{\beta}}_{\boldsymbol{x}}$. For example, when the subregion $A$ is far from the loop surrounding the hole $\boldsymbol{x}$, the change of the configuration on the loop cannot be reflected on $A$, hence, sub-configurations on $A$ can match entanglement patterns corresponding to different $\overline{\boldsymbol{\beta}'''}_{\boldsymbol{x}}\ne\overline{\boldsymbol{\beta}''}_{\boldsymbol{x}}\ne\overline{\boldsymbol{\beta}'}_{\boldsymbol{x}}\ne\overline{\boldsymbol{\beta}}_{\boldsymbol{x}}$, and ``wrong'' states will be ``picked up''.

\onecolumngrid
\begin{center}
\begin{figure}[ht]
\centering
    \includegraphics[width=17cm]{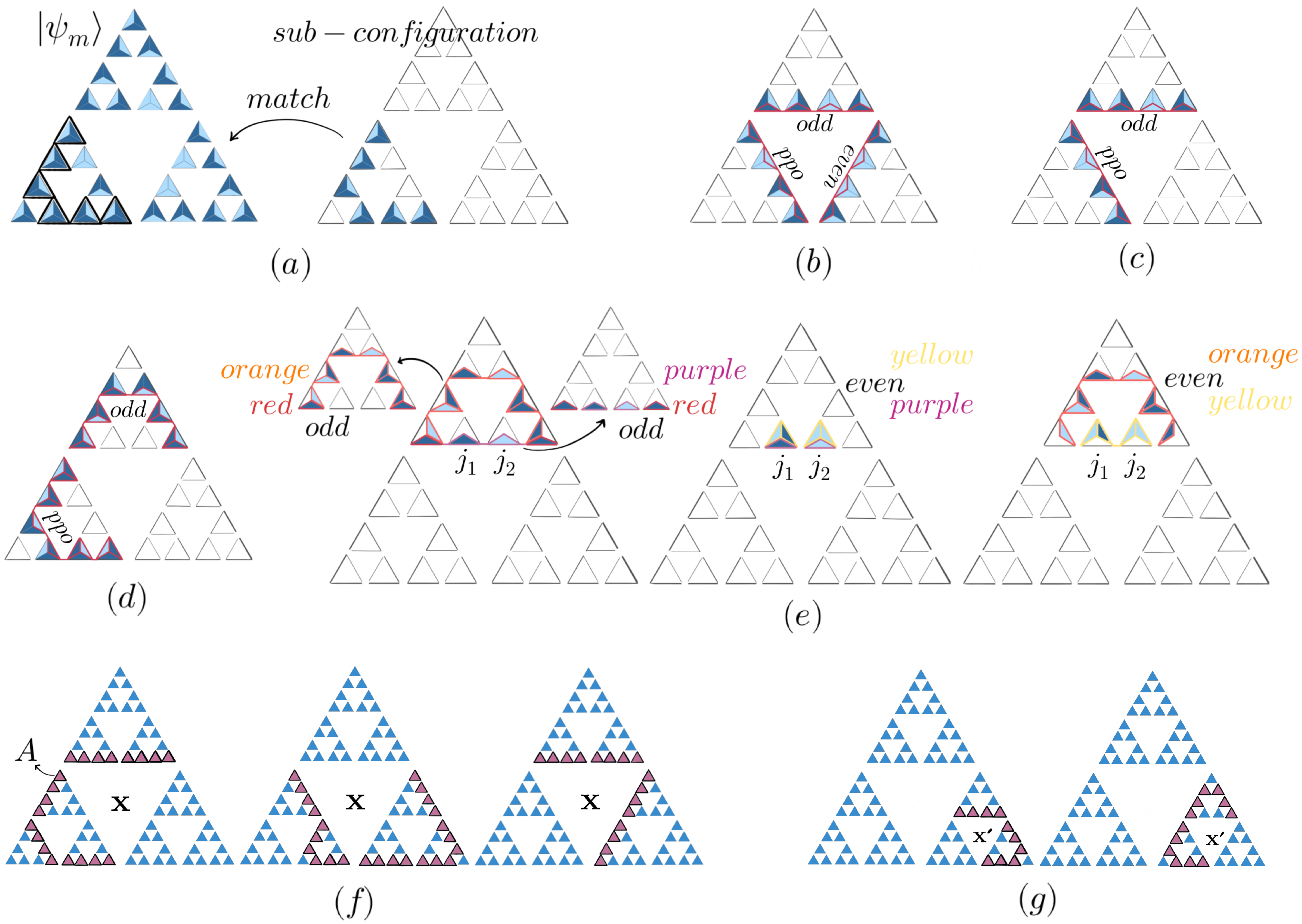}   
\phantomsubfloat{\label{12a}}\phantomsubfloat{\label{12b}}
\phantomsubfloat{\label{12c}}\phantomsubfloat{\label{12d}}
\phantomsubfloat{\label{12e}}\phantomsubfloat{\label{12f}}\phantomsubfloat{\label{12g}}
\caption{(a) The meaning of sub-configuration matching a qudit-product-state. (b) Sub-configuration on the loop (of qudits) surrounding the central hole. The red-highlighted is the loop of sides as illustrated in Fig.~\ref{8a} and \ref{9a}. The parity of the number of dark sides in the red-highlighted for each lateral is specified. (c) Sub-configuration on two out of the three laterals in the loop. The corresponding red-highlighted sides are the same as in the loop. (d) Sub-configuration on two connected paths with the same ends as the laterals. Note that while in the case of laterals each triangle (qudit) contributes one side to the red-highlighted, in the case of the path a triangle can contribute two sides. (e) The red-highlighted only include the two sides in the two ending triangles respectively. The orange-highlighted together with the two red sides are the same as the red-highlighted in one path in (d). The two purple-highlighted together with the two red sides are the same as the red-highlighted in one lateral in (c). The yellow -highlighted are on the triangle $j_1$ and $j_2$. The yellow together with the purple exactly cover two triangles. The yellow together with the orange exactly form three loops of sides. (f) Typical subregions for recoveries of $\dyad{\overline{\boldsymbol{\beta}}_{\boldsymbol{x}}}$. (g) Typical subregions for recoveries of $\dyad{\overline{\boldsymbol{\beta}}_{\boldsymbol{x}'}}$.}
\label{fig12}
\end{figure}
\end{center}
\twocolumngrid

The above heuristic arguments have pointed out that if the ``pick-up'' can avoid all the ``wrong'' basis states, some operator can reconstruct $\dyad{\overline{\boldsymbol{\beta}}_{\boldsymbol{x}}}$ on $A$. To formalize the idea, we specify the collection of all such sub-configurations on $A$, i.e.
 \begin{align}\label{collection}
\begin{split}
\{&\alpha_{i_1},\alpha_{i_2},\ldots~on~A\big|it~matches~some~pattern\\
&~corresponding~to~a~\ket{\boldsymbol{\beta}_{\boldsymbol{1}}\cdots\overline{\boldsymbol{\beta}}_{\boldsymbol{x}}\cdots\boldsymbol{\beta}_{K}},\\
&\boldsymbol{\beta}_{\boldsymbol{x}'}~arbitrary~for~\boldsymbol{x}'\ne\boldsymbol{x}\}.
\end{split}
\end{align}
And we formalize the combined effects into $Q_A$, an equal-weight-sum of all the \emph{distinct} $Q_A(\alpha_{i_1},\alpha_{i_2},\ldots)$ operators for sub-configurations $\alpha_{i_1},\alpha_{i_2},\ldots$s in the collection. That is,
\begin{equation}\label{qoa2}
Q_A=\frac{1}{c_A}\sum_{\alpha_{i_1},\alpha_{i_2},\ldots}Q_A(\alpha_{i_1}.\alpha_{i_2},\ldots)
\end{equation}
Note that according to Cor.~\ref{qaa}, the operators $Q_A(\alpha_{i_1},\alpha_{i_2},\ldots)$ and $Q_A(\alpha'_{i_1},\alpha'_{i_2},\ldots)$ that are defined for different sub-configurations on the same $A$ might be equal, hence the index in the sum does not go through all the sub-configurations in the collection, but only takes those representatives.

Then, we can formalize the above heuristic arguments into the following proposition. And the rigorous proof is given in App.~\ref{poqaaa}.

\begin{proposition}\label{qaaa}
In our code, we consider $\overline{\boldsymbol{\beta}}_{\boldsymbol{x}}$ which can be any of $\boldsymbol0,\boldsymbol1,\boldsymbol2,\boldsymbol3$, and consider a boundary subregion $A$ together with the collection of sub-configurations specified correspondingly by Eq.~\ref{collection}. If the sub-configurations in the collection only match the entanglement patterns correspond to the $\ket{\boldsymbol{\beta}_{\boldsymbol{1}}\cdots\overline{\boldsymbol{\beta}}_{\boldsymbol{x}}\cdots\boldsymbol{\beta}_{K}}$ states ($\boldsymbol{\beta}_{\boldsymbol{x}'}$ arbitrary for $\boldsymbol{x}'\ne\boldsymbol{x}$) in the encoding, then the operator $Q_A$ defined in Eq.~\ref{qoa2} reconstructs $\dyad{\overline{\boldsymbol{\beta}}_{\boldsymbol{x}}}$ on $A$.
\end{proposition}


\subsubsection{Typical boundary reconstructions of $\dyad{\pmb{\beta}_{\pmb{x}}}$}
Now, we elucidate the meaning of the above proposition and the distinguishability of the entanglement patterns in a step-wise manner with concrete examples. We show typical examples of boundary subregions supporting reconstructions of $\dyad{\pmb{\beta}_{\pmb{x}}}$. Our arguments take advantage of the pictorial representation of the $\ket{\psi_m}$ states (see Fig.~\ref{fig7}, \ref{fig8} and \ref{fig9}).

\paragraph*{\textbf{Reconstruction on the loop of qudits}} The simplest example of a reconstruction $Q_A$ is on the loop of qudits surrounding the hole $\boldsymbol{x}$ (see Fig.~\ref{12b}). To check this with Prop.~\ref{qaaa}, we consider the collection of sub-configurations on the loop as defined in Eq.~\ref{collection}. If $\alpha_{i_1},\alpha_{i_2},\dots$ is such a sub-configuration as illustrated in Fig.~\ref{12b}, then it matches some entanglement pattern corresponding to $\ket{\boldsymbol{\beta}_{\boldsymbol{1}}\cdots\overline{\boldsymbol{\beta}}_{\boldsymbol{x}}\cdots\boldsymbol{\beta}_{K}}$ in the encoding. However, by definition of the emergent bulk degrees of freedom (see Sec.~\ref{pattern1}) and the pictorial representation (see Fig.~\ref{9a}), the sub-configuration is precisely what we use to identify a concrete possibility of the parity of numbers of dark sides on the three laterals, which corresponds to $\overline{\boldsymbol{\beta}}_{\boldsymbol{x}}$. In other words, if any other entanglement pattern matches this sub-configuration, according to the concrete possibility of parity, the emergent degrees of freedom must be $\overline{\boldsymbol{\beta}}_{\boldsymbol{x}}$, and hence the entanglement pattern must correspond to some $\ket{\boldsymbol{\beta}'_{\boldsymbol{1}}\cdots\overline{\boldsymbol{\beta}}_{\boldsymbol{x}}\cdots\boldsymbol{\beta}'_{K}}$. Based on this observation, it has been sufficient to apply Prop.~\ref{qaaa}, and we can conclude that there is a reconstruction of $\dyad{\overline{\boldsymbol{\beta}}_{\boldsymbol{x}}}$ on the loop.

Now, it is clear that the existence of a reconstruction of $\dyad{\overline{\boldsymbol{\beta}}_{\boldsymbol{x}}}$ on a given boundary subregion simply relies on whether the entanglement patterns can be exactly distinguished on the subregion with respect to the emergent bulk degrees of freedom. Following this line of thinking, we can identify more boundary subregions for reconstructing $\dyad{\overline{\boldsymbol{\beta}}_{\boldsymbol{x}}}$.

Indeed, the loop surrounding the hole $\boldsymbol{x}$ is redundant to distinguish the ``right'' entanglement patterns, and two laterals of the loop already suffice (see Fig.~\ref{12c}). To see how this is possible, we consider an arbitrary $\ket{\psi_m}$ state. Recall that according to the description in Par.~\ref{psim} and the pictorial representation in Fig.~\ref{8a}, the total number of dark sides on any loop, i.e., the sum of the numbers in the three laterals, is even. Due to this restriction, on the two laterals the four possibilities of parity, i.e., (even, even), (even, odd), (odd, even) and (odd, odd), precisely determines the four possibilities on the loop, i.e., (even, even, even), (even, odd, odd), (odd, even, odd) and (odd, odd, even) (see Fig.~\ref{9a}). In other words, on the $\ket{\psi_m}$ state, identifying one out of the four possibilities on the loop corresponding to the emergent $\overline{\boldsymbol{\beta}}_{\boldsymbol{x}}$ is equivalent to identifying certain possibility out of the four on the two laterals.

\paragraph*{\textbf{Reconstruction on two laterals}} Now, as illustrated in Fig.~\ref{12c}, consider a sub-configuration $\alpha_{i_1},\alpha_{i_2},\dots$ on two specific laterals (of the loop) in the collection as defined in Eq.~\ref{collection}. Since $\alpha_{i_1},\alpha_{i_2},\dots$ matches a $\ket{\psi_m}$ state specified by a entanglement pattern corresponding to $\ket{\boldsymbol{\beta}_{\boldsymbol{1}}\cdots\overline{\boldsymbol{\beta}}_{\boldsymbol{x}}\cdots\boldsymbol{\beta}_{K}}$, then, in the pictorial representation of the $\ket{\psi_m}$ state, the sub-configuration on the loop must shows a concrete possibility, say, (even, odd, odd) as illustrated in Fig.~\ref{12b}, corresponding to the emergent $\overline{\boldsymbol{\beta}}_{\boldsymbol{x}}$. And in terms of the equivalence stated above, the sub-configuration $\alpha_{i_1},\alpha_{i_2},\dots$ on the two laterals must show a unique corresponding possibility (odd, odd) for the parity of the number of dark sides as illustrated in Fig.~\ref{12c}. Reversely, if another arbitrary entanglement pattern matches $\alpha_{i_1},\alpha_{i_2},\dots$ on the two laterals, i.e., when $\alpha_{i_1},\alpha_{i_2},\dots$ can be viewed as extracted from some $\ket{\psi_{m'}}$ specified by this entanglement pattern, then due to the equivalence again, in the pictorial representation of $\ket{\psi_{m'}}$, the sub-configuration on the loop must shows the concrete possibility (even, odd, odd). And hence this arbitrary entanglement pattern must correspond to some $\ket{\boldsymbol{\beta}'_{\boldsymbol{1}}\cdots\overline{\boldsymbol{\beta}}_{\boldsymbol{x}}\cdots\boldsymbol{\beta}'_{K}}$ in the encoding. Then, applying Prop.~\ref{qaaa}, we can conclude that there exists a reconstruction of $\dyad{\overline{\boldsymbol{\beta}}_{\boldsymbol{x}}}$ on the two laterals.

\paragraph*{\textbf{Simplified criterion for Prop.~\ref{qaaa}}}\label{simpcri} By summarizing the above arguments, we can further simplify the criterion for the existence of a reconstruction of $\dyad{\overline{\boldsymbol{\beta}}_{\boldsymbol{x}}}$ in terms of the distinguishability of entanglement patterns on a give boundary subregion $A$. Indeed, the above example illustrates that according to the way we define the $\ket{\psi_m}$ states (see Par.~\ref{psim}), there exists correlation between sub-configurations on different subregions, which can present in the same $\ket{\psi_m}$. Then, taking advantage of such correlation and referring to the above example, it is easy to show that the following condition guarantees the existence of recosntruction: There exists certain feature or restriction of the sub-configuration $\alpha_{i_1},\alpha_{i_2},\dots$ on $A$ such that for any $\ket{\psi_m}$ state, the feature or restriction can be read off $A$ if and only if in the pictorial representation the loop surrounding the hole $\boldsymbol{x}$ shows the concrete possibility (of the parity on laterals) corresponding to the emergence of $\overline{\boldsymbol{\beta}}_{\boldsymbol{x}}$. If revisiting the above example with this condition, the feature on the two laterals is simply the specification of one possibility out of (even, even), (even, odd), (odd, even) and (odd, odd). Note that if a subregion does not hold any such features, it simply means the entanglement patterns are indistinguishable on the subregions. This simplified criterion further elucidates the meaning of the distinguishability of entanglement patterns on boundary subregions with respect to the emergent bulk degrees of freedom. Based on this, it will be clear that the entanglement patterns specified in Sec.~\ref{pattern1} do realize the prediction discussed in Sec.~\ref{cencoding} (especially Condition 4 and 5).

\paragraph*{\textbf{Reconstruction on two connected paths}}\label{trdyad} According to the simplified criterion, we can even generalize the two laterals to two connected paths with the same ends as in the two laterals (see Fig.~\ref{12d}). We only require such paths to lie within the same triangular blocks as the two laterals, and to be ``economic'' enough to ensure the connectedness. We show that certain feature of the parity of the number of dark sides can be read off the paths if and only if the two laterals shows a concrete possibility of the parity, which is equivalent to the concrete possibility on the loop corresponding to the emergent $\overline{\boldsymbol{\beta}}_{\boldsymbol{x}}$. And hence there exists a reconstruction of $\dyad{\overline{\boldsymbol{\beta}}_{\boldsymbol{x}}}$ on the two paths.

As shown in Fig.~\ref{12d} and \ref{12e}, the feature of the sub-configurations on the two connected paths is simply the specific parity of the total number of the dark sides in the red-highlighted for each path in Fig.~\ref{12d}, or the ``red+orange'' highlighted for each path in Fig.~\ref{12e}. We can show that this parity for each path equals the parity of the dark sides in the highlighted for each lateral (comparing Fig.~\ref{12c} and \ref{12d}), and hence the simplified criterion is satisfied.

Without loss of generality, we can illustrate the arguments with Fig.~\ref{12e}. As shown in Fig.~\ref{12e}, the sub-configuration $\alpha_{i_1},\alpha_{i_2},\ldots$ on one path half-surrounds the sub-configuration $\alpha_{j_1},\alpha_{j_2},\ldots$ of some other qudits within the same triangular block. Now, we highlight by four colors certain collections of sides on the triangles (qudits) as shown in Fig.~\ref{12e}: the ``red+orange'' are on one of the path $\alpha_{i_1},\alpha_{i_2},\ldots$, the ``red+purple'' are on one lateral, the ``yellow+purple'' exactly cover the half-surrounded triangles $\alpha_{j_1},\alpha_{j_2},\ldots$, and the ``yellow+orange'' exactly form three loops. Now, since for each triangle representing a qudit there are even number of dark sides (see Fig.~\ref{fig7}), and in each loop, the total number of darks sides is also even (see Fig.~\ref{8a}), we can show that the parity of the total number of darks for ``yellow+purple'' equals that for ``yellow+orange'' which is even. Hence, the parity of dark sides in the  ``orange'' equals the parity of dark sides in the ``purple'', and hence the parity of dark sides on one paths (``red+orange'') equals the parity of dark sides on the corresponding lateral (``red+purple'').

Now, the typical subregions supporting reconstruction of $\dyad{\overline{\boldsymbol{\beta}}_{\boldsymbol{x}}}$ can be simply viewed as such connected paths (including the laterals) as illustrated in Fig.~\ref{12f}. By self-similarity of the alternative geometry, we have the same results for bulk qudits not in the center (see Fig.~\ref{12g}). Note that $\overline{\boldsymbol{\beta}}_{\boldsymbol{x}}$ can be any of $\boldsymbol{0},\boldsymbol{1},\boldsymbol{2},\boldsymbol{3}$.

\begin{center}
\begin{figure}[ht]
\centering
    \includegraphics[width=8.5cm]{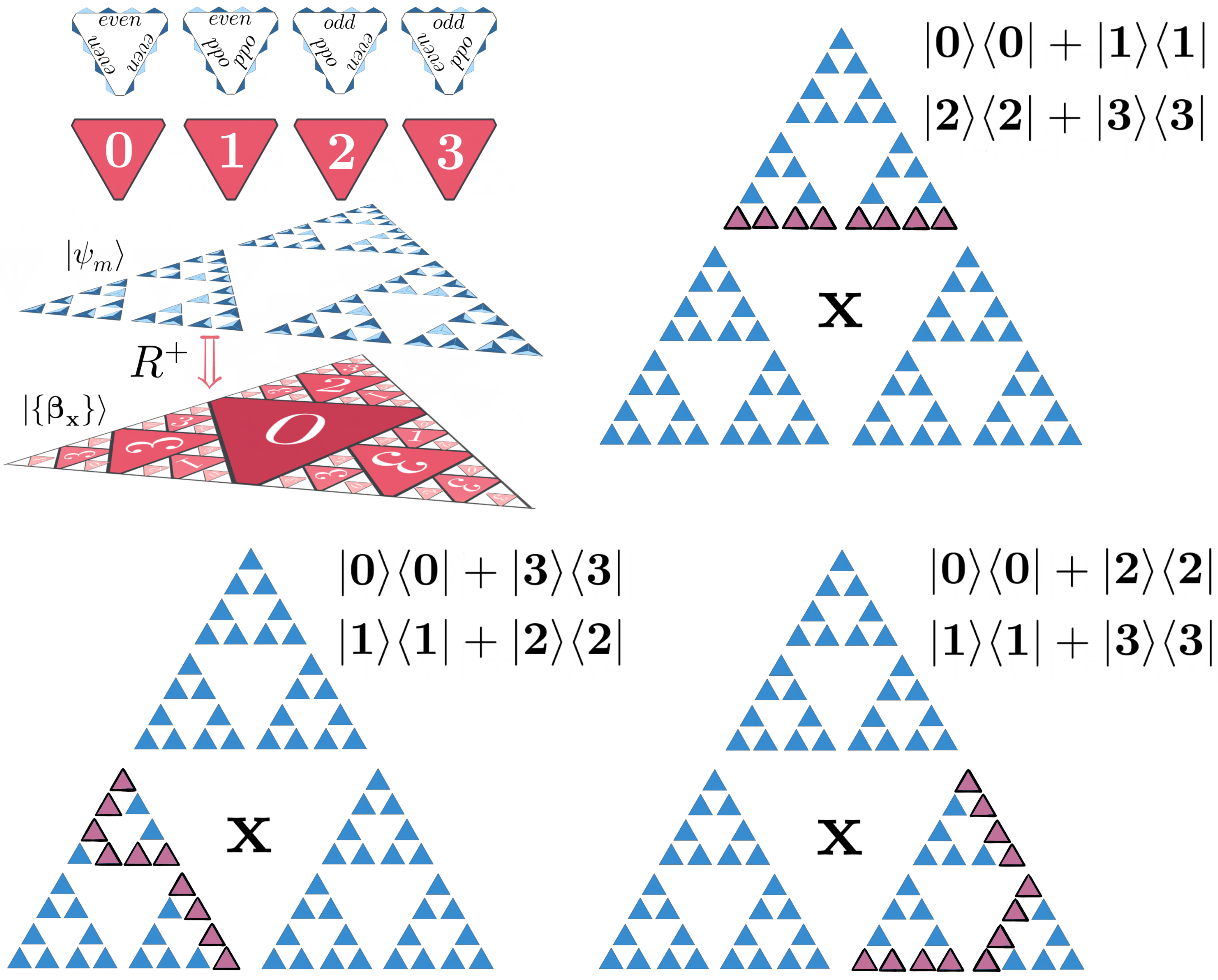}   
\caption{Typical subregions of qudits (the highlighted) supporting the reconstructions of the $\dyad{\overline{\boldsymbol{\beta}}_{\boldsymbol{x}}}+\dyad{\overline{\boldsymbol{\beta}'}_{\boldsymbol{x}}}$ operators. According to the orientation (upper left, the same as in Fig.~\ref{fig9}), the $\overline{\boldsymbol{\beta}}_{\boldsymbol{x}}$ and $\overline{\boldsymbol{\beta}'}_{\boldsymbol{x}}$ in the reconstructed logical operators have dependence on the typical subregions. The dependence is illustrated in the three pictures.}
\label{fig13}
\end{figure}
\end{center}

\subsubsection{Typical boundary reconstruction of $\dyad{\pmb{\beta}_{\pmb{x}}}+\dyad{\pmb{\beta'}_{\pmb{x}}}$}\label{trddyad}
To show more reconstructions of operators on a single bulk qudit $\boldsymbol{x}$, we can generalize the above approach based on Prop.~\ref{qaaa}. Here, we consider the reconstruction of $\dyad{\overline{\boldsymbol{\beta}}_{\boldsymbol{x}}}+\dyad{\overline{\boldsymbol{\beta}'}_{\boldsymbol{x}}}$ with $\overline{\boldsymbol{\beta}}_{\boldsymbol{x}}\ne\overline{\boldsymbol{\beta}'}_{\boldsymbol{x}}$, which will contribute to our arguments in establishing the subregion duality in the next section.

We can emulate the preceding process. The only difference is how $\dyad*{\overline{\boldsymbol{\beta}}_{\boldsymbol{x}}}+\dyad*{\overline{\boldsymbol{\beta}'}_{\boldsymbol{x}}}$ acts on the bulk qudit-product-states: $(\dyad*{\overline{\boldsymbol{\beta}}_{\boldsymbol{x}}}+\dyad*{\overline{\boldsymbol{\beta}'}_{\boldsymbol{x}}})\ket{\boldsymbol{\beta}_{\boldsymbol{1}}\cdots\boldsymbol{\beta}_{\boldsymbol{x}}\cdots\boldsymbol{\beta}_{K}}\ne0$ if and only if $\boldsymbol{\beta}_{\boldsymbol{x}}=\overline{\boldsymbol{\beta}}_{\boldsymbol{x}}$ or $\overline{\boldsymbol{\beta}'}_{\boldsymbol{x}}$. Accordingly, the equivalent condition for an operator $Q_A$ to reconstruct $\dyad{\overline{\boldsymbol{\beta}}_{\boldsymbol{x}}}+\dyad{\overline{\boldsymbol{\beta}'}_{\boldsymbol{x}}}$ on a boundary subregion $A$ only differs slightly: \emph{i)} $[Q_A,P_{\mathrm{code}}]=0$; \emph{ii)} $Q_A\ket*{\tilde{\varphi}_n}\ne0$ if and only if $\ket*{\tilde{\varphi}_n}=R\ket{\boldsymbol{\beta}_{\boldsymbol{1}}\cdots\overline{\boldsymbol{\beta}}_{\boldsymbol{x}}\cdots\boldsymbol{\beta}_{K}}$ or $\ket*{\tilde{\varphi}_n}=R\ket{\boldsymbol{\beta}_{\boldsymbol{1}}\cdots\overline{\boldsymbol{\beta}}'_{\boldsymbol{x}}\cdots\boldsymbol{\beta}_{K}}$, irrespective of $\boldsymbol{\beta}_{\boldsymbol{x}'}$ ($\boldsymbol{x}'\ne\boldsymbol{x}$); \emph{iii)} in the nonzero case, we have $Q_A\ket*{\tilde{\varphi}_n}=\ket*{\tilde{\varphi}_n}$.

Reconstructing $\dyad{\overline{\boldsymbol{\beta}}_{\boldsymbol{x}}}+\dyad{\overline{\boldsymbol{\beta}'}_{\boldsymbol{x}}}$ on $A$ requires looser distinguishability of the entanglement patterns on $A$, but the spirit is the same. Following the arguments in the reconstruction of $\dyad{\overline{\boldsymbol{\beta}}_{\boldsymbol{x}}}$, we can specify the collection of sub-configurations on $A$ similar to that in Eq.~\ref{collection}, and define the operator $Q_A$ similar to that in Eq.~\ref{qoa2}. Then, we can form a similar proposition as Prop.~\ref{qaaa}. Eventually, what can ensure $Q_A$ to be the desired reconstruction is that the subregion $A$ possess distinguishable feature to be read on every $\ket{\psi_m}$ state, which is equivalent to specifying two possibilities of parity out of the four (even, even, even), (even, odd, odd), (odd, even, odd) and (odd, odd, even) on the loop of sides surrounding the hole $\boldsymbol{x}$ (see Fig.~\ref{fig13}).

It is easy to show that typical qualified subregions supporting the reconstruction of $\dyad{\overline{\boldsymbol{\beta}}_{\boldsymbol{x}}}+\dyad{\overline{\boldsymbol{\beta}'}_{\boldsymbol{x}}}$ include each lateral of the loop of qudits surrounding the hole $\boldsymbol{x}$ and any connected path sharing the same ends as the lateral as illustrated in Fig.~\ref{fig13}. It is noticeable that according to the convention of orientation (see Fig.~\ref{9a}, \ref{9b} and \ref{fig13}), for paths lying within different triangular blocks, the $\overline{\boldsymbol{\beta}}_{\boldsymbol{x}}$ and $\overline{\boldsymbol{\beta}'}_{\boldsymbol{x}}$ in the reconstructed $\dyad{\overline{\boldsymbol{\beta}}_{\boldsymbol{x}}}+\dyad{\overline{\boldsymbol{\beta}'}_{\boldsymbol{x}}}$ operators, as indicated in Fig.~\ref{fig13}, are different. For $\boldsymbol{x}'$ not in the center, the results are the same.

\subsubsection{Reconstruction of $\widetilde{\pmb{S}}^{\pmb{\sigma}}_{\pmb{x}}$ as the gates}
It is noticeable that the distinguishability of entanglement patterns for reconstructing $\dyad{\overline{\boldsymbol{\beta}}_{\boldsymbol{x}}}$ does not guarantee the boundary subregions to be minimal reconstructions of $\mathcal{M}(\boldsymbol{x})$, because according to Condition 2 in Sec.~\ref{cencoding}, we still have to show the indistinguishability of the entanglement patterns on the complement subregions. For instance, in the above examples of distinguishability, while $\dyad{\boldsymbol{\beta}_{\boldsymbol{x}}}$ can be reconstructed on the boundary subregion $A$ as two laterals, $\dyad{\boldsymbol{\beta}_{\boldsymbol{x}}}+\dyad{\boldsymbol{\beta}'_{\boldsymbol{x}}}$ can be also reconstructed on the other lateral that is within the complement $\overline{A}$. It means that the entanglement patterns are still distinguishable on $\overline{A}$ with respect to the emergent bulk qudit, and hence $\mathcal{M}(\boldsymbol{x})$ is not completely recovered on $A$.

To construct a recovery of $\mathcal{M}(\boldsymbol{x})$ on a boundary subregion $A$, instead of investigating the indistinguishability on the complement $\overline{A}$, we can consider the reconstruction of the $\widetilde{\boldsymbol{S}}^{\boldsymbol{\sigma}}_{\boldsymbol{x}}$ operators, since the $\widetilde{\boldsymbol{S}}^{\boldsymbol{\sigma}}_{\boldsymbol{x}}$s together with the $\dyad{\boldsymbol{\beta}_{\boldsymbol{x}}}$s form generators of $\mathcal{M}(\boldsymbol{x})$. Indeed, since the $\widetilde{\boldsymbol{S}}^{\boldsymbol{\sigma}}_{\boldsymbol{x}}$s does not commute with the $\dyad{\boldsymbol{\beta}_{\boldsymbol{x}}}$s (see Eq.~\ref{cts}), Lemma~\ref{oaqec} implies that the reconstructions of the $\widetilde{\boldsymbol{S}}^{\boldsymbol{\sigma}}_{\boldsymbol{x}}$s for $\boldsymbol{\sigma}=\boldsymbol{1},\boldsymbol{2},\boldsymbol{3}$ on $A$ simply disables any reconstruction of the $\dyad{\boldsymbol{\beta}_{\boldsymbol{x}}}$s on $\overline{A}$, and negates the distinguishability of the entanglement patterns on $\overline{A}$ with respect to the emergent bulk degrees of freedom.

For convenience in discussion and without loss of generality, we describe a $\widetilde{\boldsymbol{S}}^{\boldsymbol{\sigma}}_{\boldsymbol{x}}$ operator by 
\begin{equation*}
\widetilde{\boldsymbol{S}}^{\boldsymbol{\sigma}}=\dyad*{\bar{\boldsymbol{\beta}}_{\boldsymbol{x}}}{\tilde{\boldsymbol{\beta}}_{\boldsymbol{x}}}+\dyad*{\tilde{\boldsymbol{\beta}}_{\boldsymbol{x}}}{\bar{\boldsymbol{\beta}}_{\boldsymbol{x}}}+\dyad*{\check{\boldsymbol{\beta}}_{\boldsymbol{x}}}{\hat{\boldsymbol{\beta}}_{\boldsymbol{x}}}+\dyad*{\hat{\boldsymbol{\beta}}_{\boldsymbol{x}}}{\check{\boldsymbol{\beta}}_{\boldsymbol{x}}}.
\end{equation*}
Here, we use $\bar{\boldsymbol{\beta}}_{\boldsymbol{x}},\tilde{\boldsymbol{\beta}}_{\boldsymbol{x}},\check{\boldsymbol{\beta}}_{\boldsymbol{x}},\hat{\boldsymbol{\beta}}_{\boldsymbol{x}}$ to represent $\boldsymbol{0},\boldsymbol{1},\boldsymbol{2},\boldsymbol{3}$ in any order, and we use $\boldsymbol{\beta}_{\boldsymbol{x}'}$ for arbitrarness. Then, an operator $Q_A$ ($[Q_A,P_{\mathrm{code}}]=0$) reconstructs $\widetilde{\boldsymbol{S}}^{\boldsymbol{\sigma}}_{\boldsymbol{x}}$ on $A$ if and only if the action of $Q_A$ on the basis states $\ket*{\tilde{\varphi}_n}$s, as a unitary and self-adjoint operator, simply exchanges between two code basis states
\begin{equation*}
\ket{\tilde{\varphi}_{\bar n}}=R\ket*{\boldsymbol{\beta}_{\boldsymbol{1}}\cdots\bar{\boldsymbol{\beta}}_{\boldsymbol{x}}\cdots\boldsymbol{\beta}_K},~\ket{\tilde{\varphi}_{\tilde n}}=R\ket*{\boldsymbol{\beta}_{\boldsymbol{1}}\cdots\tilde{\boldsymbol{\beta}}_{\boldsymbol{x}}\cdots\boldsymbol{\beta}_K}
\end{equation*}
and between
\begin{equation*}
\ket{\tilde{\varphi}_{\check n}}=R\ket*{\boldsymbol{\beta}_{\boldsymbol{1}}\cdots\check{\boldsymbol{\beta}}_{\boldsymbol{x}}\cdots\boldsymbol{\beta}_K},~\ket{\tilde{\varphi}_{\hat n}}=R\ket*{\boldsymbol{\beta}_{\boldsymbol{1}}\cdots\hat{\boldsymbol{\beta}}_{\boldsymbol{x}}\cdots\boldsymbol{\beta}_K}
\end{equation*}
with no effect on and irrespective of the $\boldsymbol{\beta}_{\boldsymbol{x}'}$s ($\boldsymbol{x}'\ne\boldsymbol{x}$).

To see a simple reconstruction of $\widetilde{\boldsymbol{S}}^{\boldsymbol{\sigma}}_{\boldsymbol{x}}$, we consider $T(\boldsymbol{x})=T_{i_1i_2}T_{i_3i_4}T_{i_5i_6}$ as illustrated in Fig.~\ref{9c}, where the three $T_{ii'}$ gates live respectively on the three corners of the loop surrounding the hole $\boldsymbol{x}$. And we compare the action of $P_{\mathrm{code}}T_{ii'}P_{\mathrm{code}}$ with that of $\widetilde{\boldsymbol{S}}^{\boldsymbol{\sigma}}_{\boldsymbol{x}}$.

Based on the discussion in Sec.~\ref{exmodel}, we can easily derive the following properties: (1) According to the definition of $P_{\mathrm{code}}$ in Thm.~\ref{cp1}, each of the three $T_{ii'}$ gates commute with $P_{\mathrm{code}}$, and $P_{\mathrm{code}}T_{ii'}P_{\mathrm{code}}$ is unitary and self-adjoint within $\mathcal{H}_{\mathrm{code}}$. (2) According to Prop.~\ref{pp1} and Eq.~\ref{cbs1}, the action of such a $T_{ii'}$ gates on any $\ket*{\tilde{\varphi}_n}=R\ket{\boldsymbol{\beta}_{\boldsymbol{1}}\cdots\boldsymbol{\beta}_{\boldsymbol{x}}\cdots\boldsymbol{\beta}_{K}}$ results in another $\ket*{\tilde{\varphi}_{n'}}$ and has no effect on $\boldsymbol{\beta}_{\boldsymbol{x}'}$ with $\boldsymbol{x}'\ne\boldsymbol{x}$, hence $P_{\mathrm{code}}T_{ii'}P_{\mathrm{code}}$ belongs to the subalgebra $\mathcal{M}(\boldsymbol{x})$. (3) According to the illustration of the emergent bulk qudits in Fig.~\ref{9a} and Fig.~\ref{figapp2}, the effect of the $T_{ii'}$ on the bulk qudit $\boldsymbol{x}$ is exactly exchanging the degrees of freedom between $\bar{\boldsymbol{\beta}}_{\boldsymbol{x}}$ and $\tilde{\boldsymbol{\beta}}_{\boldsymbol{x}}$, and between $\check{\boldsymbol{\beta}}_{\boldsymbol{x}}$ and $\hat{\boldsymbol{\beta}}_{\boldsymbol{x}}$. And the difference in the effects of $T_{i_1i_2}$, $T_{i_3i_4}$ and $T_{i_5i_6}$ simply corresponds to the difference in the order of $\bar{\boldsymbol{\beta}}_{\boldsymbol{x}},\tilde{\boldsymbol{\beta}}_{\boldsymbol{x}},\check{\boldsymbol{\beta}}_{\boldsymbol{x}},\hat{\boldsymbol{\beta}}_{\boldsymbol{x}}$ in the exchange.

The above comparison clearly shows that the three gates $T_{i_1i_2}$, $T_{i_3i_4}$ and $T_{i_5i_6}$ reconstruct the operator $\widetilde{\boldsymbol{S}}^{\boldsymbol{1}}_{\boldsymbol{x}}$, $\widetilde{\boldsymbol{S}}^{\boldsymbol{2}}_{\boldsymbol{x}}$ and $\widetilde{\boldsymbol{S}}^{\boldsymbol{3}}_{\boldsymbol{x}}$ respectively. Similarly, the products $T_{i_1i_2}T_{i_3i_4}$, $T_{i_3i_4}T_{i_5i_6}$ and $T_{i_3i_4}T_{i_5i_6}$ also reconstruct these operators. Note that the correspondence between the gates and the reconstructed operators depends on the convention in the illustration of the emergent bulk degrees of freedom (see Fig.~\ref{9a} and \ref{9b}).

Based on the above observation on the reconstruction of the $\widetilde{\boldsymbol{S}}^{\boldsymbol{\sigma}}_{\boldsymbol{x}}$ operators, it is sufficient to prove Thm.~\ref{cdistance}.

\paragraph*{\textbf{Proof of Thm.~\ref{cdistance} on the connected code distance}} According to Fig~\ref{fig11} and \ref{12f}, for any type of the geometry of the subregion $\overline{A}$ to be erased, there is always reconstruction of the $\dyad{\overline{\boldsymbol{\beta}}_{\boldsymbol{x}}}$ operators for $\overline{\boldsymbol{\beta}}_{\boldsymbol{x}}=\boldsymbol{0},\boldsymbol{1},\boldsymbol{2},\boldsymbol{3}$ on two connected paths (including laterals) which avoid $\overline{A}$ and hence lies within $A$. As illustrated in Fig.~\ref{fig14}, since the $\overline{A}$ subregions cover at most one out of the six of the corner qudits of the loop surrounding the hole $\boldsymbol{x}$, the two paths for reconstructing $\dyad{\overline{\boldsymbol{\beta}}_{\boldsymbol{x}}}$ and lying with two big triangular blocks can be always extended to the third triangular block by adding one more corner qudit (out of the six). Let us denote this enlarged subregion by $A_{\mathrm{min}}$, which will be proved to be a minimal subregion for recovery $\mathcal{M}(\boldsymbol{x})$ in later discussion. Then, according to the above discussion, $A_{\mathrm{min}}$ supports the reconstructions of two out of the three operators $\widetilde{\boldsymbol{S}}^{\boldsymbol{1}}_{\boldsymbol{x}},\widetilde{\boldsymbol{S}}^{\boldsymbol{2}}_{\boldsymbol{x}},\widetilde{\boldsymbol{S}}^{\boldsymbol{3}}_{\boldsymbol{x}}$; and the other one is simply a product of the two ($\widetilde{\boldsymbol{S}}^{\boldsymbol{\sigma}''}_{\boldsymbol{x}}=\widetilde{\boldsymbol{S}}^{\boldsymbol{\sigma}}_{\boldsymbol{x}}\widetilde{\boldsymbol{S}}^{\boldsymbol{\sigma}'}_{\boldsymbol{x}}$). If follows that the generators of $\mathcal{M}(\boldsymbol{x})$ can all be reconstructed on the $A_{\mathrm{min}}$ subregions as shown in Fig.~\ref{fig14}, which can completely avoid the $\overline{A}$ subregions of any type. In other words, any $\overline{A}$ of size $\mathrm{d_c}({\boldsymbol{x}})-1$ is correctable with respect to $\mathcal{M}(\boldsymbol{x})$. On the other hand, there is example of connected subregion of size $\mathrm{d_c}({\boldsymbol{x}})$ which is not correctable. Indeed, the subregion $\overline{A}$ shown in Fig.~\ref{11d} with the corner qudit in the upper big triangular block added will cover two physical qudits as one corner of the loop and hence support the reconstruction of a $\widetilde{\boldsymbol{S}}^{\boldsymbol{\sigma}}_{\boldsymbol{x}}$ operator. Combing the above observations and noting that our arguments also apply to $\boldsymbol{x}'$ no in the center, it is clear that Thm.~\ref{cdistance} has been proved.

\onecolumngrid
\begin{center}
\begin{figure}[ht]
\centering
    \includegraphics[width=15cm]{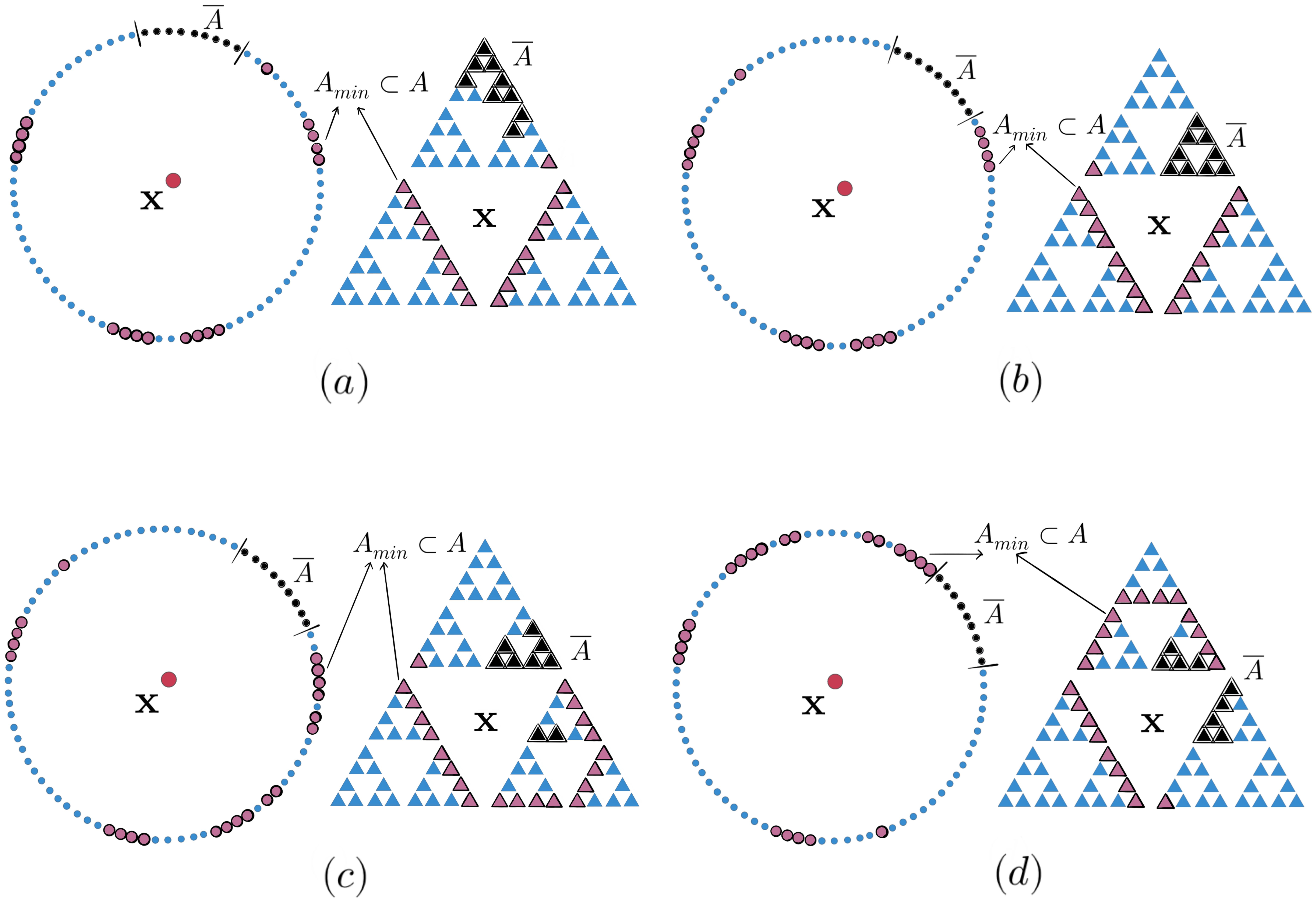}   
\phantomsubfloat{\label{14a}}\phantomsubfloat{\label{14b}}
\phantomsubfloat{\label{14c}}\phantomsubfloat{\label{14d}}
\caption{The purple-highlighted subregions represent the $A_{\mathrm{min}}\subset A$ subregions that avoid the $\overline{A}$ subregions in both the standard geometry and the alternative geometry. The $\overline{A}$ subregions are the same as illustrated in Fig.~\ref{fig11}. In the alternative geometry, each $A_{\mathrm{min}}$ consists of two connected paths in two big triangular blocks together with one more qudit in the third block. $A_{\mathrm{min}}$ supports the reconstructions of $\dyad{\overline{\boldsymbol{\beta}}_{\boldsymbol{x}}}$ operators for $\overline{\boldsymbol{\beta}}_{\boldsymbol{x}}=\boldsymbol{0},\boldsymbol{1},\boldsymbol{2},\boldsymbol{3}$ on the two paths, and support the reconstructions of $\widetilde{\boldsymbol{S}}^{\boldsymbol{\sigma}}_{\boldsymbol{x}}$, $\widetilde{\boldsymbol{S}}^{\boldsymbol{\sigma}'}_{\boldsymbol{x}}$ and $\widetilde{\boldsymbol{S}}^{\boldsymbol{\sigma}''}_{\boldsymbol{x}}=\widetilde{\boldsymbol{S}}^{\boldsymbol{\sigma}}_{\boldsymbol{x}}\widetilde{\boldsymbol{S}}^{\boldsymbol{\sigma}'}_{\boldsymbol{x}}$ on the two corners of the loop. In the standard geometry, $A_{\mathrm{min}}$ appears disconnected and scattered.}
\label{fig14}
\end{figure}
\end{center}
\twocolumngrid

Note that the above proof and the illustration in Fig.~\ref{fig14} elucidate how the alternative geometry and the rearrangement underlie the expected HQEC characteristics regarding the connected reconstruction of a single bulk qudit, as proposed in Sec.~\ref{demys}. Indeed, any connected subregion for reconstructing $\mathcal{M}(\boldsymbol{x})$ must include a minimal subregion $A_{\mathrm{min}}$ for the reconstruction. Hence, Fig.~\ref{fig14} also indicate possible connected boundary subregion for recovering the single bulk qudit $\boldsymbol{x}$.

\subsubsection{Typical reconstruction of $\widetilde{\pmb{S}}^{\pmb{\sigma}}_{\pmb{x}}$}\label{trs}
In more general cases, a single $\widetilde{\boldsymbol{S}}^{\boldsymbol{\sigma_{\boldsymbol{x}}}}_{\boldsymbol{x}}$ operator can be recovered as a product of the $S^{\sigma_i}_i$ operators on scattered physical qudits. These types of reconstruction will be important in our following arguments in explicitly showing the minimal boundary reconstruction of $\mathcal{M}(\boldsymbol{x})$.

Generally, we can consider an arbitrary product of the gates $(T_{ii'}T_{jj'}\cdots)$ which, according to the above discussion, reconstruct a product $\widetilde{\boldsymbol{S}}^{\boldsymbol{\sigma_{\boldsymbol{x}}}}_{\boldsymbol{x}}\widetilde{\boldsymbol{S}}^{\boldsymbol{\sigma_{\boldsymbol{x}'}}}_{\boldsymbol{x}'}\cdots$. However, in certain cases we may organize the product to rewrite it as
\begin{align}\label{organize1}
\begin{split}
&T_{ii'}T_{jj'}\cdots=T_{i_1i_2}T(\boldsymbol{x}_1)T(\boldsymbol{x}_{2})\cdots,\\
&or\\
&T_{ii'}T_{jj'}\cdots=T_{i_1i_2}T_{i_5i_6}T(\boldsymbol{x}_1)T(\boldsymbol{x}_{2})\cdots.
\end{split}
\end{align}
with exactly one or two $T_{ii'}$ gates multiplied with assembled gates. Here, $T_{i_1i_2}$ and $T_{i_5i_6}$ are two out of the three in an assemble gates $T(\boldsymbol{x})=T_{i_1i_2}T_{i_3i_4}T_{i_5i_6}$. In this case, the projection of $(T_{ii'}T_{jj'}\cdots)$ onto the code subspace indeed cancels the assembled gates and hence only a single $\widetilde{\boldsymbol{S}}^{\boldsymbol{\sigma_{\boldsymbol{x}}}}_{\boldsymbol{x}}$ operator is reconstructed.

\onecolumngrid
\begin{center}
\begin{figure}[ht]
\centering
    \includegraphics[width=16cm]{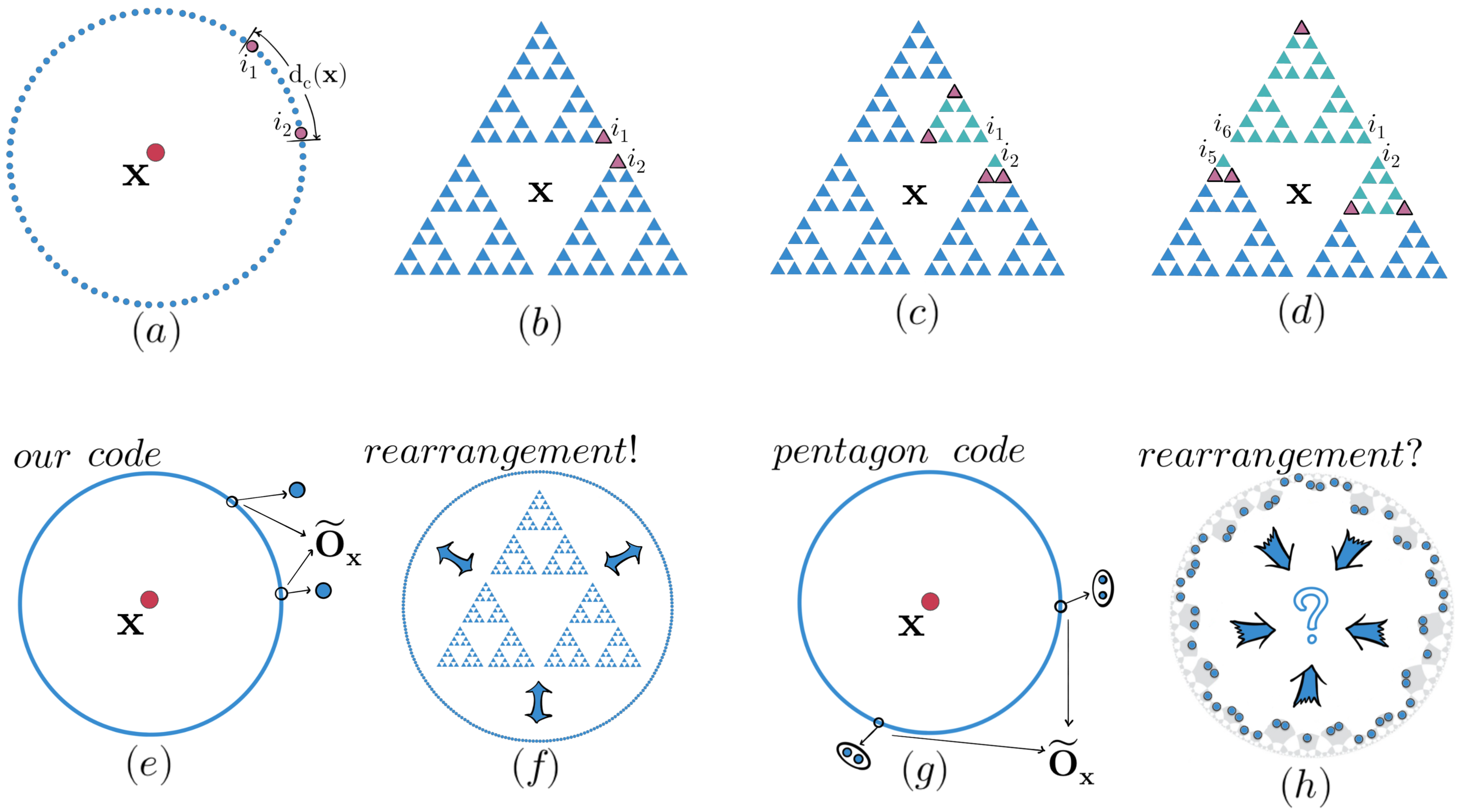}   
\phantomsubfloat{\label{15a}}\phantomsubfloat{\label{15b}}
\phantomsubfloat{\label{15c}}\phantomsubfloat{\label{15d}}
\phantomsubfloat{\label{15e}}\phantomsubfloat{\label{15f}}
\phantomsubfloat{\label{15g}}\phantomsubfloat{\label{15h}}
\caption{(a) (b) Through the rearrangement from the alternative geometric setting to the standard geometric setting, the connected code distance with respect to the central qudit $\boldsymbol{x}$ separates the two physical qudits $i_1$ and $i_2$ that supports a $T_{i_a i_2}$ gate. (c) (d) Two typical reconstructions of a $\widetilde{\boldsymbol{S}}^{\boldsymbol{\sigma_{\boldsymbol{x}}}}_{\boldsymbol{x}}$ operator for the central bulk qudit $\boldsymbol{x}$. The green-colored together with the purple-highlighted specify all the $T_{ii'}$ gates in the product $T_{ii'}T_{jj'}\cdots$ whose support only consists of the purple-highlighted. (e) In our code, two separated local qudit support a nontrivial single-bulk-qudit operator. (f) The rearrangement between the alternative geometric setting and the standard geometric setting underlies the feature in (e). (g) Depiction of Fig.16b in Ref.~\cite{pastawski2015}. In the HaPPY pentagon code, two separated local pairs of qubits support a nontrivial single-bulk-qudit operator. (h) Similarity between the features in (e) and (g) might imply a possible similar mechanism for the pentagon code. The figure (h) is adapted from figures in Ref.~\cite{pastawski2015}.}
\label{fig15}
\end{figure}
\end{center}
\twocolumngrid

On the other hand, while the product $(T_{ii'}T_{jj'}\cdots)$ seems supported on the union of the physical qudits supporting each gate, it is generally supported on a smaller boundary subregion consisting of scattered physical qudits. To show how this is possible, we recall the example indicated in the proof of Prop.~\ref{gates1} (see App.~\ref{pogates1}). As shown in Fig.~\ref{app1a}, we can consider the product of the three gates that all include a physical qudit $l$ in their supports, i.e. $T_{il}$, $T_{jl}$ and $T_{kl}$. Then their total effect on qudit $l$ is the the multiplication $S_l^{1}S_l^{2}S_l^{3}=\mathds{1}$, which means that qudit $l$ does not really belong to the support of the product of gates.

In Fig.~\ref{15c} and \ref{15d}, we illustrate two typical examples of the product $(T_{ii'}T_{jj'}\cdots)$. In such an illustration, $(T_{ii'}T_{jj'}\cdots)$ is specified by the collection of all the physical qudits supporting each $T_{ii'}$ gate in the product, i.e., both the purple-highlighted and the green-colored. In other words, for any $T_{ii'}$ gate, if it is supported on two green-colored, on one purple-highlighted together with a green-colored, or on two purple-highlighted, the gate participates in the product $(T_{ii'}T_{jj'}\cdots)$. It is clear that the true support of the product only consists of the purple-highlighted, because for each green-colored qudit, say $l$, all the three gates $T_{il}$, $T_{jl}$ and $T_{kl}$ are included in the product and hence cancels the effect on $l$. And for each of the purple-highlighted, only two out of the three gates are included in the product $(T_{ii'}T_{jj'}\cdots)$, which results in a $S^{\sigma_i}_i$ operator thereon. It follows that in both figures, the support of $(T_{ii'}T_{jj'}\cdots)$ consists of the scattered purple-highlighted on each of which the net action is the $S^{\sigma_i}_i$ operator.

It is also clear that the product $(T_{ii'}T_{jj'}\cdots)$ in Fig.~\ref{15c} and \ref{15d} can be reorganized in the two forms of Eq.~\ref{organize1} respectively. Accordingly, in Fig.~\ref{15c}, the bulk operator recovered is $\widetilde{\boldsymbol{S}}^{\boldsymbol{\sigma_{\boldsymbol{x}}}}_{\boldsymbol{x}}$, exactly the same as the case in Fig.~\ref{15b}. In Fig.~\ref{15d}, the bulk operator recovered is $\widetilde{\boldsymbol{S}}^{\boldsymbol{\sigma_{\boldsymbol{x}}}}_{\boldsymbol{x}''}=\widetilde{\boldsymbol{S}}^{\boldsymbol{\sigma_{\boldsymbol{x}}}}_{\boldsymbol{x}}\widetilde{\boldsymbol{S}}^{\boldsymbol{\sigma_{\boldsymbol{x}'}}}_{\boldsymbol{x}}$.

\subsection{Remarks on the code distance}\label{remark2}

\subsubsection{Distance and connected distance}
The fact that a two-qudit gate (see Fig.~\ref{15a}, \ref{15b} and \ref{15e}) reconstructs the $\widetilde{\boldsymbol{S}}^{\boldsymbol{\sigma}}_{\boldsymbol{x}}$ operator for the central qudit $\boldsymbol{x}$ simply means that a nontrivial single-bulk-qudit operator is supported on two physical ququarts ($\mathbb{C}^4$). Then, according to the connected distance, it is clear that the code distance with respect to $\mathcal{M}(\boldsymbol{x})$ is $2$. While this might sounds a weakness in error correction, it actually once again illustrates the essential difference between HQEC and the conventional quantum eraser correction, especially in properties corresponding to the OAQEC formalism~\cite{harlow2017}.

Similar to the HaPPY pentagon code~\cite{pastawski2015} in which a nontrivial bulk operator on the central bulk qubit has support on two pairs of neighboring physical qubits (see Fig.~\ref{15g} which depicts Fig.16b of Ref.~\cite{pastawski2015}), the small and constant code distance does not affect the demonstration of the HQEC characteristics. Indeed, in both our code and the pentagon code, the two local physical degrees of freedom, though acting nontrivially on a bulk qudit, live far apart on the boundary and are separated by the connected code distance with respect to the central bulk qudit (see Fig.~\ref{15a} and \ref{15b} for our case), which characterizes the needed correctability to explain the radial commutativity in AdS/CFT. Furthermore, the ``large enough'' connected code distance has been pointed out as a potential condition underlying the investigation of certain possible class of gates in the study of quantum computation~\cite{cree2021}.

There might be an inherent restriction on the code distance---the formalism of operator-algebra quantum error correction implies very ``poor'' protection of logical information in contrast to the conventional codes. In the conventional view, requiring genuine operator-algebra correctability for every bipartition of physical qudits (HQEC Characteristic 3, see Par.~\ref{chac3}) essentially disables the correction of any erasure errors with respect to the whole of the bulk qudits. Hence, with genuine operator-algebra correctability, it is only meaningful to study the distance of a subsystem of the bulk.



But even for the distance of bulk subsystem, the requirement of Characteristic 3 still imposes restrictions. Indeed, while certain tensor-network models that are modified based on the HaPPY code by thinning out bulk qubits~\cite{pastawski2015,cao2021} show system-size-dependent code distance, the increase in the code distance seems in tension with satisfying Characteristic 3~\footnote{First, in these models, each of the component tensor itself does not satisfy the genuine operator-algebra correctability for arbitrary bipartition of its legs. Second, in the process of thinning out the bulk degrees of freedom, tensors with a dangling leg for the bulk qubit are replaced by perfect tensors without bulk legs. For these tensors which lie at the boundary, they are usually contracted with the whole tensor network through only quite few number of legs, and most legs are uncontracted. Consequently, through the operator pushing, some of these uncontracted boundary legs can be avoided and are hence correctable with respect to the whole of the bulk. This simply means that the genuine operator-algebra correctability does not hold for boundary bipartition that separating these uncontracted legs from others.}. In an extreme case, keep thinning out the bulk qubits till only one dangling leg remains, e.g., the HaPPY one qubit code~\cite{pastawski2015}, the unique bulk qubit represents the whole logical degrees of freedom, and hence the large code distance disables those correctable boundary bipartitions to hold genuine operator-algebra correctability.

\subsubsection{``Larger'' recovery of $\dyad{\pmb{\beta}_{\pmb{x}}}$, ``smaller'' recovery of $\widetilde{\pmb{S}}^{\pmb{\sigma}}_{\pmb{x}}$}
Another noteworthy point is that the subregions supporting the reconstructions for the two types of operators, i.e., the $\dyad{\boldsymbol{\beta}_{\boldsymbol{x}}}$s and the $\widetilde{\boldsymbol{S}}^{\boldsymbol{\sigma}}_{\boldsymbol{x}}$s, appear to have different scaling on $N$. This is another property that coincides with the HaPPY code~\cite{pastawski2015}. In the latter, the price and the distance for a single bulk qudit can be ``badly mismatched''~\cite{pastawski2017}, which exactly means that the reconstruction of certain single-bulk-qubit operator requires the size of the subregion to be as large as the price, while some other operator, as mentioned above, is supported only on two local pairs of qubits.

This and the above comparison between our code and the HaPPY pentagon code appears to indicate that the latter might be equivalently described within our framework, i.e., with an alternative geometry and rearrangement. Indeed, for our code, the relation between code distance and connected code distance arises inherently from how the standard geometric setting can be rearranged from the alternative geometry (see Fig.~\ref{fig4}, \ref{15a} and \ref{15b}). And in the alternative geometry, the explicit structures of the code states naturally give rise to the different scaling of the sizes for reconstructing different single-bulk operators. Hence, it should be inspiring to investigate how the same features in the pentagon code can be interpreted in terms of a similar mechanism.

In further studies of holographic code in the formalism of operator-algebra quantum error correction, it might be important to understand the above two features regarding the sizes of subregions for the operators' reconstructions. While such correspondence is currently elusive in the tensor-network construction, our model presents a clear example that can serve as a clue for the potential studies.


\subsection{Minimal boundary subregions for reconstructing $\mathcal{M}(\pmb{x})$}\label{minirc}
Now, we explicitly describe all types of minimal boundary reconstructions of $\mathcal{M}(\boldsymbol{x})$ and show that the universal scaling behavior for them is $\sim N^{1/h}$. According to our discussion in Sec.~\ref{h1h}, it suffices to prove that all the minimal boundary recoveries scale linearly on the linear size $N_0$ (the length of a lateral) of the alternative geometry. Our arguments will mainly focus on the bulk qudit $\boldsymbol{x}$ in the center, for other bulk qudits we will take advantage of the self-similarity of the alternative geometry together with limited annotations. The basic idea is to combine the minimal recoveries of the $\dyad{\boldsymbol{\beta}_{\boldsymbol{x}}}$s and the $\widetilde{\boldsymbol{S}}^{\boldsymbol{\sigma}}_{\boldsymbol{x}}$s. And we can describe the line of thinking as follows.

\onecolumngrid
\begin{center}
\begin{figure}[ht]
\centering
    \includegraphics[width=17cm]{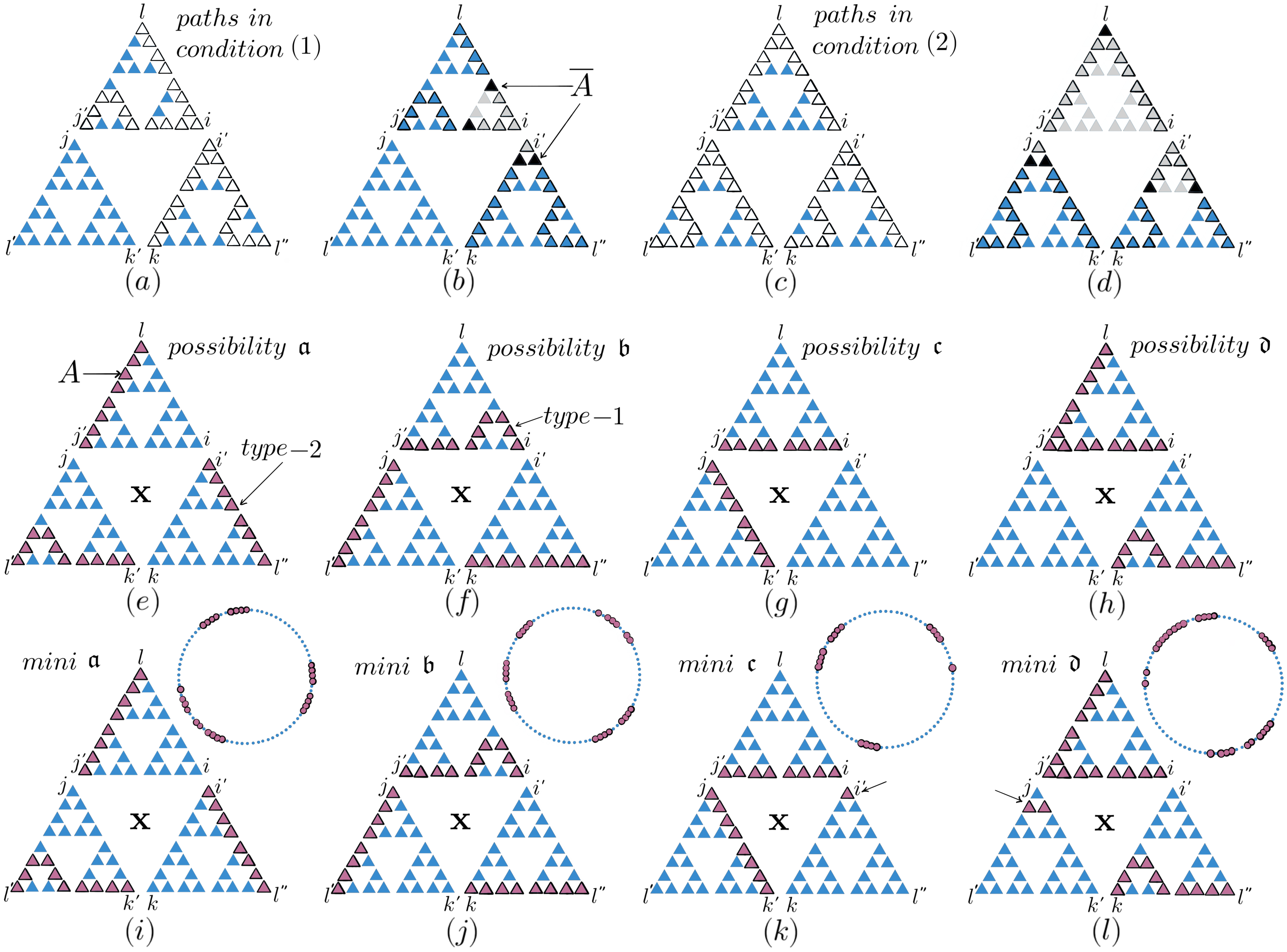}   
\phantomsubfloat{\label{16a}}\phantomsubfloat{\label{16b}}
\phantomsubfloat{\label{16c}}\phantomsubfloat{\label{16d}}
\phantomsubfloat{\label{16e}}\phantomsubfloat{\label{16f}}
\phantomsubfloat{\label{16g}}\phantomsubfloat{\label{16h}}
\phantomsubfloat{\label{16i}}\phantomsubfloat{\label{16j}}
\phantomsubfloat{\label{16k}}\phantomsubfloat{\label{16l}}
\caption{The white-highlighted denote paths described in the conditions of Prop.~\ref{qaaaa}. The purple-highlighted denote qudits in subregion $A$. The black-highlighted denote qudits in subregion $\overline{A}$. (a) (c) Examples of the paths mentioned in conditions (1) and (2) in Prop.~\ref{qaaaa}. (b) (d) Replica of Fig.~\ref{15c} and \ref{15d} with colored changed from green to grey and from purple to black as standing for $\overline{A}$. With reference to (a) and (c), it is clear that the black-highlighted live on the paths. (e) (f) (g) (h) Illustrations of the four possibilities $\mathfrak{a}$, $\mathfrak{b}$, $\mathfrak{c}$ and $\mathfrak{d}$ respectively. (i) (j) (k) (l) Illustrations of the four types of minimal subregions for recovery in both the alternative geometry and the standard 1D geometry of the physical qudits.}
\label{fig16}
\end{figure}
\end{center}
\twocolumngrid

If a boundary subregion $A$ is minimal to reconstruct $\mathcal{M}(\boldsymbol{x})$, then the $\dyad{\boldsymbol{\beta}_{\boldsymbol{x}}}$s can be surely reconstructed on $A$. According to Lemma~\ref{oaqec} and the noncommutativity shown in Eq.~\ref{cts}, there must be no reconstruction of any product $\widetilde{\boldsymbol{S}}^{\boldsymbol{\sigma}_{\boldsymbol{x}}}_{\boldsymbol{x}}\widetilde{\boldsymbol{S}}^{\boldsymbol{\sigma}_{\boldsymbol{x}'}}_{\boldsymbol{x}'}\cdots$ including the $\widetilde{\boldsymbol{S}}^{\boldsymbol{\sigma}_{\boldsymbol{x}}}_{\boldsymbol{x}}$ operators ($\boldsymbol{\sigma}_{\boldsymbol{x}}\ne\boldsymbol{0}$) on the complement $\overline{A}$. Based on this fact, we adopt a reverse approach and start with sufficient condition for the complement $\overline{A}$ to support the recoveries of the product $\widetilde{\boldsymbol{S}}^{\boldsymbol{\sigma}_{\boldsymbol{x}}}_{\boldsymbol{x}}\widetilde{\boldsymbol{S}}^{\boldsymbol{\sigma}_{\boldsymbol{x}'}}_{\boldsymbol{x}'}\cdots$. Then, negating such condition leads to necessary condition for subregion $A$ to support the recoveries of the $\dyad{\boldsymbol{\beta}_{\boldsymbol{x}}}$s. Indeed, as we will show, in terms of the necessary conditions we can already describe all possible minimal subregions supporting the recoveries of the $\dyad{\boldsymbol{\beta}_{\boldsymbol{x}}}$s. Furthermore, the description will guide how to include the recoveries of the $\widetilde{\boldsymbol{S}}^{\boldsymbol{\sigma}}_{\boldsymbol{x}}$s on $A$ and hence complete the recovery of the whole $\mathcal{M}(\boldsymbol{x})$.

For convenience in the following discussion, as illustrated in Fig.~\ref{fig16}, we let the pairs of physical qudits $(i,i'),(j,j'),(k,k')$ be the three corners of the loop surrounding the hole, and let the physical qudits $l$, $l'$ and $l'''$ be at the top of the three big triangular blocks that include $i$, $j$ and $k$ respectively. As illustrated in Fig.~\ref{15c} and \ref{15d}, the typical reconstructions of the product $\widetilde{\boldsymbol{S}}^{\boldsymbol{\sigma}_{\boldsymbol{x}}}_{\boldsymbol{x}}\widetilde{\boldsymbol{S}}^{\boldsymbol{\sigma}_{\boldsymbol{x}'}}_{\boldsymbol{x}'}\cdots$ are supported on subregions across two or three big blocks. If we consider such typical reconstructions as alternatively on $\overline{A}$ (see Fig.~\ref{16b} and \ref{16d}), i.e., as to be erased, their supports can be viewed as living on certain paths, i.e, the supports intersect with the paths and prevent the whole of such paths from being included in the subregion $A$. The proposition below proves that if the geometry of subregion $A$ satisfies certain conditions, i.e., none of certain type of connected paths is completely included in $A$, its complement $\overline{A}$ suffices to support the recoveries of the product $\widetilde{\boldsymbol{S}}^{\boldsymbol{\sigma}_{\boldsymbol{x}}}_{\boldsymbol{x}}\widetilde{\boldsymbol{S}}^{\boldsymbol{\sigma}_{\boldsymbol{x}'}}_{\boldsymbol{x}'}\cdots$, which generalize the reconstruction shown in Fig.~\ref{15c} and \ref{15d}. 

\begin{proposition}\label{qaaaa}
In the setting of our code, we take advantage of the above notations and consider the central hole $\boldsymbol{x}$. Then, either of the following two conditions on whether a subregion $A$ can include non-repeating connected paths within a block can ensure the existence of a reconstructions of the product $\widetilde{\boldsymbol{S}}^{\boldsymbol{\sigma}_{\boldsymbol{x}}}_{\boldsymbol{x}}\widetilde{\boldsymbol{S}}^{\boldsymbol{\sigma}_{\boldsymbol{x}'}}_{\boldsymbol{x}'}\cdots$ with $\widetilde{\boldsymbol{S}}^{\boldsymbol{\sigma}_{\boldsymbol{x}}}_{\boldsymbol{x}}$ included and $\boldsymbol{\sigma}_{\boldsymbol{x}}\ne\boldsymbol{0}$ on the complement $\overline{A}$.

(1). On any two blocks out of the three, the subregion $A$ does not include any of the four types of paths: paths linking qudits $i$ to $j'$, paths linking qudits $i$ to $l$, paths linking qudits $i'$ to $k$, or paths linking $i'$ to $l''$ (see Fig.~\ref{16a}).

(2). In any direction to view the alternative geometry, $A$ does not include any of the six types of paths: paths linking qudits $i'$ to $k$, paths linking qudits $i'$ to $l''$, paths linking qudits $j$ to $k'$, paths linking $j$ to $l'$, paths linking $l$ to $i$, or paths linking $l$ to $j'$(see Fig.~\ref{16c}).
\end{proposition}

The meaning of the conditions in the above proposition for our arguments is that if we reversely include such paths in $A$, it would be necessary for $A$ to support the $\dyad{\boldsymbol{\beta}_{\boldsymbol{x}}}$s. The proof of the above proposition is given in App.~\ref{poqaaaa}, which also proves the following corollary.

\begin{corollary}\label{qaaaaa}
In the condition (1) of Prop.~\ref{qaaaa}, the minimal number of qudits in $\overline{A}$ which can prevent the existence of non-repeating connected paths in $A$ linking qudits $i$ to $j'$ or to $l$, is upper bounded by $N_0/2$.
\end{corollary}

\subsubsection{Four possibilities of minimal recoveries}
Now, we suppose that $A$ is a minimal boundary subregion supporting the reconstruction of $\mathcal{M}(\boldsymbol{x})$. According to Prop.~\ref{qaaaa}, subregion $A$ must include certain non-repeating connected paths in the big blocks so that the complement $\overline{A}$ cannot support any reconstruction of the $\widetilde{\boldsymbol{S}}^{\boldsymbol{\sigma}}_{\boldsymbol{x}}$ operators. Based on this observation, we can study how certain paths are necessarily included in $A$. We will show that this study inherently unveils all the possible minimal subregions.

For convenience in the following arguments, for a non-repeating connected path with one big block, we call it of type-1 if it links the two ends of a lateral, e.g., a path linking qudits $i$ to $j'$ as in condition (1) of Prop.~\ref{qaaaa} (see Fig.~\ref{16f}); we call it of type-2 if it links one end of a lateral and the top in the same block, e.g., a path linking qudits $i'$ to $l''$ (see Fig.~\ref{16e}). Accordingly, we can classify all the possible minimal subregions into two classes: \emph{i)} On each of the three big blocks, $A$ includes at least a type-1 path or a type-2 path. \emph{ii)} On one block, $A$ includes neither type-1 path nor type-2 path. Note that $A$ must include paths within at least two blocks, otherwise condition (1) in Prop.~\ref{qaaaa} will be satisfied.

According to this classification, we directly describe four possibilities on how the subregion $A$ include the necessary paths, and show how these paths themselves, or together with additional qudits, can form four possible minimal subregion for reconstructing $\mathcal{M}(\boldsymbol{x})$. Then, we will discuss how to show that any possible minimal recovery can be viewed as one of the four possibilities.

In class \emph{i)}, there are two possibilities for the paths to negate the two conditions: $\mathfrak{a}$: Three type-2 paths lie within the three blocks respectively (see Fig.~\ref{16e}). $\mathfrak{b}$: A type-1 path lies in one block, and two type-2 paths lie within the other two blocks respectively (see Fig.~\ref{16f}). In class \emph{ii)}, we can show that paths described in either of the following two possibilities must be present in $A$. $\mathfrak{c}$: Two type-1 paths lie within two blocks respectively (see Fig.~\ref{16g}), coinciding with the typical reconstruction of the $\dyad{\boldsymbol{\beta}_{\boldsymbol{x}}}$s as mentioned in Par.~\ref{trdyad} (see Fig.~\ref{12f}). $\mathfrak{d}$: Both a type-1 path and a type-2 path lie within one block, and a type-2 paths lies within another block (see Fig.~\ref{16h}). Note that to ensure the reconstruction of all the $\dyad{\boldsymbol{\beta}_{\boldsymbol{x}}}$s on a subregion $A$, paths described in $\mathfrak{a}$, $\mathfrak{b}$, $\mathfrak{c}$ and $\mathfrak{d}$ must be present. Otherwise, if any qudit in these paths is removed, it is easy to see that certain product $\widetilde{\boldsymbol{S}}^{\boldsymbol{\sigma}_{\boldsymbol{x}}}_{\boldsymbol{x}}\widetilde{\boldsymbol{S}}^{\boldsymbol{\sigma}_{\boldsymbol{x}'}}_{\boldsymbol{x}'}\cdots$ including the $\widetilde{\boldsymbol{S}}^{\boldsymbol{\sigma}_{\boldsymbol{x}}}_{\boldsymbol{x}}$ operators ($\boldsymbol{\sigma}_{\boldsymbol{x}}\ne\boldsymbol{0}$) can be reconstructed on the complement $\overline{A}$ in a similar fashion as in Fig.~\ref{16b} and \ref{16d}.

Indeed, the paths in $\mathfrak{a}$ or $\mathfrak{b}$ are not only necessary but also sufficient to support a reconstruction of $\mathcal{M}(\boldsymbol{x})$. That is because we can view them as examples of typical reconstructions of the $\dyad{\boldsymbol{\beta}_{\boldsymbol{x}}}$s, which can be easily derived from the criterion in Par.~\ref{simpcri}~\footnote{That is to consider the features on the paths whether it exactly implies the parity on the loops surrounding the central hole in specifying the emergent bulk degrees of freedom. For $\mathfrak{a}$, the correspondence between the parity is: (1) (even, even, even) or (odd, odd, odd) on the three paths corresponds to (even, even, even) on the loop; (2) (even, odd, odd) or (odd, even, even) on the paths corresponds to (odd, even, odd) on the loop... It is important to note that for each $\ket{\psi_m}$ the total number of dark sides on each of the three laterals of the alternative geometry is even.}, and also as typical reconstructions of the $\widetilde{\boldsymbol{S}}^{\boldsymbol{\sigma}}_{\boldsymbol{x}}$s (comparing Fig.~\ref{15c} and Fig.~\ref{16i}), which are described in Sec.~\ref{trs} (see Fig.~\ref{15c} and \ref{15d}). It follows that if a boundary subregion $A$ consists of paths in $\mathfrak{a}$ or $\mathfrak{b}$, it already supports a reconstruction of $\mathcal{M}(\boldsymbol{x})$. In addition, since no more qudit can be removed from the paths (otherwise the conditions in Prop.~\ref{qaaaa} will be satisfied), such a subregion $A$ supports a minimal reconstruction of $\mathcal{M}(\boldsymbol{x})$. These cases of subregions for minimal recoveries are illustrated in Fig.~\ref{16i} and \ref{16j} together with their illustration when rearranged to the standard geometric setting.

In $\mathfrak{c}$ and $\mathfrak{d}$, according to similar arguments as above, the paths surely reconstruct the the $\dyad{\boldsymbol{\beta}_{\boldsymbol{x}}}$s. However, since subregion $A$ does not include any of the type-1 or type-2 path in the third block, it is possible that $\overline{A}$ can include a type-1 path in this block so that the $\dyad{\boldsymbol{\beta}_{\boldsymbol{x}}}+\dyad{\boldsymbol{\beta}'_{\boldsymbol{x}}}$ operators can be reconstructed on $\overline{A}$ (see Sec.~\ref{trddyad}), which do not commute with all the $\widetilde{\boldsymbol{S}}^{\boldsymbol{\sigma}}_{\boldsymbol{x}}$ operators (see Eq.~\ref{ctss}). Then, if $\mathcal{M}(\boldsymbol{x})$ is reconstructed on $A$, $A$ must include a minimal subset of qudits on the third block to prevent any type-1 path to be included in $\overline{A}$, as illustrated in Fig.~\ref{16k} and \ref{16l}.

It is easy to show that preventing any type-1 path on a block is equivalent to preventing both type-1 and type-2 paths sharing the same end of a lateral, which is exactly the condition described in condition (1) of Prop.~\ref{qaaaa}. Hence, according to Cor.~\ref{qaaaaa}, the number of qudits in this minimal subset scales linearly or sublinearly on $N_0$. And according to the proof of Prop.~\ref{qaaaa} (see App.~\ref{poqaaaa}) and the typical reconstruction of the $\widetilde{\boldsymbol{S}}^{\boldsymbol{\sigma}}_{\boldsymbol{x}}$ operators (see Sec.~\ref{trs}), these necessary qudits included in $A$ together with one of the path can reconstruct one $\widetilde{\boldsymbol{S}}^{\boldsymbol{\sigma}}_{\boldsymbol{x}}$ operator with $\boldsymbol{\sigma}_{\boldsymbol{x}}\ne\boldsymbol{0}$. Similarly, the two paths can reconstruct another $\widetilde{\boldsymbol{S}}^{\boldsymbol{\sigma}'}_{\boldsymbol{x}}$ operator with $\boldsymbol{\sigma}'_{\boldsymbol{x}}\ne\boldsymbol{0}$. Hence, the two paths together with the minimal subset of qudits can already reconstruct $\widetilde{\boldsymbol{S}}^{\boldsymbol{\sigma}}_{\boldsymbol{x}}$, $\widetilde{\boldsymbol{S}}^{\boldsymbol{\sigma}'}_{\boldsymbol{x}}$ and $\widetilde{\boldsymbol{S}}^{\boldsymbol{\sigma}''}_{\boldsymbol{x}}=\widetilde{\boldsymbol{S}}^{\boldsymbol{\sigma}}_{\boldsymbol{x}}\widetilde{\boldsymbol{S}}^{\boldsymbol{\sigma}'}_{\boldsymbol{x}}$, and hence reconstructs $\mathcal{M}(\boldsymbol{x})$. Furthermore, it is easy to see that no more qudit can be removed, then the two paths together with the minimal subset is a minimal subregion. These cases together with their geometry after rearranged to the standard setting, as illustrated in Fig.~\ref{16k} and \ref{16l}.

Note that $\mathfrak{c}$ describes the smallest subregion for recovering $\mathcal{M}(\boldsymbol{x})$ which only consists of the two type-1 paths plus one more qudit at the corner in the other block (see Fig.~\ref{16k}). The size of this subregion is $N_0+1$. Moreover, according to the rearrangement between the alternative geometry and the standard 1D geometry, $\mathfrak{c}$ also describes the possibility with smallest connected price (see Fig.~\ref{14d}).

It can be shown that any minimal subregion for recovering $\mathcal{M}(\boldsymbol{x})$ can be viewed as belonging to one of the four possibilities $\mathfrak{a}$, $\mathfrak{b}$, $\mathfrak{c}$ and $\mathfrak{d}$. Indeed, if we classify all the subregion $A$ that can recover $\mathcal{M}(\boldsymbol{x})$ according to the number and the type of paths included in $A$ on each block, it will be clearly that in any specific case, either it essentially matches the four possibilities with difference only in the symmetry or the deformation of the paths, or it includes extra paths or qudits that can be removed to reach the four possibilities. Based on this observation, we can reach the following theorem.

\begin{theorem}\label{minithm}
Consider the central bulk qudit $\boldsymbol{x}$ in our code. If $A$ is a minimal boundary subregion on which $\mathcal{M}(\boldsymbol{x})$ can be reconstructed, then $A$ must fall into one possibility of $\mathfrak{a}$, $\mathfrak{b}$, $\mathfrak{c}$ and $\mathfrak{d}$ as described above and illustrated in Fig.~\ref{16i}, \ref{16j}, \ref{16k} and \ref{16l}. And such a minimal subregion $A$ scales linearly on $N_0$, or sublinearly on $N$ as $\sim N^{1/h}$.
\end{theorem}

\subsubsection{Minimal recoveries of non-central bulk qudits}
To describe all possible minimal subregions for reconstructing $\mathcal{M}(\boldsymbol{x})$ with $\boldsymbol{x}$ not in the center, we can apply similar method, including typical reconstructions and necessary paths as described in Prop.~\ref{qaaaa}. In this case, the possibility $\mathfrak{c}$ still describes a minimal recoveries and shows the smallest size for both the price and the connected price. The difference is that when specify the reconstructions of the $\dyad{\boldsymbol{\beta}_{\boldsymbol{x}}}$s with the criterion developed in Par.~\ref{simpcri}, we cannot directly utilize the fact that in each $\ket{\psi_m}$ the number of darks side on each laterals of the alternative geometry is even. Hence, possibilities $\mathfrak{a}$, $\mathfrak{b}$ and $\mathfrak{d}$ need to be changed according to a slightly modified version of Prop.~\ref{qaaaa}. Completing the detailed arguments should be straightforward based on the above discussion. Here, we only indicate that the key in the difference: A modified version of type-2 path needs to link an end of a lateral of the loop surround the non-central hole $\boldsymbol{x}$ to the top of the big block with respect to the central hole, instead of the smaller block. Then, the modified $\mathfrak{a}$ simply corresponds to the expectation that the boundary reconstruction of $\boldsymbol{x}$ on the farther side of the boundary is more consuming than that in the nearer boundary.

Now, based on all the above arguments for both central and non-central bulk qudit, we can already reach the price and the connected price.

\begin{corollary}
In our code, the price of the bulk qudits takes the value $\mathrm{p}(\boldsymbol{x})=N_0+1,N_0/2+1,N_0/4+1,\ldots$. The connected price takes the value $\mathrm{p_c}(\boldsymbol{x})\sim (68/100)N,(68/300)N,(68/900)N,\ldots.$
\end{corollary}

\section{Subregion duality and RT formula in OAQEC}\label{sdrt}
Based on the results of the previous sections, especially on the typical reconstructions of bulk local operators (see Sec.~\ref{rcgmx}) and the minimal recoveries of bulk local operator algebras (see Sec.~\ref{minirc}), we systematically demonstrate the HQEC characteristics regarding subregion duality, and show the entanglement entropy as a version of the RT/FLM formula (see Characteristic 3, 4.1, 4.2 and 5 in Par.~\ref{chac3} and Sec.~\ref{sdew}). The illustrations of subregion duality for disconnected boundary bipartitions will also complete the demonstration for uberholography (see Characteristic 6 in Par.~\ref{chac6}). Indeed, a systematic demonstration of subregion duality in the formalism of genuine OAQEC has been rarely reported. Therefore, our demonstrations could potentially serve as a useful reference for future studies.

As expected in the previous discussion (see Sec.~\ref{sdew}), our arguments will embody a framework which translates between the algebraic and the geometric aspects of subregion duality in the formalism of genuine OAQEC. Upon the demonstration of the genuine OAQEC (Characteristic 3) for boundary bipartitions $A\overline{A}$, our framework can be presented in the following aspects: We (1) formalize the meaning of ``splits'' for bulk qudits in the entangling surface; (2) specify the entangling surface and the entanglement wedges according to the splits; (3) elaborate the explicit structure of the two von Neumann algebras $R^+\mathcal{M}_A R$ and $R^+\mathcal{M}_{\overline{A}}R$ in terms of the splits in the form of von Neumann algebra tensor product; (4) and realize the ``standard'' prescription for extracting the bulk term and the area term in the RT formula~\cite{harlow2017}. Indeed, while (1) and (2) might depend on the specifics of exact models, the way to achieve (3) and to specify the setting for (4) is general, and the form of Eq.~\ref{ew2}, Prop.~\ref{sova} together with Prop.~\ref{declprop} can be applied generally.


One might make an analogy between the roles of the framework described below and the greedy algorithm developed in the tensor network paradigm~\cite{pastawski2015,jahn2021} where the later can work efficiently in specifying the greedy wedge in the subsystem-code formalism. As we will show, our framework enables us to capture the the exact and complete description of subregion duality in the genuine OAQEC formalism, including both the algebraic and the geometric aspects. Accordingly, not only can we demonstrate the expected properties of subregion duality, but we can also conduct future studies analytically based on these properties.

\onecolumngrid
\begin{center}
\begin{figure}[ht]
\centering
    \includegraphics[width=15cm]{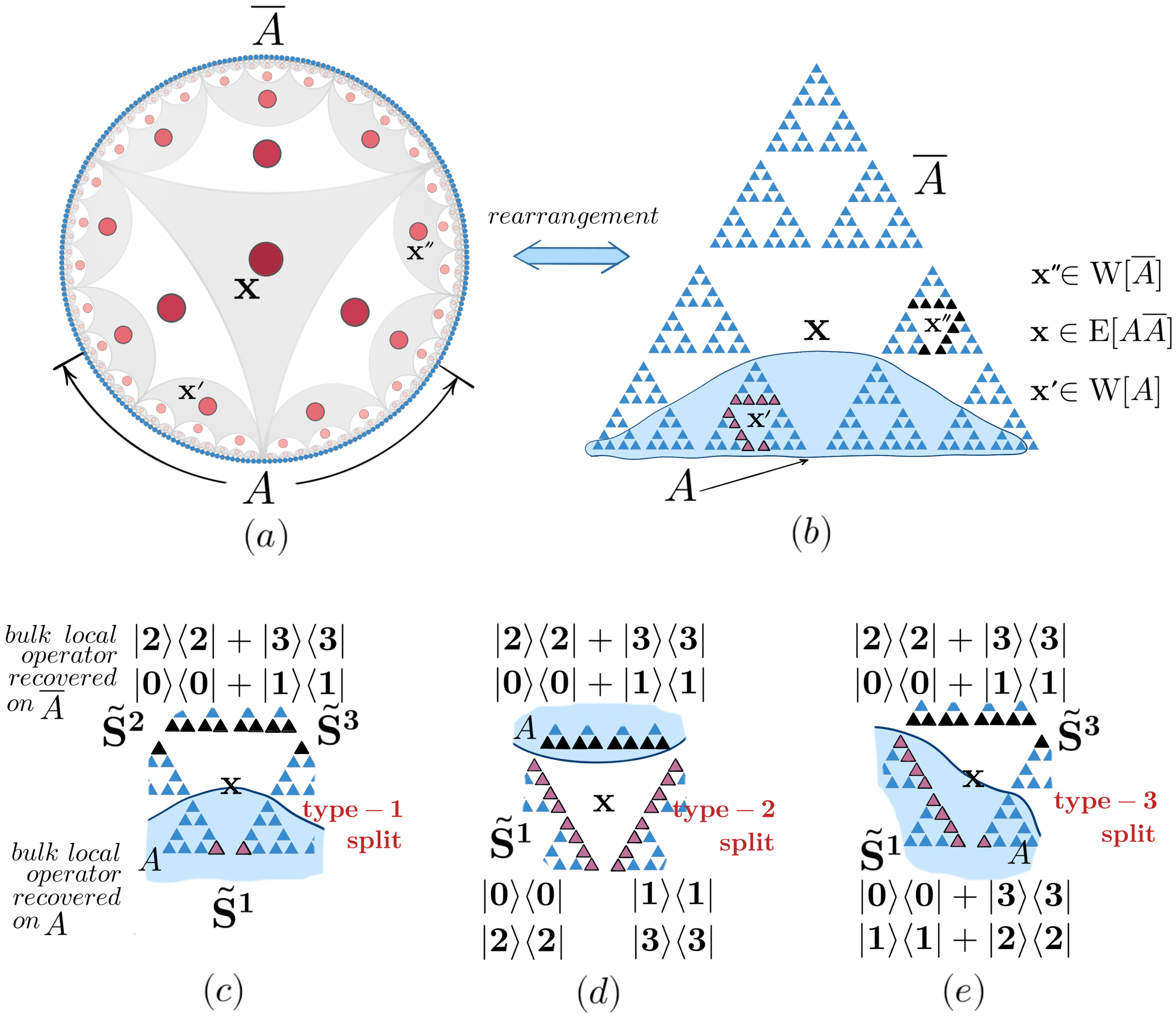}   
\phantomsubfloat{\label{17a}}\phantomsubfloat{\label{17b}}
\phantomsubfloat{\label{17c}}\phantomsubfloat{\label{17d}}
\phantomsubfloat{\label{17e}}
\caption{A connected boundary bipartition $A\overline{A}$ in the standard geometric setting, as illustrated in (a), has a more compact presentation in the alternative geometry, as illustrated in (b). In (b), the purple-highlighted indicate a minimal subregion lying within $A$ for recovering $\mathcal{M}(\boldsymbol{x}')$ of some bulk qudits (holes) $\boldsymbol{x}'$ (see Fig.~{16k} for reference), and similarly the black-highlighted indicate a minimal subregion lying within $\overline{A}$ for recovering $\mathcal{M}(\boldsymbol{x}'')$ of some bulk qudits (holes) $\boldsymbol{x}''$. (c) (d) (e) represent the three types of splits that can be identified in (b) or other bipartitions. As listed for each type, we can identify typical reconstructions of bulk local operators according to the results of the preceding section. The highlighted indicate the subregions for those typical reconstructions (see Fig.~\ref{fig12}, \ref{fig13} and \ref{fig15} for references).}
\label{fig17}
\end{figure}
\end{center}
\twocolumngrid

\subsection{Split of bulk local operator algebra}\label{splitsec}

We keep considering the code as in both the standard geometric setting and the alternative geometric setting, and keep utilizing the rearrangement bridging the two geometries (see Fig.~\ref{fig4}, \ref{fig5} and \ref{fig6}). Now, to demonstrate the HQEC characteristics regarding subregion duality in the standard geometric setting, we take advantage of what we have investigated the code in the alternative geometry---the typical reconstructions of operators $R(\cdots\otimes\dyad{\boldsymbol{\beta}_{\boldsymbol{x}}}\otimes\cdots)R^+$, $R(\cdots\otimes\dyad{\boldsymbol{\beta}_{\boldsymbol{x}}}+\dyad{\boldsymbol{\beta}'_{\boldsymbol{x}}}\otimes\cdots)R^+$ and $R(\cdots\otimes\widetilde{\boldsymbol{S}}^{\boldsymbol{\sigma}}_{\boldsymbol{x}}\otimes\cdots)R^+$ (see Sec.~\ref{rcgmx}) for single bulk qudit $\boldsymbol{x}$ (or a hole $\boldsymbol{x}$), together with the minimal recoveries of $\mathcal{M}(\boldsymbol{x})$ (see Sec.~\ref{minirc}). The basic idea is as follows: Whether a subregion $A$ of physical/boundary qudits can completely or partially recover a logical/bulk qudit $\boldsymbol{x}$ relies on whether a minimal recovery of the whole bulk local operator algebra $\mathcal{M}(\boldsymbol{x})$ or the reconstructions of certain operators in the algebra are supported on $A$; and classifying the bulk qudits according to such different ways of recoveries is expected to specify the entanglement wedges defined in terms of $\mathrm{W}[A]$, $\mathrm{W}[\overline{A}]$ and $\mathrm{E}[A\overline{A}]$ (see Sec.~\ref{sdew}).


More explicitly as illustrated in Fig.~\ref{fig17}, a boundary bipartition $A\overline{A}$ in the standard geometry (see Fig.~\ref{17a}) has a more compact presentation in the alternative geometry (see Fig.~\ref{17b}). Then, based on the latter and according to the results on the minimal recoveries, for certain bulk qudits or holes, e.g., $\boldsymbol{x}'$ and $\boldsymbol{x}''$ in Fig.~\ref{17a} and \ref{17b}, we can identify minimal subregions for the recoveries of $\mathcal{M}(\boldsymbol{x}')$ and $\mathcal{M}(\boldsymbol{x}'')$ within the subregions $A$ and $\overline{A}$ respectively. This simply implies that the bulk qudit $\boldsymbol{x}'$ and $\boldsymbol{x}''$ belong to $\mathrm{W}[A]$ and $\mathrm{W}[\overline{A}]$ respectively. Indeed, thanks to the minimal recoveries described in Sec.~\ref{minirc} and the convenience of the alternative geometry, we can identify all such bulk qudits and hence specify $\mathrm{W}[A]$ and $\mathrm{W}[\overline{A}]$.

It is crucial to note that for each bulk qudit (holes) in the illustration, if we cannot identify any minimal subregion for the recovery of $\mathcal{M}(\boldsymbol{x})$ within $A$ or $\overline{A}$, e.g., for qudits like $\boldsymbol{x}$ in Fig.~\ref{fig17}, then we can always identify certain nontrivial bulk local operators reconstructable on $A$ and $\overline{A}$ respectively. As illustrated in Fig.~\ref{17c}, \ref{17d} and \ref{17e}, according to the results in the preceding section, for each of the three types of $\boldsymbol{x}$ we can present typical bulk local operators of the form $\dyad{\boldsymbol{\beta}_{\boldsymbol{x}}}$, $\dyad{\boldsymbol{\beta}_{\boldsymbol{x}}}+\dyad{\boldsymbol{\beta}'_{\boldsymbol{x}}}$ or $\widetilde{\boldsymbol{S}}^{\boldsymbol{\sigma}}_{\boldsymbol{x}}$ that can be reconstructed on $A$ and $\overline{A}$. This observation exactly illustrates how the bulk local operator algebra can split (see Sec.~\ref{sdew}): (1) It shows clues to prove the nontrivial center for the genuine OAQEC: we can search the nontrivial typical reconstructions which can be supported on both $A$ and $\overline{A}$ in order to ensure the nontrivial center $\mathrm{Z}(R^+\mathcal{M}_A R)$. (2) It can also guide the specification of the entangling surface $\mathrm{E}[A\overline{A}]$ which is the characterization of subregion duality in genuine OAQEC (see Sec.~\ref{sdew}).

What is not apparent from the illustration is how to specify the complete structure of the von Neumann algebras $R^+\mathcal{M}_A R$, $R^+\mathcal{M}_{\overline{A}}R$ and $\mathrm{Z}(R^+\mathcal{M}_A R)$, and how to explicitly extract/compute the bulk and area terms in the RT formula (see Sec.~\ref{sdew}). A clue for the former is that if $\widetilde{\boldsymbol{O}}_1,\widetilde{\boldsymbol{O}}_2,\ldots$ can be reconstructed on $A$, the von Neumann algebra generated by them can also be reconstructed on $A$ (see the definition of generators of a von Neumann algebra in Sec.~\ref{basic}). Hence, it is natural to expect that those split typical reconstructable operators together with those fully reconstructed, though not yet presenting all operators in a von Neumann algebra, can generate $R^+\mathcal{M}_A R$ and $R^+\mathcal{M}_{\overline{A}}R$.

In the following, we generalize and formalize the above heuristic arguments into a framework within which we show how to reach the RT formula for an arbitrary state $\widetilde{\rho}$ in $\mathcal{H}_{\mathrm{code}}$. A key question to be answered is whether the split typical reconstructions (see Fig.~\ref{17c}, \ref{17d} and \ref{17e}) are sufficient to represent or generate all the split bulk local operators and hence sufficiently capture the split. To address this question, we prove a condition of local complementarity on the bulk qudits to formalize the meaning of bulk local splits, which applies generally to the boundary bipartitions considered in our demonstration and underlies the framework.

\subsubsection{Bulk local complementarity and formal meaning of splits}

To elucidate this condition, we start with the following consideration which can apply generally. That is, for each bulk qudit $\boldsymbol{x}$, we can generally specify two von Neumann algebras on $\mathcal{H}_{\mathrm{code}}$: $\mathcal{M}_a(\boldsymbol{x})=\mathcal{M}(\boldsymbol{x})\cap\mathcal{M}_A$ consisting of all logical operators that lie within $\mathcal{M}(\boldsymbol{x})$ and can be reconstructed on $A$; $\mathcal{M}_{\overline{a}}(\boldsymbol{x})=\mathcal{M}(\boldsymbol{x})\cap\mathcal{M}_{\overline{A}}$ consisting of those corresponding to $\overline{A}$. Note that the subscripts follow the convention in the literature~\cite{harlow2017,pastawski2017,cao2021}: on the boundary we use the capital $A$ and $\overline{A}$, while in the bulk we use the lowercase $a$ and $\overline{a}$.

Equivalently, through the algebra isomorphism $\mathbf{L}(\mathcal{H}_{\mathrm{code}})\xrightarrow{R^+\boldsymbol{\cdot} R}\mathbf{L}(\mathcal{E})$, the above von Neumann algebras on $\mathcal{H}_{\mathrm{code}}$ uniquely specifies two von Neumann algebras on the single-bulk-qudit Hilbert space $\mathfrak{e}_{\boldsymbol{x}}$,
\begin{align}\label{structure0}
\begin{split}
&\quad\quad\quad\boldsymbol{\mathcal{M}}_a(\boldsymbol{x})\subset\mathbf{L}(\mathfrak{e}_{\boldsymbol{x}}) \quad\quad \boldsymbol{\mathcal{M}}_{\overline{a}}(\boldsymbol{x})\subset\mathbf{L}(\mathfrak{e}_{\boldsymbol{x}}),\\
&\cdots\otimes\mathbb{C}\mathds{1}_{\mathfrak{e}_{\boldsymbol{x}'}}\otimes\boldsymbol{\mathcal{M}}_a(\boldsymbol{x})\otimes\mathbb{C}\mathds{1}_{\mathfrak{e}_{\boldsymbol{x}''}}\otimes\cdots\\
&\quad\quad\quad\quad\quad\quad\quad\quad=R^+\mathcal{M}_a(\boldsymbol{x})R\subset R^+\mathcal{M}_A R,\\
&\cdots\otimes\mathbb{C}\mathds{1}_{\mathfrak{e}_{\boldsymbol{x}'}}\otimes\boldsymbol{\mathcal{M}}_{\overline{a}}(\boldsymbol{x})\otimes\mathbb{C}\mathds{1}_{\mathfrak{e}_{\boldsymbol{x}''}}\otimes\cdots\\
&\quad\quad\quad\quad\quad\quad\quad\quad=R^+\mathcal{M}_{\overline{a}}(\boldsymbol{x})R\subset R^+\mathcal{M}_{\overline{A}} R. 
\end{split}
\end{align}
Here, $\boldsymbol{\mathcal{M}}_a(\boldsymbol{x})$ and $\boldsymbol{\mathcal{M}}_{\overline{a}}(\boldsymbol{x})$ simply consist of the local operators that can be reconstructed on $A$ and on $\overline{A}$ respectively. And the two tensor products of local operator algebras consist of operators of the form $\cdots\otimes\mathds{1}_{\mathfrak{e}_{\boldsymbol{x}'}}\otimes\widetilde{\boldsymbol{O}}_{\boldsymbol{x}}\otimes\mathds{1}_{\mathfrak{e}_{\boldsymbol{x}''}}\otimes\cdots$ with $\widetilde{\boldsymbol{O}}_{\boldsymbol{x}}\in\boldsymbol{\mathcal{M}}_a(\boldsymbol{x})$ or $\widetilde{\boldsymbol{O}}_{\boldsymbol{x}}\in\boldsymbol{\mathcal{M}}_{\overline{a}}(\boldsymbol{x})$ respectively.

\paragraph*{\textbf{Bulk local complementarity}}\label{blcpar} Now, for a given boundary bipartition $A\overline{A}$, the condition of complementarity formally means that for any bulk qudit $\boldsymbol{x}$, $\mathcal{M}_{\overline{a}}(\boldsymbol{x})$ can be viewed as consisting of all logical operators that lie within $\mathcal{M}(\boldsymbol{x})$ and commute with every operator in $\mathcal{M}_a(\boldsymbol{x})$, i.e. $\mathcal{M}_{\overline{a}}(\boldsymbol{x})=\mathcal{M}(\boldsymbol{x})\cap\mathcal{M}'_a(\boldsymbol{x})$. Or equivalently, 
\begin{equation}\label{structure1}
\boldsymbol{\mathcal{M}}_{\overline{a}}(\boldsymbol{x})=\boldsymbol{\mathcal{M}}'_a(\boldsymbol{x})~in~\mathbf{L}(\mathfrak{e}_{\boldsymbol{x}}).
\end{equation}

The importance of the bulk local complementarity for our demonstration is based on the following fact which can be straightforward checked for all connected boundary bipartitions, and also for all disconnected bipartitions of importance in the demonstration of uberholography. That is, similar to the illustration in Fig.~\ref{fig17}, based on the results in the preceding section and taking advantage of the alternative geometry, if we cannot identify a minimal subregion for recovering $\mathcal{M}(\boldsymbol{x})$ of a bulk qudit $\boldsymbol{x}$ within $A$ or $\overline{A}$, we can always identify nontrivial typical reconstructions of its bulk local operators on $A$ and $\overline{A}$ respectively, i.e., identifying split.

Then, as will be shown later, the importance include the following aspects: (1) In the case that we can identify split, the general structures $\boldsymbol{\mathcal{M}}_a(\boldsymbol{x})$ and $\boldsymbol{\mathcal{M}}_{\overline{a}}(\boldsymbol{x})$ are simply generated by the nontrivial typical reconstructions as illustrated in Fig.~\ref{fig17}, and properties of these typical reconstructions can guarantee $\boldsymbol{\mathcal{M}}_{\overline{a}}(\boldsymbol{x})=\boldsymbol{\mathcal{M}}'_a(\boldsymbol{x})$. (2) In this way, the pair $(\boldsymbol{\mathcal{M}}_a(\boldsymbol{x}),\boldsymbol{\mathcal{M}}_{\overline{a}}(\boldsymbol{x}))$ endows the split of $\mathbf{L}(\mathfrak{e}_{\boldsymbol{x}})$ with a formal meaning, which can be completely described by the typical reconstructions of bulk local operators. (3) Based on the formal meaning of splits, merely with the typical reconstructions, can we completely characterize the structures of $R^+\mathcal{M}_A R$, $R^+\mathcal{M}_{\overline{A}}R$ and $\mathrm{Z}(R^+\mathcal{M}_A R)$. (4) The splits are essential for extracting the bulk and the area terms from the RT formula following the ``standard'' prescription.

Note that according to what we have discussed in Sec.~\ref{sdew}, if for any $\boldsymbol{x}$ the whole of $\mathcal{M}({\boldsymbol{x}})$ can be reconstructed on either $A$ or $\overline{A}$, i.e., $\boldsymbol{\mathcal{M}}_{a}=\mathbf{L}(\mathfrak{e}_{\boldsymbol{x}}),\boldsymbol{\mathcal{M}}_{\overline{a}}=\mathbb{C}\mathds{1}_{\mathfrak{e}_{\boldsymbol{x}}}$ or $\boldsymbol{\mathcal{M}}_{a}=\mathbb{C}\mathds{1}_{\mathfrak{e}_{\boldsymbol{x}}},\boldsymbol{\mathcal{M}}_{\overline{a}}=\mathbf{L}(\mathfrak{e}_{\boldsymbol{x}})$, then the condition of geometric complementarity is satisfied and we will be in the subsystem-code formalism so that the structure of the von Neumann algebras is simply the tensor product of those $\mathbf{L}(\mathfrak{e}_{\boldsymbol{x}})$s. In this case, we definitely have the condition of bulk local complementarity, though the split is trivial. As will be clear, this simplicity does not happen to our code. Instead, we can show the richness of splits and how such richness shapes $R^+\mathcal{M}_A R$ and $R^+\mathcal{M}_{\overline{A}}R$.

\subsubsection{Generality of bulk local complementarity for splits}
To confirm the generality of bulk local complementarity in our code, we start with considering how to prove the condition for a given bipartition $A\overline{A}$. For each bulk qudit $\boldsymbol{x}$ exhibiting split, we denote the two von Neumann algebras generated by the typical local reconstructions (of the form $\dyad{\boldsymbol{\beta}_{\boldsymbol{x}}}$, $\dyad{\boldsymbol{\beta}_{\boldsymbol{x}}}+\dyad{\boldsymbol{\beta}'_{\boldsymbol{x}}}$ or $\widetilde{\boldsymbol{S}}^{\boldsymbol{\sigma}}_{\boldsymbol{x}}$) on $A$ and $\overline{A}$ (see Fig.~\ref{17c}, \ref{17d} and \ref{17e}) by $\boldsymbol{\mathcal{N}}_{\boldsymbol{x}}$ and $\overline{\boldsymbol{\mathcal{N}}}_{\boldsymbol{x}}$ respectively. Obviously, any operator in $\boldsymbol{\mathcal{N}}_{\boldsymbol{x}}$ ($\overline{\boldsymbol{\mathcal{N}}}_{\boldsymbol{x}}$), as a linear sum of products of those typical operators, can be reconstructed on $A$ ($\overline{A}$), i.e., we have $\boldsymbol{\mathcal{N}}_{\boldsymbol{x}}\subset\boldsymbol{\mathcal{M}}_a(\boldsymbol{x})$, $\overline{\boldsymbol{\mathcal{N}}}_{\boldsymbol{x}}\subset\boldsymbol{\mathcal{M}}_{\overline{a}}(\boldsymbol{x})$. Then, if we can prove that $\overline{\boldsymbol{\mathcal{N}}}_{\boldsymbol{x}}=\boldsymbol{\mathcal{N}}'_{\boldsymbol{x}}$, i.e., one is the commutant of the other, by simple arguments~\footnote{By definition, we have $\boldsymbol{\mathcal{M}}_{\overline{a}}(\boldsymbol{x})\subset\boldsymbol{\mathcal{M}}'_a(\boldsymbol{x})$. Then, $\overline{\boldsymbol{\mathcal{N}}}_{\boldsymbol{x}}=\boldsymbol{\mathcal{N}}_{\boldsymbol{x}}'\subset\boldsymbol{\mathcal{M}}_{\overline{a}}(\boldsymbol{x})\subset\boldsymbol{\mathcal{M}}'_a(\boldsymbol{x})$ means that $\boldsymbol{\mathcal{M}}_a(\boldsymbol{x})\subset\boldsymbol{\mathcal{N}}_{\boldsymbol{x}}$. Combining $\boldsymbol{\mathcal{N}}_{\boldsymbol{x}}\subset\boldsymbol{\mathcal{M}}_a(\boldsymbol{x})$ in the main context, we have $\boldsymbol{\mathcal{N}}_{\boldsymbol{x}}=\boldsymbol{\mathcal{M}}_a(\boldsymbol{x})$. Similarly, we can show $\overline{\boldsymbol{\mathcal{N}}}_{\boldsymbol{x}}=\boldsymbol{\mathcal{M}}_{\overline{a}}(\boldsymbol{x})$. And according to the condition $\overline{\boldsymbol{\mathcal{N}}}_{\boldsymbol{x}}=\boldsymbol{\mathcal{N}}_{\boldsymbol{x}}'$ in the main context, we have $\boldsymbol{\mathcal{M}}_{\overline{a}}(\boldsymbol{x})=\boldsymbol{\mathcal{M}}'_a(\boldsymbol{x})$.}, we can easily show $\boldsymbol{\mathcal{M}}_a(\boldsymbol{x})=\boldsymbol{\mathcal{N}}_{\boldsymbol{x}}$ and $\boldsymbol{\mathcal{M}}_{\overline{a}}(\boldsymbol{x})=\overline{\boldsymbol{\mathcal{N}}}_{\boldsymbol{x}}$. In other words, the typical reconstructable operators that we can directly identify as in Fig.~\ref{fig7} simply generate $\boldsymbol{\mathcal{M}}_a(\boldsymbol{x})$ and $\boldsymbol{\mathcal{M}}_{\overline{a}}(\boldsymbol{x})$ respectively. And automatically, we will have $\boldsymbol{\mathcal{M}}_{\overline{a}}(\boldsymbol{x})=\boldsymbol{\mathcal{M}}'_a(\boldsymbol{x})$, i.e., the local complementarity is satisfied so that $\boldsymbol{\mathcal{M}}_a(\boldsymbol{x})$ and $\boldsymbol{\mathcal{M}}_{\overline{a}}(\boldsymbol{x})$ will formally describe the split of $\mathbf{L}(\mathfrak{e}_{\boldsymbol{x}})$.

According to the above analysis, we simply need to show that for every bulk qudit exhibiting split, $\boldsymbol{\mathcal{N}}_{\boldsymbol{x}}$ and $\overline{\boldsymbol{\mathcal{N}}}_{\boldsymbol{x}}$ are commutant to each other. Indeed, we can view $\boldsymbol{\mathcal{N}}_{\boldsymbol{x}}$ and $\overline{\boldsymbol{\mathcal{N}}}_{\boldsymbol{x}}$ as von Neumann algebras on $\mathbb{C}^4$, and it is clear that we have limited number of the types of $\boldsymbol{\mathcal{N}}_{\boldsymbol{x}}$ and $\overline{\boldsymbol{\mathcal{N}}}_{\boldsymbol{x}}$ if going through all $\boldsymbol{x}$ for all bipartitions as mentioned above. Hence, if we can simply list all the types, and rigorously show that $\overline{\boldsymbol{\mathcal{N}}}_{\boldsymbol{x}}=\boldsymbol{\mathcal{N}}'_{\boldsymbol{x}}$ is always satisfied, then the desired conditions can be proved. Indeed, these listed $\boldsymbol{\mathcal{N}}_{\boldsymbol{x}}$ and $\overline{\boldsymbol{\mathcal{N}}}_{\boldsymbol{x}}$ will be ``small'' examples of OAQEC, which can construct ``large'' examples through von Neumann algebra tensor product.

\subsubsection{Examples of splits}\label{exsplitsec}
Now, we list all types of $(\boldsymbol{\mathcal{N}}_{\boldsymbol{x}},\overline{\boldsymbol{\mathcal{N}}}_{\boldsymbol{x}})$. We rigorously prove that the above desired qualifying conditions on the bulk local complementarity are always satisfied, and these $(\boldsymbol{\mathcal{N}}_{\boldsymbol{x}},\overline{\boldsymbol{\mathcal{N}}}_{\boldsymbol{x}})$ have the formal meaning of splits $(\boldsymbol{\mathcal{M}}_a(\boldsymbol{x}),\boldsymbol{\mathcal{M}}_{\overline{a}}(\boldsymbol{x}))$. Note that we will not spend texts on showing how the following types of splits cover all the cases in the boundary bipartitions considered in our demonstration, since it will be straightforward to check in a similar way to the illustration in Fig.~\ref{17a} and \ref{17b}. Hence, the following arguments can be viewed as the proof of the condition of bulk local complementarity for all connected boundary bipartitions, and also for all disconnected bipartitions of importance in the demonstration of uberholography.

In the following, each case of split will be specified by the generators of $\boldsymbol{\mathcal{N}}_{\boldsymbol{x}}$ and $\overline{\boldsymbol{\mathcal{N}}}_{\boldsymbol{x}}$, i.e., those typical reconstructions of bulk local operators on $A$ and $\overline{A}$ respectively. There are three types in total, exactly corresponding to the illustrations in Fig.~\ref{17c}, \ref{17d} and \ref{17e}, and each type include three symmetric cases.

We start with describing the first case of type-1 and sketching the proof for showing $\overline{\boldsymbol{\mathcal{N}}}_{\boldsymbol{x}}=\boldsymbol{\mathcal{N}}'_{\boldsymbol{x}}$, which is detailed in App.~\ref{posplit}. As illustrated in Fig.~\ref{17c} and with reference to Fig.~\ref{fig13} and \ref{fig15}, we can reconstruct for bulk qudit (hole) $\boldsymbol{x}$ the single-bulk-qudit operators
\begin{equation}\label{spliteq1}
\{\widetilde{\boldsymbol{S}}^{\boldsymbol{1}}_{\boldsymbol{x}},~ \mathds{1}_{\mathfrak{e}_{\boldsymbol{x}}}\}
\end{equation}
on the boundary subregion $A$, and reconstruct
\begin{equation}\label{spliteq2}
\{\widetilde{\boldsymbol{S}}^{\boldsymbol{2}}_{\boldsymbol{x}},~\widetilde{\boldsymbol{S}}^{\boldsymbol{3}}_{\boldsymbol{x}},~ \dyad{\boldsymbol0}+\dyad{\boldsymbol1},~ \dyad{\boldsymbol2}+\dyad{\boldsymbol3},~ \mathds{1}_{\mathfrak{e}_{\boldsymbol{x}}}\}
\end{equation}
on the complement subregion $\overline{A}$, which generate $\boldsymbol{\mathcal{N}}_{\boldsymbol{x}}$ and $\overline{\boldsymbol{\mathcal{N}}}_{\boldsymbol{x}}$ respectively.

In our proof, the arguments rely on a decomposition structure of the local bulk Hilbert space $\mathfrak{e}_{\boldsymbol{x}}$, 
\begin{equation}\label{decl0}
\mathfrak{e}_{\boldsymbol{x}}=\oplus_{\mu}(\mathfrak{e}_{\boldsymbol{x}a}^{\mu}\otimes\mathfrak{e}_{\boldsymbol{x}\overline{a}}^{\mu}).
\end{equation}
The decomposition describes how each tensor product $\mathfrak{e}_{\boldsymbol{x}a}^{\mu}\otimes\mathfrak{e}_{\boldsymbol{x}\overline{a}}^{\mu}$ can be identified (through an isometry) as a subspace $\widetilde{\boldsymbol{P}}^{\mu}_{\boldsymbol{x}}\mathfrak{e}_{\boldsymbol{x}}$ for some projection operator $\widetilde{\boldsymbol{P}}^{\mu}_{\boldsymbol{x}}$ on $\mathfrak{e}_{\boldsymbol{x}}$ such that these subspaces are orthogonal, i.e. $\widetilde{\boldsymbol{P}}^{\mu}_{\boldsymbol{x}}\widetilde{\boldsymbol{P}}^{\mu'}_{\boldsymbol{x}}=0$ for $\mu\ne\mu'$, and the projections are complete, i.e. $\sum_{\mu}\widetilde{\boldsymbol{P}}^{\mu}_{\boldsymbol{x}}=\mathds{1}_{\mathfrak{e}_{\boldsymbol{x}}}$. Then, operators on each tensor product, i.e., $\widetilde{\boldsymbol{O}}_{\boldsymbol{x}a}^{\mu}\otimes\mathds{1}^{\mu}_{\boldsymbol{x}\overline{a}}$, $\mathds{1}^{\mu}_{\boldsymbol{x}a}\otimes\widetilde{\boldsymbol{O}}_{\boldsymbol{x}\overline{a}}^{\mu}$, or linear sums of their products, can be identified as certain operators $\widetilde{\boldsymbol{O}}_{\boldsymbol{x}}$ on $\mathfrak{e}_{\boldsymbol{x}}$ that satisifies $\widetilde{\boldsymbol{O}}_{\boldsymbol{x}}=\widetilde{\boldsymbol{P}}^{\mu}_{\boldsymbol{x}}\widetilde{\boldsymbol{O}}_{\boldsymbol{x}}\widetilde{\boldsymbol{P}}^{\mu}_{\boldsymbol{x}}$. Note that this condition means that the action of $\widetilde{\boldsymbol{O}}_{\boldsymbol{x}}$ leaves $\widetilde{\boldsymbol{P}}^{\mu}_{\boldsymbol{x}}\mathfrak{e}_{\boldsymbol{x}}$ invariant, and $\widetilde{\boldsymbol{O}}_{\boldsymbol{x}}$ acts as the zero operator on the complement subspace orthogonal to $\widetilde{\boldsymbol{P}}^{\mu}_{\boldsymbol{x}}\mathfrak{e}_{\boldsymbol{x}}$.

According to the basic properties of von Neumann algebra on a finite-dimensional Hilbert space~\cite{harlow2017}, the important property of such a decomposition that we utilize is the following: All the operators on $\mathfrak{e}_{\boldsymbol{x}}$ that can be identified as the sums $\sum_{\mu}\widetilde{\boldsymbol{O}}_{\boldsymbol{x}a}^{\mu}\otimes\mathds{1}^{\mu}_{\boldsymbol{x}\overline{a}}$ form a von Neumann algebra on $\mathfrak{e}_{\boldsymbol{x}}$, and its commutant simply consists of all the operators that can be identified as the sums $\sum_{\mu}\mathds{1}^{\mu}_{\boldsymbol{x}a}\otimes\widetilde{\boldsymbol{O}}_{\boldsymbol{x}\overline{a}}^{\mu}$~\footnote{It is shown in Ref.~\cite{harlow2017} that any given von Neumann algebra ensures such a decomposition. Based on the orthogonality of the subspaces in a direct sum of Hilbert spaces, it is also easy to prove that any such decomposition uniquely determines a von Neumann algebra.}. Furthermore, operators in the center (within $\mathbf{L}(\mathfrak{e}_{\boldsymbol{x}})$) are simply linear combinations $\sum_{\mu}c_{\mu}\widetilde{\boldsymbol{P}}^{\mu}_{\boldsymbol{x}}$ with $c_{\mu}\in\mathbb{C}$.

According to the above decomposition, in App.~\ref{posplit}, we show that all operators of the form $\sum_{\mu}\widetilde{\boldsymbol{O}}_{\boldsymbol{x}a}^{\mu}\otimes\mathds{1}^{\mu}_{\boldsymbol{x}\overline{a}}$ can be generated by $\{\widetilde{\boldsymbol{S}}^{\boldsymbol{1}}_{\boldsymbol{x}},\mathds{1}_{\mathfrak{e}_{\boldsymbol{x}}}\}$, and reversely, $\widetilde{\boldsymbol{S}}^{\boldsymbol{1}}_{\boldsymbol{x}}$ and $\mathds{1}_{\mathfrak{e}_{\boldsymbol{x}}}$ can both be written in the form $\sum_{\mu}\widetilde{\boldsymbol{O}}_{\boldsymbol{x}a}^{\mu}\otimes\mathds{1}^{\mu}_{\boldsymbol{x}\overline{a}}$. In other words, the von Neumann algebra consisting of all $\sum_{\mu}\widetilde{\boldsymbol{O}}_{\boldsymbol{x}a}^{\mu}\otimes\mathds{1}^{\mu}_{\boldsymbol{x}\overline{a}}$ operators is exactly $\boldsymbol{\mathcal{N}}_{\boldsymbol{x}}$, i.e., generated by $\{\widetilde{\boldsymbol{S}}^{\boldsymbol{1}}_{\boldsymbol{x}},\mathds{1}_{\mathfrak{e}_{\boldsymbol{x}}}\}$. Similarly, we show that the von Neumann algebra consisting of all $\sum_{\mu}\mathds{1}^{\mu}_{\boldsymbol{x}a}\otimes\widetilde{\boldsymbol{O}}_{\boldsymbol{x}\overline{a}}^{\mu}$ operators is exactly $\overline{\boldsymbol{\mathcal{N}}}_{\boldsymbol{x}}$, i.e., generated by $\{\widetilde{\boldsymbol{S}}^{\boldsymbol{2}}_{\boldsymbol{x}},\widetilde{\boldsymbol{S}}^{\boldsymbol{3}}_{\boldsymbol{x}}, \dyad{\boldsymbol0}+\dyad{\boldsymbol1}, \dyad{\boldsymbol2}+\dyad{\boldsymbol3}, \mathds{1}_{\mathfrak{e}_{\boldsymbol{x}}}\}$. Then, we will have $\overline{\boldsymbol{\mathcal{N}}}_{\boldsymbol{x}}=\boldsymbol{\mathcal{N}}'_{\boldsymbol{x}}$. And according to the above discussion, what the two collections of operators generated respectively are exactly $\boldsymbol{\mathcal{M}}_a(\boldsymbol{x})$ and $\boldsymbol{\mathcal{M}}_{\overline{a}}(\boldsymbol{x})$ with $\boldsymbol{\mathcal{M}}_{\overline{a}}(\boldsymbol{x})=\boldsymbol{\mathcal{M}}'_a(\boldsymbol{x})$. Particularly, as shown in App.~\ref{posplit}, in this example of type-1 split, $\mathrm{Z}(\boldsymbol{\mathcal{M}}_a(\boldsymbol{x}))=\boldsymbol{\mathcal{M}}_a(\boldsymbol{x})$ which is surely nontrivial since it contains $\widetilde{\boldsymbol{S}}^{\boldsymbol{1}}_{\boldsymbol{x}}$.

According to the above results, we can easily show another two examples of the same type. We list all the three type-1 splits of $\mathbf{L}(\mathfrak{e}_{\boldsymbol{x}})$ in Tab.~\ref{splittab1}.
\begin{table}\label{splittab1}
\begin{center}
\begin{tabular}{ |c|c| } 
\hline
  Splits & Generators/typical reconsturctions  \\ 
 \hline
$\boldsymbol{\mathcal{M}}_a(\boldsymbol{x})$ & $\widetilde{\boldsymbol{S}}^{\boldsymbol{1}}_{\boldsymbol{x}}$,~ $\mathds{1}_{\mathfrak{e}_{\boldsymbol{x}}}$ \\ 
$\boldsymbol{\mathcal{M}}_{\overline{a}}(\boldsymbol{x})$ & $\widetilde{\boldsymbol{S}}^{\boldsymbol{2}}_{\boldsymbol{x}}$,~$\widetilde{\boldsymbol{S}}^{\boldsymbol{3}}_{\boldsymbol{x}}$,~ $\dyad{\boldsymbol0}+\dyad{\boldsymbol1}$,~ $\dyad{\boldsymbol2}+\dyad{\boldsymbol3}$,~ $\mathds{1}_{\mathfrak{e}_{\boldsymbol{x}}}$ \\ 
$\mathrm{Z}(\boldsymbol{\mathcal{M}}_a(\boldsymbol{x}))$& $\mathrm{Z}(\boldsymbol{\mathcal{M}}_a(\boldsymbol{x}))=\boldsymbol{\mathcal{M}}_a(\boldsymbol{x})\ne\mathbb{C}\mathds{1}_{\mathfrak{e}_{\boldsymbol{x}}}$\\
\hline
&\\
\hline
$\boldsymbol{\mathcal{M}}_a(\boldsymbol{x})$ & $\widetilde{\boldsymbol{S}}^{\boldsymbol{2}}_{\boldsymbol{x}}$,~ $\mathds{1}_{\mathfrak{e}_{\boldsymbol{x}}}$ \\ 
$\boldsymbol{\mathcal{M}}_{\overline{a}}(\boldsymbol{x})$ & $\widetilde{\boldsymbol{S}}^{\boldsymbol{3}}_{\boldsymbol{x}}$,~$\widetilde{\boldsymbol{S}}^{\boldsymbol{1}}_{\boldsymbol{x}}$,~ $\dyad{\boldsymbol0}+\dyad{\boldsymbol2}$,~ $\dyad{\boldsymbol3}+\dyad{\boldsymbol1}$,~ $\mathds{1}_{\mathfrak{e}_{\boldsymbol{x}}}$ \\ 
$\mathrm{Z}(\boldsymbol{\mathcal{M}}_a(\boldsymbol{x}))$& $\mathrm{Z}(\boldsymbol{\mathcal{M}}_a(\boldsymbol{x}))=\boldsymbol{\mathcal{M}}_a(\boldsymbol{x})\ne\mathbb{C}\mathds{1}_{\mathfrak{e}_{\boldsymbol{x}}}$\\
\hline
&\\
\hline
$\boldsymbol{\mathcal{M}}_a(\boldsymbol{x})$ & $\widetilde{\boldsymbol{S}}^{\boldsymbol{3}}_{\boldsymbol{x}}$,~ $\mathds{1}_{\mathfrak{e}_{\boldsymbol{x}}}$ \\ 
$\boldsymbol{\mathcal{M}}_{\overline{a}}(\boldsymbol{x})$ & $\widetilde{\boldsymbol{S}}^{\boldsymbol{1}}_{\boldsymbol{x}}$,~$\widetilde{\boldsymbol{S}}^{\boldsymbol{2}}_{\boldsymbol{x}}$,~ $\dyad{\boldsymbol0}+\dyad{\boldsymbol3}$,~ $\dyad{\boldsymbol1}+\dyad{\boldsymbol2}$,~ $\mathds{1}_{\mathfrak{e}_{\boldsymbol{x}}}$ \\ 
$\mathrm{Z}(\boldsymbol{\mathcal{M}}_a(\boldsymbol{x}))$& $\mathrm{Z}(\boldsymbol{\mathcal{M}}_a(\boldsymbol{x}))=\boldsymbol{\mathcal{M}}_a(\boldsymbol{x})\ne\mathbb{C}\mathds{1}_{\mathfrak{e}_{\boldsymbol{x}}}$\\
\hline
\end{tabular}
\end{center} 
\caption{Type-1 splits}
\end{table}

The type-2 split is illustrated in Fig.~\ref{17d}. In this case, the center is also nontrivial, and the splits together with their generators are lists in Tab.~\ref{splittab2}. The detailed arguments are given in App.~\ref{posplit}.
\begin{table}\label{splittab2}
\begin{center}
\begin{tabular}{ |c|c| } 
\hline
  Splits & Generators/typical reconsturctions  \\ 
 \hline
$\boldsymbol{\mathcal{M}}_a(\boldsymbol{x})$ & $\dyad{\boldsymbol0}+\dyad{\boldsymbol1}$,~ $\dyad{\boldsymbol2}+\dyad{\boldsymbol3}$,~ $\mathds{1}_{\mathfrak{e}_{\boldsymbol{x}}}$ \\ 
$\boldsymbol{\mathcal{M}}_{\overline{a}}(\boldsymbol{x})$ & $\widetilde{\boldsymbol{S}}^{\boldsymbol{1}}_{\boldsymbol{x}}$, ~$\dyad{\boldsymbol0}$, ~$\dyad{\boldsymbol1}$, ~$\dyad{\boldsymbol2}$, ~$\dyad{\boldsymbol3}$,~ $\mathds{1}_{\mathfrak{e}_{\boldsymbol{x}}}$ \\ 
$\mathrm{Z}(\boldsymbol{\mathcal{M}}_a(\boldsymbol{x}))$& $\mathrm{Z}(\boldsymbol{\mathcal{M}}_a(\boldsymbol{x}))=\boldsymbol{\mathcal{M}}_a(\boldsymbol{x})\ne\mathbb{C}\mathds{1}_{\mathfrak{e}_{\boldsymbol{x}}}$\\
\hline
&\\
\hline
$\boldsymbol{\mathcal{M}}_a(\boldsymbol{x})$ & $\dyad{\boldsymbol0}+\dyad{\boldsymbol2}$,~ $\dyad{\boldsymbol1}+\dyad{\boldsymbol3}$,~ $\mathds{1}_{\mathfrak{e}_{\boldsymbol{x}}}$ \\ 
$\boldsymbol{\mathcal{M}}_{\overline{a}}(\boldsymbol{x})$ & $\widetilde{\boldsymbol{S}}^{\boldsymbol{2}}_{\boldsymbol{x}}$, ~$\dyad{\boldsymbol0}$, ~$\dyad{\boldsymbol1}$, ~$\dyad{\boldsymbol2}$, ~$\dyad{\boldsymbol3}$,~ $\mathds{1}_{\mathfrak{e}_{\boldsymbol{x}}}$ \\ 
$\mathrm{Z}(\boldsymbol{\mathcal{M}}_a(\boldsymbol{x}))$& $\mathrm{Z}(\boldsymbol{\mathcal{M}}_a(\boldsymbol{x}))=\boldsymbol{\mathcal{M}}_a(\boldsymbol{x})\ne\mathbb{C}\mathds{1}_{\mathfrak{e}_{\boldsymbol{x}}}$\\
\hline
&\\
\hline
$\boldsymbol{\mathcal{M}}_a(\boldsymbol{x})$ & $\dyad{\boldsymbol0}+\dyad{\boldsymbol3}$,~ $\dyad{\boldsymbol1}+\dyad{\boldsymbol2}$,~ $\mathds{1}_{\mathfrak{e}_{\boldsymbol{x}}}$ \\ 
$\boldsymbol{\mathcal{M}}_{\overline{a}}(\boldsymbol{x})$ & $\widetilde{\boldsymbol{S}}^{\boldsymbol{3}}_{\boldsymbol{x}}$, ~$\dyad{\boldsymbol0}$, ~$\dyad{\boldsymbol1}$, ~$\dyad{\boldsymbol2}$, ~$\dyad{\boldsymbol3}$,~ $\mathds{1}_{\mathfrak{e}_{\boldsymbol{x}}}$ \\ 
$\mathrm{Z}(\boldsymbol{\mathcal{M}}_a(\boldsymbol{x}))$& $\mathrm{Z}(\boldsymbol{\mathcal{M}}_a(\boldsymbol{x}))=\boldsymbol{\mathcal{M}}_a(\boldsymbol{x})\ne\mathbb{C}\mathds{1}_{\mathfrak{e}_{\boldsymbol{x}}}$\\
\hline
\end{tabular}
\end{center} 
\caption{Type-2 splits}
\end{table}

The type-3 split is illustrated in Fig.~\ref{17e}. Different from the previous two types, in this case, the center is trivial, i.e. $\mathrm{Z}(\boldsymbol{\mathcal{M}}_a(\boldsymbol{x}))=\mathbb{C}\mathds{1}_{\mathfrak{e}_{\boldsymbol{x}}}$, and the two generated von Neumann algebras are factors. The splits together with their generators are listed in Tab.~\ref{splittab3}, and the detailed arguments are given in App.~\ref{posplit}.
\begin{table}\label{splittab3}
\begin{center}
\begin{tabular}{ |c|c| } 
\hline
  Splits & Generators/typical reconsturctions  \\ 
 \hline
$\boldsymbol{\mathcal{M}}_a(\boldsymbol{x})$ & $\widetilde{\boldsymbol{S}}^{\boldsymbol{1}}_{\boldsymbol{x}}$,~ $\dyad{\boldsymbol0}+\dyad{\boldsymbol3}$,~ $\dyad{\boldsymbol1}+\dyad{\boldsymbol2}$,~ $\mathds{1}_{\mathfrak{e}_{\boldsymbol{x}}}$ \\ 
$\boldsymbol{\mathcal{M}}_{\overline{a}}(\boldsymbol{x})$ & $\widetilde{\boldsymbol{S}}^{\boldsymbol{3}}_{\boldsymbol{x}}$, ~$\dyad{\boldsymbol0}+\dyad{\boldsymbol1}$, ~$\dyad{\boldsymbol2}+\dyad{\boldsymbol3}$,~ $\mathds{1}_{\mathfrak{e}_{\boldsymbol{x}}}$ \\ 
$\mathrm{Z}(\boldsymbol{\mathcal{M}}_a(\boldsymbol{x}))$& $\mathrm{Z}(\boldsymbol{\mathcal{M}}_a(\boldsymbol{x}))=\mathbb{C}\mathds{1}_{\mathfrak{e}_{\boldsymbol{x}}}$\\
\hline
&\\
\hline
$\boldsymbol{\mathcal{M}}_a(\boldsymbol{x})$ & $\widetilde{\boldsymbol{S}}^{\boldsymbol{3}}_{\boldsymbol{x}}$,~ $\dyad{\boldsymbol0}+\dyad{\boldsymbol2}$,~ $\dyad{\boldsymbol1}+\dyad{\boldsymbol3}$,~ $\mathds{1}_{\mathfrak{e}_{\boldsymbol{x}}}$ \\ 
$\boldsymbol{\mathcal{M}}_{\overline{a}}(\boldsymbol{x})$ & $\widetilde{\boldsymbol{S}}^{\boldsymbol{2}}_{\boldsymbol{x}}$, ~$\dyad{\boldsymbol0}+\dyad{\boldsymbol3}$, ~$\dyad{\boldsymbol1}+\dyad{\boldsymbol2}$,~ $\mathds{1}_{\mathfrak{e}_{\boldsymbol{x}}}$ \\ 
$\mathrm{Z}(\boldsymbol{\mathcal{M}}_a(\boldsymbol{x}))$& $\mathrm{Z}(\boldsymbol{\mathcal{M}}_a(\boldsymbol{x}))=\mathbb{C}\mathds{1}_{\mathfrak{e}_{\boldsymbol{x}}}$\\
\hline
&\\
\hline
$\boldsymbol{\mathcal{M}}_a(\boldsymbol{x})$ & $\widetilde{\boldsymbol{S}}^{\boldsymbol{2}}_{\boldsymbol{x}}$,~ $\dyad{\boldsymbol0}+\dyad{\boldsymbol1}$,~ $\dyad{\boldsymbol2}+\dyad{\boldsymbol3}$,~ $\mathds{1}_{\mathfrak{e}_{\boldsymbol{x}}}$ \\ 
$\boldsymbol{\mathcal{M}}_{\overline{a}}(\boldsymbol{x})$ & $\widetilde{\boldsymbol{S}}^{\boldsymbol{1}}_{\boldsymbol{x}}$, ~$\dyad{\boldsymbol0}+\dyad{\boldsymbol2}$, ~$\dyad{\boldsymbol1}+\dyad{\boldsymbol3}$,~ $\mathds{1}_{\mathfrak{e}_{\boldsymbol{x}}}$ \\ 
$\mathrm{Z}(\boldsymbol{\mathcal{M}}_a(\boldsymbol{x}))$& $\mathrm{Z}(\boldsymbol{\mathcal{M}}_a(\boldsymbol{x}))=\mathbb{C}\mathds{1}_{\mathfrak{e}_{\boldsymbol{x}}}$\\
\hline
\end{tabular}
\end{center} 
\caption{Type-3 splits}
\end{table}

Note that generally in the OAQEC formalism, the splits in the entangling surface need not to happen to individual bulk qudits (see the discussion Sec.~\ref{sdew}), but can happen to multiple bulk qudits as a whole (in the form of tensor-product degrees of freedom). However, in either case, the framework embodied in our arguments can be generally applied. Indeed, the latter case can be observed for some disconnected boundary bipartition in our code. To cover those cases, we can simply extend the previous study of the reconstruction of a single bulk qudit to the multiple case, and apply similar arguments as above to generalize the condition of bulk local complementarity. In reality, since those disconnected boundary bipartitions have no direct importance in our demonstrations, we leave them for future studies.

\begin{center}
\begin{figure}[ht]
\centering
    \includegraphics[width=6.5cm]{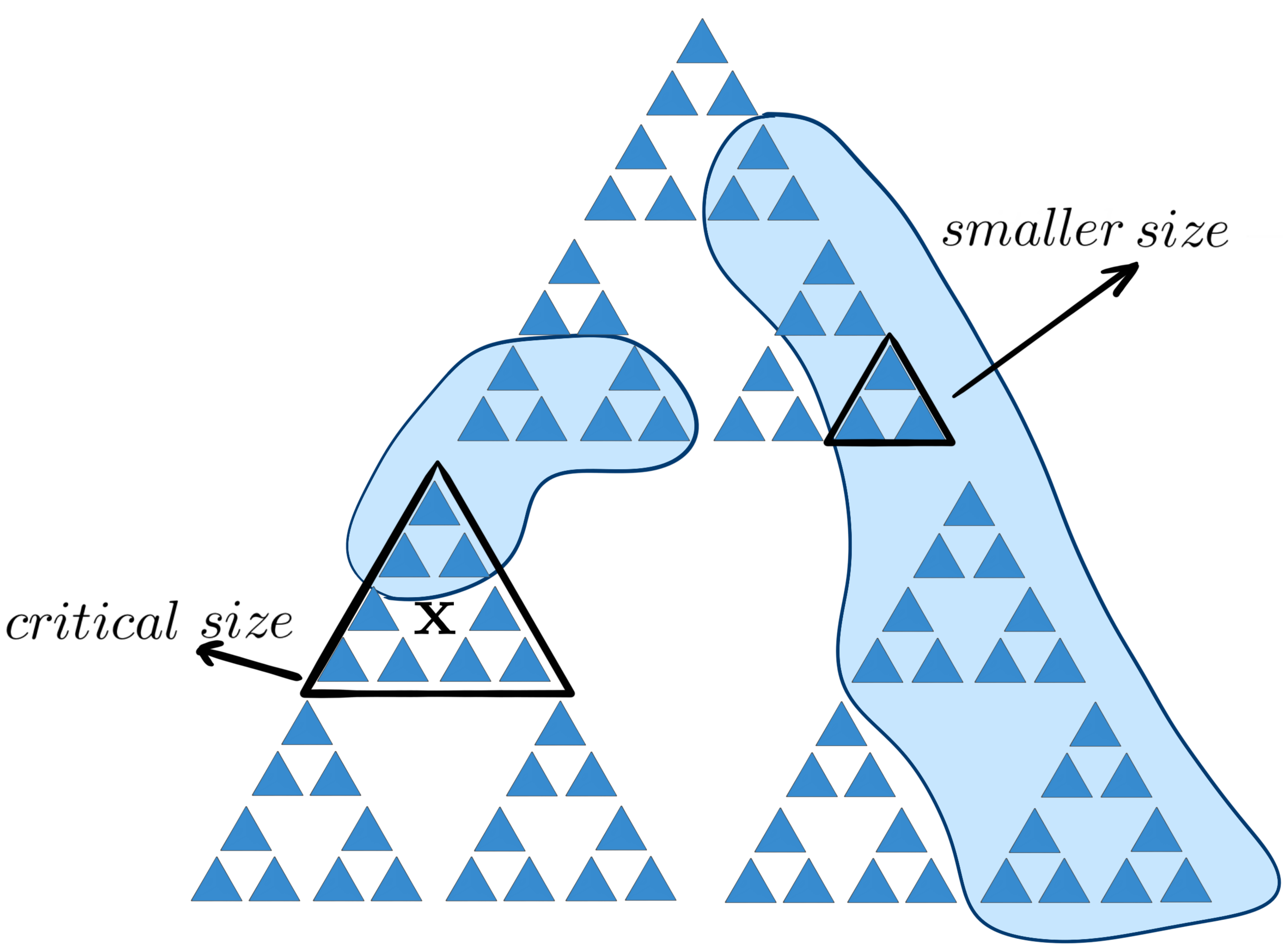}   
\caption{For an arbitrary bipartition $A\overline{A}$, there is a ``critical'' size of the triangular blocks. Blocks of smaller size are completely included in either $A$ or $\overline{A}$. Blocks of the ``critical'' or larger size are separated by the bipartition.}
\label{fig18}
\end{figure}
\end{center}

\subsection{Genuine OAQEC for arbitrary $A\overline{A}$}
Before we show how the above examples of the split of $\mathbf{L}(\mathfrak{e}_{\boldsymbol{x}})$ lead to explicit demonstration of the subregion duality, it is needed to prove that for arbitrary boundary bipartition $A\overline{A}$, the condition of the genuine OAQEC formalism, i.e. HQEC Characteristic 3 (see Par.~\ref{chac3}) is satisfied. Only upon this condition can we describe the meaning of entanglement wedges and subregion duality within genuine OAQEC. And the above examples of splits have already enable us to prove Characteristic 3. Indeed, it suffices to show that $\mathrm{Z}(\mathcal{M}_A)$ is nontrivial for arbitrary $A\overline{A}$. In the following, we show that the type-2 split always appears for some bulk qudit $\boldsymbol{x}$ so that the nontrivial logical operators generated by 
\begin{align}
\begin{split}
&R(\cdots\otimes\dyad{\boldsymbol{\beta}_{\boldsymbol{x}}}+\dyad{\boldsymbol{\beta}'_{\boldsymbol{x}}}\otimes\cdots)R^+,\\
&R(\cdots\otimes\dyad{\boldsymbol{\beta}''_{\boldsymbol{x}}}+\dyad{\boldsymbol{\beta}'''_{\boldsymbol{x}}}\otimes\cdots)R^+,\\
&\mathds{1}_{\mathcal{H}_{\mathrm{code}}}
\end{split}
\end{align}
are shared by subregion $A$ and $\overline{A}$, and belong to the center $\mathrm{Z}(\mathcal{M}_A)$.

Our arguments take advantage of the alternative geometry which shows self-similar structures on the triangular blocks. Note that those triangular blocks can have different size. Then, as shown in Fig.~\ref{fig18}, for any bipartition $A\overline{A}$ of the physical qudits (viewed in the alternative geometry), each block is either completely covered by one subregion, or separated by the two complementary subregions. If we view each triangle as the smallest blocks, then each of these triangles must be either covered by $A$ or by $\overline{A}$. However, if we view the whole lattice as a big block, then it must be separated by the two subregions. Therefore, there must be a ``critical'' size of these blocks such that each block of size below the ``critical'' size is completely covered either by $A$ or by $\overline{A}$. And importantly, there must exist one block of the ``critical'' size which is separated by the bipartition, i.e., with one of its three parts (blocks of smaller size) completely included in $A$ and the other two completely in $\overline{A}$, as shown in Fig.~\ref{fig18}. Then, we consider the corresponding hole (bulk qudit) $\boldsymbol{x}$ of this block. With reference to the examples of the type-2 split (see Fig.~\ref{17d}), we can definitely reconstruct $\dyad{\boldsymbol{\beta}_{\boldsymbol{x}}}+\dyad{\boldsymbol{\beta}'_{\boldsymbol{x}}}$ and $\dyad{\boldsymbol{\beta}''_{\boldsymbol{x}}}+\dyad{\boldsymbol{\beta}'''_{\boldsymbol{x}}}$ on both subregions $A$ and $\overline{A}$, which is the desired result.

We can conclude the above arguments into the following theorem
\begin{theorem}[Genuine OAQEC for arbitrary $A\overline{A}$]
In our code, for arbitrary boundary bipartition $A\overline{A}$, the center is nontrivial, i.e. $\mathrm{Z}(\mathcal{M}_A)=\mathcal{M}_A\cap\mathcal{M}_{\overline{A}}\ne\mathbb{C}\mathds{1}_{\mathcal{H}_{
\mathrm{code}}}$. Equivalently, we also have $\mathrm{Z}(R^+\mathcal{M}_A R)\ne\mathbb{C}\mathds{1}_{\mathcal{E}}$.  
\end{theorem}

Note that the above proof is independent on the condition of bulk local complementarity, but only relies on the typical reconstructions and the self-similarity of the alternative geometry.

\subsection{Entangling surface and entanglement wedge}
With the results on the genuine OAQEC formalism, the condition of bulk local complementarity and the formal meaning of splits, we can rigorously show how the structures of $R^+\mathcal{M}_A R$, $R^+\mathcal{M}_{\overline{A}}R$ and $\mathrm{Z}(R^+\mathcal{M}_A R)$ can be represented in terms of bulk degrees of freedom, translating between the algebraic aspects and the geometric aspects of subregion duality. Now, we complete the framework for specifying the subregion duality for a given boundary bipartition $A\overline{A}$ in the following two steps. The first step is to specify the entangling surface and the entanglement wedges as the support of $R^+\mathcal{M}_A R$ and $R^+\mathcal{M}_{\overline{A}}R$, in terms of which we can already show the basic geometric representation of the subregion duality.

We recall the discussion on the genuine OAQEC in Sec.~\ref{sdew}. Generally, for a connected boundary bipartition $A\overline{A}$, the three disjoint sub-collections of the bulk qudits have the following meaning: $\mathrm{W}[A]$ consists of all bulk qudits that can be completely reconstructed on $A$; $\mathrm{W}[\overline{A}]$ consists of all bulk qudits that can be completely reconstructed on $\overline{A}$; and $\mathrm{E}[A\overline{A}]$, the entangling surface as the feature of OAQEC, consists of all bulk qudits that splits. Note that the entanglement wedges should be understood as $\mathrm{W}[A]\cup\mathrm{E}[A\overline{A}]$ and $\mathrm{W}[\overline{A}]\cup\mathrm{E}[A\overline{A}]$.

Based on the condition of bulk local complementarity of the splits, (see Eq.~\ref{structure1}), we can specify:
\begin{align}\label{ew2}
\begin{split}
\mathrm{W}[A]&=\{\boldsymbol{x}:\boldsymbol{\mathcal{M}}_a(\boldsymbol{x})=\mathbf{L}(\mathfrak{e}_{\boldsymbol{x}})\}\\
&=\{\boldsymbol{x}:\boldsymbol{\mathcal{M}}_{\overline{a}}(\boldsymbol{x})=\mathbb{C}\mathds{1}_{\mathfrak{e}_{\boldsymbol{x}}}\},\\
\mathrm{W}[\overline{A}]&=\{\boldsymbol{x}:\boldsymbol{\mathcal{M}}_{\overline{a}}(\boldsymbol{x})=\mathbf{L}(\mathfrak{e}_{\boldsymbol{x}})\}\\
&=\{\boldsymbol{x}:\boldsymbol{\mathcal{M}}_a(\boldsymbol{x})=\mathbb{C}\mathds{1}_{\mathfrak{e}_{\boldsymbol{x}}}\},\\
\mathrm{E}[A\overline{A}]&=\{\boldsymbol{x}:\boldsymbol{\mathcal{M}}_a(\boldsymbol{x})\ne\mathbf{L}(\mathfrak{e}_{\boldsymbol{x}}),~\boldsymbol{\mathcal{M}}_{\overline{a}}(\boldsymbol{x})\ne\mathbf{L}(\mathfrak{e}_{\boldsymbol{x}})\}\\
&=\{\boldsymbol{x}:\boldsymbol{\mathcal{M}}_a(\boldsymbol{x})\ne\mathbb{C}\mathds{1}_{\mathfrak{e}_{\boldsymbol{x}}},~\boldsymbol{\mathcal{M}}_{\overline{a}}(\boldsymbol{x})\ne\mathbb{C}\mathds{1}_{\mathfrak{e}_{\boldsymbol{x}}}\}.
\end{split}
\end{align}
Here, $\boldsymbol{\mathcal{M}}_a(\boldsymbol{x})=\mathbf{L}(\mathfrak{e}_{\boldsymbol{x}})$ simply means operators on the bulk qudit $\boldsymbol{x}$ can be completely reconstructed on $A$. And since $\boldsymbol{\mathcal{M}}_{\overline{a}}(\boldsymbol{x})=\boldsymbol{\mathcal{M}}'_a(\boldsymbol{x})$ within $\mathbf{L}(\mathfrak{e}_{\boldsymbol{x}})$, $\boldsymbol{\mathcal{M}}_a(\boldsymbol{x})=\mathbf{L}(\mathfrak{e}_{\boldsymbol{x}})$ is equivalent to $\boldsymbol{\mathcal{M}}_{\overline{a}}(\boldsymbol{x})=\mathbb{C}\mathds{1}_{\mathfrak{e}_{\boldsymbol{x}}}$. When neither $\boldsymbol{\mathcal{M}}_a(\boldsymbol{x})$ nor $\boldsymbol{\mathcal{M}}_{\overline{a}}(\boldsymbol{x})$ equals the whole local operator algebra, the bulk qudit $\boldsymbol{x}$ nontrivially splits, and the center $\mathrm{Z}(\boldsymbol{\mathcal{M}}_a(\boldsymbol{x}))=\boldsymbol{\mathcal{M}}_a(\boldsymbol{x})\cap\boldsymbol{\mathcal{M}}_{\overline{a}}(\boldsymbol{x})$ (within $\mathbf{L}(\mathfrak{e}_{\boldsymbol{x}})$) describes the part shared by $A$ and $\overline{A}$ on the qudit $\boldsymbol{x}$.

For a given boundary bipartition $A\overline{A}$, the above formal way of specifying $\mathrm{W}[A]$, $\mathrm{W}[\overline{A}]$ and $E[A\overline{A}]$, together with the generators of $\boldsymbol{\mathcal{M}}_a(\boldsymbol{x})$ and $\boldsymbol{\mathcal{M}}_{\overline{a}}(\boldsymbol{x})$ as listed in Tab.~\ref{splittab1}, \ref{splittab2} and \ref{splittab3}, has indeed bridged the basics geometric representation of subregion duality to the typical reconstructions of bulk local operators and the minimal recoveries of the $\mathcal{M}(\boldsymbol{x})$s that can be directly identified from the alternative geometry.

\onecolumngrid
\begin{center}
\begin{figure}[ht]
\centering
    \includegraphics[width=17cm]{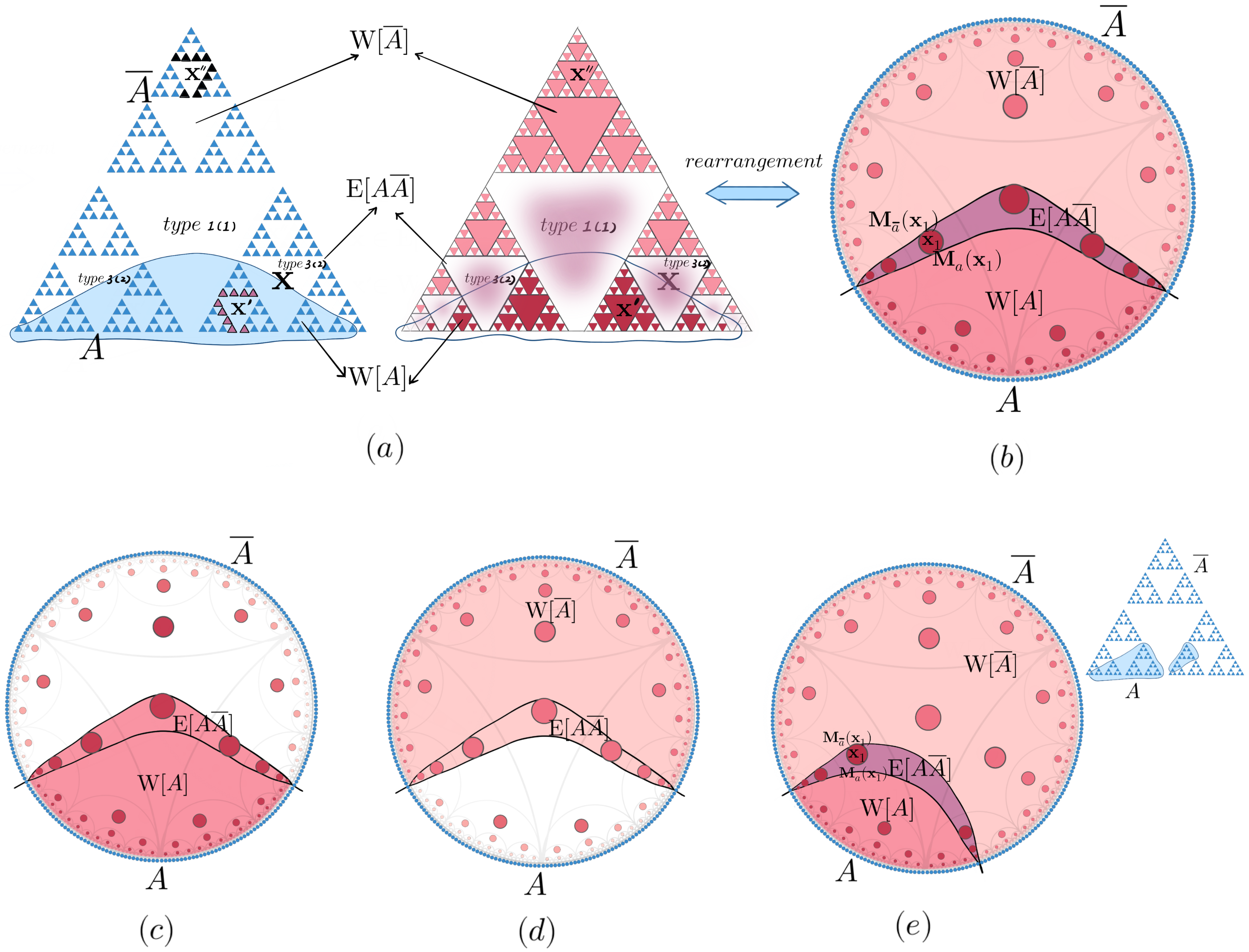}   
\phantomsubfloat{\label{19a}}\phantomsubfloat{\label{19b}}
\phantomsubfloat{\label{19c}}\phantomsubfloat{\label{19d}}
\phantomsubfloat{\label{19e}}
\caption{(a) Given the compact presentation of the boundary bipartition $A\overline{A}$ in the alternative geometry, the minimal recoveries and the typical reconstructions of bulk qudits can be identified to classify the bulk into $\mathrm{W}[A]$, $\mathrm{W}[\overline{A}]$ and $\mathrm{E}[A\overline{A}]$. For the bulk qudits (holes) with split, the type of splits can also be identified according to Tab.~\ref{splittab1}, \ref{splittab2} and \ref{splittab3}. (b) The geometric representation of subregion duality as specified through the rearrangement from (a). The red shades $\mathrm{W}[A]$, the pink shades $\mathrm{W}[\overline{A}]$ and the purple shades $\mathrm{E}[A\overline{A}]$. In $\mathrm{E}[A\overline{A}]$, the pairs $(\boldsymbol{\mathcal{M}}_a(\boldsymbol{x}),\boldsymbol{\mathcal{M}}_{\overline{a}}(\boldsymbol{x}))$ are associated to bulk qudits. (c) The entanglement wedge $\mathrm{W}[A]\cup\mathrm{E}[A\overline{A}]$ shaded by red. (d) The entanglement wedge $\mathrm{E}[A\overline{A}]\cup\mathrm{W}[\overline{A}]$ shaded by pink. The entangling surface $\mathrm{E}[A\overline{A}]$ in (b) is the overlap of the entanglement wedges in (c) and (d). (e) Another example of geometric representation of subregion duality in OAQEC.}
\label{fig19}
\end{figure}
\end{center}
\twocolumngrid

Now we consider the representative example of connected boundary bipartition as illustrated in Fig.~\ref{fig17} and \ref{fig19}, we can continue to specify the entanglement wedges, i.e., the supports of $R^+\mathcal{M}_A R$ and $R^+\mathcal{M}_{\overline{A}}R$: Through the rearrangement, those bulk qudits with identified minimal recoveries on $A$ and on $\overline{A}$ respectively (see Fig.~\ref{19a}) form $\mathrm{W}[A]$ and $\mathrm{W}[\overline{A}]$; and those with identified split typical nontrivial reconstructions form the entangling surface $E[A\overline{A}]$ (see Fig.~\ref{19b}, \ref{19c} and \ref{19d}). Clearly, in the illustration, the entangling surface in Fig.~\ref{19b} as the overlap between the two entanglement wedges in Fig.~\ref{19c} and \ref{19d} qualitatively agree with the general description in Fig.~\ref{1c}, \ref{1d} and \ref{1e}, and the corresponding discussion in Sec.~\ref{sdew}.

\subsection{Structure of subalgebra and geometric meaning}
In the second step, we complete the framework by showing how the bulk qudits in the supports of $R^+\mathcal{M}_A R$, $R^+\mathcal{M}_{\overline{A}} R$ contribute to the structures of the subalgebras. We start with considering the two von Neumann algebra tensor products $\otimes_{\boldsymbol{x}}\boldsymbol{\mathcal{M}}_a(\boldsymbol{x})$ and $\otimes_{\boldsymbol{x}}\boldsymbol{\mathcal{M}}_{\overline{a}}(\boldsymbol{x})$, where each product goes through the whole of bulk qudits. According to Eq.~\ref{ew2}, the two products can be expressed in a clear form which coincides the tripartition of the bulk,
\begin{align*}
\begin{split}
&\cdots\otimes\boldsymbol{\mathcal{M}}_a(\boldsymbol{x})\otimes\boldsymbol{\mathcal{M}}_a(\boldsymbol{x}')\otimes\boldsymbol{\mathcal{M}}_a(\boldsymbol{x}'')\otimes\cdots=\\
&(\otimes_{\boldsymbol{x}\in\mathrm{W}[A]}\mathbf{L}(\mathfrak{e}_{\boldsymbol{x}}))\otimes(\otimes_{\boldsymbol{x}\in\mathrm{E}[A\overline{A}]}\boldsymbol{\mathcal{M}}_a(\boldsymbol{x}))\otimes(\otimes_{\boldsymbol{x}\in\mathrm{W}[\overline{A}]}\mathbb{C}\mathds{1}_{\mathfrak{e}_{\boldsymbol{x}}})
\end{split}
\end{align*}
and
\begin{align*}
\begin{split}
&\cdots\otimes\boldsymbol{\mathcal{M}}_{\overline{a}}(\boldsymbol{x})\otimes\boldsymbol{\mathcal{M}}_{\overline{a}}(\boldsymbol{x}')\otimes\boldsymbol{\mathcal{M}}_{\overline{a}}(\boldsymbol{x}'')\otimes\cdots=\\
&(\otimes_{\boldsymbol{x}\in\mathrm{W}[A]}\mathbb{C}\mathds{1}_{\mathfrak{e}_{\boldsymbol{x}}})\otimes(\otimes_{\boldsymbol{x}\in\mathrm{E}[A\overline{A}]}\boldsymbol{\mathcal{M}}_{\overline{a}}(\boldsymbol{x}))\otimes(\otimes_{\boldsymbol{x}\in\mathrm{W}[\overline{A}]}\mathbf{L}(\mathfrak{e}_{\boldsymbol{x}})).
\end{split}
\end{align*}
For simplicity in notation, if we omit the part of trivial von Neumann algebra $\mathbb{C}\mathds{1}_{\mathfrak{e}_{\boldsymbol{x}}}$, we have
\begin{align}\label{structure2}
\begin{split}
\otimes_{\boldsymbol{x}}\boldsymbol{\mathcal{M}}_a&(\boldsymbol{x})\\
&=(\otimes_{\boldsymbol{x}\in\mathrm{W}[A]}\mathbf{L}(\mathfrak{e}_{\boldsymbol{x}}))\otimes(\otimes_{\boldsymbol{x}\in\mathrm{E}[A\overline{A}]}\boldsymbol{\mathcal{M}}_a(\boldsymbol{x})),\\
\otimes_{\boldsymbol{x}}\boldsymbol{\mathcal{M}}_{\overline{a}}&(\boldsymbol{x})\\
&=(\otimes_{\boldsymbol{x}\in\mathrm{E}[A\overline{A}]}\boldsymbol{\mathcal{M}}_{\overline{a}}(\boldsymbol{x}))\otimes(\otimes_{\boldsymbol{x}\in\mathrm{W}[\overline{A}]}\mathbf{L}(\mathfrak{e}_{\boldsymbol{x}})).
\end{split}
\end{align}

Then, the following proposition (see the proof in App.~\ref{posova}) directly shows that how $R^+\mathcal{M}_A R$, $R^+\mathcal{M}_{\overline{A}} R$ and the center $\mathrm{Z}(R^+\mathcal{M}_A R)$ can be explicitly described by the splits.


\begin{proposition}[Structure of the von Neumann algebras]\label{sova}
We follow the general setup in Sec.~\ref{hqecc}. If for a boundary bipartition $A\overline{A}$, the condition of complementary recovery, the genuine OAQEC, and the condition of bulk local complementarity are satisfied, then we have
\begin{align}
\begin{split}
&R^+\mathcal{M}_A R=(\otimes_{\boldsymbol{x}\in\mathrm{W}[A]}\mathbf{L}(\mathfrak{e}_{\boldsymbol{x}}))\otimes(\otimes_{\boldsymbol{x}\in\mathrm{E}[A\overline{A}]}\boldsymbol{\mathcal{M}}_a(\boldsymbol{x})),\\
&R^+\mathcal{M}_{\overline{A}} R=(\otimes_{\boldsymbol{x}\in\mathrm{E}[A\overline{A}]}\boldsymbol{\mathcal{M}}_{\overline{a}}(\boldsymbol{x}))\otimes(\otimes_{\boldsymbol{x}\in\mathrm{W}[\overline{A}]}\mathbf{L}(\mathfrak{e}_{\boldsymbol{x}})),\\
&\mathrm{Z}(R^+\mathcal{M}_A R)=\otimes_{\boldsymbol{x}\in\mathrm{E}[A\overline{A}]}\mathrm{Z}(\boldsymbol{\mathcal{M}}_a(\boldsymbol{x})).
\end{split}
\end{align}
Here, the center $\mathrm{Z}(R^+\mathcal{M}_A R)$ is within $\mathbf{L}(\mathcal{E})$, while $\mathrm{Z}(\boldsymbol{\mathcal{M}}_a(\boldsymbol{x}))=\boldsymbol{\mathcal{M}}_a(\boldsymbol{x})\cap\boldsymbol{\mathcal{M}}_{\overline{a}}(\boldsymbol{x})$ is the local center in $\mathbf{L}(\mathfrak{e}_{\boldsymbol{x}})$.
\end{proposition}


According to the above proposition combined with Eq.~\ref{ew2}, we have already establish a framework for translating between the algebraic characterization and the geometric representation of the subregion duality, which is as exact as how Eq.~\ref{ssc} does for the subsystem-code formalism. Such translation can be explained as follows, which also instructs how to access the complete description of subregion duality for a given boundary bipartition merely based on the typical reconstructions and the minimal recoveries that can be directly identified in the alternative geometry.

Given the boundary bipartition $A\overline{A}$ which determines the algebraic aspects, i.e., $\mathcal{M}_A$ and $\mathcal{M}_{\overline{A}}$, we can (1) classify the bulk qudits into $\mathrm{W}[A]$, $\mathrm{W}[\overline{A}]$ and $\mathrm{E}[A\overline{A}]$ with specific bulk geometry, (2) and specify how each bulk qudit $\boldsymbol{x}\in\mathrm{E}[A\overline{A}]$ is associated with the pair $(\boldsymbol{\mathcal{M}}_a(\boldsymbol{x}),\boldsymbol{\mathcal{M}}_{\overline{a}}(\boldsymbol{x}))$. As illustrated in Fig.~\ref{19b} and \ref{19e}, (1) and (2) specify the geometric representation of subregion duality in genuine OAQEC. Different from the geometric representation in the subsystem-code formalism (see Fig.~\ref{1b}), we need additionally the $(\boldsymbol{\mathcal{M}}_a(\boldsymbol{x}),\boldsymbol{\mathcal{M}}_{\overline{a}}(\boldsymbol{x}))$ associated to bulk qudits in the entangling surface. Prop.~\ref{sova} endows those pairs with the following meaning: $\boldsymbol{\mathcal{M}}_a(\boldsymbol{x})$ represents how the bulk qudit $\boldsymbol{x}$ contributes to $R^+\mathcal{M}_A R$, $\boldsymbol{\mathcal{M}}_{\overline{a}}(\boldsymbol{x})$ represents how it contributes to $R^+\mathcal{M}_{\overline{A}} R$ and $\mathrm{Z}(R^+\mathcal{M}_A R)$, and $\boldsymbol{\mathcal{M}}_a(\boldsymbol{x})\cap\boldsymbol{\mathcal{M}}_{\overline{a}}(\boldsymbol{x})$ represents how it contributes to $\mathrm{Z}(R^+\mathcal{M}_A R)$. Then, from such geometric representation, the complete structures of $R^+\mathcal{M}_A R$, $R^+\mathcal{M}_{\overline{A}} R$ and $\mathrm{Z}(R^+\mathcal{M}_A R)$ are simply the von Neumann tensor products of such contributions together with the whole bulk local operators algebras of bulk qudits in $\mathrm{W}[A]$ and $\mathrm{W}[\overline{A}]$.

The above description as illustrated in Fig.~\ref{19b} and \ref{19e} can be straightforwardly applied to other connected boundary bipartitions, and the expected properties can be easily checked. Hence, we have indeed demonstrated HQEC Characteristic 4.1 (see Par.~\ref{chac41}). The above framework applies also to all disconnected boundary bipartition of the importance for uberholography as can seen in following examples.

It is important to note that the essence of Prop.~\ref{sova} is general to holographic codes with OAQEC formalism. Indeed, as pointed out in Sec.~\ref{sdew}, the splits need not to happen to individual bulk qudits in the entangling surface. More generally, it happens to disjoint subsets of bulk qudits in the entangling surface, which can be seen for other disconnected boundary bipartition in our code. In that case, there should be a similar version of bulk local complementarity in which the $\boldsymbol{x}$ in $(\boldsymbol{\mathcal{M}}_a(\boldsymbol{x}),\boldsymbol{\mathcal{M}}_{\overline{a}}(\boldsymbol{x}))$ is simply replaced by the label for the disjoint subsets. Then, with the similar condition we still can prove Prop.~\ref{sova} and present a framework of the same form as above.

\subsection{Examples of geometric representations}
We can illustrate further examples of subregion duality for both connected and disconnected boundary bipartitions, which can demonstrate the derived properties as described in Characteristic 4.2 (see Par.~\ref{chac42}) and also complete the demonstration of the geometric aspects of uberholography, i.e. Characteristic 6 (see Sec.~\ref{uberintro}). In Fig.~\ref{fig20} we illustrate representative examples, other relevant cases can be illustrated and examined in a similar way.

\onecolumngrid
\begin{center}
\begin{figure}[ht]
\centering
    \includegraphics[width=17cm]{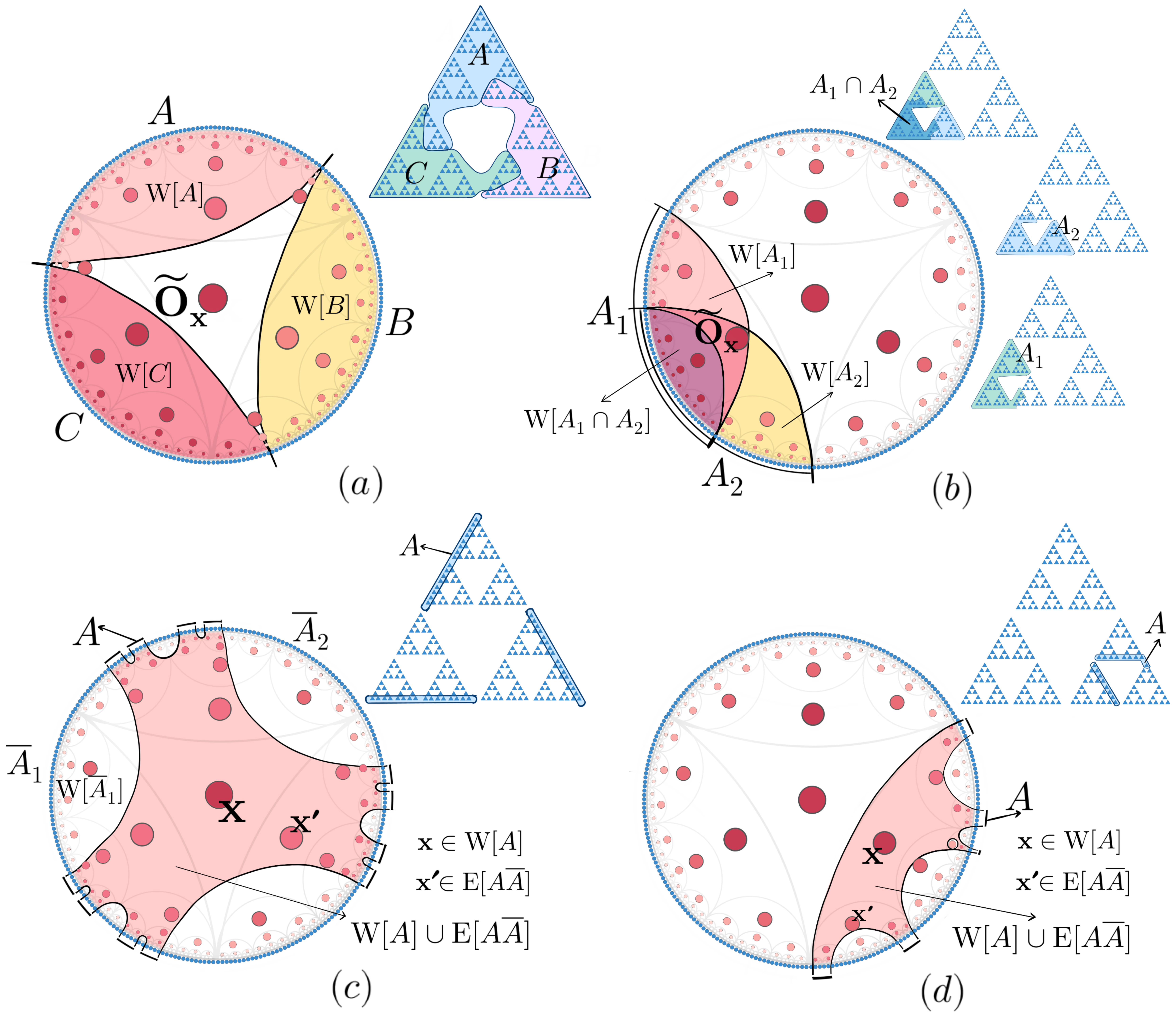}   
\phantomsubfloat{\label{20a}}\phantomsubfloat{\label{20b}}
\phantomsubfloat{\label{20c}}\phantomsubfloat{\label{20d}}
\caption{In each subfigure we present both the standard and the alternative geometric settings to present the boundary subregions. Both the wedges in the bulk and the boundary subregions present in the alternative geometric are shaded for indication. In (a) and (b), the shaded are the $\mathrm{W}[A]$s, while in (c) and (d), the shaded are the entanglement wedges $\mathrm{W}[A]\cup\mathrm{E}[A\overline{A}]$. Within the entanglement wedges in (c) and (d), the red (dark) dots lies within $\mathrm{W}[A]$, while other pink (light) ones lie within $\mathrm{E}[A\overline{A}]$.}
\label{fig20}
\end{figure}
\end{center}
\twocolumngrid

As shown in Fig.~\ref{20a}, the boundary is tripartite as disjoint subregions $A$, $B$ and $C$, and we consider the reconstruction of the bulk qudit $\boldsymbol{x}$ in the center. According to how the boundary subregions presented in the alternative geometry (the upper right) and following the above description for specifying the bulk wedges, we can specify $\mathrm{W}[A]$, $\mathrm{W}[B]$ and $\mathrm{W}[C]$ in the bulk.

In terms of the typical reconstructions studied in the alternative geometry (see Fig.~\ref{fig12} and \ref{fig15}), bulk local operators on $\boldsymbol{x}$ and of the form $\widetilde{\boldsymbol{O}}_{\boldsymbol{x}}=\dyad{\boldsymbol{\beta}_{\boldsymbol{x}}}$ can be reconstructed on none of $A$, $B$ or $C$, since any of $A\cup B$, $A\cup C$ and $B\cup C$ supports the reconstruction of a nontrivial $\widetilde{\boldsymbol{S}}^{\boldsymbol{\sigma_{\boldsymbol{x}}}}_{\boldsymbol{x}}$ operator which does not commute with $\dyad{\boldsymbol{\beta}_{\boldsymbol{x}}}$ (see Eq.~\ref{cts}). It is also clear from the typical reconsturctions that $\widetilde{\boldsymbol{O}}_{\boldsymbol{x}}$ can be reconstructed on any of $A\cup B$, $A\cup C$ and $B\cup C$ (see Fig.~\ref{12f}).

While the above arguments have shown how the example formally demonstrates what is expected in Characteristic 4.2, this example, as representing the genuine OAQEC formalism, differs slightly from the conventional illustration in the subsystem-code formalism. It is clear from Fig.~\ref{20a} that the bulk qudit $\boldsymbol{x}$ belongs to none of $\mathrm{W}[A]$, $\mathrm{W}[B]$ or $\mathrm{W}[C]$, i.e., $\mathcal{M}(\boldsymbol{x})$ has no recovery on $A$, $B$ or $C$. Indeed, the full of $\mathcal{M}(\boldsymbol{x})$ has no recovery in  $A\cup B$, $A\cup C$ or $B\cup C$ either. That is because of the nontrivial center in the genuine OAQEC: Subregions $A$ and $B\cup C=\overline{A}$ share certain operator on $\boldsymbol{x}$, e.g., $\dyad{\boldsymbol{\beta}_{\boldsymbol{x}}}+\dyad{\boldsymbol{\beta}'_{\boldsymbol{x}}}$ (see Fig.~\ref{fig13}), and $\boldsymbol{x}$ lies within each of the entangling surfaces $\mathrm{E}[A\overline{A}]$, $\mathrm{E}[B\overline{B}]$ and $\mathrm{E}[C\overline{C}]$.

As shown in Fig.~\ref{20b}, in the same way with the alternative geometry as above, for the boundary subregions $A_1$ and $A_2$ we can specify $\mathrm{W}[A_1]$, $\mathrm{W}[A_2]$ and $\mathrm{W}[A_1\cap A_2]$. Consider the bulk qudit $\boldsymbol{x}$ in the illustration, which clearly lies within $\mathrm{W}[A_1]$ and $\mathrm{W}[A_2]$, i.e., the whole of $\mathcal{M}(\boldsymbol{x})$ can be recovered there in, but not in $\mathrm{W}[A_1\cap A_2]$. It implies that there exists an operator $\widetilde{\boldsymbol{O}}_{\boldsymbol{x}}$ which can be reconstructed on both $A_1$ and $A_2$ but not on the intersection $A_1\cap A_2$. Indeed, this example also manifests the difference between the genuine OAQEC and the subsystem-code formalism: Although $\boldsymbol{x}$ does not lie within $\mathrm{W}[A_1\cap A_2]$, some operator on it, e.g., $\dyad{\boldsymbol{\beta}_{\boldsymbol{x}}}+\dyad{\boldsymbol{\beta}'_{\boldsymbol{x}}}$, can be reconstructed on $A_1\cap A_2$ since it is included in the nontrivial center shared by $A_1\cap A_2$ and $\overline{A_1\cap A_2}$.

The geometric representations of uberholography are illustrated in Fig.~\ref{20c} and \ref{20d}, in which the minimal subregions $A$ for recovering $\mathcal{M}(\boldsymbol{x})$ have been demonstrated in Fig.~\ref{16i} and \ref{16k} with scaling behavior $\sim N^{1/h}$. Here, we show the geometric aspects and the entanglement wedges in the same as above. In both Fig.~\ref{20c} and \ref{20d}, the shaded is the whole entanglement wedge, i.e. $\mathrm{W}[A]\cup\mathrm{E}[A\overline{A}]$. It can be viewed as hollowed out from the bulk by removing $\mathrm{W}[\overline{A}_1],\mathrm{W}[\overline{A}_2],\ldots$ for the insulated parts $\overline{A}_1,\overline{A}_2,\ldots$ of the complement $\overline{A}$.

It is noticeable that in Fig.~\ref{20c} and \ref{20d}, while the bulk qudit $\boldsymbol{x}$ lies within $\mathrm{W}[A]$, other bulk qudits in the same entanglement wedge, e.g., $\boldsymbol{x}'$, lie within the entangling surface $\mathrm{E}[A\overline{A}]$, since certain operators on them are shared in the center corresponding to $A$ and $\overline{A}$. This again manifests the difference between the genuine OAQEC and the subsystem-code formalism.

\subsection{RT formula in OAQEC}

Now, we extend the framework developed for the demonstration and the study of subregion duality. As we will show, the formal meaning of splits in the entangling surface and how the splits contribute to the structures of the subalgebras underlie a natural way to extract the bulk and area terms of the RT formula in genuine OAQEC following the ``standard'' prescription~\cite{harlow2017}. Our arguments will not only demonstrate HQEC Characteristic 5 (see Par.~\ref{chac5}), but also reveal how the entanglement patterns of the boundary basis states manifest in shaping the area term of the RT formula. In the following, we start with a brief summary of the ``standard'' prescription.

\subsubsection{Standard prescription}
Consider a connected boundary bipartition $A\overline{A}$ ($\mathcal{H}=\mathcal{H}_A\otimes\mathcal{H}_{\overline{A}}$). We have proved the condition of complementary recovery (Characteristic 2) and the genuine OAQEC (Characteristic 3), i.e., $\mathcal{M}_{\overline{A}}=\mathcal{M}'_A$ and $\mathrm{Z}(\mathcal{M}_A)\ne\mathbb{C}\mathds{1}$. Then, according to the standard prescription given by Ref.~\cite{harlow2017}, although the bulk and the area terms in
\begin{align}\label{rtformula0}
\begin{split}
\mathrm{S}(\widetilde{\rho}_A)&=\mathrm{Tr}[\widetilde{\rho}\mathcal{L}_A]+\mathrm{S}(\widetilde{\rho},\mathcal{M}_A)\\
\mathrm{S}(\widetilde{\rho}_{\overline{A}})&=\mathrm{Tr}[\widetilde{\rho}\mathcal{L}_A]+\mathrm{S}(\widetilde{\rho},\mathcal{M}_{\overline{A}}),
\end{split}
\end{align}
have ``universal'' meaning only dependent on the von Neumann algebras $\mathcal{M}_A$ and $\mathcal{M}_{\overline{A}}$, we can extract them from specific nontrivial decompositions of the Hilbert spaces $\mathcal{E}$, $\mathcal{H}_A$ and $\mathcal{H}_{\overline{A}}$, as long as the decompositions satisfy certain conditions as described below. 

As a basic properties of von Neumann algebra~\cite{harlow2017}, similar to the case described in Sec.~\ref{exsplitsec}, there exists a decomposition $\mathcal{H}_{\mathrm{code}}=\oplus_{\boldsymbol{\mu}}(\mathcal{H}^{\boldsymbol{\mu}}_a\otimes\mathcal{H}^{\boldsymbol{\mu}}_{\overline{a}})$ such that logical operators $\widetilde{O}\in\mathcal{M}_A$ are simply all operators of the form $\sum_{\boldsymbol{\mu}}(\widetilde{O}^{\boldsymbol{\mu}}_a\otimes\mathds{1}^{\boldsymbol{\mu}}_{\overline{a}})$, logical operators $\widetilde{O'}\in\mathcal{M}_{\overline{A}}$ are simply those $\sum_{\boldsymbol{\mu}}(\mathds{1}^{\boldsymbol{\mu}}_a\otimes\widetilde{O}^{\boldsymbol{\mu}}_{\overline{a}})$, and logical operators $\widetilde{O}\in\mathrm{Z}(\mathcal{M}_A)$ are simply those $\sum_{\boldsymbol{\mu}}c_{\boldsymbol{\mu}}(\mathds{1}^{\boldsymbol{\mu}}_a\otimes\mathds{1}^{\boldsymbol{\mu}}_{\overline{a}})$ with ($c_{\boldsymbol{\mu}}\in\mathbb{C}$). Note that here the subscripts $a$ and $\overline{a}$ can be viewed as the bulk counterpart of the subscripts $A$ and $\overline{A}$ for the boundary, indicating a specific boundary bipartition.

For convenience in our arguments, in terms of the unitary map $\mathcal{E}\xrightarrow{R}\mathcal{H}_{\mathrm{code}}$, we can consider the equivalent decomposition of $\mathcal{E}$, i.e., the bulk Hilbert space. Then, the basic condition can be also formally summarized as follows 
\begin{align}\label{decl1}
\begin{split}
&\mathcal{E}=\oplus_{\boldsymbol{\mu}}(\mathcal{E}^{\boldsymbol{\mu}}_a\otimes\mathcal{E}^{\boldsymbol{\mu}}_{\overline{a}}), \quad \mathcal{E}^{\boldsymbol{\mu}}_a\otimes\mathcal{E}^{\boldsymbol{\mu}}_{\overline{a}}\xrightarrow{\boldsymbol{I}^{\boldsymbol{\mu}}}\mathcal{E},\\
&R^+\mathcal{M}_A R\ni\widetilde{\boldsymbol{O}}=\sum_{\boldsymbol{\mu}}{\boldsymbol{I}^{\boldsymbol{\mu}}}(\widetilde{\boldsymbol{O}}^{\boldsymbol{\mu}}_a\otimes\mathds{1}^{\boldsymbol{\mu}}_{\overline{a}}){\boldsymbol{I}^{\boldsymbol{\mu}}}^+,\\
&R^+\mathcal{M}_{\overline{A}} R\ni\widetilde{\boldsymbol{O}}=\sum_{\boldsymbol{\mu}}{\boldsymbol{I}^{\boldsymbol{\mu}}}(\mathds{1}^{\boldsymbol{\mu}}_a\otimes\widetilde{\boldsymbol{O}}^{\boldsymbol{\mu}}_{\overline{a}}){\boldsymbol{I}^{\boldsymbol{\mu}}}^+,\\
&\mathrm{Z}(R^+\mathcal{M}_A R)\ni\widetilde{\boldsymbol{O}}=\sum_{\boldsymbol{\mu}}{\boldsymbol{I}^{\boldsymbol{\mu}}}(\mathds{1}^{\boldsymbol{\mu}}_a\otimes\mathds{1}^{\boldsymbol{\mu}}_{\overline{a}}){\boldsymbol{I}^{\boldsymbol{\mu}}}^+,
\end{split}
\end{align}
where we have $\widetilde{\boldsymbol{O}}^{\boldsymbol{\mu}}_a\in\mathbf{L}(\mathcal{E}^{\boldsymbol{\mu}}_a),  \widetilde{\boldsymbol{O}}^{\boldsymbol{\mu}}_{\overline{a}}\in\mathbf{L}(\mathcal{E}^{\boldsymbol{\mu}}_{\overline{a}})$. The isometry $\boldsymbol{I}^{\boldsymbol{\mu}}$ underlies the formal meaning of ``identified as'', based on which, the formal meaning of the decomposition is ${\boldsymbol{I}^{\boldsymbol{\mu}'}}^+\boldsymbol{I}^{\boldsymbol{\mu}}=\delta_{\boldsymbol{\mu}\boldsymbol{\mu}'}\mathds{1}^{\boldsymbol{\mu}},\widetilde{\boldsymbol{P}}^{\boldsymbol{\mu}}=\boldsymbol{I}^{\boldsymbol{\mu}}{\boldsymbol{I}^{\boldsymbol{\mu}}}^+, \sum_{\boldsymbol{\mu}}\widetilde{\boldsymbol{P}}^{\boldsymbol{\mu}}=\mathds{1}$. And we denote by $\{\ket*{\boldsymbol{B}^{\boldsymbol{\mu}}_{a,\boldsymbol{n}}}\}_{1\le\boldsymbol{n}\le\mathrm{dim}\mathcal{E}^{\boldsymbol{\mu}}_a}$ and $\{\ket*{\boldsymbol{B}^{\boldsymbol{\mu}}_{\overline{a},\boldsymbol{n}'}}\}_{1\le\boldsymbol{n}'\le\mathrm{dim}\mathcal{E}^{\boldsymbol{\mu}}_{\overline{a}}}$ as bases of $\mathcal{E}^{\boldsymbol{\mu}}_a$ and $\mathcal{E}^{\boldsymbol{\mu}}_{\overline{a}}$ respectively.

According to Theorem 5.1 in Ref.~\cite{harlow2017}, there exist two decompositions (of subspaces) of $\mathcal{H}_A$ and $\mathcal{H}_{\overline{A}}$ respectively together with a family of isometric linear maps $\{\boldsymbol{J}^{\boldsymbol{\mu}}\}$, i.e.,
\begin{align}\label{decl2}
\begin{split}
&\mathcal{E}^{\boldsymbol{\mu}}_a\otimes \mathcal{F}^{\boldsymbol{\mu}}_a\xrightarrow{\boldsymbol{U}^{\boldsymbol{\mu}}_A}\mathcal{H}_A, \quad \mathcal{E}^{\boldsymbol{\mu}}_{\overline{a}}\otimes \mathcal{F}^{\boldsymbol{\mu}}_{\overline{a}}\xrightarrow{\boldsymbol{U}^{\boldsymbol{\mu}}_{\overline{A}}}\mathcal{H}_{\overline{A}},\\
&\mathcal{E}^{\boldsymbol{\mu}}_a\otimes\mathcal{E}^{\boldsymbol{\mu}}_{\overline{a}}\xrightarrow{\boldsymbol{J}^{\boldsymbol{\mu}}}\mathcal{E}^{\boldsymbol{\mu}}_a\otimes(\mathcal{F}^{\boldsymbol{\mu}}_a\otimes\mathcal{F}^{\boldsymbol{\mu}}_{\overline{a}})\otimes\mathcal{E}^{\boldsymbol{\mu}}_{\overline{a}},\\
&\boldsymbol{J}^{\boldsymbol{\mu}}(\ket*{\boldsymbol{B}^{\boldsymbol{\mu}}_{a,\boldsymbol{n}}}\otimes\ket*{\boldsymbol{B}^{\boldsymbol{\mu}}_{\overline{a},\boldsymbol{n}'}})=\ket*{\boldsymbol{B}^{\boldsymbol{\mu}}_{a,\boldsymbol{n}}}\otimes\ket{\chi^{\boldsymbol{\mu}}}\otimes\ket*{\boldsymbol{B}^{\boldsymbol{\mu}}_{\overline{a},\boldsymbol{n}'}}\\
&\ket{\chi^{\boldsymbol{\mu}}}\in\mathcal{F}^{\boldsymbol{\mu}}_a\otimes\mathcal{F}^{\boldsymbol{\mu}}_{\overline{a}}, \quad \braket{\chi^{\boldsymbol{\mu}}}{\chi^{\boldsymbol{\mu}}}=1,
\end{split}
\end{align}
where the index of decomposition is consistent with Eq.~\ref{decl1}. Note that we use equivalent expression~\footnote{Note that in Ref.~\cite{harlow2017}, the part for decompositions of $\mathcal{H}_A$ and $\mathcal{H}_{\overline{A}}$ are described by decomposition followed by unitary operators. Indeed, since the decomposition is in nature a family of ``orthogonal'' isometry statifying properties as above for the $\boldsymbol{I}^{\boldsymbol{\mu}}$s, decompositions (``orthogonal'' isometries) of $\mathcal{H}_A$ and $\mathcal{H}_{\overline{A}}$ followed by unitary operators on $\mathcal{H}_A$ and $\mathcal{H}_{\overline{A}}$ are again decompositions (``orthogonal'' isometries). Hence, we directly, and also equivalently, use $\boldsymbol{U}^{\boldsymbol{\mu}}_A$ and $\boldsymbol{U}^{\boldsymbol{\mu}}_{\overline{A}}$. In terms of this perspective, in the decompositions of $\mathcal{H}_A$ and $\mathcal{H}_{\overline{A}}$, we equivalently and directly use $\mathcal{E}^{\boldsymbol{\mu}}_a$ and $\mathcal{E}^{\boldsymbol{\mu}}_{\overline{a}}$ that decompose $\mathcal{E}$, instead of $\mathcal{H}^{\boldsymbol{\mu}}_{A_1}$ and $\mathcal{H}^{\boldsymbol{\mu}}_{\overline{A}_1}$ (see Ref.~\cite{harlow2017} where the superscript was $\alpha$) which were used as general copies of $\mathcal{E}^{\boldsymbol{\mu}}_a$ and $\mathcal{E}^{\boldsymbol{\mu}}_{\overline{a}}$ respectively. And we specify the necessary isometry $\boldsymbol{J}^{\boldsymbol{\mu}}:\ket{\boldsymbol{B}^{\boldsymbol{\mu}}_a}\otimes\ket{\boldsymbol{B}^{\boldsymbol{\mu}}_{\overline{a}}}\mapsto\ket{\boldsymbol{B}^{\boldsymbol{\mu}}_a}\otimes\ket{\chi^{\boldsymbol{\mu}}}\otimes\ket{\boldsymbol{B}^{\boldsymbol{\mu}}_{\overline{a}}}$ by directly using $\ket{\boldsymbol{B}^{\boldsymbol{\mu}}_a}$ and $\ket{\boldsymbol{B}^{\boldsymbol{\mu}}_{\overline{a}}}$ on the right-hand-side, instead of using the copies $\ket{\boldsymbol{B}^{\boldsymbol{\mu}}_a}_{A_1}$ and $\ket{\boldsymbol{B}^{\boldsymbol{\mu}}_{\overline{a}}}_{\overline{A}_1}$. In our notation, there is no loss of generality, since the generality is contained the ``orthogonal'' isometries that underlie the decomposition.} which is more concise for the rigorous proving, but with slightly different appearance from Ref.~\cite{harlow2017}.

We can regard the two decompositions for $\mathcal{H}_A$ and $\mathcal{H}_{\overline{A}}$ and the family of isometry as to show how the decomposition structure of $\mathcal{E}$ (see Eq.~\ref{decl1}) is paralleled in $\mathcal{H}_A\otimes\mathcal{H}_{\overline{A}}$, and hence they are required to satisfy that
\begin{align}\label{diagram10}
\begin{split}
&R\boldsymbol{I}^{\boldsymbol{\mu}}=(\boldsymbol{U}^{\boldsymbol{\mu}}_A\otimes\boldsymbol{U}^{\boldsymbol{\mu}}_{\overline{A}})\boldsymbol{J}^{\boldsymbol{\mu}},\\
&R\ket*{\boldsymbol{B}^{\boldsymbol{\mu}}_{a,\boldsymbol{n}}\boldsymbol{B}^{\boldsymbol{\mu}}_{\overline{a},\boldsymbol{n}'}}\\
&=R\boldsymbol{I}^{\boldsymbol{\mu}}(\ket*{\boldsymbol{B}^{\boldsymbol{\mu}}_{a,\boldsymbol{n}}}\otimes\ket*{\boldsymbol{B}^{\boldsymbol{\mu}}_{\overline{a},\boldsymbol{n}'}})\\
&=(\boldsymbol{U}^{\boldsymbol{\mu}}_A\otimes\boldsymbol{U}^{\boldsymbol{\mu}}_{\overline{A}})\boldsymbol{J}^{\boldsymbol{\mu}}(\ket*{\boldsymbol{B}^{\boldsymbol{\mu}}_{a,\boldsymbol{n}}}\otimes\ket*{\boldsymbol{B}^{\boldsymbol{\mu}}_{\overline{a},\boldsymbol{n}'}})\\
&=(\boldsymbol{U}^{\boldsymbol{\mu}}_A\otimes\boldsymbol{U}^{\boldsymbol{\mu}}_{\overline{A}})(\ket*{\boldsymbol{B}^{\boldsymbol{\mu}}_{a,\boldsymbol{n}}}\otimes\ket{\chi^{\boldsymbol{\mu}}}\otimes\ket*{\boldsymbol{B}^{\boldsymbol{\mu}}_{\overline{a},\boldsymbol{n}'}}),
\end{split}
\end{align}
where we use $\ket*{\boldsymbol{B}^{\boldsymbol{\mu}}_{a,\boldsymbol{n}}\boldsymbol{B}^{\boldsymbol{\mu}}_{\overline{a},\boldsymbol{n}'}}$ abbreviated for $\boldsymbol{I}^{\boldsymbol{\mu}}(\ket*{\boldsymbol{B}^{\boldsymbol{\mu}}_{a,\boldsymbol{n}}}\otimes\ket*{\boldsymbol{B}^{\boldsymbol{\mu}}_{\overline{a},\boldsymbol{n}'}})$.
A more succinct description of the above equations can be given by the commutative diagram
\begin{equation}\label{diagram11}
\begin{tikzcd}  
&\mathcal{E}^{\boldsymbol{\mu}}_a\otimes\mathcal{E}^{\boldsymbol{\mu}}_{\overline{a}}\arrow{r}{\boldsymbol{J}^{\boldsymbol{\mu}}}\arrow{d}[swap]{\boldsymbol{I}^{\boldsymbol{\mu}}}&\mathcal{E}^{\boldsymbol{\mu}}_a\otimes\mathcal{F}^{\boldsymbol{\mu}}_a\otimes\mathcal{F}^{\boldsymbol{\mu}}_{\overline{a}}\otimes\mathcal{E}^{\boldsymbol{\mu}}_{\overline{a}}\arrow{d}{\boldsymbol{U}^{\boldsymbol{\mu}}_A\otimes\boldsymbol{U}^{\boldsymbol{\mu}}_{\overline{A}}}&\\
&\mathcal{E}\arrow{r}{R}&\mathcal{H}_A\otimes\mathcal{H}_{\overline{A}}&
\end{tikzcd},
\end{equation}
in which ``commutative'' simply means that the two compositions of the linear maps are equal, i.e. Eq.~\ref{diagram10}. This diagram clearly shows how the decomposition of $\mathcal{E}$ corresponding to $R^+\mathcal{M}_A R$ and $R^+\mathcal{M}_{\overline{A}} R$ is ``encoded'' in the decompositions of the complementary boundary parts. It essentially bridges the von Neumann algebra description of subregion duality to the RT formula of entanglement entropy~\cite{harlow2017}.

Then, we consider an arbitrary $\widetilde{\rho}$ on $\mathcal{H}_{\mathrm{code}}$. According to Ref.~\cite{harlow2017}, once the decompositions of $\mathcal{E}$, $\mathcal{H}_A$ and $\mathcal{H}_{\overline{A}}$, together with the isometries $\{\boldsymbol{J}^{\boldsymbol{\mu}}\}$ (or the $\ket{\chi^{\boldsymbol{\mu}}}$s) are specified and proved to satisfy the conditions described in Eq.~\ref{decl1} and in Eq.~\ref{diagram10} (\ref{diagram11}), we can define the projected part of  $\widetilde{\rho}$ onto $\mathcal{E}^{\boldsymbol{\mu}}_a\otimes\mathcal{E}^{\boldsymbol{\mu}}_{\overline{a}}$ as $\widetilde{\boldsymbol{\rho}}^{\boldsymbol{\mu}}_{a\overline{a}}={\boldsymbol{I}^{\boldsymbol{\mu}}}^+R^+\widetilde{\rho}R\boldsymbol{I}^{\boldsymbol{\mu}}$, the partial-traced density operator on $\mathcal{E}^{\boldsymbol{\mu}}_a$ as $p_{\boldsymbol{\mu}}\widetilde{\boldsymbol{\rho}}^{\boldsymbol{\mu}}_a=\mathrm{Tr}^{\boldsymbol{\mu}}_{\overline{a}}[\widetilde{\boldsymbol{\rho}}^{\boldsymbol{\mu}}_{a\overline{a}}]$ and the partial-traced density operator of $\ket{\chi^{\boldsymbol{\mu}}}$ on $\mathcal{F}^{\boldsymbol{\mu}}_a$ as $\chi^{\boldsymbol{\mu}}_a$, so that the terms in Eq.~\ref{rtformula0} can be extracted as
\begin{align}\label{rtformula1}
\begin{split}
&\mathrm{Z}(\mathcal{M}_A)\ni\mathcal{L}_A=\sum_{\boldsymbol{\mu}}\mathrm{S}(\chi^{\boldsymbol{\mu}}_a)R\widetilde{\boldsymbol{P}}^{\boldsymbol{\mu}}R^+,\\
&\mathrm{S}(\widetilde{\rho},\mathcal{M}_A)=-\sum_{\boldsymbol{\mu}}p_{\boldsymbol{\mu}}\mathrm{log}p_{\boldsymbol{\mu}}+\sum_{\boldsymbol{\mu}}p_{\boldsymbol{\mu}}\mathrm{S}(\widetilde{\boldsymbol{\rho}}^{\boldsymbol{\mu}}_a),\\
&\mathrm{S}(\widetilde{\rho},\mathcal{M}_{\overline{A}})=-\sum_{\boldsymbol{\mu}}p_{\boldsymbol{\mu}}\mathrm{log}p_{\boldsymbol{\mu}}+\sum_{\boldsymbol{\mu}}p_{\boldsymbol{\mu}}\mathrm{S}(\widetilde{\boldsymbol{\rho}}^{\boldsymbol{\mu}}_{\overline{a}}),
\end{split}
\end{align}
where $R\widetilde{\boldsymbol{P}}^{\boldsymbol{\mu}}R^+$ is the projection operator on $\mathcal{H}_{\mathrm{code}}$ induced by the decomposition.

\subsubsection{Splits and the bulk terms}\label{sbtee}
According to the above prescription, to extract the area and the bulk terms as described in Eq.~\ref{rtformula1} for an arbitrary state $\widetilde{\rho}$, we simply need to specify concrete decompositions and isometric maps as described in Eq.~\ref{decl1} and \ref{decl2}, and show that they satisfy the condition described by Eq.~\ref{diagram10} or \ref{diagram11}. We firstly show that the decomposition $\mathcal{E}=\oplus_{\boldsymbol{\mu}}(\mathcal{E}^{\boldsymbol{\mu}}_a\otimes\mathcal{E}^{\boldsymbol{\mu}}_{\overline{a}})$ (in Eq.~\ref{decl1}) can be naturally specified in terms of the splits on the entangling surface, based on which we can explicitly compute the bulk terms $\mathrm{S}(\widetilde{\rho},\mathcal{M}_A)$ and $\mathrm{S}(\widetilde{\rho},\mathcal{M}_{\overline{A}})$.

Recall that as discussed in Sec.~\ref{exsplitsec}, corresponding to the split $(\boldsymbol{\mathcal{M}}_a(\boldsymbol{x}),\boldsymbol{\mathcal{M}}_{\overline{a}}(\boldsymbol{x}))$, the local Hilbert space of each bulk qudit $\boldsymbol{x}$ in the entangling surface $\mathrm{E}[A\overline{A}]$ also has a decomposition $\mathfrak{e}_{\boldsymbol{x}}=\oplus_{\mu}(\mathfrak{e}_{\boldsymbol{x}a}^{\mu}\otimes\mathfrak{e}_{\boldsymbol{x}\overline{a}}^{\mu})$. According to such decompositions, the Hilbert space $\otimes_{\boldsymbol{x}\in\mathrm{E}[A\overline{A}]}\mathfrak{e}_{\boldsymbol{x}}$ for the bulk subsystem on the entangling surface has the following structure
\begin{align}\label{decl3}
\begin{split}
&~~~\otimes_{\boldsymbol{x}\in\mathrm{E}[A\overline{A}]}\mathfrak{e}_{\boldsymbol{x}}\\
&=\otimes_{\boldsymbol{x}\in\mathrm{E}[A\overline{A}]}[\oplus_{\mu_{\boldsymbol{x}}}(\mathfrak{e}_{\boldsymbol{x}a}^{\mu_{\boldsymbol{x}}}\otimes\mathfrak{e}_{\boldsymbol{x}\overline{a}}^{\mu_{\boldsymbol{x}}})]\\
&=[\oplus_{\mu_{\boldsymbol{x}_1}}(\mathfrak{e}_{\boldsymbol{x}a}^{\mu_{\boldsymbol{x}_1}}\otimes\mathfrak{e}_{\boldsymbol{x}\overline{a}}^{\mu_{\boldsymbol{x}_1}})]\otimes[\oplus_{\mu_{\boldsymbol{x}_2}}(\mathfrak{e}_{\boldsymbol{x}a}^{\mu_{\boldsymbol{x}_2}}\otimes\mathfrak{e}_{\boldsymbol{x}\overline{a}}^{\mu_{\boldsymbol{x}_2}})]\otimes\cdots\\
&=\oplus_{\{\mu_{\boldsymbol{x}}\}}[(\mathfrak{e}_{\boldsymbol{x}a}^{\mu_{\boldsymbol{x}_1}}\otimes\mathfrak{e}_{\boldsymbol{x}\overline{a}}^{\mu_{\boldsymbol{x}_1}})\otimes(\mathfrak{e}_{\boldsymbol{x}a}^{\mu_{\boldsymbol{x}_2}}\otimes\mathfrak{e}_{\boldsymbol{x}\overline{a}}^{\mu_{\boldsymbol{x}_2}})\otimes\cdots]\\
&=\oplus_{\{\mu_{\boldsymbol{x}}\}}[(\mathfrak{e}_{\boldsymbol{x}a}^{\mu_{\boldsymbol{x}_1}}\otimes\mathfrak{e}_{\boldsymbol{x}a}^{\mu_{\boldsymbol{x}_2}}\cdots)\otimes(\mathfrak{e}_{\boldsymbol{x}\overline{a}}^{\mu_{\boldsymbol{x}_1}}\otimes\mathfrak{e}_{\boldsymbol{x}\overline{a}}^{\mu_{\boldsymbol{x}_2}}\cdots)]\\
&=\oplus_{\{\mu_{\boldsymbol{x}}\}}[(\otimes_{\boldsymbol{x}\in\mathrm{E}[A\overline{A}]}\mathfrak{e}_{\boldsymbol{x}a}^{\mu_{\boldsymbol{x}}})\otimes(\otimes_{\boldsymbol{x}\in\mathrm{E}[A\overline{A}]}\mathfrak{e}_{\boldsymbol{x}\overline{a}}^{\mu_{\boldsymbol{x}}})].
\end{split}
\end{align}
Here, the index $\{\mu_{\boldsymbol{x}}\}=\mu_{\boldsymbol{x}_1},\mu_{\boldsymbol{x}_2},\ldots$ of the direct sum runs through all bulk qudits in the entangling surface $\mathrm{E}[A\overline{A}]$ and includes all possible $\mu_{\boldsymbol{x}}$ for each $\boldsymbol{x}$. Comparing the above factor spaces $\otimes_{\boldsymbol{x}\in\mathrm{E}[A\overline{A}]}\mathfrak{e}_{\boldsymbol{x}a}^{\mu_{\boldsymbol{x}}}$ and $\otimes_{\boldsymbol{x}\in\mathrm{E}[A\overline{A}]}\mathfrak{e}_{\boldsymbol{x}\overline{a}}^{\mu_{\boldsymbol{x}}}$ with $\mathcal{E}^{\boldsymbol{\mu}}_a$ and $\mathcal{E}^{\boldsymbol{\mu}}_{\overline{a}}$ in Eq.~\ref{decl1}, or, with $\mathfrak{e}_{\boldsymbol{x}a}^{\mu}$ and $\mathfrak{e}_{\boldsymbol{x}\overline{a}}^{\mu}$ in Eq.~\ref{decl0}, it is clear that the decomposition $\otimes_{\boldsymbol{x}\in\mathrm{E}[A\overline{A}]}\mathfrak{e}_{\boldsymbol{x}}=\oplus_{\{\mu_{\boldsymbol{x}}\}}[(\otimes_{\boldsymbol{x}\in\mathrm{E}[A\overline{A}]}\mathfrak{e}_{\boldsymbol{x}a}^{\mu_{\boldsymbol{x}}})\otimes(\otimes_{\boldsymbol{x}\in\mathrm{E}[A\overline{A}]}\mathfrak{e}_{\boldsymbol{x}\overline{a}}^{\mu_{\boldsymbol{x}}})]$ determines two von Neumann algebras as commutants to one another (in analogous to $R^+\mathcal{M}_A R$ and $R^+\mathcal{M}_{\overline{A}} R$, or to $\boldsymbol{\mathcal{M}}_a(\boldsymbol{x})$ and $\boldsymbol{\mathcal{M}}_{\overline{a}}(\boldsymbol{x})$). And such von Neumann algebras essentially represent the total split of the whole bulk degrees of freedom in the entangling surface.

Now, inserting the above decomposition for the total split into the expression of the bulk Hilbert space $\mathcal{E}=(\otimes_{\boldsymbol{x}\in\mathrm{W}[A]}\mathfrak{e}_{\boldsymbol{x}})\otimes(\otimes_{\boldsymbol{x}\in\mathrm{E}[A\overline{A}]}\mathfrak{e}_{\boldsymbol{x}})\otimes(\otimes_{\boldsymbol{x}\in\mathrm{W}[\overline{A}]}\mathfrak{e}_{\boldsymbol{x}})$, we have
\begin{align}\label{decl4}
\begin{split}
\mathcal{E}&=(\otimes_{\boldsymbol{x}\in\mathrm{W}[A]}\mathfrak{e}_{\boldsymbol{x}})\otimes(\otimes_{\boldsymbol{x}\in\mathrm{E}[A\overline{A}]}\mathfrak{e}_{\boldsymbol{x}})\otimes(\otimes_{\boldsymbol{x}\in\mathrm{W}[\overline{A}]}\mathfrak{e}_{\boldsymbol{x}})\\
&=\{\otimes_{\boldsymbol{x}\in\mathrm{W}[A]}\mathfrak{e}_{\boldsymbol{x}}\}\otimes\{\oplus_{\{\mu_{\boldsymbol{x}}\}}[(\otimes_{\boldsymbol{x}\in\mathrm{E}[A\overline{A}]}\mathfrak{e}_{\boldsymbol{x}a}^{\mu_{\boldsymbol{x}}})\\
&~~~\otimes(\otimes_{\boldsymbol{x}\in\mathrm{E}[A\overline{A}]}\mathfrak{e}_{\boldsymbol{x}\overline{a}}^{\mu_{\boldsymbol{x}}})]\}\otimes\{\otimes_{\boldsymbol{x}\in\mathrm{W}[\overline{A}]}\mathfrak{e}_{\boldsymbol{x}}\}\\
&=\oplus_{\{\mu_{\boldsymbol{x}}\}}\{[(\otimes_{\boldsymbol{x}\in\mathrm{W}[A]}\mathfrak{e}_{\boldsymbol{x}})\otimes(\otimes_{\boldsymbol{x}\in\mathrm{E}[A\overline{A}]}\mathfrak{e}_{\boldsymbol{x}a}^{\mu_{\boldsymbol{x}}})]\\
&~~~\otimes[(\otimes_{\boldsymbol{x}\in\mathrm{E}[A\overline{A}]}\mathfrak{e}_{\boldsymbol{x}{\overline{a}}}^{\mu_{\boldsymbol{x}}})\otimes(\otimes_{\boldsymbol{x}\in\mathrm{W}[\overline{A}]}\mathfrak{e}_{\boldsymbol{x}})]\},
\end{split}
\end{align}
which can specify a decomposition and corresponding bases
\begin{align}\label{decl5}
\begin{split}
&\mathcal{E}=\oplus_{\{\mu_{\boldsymbol{x}}\}}(\mathcal{E}^{\{\mu_{\boldsymbol{x}}\}}_a\otimes\mathcal{E}^{\{\mu_{\boldsymbol{x}}\}}_{\overline{a}}),\\
&\mathcal{E}^{\{\mu_{\boldsymbol{x}}\}}_a=(\otimes_{\boldsymbol{x}\in\mathrm{W}[A]}\mathfrak{e}_{\boldsymbol{x}})\otimes(\otimes_{\boldsymbol{x}\in\mathrm{E}[A\overline{A}]}\mathfrak{e}_{\boldsymbol{x}a}^{\mu_{\boldsymbol{x}}}),\\
&\mathcal{E}^{\{\mu_{\boldsymbol{x}}\}}_{\overline{a}}=(\otimes_{\boldsymbol{x}\in\mathrm{E}[A\overline{A}]}\mathfrak{e}_{\boldsymbol{x}{\overline{a}}}^{\mu_{\boldsymbol{x}}})\otimes(\otimes_{\boldsymbol{x}\in\mathrm{W}[\overline{A}]}\mathfrak{e}_{\boldsymbol{x}})\\
&\ket*{\boldsymbol{B}^{\{\mu_{\boldsymbol{x}}\}}_{a,\boldsymbol{n}}}=(\otimes_{\boldsymbol{x}\in\mathrm{W}[A]}\ket{\boldsymbol{\beta}_{\boldsymbol{x}}})\otimes(\otimes_{\boldsymbol{x}\in\mathrm{E}[A\overline{A}]}\ket{\boldsymbol{\beta}^{\mu_{\boldsymbol{x}}}_{\boldsymbol{x}a}})\\
&\ket*{\boldsymbol{B}^{\{\mu_{\boldsymbol{x}}\}}_{\overline{a},\boldsymbol{n}'}}=(\otimes_{\boldsymbol{x}\in\mathrm{E}[A\overline{A}]}\ket{\boldsymbol{\beta}^{\mu_{\boldsymbol{x}}}_{\boldsymbol{x}\overline{a}}})\otimes(\otimes_{\boldsymbol{x}\in\mathrm{W}[\overline{A}]}\ket{\boldsymbol{\beta}_{\boldsymbol{x}}}),
\end{split}
\end{align}
where $\{\ket{\boldsymbol{\beta}^{\mu_{\boldsymbol{x}}}_{\boldsymbol{x}a}}\}$ and $\{\ket{\boldsymbol{\beta}^{\mu_{\boldsymbol{x}}}_{\boldsymbol{x}\overline{a}}}\}$, as bases of $\mathfrak{e}_{\boldsymbol{x}a}^{\mu_{\boldsymbol{x}}}$ and $\mathfrak{e}_{\boldsymbol{x}\overline{a}}^{\mu_{\boldsymbol{x}}}$ respectively have been specified in App.~\ref{posplit}.

To show that this is the desired decomposition, we have to check the relationship on the operators as shown in Eq.~\ref{decl1}, i.e. to prove that the von Neumann algebras specified by this decomposition are exactly $R^+\mathcal{M}_A R$ and $R^+\mathcal{M}_{\overline{A}} R$. Indeed, according to the derivation detailed in App.~\ref{podecl}, which takes advantage of Prop.~\ref{sova}, we can prove the following proposition.
\begin{proposition}\label{declprop}
Following the same conditions as in Prop.~\ref{sova}, the decomposition given by Eq.~\ref{decl5} is the desired one, i.e., we have
\begin{align}
\begin{split}
&R^+\mathcal{M}_A R\ni\widetilde{\boldsymbol{O}}=\sum_{\{\mu_{\boldsymbol{x}}\}}\widetilde{\boldsymbol{O}}^{\{\mu_{\boldsymbol{x}}\}}_a\otimes\mathds{1}^{\{\mu_{\boldsymbol{x}}\}}_{\overline{a}},\\
&R^+\mathcal{M}_{\overline{A}} R\ni\widetilde{\boldsymbol{O}}=\sum_{\{\mu_{\boldsymbol{x}}\}}\mathds{1}^{\{\mu_{\boldsymbol{x}}\}}_a\otimes\widetilde{\boldsymbol{O}}^{\{\mu_{\boldsymbol{x}}\}}_{\overline{a}},\\
&\mathrm{Z}(R^+\mathcal{M}_A R)\ni\widetilde{\boldsymbol{O}}=\sum_{\{\mu_{\boldsymbol{x}}\}}c_{\{\mu_{\boldsymbol{x}}\}}\mathds{1}^{\{\mu_{\boldsymbol{x}}\}}_a\otimes\mathds{1}^{\{\mu_{\boldsymbol{x}}\}}_{\overline{a}},
\end{split}
\end{align}
i.e., the von Neumann algebra on the left consists of operators of the form on the right. Here, compared with Eq.~\ref{decl1}, we simply omit the isometry $\boldsymbol{I}^{\{\mu_{\boldsymbol{x}}\}}$ for simplicity.
\end{proposition}

Note that the specification of the decomposition (Eq.~\ref{decl5}) and the above proposition is general for holographic code within the formalism of OAQEC, since the splits in the entangling surface is an intrinsic structure (see Sec.~\ref{sdew}).

The decomposition given in Eq.~\ref{decl5} enables us to explicitly compute the bulk term for any state $\widetilde{\rho}$, since it relates the $\{\ket*{\boldsymbol{B}^{\{\mu_{\boldsymbol{x}}\}}_{a,\boldsymbol{n}}}\otimes\ket*{\boldsymbol{B}^{\{\mu_{\boldsymbol{x}}\}}_{\overline{a},\boldsymbol{n}'}}\}$ basis that is compatible with the decomposition to the boundary code basis $\{\ket{\widetilde{\varphi}_n}\}$ that are used to define the model (see Eq.~\ref{cbs1}). Indeed, the following relation (based on Eq.~\ref{decl4})
\begin{align}\label{area0}
\begin{split}
&\ket*{\boldsymbol{B}^{\{\mu_{\boldsymbol{x}}\}}_{a,\boldsymbol{n}}}\otimes\ket*{\boldsymbol{B}^{\{\mu_{\boldsymbol{x}}\}}_{\overline{a},\boldsymbol{n}'}}\\
&=(\otimes_{\boldsymbol{x}\in\mathrm{W}[A]}\ket{\boldsymbol{\beta}_{\boldsymbol{x}}})\otimes(\otimes_{\boldsymbol{x}\in\mathrm{E}[A\overline{A}]}\ket{\boldsymbol{\beta}^{\mu_{\boldsymbol{x}}}_{\boldsymbol{x}a}})\\
&\otimes(\otimes_{\boldsymbol{x}\in\mathrm{E}[A\overline{A}]}\ket{\boldsymbol{\beta}^{\mu_{\boldsymbol{x}}}_{\boldsymbol{x}\overline{a}}})\otimes(\otimes_{\boldsymbol{x}\in\mathrm{W}[\overline{A}]}\ket{\boldsymbol{\beta}_{\boldsymbol{x}}})\\
&=(\otimes_{\boldsymbol{x}\in\mathrm{W}[A]}\ket{\boldsymbol{\beta}_{\boldsymbol{x}}})\otimes[\otimes_{\boldsymbol{x}\in\mathrm{E}[A\overline{A}]}(\ket{\boldsymbol{\beta}^{\mu_{\boldsymbol{x}}}_{\boldsymbol{x}a}}\otimes\ket{\boldsymbol{\beta}^{\mu_{\boldsymbol{x}}}_{\boldsymbol{x}\overline{a}}})]\\
&\otimes(\otimes_{\boldsymbol{x}\in\mathrm{W}[\overline{A}]}\ket{\boldsymbol{\beta}_{\boldsymbol{x}}}),
\end{split}
\end{align}
can be further clarified in terms of the details of the splits as listed in Sec.~\ref{exsplitsec} and App.~\ref{posplit}, i.e., on whether $\ket{\boldsymbol{\beta}^{\mu_{\boldsymbol{x}}}_{\boldsymbol{x}a}}\otimes\ket{\boldsymbol{\beta}^{\mu_{\boldsymbol{x}}}_{\boldsymbol{x}\overline{a}}}\in\mathfrak{e}_{\boldsymbol{x}a}^{\mu_{\boldsymbol{x}}}\otimes\mathfrak{e}_{\boldsymbol{x}\overline{a}}^{\mu_{\boldsymbol{x}}}$ is simply identified as a bulk local basis state $\ket{\boldsymbol{\beta}_{\boldsymbol{x}}}\in\mathfrak{e}_{\boldsymbol{x}}$ or identified as a linear sum of the basis states.

According to App.~\ref{posplit}, for type-2 and type-3 splits, $\ket{\boldsymbol{\beta}^{\mu_{\boldsymbol{x}}}_{\boldsymbol{x}a}}\otimes\ket{\boldsymbol{\beta}^{\mu_{\boldsymbol{x}}}_{\boldsymbol{x}\overline{a}}}$ is simply identified as a bulk local basis state $\ket{\boldsymbol{\beta}_{\boldsymbol{x}}}$, while for type-1 split, $\ket{\boldsymbol{\beta}^{\mu_{\boldsymbol{x}}}_{\boldsymbol{x}a}}\otimes\ket{\boldsymbol{\beta}^{\mu_{\boldsymbol{x}}}_{\boldsymbol{x}\overline{a}}}$ is identified as $(1/\sqrt{2})(\ket{\boldsymbol{\beta}_{\boldsymbol{x}}}\pm\ket{\boldsymbol{\beta}'_{\boldsymbol{x}}})$. Then we have
\begin{align}\label{area1}
\begin{split}
&\ket*{\boldsymbol{B}^{\{\mu_{\boldsymbol{x}}\}}_{a,\boldsymbol{n}}}\otimes\ket*{\boldsymbol{B}^{\{\mu_{\boldsymbol{x}}\}}_{\overline{a},\boldsymbol{n}'}}\\
&=2^{-\frac{N_1}{2}}\ket{\boldsymbol{\beta}_{\boldsymbol{1}}\cdots(\boldsymbol{\beta}_{\boldsymbol{x}}\pm\boldsymbol{\beta}'_{\boldsymbol{x}})\cdots(\boldsymbol{\beta}_{\boldsymbol{x}'}\pm\boldsymbol{\beta}'_{\boldsymbol{x}'})\cdots},\\
&R\ket*{\boldsymbol{B}^{\{\mu_{\boldsymbol{x}}\}}_{a,\boldsymbol{n}} \boldsymbol{B}^{\{\mu_{\boldsymbol{x}}\}}_{\overline{a},\boldsymbol{n}'}}\\
&=2^{-\frac{N_1}{2}}R(\ket{\boldsymbol{\beta}_{\boldsymbol{1}}\cdots\boldsymbol{\beta}_{\boldsymbol{x}}\cdots\boldsymbol{\beta}_{\boldsymbol{x}'}\cdots}\\
&\pm \ket{\boldsymbol{\beta}_{\boldsymbol{1}}\cdots\boldsymbol{\beta}_{\boldsymbol{x}}\cdots\boldsymbol{\beta}'_{\boldsymbol{x}'}\cdots}\pm\ket{\boldsymbol{\beta}_{\boldsymbol{1}}\cdots\boldsymbol{\beta}'_{\boldsymbol{x}}\cdots\boldsymbol{\beta}_{\boldsymbol{x}'}\cdots}\\
&\pm\ket{\boldsymbol{\beta}_{\boldsymbol{1}}\cdots\boldsymbol{\beta}'_{\boldsymbol{x}}\cdots\boldsymbol{\beta}'_{\boldsymbol{x}'}\cdots}\pm\cdots)\\
&=2^{-\frac{N_1}{2}}(\ket{\widetilde{\varphi}_n}\pm\ket{\widetilde{\varphi}_{n'}}\pm\ket{\widetilde{\varphi}_{n''}}\pm\cdots).
\end{split}
\end{align}
It follows that if the expansion coefficients of $\widetilde{\rho}$ in terms of $\dyad{\widetilde{\varphi}_n}{\widetilde{\varphi}_{n'}}$ can be specified, through Eq.~\ref{area1} we can specify how $\widetilde{\boldsymbol{\rho}}=R^+\widetilde{\rho}R$ is expanded by the $\ket*{\boldsymbol{B}^{\{\mu_{\boldsymbol{x}}\}}_{a,\boldsymbol{n}}}\otimes\ket*{\boldsymbol{B}^{\{\mu_{\boldsymbol{x}}\}}_{\overline{a},\boldsymbol{n}'}}$ states  and hence compute the bulk terms according to Eq.~\ref{rtformula1}.

\subsubsection{Extracting the area term: entanglement patterns revisited}\label{rtformula}
Now, according to Eq.~\ref{decl5} and Prop.~\ref{declprop}, we show how to specify the desired decompositions for the boundary complementary subsystem $\mathcal{H}_A$ and $\mathcal{H}_{\overline{A}}$ respectively, and also the corresponding family of isometric linear maps (see Eq.~\ref{decl2}), i.e.,
\begin{align}\label{area2}
\begin{split}
&\boldsymbol{U}^{\{\mu_{\boldsymbol{x}}\}}_A: \mathcal{E}^{\{\mu_{\boldsymbol{x}}\}}_a\otimes\mathcal{F}^{\{\mu_{\boldsymbol{x}}\}}_a\rightarrow\mathcal{H}_A,\\
&\boldsymbol{U}^{\{\mu_{\boldsymbol{x}}\}}_A(\ket*{\boldsymbol{B}^{\{\mu_{\boldsymbol{x}}\}}_{a,\boldsymbol{n}}}\otimes\ket*{\boldsymbol{Z}^{\{\mu_{\boldsymbol{x}}\}}_{a,\boldsymbol{l}}})=\ket*{\widetilde{\varphi}^{\{\mu_{\boldsymbol{x}}\}}_{A, \boldsymbol{n}\boldsymbol{l}}},\\
&\boldsymbol{U}^{\{\mu_{\boldsymbol{x}}\}}_{\overline{A}}: \mathcal{E}^{\{\mu_{\boldsymbol{x}}\}}_{\overline{a}}\otimes\mathcal{F}^{\{\mu_{\boldsymbol{x}}\}}_{\overline{a}}\rightarrow\mathcal{H}_{\overline{A}},\\
&\boldsymbol{U}^{\{\mu_{\boldsymbol{x}}\}}_{\overline{A}}(\ket*{\boldsymbol{B}^{\{\mu_{\boldsymbol{x}}\}}_{{\overline{a}},\boldsymbol{n}'}}\otimes\ket*{\boldsymbol{Z}^{\{\mu_{\boldsymbol{x}}\}}_{{\overline{a}},\boldsymbol{l}'}})=\ket*{\widetilde{\varphi}^{\{\mu_{\boldsymbol{x}}\}}_{{\overline{A}}, \boldsymbol{n}'\boldsymbol{l}'}},\\
&\boldsymbol{J}^{\{\mu_{\boldsymbol{x}}\}}(\ket*{\boldsymbol{B}^{\{\mu_{\boldsymbol{x}}\}}_{a,\boldsymbol{n}}}\otimes\ket*{\boldsymbol{B}^{\{\mu_{\boldsymbol{x}}\}}_{\overline{a},\boldsymbol{n}'}})\\
&=\ket*{\boldsymbol{B}^{\{\mu_{\boldsymbol{x}}\}}_{a,\boldsymbol{n}}}\otimes\ket*{\chi^{\{\mu_{\boldsymbol{x}}\}}}\otimes\ket*{\boldsymbol{B}^{\{\mu_{\boldsymbol{x}}\}}_{\overline{a},\boldsymbol{n}'}},
\end{split}
\end{align}
where $\{\ket*{\boldsymbol{Z}^{\{\mu_{\boldsymbol{x}}\}}_{a,\boldsymbol{l}}}\}$ and $\{\ket*{\boldsymbol{Z}^{\{\mu_{\boldsymbol{x}}\}}_{{\overline{a}},\boldsymbol{l}'}}\}$ are bases of $\mathcal{F}^{\{\mu_{\boldsymbol{x}}\}}_a$ and $\mathcal{F}^{\{\mu_{\boldsymbol{x}}\}}_{\overline{a}}$ respectively. We will prove that the way we specify the structures will ensure Eq.~\ref{diagram10} (\ref{diagram11}), i.e. $R\boldsymbol{I}^{\{\mu_{\boldsymbol{x}}\}}=(\boldsymbol{U}^{\{\mu_{\boldsymbol{x}}\}}_A\otimes\boldsymbol{U}^{\{\mu_{\boldsymbol{x}}\}}_{\overline{A}})\boldsymbol{J}^{\{\mu_{\boldsymbol{x}}\}}$ so that the $\{\ket*{\chi^{\{\mu_{\boldsymbol{x}}\}}}\}$ states will give rise to the area term.

As discussed previously, these decompositions are expected to represent how those factors $\mathcal{E}^{\{\mu_{\boldsymbol{x}}\}}_a$ and $\mathcal{E}^{\{\mu_{\boldsymbol{x}}\}}_{\overline{a}}$ can be read off the complementary boundary subregions. Hence, the structures of such decompositions in a holographic code should reflect general properties of the boundary entanglement with respect to how the bulk structure is emergent.

The general idea for specifying these structures can be described as follows with backward reasoning. First, if $\boldsymbol{U}^{\{\mu_{\boldsymbol{x}}\}}_A$ maps $\ket*{\boldsymbol{B}^{\{\mu_{\boldsymbol{x}}\}}_{a,\boldsymbol{n}}}\otimes\ket*{\boldsymbol{Z}^{\{\mu_{\boldsymbol{x}}\}}_{a,\boldsymbol{l}}}$ to some partial boundary state $\ket*{\widetilde{\varphi}^{\{\mu_{\boldsymbol{x}}\}}_{A, \boldsymbol{n}\boldsymbol{l}}}$ on subregion $A$, the basis state $\ket*{\boldsymbol{Z}^{\{\mu_{\boldsymbol{x}}\}}_{a,\boldsymbol{l}}}\in\mathcal{F}^{\{\mu_{\boldsymbol{x}}\}}_a$ can be understood as an auxiliary state that specify the partial boundary states together with $\ket*{\boldsymbol{B}^{\{\mu_{\boldsymbol{x}}\}}_{a,\boldsymbol{n}}}$, and so are $\ket*{\boldsymbol{Z}^{\{\mu_{\boldsymbol{x}}\}}_{\overline{a},\boldsymbol{l}'}}\in\mathcal{F}^{\{\mu_{\boldsymbol{x}}\}}_{\overline{a}}$ and $\ket*{\widetilde{\varphi}^{\{\mu_{\boldsymbol{x}}\}}_{\overline{A}, \boldsymbol{n}'\boldsymbol{l}'}}$ for the boundary subregion $\overline{A}$. Then, if the decompositions and isometries satisfy $R\ket*{\boldsymbol{B}^{\{\mu_{\boldsymbol{x}}\}}_{a,\boldsymbol{n}}\boldsymbol{B}^{\{\mu_{\boldsymbol{x}}\}}_{\overline{a},\boldsymbol{n}'}}=(\boldsymbol{U}^{\{\mu_{\boldsymbol{x}}\}}_A\otimes\boldsymbol{U}^{\{\mu_{\boldsymbol{x}}\}}_{\overline{A}})(\ket*{\boldsymbol{B}^{\{\mu_{\boldsymbol{x}}\}}_{a,\boldsymbol{n}}}\otimes\ket*{\chi^{\{\mu_{\boldsymbol{x}}\}}}\otimes\ket*{\boldsymbol{B}^{\{\mu_{\boldsymbol{x}}\}}_{\overline{a},\boldsymbol{n}'}})$, the total effect of $\boldsymbol{U}^{\{\mu_{\boldsymbol{x}}\}}_A$, $\boldsymbol{U}^{\{\mu_{\boldsymbol{x}}\}}_{\overline{A}}$ and the entangled $\ket*{\chi^{\{\mu_{\boldsymbol{x}}\}}}=\sum_{\boldsymbol{l}\boldsymbol{l}'}c_{\boldsymbol{l}\boldsymbol{l}'}\ket*{\boldsymbol{Z}^{\{\mu_{\boldsymbol{x}}\}}_{a,\boldsymbol{l}}}\otimes\ket*{\boldsymbol{Z}^{\{\mu_{\boldsymbol{x}}\}}_{\overline{a},\boldsymbol{l}'}}$ are expected to realize the meaning of ``sewing'' the partial boundary states $\{\ket*{\widetilde{\varphi}^{\{\mu_{\boldsymbol{x}}\}}_{A, \boldsymbol{n}\boldsymbol{l}}}\}$ and $\{\ket*{\widetilde{\varphi}^{\{\mu_{\boldsymbol{x}}\}}_{\overline{A}, \boldsymbol{n}'\boldsymbol{l}'}}\}$ into the full boundary code state $R\ket*{\boldsymbol{B}^{\{\mu_{\boldsymbol{x}}\}}_{a,\boldsymbol{n}}\boldsymbol{B}^{\{\mu_{\boldsymbol{x}}\}}_{\overline{a},\boldsymbol{n}'}}=2^{-\frac{N_1}{2}}(\ket{\widetilde{\varphi}_n}\pm\ket{\widetilde{\varphi}_{n'}}\pm\ket{\widetilde{\varphi}_{n''}}\pm\cdots)$.

To see how two of the partial boundary states $\ket*{\widetilde{\varphi}^{\{\mu_{\boldsymbol{x}}\}}_{A, \boldsymbol{n}\boldsymbol{l}}}$ and $\ket*{\widetilde{\varphi}^{\{\mu_{\boldsymbol{x}}\}}_{\overline{A}, \boldsymbol{n}'\boldsymbol{l}'}}$ can contribute to the ``sewing'' in our model, we switch to the alternative geometric setting and recall that the boundary code basis states $\ket{\widetilde{\varphi}_n}$s are equal-weight sums of the physical qudit-product-states $\ket{\psi_m}$s whose properties underlie the entanglement patterns of the $\ket{\widetilde{\varphi}_n}$s. According to the boundary bipartition, e.g., the one shown in Fig.~\ref{21a}, we can tear $\ket{\psi_m}$ into halves $\ket*{\psi_{A,m_A}}$ and $\ket*{\psi_{\overline{A},m_{\overline{A}}}}$ as qudit-product-states on boundary subregions $A$ and $\overline{A}$ respectively, as shown in Fig.~\ref{21b}. Then, viewing $\ket*{\widetilde{\varphi}^{\{\mu_{\boldsymbol{x}}\}}_{A, \boldsymbol{n}\boldsymbol{l}}}$ and $\ket*{\widetilde{\varphi}^{\{\mu_{\boldsymbol{x}}\}}_{\overline{A}, \boldsymbol{n}'\boldsymbol{l}'}}$ also as sums of certain $\{\ket*{\psi_{A,m_A}}\}$ and $\{\ket*{\psi_{\overline{A},m_{\overline{A}}}}\}$, the ``sewing'' is essentially matching the $\ket*{\psi_{A,m_A}}$s and the $\ket*{\psi_{\overline{A},m_{\overline{A}}}}$s and concatenating the matched into the $\ket{\psi_m}$s that participate in the expansion of $R\ket*{\boldsymbol{B}^{\{\mu_{\boldsymbol{x}}\}}_{a,\boldsymbol{n}}\boldsymbol{B}^{\{\mu_{\boldsymbol{x}}\}}_{\overline{a},\boldsymbol{n}'}}$.

\begin{center}
\begin{figure}[ht]
\centering
    \includegraphics[width=8.5cm]{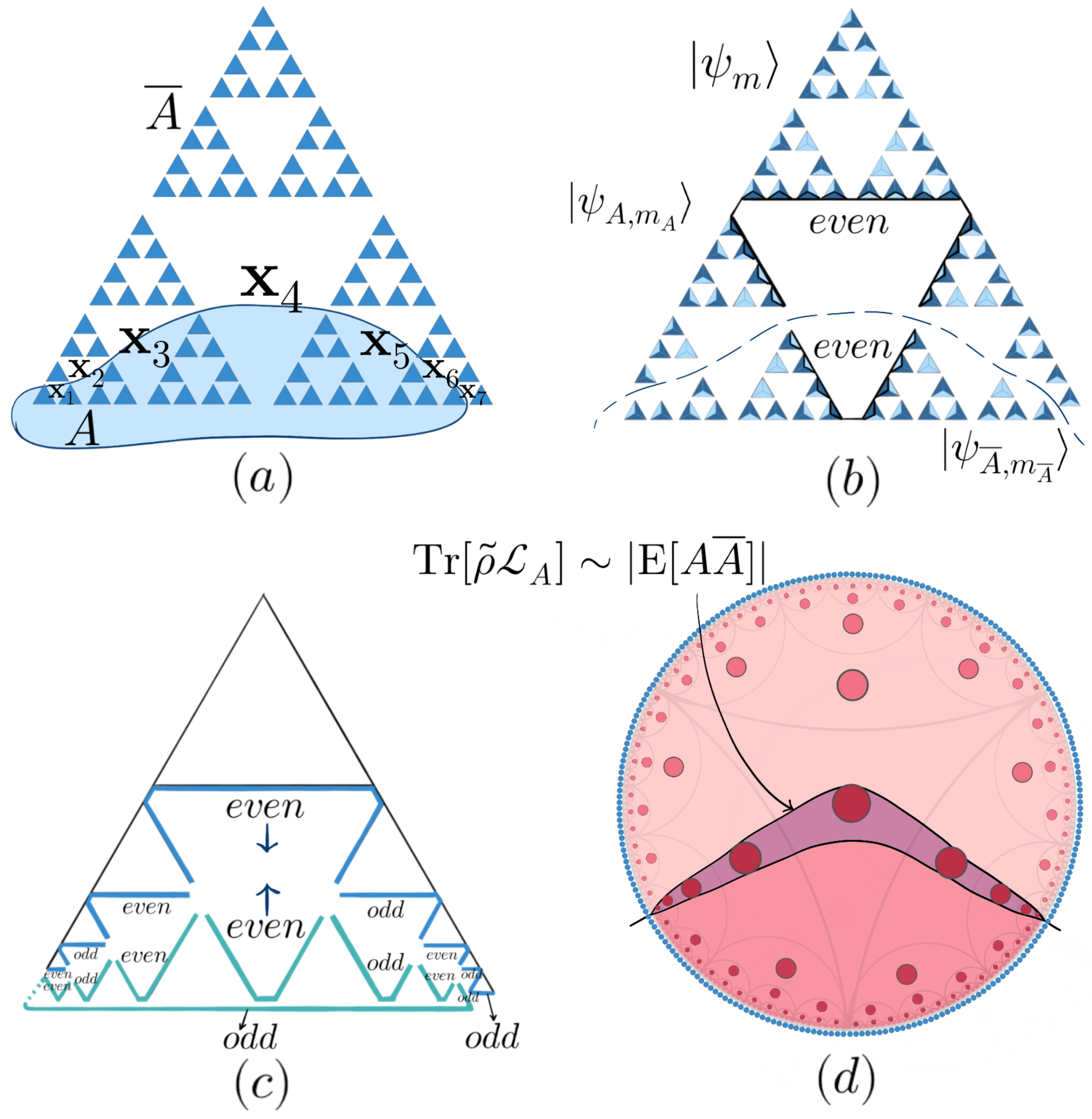}   
\phantomsubfloat{\label{21a}}\phantomsubfloat{\label{21b}}
\phantomsubfloat{\label{21c}}\phantomsubfloat{\label{21d}}
\caption{(a) An example of boundary bipartition $A\overline{A}$ as present in the alternative geometry. (b) A $\ket{\psi_m}$ state is torn into halves by the bipartition, and so is the parity of number of dark sides on the loop surrounding a hole. (c) A necessary condition for a $\ket*{\psi_{A,m_A}}$ and a $\ket*{\psi_{\overline{A},m_{\overline{A}}}}$ to be matched and concatenated into a $\ket{\psi_m}$ is that the parities of the number of dark sides on the torn halves (of the loop) must match. (d) The area term scales linearly on the size of the entangling surface.}
\label{fig21}
\end{figure}
\end{center}

In this perspective, in satisfying $R\ket*{\boldsymbol{B}^{\{\mu_{\boldsymbol{x}}\}}_{a,\boldsymbol{n}}\boldsymbol{B}^{\{\mu_{\boldsymbol{x}}\}}_{\overline{a},\boldsymbol{n}'}}=(\boldsymbol{U}^{\{\mu_{\boldsymbol{x}}\}}_A\otimes\boldsymbol{U}^{\{\mu_{\boldsymbol{x}}\}}_{\overline{A}})(\ket*{\boldsymbol{B}^{\{\mu_{\boldsymbol{x}}\}}_{a,\boldsymbol{n}}}\otimes\ket*{\chi^{\{\mu_{\boldsymbol{x}}\}}}\otimes\ket*{\boldsymbol{B}^{\{\mu_{\boldsymbol{x}}\}}_{\overline{a},\boldsymbol{n}'}})$, the role of the pairs of states $(\ket*{\boldsymbol{B}^{\{\mu_{\boldsymbol{x}}\}}_{a,\boldsymbol{n}}},\ket*{\boldsymbol{Z}^{\{\mu_{\boldsymbol{x}}\}}_{a,\boldsymbol{l}}})$ and $(\ket*{\boldsymbol{B}^{\{\mu_{\boldsymbol{x}}\}}_{\overline{a},\boldsymbol{n}'}},\ket*{\boldsymbol{Z}^{\{\mu_{\boldsymbol{x}}\}}_{\overline{a},\boldsymbol{l}'}})$ is to qualify that $\ket*{\psi_{A,m_A}}$s and $\ket*{\psi_{\overline{A},m_{\overline{A}}}}$s, upon being matched, are concatenated into the expected $\ket{\psi_m}$s (contributing to the expansion of $R\ket*{\boldsymbol{B}^{\{\mu_{\boldsymbol{x}}\}}_{a,\boldsymbol{n}}\boldsymbol{B}^{\{\mu_{\boldsymbol{x}}\}}_{\overline{a},\boldsymbol{n}'}}$), while the entanglement in $\ket*{\chi^{\{\mu_{\boldsymbol{x}}\}}}=\sum_{\boldsymbol{l}\boldsymbol{l}'}c_{\boldsymbol{l}\boldsymbol{l}'}\ket*{\boldsymbol{Z}^{\{\mu_{\boldsymbol{x}}\}}_{a,\boldsymbol{l}}}\otimes\ket*{\boldsymbol{Z}^{\{\mu_{\boldsymbol{x}}\}}_{\overline{a},\boldsymbol{l}'}}$ are in charge of judging which pairs of $\ket*{\psi_{A,m_A}}$ and $\ket*{\psi_{\overline{A},m_{\overline{A}}}}$ can be matched.

In App.~\ref{portformula}, we explicitly realize the above general idea in a representative and concrete example of connected boundary bipartition. Based on the example, the detailed process of specifying the decompositions and isometries in Eq.~\ref{area2}, and also the proof for them to satisfy Eq.~\ref{diagram10} (\ref{diagram11}) can be straightforwardly generalized to other connected boundary bipartitions and also disconnected bipartitions relevant for uberholography.

As shown in App.~\ref{portformula}, we can define
\begin{align}
\begin{split}
&\mathcal{F}^{\{\mu_{\boldsymbol{x}}\}}_a=\mathfrak{f}_{\boldsymbol{x}_1a}\otimes\mathfrak{f}_{\boldsymbol{x}_2a}\otimes\cdots,\\
&\mathfrak{f}_{\boldsymbol{x}a}=\mathbb{C}^2, \quad \mathfrak{f}_{\boldsymbol{x}a}\ni\ket{\zeta_{\boldsymbol{x}a}}=\ket{even},\ket{odd},\\
&\ket*{\boldsymbol{Z}^{\{\mu_{\boldsymbol{x}}\}}_{a,\boldsymbol{l}}}=\ket{\zeta_{\boldsymbol{x}_1a}\zeta_{\boldsymbol{x}_2a}\cdots},\\
&\mathcal{F}^{\{\mu_{\boldsymbol{x}}\}}_{\overline{a}}=\mathfrak{f}_{\boldsymbol{x}_1{\overline{a}}}\otimes\mathfrak{f}_{\boldsymbol{x}_2{\overline{a}}}\otimes\cdots,\\
&\mathfrak{f}_{\boldsymbol{x}{\overline{a}}}=\mathbb{C}^2, \quad \mathfrak{f}_{\boldsymbol{x}{\overline{a}}}\ni\ket{\zeta_{\boldsymbol{x}\overline{a}}}=\ket{even},\ket{odd},\\
&\ket*{\boldsymbol{Z}^{\{\mu_{\boldsymbol{x}}\}}_{{\overline{a}},\boldsymbol{l}'}}=\ket{\zeta_{\boldsymbol{x}_1\overline{a}}\zeta_{\boldsymbol{x}_2\overline{a}}\cdots},
\end{split}
\end{align}
where the $\boldsymbol{x}_1,\boldsymbol{x}_2,\ldots,\boldsymbol{x}_{N_E}$ basically agrees with the bulk qudits in the entangling surface (see Fig.~\ref{21a}), and $N_E$ scales linearly on the size of the entangling surface.

To see the meaning of each $\mathfrak{f}_{\boldsymbol{x}a}$ or $\mathfrak{f}_{\boldsymbol{x}\bar a}$, we note that the boundary bipartition $A\overline{A}$ tears the alternative geometry and also the holes corresponding to the entangling surface into halves (see Fig.\ref{21b} and \ref{21c}). And we recall that the characterizing property of the $\ket{\psi_m}$ states is that the parity of the number of dark sides on each loop surrounding each hole $\boldsymbol{x}$ (also on each of the three laterals of the lattice) is even, which underlies the description of the entanglement patterns of the boundary code states $\ket{\widetilde{\varphi}_n}$s (see Sec.~\ref{pattern1}). Hence, a necessary condition for judging $\ket*{\psi_{A,m_A}}$ and $\ket*{\psi_{\overline{A},m_{\overline{A}}}}$ to be matched and concatenated into a $\ket{\psi_m}$ state is that the parities on the torn halves of each hole match as $(even,even)$ or $(odd,odd)$, as shown in Fig.~\ref{21c}. In view of this, we index the basis states $\ket{\zeta_{\boldsymbol{x}a}}$ or $\ket{\zeta_{\boldsymbol{x}\bar a}}$ in each $\mathfrak{f}_{\boldsymbol{x}a}$ or $\mathfrak{f}_{\boldsymbol{x}\bar a}$ simply by the parity so that the basis states $\ket*{\boldsymbol{Z}^{\{\mu_{\boldsymbol{x}}\}}_{a,\boldsymbol{l}}}$ and $\ket*{\boldsymbol{Z}^{\{\mu_{\boldsymbol{x}}\}}_{{\overline{a}},\boldsymbol{l}'}}$ are simply indexed by the configurations of the parities on the torn halves on $A$ and $\overline{A}$ respectively.

Then, according to App.~\ref{portformula}, upon the detailed definition of the decompositions $\boldsymbol{U}^{\{\mu_{\boldsymbol{x}}\}}_A$ and $\boldsymbol{U}^{\{\mu_{\boldsymbol{x}}\}}_{\overline{A}}$ following the above general idea, if the isometries $\{\boldsymbol{J}^{\{\mu_{\boldsymbol{x}}\}}\}$ are simply specified by defining the ``judge'' as
\begin{align}
\begin{split}
&\ket*{\chi^{\{\mu_{\boldsymbol{x}}\}}}=\frac{1}{\sqrt{2^{N_E}}}\sum_{\boldsymbol{l}}\ket*{\boldsymbol{Z}^{\{\mu_{\boldsymbol{x}}\}}_{a,\boldsymbol{l}}}\otimes\ket*{\boldsymbol{Z}^{\{\mu_{\boldsymbol{x}}\}}_{\bar a,\boldsymbol{l}}}\\
&=\frac{1}{\sqrt{2^{N_E}}}(\ket{even,even,\ldots}\otimes\ket{even,even,\ldots}\\
&+\ket{even,odd,\ldots}\otimes\ket{even,odd,\ldots}\\
&+\cdots+\ket{odd,odd,\ldots}\otimes\ket{odd,odd,\ldots}),
\end{split}
\end{align} 
where the sum goes through all the $2^{N_E}$ possibilities, then we can prove Eq.~\ref{diagram10} (\ref{diagram11}) for these defined structures.

It follows that we have
\begin{align}
\begin{split}
&\mathrm{S}(\chi^{\{\mu_{\boldsymbol{x}}\}}_a)=N_E\mathrm{log}2,\\
&\mathrm{Z}(\mathcal{M}_A)\ni\mathcal{L}_A=\sum_{\boldsymbol{\mu}}N_E\mathrm{log}2R\widetilde{\boldsymbol{P}}^{\boldsymbol{\mu}}R^+\\
&=N_E\mathrm{log}2\mathds{1}_{\mathcal{H}_{\mathrm{code}}},
\end{split}
\end{align}
i.e., the area term for any $\widetilde{\rho}$ scales linearly on the size of the entangling surface (see Fig.~\ref{21d}). It can be checked by more examples for subregion duality in our code to see that this number also scales linearly on the number of tiles that the minimal surface crosses or passes by.

\subsection{Remarks on the exact model}\label{remark3}
According to all the arguments above and in the previous sections, we have formally and comprehensively demonstrated the HQEC characteristics described in Sec.~\ref{hqecc} for our code, in the forms of both theorems and descriptions combined with pictorial illustrations. We believe that these demonstrations sufficiently confirms the feasibility and the advantage of our proposed novel approach for constructing exact models of holographic code.

While the demonstrations ensure that our code can capture certain essence of the hypothesized quantum-information interpretation of the AdS/CFT correspondence, our code might be regarded as a minimalistic model of HQEC in the sense that our construction is to be further developed for more complex and hence more optimal models in studying AdS/CFT. We briefly comment on such development from the following three aspects.

Firstly, as we have shown, the area operator of the RT formula in our code appears to be a constant operator, though the bulk term confirms the essential difference from the subsystem-code formalism. While it is an attempt to realize a non-constant area operator in an exact holographic code constructed within our approach, as pointed out in Ref.~\cite{cao2023}, this might be a challenge not just a matter of model construction, but require further investigations on the hypothetical HQEC characteristics themselves. That is because the existence of a non-constant area operator seems conflicting to certain properties also expected in the construction of exact models~\cite{cao2023}. Alternatively, construction of approximate holographic code based on our model might be a possible route towards realizing the physics of bulk gravity with subleading corrections~\cite{cao2021,cao2023}. Indeed, the explicit boundary state structures of our model might facilitate investigations in this direction.

Second, while the area term in our code scales linearly on the size of the entangling surface, there are certain cases of boundary bipartitions for which the entangling surface contains constant number of bulk qudits and is inconsistent with the expected meaning (see Fig.~\ref{22a}). According to our discussion in Sec.~\ref{remark1}, this implies that the alternative geometry and hence the bulk discretization in our construction, due to its simplicity, is barely sufficient to underlie the formal demonstration of the HQEC characteristics, but inadequate to ideally describe the AdS bulk geometry as cut off for a finite system.

Instead of searching for more complex alternative geometry, which should guarantee more effective modeling of the AdS bulk geometry, a simple extension of our code can already significantly improve the insufficiency. That is, we can consider a double-layer stacking of our code as shown in Fig.~\ref{22b} and \ref{22c}, where the upper copy is rotated for $\pi/3$. In the stacked code we can view each boundary qudit as $\mathbb{C}^8$ and its construction can also be viewed as lying with our approach, for which the alternative geometry is shown in Fig.~\ref{22b}. It should be easy to check in the stacked code the connected code distance, the complementary recovery, the genuine OAQEC and other characteristics. Then, as shown in Fig.~\ref{22d}, the improved entangling surface can be viewed as containing bulk qudits from bulk copies and hence appears consistent with the expected properties.

Third, the above observation, e.g., a comparison between the bulk geometry in Fig.~\ref{22a} and that in Fig.~\ref{22d}, might imply that the boundary entanglement patterns of our code, as can be embedded in a larger code, ``partially'' captures the expected global entanglement structures modeling the AdS/CFT correspondence. Conversely, our code itself might be more relevant to some ``reduced'' version of the AdS/CFT correspondence.

Indeed, there are certain qualitative similarities between the emergent bulk locality of our code and that of the $p$-adic AdS/CFT. For instances, the bulk qudit in our code can be viewed as living on the vertex of certain tree graph (comparing Fig.~\ref{22e} and \ref{22f}), i.e., the dual graph of the ideal tiling. And the entangling surfaces as shown in Fig.~\ref{21d}, \ref{22a} and \ref{22e} appear consistent with the branching structures of the tree shown in Fig.~\ref{22f}. Because the behavior of entanglement entropy for the $p$-adic AdS/CFT remains puzzling according to the current studies~\cite{bhattacharyya2018,hung2019}, much less the corresponding HQEC structures, it might be important to combine and compare the studies on the $p$-adic AdS/CFT with our code or certain extension for future studies. For example, we might extend our construction to a variant alternative geometry as shown in Fig.~\ref{22g}, which seems to underlie a bulk tree structure (in a similar way as in our code) for the $3$-adic AdS geometry without center in the infinite limit.

\onecolumngrid
\begin{center}
\begin{figure}[ht]
\centering
    \includegraphics[width=16cm]{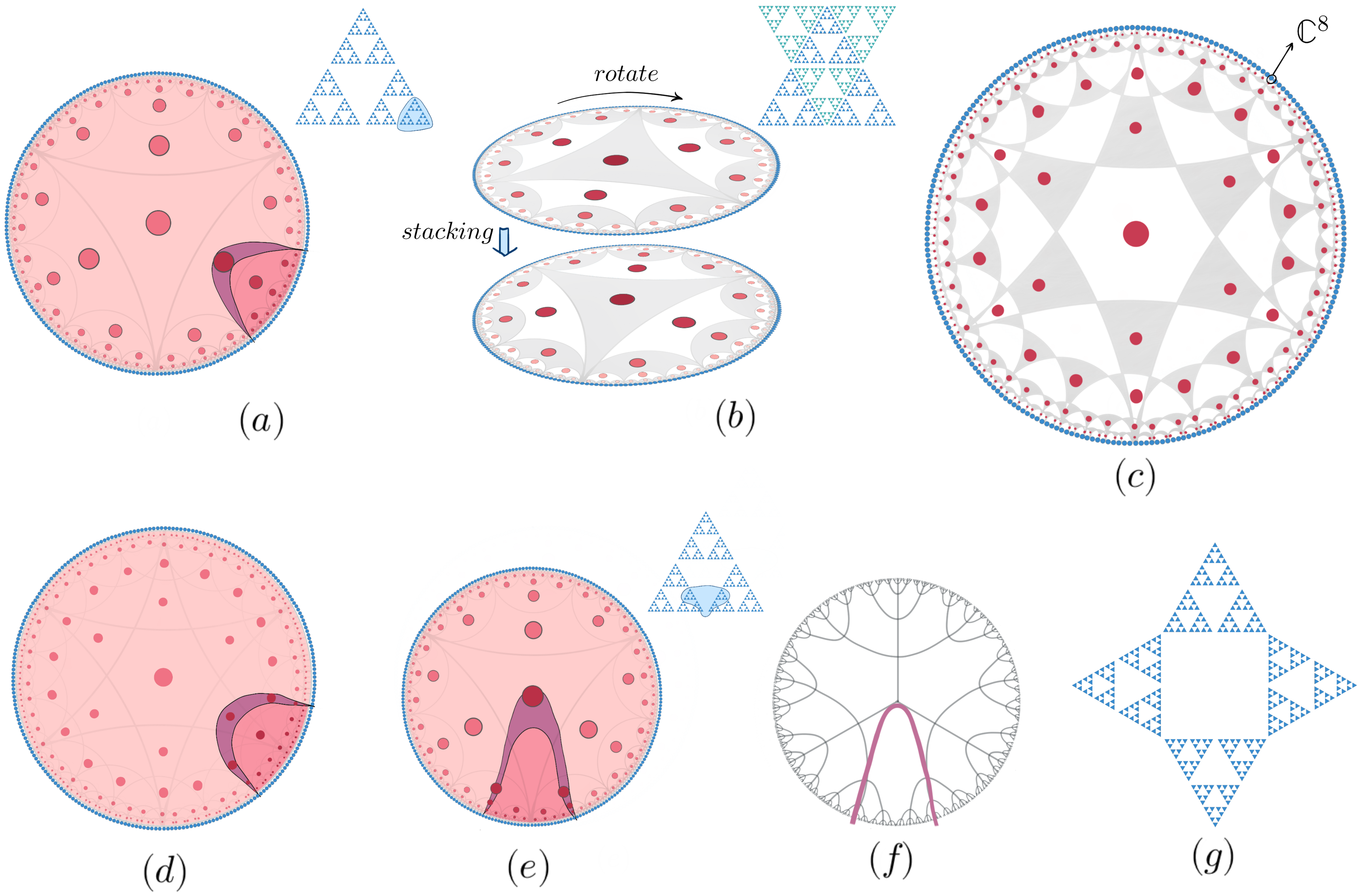}   
\phantomsubfloat{\label{22a}}\phantomsubfloat{\label{22b}}
\phantomsubfloat{\label{22c}}\phantomsubfloat{\label{22d}}
\phantomsubfloat{\label{22e}}\phantomsubfloat{\label{22f}}
\phantomsubfloat{\label{22g}}
\caption{(a) The entangling surface in the example of subregion duality includes one bulk qudit. (b) Double-layer stacking of our codes, with the upper rotated for $\pi/3$. (c) The stacked code with each boundary qudit viewed as $\mathbb{C}^8$. (d) The entangling surface in the stacked code corresponding to the same boundary bipartition as (a). The entangling surface in another example of subregion duality in (e) exhibits similarities to the geometric in a hyperbolic tree shown in (f). (g) A fractional-Hausdorff-dimension structure which might give rise to the bulk geometry of $3$-adic AdS/CFT through rearrangement.}
\label{fig22}
\end{figure}
\end{center}
\twocolumngrid

It is noticeable that our code cannot be constructed as a tree tensor-network of perfect tensors, since any two branches cannot recover the central bulk qudit (see Fig.~\ref{16i}, \ref{16j}, \ref{16k} and \ref{16l}). But our code might be related to more complicated tensor-network structures, e.g., those developed for capturing the correlation properties of $p$-adic AdS/CFT~\cite{bhattacharyya2018,hung2019}. In this possibility, it might be an entry point to unveil the potential deep relation between our approach for constructing holographic code and the tensor-network paradigm.

\section{Discussion and outlook}
With the systematic framework and the illustrative model, we have introduced and elaborated a novel approach for building exact models of holographic quantum error-correcting code (HQEC). This approach takes an opposite route to the conventional tensor-network paradigm, and explicitly exemplifies the holographic principle, i.e. how the bulk is emergent from the boundary entanglement, in the construction of an exact model. In other words, instead of contracting tensors from the bulk to the boundary, our construction starts with scalable descriptions of the boundary/physical qudits and specifies the bulk/logical degrees of freedom as emerged from the boundary entanglement patterns, based on which the hypothetical HQEC strucutre are comprehensively unfolded.

We believe that our work presents a fresh perspective for future studies and constructions of holographic code. Indeed, with the explicit structures of boundary code states, a wide range of research topics related to HQEC, e.g., the advanced topics of quantum gravity, the toy models of black holes, and the dynamics, might be studied in an analytic, and more efficient and comprehensive way. But more importantly, our work might significantly advance researches in the following directions.

Our work illustrates how the hypothetical properties of bulk reconstruction proposed for capturing the AdS/CFT correspondence can be simply unfolded from ``elementary'' descriptions of the physical qudits, which is applicable to broad contexts. Hence, for the emerging directions which attempt to realize and study HQEC within different areas, e.g., the potential quantum simulation studies of the physics related to quantum gravity, our work might have established a crucial bridge connecting the structure of HQEC to the interdisciplinary contexts, thereby substantially boosting the researches. Indeed, exploring our exact model or its extensions mentioned above in these emerging directions may already be sufficient to yield significant findings. For example, as shown in our exact model, the scalable descriptions of the boundary states can derive succinct prescription for direct quantum-circuit simulations of the boundary code states.

As illustrated in our exact model, our approach reveals explicit relationship between the description of the boundary and the characteristics of bulk operator reconstruction. According to this relationship, the exact code, as the initial model of our approach, can already formally demonstrate certain desired properties that are not yet proved in the established models. Because of this advantage, our approach is expected to play an important role in future studies of holographic code, and it might be promising to incorporate our approach and the conventional tensor-network paradigm towards the construction of more optimal models for the quantum-information studies of AdS/CFT.

It is noticeable that our arguments embody a general and systematic framework for demonstrating and studying subregion duality (including the RT formula) in the genuine formalism of operator algebra quantum error correction (OAQEC). To the best of our knowledge, such framework is rarely presented in the literature of HQEC, and hence can serve as a reference for future studies on the subleading corrections to the bulk gravity. Additionally, considering the recently significant renewed focus on OAQEC, our results might also shed light on the development of this subject in studies of quantum error correction.

It might be surprising that in demonstrating the properties of uberholography, our approach establishes a connection between HQEC and the studies of fractal or fractional-Hausdorff-dimension many-body physics which is attracting growing attention but is still in its infancy. According to our results, certain nonlocally entangled systems in fractional-Hausdorff-dimension can give rise to a 1D quantum systems with holographic dual. This observation is expected to be further investigated to drive advancements in both specific and interdisciplinary fields.

It might be also worth mentioning that our description of the boundary entanglement patterns coincides with the way the many-body entanglement is studied in the context of condensed matter physics. For example, the way the code states in our model can be expressed in terms of local quantum gates shares similarities with that of the long-range entangled states in topological matter. It might imply that further investigation of such similarity can pave the way for exploring potential connections between the quantum information structures applied in different fields, which is expected to reveal new source for future progress in quantum physics~\cite{hartnoll2021}.


\begin{acknowledgements}
The author thanks Anne E. B. Nielsen and Cristiane Morais Smith for inspiring discussions on the fractal geometry. The author thanks Jinwu Ye, Li You, Wei Ku, Chi Ming Yim and Hong Ding for inspiring discussions on relevant topics in quantum error correction, quantum gravity, quantum simulation and condensed matter physics. The author thanks Fangyuan Gu, Rajae Malek, Liu Yang, Yuzhu Cui, Weikang Lin, Zhen Wang for their help in reading the manuscript and their valuable suggestions. This work is supported by Tsung-Dao Lee Institute Postdoc Fellowship Research Funding.
\end{acknowledgements}

\appendix
\section{Proofs}
\subsection{Proof of the criteria of complementary recovery}\label{crcr}
In this section we give the proof of Prop.~\ref{cr1} and \ref{cr2}. We start with reviewing how the condition of complementary recovery for $\mathcal{H}_{\mathrm{code}}$ is stated in the literature: Not only certain $\mathcal{M}$ on $\mathcal{H}_{\mathrm{code}}$ can be reconstructed on subregion $A$, but also its commutant $\mathcal{M}'$ can be correspondingly reconstructed on the complement $\overline{A}$~\cite{harlow2017}. This requirement, though, does not explicitly specify which $\mathcal{M}$ is to be considered, it implicitly guarantees the uniqueness of such $\mathcal{M}$.

Indeed, according to definition of $\mathcal{M}_A, \mathcal{M}_{\overline{A}}$ and Lemma.~\ref{oaqec}, we have $\mathcal{M}_{\overline{A}}\subset\mathcal{M}'_A$. Now, suppose the condition of complementary recovery, i.e., for some von Neumann algebra $\mathcal{M}$ we have $\mathcal{M}\subset\mathcal{M}_A$ and $\mathcal{M}'\subset\mathcal{M}_{\overline{A}}$. Then, by the basic property of the commutant, the condition of complementary recovery implies $\mathcal{M}'_A\subset\mathcal{M}'\subset\mathcal{M}_{\overline{A}}$. It follows that we have $\mathcal{M}_{\overline{A}}=\mathcal{M}'_A=\mathcal{M}'$, which also means that $\mathcal{M}=\mathcal{M}_A$, i.e., $\mathcal{M}_A$ is the only possible von Neumann algebra to satisfy the complementary recovery. Reversely, it is obvious that if $\mathcal{M}_{\overline{A}}=\mathcal{M}'_A$, then the von Neumann algebra $\mathcal{M}_A$ ensures $\mathcal{H}_{\mathrm{code}}$ to satisfy the condition of complementary recovery.

The above arguments have simply proved Prop.~\ref{cr1}: \emph{For a given decomposition $\mathcal{H}=\mathcal{H}_{A}\otimes\mathcal{H}_{\overline{A}}$, the code $\mathcal{H}_{\mathrm{code}}$ exhibits complementary recovery if and only if $\mathcal{M}_A$ and $\mathcal{M}'_A$ can be reconstructed on $A$ and on $\overline{A}$ respectively, or equivalently, $\mathcal{M}'_A=\mathcal{M}_{\overline{A}}$ ($\mathcal{M}'_{\overline{A}}=\mathcal{M}_A$).}

The motivation for Prop.~\ref{cr2} is that we need to prove the equality $\mathcal{M}'_A=\mathcal{M}_{\overline{A}}$ in our code. The application of those properties requires further knowledge of operators in $\mathcal{M}'_A$ as described by the lemma below.

\begin{lemma}\label{mvn}
$\mathcal{M}'_A$ is the minimal von Neumann algebra containing all operators of the form $P_{\mathrm{code}}O_{\overline{A}}P_{\mathrm{code}}$ ($O_{\overline{A}}$ not necessarily commuting with $P_{\mathrm{code}}$).
\end{lemma}

\begin{proof}
We use $\mathcal{N}$ to denote the minimal von Neumann algebra containing all operators of the form $P_{\mathrm{code}}O_{\overline{A}}P_{\mathrm{code}}$ ($O_{\overline{A}}$ not necessarily commuting with $P_{\mathrm{code}}$). $\mathcal{N}$ is simply the intersection of all von Neumann algebras containing those operators. It is easy to check that such intersection is still a von Neumann algebra. 

According to the definition of $\mathcal{M}_A$, Lemma~\ref{oaqec} implies that $\mathcal{M}'_A$ contains all those $P_{\mathrm{code}}O_{\overline{A}}P_{\mathrm{code}}$ operators and hence we have $\mathcal{N}\subset\mathcal{M}'_A$. 

Now apply Lemma~\ref{oaqec} to $\mathcal{N}'$. Since $\mathcal{N}=(\mathcal{N}')'$ contains all the $P_{\mathrm{code}}O_{\overline{A}}P_{\mathrm{code}}$ operators, $\mathcal{N}'$ must be reconstructed on $A$, i.e., we have $\mathcal{N}'\subset\mathcal{M}_A$, or equivalently, $\mathcal{M}'_A\subset\mathcal{N}$.

Combining the above arguments, we have $\mathcal{M}'_A=\mathcal{N}$.
\end{proof}

Now, combining Proposition~\ref{cr1} and Lemma~\ref{mvn}, we can conclude: the condition of the complementary recovery is satisfied if and only if $\mathcal{M}_{\overline{A}}$ equals the minimal von Neumann algebra containing all the $P_{\mathrm{code}}O_{\overline{A}}P_{\mathrm{code}}$ operators. Based on this conclusion, we can derive a more direct and applicable criterion for for the complementary recovery. Indeed, $\mathcal{M}_{\overline{A}}$ is, by definition, included in the minimal algebra. Then, as long as we can prove the reverse, the complementary recovery will be confirmed. Obviously, the reverse simply requires that every $P_{\mathrm{code}}O_{\overline{A}}P_{\mathrm{code}}$, even with $O_{\overline{A}}$ not commuting with $P_{\mathrm{code}}$, is included in $\mathcal{M}_{\overline{A}}$. Hence, according to this observation, we can easily prove Prop.~\ref{cr2}:

\emph{For a given decomposition $\mathcal{H}=\mathcal{H}_{A}\otimes\mathcal{H}_{\overline{A}}$, the code $\mathcal{H}_{\mathrm{code}}$ exhibits complementary recovery if and only if every $P_{\mathrm{code}}O_{\overline{A}}P_{\mathrm{code}}$, with $O_{\overline{A}}$ not necessarily commuting with $P_{\mathrm{code}}$, equals $P_{\mathrm{code}}Q_{\overline{A}}P_{\mathrm{code}}$ for some $Q_{\overline{A}}$ commuting with $P_{\mathrm{code}}$.}

\subsection{Proof of Eq.~\ref{ssc}}\label{possc}
In this proof we focus on showing the equivalent version of Eq.~\ref{ssc}, i.e.,
\begin{align*}
\begin{split}
&\mathcal{M}_A=R(\otimes_{\boldsymbol{x}\in\mathrm{W}[A]}\mathbf{L}(\mathfrak{e}_{\boldsymbol{x}}))R^+,\\
&\mathcal{M}_{\overline{A}}=R(\otimes_{\boldsymbol{x}\in\mathrm{W}[\overline{A}]}\mathbf{L}(\mathfrak{e}_{\boldsymbol{x}}))R^+.
\end{split}
\end{align*}
And then based on the equalities, it will be easy to prove the condition of complementary recovery and the trivial center. Our proof utilizes the fact that according to the commutation of tensor products of von Neumann algebras~\cite{stratila2019}, we have $\big[(\otimes_{\boldsymbol{x}\in\mathrm{W}[A]}\mathbb{C}\mathds{1}_{\mathfrak{e}_{\boldsymbol{x}}})\otimes(\otimes_{\boldsymbol{x}\in\mathrm{W}[\overline{A}]}\mathbf{L}(\mathfrak{e}_{\boldsymbol{x}}))\big]'=(\otimes_{\boldsymbol{x}\in\mathrm{W}[A]}\mathbf{L}(\mathfrak{e}_{\boldsymbol{x}}))\otimes(\otimes_{\boldsymbol{x}\in\mathrm{W}[\overline{A}]}\mathbb{C}\mathds{1}_{\mathfrak{e}_{\boldsymbol{x}}})$. In the simplified notations, it is $\otimes_{\boldsymbol{x}\in\mathrm{W}[\overline{A}]}\mathbf{L}(\mathfrak{e}_{\boldsymbol{x}})=(\otimes_{\boldsymbol{x}\in\mathrm{W}[A]}\mathbf{L}(\mathfrak{e}_{\boldsymbol{x}}))'$, and through the isomorphism, we have $R(\otimes_{\boldsymbol{x}\in\mathrm{W}[\overline{A}]}\mathbf{L}(\mathfrak{e}_{\boldsymbol{x}}))R^+=(R(\otimes_{\boldsymbol{x}\in\mathrm{W}[A]}\mathbf{L}(\mathfrak{e}_{\boldsymbol{x}}))R^+)'$. It is easy to see that the center is trivial, as only the constant operators on $\mathcal{E}$ belong to both $\otimes_{\boldsymbol{x}\in\mathrm{W}[A]}\mathbf{L}(\mathfrak{e}_{\boldsymbol{x}})$ and $\otimes_{\boldsymbol{x}\in\mathrm{W}[\overline{A}]}\mathbf{L}(\mathfrak{e}_{\boldsymbol{x}})$.

To show the equalities, we first show that $R(\otimes_{\boldsymbol{x}\in\mathrm{W}[A]}\mathbf{L}(\mathfrak{e}_{\boldsymbol{x}}))R^+\subset\mathcal{M}_A$ and $R(\otimes_{\boldsymbol{x}\in\mathrm{W}[\overline{A}]}\mathbf{L}(\mathfrak{e}_{\boldsymbol{x}}))R^+\subset\mathcal{M}_{\overline{A}}$. To show that $R(\otimes_{\boldsymbol{x}\in\mathrm{W}[A]}\mathbf{L}(\mathfrak{e}_{\boldsymbol{x}}))R^+\subset\mathcal{M}_A$, we note that operators in $\otimes_{\boldsymbol{x}\in\mathrm{W}[A]}\mathbf{L}(\mathfrak{e}_{\boldsymbol{x}})$ are linear sums of product operator of the form $\otimes_{\boldsymbol{x}\in\mathrm{W}[A]}\widetilde{\boldsymbol{o}}_{\boldsymbol{x}}$ with $\widetilde{\boldsymbol{o}}_{\boldsymbol{x}}\in\mathbf{L}(\mathfrak{e}_{\boldsymbol{x}})$. And $\otimes_{\boldsymbol{x}\in\mathrm{W}[A]}\widetilde{\boldsymbol{o}}_{\boldsymbol{x}}$ simply equals the product $(\cdots\otimes\mathds{1}_{\mathfrak{e}_{\boldsymbol{x}'}}\otimes\widetilde{\boldsymbol{o}}_{\boldsymbol{x}}\otimes\mathds{1}_{\mathfrak{e}_{\boldsymbol{x}''}}\otimes\cdots)(\cdots\otimes\widetilde{\boldsymbol{o}}_{\boldsymbol{x}'}\otimes\mathds{1}_{\mathfrak{e}_{\boldsymbol{x}}}\otimes\mathds{1}_{\mathfrak{e}_{\boldsymbol{x}''}}\otimes\cdots)\cdots$ for $\boldsymbol{x},\boldsymbol{x}',\boldsymbol{x}'',\ldots\in\mathrm{W}[A]$, which means that $R\otimes_{\boldsymbol{x}\in\mathrm{W}[A]}\widetilde{\boldsymbol{o}}_{\boldsymbol{x}}R^+=R\cdots R^+R(\cdots\otimes\mathds{1}_{\mathfrak{e}_{\boldsymbol{x}'}}\otimes\widetilde{\boldsymbol{o}}_{\boldsymbol{x}}\otimes\mathds{1}_{\mathfrak{e}_{\boldsymbol{x}''}}\otimes\cdots)R^+R(\cdots\otimes\widetilde{\boldsymbol{o}}_{\boldsymbol{x}'}\otimes\mathds{1}_{\mathfrak{e}_{\boldsymbol{x}}}\otimes\mathds{1}_{\mathfrak{e}_{\boldsymbol{x}''}}\otimes\cdots)R^+R\cdots R^+$, i.e., a product of operators in $\mathcal{M}_A$ which also belongs to $\mathcal{M}_A$. Then, $\mathcal{M}_A$ surely includes $R(\otimes_{\boldsymbol{x}\in\mathrm{W}[A]}\mathbf{L}(\mathfrak{e}_{\boldsymbol{x}}))R^+$. Similarly, we can show that $R(\otimes_{\boldsymbol{x}\in\mathrm{W}[\overline{A}]}\mathbf{L}(\mathfrak{e}_{\boldsymbol{x}}))R^+\subset\mathcal{M}_{\overline{A}}$.

Now, we note that by definition, we have $\mathcal{M}_{\overline{A}}\subset\mathcal{M}'_A$, hence what we really have is $R(\otimes_{\boldsymbol{x}\in\mathrm{W}[\overline{A}]}\mathbf{L}(\mathfrak{e}_{\boldsymbol{x}}))R^+\subset\mathcal{M}'_A$. Then, according to the property of commutant, we have $\mathcal{M}_A\subset\big[R(\otimes_{\boldsymbol{x}\in\mathrm{W}[\overline{A}]}\mathbf{L}(\mathfrak{e}_{\boldsymbol{x}}))R^+\big]'= R(\otimes_{\boldsymbol{x}\in\mathrm{W}[A]}\mathbf{L}(\mathfrak{e}_{\boldsymbol{x}}))R^+$, which means that $R(\otimes_{\boldsymbol{x}\in\mathrm{W}[A]}\mathbf{L}(\mathfrak{e}_{\boldsymbol{x}}))R^+$ must equal $\mathcal{M}_A$. The arguments for $R(\otimes_{\boldsymbol{x}\in\mathrm{W}[\overline{A}]}\mathbf{L}(\mathfrak{e}_{\boldsymbol{x}}))R^+=\mathcal{M}_{\overline{A}}$ are similar.

Finally, since $R(\otimes_{\boldsymbol{x}\in\mathrm{W}[\overline{A}]}\mathbf{L}(\mathfrak{e}_{\boldsymbol{x}}))R^+=(R(\otimes_{\boldsymbol{x}\in\mathrm{W}[A]}\mathbf{L}(\mathfrak{e}_{\boldsymbol{x}}))R^+)'$, we have
$\mathcal{M}_{\overline{A}}=\mathcal{M}'_A$. And the trivial center is clear.

\subsection{Proof of Prop.~\ref{gates1}}\label{pogates1}
(1) To prove that states of the form $(T_{ii'}T_{jj'}\cdots)\ket{\psi_{m_0}}$ go through the whole collection $\{\ket{\psi_m}\}$, we can simply replace $\ket{\psi_{m_0}}$ by $(T_{kk'}T_{ll'}\cdots)\ket{0\cdots0\cdots0}$. Then we exactly have the form as given by Eq.~\ref{dpsi}, which is the definition of all the states in the collection $\{\ket{\psi_m}\}$.

(2) To show that distinct products $(T_{i_1i'_1}T_{i_2i'_2}\cdots)$ and $(T_{j_1j'_1}T_{j_2j'_2}\cdots)$ map $\ket{\psi_{m_0}}$ to distinct qudit-product-states, we prove by contradiction and assume the opposite, i.e., the equality
\begin{equation*}
(T_{i_1i'_1}T_{i_2i'_2}\cdots)\ket{\psi_{m_0}}=(T_{j_1j'_1}T_{j_2j'_2}\cdots)\ket{\psi_{m_0}}.
\end{equation*}
According to the basic properties of the $T_{ii'}$ operators ($T_{ii'}T_{ii'}=\mathds{1}$ and their commutativity), we have $(T_{i_1i'_1}T_{i_2i'_2}\cdots)(T_{i_1i'_1}T_{i_2i'_2}\cdots)=\mathds{1}$ and can rewrite the equality as
\begin{equation*}
\ket{\psi_{m_0}}=(T_{i_1i'_1}T_{i_2i'_2}\cdots)(T_{j_1j'_1}T_{j_2j'_2}\cdots)\ket{\psi_{m_0}}.
\end{equation*}
Thus, we just need to prove that this equality leads to contradiction. In the following arguments, we will rewrite the product $(T_{i_1i'_1}T_{i_2i'_2}\cdots)(T_{j_1j'_1}T_{j_2j'_2}\cdots)$ stepwise into a concise form as a product of the commutative $S_i^{\sigma_i}$ operators with a nontrivial support of qudits, and then show that the action of such product cannot leave $\ket{\psi_{m_0}}$ unchanged.

Step 1. Since $(T_{i_1i'_1}T_{i_2i'_2}\cdots)$ and $(T_{j_1j'_1}T_{j_2j'_2}\cdots)$ are distinct, their multiplication cancels the possible $T_{ii'}$s that appear in both the two product, and only certain distinct gates $T_{kk'},T_{ll'},\ldots$s that appear in one of the two brackets remain. Then, we have a new form of the product $(T_{i_1i'_1}T_{i_2i'_2}\cdots)(T_{j_1j'_1}T_{j_2j'_2}\cdots)=(T_{kk'}T_{ll'}\cdots)$, and can rewrite the equality as
\begin{equation*}
\ket{\psi_{m_0}}=(T_{kk'}T_{ll'}\cdots)\ket{\psi_{m_0}}.
\end{equation*}

Step 2. Since each $T_{kk'}$ is a product of two unitary operators $S_k^{\sigma_k}$ and $S_{k'}^{\sigma_{k'}}$ on qudits $k$ and $k'$ respectively, we can rewrite $(T_{kk'}T_{ll'}\cdots)$ as $(S_k^{\sigma_k}S_{k'}^{\sigma_{k'}}S_l^{\sigma_l}S_{l'}^{\sigma_{l'}}\cdots)$ with $\sigma_k,\sigma_{k'},\ldots=1,2,3$.

Step 3.  The product $(T_{kk'}T_{ll'}\cdots)=(S_k^{\sigma_k}S_{k'}^{\sigma_{k'}}S_l^{\sigma_l}S_{l'}^{\sigma_{l'}}\cdots)$ can be further reduced. That is, some qudits $l$ can repeat appearing in the supports of more than one gate from the three $T_{il}$, $T_{jl}$ and $T_{kl}$ corresponding to its three neighbors $i$, $j$ and $k$ respectively (see Fig.~\ref{app1a}). In that case, more than one out of the three single-qudit operators $S_l^{1}$, $S_l^{2}$ and $S_l^{3}$ appear in the product. Utilizing the commutativity of the single-qudit operators, we can multiply together the operators on the same qudit $l$. Then, according to the basic properties stated in Eq.~\ref{ss}, such multiplication on qudit $l$ has two possible results: (i) When all the three gates $T_{il}$, $T_{jl}$ and $T_{kl}$ appears in the product, the multiplication must be $S_l^{1}S_l^{2}S_l^{3}=\mathds{1}$ (see Fig.~\ref{app1a}), i.e., operators on the qudit $l$ cancel each other and hence qudit $l$ is not really in the support of the product. (ii) When only one or two out of the three gates $T_{il}$, $T_{jl}$ and $T_{kl}$ appear, the multiplication always results in some $S_l^{\bar{\sigma}_l}$ out of $S_l^{1}$, $S_l^{2}$ and $S_l^{3}$, and qudit $l$ belongs to the support of the product. Hence, the product $(S_k^{\sigma_k}S_{k'}^{\sigma_{k'}}S_l^{\sigma_l}S_{l'}^{\sigma_{l'}}\cdots)$ can be further reduced to a concise form $(S_k^{\bar{\sigma}_k}S_l^{\bar{\sigma}_l}\cdots)$ with $\bar{\sigma}_k,\bar{\sigma}_k,\ldots=1,2,3$ and with no repetition in the index of qudits $k,l,\ldots$. And the assumed equality can be eventually rewritten as
\begin{equation*}
\ket{\psi_{m_0}}=(S_k^{\bar{\sigma}_k}S_l^{\bar{\sigma}_l}\cdots)\ket{\psi_{m_0}}.
\end{equation*}

\begin{center}
\begin{figure}[ht]
\centering
    \includegraphics[width=7.5cm]{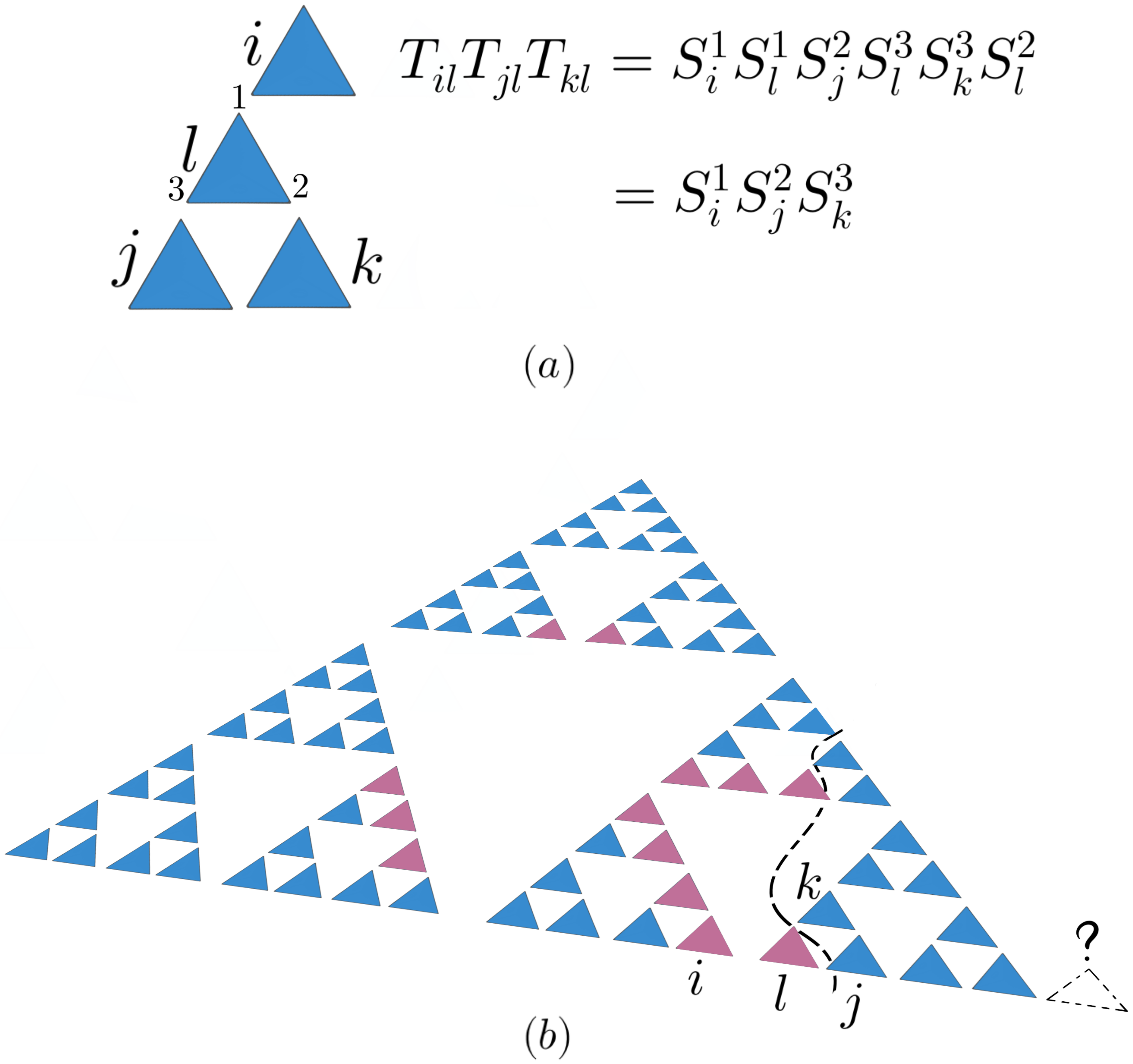}   
\phantomsubfloat{\label{app1a}}
\phantomsubfloat{\label{app1b}}
\caption{(a) The action of all the three gates sharing the qudit $l$ in their supports cancels the action on $l$. (b) The corner qudit (right) has only two neighbors. The purple indicate the qudits supporting at least one gate in $T_{kk'}T_{ll'}\cdots$, among which $l$ is the closest to the corner.}
\label{figapp1}
\end{figure}
\end{center}

Now, we show that the support of $(T_{kk'}T_{ll'}\cdots)=(S_k^{\bar{\sigma}_k}S_l^{\bar{\sigma}_l}\cdots)$ is nonempty. There are two possibilities. First, some gate in the product $(T_{kk'}T_{ll'}\cdots)$ acts on a corner qudits in the Sierpi\'nski geometry. Then, since the corner qudit has only two neighbors (see Fig.~\ref{app1b}), according to Step 3 above, the single-qudit operators on this corner qudit do not cancel, and hence the qudit belongs to the support of the product. Second, if no gate acts on a corner qudit, there must be a qudit $l$ underlying a gate $T_{ll'}$ in the product $(T_{kk'}T_{ll'}\cdots)$ such that qudit $l$ is the closest one to a corner qudit among all the qudits that underlie some gates in the product (see Fig.~\ref{app1b}). In this case, qudit $l$ has one neighboring qudit $j$ which is even closer to the corner but does not underlie any gates in the product. Hence, in Step 3 above, the single-qudit operators on qudit $l$ do not cancel, and qudit $l$ belongs to the support of the product.

According to the arguments above, the support of $(S_k^{\bar{\sigma}_k}S_l^{\bar{\sigma}_l}\cdots)$ is nonempty. Then, considering that $\ket{\psi_{m_0}}=\ket{\alpha_1\cdots\alpha_k\cdots\alpha_l\cdots\alpha_N}$, the product must map $\ket{\psi_{m_0}}$ to a different qudit-product-state $\ket{\alpha_1\cdots\alpha'_k\cdots\alpha'_l\cdots\alpha_N}$ with $\alpha'_k=S_k^{\bar{\sigma}_k}\ket{\alpha_k}\ne\ket{\alpha_k}, \alpha'_l=S_l^{\bar{\sigma}_l}\ket{\alpha_l}\ne\ket{\alpha_l},\ldots$. Hence, we have the inequality
\begin{equation*}
\ket{\psi_{m_0}}\ne(S_k^{\bar{\sigma}_k}S_l^{\bar{\sigma}_l}\cdots)\ket{\psi_{m_0}},
\end{equation*}
which contradicts the assumed equality.

\subsection{Proof of Prop.~\ref{pp1}}\label{popp1}
\begin{center}
\begin{figure}[ht]
\centering
    \includegraphics[width=8.5cm]{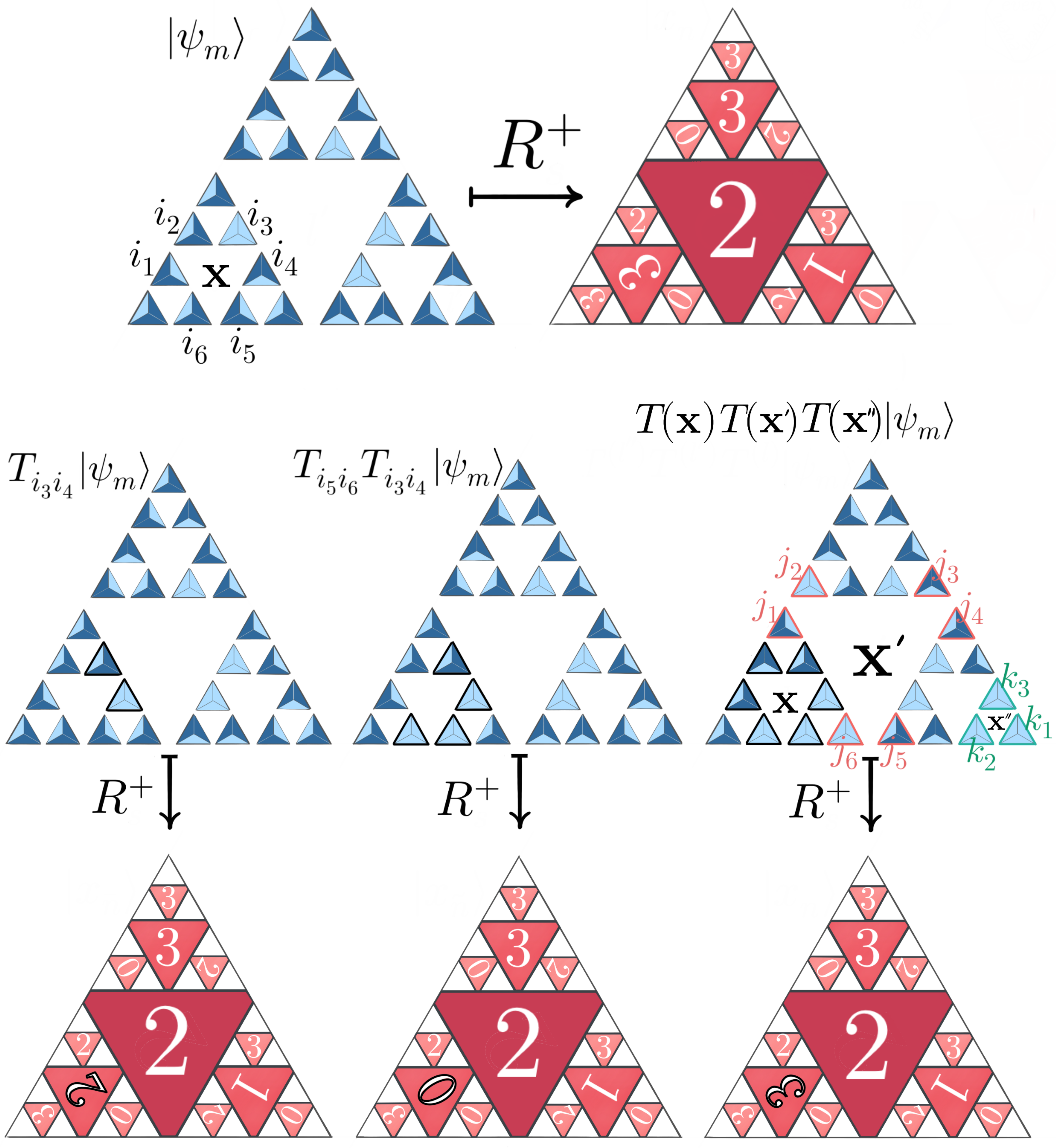}   
\caption{Illustration of how the actions of the gates on $\ket{\psi_m}$ change the associated emergent degrees of freedom.}
\label{figapp2}
\end{figure}
\end{center}

(1) We consider an arbitrary $\ket{\psi_m}$ and fix an arbitrary hole $\boldsymbol{x}$ with the gate $T(\boldsymbol{x})=T_{i_1i_2}T_{i_3i_4}T_{i_5i_6}$. We also consider $\ket{\boldsymbol{\beta}_{\boldsymbol{1}}\cdots\boldsymbol{\beta}_{\boldsymbol{x}}\cdots\boldsymbol{\beta}_{K}}=\sqrt{2^K}R^+\ket{\psi_m}$, or equivalently the emergent degrees of freedom $\boldsymbol{\beta}_{\boldsymbol{1}},\ldots,\boldsymbol{\beta}_{\boldsymbol{x}},\ldots,\boldsymbol{\beta}_{K}$ that can be read off the pictorial representation of $\ket{\psi_m}$ (see Fig.~\ref{figapp2}). Then, the following properties are clear from the pictorial representations as shown in Fig.~\ref{figapp2}.

Applying any gate out of the three $T_{i_1i_2},T_{i_3i_4},T_{i_5i_6}$, or applying the product of any two out of the three will map $\ket{\psi_m}$ to a different $\ket{\psi_{m'}}$, changing the degrees of freedom $\boldsymbol{\beta}_{\boldsymbol{x}}$ as manifest in $\ket{\psi_m}$ to a different $\boldsymbol{\beta}'_{\boldsymbol{x}}$ as manifest in $\ket{\psi_{m'}}$. And by suitably choosing the gate or the product, the changed $\boldsymbol{\beta}'_{\boldsymbol{x}}$ can be arbitrary but different from $\boldsymbol{\beta}_{\boldsymbol{x}}$.

It is important to note that such change has no effect on other emergent degrees of freedom $\boldsymbol{\beta}_{\boldsymbol{x}'}$ associated to other hole $\boldsymbol{x}'$. That is because (i) any of the three gates is either apart from the hole $\boldsymbol{x}'$ and hence does not act on the sub-configuration $\alpha_{i'_1},\alpha_{i'_2},\ldots$ surrounding the $\boldsymbol{x}'$, (ii) or such gate only locally change the the sub-configuration on one lateral sides and hence does not affect $\boldsymbol{\beta}_{\boldsymbol{x}'}$ (see Fig.~\ref{figapp2}).

Now, if we formalize these properties using the $R^+$ operator, we exactly get the desired conditions.

(2) Following the above arguments, if we apply $T(\boldsymbol{x})=T_{i_1i_2}T_{i_3i_4}T_{i_5i_6}$ to $\ket{\psi_m}$, the effects from the three three gates $T_{i_1i_2},T_{i_3i_4},T_{i_5i_6}$ on $\boldsymbol{\beta}_{\boldsymbol{x}}$ as manifest in $\ket{\psi_m}$ cancel each other, although $T(\boldsymbol{x})\ket{\psi_m}=(T_{i_1i_2}T_{i_3i_4}T_{i_5i_6})\ket{\psi_m}\ne\ket{\psi_m}$ (see Prop.~\ref{gates1}). Because this argument also applies to the gates associated to other holes, applying any nontrivial product of the gates $(T(\boldsymbol{x})T(\boldsymbol{x}')\cdots)$ will not change the entanglement pattern, or the emergent degrees o freedom as manifest on $\ket{\psi_m}$.

However, if a nontrivial product $(T_{ii'}T_{jj'}\cdots)$ cannot be written as the form $(T(\boldsymbol{x})T(\boldsymbol{x}')\cdots)$, then for some hole $\boldsymbol{x}'$, only one or two out of the three gates $T_{i'_1i'_2},T_{i'_3i'_4},T_{i'_5i'_6}$ appear in the product. Then, according to our arguments for (1), the product $(T_{ii'}T_{jj'}\cdots)$ will definitely change the emergent degrees of freedom $\boldsymbol{\beta}_{\boldsymbol{x}'}$ as manifest in $\ket{\psi_m}$ and hence change the entanglement pattern.

These properties can also be formalized with $R^+$, which exactly gives the desired condition.

\subsection{Proof of Lemma~\ref{stot1}}\label{postot1}
We only need to prove the forward. The proof for the backward is trivial since the gates commute with $P_0$. We assume that $P_0(\otimes_{i\in A}S_i^{\sigma_i})P_0\ne 0$ and $\otimes_{i\in A}S_i^{\sigma_i}\ne\mathds{1}$ (at least one $\sigma_i\ne0$).

Since $\otimes_{i\in A}S_i^{\sigma_i}$ maps a qudit-product-state to another different one, according to the assumption, $\otimes_{i\in A}S_i^{\sigma_i}$ must map some $\ket{\psi_m}\in\mathcal{H}_0$ to another different $\ket{\psi_{m'}}\in\mathcal{H}_0$. Otherwise, $\otimes_{i\in A}S_i^{\sigma_i}$ maps all $\ket{\psi_m}$ states out of $\mathcal{H}_0$, then the projection is zero.

Then, according to Prop.~\ref{gates1}, we must have $(T_{ii'}T_{jj'}\cdots)\ket{\psi_m}=\ket{\psi_{m'}}$. And we have the equality
\begin{equation*}
\otimes_{i\in A}S_i^{\sigma_i}\ket{\psi_m}=(T_{ii'}T_{jj'}\cdots)\ket{\psi_m}.
\end{equation*}
Indee, the product $(T_{ii'}T_{jj'}\cdots)$ itself is certain tensor product of the $S^\sigma_i$ operators, i.e. $(T_{ii'}T_{jj'}\cdots)=(S_{i'}^{\sigma_{i'}}S_{j'}^{\sigma_{j'}}\cdots)$ with $\sigma_{i'},\sigma_{j'},\ldots\ne0$. Hence, we have
\begin{equation*}
\otimes_{i\in A}S_i^{\sigma_i}\ket{\psi_m}=(S_{i'}^{\sigma_{i'}}S_{j'}^{\sigma_{j'}}\cdots)\ket{\psi_m}.
\end{equation*}

Note that in $(S_{i'}^{\sigma_{i'}}S_{j'}^{\sigma_{j'}}\cdots)$, each $\sigma_{i'}\ne0$ means that each $S_{i'}^{\sigma_{i'}}\ne\mathds{1}$. Hence, all the qudits engaged in $(S_{i'}^{\sigma_{i'}}S_{j'}^{\sigma_{j'}}\cdots)$ form its support. However in $\otimes_{i\in A}S_i^{\sigma_i}$, generally, not every $\sigma_i$ is guaranteed nonzero ($S_i^0=\mathds{1}$ possibly for some qudits in $A$), but only those qudits with nonzero $\sigma_i$ specify the support.

Thus, according to the definition of the $S^\sigma$ operator (see Eq.~\ref{ds}) and recalling that $\ket{\psi_m}=\ket{\alpha_1\cdots\alpha_i\cdots\alpha_N}$, in the above equality, if a $S_i^{\sigma_i}$ operator with $\sigma_i\ne0$ on one side maps one qudit state $\ket{\alpha_i}$ to $\ket{\alpha'_i}\ne\ket{\alpha_i}$, the same operator must appear in the other side. In other words, the above equality holds if and only if the products on the two sides share the same support, and for each qudit within the support, the $S^\sigma$ operator in the two products are the same. Therefore, we have $\otimes_{i\in A}S_i^{\sigma_i}=(S_{i'}^{\sigma_{i'}}S_{j'}^{\sigma_{j'}}\cdots)=(T_{ii'}T_{jj'}\cdots)$.

\subsection{Proof of Prop.~\ref{qa}}\label{poqa}
(1) To check that the support of $Q_A$ is $A$, we note that each $T(\boldsymbol{x})$ is a product of the $S^{\sigma_i}_i$ operators, and so is $T(\boldsymbol{x})T(\boldsymbol{x}')\cdots$. Then, we can re-express each term in Eq.~\ref{qoa1} as
\begin{align}\label{qaapp1}
\begin{split}
&(T(\boldsymbol{x})T(\boldsymbol{x}')\cdots)(\otimes_{i\in A}\dyad{\alpha_i})(T(\boldsymbol{x})T(\boldsymbol{x}')\cdots)\\
=&(S^{\sigma_i}_iS^{\sigma_j}_j\cdots)(\dyad{\alpha_{i_1}\alpha_{i_2}\cdots})(S^{\sigma_i}_iS^{\sigma_j}_j\cdots)\\
=&\dyad{\alpha'_{i_1}\alpha'_{i_2}\cdots}\\
=&\otimes_{i\in A}\dyad{\alpha'_i}.
\end{split}
\end{align}
The derivation takes advantage of the fact that $S^{\sigma_i}_i=(S^{\sigma_i}_i)^+=(S^{\sigma_i}_i)^{-1}$ (see Eq.~\ref{ds} and \ref{ss}). It says that for qudit $j$ outside the support of $\otimes_{i\in A}\dyad{\alpha_i}$, the $S^{\sigma_j}_j$ on the two sides of the expression cancels, i.e. $S^{\sigma_j}_jS^{\sigma_j}_j=\mathds{1}$; and for qudit $i$ within the support, $S^{\sigma_i}_i$ simply change the local state $\ket{\alpha_i}$ into $\ket{\alpha'_i}\ne\ket{\alpha_i}$. And hence, each term of $Q_A$ obviously has the same support as $\otimes_{i\in A}\dyad{\alpha_i}$.

(2) It is obvious from Eq.~\ref{qoa1} that $Q_A$ is self-adjoint (Hermitian). Hence, to show that $Q_A$ commutes with $P_{\mathrm{code}}$, it suffices to show that $\mathcal{H}_{\mathrm{code}}$ is invariant under the action of $Q_A$~\footnote{A basic property in operator algebra on finite-dimension Hilbert space is that the subspace $\mathcal{H}_{\mathrm{code}}$ is invariant under the action of both operator $O$ and its adjoint $O^+$ if and only if $[P_{\mathrm{code}},O]=0$.}. Consider a basis state $\ket*{\tilde{\varphi}_n}$ of $\mathcal{H}_{\mathrm{code}}$. We show that $Q_A\ket*{\tilde{\varphi}_n}$ is propositional to $\ket*{\tilde{\varphi}_n}$ itself.

Recall that $\ket*{\tilde{\varphi}_n}$ is an equal-weight sum of a sub-collection of the qudit-product-states $\ket{\psi_m}$s (see Par.~\ref{cbs0}). And if fixing an arbitrary $\ket{\psi_m}$ in the summation, all the qudit-product-states in the sum can be expressed as $(T(\boldsymbol{x})T(\boldsymbol{x}')\cdots)\ket{\psi_m}$, i.e., in a one-to-one correspondence to the products $(T(\boldsymbol{x})T(\boldsymbol{x}')\cdots)$s in Eq.~\ref{qoa1} (or in Eq.~\ref{exp1} and Eq.~\ref{cbs1}). Based on this fact, we can understand the action of $Q_A$ on $\ket*{\tilde{\varphi}_n}$ as following.

Indeed, we can firstly consider the action of one term in $Q_A$ on $\ket{\psi_m}$, i.e., the action of 
\begin{align*}
\begin{split}
&(T(\boldsymbol{x})T(\boldsymbol{x}')\cdots)(\otimes_{i\in A}\dyad{\alpha_i})(T(\boldsymbol{x})T(\boldsymbol{x}')\cdots)=\\
&(T(\boldsymbol{x})T(\boldsymbol{x}')\cdots)(\dyad{\alpha_{i_1}\alpha_{i_2}\cdots})(T(\boldsymbol{x})T(\boldsymbol{x}')\cdots) 
\end{split}
\end{align*}
on $\ket{\psi_m}=\ket{\alpha_1\cdots\alpha_i\cdots\alpha_N}$. It is easy to see that $(T(\boldsymbol{x})T(\boldsymbol{x}')\cdots)(\otimes_{i\in A}\dyad{\alpha_i})(T(\boldsymbol{x})T(\boldsymbol{x}')\cdots)\ket{\psi_m}\ne0$ if and only if the state $\ket{\psi_{m'}}=(T(\boldsymbol{x})T(\boldsymbol{x}')\cdots)\ket{\psi_m}$, which also participates in expanding $\ket*{\tilde{\varphi}_n}$, matches the sub-configuration $\alpha_{i_1},\alpha_{i_2},\ldots$ on the support $A$. And if so, the action $(T(\boldsymbol{x})T(\boldsymbol{x}')\cdots)(\otimes_{i\in A}\dyad{\alpha_i})(T(\boldsymbol{x})T(\boldsymbol{x}')\cdots)\ket{\psi_m}$ simply turns back to $\ket{\psi_m}$ itself. That is
\begin{align*}
\begin{split}
&(T(\boldsymbol{x})T(\boldsymbol{x}')\cdots)(\otimes_{i\in A}\dyad{\alpha_i})(T(\boldsymbol{x})T(\boldsymbol{x}')\cdots)\ket{\psi_m}\\
=&(T(\boldsymbol{x})T(\boldsymbol{x}')\cdots)(\otimes_{i\in A}\dyad{\alpha_i})\ket{\psi_{m'}}\\
=&(T(\boldsymbol{x})T(\boldsymbol{x}')\cdots)\ket{\psi_{m'}}\\
=&\ket{\psi_m},
\end{split}
\end{align*}
where we utilized the fact that $(T(\boldsymbol{x})T(\boldsymbol{x}')\cdots)(T(\boldsymbol{x})T(\boldsymbol{x}')\cdots)=\mathds{1}$.

Then, we note that the terms in $Q_A$, i.e., $(T(\boldsymbol{x})T(\boldsymbol{x}')\cdots)(\otimes_{i\in A}\dyad{\alpha_i})(T(\boldsymbol{x})T(\boldsymbol{x}')\cdots)$s, and all the qudit-product-states in the expansion of $\ket*{\tilde{\varphi}_n}$, i.e., $\ket{\psi_{m'}}=(T(\boldsymbol{x})T(\boldsymbol{x}')\cdots)\ket{\psi_m}$s, are both indexed by the product $(T(\boldsymbol{x})T(\boldsymbol{x}')\cdots)$ in Eq.~\ref{exp1} and in a one-to-one correspondence to each other. It means that the action of $Q_A$ on $\ket{\psi_m}$ simply checks with every state $\ket{\psi_{m'}}$ in the expansion of $\ket*{\tilde{\varphi}_n}$ on whether the state matches the sub-configuration $\alpha_{i_1},\alpha_{i_2},\ldots$.

Then, if no $\ket{\psi_{m'}}$ in the expansion matches the sub-configuration, we have $Q_A\ket*{\tilde{\varphi}_n}=0$. And hence, we only need to consider the case where some $\ket{\psi_{m'}}$ matches the sub-configuration. And in this case the result of this action is $c_A\ket{\psi_m}$. Here, $c_A$ exactly equals $c_n/{2^K}$ and $c_n$ is the number of all the qudit-product-states in the expansion, which match the sub-configuration. Note that the constant $c_A$ is, by definition, independent on the choice of $\ket{\psi_m}$ in consideration.

Now, for $\ket*{\tilde{\varphi}_n}$ with $Q_A\ket*{\tilde{\varphi}_n}\ne0$, since the $\ket{\psi_m}$ in the above arguments is arbitrary among all the states expanding $\ket*{\tilde{\varphi}_n}$, we have the desired result:
\begin{align*}
\begin{split}
Q_A\ket*{\tilde{\varphi}_n}&=\frac{1}{\sqrt{2^K}}\sum Q_A\ket{\psi_m}\\
&=c_A\frac{1}{\sqrt{2^K}}\sum \ket{\psi_m}\\
&=c_A\ket*{\tilde{\varphi}_n},
\end{split}
\end{align*}
where we keep in mind that the sum only goes through a sub-collection of the $\ket{\psi_m}$ states (see Eq.~\ref{cbs1}).

(3) To show that $c_A$ is independent on $\ket*{\tilde{\varphi}_n}$ with $Q_A\ket*{\tilde{\varphi}_n}\ne0$, we simply need to show that for different $\ket*{\tilde{\varphi}_n}$s, the number $c_n$ of the qudit-product-states that participate in each expansion and match the sub-configuration $\alpha_{i_1},\alpha_{i_2},\ldots$ on $A$ is the same.

Suppose that a $\ket{\psi_m}$ state in the expansion of $\ket*{\tilde{\varphi}_n}$ matches the sub-configuration. Then an arbitrary state $\ket{\psi_{m'}}=(T(\boldsymbol{x})T(\boldsymbol{x}')\cdots)\ket{\psi_m}$ in the expansion matches the sub-configuration if and only if its whole configuration $\alpha'_i,\alpha'_j,\ldots$ differs from that of $\ket{\psi_m}$ in the subregion that avoids the support of $(T(\boldsymbol{x})T(\boldsymbol{x}')\cdots)$. It follows that counting the states that match the sub-configuration is equivalent to counting the products $(T(\boldsymbol{x})T(\boldsymbol{x}')\cdots)$s with support disjoint from $A$. Obviously, the latter is only determined by the geometry of $A$ on the lattice, and is independent on the content of $\otimes_{i\in A}\dyad{\alpha_i}$, i.e., the sub-configuration $\alpha_{i_1},\alpha_{i_2},\ldots$.

(4) To show that $P_{\mathrm{code}}Q_AP_{\mathrm{code}}=P_{\mathrm{code}}\otimes_{i\in A}\dyad{\alpha_i}P_{\mathrm{code}}$, we compare the matrix elements as represented on the basis $\{\ket*{\tilde{\varphi}_n}\}$. According to the previous arguments, the action of $\otimes_{i\in A}\dyad{\alpha_i}$ on $\ket*{\tilde{\varphi}_n}$ simply ``picks up'' only the $\ket{\psi_m}$ states in the expansion, which match the sub-configuration on $A$ and are in total $c_n$ states. Hence, we have $\mel{\tilde{\varphi}_n}{(\otimes_{i\in A}\dyad{\alpha_i})}{\tilde{\varphi}_{n'}}=\delta_{nn'}(c_n/2^K)=c_A=\mel{\tilde{\varphi}_n}{Q_A}{\tilde{\varphi}_{n'}}$. This has proved the desired results.

\subsection{Proof of Cor.~\ref{qaa}}\label{poqaa}
(1) $\Rightarrow$ (2): Since $P_{\mathrm{code}}Q_AP_{\mathrm{code}}\ne0$ and $P_{\mathrm{code}}Q'_AP_{\mathrm{code}}\ne0$, then, $Q_A=Q'_A$ means that for some $\ket*{\tilde{\varphi}_n}$ we have $Q_A\ket*{\tilde{\varphi}_n}=Q'_A\ket*{\tilde{\varphi}_n}=c_A\ket*{\tilde{\varphi}_n}$. According to Prop.~\ref{qa}, some $\ket{\psi_m}$ in the expansion of $\ket*{\tilde{\varphi}_n}$ must match the sub-configuration $\alpha_{i_1},\alpha_{i_2},\ldots$. And similarly, some $\ket{\psi_{m'}}$ in the expansion of $\ket*{\tilde{\varphi}_n}$ must match the sub-configuration $\alpha'_{i_1},\alpha'_{i_2},\ldots$.

(2) $\Rightarrow$ (3): Suppose that $\ket{\psi_m}$ and $\ket{\psi_{m'}}$ that both participate the expansion of the same $\ket*{\tilde{\varphi}_n}$ match the two sub-configurations $\alpha_{i_1},\alpha_{i_2},\ldots$ and $\alpha'_{i_1},\alpha'_{i_2},\ldots$ (supported on boundary subregion $A$) respectively. Then, according to the definition of $\ket*{\tilde{\varphi}_n}$ as given by Eq.~\ref{cbs1}, we have $\ket{\psi_{m'}}=(T(\boldsymbol{x})T(\boldsymbol{x}')\cdots)\ket{\psi_m}$.

To show what this equality means, we can write $(T(\boldsymbol{x})T(\boldsymbol{x}')\cdots)=(S^{\sigma_i}_iS^{\sigma_j}_j\cdots)$. And denoting the support of this product by $A_0\cup A'_0$ with $A_0\subset A$ and $A'_0\subset \overline{A}$, we can further write $(T(\boldsymbol{x})T(\boldsymbol{x}')\cdots)=(\otimes_{i\in A_0}S^{\sigma_i}_i)(\otimes_{j\in A'_0}S^{\sigma_j}_j)$. Then, $\ket{\psi_{m'}}=(T(\boldsymbol{x})T(\boldsymbol{x}')\cdots)\ket{\psi_m}$ simply means that $(\otimes_{i\in A_0}S^{\sigma_i}_i)$ exactly changes the sub-configuration $\alpha_{i_1},\alpha_{i_2},\ldots$ into $\alpha'_{i_1},\alpha'_{i_2},\ldots$, i.e., $\ket{\alpha'_{i_1}\alpha'_{i_2}\cdots}_{A}=(\otimes_{i\in A_0}S^{\sigma_i}_i)\ket{\alpha_{i_1}\alpha_{i_2}\cdots}_{A}$.

Now, we have
\begin{align*}
\begin{split}
&(T(\boldsymbol{x})T(\boldsymbol{x}')\cdots)(\otimes_{i\in A}\dyad{\alpha_i})(T(\boldsymbol{x})T(\boldsymbol{x}')\cdots)\\
=&(\otimes_{i\in A_0}S^{\sigma_i}_i)(\otimes_{i\in A}\dyad{\alpha_i})(\otimes_{i\in A_0}S^{\sigma_i}_i)\\
=&(\otimes_{i\in A_0}S^{\sigma_i}_i)\dyad{\alpha_{i_1}\alpha_{i_2}\cdots}_A(\otimes_{i\in A_0}S^{\sigma_i}_i)\\
=&\dyad{\alpha'_{i_1}\alpha'_{i_2}\cdots}_A\\
=&\otimes_{i\in A}\dyad{\alpha'_i},
\end{split}
\end{align*}
as desired. Note that in the first step of the above derivation we utilize the fact that $(\otimes_{j\in A'_0}S^{\sigma_j}_j)$ lies outside the support of $(\otimes_{i\in A}\dyad{\alpha_i})$ and hence cancels from the two sides, i.e., $(\otimes_{j\in A'_0}S^{\sigma_j}_j)(\otimes_{j\in A'_0}S^{\sigma_j}_j)=\mathds{1}$.

(3) $\Rightarrow$ (1): We assume that $\otimes_{i\in A}\dyad{\alpha'_i}$ equals $(T(\boldsymbol{x})T(\boldsymbol{x}')\cdots)(\otimes_{i\in A}\dyad{\alpha_i})(T(\boldsymbol{x})T(\boldsymbol{x}')\cdots)$. Then , according to Eq.~\ref{qoa1}, $Q'_A$ equals the summation running through all the possible terms of the form 
\begin{multline}\label{appqaa1}
(T(\boldsymbol{x}'')T(\boldsymbol{x}''')\cdots)\big[(T(\boldsymbol{x})T(\boldsymbol{x}')\cdots)(\otimes_{i\in A}\dyad{\alpha_i})\\
\times(T(\boldsymbol{x})T(\boldsymbol{x}')\cdots)\big](T(\boldsymbol{x}'')T(\boldsymbol{x}''')\cdots).
\end{multline}
Note that the product $(T(\boldsymbol{x}'')T(\boldsymbol{x}''')\cdots)(T(\boldsymbol{x})T(\boldsymbol{x}')\cdots)$ is again a product of the assembled gates so that the term in Eq.~\ref{appqaa1} is also a term in the definition of $Q_A$ as given by Eq.~\ref{qoa1}. Indeed, since we have $(T(\boldsymbol{x})T(\boldsymbol{x}')\cdots)(T(\boldsymbol{x})T(\boldsymbol{x}')\cdots)=\mathds{1}$, multiplying by $(T(\boldsymbol{x})T(\boldsymbol{x}')\cdots)$ simply maps the collection of all the possible products of the assembled gates one-to-one onto the collection itself. And hence the terms in the summation for $Q'_A$, i.e., of the form as in Eq.~\ref{appqaa1}, are exactly the terms in the summation for $Q_A$. That is, $Q_A=Q'_A$ as desired.

\subsection{Proof of Prop.~\ref{qaaa}}\label{poqaaa}
In the proof, we check with the condition for reconstructing $\dyad{\overline{\boldsymbol{\beta}}_{\boldsymbol{x}}}$ on $A$. That is \emph{i)} $[Q_A,P_{\mathrm{code}}]=0$; \emph{ii)} $Q_A\ket*{\tilde{\varphi}_n}\ne0$ if and only if $\ket*{\tilde{\varphi}_n}=R\ket{\boldsymbol{\beta}_{\boldsymbol{1}}\cdots\overline{\boldsymbol{\beta}}_{\boldsymbol{x}}\cdots\boldsymbol{\beta}_{K}}$, irrespective of $\boldsymbol{\beta}_{\boldsymbol{x}'}$ ($\boldsymbol{x}'\ne\boldsymbol{x}$); \emph{iii)} in the nonzero case, we have $Q_A\ket*{\tilde{\varphi}_n}=\ket*{\tilde{\varphi}_n}$.

\emph{i)} It is easy to show that $Q_A$ commutes with $P_{\mathrm{code}}$, since every $Q_A(\alpha_{i_1},\alpha_{i_2},\ldots)$ does (see Prop.~\ref{qa}).

\emph{iii)} In terms of Cor.~\ref{qaa}, we can also show that if $Q_A\ket*{\tilde{\varphi}_n}\ne0$, then $Q_A\ket*{\tilde{\varphi}_n}=\ket*{\tilde{\varphi}_n}$. Indeed, if $Q_A\ket*{\tilde{\varphi}_n}\ne0$, then for some $Q_A(\alpha_{i_1},\alpha_{i_2},\ldots)$, we have $Q_A(\alpha_{i_1},\alpha_{i_2},\ldots)\ket*{\tilde{\varphi}_n}=c_A\ket*{\tilde{\varphi}_n}$. To show that $Q_A\ket*{\tilde{\varphi}_n}=\ket*{\tilde{\varphi}_n}$, it suffices to show that if we have both $Q_A(\alpha_{i_1},\alpha_{i_2},\ldots)\ket*{\tilde{\varphi}_n}=c_A\ket*{\tilde{\varphi}_n}$ and $Q_A(\alpha'_{i_1},\alpha'_{i_2},\ldots)\ket*{\tilde{\varphi}_n}=c_A\ket*{\tilde{\varphi}_n}$, then we must have $Q_A(\alpha_{i_1},\alpha_{i_2},\ldots)=Q_A(\alpha'_{i_1},\alpha'_{i_2},\ldots)$, i.e., the operator only appears once in the summation for $Q_A$ defined in Eq.~\ref{qoa2}. Indeed, this is guaranteed by Prop.~\ref{qa} and Cor.~\ref{qaa}. According to Prop.~\ref{qa}, the two equalities means that sub-configurations $\alpha_{i_1},\alpha_{i_2},\ldots$ and $\alpha'_{i_1},\alpha'_{i_2},\ldots$ matches some qudit-product-states $\ket{\psi_m}$ and $\ket{\psi_{m'}}$ in the expansion of $\ket*{\tilde{\varphi}_n}$ respectively. Then, Cor.~\ref{qaa} implies that in this case $Q_A(\alpha_{i_1},\alpha_{i_2},\ldots)$ must be equal to $Q_A(\alpha'_{i_1},\alpha'_{i_2},\ldots)$.

\emph{ii)} It is clear from Eq.~\ref{collection} and Eq.~\ref{qoa2} that for any $\ket*{\tilde{\varphi}_n}=R\ket{\boldsymbol{\beta}_{\boldsymbol{1}}\cdots\overline{\boldsymbol{\beta}}_{\boldsymbol{x}}\cdots\boldsymbol{\beta}_{K}}$, irrespective of $\boldsymbol{\beta}_{\boldsymbol{x}'}$ ($\boldsymbol{x}'\ne\boldsymbol{x}$), some sub-configuration $\alpha_{i_1},\alpha_{i_2},\ldots$ matches the $\ket*{\tilde{\varphi}_n}$, and hence $Q_A(\alpha_{i_1},\alpha_{i_2},\ldots)\ket*{\tilde{\varphi}_n}=c_A\ket*{\tilde{\varphi}_n}$, $Q_A\ket*{\tilde{\varphi}_n}\ne0$. Then, to ensure that $Q_A$ reconstructs $\dyad{\overline{\boldsymbol{\beta}}_{\boldsymbol{x}}}$ on $A$, we are only left to prove that if $Q_A\ket*{\tilde{\varphi}_n}\ne0$, then $\ket*{\tilde{\varphi}_n}=R\ket{\boldsymbol{\beta}'_{\boldsymbol{1}}\cdots\overline{\boldsymbol{\beta}}_{\boldsymbol{x}}\cdots\boldsymbol{\beta}'_{K}}$.

Indeed, $Q_A\ket*{\tilde{\varphi}_n}\ne0$ implies $Q_A(\alpha_{i_1},\alpha_{i_2},\ldots)\ket*{\tilde{\varphi}_n}\ne0$ for some sub-configuration $\alpha_{i_1},\alpha_{i_2},\ldots$ specified in Eq.~\ref{collection}, i.e., the sub-configuration matches the $\ket*{\tilde{\varphi}_n}$ (and the entanglement patterns). Hence, if the condition in the proposition is satisfied, i.e., any entanglement pattern matching a sub-configuration specified by Eq.~\ref{collection} must correspond to some $\ket{\boldsymbol{\beta}_{\boldsymbol{1}}\cdots\overline{\boldsymbol{\beta}}_{\boldsymbol{x}}\cdots\boldsymbol{\beta}_{K}}$ in the encoding, it is clear that we have $\ket*{\tilde{\varphi}_n}=R\ket{\boldsymbol{\beta}'_{\boldsymbol{1}}\cdots\overline{\boldsymbol{\beta}}_{\boldsymbol{x}}\cdots\boldsymbol{\beta}'_{K}}$.

\subsection{Proof of Prop.~\ref{qaaaa}}\label{poqaaaa}
\subsubsection{Condition (1)}
Assume that condition (1) is satisfied. Then, starting from the qudit $i$, any non-repeating connected path linking it to $j'$ or $l$ (within the same block) has the form of $i,i_1,i_2,\ldots,\bar{i},\ldots$, as illustrated in Fig.~\ref{app3a}. Here, $\bar{i}$ is the first qudit belonging to $\overline{A}$ in the path. Going through all such paths, we can denote by $\bar{\mathscr{I}}\subset\overline{A}$ the collection of all the $\bar{i}$ qudits (see Fig.~\ref{app3a}, the black-colored). Similarly, in the another block, we can specify the collections $\bar{\mathscr{K}}$ (see Fig.~\ref{app3a}, the black colored). In the following, we show that there exists a reconstruction of $\widetilde{\boldsymbol{S}}^{\boldsymbol{\sigma_{\boldsymbol{x}}}}_{\boldsymbol{x}}\widetilde{\boldsymbol{S}}^{\boldsymbol{\sigma_{\boldsymbol{x}'}}}_{\boldsymbol{x}'}\cdots$ (including $\widetilde{\boldsymbol{S}}^{\boldsymbol{\sigma_{\boldsymbol{x}}}}_{\boldsymbol{x}}$ for the central hole $\boldsymbol{x}$ and with $\boldsymbol{\sigma}_{\boldsymbol{x}}\ne0$) supported on $\bar{\mathscr{I}}\cup\bar{\mathscr{K}}$, as described in Sec.~\ref{trs} (see Fig.~\ref{15c}).

We only consider the case where $i$ belongs to $A$. The case with $i\in\overline{A}$ can be viewed as a trivial case in our description. We begin with describing a process to specify a tree of paths lying within $A$ and within the same big block as $i$ (see Fig.~\ref{app3a}, the pruple-colored). In the tree, the path will start at qudit $i$, and end at a qudit in $A$ neighboring to a qudit in $\bar{\mathscr{I}}$ or a qudit with all its three neighbors already appeared in the branches. In the notation for each qudit $i^n_{b_n}$ in a path, the superscript $n$ indicates the step in which the qudit is included in the tree, while the subscript $b_n=1,2$ indicates the branch the qudit is derived from some $i^{n-1}_{b_{n-1}}$. Note that for each $i^n_{b_n}$, there are only three neighbors, one is the $i^{n-1}_{b_{n-1}}$ and the other two correspond to $b_n=1,2$. It can be expected that till the $n$th step before the process ends, a path can be of the form $i,i^1_{b_1},i^2_{b_2},\ldots,i^n_{b_n}$. Our specification will be step-wise, and along these steps, we keep building a product of the gates, and also enlarging a collection $\mathscr{I}\in A$ that consists of the qudits already present in these ``growing'' paths.

We initiate the product of gates before the first step. That is $T_{ii^1_1}T_{ii^1_2}$, i.e., on $i$ and its only two neighbors $i^1_1$ and $i^1_2$ in the big block (the other neighbor is $i'$ on another block). In the first step, we derive two branch of paths $i,i^1_1$ and $i,i^1_2$ from $i$ to these two neighbors. If both $i^1_1$ and $i^1_2$ belong to $A$, we can define $\mathscr{I}=\{i,i^1_1,i^1_2\}$, and before going to the next step we update the product of gates by adding all new gates with supports including $i^1_1$ or $i^1_2$. Note that in the updated product of gates, no gate appears repeatedly.

Now we can keep deriving the tree of paths by considering the two neighbors for each of $i^1_1$ and $i^1_2$. There are general rules in deriving the paths.

Once we have reached $i^n_{b_n}$ in the $n$th step, we add $i^n_{b_n}$ to $\mathscr{I}$ and add all new gates with supports including $i^n_{b_n}$ to the product of gates. Then, for whether go to the next step from $i^n_{b_n}$, there are two possibilities on its only two neighbors (excluding the one from which $i^n_{b_n}$ is derived): \emph{i)} If for both of the two neighbors, each one either belongs to $\overline{A}$, or has already appeared in the established part of the tree, we simply end this path with $i^n_{b_n}$. \emph{ii)} If one neighbor still belongs to $A$ and is new to the established part of the tree, then we go to the next step with this new neighbor.

\begin{center}
\begin{figure}[ht]
\centering
    \includegraphics[width=8.5cm]{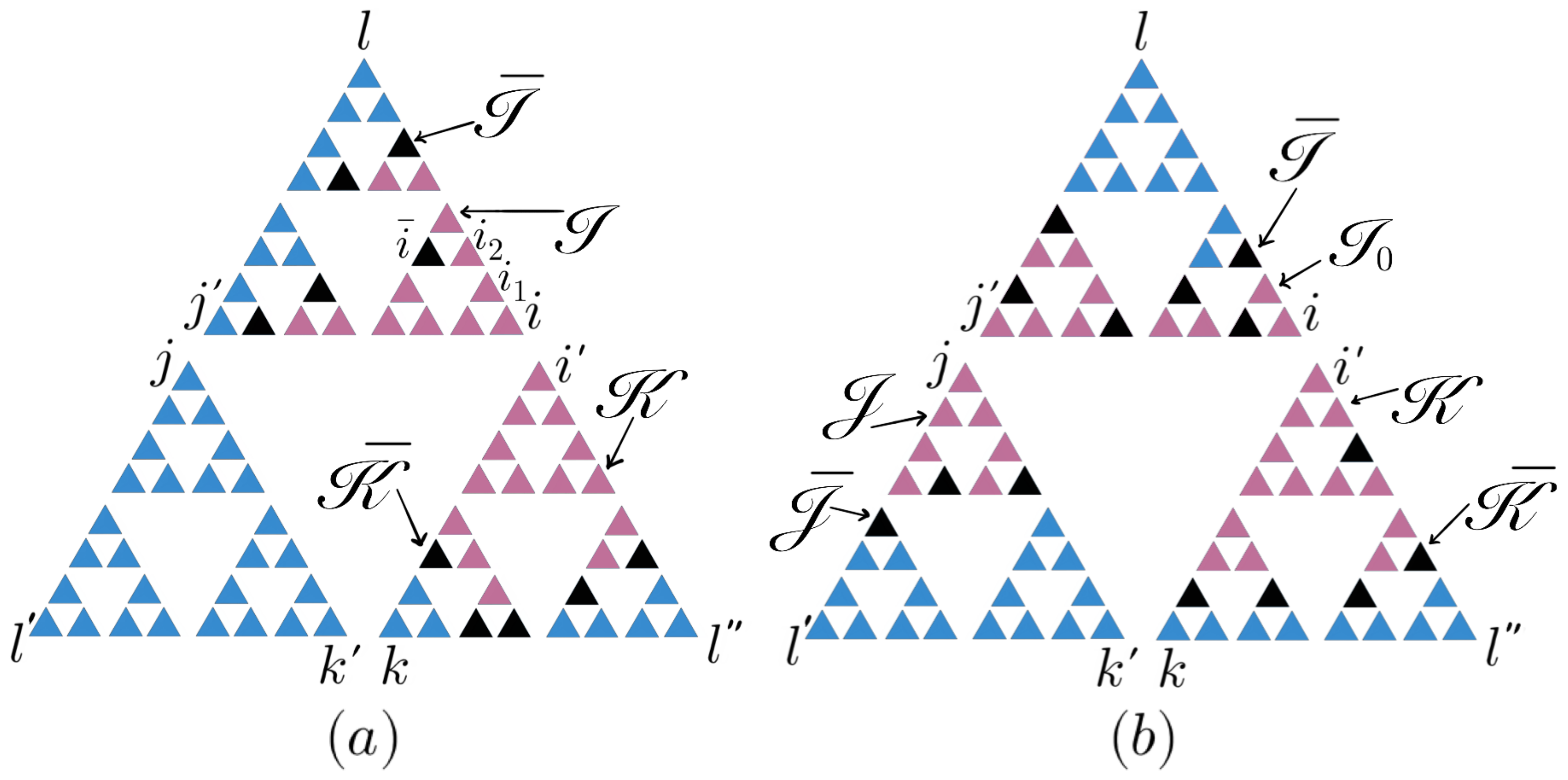}   
\phantomsubfloat{\label{app3a}}\phantomsubfloat{\label{app3b}}
\caption{Purple-colored represent qudits in $\mathscr{I}\subset A$ for (a), in $\mathscr{I}_0\subset A$ for (b), in $\mathscr{K}\subset A$ or in $\mathscr{J}\subset A$. Note that for simplicity not all qudits in $A$ are indicated. Black-colored represent qudits in $\bar{\mathscr{I}}\subset \overline{A}$, in $\bar{\mathscr{K}}\subset \overline{A}$ or in $\bar{\mathscr{J}}\subset \overline{A}$. Similarly, not all qudits in $\overline{A}$ are indicated. (a) and (b) describe the cases for condition (1) and (2) respectively, and are simply generalization of Fig.~\ref{15c} and \ref{15d}, but considering the reconstruction on $\overline{A}$ instead of on $A$.}
\label{figapp3}
\end{figure}
\end{center}

According to the rules, starting from the paths $i,i^1_1$ and $i,i^1_2$ that we have established, after finishing the $n$th step, the updated $\mathscr{I}$ always lies within $A$, and the product of gates $T_{ii^1_1}T_{ii^1_2}\cdots T_{i^n_{b_n}i^{n+1}_{b_{n+1}}}$ always consists of all distinct gates with supports including a qudit in $\mathscr{I}$ (except that $T_{ii'}$ is not included). Furthermore, for any gates in the product, its support either consists of two qudits in $\mathscr{I}$, or one qudit in $\mathscr{I}$ with the other in $A$ for extending the path, or one in $\mathscr{I}$ and the other in $\overline{A}$. Then, obviously, according to what we have mentioned in the proof App.~\ref{pogates1}, for each qudit in $\mathscr{I}$ (except $i$), all the three gates including this qudit participate in the product, and hence the qudit lies outside the support of the product of $T_{ii^1_1}T_{ii^1_2}\cdots T_{i^n_{b_n}i^{n+1}_{b_{n+1}}}$.

According to condition (1), the tree will end in the situation where possibility \emph{i)} applies to every ending qudit. It follows that in the last step, for any gates in the product, either its support consists of two qudit in the updated $\mathscr{I}$, or one qudit in $\mathscr{I}$ and the other in $\overline{A}$. Hence the support of the whole product of the gates will be on $i$ and qudits in $\bar{\mathscr{I}}\subset\overline{A}$. It is noticeable that these qudits in $\bar{\mathscr{I}}\subset\overline{A}$ are generically of very limited number. As shown in Fig.~\ref{app3a}, these qudits do not need to go through the whole intersection of $\overline{A}$ and the big block. And it is easy to show that their total number scales linearly or sublinearly on the linear size $N_0$. Indeed, the number must be upper bounded by half of the lateral which can be viewed as the ends of the tree if we completely derive from $i$ through the whole big block.

Now, we can apply the whole process to the other big block including $i'$ and result in a product of gates with support on $i'$ and qudits in $\bar{\mathscr{K}}\subset\overline{A}$. Then, according to what we have shown in Sec.~\ref{trs}, the combination of the two products together with the gate $T_{ii'}$ will be a reconstruction of $\widetilde{\boldsymbol{S}}^{\boldsymbol{\sigma_{\boldsymbol{x}}}}_{\boldsymbol{x}}\widetilde{\boldsymbol{S}}^{\boldsymbol{\sigma_{\boldsymbol{x}'}}}_{\boldsymbol{x}'}\cdots$ ($\boldsymbol{\sigma}_{\boldsymbol{x}}\ne\boldsymbol{0}$) on $\overline{A}$.

\subsubsection{Condition (2)}
Similar to the above arguments, we can define $\bar{\mathscr{J}}$ and $\bar{\mathscr{K}}$. Accordingly, we can also define a product of gates with support on $j$ together with qudits in $\bar{\mathscr{J}}$, and a product of gates with support on $i'$ together with qudits in $\bar{\mathscr{K}}$ (see Fig.~\ref{app3b}).

In addition, on the block including $l$, we repeat the process to define a tree of path starting from $l$ and define $\mathscr{I}$. Note that here $\mathscr{I}$ has different geometric meaning from that in the above arguments, since the tree starts from $l$, instead of from $i$. As another difference, instead of considering the product of gates defined along the tree, we define $\mathscr{I}_0$ as consisting of all qudits in $A$, in the block, but not in $\mathscr{I}$, and consider the product of all gates with one supporting qudit belonging to $\mathscr{I}_0$ (see Fig.~\ref{app3b}). Obviously, for each gate in the considered product, say $T_{\check{i}\hat{i}}$, since by definition one qudit $\check{i}\in A$ does not belong to $\mathscr{I}$ (not present in the tree), either both $\check{i}$ and $\hat{i}$ belong to $\mathscr{I}_0$, or $\hat{i}$ belongs to $\overline{A}$, or be qudits $i'$ or $j$ in the other blocks. Indeed, if it is not the case, i.e., $\hat{i}$ belongs to $\mathscr{I}$, a path in the tree will extend to $\hat{i}$, which leads to contradiction. Consequently, the support of the considered product of gates is on qudits in $\overline{A}$, possibly together with $i'$ and $j$.

Now, we combine the three products. If $T_{ii'}$ or $T_{jj'}$ are not included, i.e., when $i$ or $j'$ does not belong to $\mathscr{I}_0$ and hence not in $A$, we add them in. It is easy to show that the support of the whole product is within $\overline{A}$. Then according to similar arguments as in Sec.~\ref{trs}, the finial product of gates reconstructs $\widetilde{\boldsymbol{S}}^{\boldsymbol{\sigma_{\boldsymbol{x}}}}_{\boldsymbol{x}}\widetilde{\boldsymbol{S}}^{\boldsymbol{\sigma_{\boldsymbol{x}'}}}_{\boldsymbol{x}'}\cdots$ ($\boldsymbol{\sigma}_{\boldsymbol{x}}\ne\boldsymbol{0}$) on $\overline{A}$.

\subsection{Proofs for Sec.~\ref{exsplitsec}}\label{posplit}

\subsubsection{Detailed arguments for type-1 split (the first example)}
We show that for certain specific decomposition $\mathfrak{e}_{\boldsymbol{x}}=\oplus_{\mu}(\mathfrak{e}_{\boldsymbol{x}a}^{\mu}\otimes\mathfrak{e}_{\boldsymbol{x}\overline{a}}^{\mu})$
the von Neumann algebras generated by the collections of operators
\begin{align*}
\begin{split}
&\{\widetilde{\boldsymbol{S}}^{\boldsymbol{1}}_{\boldsymbol{x}},~ \mathds{1}_{\mathfrak{e}_{\boldsymbol{x}}}\},\\
&\{\widetilde{\boldsymbol{S}}^{\boldsymbol{2}}_{\boldsymbol{x}},~\widetilde{\boldsymbol{S}}^{\boldsymbol{3}}_{\boldsymbol{x}},~ \dyad{\boldsymbol0}+\dyad{\boldsymbol1},~ \dyad{\boldsymbol2}+\dyad{\boldsymbol3},~ \mathds{1}_{\mathfrak{e}_{\boldsymbol{x}}}\}
\end{split}
\end{align*}
are exactly the von Neumann algebra consisting of operators of the form $\sum_{\mu}\widetilde{\boldsymbol{O}}_{\boldsymbol{x}a}^{\mu}\otimes\mathds{1}^{\mu}_{\boldsymbol{x}\overline{a}}$
and the one consisting of operators of the form $\sum_{\mu}\mathds{1}^{\mu}_{\boldsymbol{x}a}\otimes\widetilde{\boldsymbol{O}}_{\boldsymbol{x}\overline{a}}^{\mu}$ respectively.

A convenient way to specify such a decomposition is identifying the basis states $\ket{\boldsymbol{\beta}^{\mu}_{\boldsymbol{x}a}}\otimes\ket{\boldsymbol{\beta}^{\mu}_{\boldsymbol{x}\overline{a}}}$ (with $\ket{\boldsymbol{\beta}^{\mu}_{\boldsymbol{x}a}}\in\mathfrak{e}_{\boldsymbol{x}a}^{\mu}$ and $\ket{\boldsymbol{\beta}^{\mu}_{\boldsymbol{x}\overline{a}}}\in\mathfrak{e}_{\boldsymbol{x}\overline{a}}^{\mu}$) of each tensor product as certain orthonormal states in $\mathfrak{e}_{\boldsymbol{x}}$. Here, by ``identifying as'', we mean mapping each tensor product $\mathfrak{e}_{\boldsymbol{x}a}^{\mu}\otimes\mathfrak{e}_{\boldsymbol{x}\overline{a}}^{\mu}$ onto a subspace of $\mathfrak{e}_{\boldsymbol{x}}$ through an isometry $\iota^{\mu}$ with $\iota^{\mu}(\iota^{\mu})^+=\widetilde{\boldsymbol{P}}^{\mu}_{\boldsymbol{x}}$. In the arguments, it is sufficient to use ``$\mapsto$'' to specify the details of the isometry, and hence we do not explicitly work with the notation ``$\iota^{\mu}$''. Furthermore, we will also use ``$\mapsto$'' for the detail of the mapping of operators $\iota^{\mu}\boldsymbol{\cdot}(\iota^{\mu})^+$. In this sense, we can also ``identify'' $\widetilde{\boldsymbol{O}}_{\boldsymbol{x}a}^{\mu}\otimes\mathds{1}^{\mu}_{\boldsymbol{x}\overline{a}}$ and $\mathds{1}^{\mu}_{\boldsymbol{x}a}\otimes\widetilde{\boldsymbol{O}}_{\boldsymbol{x}\overline{a}}^{\mu}$ as operators on $\mathfrak{e}_{\boldsymbol{x}}$.

We specify the decomposition structure as follows. 
\begin{align}\label{splitex1}
\begin{split}
&\mu=1,2,~\mathfrak{e}_{\boldsymbol{x}a}^{1}=\mathbb{C},~\mathfrak{e}_{\boldsymbol{x}a}^{2}=\mathbb{C},~\mathfrak{e}_{\boldsymbol{x}\overline{a}}^{1}=\mathbb{C}^2,~\mathfrak{e}_{\boldsymbol{x}\overline{a}}^{2}=\mathbb{C}^2\\
&\ket{\boldsymbol{\beta}^1_{\boldsymbol{x}a}}=\ket{1}^1_a, \quad \ket{\boldsymbol{\beta}^2_{\boldsymbol{x}a}}=\ket{1}^2_a,\\
&\ket{\boldsymbol{\beta}^1_{\boldsymbol{x}\overline{a}}}=\ket{\uparrow}^1_{\overline{a}},~\ket{\downarrow}^1_{\overline{a}}, \quad \ket{\boldsymbol{\beta}^2_{\boldsymbol{x}\overline{a}}}=\ket{\uparrow}^2_{\overline{a}},~\ket{\downarrow}^2_{\overline{a}},\\
&\ket{1}^1_a\otimes\ket{\uparrow}^1_{\overline{a}}\mapsto(1/\sqrt{2})(\ket{\boldsymbol{0}}+\ket{\boldsymbol{1}})\in\mathfrak{e}_{\boldsymbol{x}},\\
&\ket{1}^1_a\otimes\ket{\downarrow}^1_{\overline{a}}\mapsto(1/\sqrt{2})(\ket{\boldsymbol{2}}+\ket{\boldsymbol{3}})\in\mathfrak{e}_{\boldsymbol{x}},\\
&\ket{1}^2_a\otimes\ket{\uparrow}^2_{\overline{a}}\mapsto(1/\sqrt{2})(\ket{\boldsymbol{0}}-\ket{\boldsymbol{1}})\in\mathfrak{e}_{\boldsymbol{x}},\\
&\ket{1}^2_a\otimes\ket{\downarrow}^2_{\overline{a}}\mapsto(1/\sqrt{2})(\ket{\boldsymbol{2}}-\ket{\boldsymbol{3}})\in\mathfrak{e}_{\boldsymbol{x}},\\
&\widetilde{\boldsymbol{P}}^{1}_{\boldsymbol{x}}=(1/2)(\dyad{\boldsymbol{0}+\boldsymbol{1}}{\boldsymbol{0}+\boldsymbol{1}}+\dyad{\boldsymbol{2}+\boldsymbol{3}}{\boldsymbol{2}+\boldsymbol{3}}),\\
&\widetilde{\boldsymbol{P}}^{2}_{\boldsymbol{x}}=(1/2)(\dyad{\boldsymbol{0}-\boldsymbol{1}}{\boldsymbol{0}-\boldsymbol{1}}+\dyad{\boldsymbol{2}-\boldsymbol{3}}{\boldsymbol{2}-\boldsymbol{3}}).
\end{split}
\end{align}
Here, we view $\mathbb{C}$ as a one-dimensional Hilbert space consisting of states $c\ket{1}$ ($c$ a complex number), and view $\mathbb{C}^2$ as a two-dimensional Hilbert space spanned by $\ket{\uparrow},\ket{\downarrow}$. Note that $\dyad{\boldsymbol{0}+\boldsymbol{1}}{\boldsymbol{0}+\boldsymbol{1}}$ is abbreviated for $(\ket{\boldsymbol{0}}+\ket{\boldsymbol{1}})(\bra{\boldsymbol{0}}+\bra{\boldsymbol{1}})$.

Now, all operators of the form $\sum_{\mu}\widetilde{\boldsymbol{O}}_{\boldsymbol{x}a}^{\mu}\otimes\mathds{1}^{\mu}_{\boldsymbol{x}\overline{a}}$ are simply linear expansions of
\begin{align}
\begin{split}
&\dyad{1}^1_a\otimes\mathds{1}^{1}_{\overline{a}}\mapsto\widetilde{\boldsymbol{P}}^{1}_{\boldsymbol{x}}=\mathds{1}_{\mathfrak{e}_{\boldsymbol{x}}}+\widetilde{\boldsymbol{S}}^{\boldsymbol{1}}_{\boldsymbol{x}},\\
&\dyad{1}^2_a\otimes\mathds{1}^{2}_{\overline{a}}\mapsto\widetilde{\boldsymbol{P}}^{2}_{\boldsymbol{x}}=\mathds{1}_{\mathfrak{e}_{\boldsymbol{x}}}-\widetilde{\boldsymbol{S}}^{\boldsymbol{1}}_{\boldsymbol{x}}.
\end{split}
\end{align}
Obviously, operators on the two sides of the above two equations expand each other. This has proved that the von Neumann algebras generated by the collection $\{\widetilde{\boldsymbol{S}}^{\boldsymbol{1}}_{\boldsymbol{x}},~ \mathds{1}_{\mathfrak{e}_{\boldsymbol{x}}}\}$ is exactly the von Neumann algebra consisting of operators of the form $\sum_{\mu}\widetilde{\boldsymbol{O}}_{\boldsymbol{x}a}^{\mu}\otimes\mathds{1}^{\mu}_{\boldsymbol{x}\overline{a}}$. Furthermore, since both $\mathfrak{e}_{\boldsymbol{x}a}^{1}$ and $\mathfrak{e}_{\boldsymbol{x}a}^{2}$ are one-dimensional, a $\sum_{\mu}\widetilde{\boldsymbol{O}}_{\boldsymbol{x}a}^{\mu}\otimes\mathds{1}^{\mu}_{\boldsymbol{x}\overline{a}}$ are identified simply as a linear combination of the projection operators. It follows that this von Neumann algebra is also the center of itself (see the discussion in Sec.~\ref{exsplitsec}).

Then, all the operators of the form $\sum_{\mu}\mathds{1}^{\mu}_{\boldsymbol{x}a}\otimes\widetilde{\boldsymbol{O}}_{\boldsymbol{x}\overline{a}}^{\mu}$ are simply linear expansions of
\begin{align}
\begin{split}
\mathds{1}^{1}_{a}\otimes\dyad{\uparrow}^1_{\overline{a}}\mapsto&(1/2)\dyad{\boldsymbol{0}+\boldsymbol{1}}\\
&=(1/2)(\dyad{\boldsymbol{0}}+\dyad{\boldsymbol{1}})(\mathds{1}_{\mathfrak{e}_{\boldsymbol{x}}}+\widetilde{\boldsymbol{S}}^{\boldsymbol{1}}_{\boldsymbol{x}})\\
\mathds{1}^{1}_{a}\otimes\dyad{\downarrow}^1_{\overline{a}}\mapsto&(1/2)\dyad{\boldsymbol{2}+\boldsymbol{3}}\\
&=(1/2)(\dyad{\boldsymbol{2}}+\dyad{\boldsymbol{3}})(\mathds{1}_{\mathfrak{e}_{\boldsymbol{x}}}+\widetilde{\boldsymbol{S}}^{\boldsymbol{1}}_{\boldsymbol{x}})\\
\mathds{1}^{1}_{a}\otimes\dyad{\uparrow}{\downarrow}^1_{\overline{a}}\mapsto&(1/2)\dyad{\boldsymbol{0}+\boldsymbol{1}}{\boldsymbol{2}+\boldsymbol{3}}\\
&=(1/2)(\dyad{\boldsymbol{0}}+\dyad{\boldsymbol{1}})(\widetilde{\boldsymbol{S}}^{\boldsymbol{2}}_{\boldsymbol{x}}+\widetilde{\boldsymbol{S}}^{\boldsymbol{3}}_{\boldsymbol{x}})\\
\mathds{1}^{1}_{a}\otimes\dyad{\downarrow}{\uparrow}^1_{\overline{a}}\mapsto&(1/2)\dyad{\boldsymbol{2}+\boldsymbol{3}}{\boldsymbol{0}+\boldsymbol{1}}\\
&=(1/2)(\dyad{\boldsymbol{2}}+\dyad{\boldsymbol{3}})(\widetilde{\boldsymbol{S}}^{\boldsymbol{2}}_{\boldsymbol{x}}+\widetilde{\boldsymbol{S}}^{\boldsymbol{3}}_{\boldsymbol{x}})\\
\mathds{1}^{2}_{a}\otimes\dyad{\uparrow}^2_{\overline{a}}\mapsto&(1/2)\dyad{\boldsymbol{0}-\boldsymbol{1}}\\
&=(1/2)(\dyad{\boldsymbol{0}}+\dyad{\boldsymbol{1}})(\mathds{1}_{\mathfrak{e}_{\boldsymbol{x}}}-\widetilde{\boldsymbol{S}}^{\boldsymbol{1}}_{\boldsymbol{x}})\\
\mathds{1}^{2}_{a}\otimes\dyad{\downarrow}^2_{\overline{a}}\mapsto&(1/2)\dyad{\boldsymbol{2}-\boldsymbol{3}}\\
&=(1/2)(\dyad{\boldsymbol{2}}+\dyad{\boldsymbol{3}})(\mathds{1}_{\mathfrak{e}_{\boldsymbol{x}}}-\widetilde{\boldsymbol{S}}^{\boldsymbol{1}}_{\boldsymbol{x}})\\
\mathds{1}^{2}_{a}\otimes\dyad{\uparrow}{\downarrow}^2_{\overline{a}}\mapsto&(1/2)\dyad{\boldsymbol{0}-\boldsymbol{1}}{\boldsymbol{2}-\boldsymbol{3}}\\
&=(1/2)(\dyad{\boldsymbol{0}}+\dyad{\boldsymbol{1}})(\widetilde{\boldsymbol{S}}^{\boldsymbol{2}}_{\boldsymbol{x}}-\widetilde{\boldsymbol{S}}^{\boldsymbol{3}}_{\boldsymbol{x}})\\
\mathds{1}^{2}_{a}\otimes\dyad{\downarrow}{\uparrow}^2_{\overline{a}}\mapsto&(1/2)\dyad{\boldsymbol{2}-\boldsymbol{3}}{\boldsymbol{0}-\boldsymbol{1}}\\
&=(1/2)(\dyad{\boldsymbol{2}}+\dyad{\boldsymbol{3}})(\widetilde{\boldsymbol{S}}^{\boldsymbol{2}}_{\boldsymbol{x}}-\widetilde{\boldsymbol{S}}^{\boldsymbol{3}}_{\boldsymbol{x}})
\end{split}
\end{align}
Note that we have $\widetilde{\boldsymbol{S}}^{\boldsymbol{1}}_{\boldsymbol{x}}=\widetilde{\boldsymbol{S}}^{\boldsymbol{2}}_{\boldsymbol{x}}\widetilde{\boldsymbol{S}}^{\boldsymbol{3}}_{\boldsymbol{x}}$ and $\dyad{\boldsymbol0}+\dyad{\boldsymbol1}+\dyad{\boldsymbol2}+\dyad{\boldsymbol3}=\mathds{1}_{\mathfrak{e}_{\boldsymbol{x}}}$. Then, in the above equations, it is clear that operators on the left side can be generated from $\{\widetilde{\boldsymbol{S}}^{\boldsymbol{2}}_{\boldsymbol{x}},\widetilde{\boldsymbol{S}}^{\boldsymbol{3}}_{\boldsymbol{x}}, \dyad{\boldsymbol0}+\dyad{\boldsymbol1},\dyad{\boldsymbol2}+\dyad{\boldsymbol3},\mathds{1}_{\mathfrak{e}_{\boldsymbol{x}}}\}$. On the other hand, through combination of the above equations, it will be easy to see that operators on the right can be linearly expanded by those on the left. This observation has proved that the von Neumann algebra generated by the collection $\{\widetilde{\boldsymbol{S}}^{\boldsymbol{2}}_{\boldsymbol{x}},\widetilde{\boldsymbol{S}}^{\boldsymbol{3}}_{\boldsymbol{x}}, \dyad{\boldsymbol0}+\dyad{\boldsymbol1},\dyad{\boldsymbol2}+\dyad{\boldsymbol3},\mathds{1}_{\mathfrak{e}_{\boldsymbol{x}}}\}$ is exactly the one consisting of operators of the form $\sum_{\mu}\mathds{1}^{\mu}_{\boldsymbol{x}a}\otimes\widetilde{\boldsymbol{O}}_{\boldsymbol{x}\overline{a}}^{\mu}$.

\subsubsection{Detailed arguments for type-2 split (the first example)}
For the first example of type-2 split, we consider the collections
\begin{align*}
\begin{split}
&\{\dyad{\boldsymbol{0}}+\dyad{\boldsymbol{1}},~\dyad{\boldsymbol{2}}+\dyad{\boldsymbol{3}},~\mathds{1}_{\mathfrak{e}_{\boldsymbol{x}}}\},\\
&\{\widetilde{\boldsymbol{S}}^{\boldsymbol{1}}_{\boldsymbol{x}},~\dyad{\boldsymbol{0}},~\dyad{\boldsymbol{1}},~\dyad{\boldsymbol{2}},~\dyad{\boldsymbol{3}},~\mathds{1}_{\mathfrak{e}_{\boldsymbol{x}}}\},
\end{split}
\end{align*}
which can be reconstructed on a boundary subregion $A$ and the its complement $\overline{A}$ respectively. To show that they generate two von Neumann algebras that are commutant to each other, we specify the decomposition $\mathfrak{e}_{\boldsymbol{x}}=\oplus_{\mu}(\mathfrak{e}_{\boldsymbol{x}a}^{\mu}\otimes\mathfrak{e}_{\boldsymbol{x}\overline{a}}^{\mu})$ as follows,
\begin{align}
\begin{split}
&\mu=1,2,~\mathfrak{e}_{\boldsymbol{x}a}^{1}=\mathbb{C},~\mathfrak{e}_{\boldsymbol{x}a}^{2}=\mathbb{C},~\mathfrak{e}_{\boldsymbol{x}\overline{a}}^{1}=\mathbb{C}^2,~\mathfrak{e}_{\boldsymbol{x}\overline{a}}^{2}=\mathbb{C}^2\\
&\ket{\boldsymbol{\beta}^1_{\boldsymbol{x}a}}=\ket{1}^1_a, \quad \ket{\boldsymbol{\beta}^2_{\boldsymbol{x}a}}=\ket{1}^2_a,\\
&\ket{\boldsymbol{\beta}^1_{\boldsymbol{x}\overline{a}}}=\ket{\uparrow}^1_{\overline{a}},~\ket{\downarrow}^1_{\overline{a}}, \quad \ket{\boldsymbol{\beta}^2_{\boldsymbol{x}\overline{a}}}=\ket{\uparrow}^2_{\overline{a}},~\ket{\downarrow}^2_{\overline{a}},\\
&\ket{1}^1_a\otimes\ket{\uparrow}^1_{\overline{a}}\mapsto\ket{\boldsymbol{0}}\in\mathfrak{e}_{\boldsymbol{x}},\\
&\ket{1}^1_a\otimes\ket{\downarrow}^1_{\overline{a}}\mapsto\ket{\boldsymbol{1}}\in\mathfrak{e}_{\boldsymbol{x}},\\
&\ket{1}^2_a\otimes\ket{\uparrow}^2_{\overline{a}}\mapsto\ket{\boldsymbol{2}}\in\mathfrak{e}_{\boldsymbol{x}},\\
&\ket{1}^2_a\otimes\ket{\downarrow}^2_{\overline{a}}\mapsto\ket{\boldsymbol{3}}\in\mathfrak{e}_{\boldsymbol{x}},\\
&\widetilde{\boldsymbol{P}}^{1}_{\boldsymbol{x}}=\dyad{\boldsymbol{0}}+\dyad{\boldsymbol{1}},~\widetilde{\boldsymbol{P}}^{2}_{\boldsymbol{x}}=\dyad{\boldsymbol{2}}+\dyad{\boldsymbol{3}}.
\end{split}
\end{align}

Now, all operators of the form $\sum_{\mu}\widetilde{\boldsymbol{O}}_{\boldsymbol{x}a}^{\mu}\otimes\mathds{1}^{\mu}_{\boldsymbol{x}\overline{a}}$ are simply linear expansions of
\begin{align}
\begin{split}
&\dyad{1}^1_a\otimes\mathds{1}^{1}_{\overline{a}}\mapsto\widetilde{\boldsymbol{P}}^{1}_{\boldsymbol{x}}=\dyad{\boldsymbol{0}}+\dyad{\boldsymbol{1}},\\
&\dyad{1}^2_a\otimes\mathds{1}^{2}_{\overline{a}}\mapsto\widetilde{\boldsymbol{P}}^{2}_{\boldsymbol{x}}=\dyad{\boldsymbol{2}}+\dyad{\boldsymbol{3}}.
\end{split}
\end{align}
Obviously, operators on the two sides of the above two equations expand each other. This has proved that the von Neumann algebras generated by the collection $\{\dyad{\boldsymbol{0}}+\dyad{\boldsymbol{1}},~\dyad{\boldsymbol{2}}+\dyad{\boldsymbol{3}},~\mathds{1}_{\mathfrak{e}_{\boldsymbol{x}}}\}$ is exactly the von Neumann algebra consisting of operators of the form $\sum_{\mu}\widetilde{\boldsymbol{O}}_{\boldsymbol{x}a}^{\mu}\otimes\mathds{1}^{\mu}_{\boldsymbol{x}\overline{a}}$. Similar to the case of type-1 split, this von Neumann algebra is also the center of itself.

Then, all the operators of the form $\sum_{\mu}\mathds{1}^{\mu}_{\boldsymbol{x}a}\otimes\widetilde{\boldsymbol{O}}_{\boldsymbol{x}\overline{a}}^{\mu}$ are simply linear expansions of
\begin{align}
\begin{split}
&\mathds{1}^{1}_{a}\otimes\dyad{\uparrow}^1_{\overline{a}}\mapsto\dyad{\boldsymbol{0}}\\
&\mathds{1}^{1}_{a}\otimes\dyad{\downarrow}^1_{\overline{a}}\mapsto\dyad{\boldsymbol{1}}\\
&\mathds{1}^{1}_{a}\otimes\dyad{\uparrow}{\downarrow}^1_{\overline{a}}\mapsto\dyad{\boldsymbol{0}}{\boldsymbol{1}}=\dyad{\boldsymbol{0}}\widetilde{\boldsymbol{S}}^{\boldsymbol{1}}_{\boldsymbol{x}}\\
&\mathds{1}^{1}_{a}\otimes\dyad{\downarrow}{\uparrow}^1_{\overline{a}}\mapsto\dyad{\boldsymbol{1}}{\boldsymbol{0}}=\dyad{\boldsymbol{1}}\widetilde{\boldsymbol{S}}^{\boldsymbol{1}}_{\boldsymbol{x}}\\
&\mathds{1}^{2}_{a}\otimes\dyad{\uparrow}^2_{\overline{a}}\mapsto\dyad{\boldsymbol{2}}\\
&\mathds{1}^{2}_{a}\otimes\dyad{\downarrow}^2_{\overline{a}}\mapsto\dyad{\boldsymbol{3}}\\
&\mathds{1}^{2}_{a}\otimes\dyad{\uparrow}{\downarrow}^2_{\overline{a}}\mapsto\dyad{\boldsymbol{2}}{\boldsymbol{3}}=\dyad{\boldsymbol{2}}\widetilde{\boldsymbol{S}}^{\boldsymbol{1}}_{\boldsymbol{x}}\\
&\mathds{1}^{2}_{a}\otimes\dyad{\downarrow}{\uparrow}^2_{\overline{a}}\mapsto\dyad{\boldsymbol{3}}{\boldsymbol{2}}=\dyad{\boldsymbol{3}}\widetilde{\boldsymbol{S}}^{\boldsymbol{1}}_{\boldsymbol{x}}.
\end{split}
\end{align}
Obviously, operators on the two sides of the above two equations expand each other. Hence, the von Neumann algebra generated by $\{\widetilde{\boldsymbol{S}}^{\boldsymbol{1}}_{\boldsymbol{x}},~\dyad{\boldsymbol{0}},~\dyad{\boldsymbol{1}},~\dyad{\boldsymbol{2}},~\dyad{\boldsymbol{3}},~\mathds{1}_{\mathfrak{e}_{\boldsymbol{x}}}\}$ is exactly the one consisting of operators of the form $\sum_{\mu}\mathds{1}^{\mu}_{\boldsymbol{x}a}\otimes\widetilde{\boldsymbol{O}}_{\boldsymbol{x}\overline{a}}^{\mu}$.

\subsubsection{Detailed arguments for type-3 split (the first example)}
For the first example of type-3 split, we consider the collections
\begin{align*}
\begin{split}
&\{\widetilde{\boldsymbol{S}}^{\boldsymbol{3}}_{\boldsymbol{x}},~\dyad{\boldsymbol{0}}+\dyad{\boldsymbol{1}},~\dyad{\boldsymbol{2}}+\dyad{\boldsymbol{3}},~\mathds{1}_{\mathfrak{e}_{\boldsymbol{x}}}\},\\
&\{\widetilde{\boldsymbol{S}}^{\boldsymbol{1}}_{\boldsymbol{x}},~\dyad{\boldsymbol{0}}+\dyad{\boldsymbol{3}},~\dyad{\boldsymbol{1}}+\dyad{\boldsymbol{2}},~\mathds{1}_{\mathfrak{e}_{\boldsymbol{x}}}\},
\end{split}
\end{align*}
which can be reconstructed on a boundary subregion $A$ and the its complement $\overline{A}$ respectively. In the following, we show that they generate two von Neumann algebras that are commutant to each other, and the generated von Neumann algebras are two factors, i.e. the center is trivial. To that end, we specify the decomposition $\mathfrak{e}_{\boldsymbol{x}}=\oplus_{\mu}(\mathfrak{e}_{\boldsymbol{x}a}^{\mu}\otimes\mathfrak{e}_{\boldsymbol{x}\overline{a}}^{\mu})$ with only one sector, i.e. $\mathfrak{e}_{\boldsymbol{x}}=\mathfrak{e}_{\boldsymbol{x}a}^{1}\otimes\mathfrak{e}_{\boldsymbol{x}\overline{a}}^{1}$,
\begin{align}
\begin{split}
&\mathfrak{e}_{\boldsymbol{x}a}^{1}=\mathbb{C}^2,~\mathfrak{e}_{\boldsymbol{x}\overline{a}}^{1}=\mathbb{C}^2,\\
&\ket{\boldsymbol{\beta}^1_{\boldsymbol{x}a}}=\ket{\uparrow}^1_{a},~\ket{\downarrow}^1_{a}, \quad \ket{\boldsymbol{\beta}^1_{\boldsymbol{x}\overline{a}}}=\ket{\uparrow}^1_{\overline{a}},~\ket{\downarrow}^1_{\overline{a}},\\
&\ket{\uparrow}^1_a\otimes\ket{\uparrow}^1_{\overline{a}}\mapsto\ket{\boldsymbol{0}}\in\mathfrak{e}_{\boldsymbol{x}},\\
&\ket{\uparrow}^1_a\otimes\ket{\downarrow}^1_{\overline{a}}\mapsto\ket{\boldsymbol{1}}\in\mathfrak{e}_{\boldsymbol{x}},\\
&\ket{\downarrow}^1_a\otimes\ket{\uparrow}^1_{\overline{a}}\mapsto\ket{\boldsymbol{3}}\in\mathfrak{e}_{\boldsymbol{x}},\\
&\ket{\downarrow}^1_a\otimes\ket{\downarrow}^1_{\overline{a}}\mapsto\ket{\boldsymbol{2}}\in\mathfrak{e}_{\boldsymbol{x}}.
\end{split}
\end{align}

We want to show that the two generated von Neumann algebras consist of operators of the form $\widetilde{\boldsymbol{O}}_{\boldsymbol{x}a}^{1}\otimes\mathds{1}^{1}_{\boldsymbol{x}\overline{a}}$ and of the form $\mathds{1}^{1}_{\boldsymbol{x}a}\otimes\widetilde{\boldsymbol{O}}_{\boldsymbol{x}\overline{a}}^{1}$ respectively, and hence the two von Neumann algebras are simply $\mathbf{L}(\mathfrak{e}_{\boldsymbol{x}a}^{1})\otimes\mathbb{C}\mathds{1}_{\mathfrak{e}_{\boldsymbol{x}\overline{a}}^{1}}$ and $\mathbb{C}\mathds{1}_{\mathfrak{e}_{\boldsymbol{x}{a}}^{1}}\otimes\mathbf{L}(\mathfrak{e}_{\boldsymbol{x}\overline{a}}^{1})$, which are surely commutant to each other and with trivial center ($\mathbb{C}\mathds{1}_{\mathfrak{e}_{\boldsymbol{x}}}$).

According to the above specification of the decomposition, all operators of the form $\widetilde{\boldsymbol{O}}_{\boldsymbol{x}a}^{1}\otimes\mathds{1}^{1}_{\boldsymbol{x}\overline{a}}$ are simply linear expansions of
\begin{align}
\begin{split}
\dyad{\uparrow}^1_a\otimes\mathds{1}^{1}_{\overline{a}}\mapsto&\dyad{\boldsymbol{0}}+\dyad{\boldsymbol{1}}\\
\dyad{\downarrow}^1_a\otimes\mathds{1}^{1}_{\overline{a}}\mapsto&\dyad{\boldsymbol{2}}+\dyad{\boldsymbol{3}}\\
\dyad{\uparrow}{\downarrow}^1_a\otimes\mathds{1}^{1}_{\overline{a}}\mapsto&\dyad{\boldsymbol{0}}{\boldsymbol{3}}+\dyad{\boldsymbol{1}}{\boldsymbol{2}}\\
&=(\dyad{\boldsymbol{0}}+\dyad{\boldsymbol{1}})\widetilde{\boldsymbol{S}}^{\boldsymbol{3}}_{\boldsymbol{x}}\\
\dyad{\downarrow}{\uparrow}^1_a\otimes\mathds{1}^{1}_{\overline{a}}\mapsto&\dyad{\boldsymbol{3}}{\boldsymbol{0}}+\dyad{\boldsymbol{2}}{\boldsymbol{1}}\\
&=(\dyad{\boldsymbol{2}}+\dyad{\boldsymbol{3}})\widetilde{\boldsymbol{S}}^{\boldsymbol{3}}_{\boldsymbol{x}}.
\end{split}
\end{align}
It is obvious that these operators and the collection $\{\widetilde{\boldsymbol{S}}^{\boldsymbol{3}}_{\boldsymbol{x}},~\dyad{\boldsymbol{0}}+\dyad{\boldsymbol{1}},~\dyad{\boldsymbol{2}}+\dyad{\boldsymbol{3}},~\mathds{1}_{\mathfrak{e}_{\boldsymbol{x}}}\}$ expand each other.

Then, all operators of the form $\mathds{1}^{1}_{\boldsymbol{x}a}\otimes\widetilde{\boldsymbol{O}}_{\boldsymbol{x}\overline{a}}^{1}$ are simply linear expansions of
\begin{align}
\begin{split}
\mathds{1}^{1}_{a}\otimes\dyad{\uparrow}^1_{\overline{a}}\mapsto&\dyad{\boldsymbol{0}}+\dyad{\boldsymbol{3}}\\
\mathds{1}^{1}_{a}\otimes\dyad{\downarrow}^1_{\overline{a}}\mapsto&\dyad{\boldsymbol{1}}+\dyad{\boldsymbol{2}}\\
\mathds{1}^{1}_{a}\otimes\dyad{\uparrow}{\downarrow}^1_{\overline{a}}\mapsto&\dyad{\boldsymbol{0}}{\boldsymbol{1}}+\dyad{\boldsymbol{3}}{\boldsymbol{2}}\\
&=(\dyad{\boldsymbol{0}}+\dyad{\boldsymbol{3}})\widetilde{\boldsymbol{S}}^{\boldsymbol{1}}_{\boldsymbol{x}}\\
\mathds{1}^{1}_{a}\otimes\dyad{\downarrow}{\uparrow}^1_{\overline{a}}\mapsto&\dyad{\boldsymbol{1}}{\boldsymbol{0}}+\dyad{\boldsymbol{2}}{\boldsymbol{3}}\\
&=(\dyad{\boldsymbol{1}}+\dyad{\boldsymbol{2}})\widetilde{\boldsymbol{S}}^{\boldsymbol{1}}_{\boldsymbol{x}}.
\end{split}
\end{align}
Again, it is clear that these operators and the collection $\{\widetilde{\boldsymbol{S}}^{\boldsymbol{1}}_{\boldsymbol{x}},~\dyad{\boldsymbol{0}}+\dyad{\boldsymbol{3}},~\dyad{\boldsymbol{1}}+\dyad{\boldsymbol{2}},~\mathds{1}_{\mathfrak{e}_{\boldsymbol{x}}}\}$ expand each other. And hence the desired result has been proved.

\subsection{Proof of Prop.~\ref{sova}}\label{posova}
We firstly show that $\otimes_{\boldsymbol{x}}\boldsymbol{\mathcal{M}}_a(\boldsymbol{x})$ and $\otimes_{\boldsymbol{x}}\boldsymbol{\mathcal{M}}_{\overline{a}}(\boldsymbol{x})$ are included in $R^+\mathcal{M}_A R$ and $R^+\mathcal{M}_{\overline{A}} R$ respectively, and apply their alternative expression in Eq.~\ref{structure2} later.

We note that $\otimes_{\boldsymbol{x}}\boldsymbol{\mathcal{M}}_a(\boldsymbol{x})$ consists of all linear combinations of operators of the form $\cdots\otimes\widetilde{\boldsymbol{O}}_{\boldsymbol{x}}\otimes\widetilde{\boldsymbol{O}}_{\boldsymbol{x}'}\otimes\widetilde{\boldsymbol{O}}_{\boldsymbol{x}''}\otimes\cdots$ with $\widetilde{\boldsymbol{O}}_{\boldsymbol{x}}\in\boldsymbol{\mathcal{M}}_a(\boldsymbol{x}),\widetilde{\boldsymbol{O}}_{\boldsymbol{x}'}\in\boldsymbol{\mathcal{M}}_a(\boldsymbol{x}'),\widetilde{\boldsymbol{O}}_{\boldsymbol{x}''}\in\boldsymbol{\mathcal{M}}_a(\boldsymbol{x}''),\dots$. And such a tensor product of operators is simply a product of operators as
\begin{align*}
\begin{split}
&\cdots\otimes\widetilde{\boldsymbol{O}}_{\boldsymbol{x}}\otimes\widetilde{\boldsymbol{O}}_{\boldsymbol{x}'}\otimes\widetilde{\boldsymbol{O}}_{\boldsymbol{x}''}\otimes\cdots\\
&=\cdots(\cdots\otimes\widetilde{\boldsymbol{O}}_{\boldsymbol{x}}\otimes\mathds{1}_{\mathfrak{e}_{\boldsymbol{x}'}}\otimes\mathds{1}_{\mathfrak{e}_{\boldsymbol{x}''}}\otimes\cdots)(\cdots\otimes\mathds{1}_{\mathfrak{e}_{\boldsymbol{x}}}\otimes\widetilde{\boldsymbol{O}}_{\boldsymbol{x}'}\\
&\otimes\mathds{1}_{\mathfrak{e}_{\boldsymbol{x}''}}\otimes\cdots)(\cdots\otimes\mathds{1}_{\mathfrak{e}_{\boldsymbol{x}}}\otimes\mathds{1}_{\mathfrak{e}_{\boldsymbol{x}'}}\otimes\widetilde{\boldsymbol{O}}_{\boldsymbol{x}''}\otimes\cdots)\cdots
\end{split}
\end{align*}
in which, each bracket belongs to one $\cdots\otimes\mathbb{C}\mathds{1}_{\mathfrak{e}_{\boldsymbol{x}'}}\otimes\boldsymbol{\mathcal{M}}_a(\boldsymbol{x})\otimes\mathbb{C}\mathds{1}_{\mathfrak{e}_{\boldsymbol{x}''}}\otimes\cdots$ and hence belongs to the von Neumann algebra $R^+\mathcal{M}_A R$ (see Eq.~\ref{structure0}). It follows that the tensor product $\cdots\otimes\widetilde{\boldsymbol{O}}_{\boldsymbol{x}}\otimes\widetilde{\boldsymbol{O}}_{\boldsymbol{x}'}\otimes\widetilde{\boldsymbol{O}}_{\boldsymbol{x}''}\otimes\cdots$, as product of elements in $R^+\mathcal{M}_A R$ must also belong to $R^+\mathcal{M}_A R$, and so are the linear combinations. Hence, we conclude that $\otimes_{\boldsymbol{x}}\boldsymbol{\mathcal{M}}_a(\boldsymbol{x})\subset R^+\mathcal{M}_A R$, and with similar arguments, $\otimes_{\boldsymbol{x}}\boldsymbol{\mathcal{M}}_{\overline{a}}(\boldsymbol{x})\subset R^+\mathcal{M}_{\overline{A}} R$.

Now, we apply the basic property regarding the commutation of tensor products of von Neumann algebras~\cite{stratila2019}, we have 
\begin{equation*}
(\otimes_{\boldsymbol{x}}\boldsymbol{\mathcal{M}}_a(\boldsymbol{x}))'=\otimes_{\boldsymbol{x}}\boldsymbol{\mathcal{M}}'_a(\boldsymbol{x})=\otimes_{\boldsymbol{x}}\boldsymbol{\mathcal{M}}_{\overline{a}}(\boldsymbol{x}), 
\end{equation*}
(see Eq.~\ref{structure1}: $\boldsymbol{\mathcal{M}}_{\overline{a}}(\boldsymbol{x})=\boldsymbol{\mathcal{M}}'_a(\boldsymbol{x})$). Meanwhile, because $R^+\mathcal{M}_{\overline{A}} R$ is the commutant of $R^+\mathcal{M}_A R$, we have
\begin{equation*}
\otimes_{\boldsymbol{x}}\boldsymbol{\mathcal{M}}_a(\boldsymbol{x})\subset R^+\mathcal{M}_A R,\quad\quad (\otimes_{\boldsymbol{x}}\boldsymbol{\mathcal{M}}_a(\boldsymbol{x}))'\subset (R^+\mathcal{M}_A R)'.
\end{equation*}

Then, with similar arguments as in App.~\ref{possc}, we have 
\begin{equation*}
\otimes_{\boldsymbol{x}}\boldsymbol{\mathcal{M}}_a(\boldsymbol{x})= R^+\mathcal{M}_A R,\quad\quad \otimes_{\boldsymbol{x}}\boldsymbol{\mathcal{M}}_{\overline{a}}(\boldsymbol{x})= R^+\mathcal{M}_{\overline{A}} R.
\end{equation*}

Finally, if we apply the alternative expression of $\otimes_{\boldsymbol{x}}\boldsymbol{\mathcal{M}}_a(\boldsymbol{x})$ and $\otimes_{\boldsymbol{x}}\boldsymbol{\mathcal{M}}_{\overline{a}}(\boldsymbol{x})$ given in Eq.~\ref{structure2}, we eventually reach
\begin{align*}
\begin{split}
R^+\mathcal{M}_A R=(\otimes_{\boldsymbol{x}\in\mathrm{W}[A]}\mathbf{L}(\mathfrak{e}_{\boldsymbol{x}}))\otimes(\otimes_{\boldsymbol{x}\in\mathrm{E}[A\overline{A}]}\boldsymbol{\mathcal{M}}_a(\boldsymbol{x})),\\
R^+\mathcal{M}_{\overline{A}} R=(\otimes_{\boldsymbol{x}\in\mathrm{E}[A\overline{A}]}\boldsymbol{\mathcal{M}}_{\overline{a}}(\boldsymbol{x}))\otimes(\otimes_{\boldsymbol{x}\in\mathrm{W}[\overline{A}]}\mathbf{L}(\mathfrak{e}_{\boldsymbol{x}})).
\end{split}
\end{align*}

As for the center of von Neumann algebra, we also the basic property regarding the commutation of tensor products of von Neumann algebras~\cite{stratila2019}. That is 
\begin{equation*}
\mathrm{Z}(\otimes_{\boldsymbol{x}}\boldsymbol{\mathcal{M}}_a(\boldsymbol{x}))=\otimes_{\boldsymbol{x}}\mathrm{Z}(\boldsymbol{\mathcal{M}}_a(\boldsymbol{x})),
\end{equation*}
where $\mathrm{Z}(\otimes_{\boldsymbol{x}}\boldsymbol{\mathcal{M}}_a(\boldsymbol{x}))$ is within $\mathbf{L}(\mathcal{E})$ while $\mathrm{Z}(\boldsymbol{\mathcal{M}}_a(\boldsymbol{x}))$ is within $\mathbf{L}(\mathfrak{e}_{\boldsymbol{x}})$. Then, according to Eq.~\ref{ew2}, the ``local'' center $\mathrm{Z}(\boldsymbol{\mathcal{M}}_a(\boldsymbol{x}))$ is trivial for $\boldsymbol{x}\in\mathrm{W}[A]$ or $\boldsymbol{x}\in\mathrm{W}[\overline{A}]$. Hence, replacing $\otimes_{\boldsymbol{x}}\boldsymbol{\mathcal{M}}_a(\boldsymbol{x})$ by $R^+\mathcal{M}_A R$ as from above result, we have
\begin{equation*}
\mathrm{Z}(R^+\mathcal{M}_A R)=\otimes_{\boldsymbol{x}\in\mathrm{E}[A\overline{A}]}\mathrm{Z}(\boldsymbol{\mathcal{M}}_a(\boldsymbol{x})).
\end{equation*}

\subsection{Proof of Prop.~\ref{declprop}}\label{podecl}
With respect to the decomposition in Eq.~\ref{decl5}, we need to prove that $R^+\mathcal{M}_A R$ and $R^+\mathcal{M}_{\overline{A}} R$ consist of operators of the form $\sum_{\{\mu_{\boldsymbol{x}}\}}\widetilde{\boldsymbol{O}}^{\{\mu_{\boldsymbol{x}}\}}_a\otimes\mathds{1}^{\{\mu_{\boldsymbol{x}}\}}_{\overline{a}}$ and operators of the form $\sum_{\{\mu_{\boldsymbol{x}}\}}\mathds{1}^{\{\mu_{\boldsymbol{x}}\}}_a\otimes\widetilde{\boldsymbol{O}}^{\{\mu_{\boldsymbol{x}}\}}_{\overline{a}}$ respectively.

According to Prop.~\ref{sova}, any operator $\widetilde{\boldsymbol{O}}\in R^+\mathcal{M}_A R$ is a sum of operators of the form
\begin{align}
\begin{split}
\widetilde{\boldsymbol{O}}&=(\otimes_{\boldsymbol{x}\in\mathrm{W}[A]}\widetilde{\boldsymbol{O}}_{\boldsymbol{x}})\otimes(\otimes_{\boldsymbol{x}\in\mathrm{E}[A\overline{A}]}\widetilde{\boldsymbol{O}}_{\boldsymbol{x}})\\
&\otimes(\otimes_{\boldsymbol{x}\in\mathrm{W}[\overline{A}]}\mathds{1}_{\boldsymbol{x}}),\\
&\widetilde{\boldsymbol{O}}_{\boldsymbol{x}}\in\mathbf{L}(\mathfrak{e}_{\boldsymbol{x}})~for~\boldsymbol{x}\in\mathrm{W}[A],\\ &\widetilde{\boldsymbol{O}}_{\boldsymbol{x}}\in \boldsymbol{\mathcal{M}}_a(\boldsymbol{x})~for~\boldsymbol{x}\in\mathrm{E}[A\overline{A}].
\end{split}
\end{align}
And according to the decomposition $\mathfrak{e}_{\boldsymbol{x}}=\oplus_{\mu_{\boldsymbol{x}}}(\mathfrak{e}_{\boldsymbol{x}a}^{\mu_{\boldsymbol{x}}}\otimes\mathfrak{e}_{\boldsymbol{x}\overline{a}}^{\mu_{\boldsymbol{x}}})$ for $\boldsymbol{x}\in\mathrm{E}[A\overline{A}]$ (see Sec.~\ref{exsplitsec}), we can write $\widetilde{\boldsymbol{O}}_{\boldsymbol{x}}=\sum_{\mu_{\boldsymbol{x}}}(\widetilde{\boldsymbol{O}}^{\mu_{\boldsymbol{x}}}_{\boldsymbol{x}a}\otimes\mathds{1}^{\mu_{\boldsymbol{x}}}_{\boldsymbol{x}\overline{a}})$ for $\widetilde{\boldsymbol{O}}_{\boldsymbol{x}}\in\boldsymbol{\mathcal{M}}_a(\boldsymbol{x})$. Then, with respect to the equivalence between different decompositions of $\mathcal{E}$ in Eq.~\ref{decl3} and \ref{decl4}, the above form of operator can be written as
\begin{align}
\begin{split}
&\widetilde{\boldsymbol{O}}\\
=&(\otimes_{\boldsymbol{x}\in\mathrm{W}[A]}\widetilde{\boldsymbol{O}}_{\boldsymbol{x}})\otimes\{\otimes_{\boldsymbol{x}\in\mathrm{E}[A\overline{A}]}[\sum_{\mu_{\boldsymbol{x}}}(\widetilde{\boldsymbol{O}}^{\mu_{\boldsymbol{x}}}_{\boldsymbol{x}a}\otimes\mathds{1}^{\mu_{\boldsymbol{x}}}_{\boldsymbol{x}\overline{a}})]\}\\
&\otimes(\otimes_{\boldsymbol{x}\in\mathrm{W}[\overline{A}]}\mathds{1}_{\boldsymbol{x}})\\
=&(\otimes_{\boldsymbol{x}\in\mathrm{W}[A]}\widetilde{\boldsymbol{O}}_{\boldsymbol{x}})\otimes\{\sum_{\{\mu_{\boldsymbol{x}}\}}[\otimes_{\boldsymbol{x}\in\mathrm{E}[A\overline{A}]}(\widetilde{\boldsymbol{O}}^{\mu_{\boldsymbol{x}}}_{\boldsymbol{x}a}\otimes\mathds{1}^{\mu_{\boldsymbol{x}}}_{\boldsymbol{x}\overline{a}})]\}\\
&\otimes(\otimes_{\boldsymbol{x}\in\mathrm{W}[\overline{A}]}\mathds{1}_{\boldsymbol{x}})\\
=&\sum_{\{\mu_{\boldsymbol{x}}\}}\{(\otimes_{\boldsymbol{x}\in\mathrm{W}[A]}\widetilde{\boldsymbol{O}}_{\boldsymbol{x}})\otimes[\otimes_{\boldsymbol{x}\in\mathrm{E}[A\overline{A}]}(\widetilde{\boldsymbol{O}}^{\mu_{\boldsymbol{x}}}_{\boldsymbol{x}a}\otimes\mathds{1}^{\mu_{\boldsymbol{x}}}_{\boldsymbol{x}\overline{a}})]\\
&\otimes(\otimes_{\boldsymbol{x}\in\mathrm{W}[\overline{A}]}\mathds{1}_{\boldsymbol{x}})\}\\
=&\sum_{\{\mu_{\boldsymbol{x}}\}}\{[(\otimes_{\boldsymbol{x}\in\mathrm{W}[A]}\widetilde{\boldsymbol{O}}_{\boldsymbol{x}})\otimes(\otimes_{\boldsymbol{x}\in\mathrm{E}[A\overline{A}]}\widetilde{\boldsymbol{O}}^{\mu_{\boldsymbol{x}}}_{\boldsymbol{x}a})]\\
&\otimes[(\otimes_{\boldsymbol{x}\in\mathrm{E}[A\overline{A}]}\mathds{1}^{\mu_{\boldsymbol{x}}}_{\boldsymbol{x}\overline{a}})\otimes(\otimes_{\boldsymbol{x}\in\mathrm{W}[\overline{A}]}\mathds{1}_{\boldsymbol{x}})]\},
\end{split}
\end{align}
which is exactly an operator of the form $\sum_{\{\mu_{\boldsymbol{x}}\}}\widetilde{\boldsymbol{O}}^{\{\mu_{\boldsymbol{x}}\}}_a\otimes\mathds{1}^{\{\mu_{\boldsymbol{x}}\}}_{\overline{a}}$ with $\widetilde{\boldsymbol{O}}^{\{\mu_{\boldsymbol{x}}\}}_a=(\otimes_{\boldsymbol{x}\in\mathrm{W}[A]}\widetilde{\boldsymbol{O}}_{\boldsymbol{x}})\otimes(\otimes_{\boldsymbol{x}\in\mathrm{E}[A\overline{A}]}\widetilde{\boldsymbol{O}}^{\mu_{\boldsymbol{x}}}_{\boldsymbol{x}a})$. It follows that $R^+\mathcal{M}_A R$ is included in the von Neumann algebra consisting of all the $\sum_{\{\mu_{\boldsymbol{x}}\}}\widetilde{\boldsymbol{O}}^{\{\mu_{\boldsymbol{x}}\}}_a\otimes\mathds{1}^{\{\mu_{\boldsymbol{x}}\}}_{\overline{a}}$ operators.

Similarly, we can also show that $R^+\mathcal{M}_{\overline{A}} R$ is included in the von Neumann algebra consisting of all the $\sum_{\{\mu_{\boldsymbol{x}}\}}\mathds{1}^{\{\mu_{\boldsymbol{x}}\}}_a\otimes\widetilde{\boldsymbol{O}}^{\{\mu_{\boldsymbol{x}}\}}_{\overline{a}}$ operators. Then, two von Neumann algebras as commutant of one another being included in another two von Neumann algebras as commutant of one another respectively, the situation we have met for several times in the texts, implies that the two pairs are equal. And hence the desired results are proved.

\section{Extracting the RT-formula area term}\label{portformula}
To formalize and elucidate the ``sewing'' as discussed in Sec.~\ref{rtformula}, we specify the desired decompositions $\{\boldsymbol{U}^{\{\mu_{\boldsymbol{x}}\}}_A\}$ and $\{\boldsymbol{U}^{\{\mu_{\boldsymbol{x}}\}}_{\overline{A}}\}$ together with the family of isometry $\{\boldsymbol{J}^{\{\mu_{\boldsymbol{x}}\}}\}$ in a representative and concise example of boundary bipartition $A\overline{A}$ as shown in Fig.~\ref{app4a}, in which the entangling surface (after rearranged to the standard geometric setting) includes all the three types of splits but only one type-1 split. Based on this specification, we can rigorously extract the $\{\ket*{\chi^{\{\mu_{\boldsymbol{x}}\}}}\}$ states for the area terms in the RT formula. The generalization to other cases of boundary bipartitions should be straightforward.

\subsubsection{Specifying $\mathcal{F}^{\{\mu_{\pmb{x}}\}}_a$ and $\mathcal{F}^{\{\mu_{\pmb{x}}\}}_{\bar{a}}$, the pair of parities}

As shown in Fig.~\ref{app4a}, the entangling surface $\mathrm{E}[A\overline{A}]$ consists of seven bulk qudits $\boldsymbol{x}_1,\boldsymbol{x}_2,\ldots,\boldsymbol{x}_7$. We can write
\begin{align}
\begin{split}
\mathcal{E}^{\{\mu_{\boldsymbol{x}}\}}_a&=(\otimes_{\boldsymbol{x}\in\mathrm{W}[A]}\mathfrak{e}_{\boldsymbol{x}})\otimes(\otimes_{\boldsymbol{x}\in\mathrm{E}[A\overline{A}]}\mathfrak{e}_{\boldsymbol{x}a}^{\mu_{\boldsymbol{x}}})\\
&=(\otimes_{\boldsymbol{x}\in\mathrm{W}[A]}\mathfrak{e}_{\boldsymbol{x}})\otimes(\mathfrak{e}_{\boldsymbol{x}_1a}^{\mu_{\boldsymbol{x}}}\otimes\cdots\otimes\mathfrak{e}_{\boldsymbol{x}_7a}^{\mu_{\boldsymbol{x}}}),\\
\ket*{\boldsymbol{B}^{\{\mu_{\boldsymbol{x}}\}}_{a,\boldsymbol{n}}}&=(\otimes_{\boldsymbol{x}\in\mathrm{W}[A]}\ket{\boldsymbol{\beta}_{\boldsymbol{x}}})\otimes\ket*{\boldsymbol{\beta}^{\mu_{\boldsymbol{x}_1}}_{\boldsymbol{x}_1a}\cdots\boldsymbol{\beta}^{\mu_{\boldsymbol{x}_7}}_{\boldsymbol{x}_7a}},\\
\mathcal{E}^{\{\mu_{\boldsymbol{x}}\}}_{\overline{a}}&=(\otimes_{\boldsymbol{x}\in\mathrm{E}[A\overline{A}]}\mathfrak{e}_{\boldsymbol{x}{\overline{a}}}^{\mu_{\boldsymbol{x}}})\otimes(\otimes_{\boldsymbol{x}\in\mathrm{W}[\overline{A}]}\mathfrak{e}_{\boldsymbol{x}})\\
&=(\mathfrak{e}_{\boldsymbol{x}_1{\overline{a}}}^{\mu_{\boldsymbol{x}}}\otimes\cdots\otimes\mathfrak{e}_{\boldsymbol{x}_7{\overline{a}}}^{\mu_{\boldsymbol{x}}})\otimes(\otimes_{\boldsymbol{x}\in\mathrm{W}[\overline{A}]}\mathfrak{e}_{\boldsymbol{x}}),\\
\ket*{\boldsymbol{B}^{\{\mu_{\boldsymbol{x}}\}}_{\overline{a},\boldsymbol{n}'}}&=\ket*{\boldsymbol{\beta}^{\mu_{\boldsymbol{x}_1}}_{\boldsymbol{x}_1\overline{a}}\cdots\boldsymbol{\beta}^{\mu_{\boldsymbol{x}_7}}_{\boldsymbol{x}_7\overline{a}}}\otimes(\otimes_{\boldsymbol{x}\in\mathrm{W}[\overline{A}]}\ket{\boldsymbol{\beta}_{\boldsymbol{x}}}).
\end{split}
\end{align}

To specify $\mathcal{F}^{\{\mu_{\boldsymbol{x}}\}}_a$ and $\mathcal{F}^{\{\mu_{\boldsymbol{x}}\}}_{\overline{a}}$, we consider the pictorial representation of a $\ket{\psi_m}$ state as shown in Fig.~\ref{app4b} with the bipartition indicated. Recall the discussion in Par.~\ref{psim} that in such a pictorial representation, the defining characteristic of a $\ket{\psi_m}$ state is that the parity of the number of dark sides on each loop surrounding each hole $\boldsymbol{x}$ (also on each of the three laterals of the lattice) is even, which underlies the description of the entanglement patterns of the boundary code states $\ket{\widetilde{\varphi}_n}$s (see Sec.~\ref{pattern1}).

As shown in Fig.~\ref{app4b} and \ref{app4c}, we can view $\ket{\psi_m}$ as concatenated from $\ket*{\psi_{A,m_A}}$ and $\ket*{\psi_{\overline{A},m_{\overline{A}}}}$. In this way, the holes (also the laterals) are torn into halves (see the blue and green segments in Fig.~\ref{app4c}) and so are the parities of number of dark sides thereon, and the pairs of the torn halves of the holes coincide with the splits (see Fig.~\ref{19a}) in the entangling surface. Obviously, a necessary condition for any $\ket*{\psi_{A,m_A}}$ and $\ket*{\psi_{\overline{A},m_{\overline{A}}}}$ to be concatenated into a $\ket{\psi_m}$ state is that the parities on the torn halves match (the sum for each pair is even), i.e., the configurations of the parities on the $A$ part and on the $\overline{A}$ part must be $\{even,even,even,\ldots\}_A, \{even,even,even,\ldots\}_{\overline{A}}$ or $\{odd,odd,odd,\ldots\}_A, \{odd,odd,odd,\ldots\}_{\overline{A}}$ or $\{even,odd,even,\ldots\}_A, \{even,odd,even,\ldots\}_{\overline{A}}$, etc. (see Fig.~\ref{app4c}).

According to the above observation, we consider to index the basis states in $\mathcal{F}^{\{\mu_{\boldsymbol{x}}\}}_a$ and $\mathcal{F}^{\{\mu_{\boldsymbol{x}}\}}_{\overline{a}}$ by the possible configurations of parities on the torn halves, and we start with clarifying the number of independent pairs of parities (on each pair of torn halves). Here, by independent pairs of parity, we mean those parities, e.g., in the configurations $\{even,odd,even,\ldots\}_A, \{even,odd,even,\ldots\}_{\overline{A}}$, which can be arbitrary as read from $\ket*{\psi_{A,m_A}}$ and $\ket*{\psi_{\overline{A},m_{\overline{A}}}}$ that are torn from some $\ket{\psi_m}$.

\onecolumngrid
\begin{center}
\begin{figure}[ht]
\centering
    \includegraphics[width=18cm]{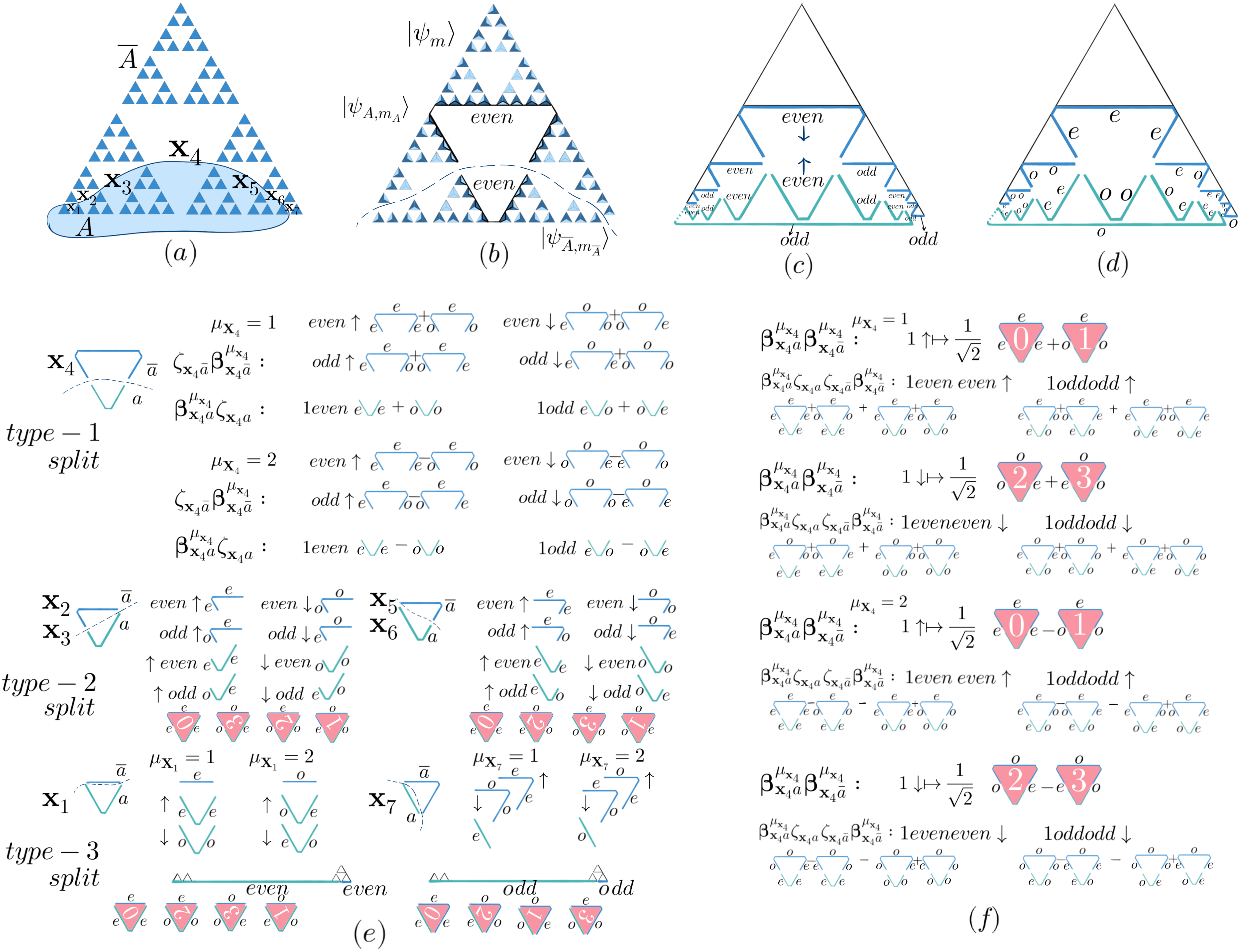}   
\phantomsubfloat{\label{app4a}}\phantomsubfloat{\label{app4b}}
\phantomsubfloat{\label{app4c}}\phantomsubfloat{\label{app4d}}
\phantomsubfloat{\label{app4e}}\phantomsubfloat{\label{app4f}}
\caption{According to the bipartition specified by (a), a $\ket{\psi_m}$ state is torn into $\ket*{\psi_{A,m_A}}$ and $\ket*{\psi_{\overline{A},m_{\overline{A}}}}$ as illustrated in (b). Also, the holes corresponding to the bulk qudits in the entangling surface are torn into halves, and so are the parity of the dark sides on the loop surrounding a hole. (c) To concatenated any $\ket*{\psi_{A,m_A}}$ and $\ket*{\psi_{\overline{A},m_{\overline{A}}}}$ into a $\ket{\psi_m}$ state, a necessary condition is the match of the parties on the torn halves. (d) The finer description of the parities on the torn halves. (e) The complete list of how $(\boldsymbol{\beta}^{\mu_{\boldsymbol{x}}}_{\boldsymbol{x}a},\zeta_{\boldsymbol{x}a})$ or $(\boldsymbol{\beta}^{\mu_{\boldsymbol{x}}}_{\boldsymbol{x}\bar a},\zeta_{\boldsymbol{x}\bar a})$ corresponds to the finer descriptions of parities on the torn halves. (f) Combining the finer description characterized by $(\boldsymbol{\beta}^{\mu_{\boldsymbol{x}_4}}_{\boldsymbol{x}_4a},\zeta_{\boldsymbol{x}_4a})$ and $(\boldsymbol{\beta}^{\mu_{\boldsymbol{x}_4}}_{\boldsymbol{x}_4\bar a},\zeta_{\boldsymbol{x}_4\bar a})$, following the match of the parities on the torn halves, gives rise to all possible descriptions of the hole.}
\label{figapp4}
\end{figure}
\end{center}
\twocolumngrid

As shown in Fig.~\ref{app4c}, generally, for any $\ket{\psi_m}$ to be concatenated from $\ket*{\psi_{A,m_A}}$ and $\ket*{\psi_{\overline{A},m_{\overline{A}}}}$, while there are nine pairs of torn halves (including the torn halves for both the two laterals as indicated by the solid green segment and the dashed green segment), according to the basics of the pictorial representation (see Par.~\ref{picrep}), there are indeed exactly eight of the parities (only the solid segments) can be chosen arbitrarily~\footnote{There are two important properties of the pictorial representations of the $\ket{\psi_m}$s as discussed in Par.~\ref{picrep}: (1) The number of all dark sides on either $A$ or $\overline{A}$ is even in total. (2) By sequentially applying the $T_{ii'}$ operators, one can arbitrarily change the parities in the torn halves of any hole or on the lateral indicated by the solid segments, without changing the parities in other torn halves except on the torn halves on the other lateral, i.e., the one indicated by the dashed segment. Here, by ``arbitrarily changed'', we mean that before and after the change, $\ket*{\psi_{A,m_A}}$ and $\ket*{\psi_{\overline{A},m_{\overline{A}}}}$ satisfying the configurations of parities can all be viewed as torn from some $\ket{\psi_m}$ state.}. If we fix the index $\{\mu_{\boldsymbol{x}}\}$, the independence can be further reduced by two. Indeed, according to the illustration in Fig.~\ref{app4a} and \ref{19a}, bulk qudit $\boldsymbol{x}_4$ exhibits type-1 split, bulk qudits $\boldsymbol{x}_1,\boldsymbol{x}_7$ exhibit type-2 splits, and bulk qudits $\boldsymbol{x}_2,\boldsymbol{x}_3,\boldsymbol{x}_5,\boldsymbol{x}_6$ exhibit type-3 splits. According to Sec.~\ref{exsplitsec} and App.~\ref{posplit}, for a fixed index $\mu_{\boldsymbol{x}}$ of a type-2 split, i.e., a fixed sector in $\mathfrak{e}_{\boldsymbol{x}}=\oplus_{\mu_{\boldsymbol{x}}}(\mathfrak{e}_{\boldsymbol{x}a}^{\mu_{\boldsymbol{x}}}\otimes\mathfrak{e}_{\boldsymbol{x}\overline{a}}^{\mu_{\boldsymbol{x}}})$, the parity of either of the corresponding torn halves is fixed. Hence, the parities on $\boldsymbol{x}_1,\boldsymbol{x}_7$ are not independent. However, it is easy to check that for connected boundary bipartition the number of type-2 splits is always finite and independent on the system size, and the torn parities thereon can be viewed as attached to the parties on the laterals (see Fig.~\ref{app4c}).

Now, we can view the independent pairs of parity as on holes $\boldsymbol{x}_2,\boldsymbol{x}_3,\boldsymbol{x}_4,\boldsymbol{x}_5,\boldsymbol{x}_6$, together with the pair on the lateral related to $(\boldsymbol{x}_1,\boldsymbol{x}_7)$. Note that it is easy to check that the number of independent pairs scales linearly on the number of qudits in the entangling surface. Accordingly, we define. 
\begin{align}
\begin{split}
&\mathcal{F}^{\{\mu_{\boldsymbol{x}}\}}_a=\mathfrak{f}_{\boldsymbol{x}_1,\boldsymbol{x}_7a}\otimes\mathfrak{f}_{\boldsymbol{x}_2a}\otimes\mathfrak{f}_{\boldsymbol{x}_3a}\otimes\mathfrak{f}_{\boldsymbol{x}_4a}\otimes\mathfrak{f}_{\boldsymbol{x}_5a}\otimes\mathfrak{f}_{\boldsymbol{x}_6a},\\
&\mathfrak{f}_{\boldsymbol{x}a}=\mathbb{C}^2, \quad \mathfrak{f}_{\boldsymbol{x}a}\ni\ket{\zeta_{\boldsymbol{x}a}}=\ket{even},\ket{odd},\\
&\ket*{\boldsymbol{Z}^{\{\mu_{\boldsymbol{x}}\}}_{a,\boldsymbol{l}}}=\ket{\zeta_{\boldsymbol{x}_1\boldsymbol{x}_7a}\zeta_{\boldsymbol{x}_2a}\cdots\zeta_{\boldsymbol{x}_6a}},\\
&\mathcal{F}^{\{\mu_{\boldsymbol{x}}\}}_{\overline{a}}=\mathfrak{f}_{\boldsymbol{x}_1,\boldsymbol{x}_7{\overline{a}}}\otimes\mathfrak{f}_{\boldsymbol{x}_2{\overline{a}}}\otimes\mathfrak{f}_{\boldsymbol{x}_3{\overline{a}}}\otimes\mathfrak{f}_{\boldsymbol{x}_4{\overline{a}}}\otimes\mathfrak{f}_{\boldsymbol{x}_5{\overline{a}}}\otimes\mathfrak{f}_{\boldsymbol{x}_6{\overline{a}}},\\
&\mathfrak{f}_{\boldsymbol{x}{\overline{a}}}=\mathbb{C}^2, \quad \mathfrak{f}_{\boldsymbol{x}{\overline{a}}}\ni\ket{\zeta_{\boldsymbol{x}\overline{a}}}=\ket{even},\ket{odd},\\
&\ket*{\boldsymbol{Z}^{\{\mu_{\boldsymbol{x}}\}}_{{\overline{a}},\boldsymbol{l}'}}=\ket{\zeta_{\boldsymbol{x}_1\boldsymbol{x}_7\overline{a}}\zeta_{\boldsymbol{x}_2\overline{a}}\cdots\zeta_{\boldsymbol{x}_6\overline{a}}},
\end{split}
\end{align}
where, according to the subscripts, each pair $(\mathfrak{f}_{\boldsymbol{x}a},\mathfrak{f}_{\boldsymbol{x}{\overline{a}}})$ corresponds to an independent pair of torn halves, and the basis states $\ket{\zeta_{\boldsymbol{x}a}}$ and $\ket{\zeta_{\boldsymbol{x}\bar a}}$ are indexed by the possible parities.

\subsubsection{Specifying $\pmb{U}^{\{\mu_{\pmb{x}}\}}_A$ and $\pmb{U}^{\{\mu_{\pmb{x}}\}}_{\overline{A}}$, finer description of parities}
To specify the decompositions of $\mathcal{H}_A$ and $\mathcal{H}_{\overline{A}}$, i.e.,
\begin{align}
\begin{split}
&\boldsymbol{U}^{\{\mu_{\boldsymbol{x}}\}}_A: \mathcal{E}^{\{\mu_{\boldsymbol{x}}\}}_a\otimes\mathcal{F}^{\{\mu_{\boldsymbol{x}}\}}_a\rightarrow\mathcal{H}_A,\\
&\boldsymbol{U}^{\{\mu_{\boldsymbol{x}}\}}_A(\ket*{\boldsymbol{B}^{\{\mu_{\boldsymbol{x}}\}}_{a,\boldsymbol{n}}}\otimes\ket*{\boldsymbol{Z}^{\{\mu_{\boldsymbol{x}}\}}_{a,\boldsymbol{l}}})=\ket*{\widetilde{\varphi}^{\{\mu_{\boldsymbol{x}}\}}_{A, \boldsymbol{n}\boldsymbol{l}}},\\
&\boldsymbol{U}^{\{\mu_{\boldsymbol{x}}\}}_{\overline{A}}: \mathcal{E}^{\{\mu_{\boldsymbol{x}}\}}_{\overline{a}}\otimes\mathcal{F}^{\{\mu_{\boldsymbol{x}}\}}_{\overline{a}}\rightarrow\mathcal{H}_{\overline{A}},\\
&\boldsymbol{U}^{\{\mu_{\boldsymbol{x}}\}}_{\overline{A}}(\ket*{\boldsymbol{B}^{\{\mu_{\boldsymbol{x}}\}}_{{\overline{a}},\boldsymbol{n}'}}\otimes\ket*{\boldsymbol{Z}^{\{\mu_{\boldsymbol{x}}\}}_{{\overline{a}},\boldsymbol{l}'}})=\ket*{\widetilde{\varphi}^{\{\mu_{\boldsymbol{x}}\}}_{{\overline{A}}, \boldsymbol{n}'\boldsymbol{l}'}},
\end{split}
\end{align}
we simply need to define $\ket*{\widetilde{\varphi}^{\{\mu_{\boldsymbol{x}}\}}_{A, \boldsymbol{n}\boldsymbol{l}}}$ and $\ket*{\widetilde{\varphi}^{\{\mu_{\boldsymbol{x}}\}}_{{\overline{A}}, \boldsymbol{n}'\boldsymbol{l}'}}$ as sums of certain $\ket*{\psi_{A,m_A}}$s and $\ket*{\psi_{\overline{A},m_{\overline{A}}}}$s respectively. To that end, we firstly note that according to Sec.~\ref{exsplitsec} and App.~\ref{posplit}, for type-1 split, we have $\ket*{\boldsymbol{\beta}^{\mu_{\boldsymbol{x}}}_{\boldsymbol{x}a}\boldsymbol{\beta}^{\mu_{\boldsymbol{x}}}_{\boldsymbol{x}{\overline{a}}}}=\ket{\boldsymbol{\beta}_{\boldsymbol{x}}}\pm\ket{\boldsymbol{\beta}'_{\boldsymbol{x}}}$ while for other types of split we have $\ket*{\boldsymbol{\beta}^{\mu_{\boldsymbol{x}}}_{\boldsymbol{x}a}\boldsymbol{\beta}^{\mu_{\boldsymbol{x}}}_{\boldsymbol{x}{\overline{a}}}}=\ket{\boldsymbol{\beta}_{\boldsymbol{x}}}$. Then, in this concrete example, we have
\begin{align}\label{sign00}
\begin{split}
&\ket*{\boldsymbol{B}^{\{\mu_{\boldsymbol{x}}\}}_{a,\boldsymbol{n}}\boldsymbol{B}^{\{\mu_{\boldsymbol{x}}\}}_{{\overline{a}},\boldsymbol{n}'}}\\
&=(\otimes_{\boldsymbol{x}\in\mathrm{W}[A]}\ket{\boldsymbol{\beta}_{\boldsymbol{x}}})\otimes\ket*{\boldsymbol{\beta}^{\mu_{\boldsymbol{x}_1}}_{\boldsymbol{x}_1a}\cdots\boldsymbol{\beta}^{\mu_{\boldsymbol{x}_7}}_{\boldsymbol{x}_7a}}\\
&\otimes\ket*{\boldsymbol{\beta}^{\mu_{\boldsymbol{x}_1}}_{\boldsymbol{x}_1\overline{a}}\cdots\boldsymbol{\beta}^{\mu_{\boldsymbol{x}_7}}_{\boldsymbol{x}_7\overline{a}}}\otimes(\otimes_{\boldsymbol{x}\in\mathrm{W}[\overline{A}]}\ket{\boldsymbol{\beta}_{\boldsymbol{x}}})\\
&=(\otimes_{\boldsymbol{x}\in\mathrm{W}[A]}\ket{\boldsymbol{\beta}_{\boldsymbol{x}}})\\
&\otimes\ket*{\boldsymbol{\beta}^{\mu_{\boldsymbol{x}_1}}_{\boldsymbol{x}_1a}\boldsymbol{\beta}^{\mu_{\boldsymbol{x}_1}}_{\boldsymbol{x}_1{\overline{a}}}}\otimes\cdots\otimes\ket*{\boldsymbol{\beta}^{\mu_{\boldsymbol{x}_7}}_{\boldsymbol{x}_7a}\boldsymbol{\beta}^{\mu_{\boldsymbol{x}_7}}_{\boldsymbol{x}_7{\overline{a}}}}\\
&\otimes(\otimes_{\boldsymbol{x}\in\mathrm{W}[\overline{A}]}\ket{\boldsymbol{\beta}_{\boldsymbol{x}}})\\
&=\frac{1}{\sqrt{2}}(\otimes_{\boldsymbol{x}\in\mathrm{W}[A]}\ket{\boldsymbol{\beta}_{\boldsymbol{x}}})\\
&\otimes\ket{\boldsymbol{\beta}_{\boldsymbol{x}_1}\boldsymbol{\beta}_{\boldsymbol{x}_2}\boldsymbol{\beta}_{\boldsymbol{x}_3}(\boldsymbol{\beta}_{\boldsymbol{x}_4}\pm\boldsymbol{\beta}'_{\boldsymbol{x}_4})\boldsymbol{\beta}_{\boldsymbol{x}_5}\boldsymbol{\beta}_{\boldsymbol{x}_6}\boldsymbol{\beta}_{\boldsymbol{x}_7}}\\
&\otimes(\otimes_{\boldsymbol{x}\in\mathrm{W}[\overline{A}]}\ket{\boldsymbol{\beta}_{\boldsymbol{x}}})\\
&=\frac{1}{\sqrt{2}}\ket{\boldsymbol{\beta}_{\boldsymbol{1}}\cdots\boldsymbol{\beta}_{\boldsymbol{x}_1}\boldsymbol{\beta}_{\boldsymbol{x}_2}\boldsymbol{\beta}_{\boldsymbol{x}_3}(\boldsymbol{\beta}_{\boldsymbol{x}_4}\pm\boldsymbol{\beta}'_{\boldsymbol{x}_4})\boldsymbol{\beta}_{\boldsymbol{x}_5}\boldsymbol{\beta}_{\boldsymbol{x}_6}\boldsymbol{\beta}_{\boldsymbol{x}_7}\cdots\boldsymbol{\beta}_{K}}\\
&=\frac{1}{\sqrt{2}}\ket{\boldsymbol{\beta}_{\boldsymbol{1}}\cdots\boldsymbol{\beta}_{\boldsymbol{x}_1}\boldsymbol{\beta}_{\boldsymbol{x}_2}\boldsymbol{\beta}_{\boldsymbol{x}_3}\boldsymbol{\beta}_{\boldsymbol{x}_4}\boldsymbol{\beta}_{\boldsymbol{x}_5}\boldsymbol{\beta}_{\boldsymbol{x}_6}\boldsymbol{\beta}_{\boldsymbol{x}_7}\cdots\boldsymbol{\beta}_{K}}\\
&\pm\ket{\boldsymbol{\beta}_{\boldsymbol{1}}\cdots\boldsymbol{\beta}_{\boldsymbol{x}_1}\boldsymbol{\beta}_{\boldsymbol{x}_2}\boldsymbol{\beta}_{\boldsymbol{x}_3}\boldsymbol{\beta}'_{\boldsymbol{x}_4}\boldsymbol{\beta}_{\boldsymbol{x}_5}\boldsymbol{\beta}_{\boldsymbol{x}_6}\boldsymbol{\beta}_{\boldsymbol{x}_7}\cdots\boldsymbol{\beta}_{K}},
\end{split}
\end{align}
and
\begin{align}
\begin{split}
&R\ket*{\boldsymbol{B}^{\{\mu_{\boldsymbol{x}}\}}_{a,\boldsymbol{n}}\boldsymbol{B}^{\{\mu_{\boldsymbol{x}}\}}_{{\overline{a}},\boldsymbol{n}'}}\\
&=\frac{1}{\sqrt{2}}R\ket{\boldsymbol{\beta}_{\boldsymbol{1}}\cdots\boldsymbol{\beta}_{\boldsymbol{x}_1}\boldsymbol{\beta}_{\boldsymbol{x}_2}\boldsymbol{\beta}_{\boldsymbol{x}_3}\boldsymbol{\beta}_{\boldsymbol{x}_4}\boldsymbol{\beta}_{\boldsymbol{x}_5}\boldsymbol{\beta}_{\boldsymbol{x}_6}\boldsymbol{\beta}_{\boldsymbol{x}_7}\cdots\boldsymbol{\beta}_{K}}\\
&\pm \frac{1}{\sqrt{2}}R\ket{\boldsymbol{\beta}_{\boldsymbol{1}}\cdots\boldsymbol{\beta}_{\boldsymbol{x}_1}\boldsymbol{\beta}_{\boldsymbol{x}_2}\boldsymbol{\beta}_{\boldsymbol{x}_3}\boldsymbol{\beta}'_{\boldsymbol{x}_4}\boldsymbol{\beta}_{\boldsymbol{x}_5}\boldsymbol{\beta}_{\boldsymbol{x}_6}\boldsymbol{\beta}_{\boldsymbol{x}_7}\cdots\boldsymbol{\beta}_{K}}\\
&=\frac{1}{\sqrt{2}}(\ket{\widetilde{\varphi}_n}\pm\ket{\widetilde{\varphi}_{n'}}).
\end{split}
\end{align}

In view of these observations, to ensure that the to-be-defined $\{\ket*{\widetilde{\varphi}^{\{\mu_{\boldsymbol{x}}\}}_{A, \boldsymbol{n}\boldsymbol{l}}}\}$ and $\{\ket*{\widetilde{\varphi}^{\{\mu_{\boldsymbol{x}}\}}_{{\overline{A}}, \boldsymbol{n}'\boldsymbol{l}'}}\}$ can be ``sewed'' into $R\ket*{\boldsymbol{B}^{\{\mu_{\boldsymbol{x}}\}}_{a,\boldsymbol{n}}\boldsymbol{B}^{\{\mu_{\boldsymbol{x}}\}}_{{\overline{a}},\boldsymbol{n}'}}=1/2(\ket{\widetilde{\varphi}_n}\pm\ket{\widetilde{\varphi}_{n'}}$, any pair of $\ket*{\psi_{A,m_A}})$ and $\ket*{\psi_{\overline{A},m_{\overline{A}}}}$ (contributing to the expansion of $\ket*{\widetilde{\varphi}^{\{\mu_{\boldsymbol{x}}\}}_{A, \boldsymbol{n}\boldsymbol{l}}}$ and $\ket*{\widetilde{\varphi}^{\{\mu_{\boldsymbol{x}}\}}_{{\overline{A}}, \boldsymbol{n}'\boldsymbol{l}'}}$ respectively), upon being matched and concatenated into some $\ket{\psi_m}$, should satisfy (according to Eq.~\ref{dr+})
\begin{align}\label{sign0}
\begin{split}
&R^+\ket{\psi_m}\\
&=\frac{1}{\sqrt{2^K}}\ket{\boldsymbol{\beta}_{\boldsymbol{1}}\cdots\boldsymbol{\beta}_{\boldsymbol{x}_1}\boldsymbol{\beta}_{\boldsymbol{x}_2}\boldsymbol{\beta}_{\boldsymbol{x}_3}\boldsymbol{\beta}_{\boldsymbol{x}_4}\boldsymbol{\beta}_{\boldsymbol{x}_5}\boldsymbol{\beta}_{\boldsymbol{x}_6}\boldsymbol{\beta}_{\boldsymbol{x}_7}\cdots\boldsymbol{\beta}_{K}}\\
&or\\
&R^+\ket{\psi_m}\\
&=\frac{1}{\sqrt{2^K}}\ket{\boldsymbol{\beta}_{\boldsymbol{1}}\cdots\boldsymbol{\beta}_{\boldsymbol{x}_1}\boldsymbol{\beta}_{\boldsymbol{x}_2}\boldsymbol{\beta}_{\boldsymbol{x}_3}\boldsymbol{\beta}'_{\boldsymbol{x}_4}\boldsymbol{\beta}_{\boldsymbol{x}_5}\boldsymbol{\beta}_{\boldsymbol{x}_6}\boldsymbol{\beta}_{\boldsymbol{x}_7}\cdots\boldsymbol{\beta}_{K}}.
\end{split}
\end{align}
And note that $1/2(\ket{\widetilde{\varphi}_n}\pm\ket{\widetilde{\varphi}_{n'}}$ is a sum of certain $\ket{\psi_m}$ states with coefficients equal to $\pm 1$.

Recall that if $\ket*{\psi_{A,m_A}}$ and $\ket*{\psi_{\overline{A},m_{\overline{A}}}}$ can be matched and concatenated into a $\ket{\psi_m}$ state that satisfies the above condition, we need finer description regarding the parities than simply the parities on the torn halves as described previously. Indeed, according to the discussion of the entanglement patterns in Par.~\ref{pattern1}, if the above condition is satisfied for $\ket{\psi_m}$, then the parity on each side of the triangular loops surrounding a hole should be fixed to represent the emergent $\boldsymbol{\beta}_{\boldsymbol{x}}$. As shown in Fig.~\ref{app4d}, we not only need the parities for each of the torn halves to match as $(even,even)$ or $(odd,odd)$, but also need the parities on the line segments of each of the halves to satisfy certain conditions. For instance, for $\boldsymbol{x}_4$ as illustrated in Fig.~\ref{app4d}, to ensure that the concatenated $\ket{\psi_m}$ maps to the bulk qudit-product-state with $\boldsymbol{\beta}_{\boldsymbol{x}_4}=\boldsymbol{1}$, corresponding to $(even,even)$ on the two torn halves, we can exactly have the finer descriptions $[(e,e,e),(o,o)],[(o,e,o),(e,e)]$ (with $e$ and $o$ abbreviated for $even$ and $odd$ respectively); and corresponding to $(odd,odd)$ on the two torn halves, we can exactly have finer descriptions $[(e,e,o),(o,e)],[(o,e,e),(e,o)]$.

In the following, we show how to exactly characterize the ``correct'' $\ket*{\psi_{A,m_A}}$s and $\ket*{\psi_{\overline{A},m_{\overline{A}}}}$s by specifying the finer description of the parities. An important property of these finer descriptions is that on each of the torn halves for each holes, the corresponding finer description can be arbitrary, e.g., being $(e,e,e),(e,e,o),\ldots$ all the $2\times2\times2$ possibilities for the upper torn half of $\boldsymbol{x}_4$, and is independent on the finer descriptions for other holes. In other words, any configuration of such finer descriptions for all the torn halves on $A$ or $\overline{A}$ can be read off some $\ket*{\psi_{A,m_A}}$s and $\ket*{\psi_{\overline{A},m_{\overline{A}}}}$ as torn form some $\ket{\psi_m}$. The reason is the same for the independence of the pairs of parities for the torn halves (see previous arguments and footnote).

Firstly, we write
\begin{align}
\begin{split}
&\ket*{\boldsymbol{B}^{\{\mu_{\boldsymbol{x}}\}}_{a,\boldsymbol{n}}}\otimes\ket*{\boldsymbol{Z}^{\{\mu_{\boldsymbol{x}}\}}_{a,\boldsymbol{l}}}\\
&=(\otimes_{\boldsymbol{x}\in\mathrm{W}[A]}\ket{\boldsymbol{\beta}_{\boldsymbol{x}}})\\
&\otimes\ket*{\boldsymbol{\beta}^{\mu_{\boldsymbol{x}_1}}_{\boldsymbol{x}_1a}\cdots\boldsymbol{\beta}^{\mu_{\boldsymbol{x}_7}}_{\boldsymbol{x}_7a}}\otimes\ket{\zeta_{\boldsymbol{x}_1\boldsymbol{x}_7a}\zeta_{\boldsymbol{x}_2a}\cdots\zeta_{\boldsymbol{x}_6a}}\\
&=(\otimes_{\boldsymbol{x}\in\mathrm{W}[A]}\ket{\boldsymbol{\beta}_{\boldsymbol{x}}})\\
&\otimes\ket*{(\boldsymbol{\beta}^{\mu_{\boldsymbol{x}_1}}_{\boldsymbol{x}_1a}\boldsymbol{\beta}^{\mu_{\boldsymbol{x}_7}}_{\boldsymbol{x}_7a}\zeta_{\boldsymbol{x}_1\boldsymbol{x}_7a})(\boldsymbol{\beta}^{\mu_{\boldsymbol{x}_2}}_{\boldsymbol{x}_2a}\zeta_{\boldsymbol{x}_2a})\cdots(\boldsymbol{\beta}^{\mu_{\boldsymbol{x}_6}}_{\boldsymbol{x}_6a}\zeta_{\boldsymbol{x}_6a})},\\
&\ket*{\boldsymbol{Z}^{\{\mu_{\boldsymbol{x}}\}}_{\overline{a},\boldsymbol{l}'}}\otimes\ket*{\boldsymbol{B}^{\{\mu_{\boldsymbol{x}}\}}_{\overline{a},\boldsymbol{n}'}}\\
&=\ket{\zeta_{\boldsymbol{x}_1\boldsymbol{x}_7\overline{a}}\zeta_{\boldsymbol{x}_2\overline{a}}\cdots\zeta_{\boldsymbol{x}_6\overline{a}}}\otimes\ket*{\boldsymbol{\beta}^{\mu_{\boldsymbol{x}_1}}_{\boldsymbol{x}_1\overline{a}}\cdots\boldsymbol{\beta}^{\mu_{\boldsymbol{x}_7}}_{\boldsymbol{x}_7\overline{a}}}\\
&\otimes(\otimes_{\boldsymbol{x}\in\mathrm{W}[\overline{A}]}\ket{\boldsymbol{\beta}_{\boldsymbol{x}}})\\
&=\ket*{(\zeta_{\boldsymbol{x}_1\boldsymbol{x}_7\overline{a}}\boldsymbol{\beta}^{\mu_{\boldsymbol{x}_1}}_{\boldsymbol{x}_1\overline{a}}\boldsymbol{\beta}^{\mu_{\boldsymbol{x}_7}}_{\boldsymbol{x}_7\overline{a}})(\zeta_{\boldsymbol{x}_2\overline{a}}\boldsymbol{\beta}^{\mu_{\boldsymbol{x}_2}}_{\boldsymbol{x}_2\overline{a}})\cdots(\zeta_{\boldsymbol{x}_6\overline{a}}\boldsymbol{\beta}^{\mu_{\boldsymbol{x}_6}}_{\boldsymbol{x}_6\overline{a}})}\\
&\otimes(\otimes_{\boldsymbol{x}\in\mathrm{W}[\overline{A}]}\ket{\boldsymbol{\beta}_{\boldsymbol{x}}}).
\end{split}
\end{align}
And recall that the degrees of freedom $\boldsymbol{\beta}_{\boldsymbol{x}},\boldsymbol{\beta}^{\mu_{\boldsymbol{x}}}_{\boldsymbol{x}a},\zeta_{\boldsymbol{x}a},\boldsymbol{\beta}^{\mu_{\boldsymbol{x}}}_{\boldsymbol{x}\overline{a}},\zeta_{\boldsymbol{x}\overline{a}}$ can be given as as below (see Sec.~\ref{exsplitsec} and App.~\ref{posplit})
\begin{align}
\begin{split}
&\boldsymbol{\beta}_{\boldsymbol{x}}={\boldsymbol{0}},{\boldsymbol{1}},{\boldsymbol{2}},{\boldsymbol{3}},\quad \boldsymbol{x}\in\mathrm{W}[A],\mathrm{W}[\overline{A}],\\
&\boldsymbol{\beta}^{1}_{\boldsymbol{x}_1a}=\uparrow, \downarrow,~\boldsymbol{\beta}^{2}_{\boldsymbol{x}_1a}=\uparrow, \downarrow,\quad \boldsymbol{\beta}^{1}_{\boldsymbol{x}_1\overline{a}}=1,~\boldsymbol{\beta}^{2}_{\boldsymbol{x}_1\overline{a}}=1\\
&\boldsymbol{\beta}^{1}_{\boldsymbol{x}_7a}=1,~\boldsymbol{\beta}^{2}_{\boldsymbol{x}_7a}=1,\quad\quad~ \boldsymbol{\beta}^{1}_{\boldsymbol{x}_7\overline{a}}=\uparrow, \downarrow,~\boldsymbol{\beta}^{2}_{\boldsymbol{x}_7\overline{a}}=\uparrow, \downarrow,\\
&\boldsymbol{\beta}^{1}_{\boldsymbol{x}_4a}=1,~\boldsymbol{\beta}^{2}_{\boldsymbol{x}_4a}=1,\quad\quad~ \boldsymbol{\beta}^{1}_{\boldsymbol{x}_4\overline{a}}=\uparrow, \downarrow,~\boldsymbol{\beta}^{2}_{\boldsymbol{x}_4\overline{a}}=\uparrow, \downarrow,\\
&\boldsymbol{\beta}^{1}_{\boldsymbol{x}_2a}=\uparrow, \downarrow,\quad\quad\quad\quad\quad\quad~ \boldsymbol{\beta}^{1}_{\boldsymbol{x}_2\overline{a}}=\uparrow, \downarrow,\\
&......
\end{split}
\end{align}
Note that for $\boldsymbol{x}_3,\boldsymbol{x}_5,\boldsymbol{x}_6$ the case is the same to $\boldsymbol{x}_2$.

Then, we can associate the pairs $(\boldsymbol{\beta}^{\mu_{\boldsymbol{x}}}_{\boldsymbol{x}a},\zeta_{\boldsymbol{x}a})$ and $(\boldsymbol{\beta}^{\mu_{\boldsymbol{x}}}_{\boldsymbol{x}\bar a},\zeta_{\boldsymbol{x}\bar a})$ the finer descriptions on the lower torn half and the upper torn half respectively. The complete association or correspondence is listed in Fig.~\ref{app4e}. Accordingly, for $(\boldsymbol{\beta}^{\mu_{\boldsymbol{x}_1}}_{\boldsymbol{x}_1a}\boldsymbol{\beta}^{\mu_{\boldsymbol{x}_7}}_{\boldsymbol{x}_7a}\zeta_{\boldsymbol{x}_1\boldsymbol{x}_7a})(\boldsymbol{\beta}^{\mu_{\boldsymbol{x}_2}}_{\boldsymbol{x}_2a}\zeta_{\boldsymbol{x}_2a})\cdots(\boldsymbol{\beta}^{\mu_{\boldsymbol{x}_6}}_{\boldsymbol{x}_6a}\zeta_{\boldsymbol{x}_6a})$ and $(\zeta_{\boldsymbol{x}_1\boldsymbol{x}_7\overline{a}}\boldsymbol{\beta}^{\mu_{\boldsymbol{x}_1}}_{\boldsymbol{x}_1\overline{a}}\boldsymbol{\beta}^{\mu_{\boldsymbol{x}_7}}_{\boldsymbol{x}_7\overline{a}})(\zeta_{\boldsymbol{x}_2\overline{a}}\boldsymbol{\beta}^{\mu_{\boldsymbol{x}_2}}_{\boldsymbol{x}_2\overline{a}})\cdots(\zeta_{\boldsymbol{x}_6\overline{a}}\boldsymbol{\beta}^{\mu_{\boldsymbol{x}_6}}_{\boldsymbol{x}_6\overline{a}})$, we have associated configurations of the finer descriptions, like the illustration in Fig.~\ref{app4d}. It is easy to check that all the possible finer descriptions are included in Fig.~\ref{app4e}. In this list, for $\boldsymbol{x}_4$, each pair $(\boldsymbol{\beta}^{\mu_{\boldsymbol{x}}}_{\boldsymbol{x}a},\zeta_{\boldsymbol{x}a})$ or $(\boldsymbol{\beta}^{\mu_{\boldsymbol{x}}}_{\boldsymbol{x}\bar a},\zeta_{\boldsymbol{x}\bar a})$ corresponds to two finer descriptions with the sign $+$ or $-$ assigned, whose meaning will be clarified below. And for others, each pair corresponds to exactly one finer description.

Now, based on the list in Fig.~\ref{app4e}, for each specific $\{\mu_{\boldsymbol{x}}\}=\mu_{\boldsymbol{x}_1},\mu_{\boldsymbol{x}_2},\ldots,\mu_{\boldsymbol{x}_7}$, and for given $\ket*{\boldsymbol{B}^{\{\mu_{\boldsymbol{x}}\}}_{a,\boldsymbol{n}}}=(\otimes_{\boldsymbol{x}\in\mathrm{W}[A]}\ket{\boldsymbol{\beta}_{\boldsymbol{x}}})\otimes\ket*{\boldsymbol{\beta}^{\mu_{\boldsymbol{x}_1}}_{\boldsymbol{x}_1a}\cdots\boldsymbol{\beta}^{\mu_{\boldsymbol{x}_7}}_{\boldsymbol{x}_7a}}$, $\ket*{\boldsymbol{Z}^{\{\mu_{\boldsymbol{x}}\}}_{a,\boldsymbol{l}}}=\ket{\zeta_{\boldsymbol{x}_1\boldsymbol{x}_7a}\zeta_{\boldsymbol{x}_2a}\cdots\zeta_{\boldsymbol{x}_6a}}$ and $\ket*{\boldsymbol{B}^{\{\mu_{\boldsymbol{x}}\}}_{\overline{a},\boldsymbol{n}'}}=\ket*{\boldsymbol{\beta}^{\mu_{\boldsymbol{x}_1}}_{\boldsymbol{x}_1\overline{a}}\cdots\boldsymbol{\beta}^{\mu_{\boldsymbol{x}_7}}_{\boldsymbol{x}_7\overline{a}}}\otimes(\otimes_{\boldsymbol{x}\in\mathrm{W}[\overline{A}]}\ket{\boldsymbol{\beta}_{\boldsymbol{x}}})$, $\ket*{\boldsymbol{Z}^{\{\mu_{\boldsymbol{x}}\}}_{\overline{a},\boldsymbol{l}'}}=\ket{\zeta_{\boldsymbol{x}_1\boldsymbol{x}_7\overline{a}}\zeta_{\boldsymbol{x}_2\overline{a}}\cdots\zeta_{\boldsymbol{x}_6\overline{a}}}$, we define $\ket*{\widetilde{\varphi}^{\{\mu_{\boldsymbol{x}}\}}_{A, \boldsymbol{n}\boldsymbol{l}}}$ and $\ket*{\widetilde{\varphi}^{\{\mu_{\boldsymbol{x}}\}}_{{\overline{A}}, \boldsymbol{n}'\boldsymbol{l}'}}$ as follows: (1) The two normalized states have the form 
\begin{align}\label{sign1}
\begin{split}
&\ket*{\widetilde{\varphi}^{\{\mu_{\boldsymbol{x}}\}}_{A, \boldsymbol{n}\boldsymbol{l}}}=\frac{1}{\sqrt{\mathscr{N}_A}}\sum_{m_A}c_{m_A}\ket*{\psi_{A,m_A}},~c_{m_A}=\pm 1,\\
&\ket*{\widetilde{\varphi}^{\{\mu_{\boldsymbol{x}}\}}_{\overline{A}, \boldsymbol{n}'\boldsymbol{l}'}}=\frac{1}{\sqrt{\mathscr{N}_{\overline{A}}}}\sum_{m_{\overline{A}}}c_{m_{\overline{A}}}\ket*{\psi_{{\overline{A}},m_{\overline{A}}}},~c_{m_{\overline{A}}}=\pm 1.
\end{split}
\end{align}
(2) A $\ket*{\psi_{A,m_A}}$ contributes to the sum of $\ket*{\widetilde{\varphi}^{\{\mu_{\boldsymbol{x}}\}}_{A, \boldsymbol{n}\boldsymbol{l}}}$ if and only if it is torn from certain $\ket{\psi_m}$ state off which we can read the emergent bulk degrees of freedom and the finer description of parities specified by $\ket*{\boldsymbol{B}^{\{\mu_{\boldsymbol{x}}\}}_{a,\boldsymbol{n}}}$ and $\ket*{\boldsymbol{Z}^{\{\mu_{\boldsymbol{x}}\}}_{a,\boldsymbol{l}}}$. More explicitly, in $R^+\ket{\psi_m}$, $\boldsymbol{\beta}_{\boldsymbol{x}}$ with $\boldsymbol{x}\in\mathrm{W}[A]$ equals the $\boldsymbol{\beta}_{\boldsymbol{x}}$ in $\ket*{\boldsymbol{B}^{\{\mu_{\boldsymbol{x}}\}}_{a,\boldsymbol{n}}}$ (see Eq.~\ref{sign0}), and the finer descriptions of the parities on the torn halves as specified by $(\boldsymbol{\beta}^{\mu_{\boldsymbol{x}_1}}_{\boldsymbol{x}_1a}\boldsymbol{\beta}^{\mu_{\boldsymbol{x}_7}}_{\boldsymbol{x}_7a}\zeta_{\boldsymbol{x}_1\boldsymbol{x}_7a})(\boldsymbol{\beta}^{\mu_{\boldsymbol{x}_1}}_{\boldsymbol{x}_2a}\zeta_{\boldsymbol{x}_2a})\cdots(\boldsymbol{\beta}^{\mu_{\boldsymbol{x}_7}}_{\boldsymbol{x}_6a}\zeta_{\boldsymbol{x}_6a})$ ($\ket*{\boldsymbol{B}^{\{\mu_{\boldsymbol{x}}\}}_{a,\boldsymbol{n}}}$) in Fig.~\ref{app4e} can be read off the pictorial representation of $\ket{\psi_m}$. (3) In reading the finer descriptions off $\ket*{\psi_{A,m_A}}$, for $\boldsymbol{x}_4$ with two possible finer description for the torn half, $\ket*{\psi_{A,m_A}}$ can match either one: if it matches the one with $``-''$ sign, the coefficient $c_{m_A}$ in Eq.~\ref{sign1} is $-1$, otherwise, the coefficient is $1$. (4) Similar rules apply to $\ket*{\psi_{\overline{A},m_{\overline{A}}}}$.

In terms of the above definition, with reference to Fig.~\ref{app4e}, it is easy to see the following properties of $\ket*{\widetilde{\varphi}^{\{\mu_{\boldsymbol{x}}\}}_{A, \boldsymbol{n}\boldsymbol{l}}}$ and $\ket*{\widetilde{\varphi}^{\{\mu_{\boldsymbol{x}}\}}_{\overline{A}, \boldsymbol{n}'\boldsymbol{l}'}}$.

\emph{i)} If we fix $\{\mu_{\boldsymbol{x}}\}=\mu_{\boldsymbol{x}_1},\mu_{\boldsymbol{x}_2},\ldots,\mu_{\boldsymbol{x}_7}$, then, for distinct $(\ket*{\boldsymbol{B}^{\{\mu_{\boldsymbol{x}}\}}_{a,\boldsymbol{n}_1}},\ket*{\boldsymbol{Z}^{\{\mu_{\boldsymbol{x}}\}}_{a,\boldsymbol{l}_1}})$ and $(\ket*{\boldsymbol{B}^{\{\mu_{\boldsymbol{x}}\}}_{a,\boldsymbol{n}_2}},\ket*{\boldsymbol{Z}^{\{\mu_{\boldsymbol{x}}\}}_{a,\boldsymbol{l}_2}})$, i.e., with difference in either $\boldsymbol{\beta}_{\boldsymbol{x}}$ for $\boldsymbol{x}\in\mathrm{W}[A]$, or in $\boldsymbol{\beta}^{\mu_{\boldsymbol{x}}}_{\boldsymbol{x}a}$, or in $\zeta_{\boldsymbol{x}a}$, the $\ket*{\psi_{A,m_A}}$s specified respectively according to the above rules must be orthogonal, and hence we have $\braket*{\widetilde{\varphi}^{\{\mu_{\boldsymbol{x}}\}}_{A, \boldsymbol{n}_1\boldsymbol{l}_1}}{\widetilde{\varphi}^{\{\mu_{\boldsymbol{x}}\}}_{A, \boldsymbol{n}_2\boldsymbol{l}_2}}=0$. This property also applies to $\ket*{\psi_{\overline{A},m_{\overline{A}}}}$ and $\ket*{\widetilde{\varphi}^{\{\mu_{\boldsymbol{x}}\}}_{\overline{A}, \boldsymbol{n}'\boldsymbol{l}'}}$.

\emph{ii)} The total number of the $\ket*{\psi_{A,m_A}}$s specified by $(\ket*{\boldsymbol{B}^{\{\mu_{\boldsymbol{x}}\}}_{a,\boldsymbol{n}_1}},\ket*{\boldsymbol{Z}^{\{\mu_{\boldsymbol{x}}\}}_{a,\boldsymbol{l}_1}})$ or by $(\ket*{\boldsymbol{B}^{\{\mu_{\boldsymbol{x}}\}}_{a,\boldsymbol{n}_2}},\ket*{\boldsymbol{Z}^{\{\mu_{\boldsymbol{x}}\}}_{a,\boldsymbol{l}_2}})$ is the same, i.e., $\mathscr{N}_A$ is independent on $(\ket*{\boldsymbol{B}^{\{\mu_{\boldsymbol{x}}\}}_{a,\boldsymbol{n}}},\ket*{\boldsymbol{Z}^{\{\mu_{\boldsymbol{x}}\}}_{a,\boldsymbol{l}}})$. Moreover, for fixed $(\ket*{\boldsymbol{B}^{\{\mu_{\boldsymbol{x}}\}}_{a,\boldsymbol{n}}},\ket*{\boldsymbol{Z}^{\{\mu_{\boldsymbol{x}}\}}_{a,\boldsymbol{l}}})$, the total number of the $\ket*{\psi_{A,m_A}}$s with $c_{m_A}=1$ (corresponding to one possibility of finer description on $\boldsymbol{x}_4$ as shown in Fig.~\ref{app4e}) equals the total number for $c_{m_A}=-1$ (corresponding to the other possibility). The reason is that these different collections of $\ket*{\psi_{A,m_A}}$s corresponding to the difference in either $\boldsymbol{\beta}_{\boldsymbol{x}}$ for $\boldsymbol{x}\in\mathrm{W}[A]$, or in $\boldsymbol{\beta}^{\mu_{\boldsymbol{x}}}_{\boldsymbol{x}a}$, or in $\zeta_{\boldsymbol{x}a}$ can always be mapped from one to another through sequentially applying the $T_{ii'}$ gates, which is unitary and hence keeps the number invariant. These properties also apply to $\ket*{\psi_{\overline{A},m_{\overline{A}}}}$ and $\ket*{\widetilde{\varphi}^{\{\mu_{\boldsymbol{x}}\}}_{\overline{A}, \boldsymbol{n}'\boldsymbol{l}'}}$.

\emph{iii)} In terms of the above properties, it is easy to show that for distinct $\{\mu_{\boldsymbol{x}}\}$ and $\{\mu'_{\boldsymbol{x}}\}$, we have $\braket{\widetilde{\varphi}^{\{\mu_{\boldsymbol{x}}\}}_{A, \boldsymbol{n}\boldsymbol{l}}}{\widetilde{\varphi}^{\{\mu'_{\boldsymbol{x}}\}}_{A, \bar{\boldsymbol{n}}\bar{\boldsymbol{l}}}}=0$ and $\braket{\widetilde{\varphi}^{\{\mu_{\boldsymbol{x}}\}}_{\overline{A}, \boldsymbol{n}'\boldsymbol{l}'}}{\widetilde{\varphi}^{\{\mu'_{\boldsymbol{x}}\}}_{\overline{A}, \bar{\boldsymbol{n}'}\bar{\boldsymbol{l}'}}}=0$.

In this way, it is clear that we have specified isometric $\boldsymbol{U}^{\{\mu_{\boldsymbol{x}}\}}_A$ as $\boldsymbol{U}^{\{\mu_{\boldsymbol{x}}\}}_A(\ket*{\boldsymbol{B}^{\{\mu_{\boldsymbol{x}}\}}_{a,\boldsymbol{n}}}\otimes\ket*{\boldsymbol{Z}^{\{\mu_{\boldsymbol{x}}\}}_{a,\boldsymbol{l}}})=\ket*{\widetilde{\varphi}^{\{\mu_{\boldsymbol{x}}\}}_{A, \boldsymbol{n}\boldsymbol{l}}}$ and isometric $\boldsymbol{U}^{\{\mu_{\boldsymbol{x}}\}}_{\overline{A}}$ as $\boldsymbol{U}^{\{\mu_{\boldsymbol{x}}\}}_{\overline{A}}(\ket*{\boldsymbol{B}^{\{\mu_{\boldsymbol{x}}\}}_{{\overline{a}},\boldsymbol{n}'}}\otimes\ket*{\boldsymbol{Z}^{\{\mu_{\boldsymbol{x}}\}}_{{\overline{a}},\boldsymbol{l}'}})=\ket*{\widetilde{\varphi}^{\{\mu_{\boldsymbol{x}}\}}_{{\overline{A}}, \boldsymbol{n}'\boldsymbol{l}'}}$ with ${\boldsymbol{U}^{\{\mu'_{\boldsymbol{x}}\}}_A}^+\boldsymbol{U}^{\{\mu_{\boldsymbol{x}}\}}_A=0$ and ${\boldsymbol{U}^{\{\mu'_{\boldsymbol{x}}\}}_{\overline{A}}}^+\boldsymbol{U}^{\{\mu_{\boldsymbol{x}}\}}_{\overline{A}}=0$ for distinct $\{\mu_{\boldsymbol{x}}\}$ and $\{\mu'_{\boldsymbol{x}}\}$.


\subsubsection{Specifying $\ket*{\chi^{\{\mu_{\pmb{x}}\}}}$, extracting the area operator}
We can infer a crucial property from the list in Fig.~\ref{app4e}, which indicates the way to define the isometries $\{\boldsymbol{J}^{\{\mu_{\pmb{x}}\}}\}$, or equivalently, the $\ket*{\chi^{\{\mu_{\pmb{x}}\}}}$s. That is, for any $(\boldsymbol{\beta}^{\mu_{\boldsymbol{x}_4}}_{\boldsymbol{x}_4 a},\boldsymbol{\beta}^{\mu_{\boldsymbol{x}_4}}_{\boldsymbol{x}_4\bar a})$ identified as $\boldsymbol{\beta}_{\boldsymbol{x}_4}\pm\boldsymbol{\beta}'_{\boldsymbol{x}_4}$, if we combine the finer descriptions of the parities for the torn two halves as specified by $[(\boldsymbol{\beta}^{\mu_{\boldsymbol{x}_4}}_{\boldsymbol{x}_4 a},even),(even\boldsymbol{\beta}^{\mu_{\boldsymbol{x}_4}}_{\boldsymbol{x}_4\bar a})]$ together with $[(\boldsymbol{\beta}^{\mu_{\boldsymbol{x}_4}}_{\boldsymbol{x}_4 a},odd),(odd,\boldsymbol{\beta}^{\mu_{\boldsymbol{x}_4}}_{\boldsymbol{x}_4\bar a})]$, i.e., keeping the parities matched (see Fig.~\ref{app4c}), then, as shown in Fig.~\ref{app4f}, the combined finer descriptions exactly give rise to all the possibilities of the parities in the hole corresponding to $\boldsymbol{\beta}_{\boldsymbol{x}_4}$ and $\boldsymbol{\beta}'_{\boldsymbol{x}_4}$. Furthermore, as shown in Fig.~\ref{app4f}, the combination of the sign $\pm$ in the conventional way exactly gives rise to the correct sign in $\boldsymbol{\beta}_{\boldsymbol{x}_4}\pm\boldsymbol{\beta}'_{\boldsymbol{x}_4}$. It is also clear from Fig.~\ref{app4e} that for other $\boldsymbol{x}\ne\boldsymbol{x}_4$ in the entangling surface, the combined finer descriptions as specified by $[(\boldsymbol{\beta}^{\mu_{\boldsymbol{x}}}_{\boldsymbol{x} a},even),(even\boldsymbol{\beta}^{\mu_{\boldsymbol{x}}}_{\boldsymbol{x}\bar a})]$ together with $[(\boldsymbol{\beta}^{\mu_{\boldsymbol{x}}}_{\boldsymbol{x} a},odd),(odd,\boldsymbol{\beta}^{\mu_{\boldsymbol{x}}}_{\boldsymbol{x}\bar a})]$ exactly give rise to all the possibilities of the parities in the hole corresponding to the identified $\boldsymbol{\beta}_{\boldsymbol{x}}$ by $(\boldsymbol{\beta}^{\mu_{\boldsymbol{x}}}_{\boldsymbol{x} a},\boldsymbol{\beta}^{\mu_{\boldsymbol{x}}}_{\boldsymbol{x}\bar a})$.

According to this observation, for each $\{\mu_{\boldsymbol{x}}\}$ we can write
\begin{align}
\begin{split}
&\mathcal{F}^{\{\mu_{\boldsymbol{x}}\}}_a\otimes\mathcal{F}^{\{\mu_{\boldsymbol{x}}\}}_{\overline{a}}\\
&=(\mathfrak{f}_{\boldsymbol{x}_1,\boldsymbol{x}_7a}\otimes\cdots\otimes\mathfrak{f}_{\boldsymbol{x}_6a})\otimes(\mathfrak{f}_{\boldsymbol{x}_1,\boldsymbol{x}_7{\overline{a}}}\otimes\cdots\otimes\mathfrak{f}_{\boldsymbol{x}_6{\overline{a}}})\\
&=(\mathfrak{f}_{\boldsymbol{x}_1,\boldsymbol{x}_7a}\otimes\mathfrak{f}_{\boldsymbol{x}_1,\boldsymbol{x}_7\overline{a}})\\
&\otimes(\mathfrak{f}_{\boldsymbol{x}_2a}\otimes\mathfrak{f}_{\boldsymbol{x}_2\overline{a}})\otimes(\mathfrak{f}_{\boldsymbol{x}_3a}\otimes\mathfrak{f}_{\boldsymbol{x}_3\overline{a}})\\
&\otimes(\mathfrak{f}_{\boldsymbol{x}_4a}\otimes\mathfrak{f}_{\boldsymbol{x}_4\overline{a}})\otimes(\mathfrak{f}_{\boldsymbol{x}_5a}\otimes\mathfrak{f}_{\boldsymbol{x}_5\overline{a}})\\
&\otimes(\mathfrak{f}_{\boldsymbol{x}_6a}\otimes\mathfrak{f}_{\boldsymbol{x}_6\overline{a}})
\end{split}
\end{align}
and define 
\begin{align}\label{chichi}
\begin{split}
&\ket*{\chi_{\boldsymbol{x}_1\boldsymbol{x}_7}}=\frac{1}{\sqrt{2}}(\ket{even}\otimes\ket{even}+\ket{odd}\otimes\ket{odd}),\\
&\ket*{\chi_{\boldsymbol{x}_2}}=\frac{1}{\sqrt{2}}(\ket{even}\otimes\ket{even}+\ket{odd}\otimes\ket{odd}),\\
&\ket*{\chi_{\boldsymbol{x}_3}}=\frac{1}{\sqrt{2}}(\ket{even}\otimes\ket{even}+\ket{odd}\otimes\ket{odd}),\\
&\ket*{\chi_{\boldsymbol{x}_4}}=\frac{1}{\sqrt{2}}(\ket{even}\otimes\ket{even}+\ket{odd}\otimes\ket{odd}),\\
&\ket*{\chi_{\boldsymbol{x}_5}}=\frac{1}{\sqrt{2}}(\ket{even}\otimes\ket{even}+\ket{odd}\otimes\ket{odd}),\\
&\ket*{\chi_{\boldsymbol{x}_6}}=\frac{1}{\sqrt{2}}(\ket{even}\otimes\ket{even}+\ket{odd}\otimes\ket{odd}),\\
&\ket*{\chi^{\{\mu_{\boldsymbol{x}}\}}}=\frac{1}{\sqrt{2^6}}\ket*{\chi_{\boldsymbol{x}_1\boldsymbol{x}_7}}\otimes\ket*{\chi_{\boldsymbol{x}_2}}\otimes\cdots\otimes\ket*{\chi_{\boldsymbol{x}_6}}\\
&=\frac{1}{\sqrt{2^6}}(\ket{even,even,\ldots}\otimes\ket{even,even,\ldots}\\
&+\ket{even,odd,\ldots}\otimes\ket{even,odd,\ldots}\\
&+\cdots+\ket{odd,odd,\ldots}\otimes\ket{odd,odd,\ldots})\\
&=\frac{1}{\sqrt{2^6}}\sum_{\boldsymbol{l}}\ket*{\boldsymbol{Z}^{\{\mu_{\boldsymbol{x}}\}}_{a,\boldsymbol{l}}}\otimes\ket*{\boldsymbol{Z}^{\{\mu_{\boldsymbol{x}}\}}_{\bar a,\boldsymbol{l}}}.
\end{split}
\end{align}

Then, for any $\ket*{\boldsymbol{B}^{\{\mu_{\boldsymbol{x}}\}}_{a,\boldsymbol{n}}}$ and $\ket*{\boldsymbol{B}^{\{\mu_{\boldsymbol{x}}\}}_{\bar a,\boldsymbol{n}'}}$ we have
\begin{align}
\begin{split}
&\ket*{\boldsymbol{B}^{\{\mu_{\boldsymbol{x}}\}}_{a,\boldsymbol{n}}}\otimes\ket*{\chi^{\{\mu_{\boldsymbol{x}}\}}}\otimes\ket*{\boldsymbol{B}^{\{\mu_{\boldsymbol{x}}\}}_{\bar a,\boldsymbol{n}'}}\\
&=\frac{1}{\sqrt{2^6}}\sum_{\boldsymbol{l}}\ket*{\boldsymbol{B}^{\{\mu_{\boldsymbol{x}}\}}_{a,\boldsymbol{n}}}\otimes\ket*{\boldsymbol{Z}^{\{\mu_{\boldsymbol{x}}\}}_{a,\boldsymbol{l}}}\otimes\ket*{\boldsymbol{Z}^{\{\mu_{\boldsymbol{x}}\}}_{\bar a,\boldsymbol{l}}}\otimes\ket*{\boldsymbol{B}^{\{\mu_{\boldsymbol{x}}\}}_{\bar a,\boldsymbol{n}'}},
\end{split}
\end{align}
and hence we have
\begin{align}
\begin{split}
&\boldsymbol{U}^{\{\mu_{\boldsymbol{x}}\}}_A\otimes\boldsymbol{U}^{\{\mu_{\boldsymbol{x}}\}}_{\overline{A}}(\ket*{\boldsymbol{B}^{\{\mu_{\boldsymbol{x}}\}}_{a,\boldsymbol{n}}}\otimes\ket*{\chi^{\{\mu_{\boldsymbol{x}}\}}}\otimes\ket*{\boldsymbol{B}^{\{\mu_{\boldsymbol{x}}\}}_{\bar a,\boldsymbol{n}'}})\\
&=\frac{1}{\sqrt{2^6}}\sum_{\boldsymbol{l}}\boldsymbol{U}^{\{\mu_{\boldsymbol{x}}\}}_A(\ket*{\boldsymbol{B}^{\{\mu_{\boldsymbol{x}}\}}_{a,\boldsymbol{n}}}\otimes\ket*{\boldsymbol{Z}^{\{\mu_{\boldsymbol{x}}\}}_{a,\boldsymbol{l}}})\\
&\otimes\boldsymbol{U}^{\{\mu_{\boldsymbol{x}}\}}_{\overline{A}}(\ket*{\boldsymbol{Z}^{\{\mu_{\boldsymbol{x}}\}}_{\bar a,\boldsymbol{l}}}\otimes\ket*{\boldsymbol{B}^{\{\mu_{\boldsymbol{x}}\}}_{\bar a,\boldsymbol{n}'}})\\
&=\frac{1}{\sqrt{2^6}}\sum_{\boldsymbol{l}}\ket*{\widetilde{\varphi}^{\{\mu_{\boldsymbol{x}}\}}_{A, \boldsymbol{n}\boldsymbol{l}}}\otimes\ket*{\widetilde{\varphi}^{\{\mu_{\boldsymbol{x}}\}}_{\overline{A}, \boldsymbol{n}'\boldsymbol{l}}}.
\end{split}
\end{align}

Note that here the index $\boldsymbol{l}$ goes through all possible configurations of the parities on the torn halves on $A$ (and also equivalent to those on $\overline{A}$), i.e., $(even,even,\ldots),(even,odd,\ldots),\ldots$. For each specific $\boldsymbol{l}$, e.g., $(even,even,\ldots)$, the states $\ket*{\widetilde{\varphi}^{\{\mu_{\boldsymbol{x}}\}}_{A, \boldsymbol{n}\boldsymbol{l}}}=\boldsymbol{U}^{\{\mu_{\boldsymbol{x}}\}}_A(\ket*{\boldsymbol{B}^{\{\mu_{\boldsymbol{x}}\}}_{a,\boldsymbol{n}}}\otimes\ket*{\boldsymbol{Z}^{\{\mu_{\boldsymbol{x}}\}}_{a,\boldsymbol{l}}})$ and $\ket*{\widetilde{\varphi}^{\{\mu_{\boldsymbol{x}}\}}_{\overline{A}, \boldsymbol{n}'\boldsymbol{l}}}=\boldsymbol{U}^{\{\mu_{\boldsymbol{x}}\}}_A(\ket*{\boldsymbol{B}^{\{\mu_{\boldsymbol{x}}\}}_{a,\boldsymbol{n}}}\otimes\ket*{\boldsymbol{Z}^{\{\mu_{\boldsymbol{x}}\}}_{a,\boldsymbol{l}}})$ are exactly determined by
\begin{align}
\begin{split}
&(\otimes_{\boldsymbol{x}\in\mathrm{W}[A]}\ket{\boldsymbol{\beta}_{\boldsymbol{x}}})\otimes\ket*{\boldsymbol{\beta}^{\mu_{\boldsymbol{x}_1}}_{\boldsymbol{x}_1a}\cdots\boldsymbol{\beta}^{\mu_{\boldsymbol{x}_7}}_{\boldsymbol{x}_7a}}\otimes\ket{even,even,\ldots}\\
&=(\otimes_{\boldsymbol{x}\in\mathrm{W}[A]}\ket{\boldsymbol{\beta}_{\boldsymbol{x}}})\\
&\otimes\ket*{(\boldsymbol{\beta}^{\mu_{\boldsymbol{x}_1}}_{\boldsymbol{x}_1a}\boldsymbol{\beta}^{\mu_{\boldsymbol{x}_7}}_{\boldsymbol{x}_7a}even)(\boldsymbol{\beta}^{\mu_{\boldsymbol{x}_2}}_{\boldsymbol{x}_2a}even)\cdots(\boldsymbol{\beta}^{\mu_{\boldsymbol{x}_6}}_{\boldsymbol{x}_6a}even)}
\end{split}
\end{align}
and 
\begin{align}
\begin{split}
&\ket{even,even,\ldots}\otimes\ket*{\boldsymbol{\beta}^{\mu_{\boldsymbol{x}_1}}_{\boldsymbol{x}_1\overline{a}}\cdots\boldsymbol{\beta}^{\mu_{\boldsymbol{x}_7}}_{\boldsymbol{x}_7\overline{a}}}\otimes(\otimes_{\boldsymbol{x}\in\mathrm{W}[\overline{A}]}\ket{\boldsymbol{\beta}_{\boldsymbol{x}}})\\
&=\ket*{(even\boldsymbol{\beta}^{\mu_{\boldsymbol{x}_1}}_{\boldsymbol{x}_1\overline{a}}\boldsymbol{\beta}^{\mu_{\boldsymbol{x}_7}}_{\boldsymbol{x}_7\overline{a}})(even\boldsymbol{\beta}^{\mu_{\boldsymbol{x}_2}}_{\boldsymbol{x}_2\overline{a}})\cdots(even\boldsymbol{\beta}^{\mu_{\boldsymbol{x}_6}}_{\boldsymbol{x}_6\overline{a}})}\\
&\otimes(\otimes_{\boldsymbol{x}\in\mathrm{W}[\overline{A}]}\ket{\boldsymbol{\beta}_{\boldsymbol{x}}})
\end{split}
\end{align}
through the list in Fig.~\ref{app4e}. And hence the product $\ket*{\widetilde{\varphi}^{\{\mu_{\boldsymbol{x}}\}}_{A, \boldsymbol{n}\boldsymbol{l}}}\otimes\ket*{\widetilde{\varphi}^{\{\mu_{\boldsymbol{x}}\}}_{\overline{A}, \boldsymbol{n}'\boldsymbol{l}}}$, as a sum, consists of the states $\frac{1}{\sqrt{\mathscr{N}_A \mathscr{N}_{\overline{A}}}}c_{m_A}c_{m_{\overline{A}}}\ket*{\psi_{A,m_A}}\otimes\ket*{\psi_{{\overline{A}},m_{\overline{A}}}}$ in which the concatenation is characterized by
\begin{align}
\begin{split}
&\boldsymbol{\beta}_{\boldsymbol{1}},\ldots,\\
&(\boldsymbol{\beta}^{\mu_{\boldsymbol{x}_1}}_{\boldsymbol{x}_1a}\boldsymbol{\beta}^{\mu_{\boldsymbol{x}_7}}_{\boldsymbol{x}_7a}eveneven\boldsymbol{\beta}^{\mu_{\boldsymbol{x}_1}}_{\boldsymbol{x}_1\overline{a}}\boldsymbol{\beta}^{\mu_{\boldsymbol{x}_7}}_{\boldsymbol{x}_7\overline{a}}),\\
&(\boldsymbol{\beta}^{\mu_{\boldsymbol{x}_2}}_{\boldsymbol{x}_2a}eveneven\boldsymbol{\beta}^{\mu_{\boldsymbol{x}_2}}_{\boldsymbol{x}_2\overline{a}}),\ldots,(\boldsymbol{\beta}^{\mu_{\boldsymbol{x}_6}}_{\boldsymbol{x}_6a}eveneven\boldsymbol{\beta}^{\mu_{\boldsymbol{x}_6}}_{\boldsymbol{x}_6\overline{a}}),\\
&\ldots,\boldsymbol{\beta}_{K}.
\end{split}
\end{align}
Then, according to the previous discussion on Fig.~\ref{app4f}, $\ket*{\psi_{A,m_A}}\otimes\ket*{\psi_{{\overline{A}},m_{\overline{A}}}}$ must be certain $\ket{\psi_m}$ that corresponds (through $R^+$) to $\ket{\boldsymbol{\beta}_{\boldsymbol{1}}\cdots\boldsymbol{\beta}_{\boldsymbol{x}_1}\boldsymbol{\beta}_{\boldsymbol{x}_2}\boldsymbol{\beta}_{\boldsymbol{x}_3}\boldsymbol{\beta}_{\boldsymbol{x}_4}\boldsymbol{\beta}_{\boldsymbol{x}_5}\boldsymbol{\beta}_{\boldsymbol{x}_6}\boldsymbol{\beta}_{\boldsymbol{x}_7}\cdots\boldsymbol{\beta}_{K}}$ or $\ket{\boldsymbol{\beta}_{\boldsymbol{1}}\cdots\boldsymbol{\beta}_{\boldsymbol{x}_1}\boldsymbol{\beta}_{\boldsymbol{x}_2}\boldsymbol{\beta}_{\boldsymbol{x}_3}\pm\boldsymbol{\beta}'_{\boldsymbol{x}_4}\boldsymbol{\beta}_{\boldsymbol{x}_5}\boldsymbol{\beta}_{\boldsymbol{x}_6}\boldsymbol{\beta}_{\boldsymbol{x}_7}\cdots\boldsymbol{\beta}_{K}}$ in
\begin{align}
\begin{split}
&(\otimes_{\boldsymbol{x}\in\mathrm{W}[A]}\ket{\boldsymbol{\beta}_{\boldsymbol{x}}})\\
&\otimes\ket*{\boldsymbol{\beta}^{\mu_{\boldsymbol{x}_1}}_{\boldsymbol{x}_1a}\boldsymbol{\beta}^{\mu_{\boldsymbol{x}_1}}_{\boldsymbol{x}_1{\overline{a}}}}\otimes\cdots\otimes\ket*{\boldsymbol{\beta}^{\mu_{\boldsymbol{x}_7}}_{\boldsymbol{x}_7a}\boldsymbol{\beta}^{\mu_{\boldsymbol{x}_7}}_{\boldsymbol{x}_7{\overline{a}}}}\\
&\otimes(\otimes_{\boldsymbol{x}\in\mathrm{W}[\overline{A}]}\ket{\boldsymbol{\beta}_{\boldsymbol{x}}})\\
&=\frac{1}{\sqrt{2}}\ket{\boldsymbol{\beta}_{\boldsymbol{1}}\cdots\boldsymbol{\beta}_{\boldsymbol{x}_1}\boldsymbol{\beta}_{\boldsymbol{x}_2}\boldsymbol{\beta}_{\boldsymbol{x}_3}(\boldsymbol{\beta}_{\boldsymbol{x}_4}\pm\boldsymbol{\beta}'_{\boldsymbol{x}_4})\boldsymbol{\beta}_{\boldsymbol{x}_5}\boldsymbol{\beta}_{\boldsymbol{x}_6}\boldsymbol{\beta}_{\boldsymbol{x}_7}\cdots\boldsymbol{\beta}_{K}},
\end{split}
\end{align}
with $c_{m_A}c_{m_{\overline{A}}}$ determines the correct sign.

Then, consider the whole expansion $\frac{1}{\sqrt{2^6}}\sum_{\boldsymbol{l}}\ket*{\widetilde{\varphi}^{\{\mu_{\boldsymbol{x}}\}}_{A, \boldsymbol{n}\boldsymbol{l}}}\otimes\ket*{\widetilde{\varphi}^{\{\mu_{\boldsymbol{x}}\}}_{\overline{A}, \boldsymbol{n}'\boldsymbol{l}}}$ in which the $\frac{1}{\sqrt{\mathscr{N}_A \mathscr{N}_{\overline{A}}}}c_{m_A}c_{m_{\overline{A}}}\ket*{\psi_{A,m_A}}\otimes\ket*{\psi_{{\overline{A}},m_{\overline{A}}}}$ are characterized by
\begin{align}
\begin{split}
&\boldsymbol{\beta}_{\boldsymbol{1}},\ldots,\\
&(\boldsymbol{\beta}^{\mu_{\boldsymbol{x}_1}}_{\boldsymbol{x}_1a}\boldsymbol{\beta}^{\mu_{\boldsymbol{x}_7}}_{\boldsymbol{x}_7a}eveneven\boldsymbol{\beta}^{\mu_{\boldsymbol{x}_1}}_{\boldsymbol{x}_1\overline{a}}\boldsymbol{\beta}^{\mu_{\boldsymbol{x}_7}}_{\boldsymbol{x}_7\overline{a}}),\\
&(\boldsymbol{\beta}^{\mu_{\boldsymbol{x}_2}}_{\boldsymbol{x}_2a}eveneven\boldsymbol{\beta}^{\mu_{\boldsymbol{x}_2}}_{\boldsymbol{x}_2\overline{a}}),\ldots,(\boldsymbol{\beta}^{\mu_{\boldsymbol{x}_6}}_{\boldsymbol{x}_6a}eveneven\boldsymbol{\beta}^{\mu_{\boldsymbol{x}_6}}_{\boldsymbol{x}_6\overline{a}}),\\
&\ldots,\boldsymbol{\beta}_{K},\\
&or\\
&\boldsymbol{\beta}_{\boldsymbol{1}},\ldots,\\
&(\boldsymbol{\beta}^{\mu_{\boldsymbol{x}_1}}_{\boldsymbol{x}_1a}\boldsymbol{\beta}^{\mu_{\boldsymbol{x}_7}}_{\boldsymbol{x}_7a}eveneven\boldsymbol{\beta}^{\mu_{\boldsymbol{x}_1}}_{\boldsymbol{x}_1\overline{a}}\boldsymbol{\beta}^{\mu_{\boldsymbol{x}_7}}_{\boldsymbol{x}_7\overline{a}}),\\
&(\boldsymbol{\beta}^{\mu_{\boldsymbol{x}_2}}_{\boldsymbol{x}_2a}eveneven\boldsymbol{\beta}^{\mu_{\boldsymbol{x}_2}}_{\boldsymbol{x}_2\overline{a}}),\ldots,(\boldsymbol{\beta}^{\mu_{\boldsymbol{x}_6}}_{\boldsymbol{x}_6a}oddodd\boldsymbol{\beta}^{\mu_{\boldsymbol{x}_6}}_{\boldsymbol{x}_6\overline{a}}),\\
&\ldots,\boldsymbol{\beta}_{K},\\
&or\\
&\boldsymbol{\beta}_{\boldsymbol{1}},\ldots,\\
&(\boldsymbol{\beta}^{\mu_{\boldsymbol{x}_1}}_{\boldsymbol{x}_1a}\boldsymbol{\beta}^{\mu_{\boldsymbol{x}_7}}_{\boldsymbol{x}_7a}eveneven\boldsymbol{\beta}^{\mu_{\boldsymbol{x}_1}}_{\boldsymbol{x}_1\overline{a}}\boldsymbol{\beta}^{\mu_{\boldsymbol{x}_7}}_{\boldsymbol{x}_7\overline{a}}),\\
&(\boldsymbol{\beta}^{\mu_{\boldsymbol{x}_2}}_{\boldsymbol{x}_2a}oddodd\boldsymbol{\beta}^{\mu_{\boldsymbol{x}_2}}_{\boldsymbol{x}_2\overline{a}}),\ldots,(\boldsymbol{\beta}^{\mu_{\boldsymbol{x}_6}}_{\boldsymbol{x}_6a}oddodd\boldsymbol{\beta}^{\mu_{\boldsymbol{x}_6}}_{\boldsymbol{x}_6\overline{a}}),\\
&\ldots,\boldsymbol{\beta}_{K},\\
&or\\
&\ldots
\end{split}
\end{align}
going through all the $2^6$ possibilities. We can conclude that 
\begin{align}
\begin{split}
&\boldsymbol{U}^{\{\mu_{\boldsymbol{x}}\}}_A\otimes\boldsymbol{U}^{\{\mu_{\boldsymbol{x}}\}}_{\overline{A}}(\ket*{\boldsymbol{B}^{\{\mu_{\boldsymbol{x}}\}}_{a,\boldsymbol{n}}}\otimes\ket*{\chi^{\{\mu_{\boldsymbol{x}}\}}}\otimes\ket*{\boldsymbol{B}^{\{\mu_{\boldsymbol{x}}\}}_{\bar a,\boldsymbol{n}'}})\\
&=\frac{1}{\sqrt{2^6}}\sum_{\boldsymbol{l}}\ket*{\widetilde{\varphi}^{\{\mu_{\boldsymbol{x}}\}}_{A, \boldsymbol{n}\boldsymbol{l}}}\otimes\ket*{\widetilde{\varphi}^{\{\mu_{\boldsymbol{x}}\}}_{\overline{A}, \boldsymbol{n}'\boldsymbol{l}}}\\
&=\frac{1}{\sqrt{2}}R\ket{\boldsymbol{\beta}_{\boldsymbol{1}}\cdots\boldsymbol{\beta}_{\boldsymbol{x}_1}\boldsymbol{\beta}_{\boldsymbol{x}_2}\boldsymbol{\beta}_{\boldsymbol{x}_3}(\boldsymbol{\beta}_{\boldsymbol{x}_4}\pm\boldsymbol{\beta}'_{\boldsymbol{x}_4})\boldsymbol{\beta}_{\boldsymbol{x}_5}\boldsymbol{\beta}_{\boldsymbol{x}_6}\boldsymbol{\beta}_{\boldsymbol{x}_7}\cdots\boldsymbol{\beta}_{K}}\\
&=R\ket*{\boldsymbol{B}^{\{\mu_{\boldsymbol{x}}\}}_{a,\boldsymbol{n}}\boldsymbol{B}^{\{\mu_{\boldsymbol{x}}\}}_{{\overline{a}},\boldsymbol{n}'}},
\end{split}
\end{align}
which has qualified the decompositions $\{\boldsymbol{U}^{\{\mu_{\boldsymbol{x}}\}}_A\}$ and $\{\boldsymbol{U}^{\{\mu_{\boldsymbol{x}}\}}_{\overline{A}}\}$ and the isometries $\{\boldsymbol{J}^{\{\mu_{\boldsymbol{x}}\}}\}$ given by $\{\ket*{\chi^{\{\mu_{\boldsymbol{x}}\}}}\}$. Therefore, the
\begin{align}
\begin{split}
&\ket*{\chi^{\{\mu_{\boldsymbol{x}}\}}}=\frac{1}{\sqrt{2^6}}\ket*{\chi_{\boldsymbol{x}_1\boldsymbol{x}_7}}\otimes\ket*{\chi_{\boldsymbol{x}_2}}\otimes\cdots\otimes\ket*{\chi_{\boldsymbol{x}_6}}\\
&=\frac{1}{\sqrt{2^6}}(\ket{even,even,\ldots}\otimes\ket{even,even,\ldots}\\
&+\ket{even,odd,\ldots}\otimes\ket{even,odd,\ldots}\\
&+\cdots+\ket{odd,odd,\ldots}\otimes\ket{odd,odd,\ldots})
\end{split}
\end{align}
is the desired.

\section{Description of the quantum-circuits}\label{poscalable}
To describe the circuit introduced in Sec.~\ref{scalable} and illustrated in Fig.~\ref{fig10} in each step, we can rewrite Eq.~\ref{cbs1} so that the commutative $(1/\sqrt{2})(\mathds{1}+T(\boldsymbol{x}))$ terms can be reordered in alignment with the steps in Fig.~\ref{fig10}, i.e.,
\begin{align*}
\begin{split}
&\ket{\tilde{\varphi}_{n_0}}\\
&=\sqrt{2^{K}}\prod_{\boldsymbol{x}=1}^{K}(\frac{\mathds{1}+T(\boldsymbol{x})}{2})\ket{0\cdots0\cdots0}\\
&=\cdots\frac{1}{\sqrt{2}}(\mathds{1}+T(\boldsymbol{x}_2))\frac{1}{\sqrt{2}}(\mathds{1}+T(\boldsymbol{x}_1))\ket{0\cdots0\cdots0}.
\end{split}
\end{align*}

It is clear that the first term acting on $\ket{0\cdots0\cdots0}$ results in 
\begin{align*}
\begin{split}
&\frac{1}{\sqrt{2}}(\mathds{1}+T(\boldsymbol{x}_1))\ket{0\cdots0\cdots0}\\
&=\frac{1}{\sqrt{2}}(\ket{0\cdots0\cdots0}+\ket{0\cdots\alpha_{i}\alpha_{i_1}\alpha_{i_2}\alpha_{i_3}\alpha_{i_4}\alpha_{i_5}\cdots0}),\\
&\alpha_i=S^{\sigma_i}_i\ket{0},~\alpha_{i_1}=S^{\sigma_{i_1}}_{i_1}\ket{0},~\alpha_{i_2}=S^{\sigma_{i_2}}_{i_2}\ket{0},\ldots,
\end{split}
\end{align*}
where the superscripts $\sigma_{i},\sigma_{i_1},\sigma_{i_2},\ldots$ are determined by $T(\boldsymbol{x}_1)$. This effect is equivalent to mapping the violet-colored qudit $i$ from the state $\ket{0}$ to $(1/\sqrt{2})(\ket{0}+\ket{\alpha_i})$ through a single qudit gate, followed by applying controlled gates as shown in Fig.~\ref{fig10} where we view $i$ as the control qudit and the rest as the target qudits. The controlled gates can be written as
\begin{align*}
\begin{split}
&\dyad{0}\otimes\mathds{1}+\dyad{\alpha_i}\otimes S^{\sigma_{i_1}}_{i_1}+\dyad{\alpha'_i}\otimes\mathds{1}+\dyad{\alpha''_i}\otimes\mathds{1},\\
&\dyad{0}\otimes\mathds{1}+\dyad{\alpha_i}\otimes S^{\sigma_{i_2}}_{i_2}+\dyad{\alpha'_i}\otimes\mathds{1}+\dyad{\alpha''_i}\otimes\mathds{1},\\
&......,
\end{split}
\end{align*}
which simply apply the desired operators $S^{\sigma_{i_1}}_{i_1},S^{\sigma_{i_2}}_{i_2},\ldots$ to the target qudits when the control qudit $i$ is in $\ket{\alpha_i}$, while do nothing on the target qudits when the control qudit $i$ is in other orthogonal states $\ket{0}$, $\ket{\alpha'_i}$ or $\ket{\alpha''_i}$. Here, the single-qudit-gate on the control qudit (purple) $i$ contributes to the depiction in Step 1 of Fig.~\ref{fig10}, and these controlled gates are described by the circuit in Step 2.

After these operation, Eq.~\ref{cbs1} can be further rewritten as 
\begin{align*}
\begin{split}
&\ket{\tilde{\varphi}_{n_0}}\\
&=\sqrt{2^{K}}\prod_{\boldsymbol{x}=1}^{K}(\frac{\mathds{1}+T(\boldsymbol{x})}{2})\ket{0\cdots0\cdots0}\\
&=\cdots\frac{1}{\sqrt{2}}(\mathds{1}+T(\boldsymbol{x}_2))\frac{1}{\sqrt{2}}\ket{0\cdots0\cdots0}\\
&+\cdots\frac{1}{\sqrt{2}}(\mathds{1}+T(\boldsymbol{x}_2))\frac{1}{\sqrt{2}}\ket{0\cdots\alpha_{i}\alpha_{i_1}\alpha_{i_2}\alpha_{i_3}\alpha_{i_4}\alpha_{i_5}\cdots0},
\end{split}
\end{align*}
which implies that we can equivalently apply single-qudit and controlled gates in a similar way. The way to apply the gates to $\ket{0\cdots0\cdots0}$ is obvious. Indeed, according to the way we select the purple qudits, in the action of 
\begin{equation*}
(\mathds{1}+T(\boldsymbol{x}_2))\frac{1}{\sqrt{2}}\ket{0\cdots\alpha_{i}\alpha_{i_1}\alpha_{i_2}\alpha_{i_3}\alpha_{i_4}\alpha_{i_5}\cdots0}
\end{equation*}
the supporting qudit $j$ of $T(\boldsymbol{x}_2)$ (see Fig.~\ref{fig10}) lies outside the overlap with the previous step and is hence in the state $\ket{0}$ in $\ket{0\cdots\alpha_{i}\alpha_{i_1}\alpha_{i_2}\alpha_{i_3}\alpha_{i_4}\alpha_{i_5}\cdots0}$. Consequently, we can exactly apply a single-qudit gate to map $\ket{0}$ to $(1/\sqrt{2})(\ket{0}+\ket{\alpha_j})$ with $\ket{\alpha_j}$ determined by the action of $T(\boldsymbol{x}_2)$ on the qudit $j$, and then apply the controlled gates with $j$ as the control qudit and with $S^{\sigma_{j_1}}_{j_1}$ and $S^{\sigma_{j_2}}_{j_2}$ (also determined by $T(\boldsymbol{x}_2)$) as the desired operators on the rest qudits in the support of $T(\boldsymbol{x}_2)$.

Then, we can utilize similar gates in all the following steps. Note that since the single-qudit operators on the purple qudit always commute with the gates in the previous steps, we can reposition them all to the first step, as described in Step 1 in Fig.~\ref{fig10}. In this way, the steps of applying the circuits illustrated in Fig.~\ref{fig10} complete the action of $\sqrt{2^{K}}\prod_{\boldsymbol{x}=1}^{K}(\frac{\mathds{1}+T(\boldsymbol{x})}{2})$ on $\ket{0\cdots0\cdots0}$. And obviously, this description can be easily scaled up to arbitrary system size.

\end{CJK*}

\bibliography{ref.bib}

\begin{thebibliography}{113}%
\makeatletter
\providecommand \@ifxundefined [1]{%
 \@ifx{#1\undefined}
}%
\providecommand \@ifnum [1]{%
 \ifnum #1\expandafter \@firstoftwo
 \else \expandafter \@secondoftwo
 \fi
}%
\providecommand \@ifx [1]{%
 \ifx #1\expandafter \@firstoftwo
 \else \expandafter \@secondoftwo
 \fi
}%
\providecommand \natexlab [1]{#1}%
\providecommand \enquote  [1]{``#1''}%
\providecommand \bibnamefont  [1]{#1}%
\providecommand \bibfnamefont [1]{#1}%
\providecommand \citenamefont [1]{#1}%
\providecommand \href@noop [0]{\@secondoftwo}%
\providecommand \href [0]{\begingroup \@sanitize@url \@href}%
\providecommand \@href[1]{\@@startlink{#1}\@@href}%
\providecommand \@@href[1]{\endgroup#1\@@endlink}%
\providecommand \@sanitize@url [0]{\catcode `\\12\catcode `\$12\catcode
  `\&12\catcode `\#12\catcode `\^12\catcode `\_12\catcode `\%12\relax}%
\providecommand \@@startlink[1]{}%
\providecommand \@@endlink[0]{}%
\providecommand \url  [0]{\begingroup\@sanitize@url \@url }%
\providecommand \@url [1]{\endgroup\@href {#1}{\urlprefix }}%
\providecommand \urlprefix  [0]{URL }%
\providecommand \Eprint [0]{\href }%
\providecommand \doibase [0]{https://doi.org/}%
\providecommand \selectlanguage [0]{\@gobble}%
\providecommand \bibinfo  [0]{\@secondoftwo}%
\providecommand \bibfield  [0]{\@secondoftwo}%
\providecommand \translation [1]{[#1]}%
\providecommand \BibitemOpen [0]{}%
\providecommand \bibitemStop [0]{}%
\providecommand \bibitemNoStop [0]{.\EOS\space}%
\providecommand \EOS [0]{\spacefactor3000\relax}%
\providecommand \BibitemShut  [1]{\csname bibitem#1\endcsname}%
\let\auto@bib@innerbib\@empty
\bibitem [{\citenamefont {Almheiri}\ \emph {et~al.}(2015)\citenamefont
  {Almheiri}, \citenamefont {Dong},\ and\ \citenamefont
  {Harlow}}]{almheiri2015}%
  \BibitemOpen
  \bibfield  {author} {\bibinfo {author} {\bibfnamefont {A.}~\bibnamefont
  {Almheiri}}, \bibinfo {author} {\bibfnamefont {X.}~\bibnamefont {Dong}},\
  and\ \bibinfo {author} {\bibfnamefont {D.}~\bibnamefont {Harlow}},\
  }\bibfield  {title} {\bibinfo {title} {Bulk locality and quantum error
  correction in {{AdS}}/{{CFT}}},\ }\href
  {https://doi.org/10.1007/JHEP04(2015)163} {\bibfield  {journal} {\bibinfo
  {journal} {Journal of High Energy Physics}\ }\textbf {\bibinfo {volume}
  {2015}},\ \bibinfo {pages} {163} (\bibinfo {year} {2015})}\BibitemShut
  {NoStop}%
\bibitem [{\citenamefont {Pastawski}\ \emph {et~al.}(2015)\citenamefont
  {Pastawski}, \citenamefont {Yoshida}, \citenamefont {Harlow},\ and\
  \citenamefont {Preskill}}]{pastawski2015}%
  \BibitemOpen
  \bibfield  {author} {\bibinfo {author} {\bibfnamefont {F.}~\bibnamefont
  {Pastawski}}, \bibinfo {author} {\bibfnamefont {B.}~\bibnamefont {Yoshida}},
  \bibinfo {author} {\bibfnamefont {D.}~\bibnamefont {Harlow}},\ and\ \bibinfo
  {author} {\bibfnamefont {J.}~\bibnamefont {Preskill}},\ }\bibfield  {title}
  {\bibinfo {title} {Holographic quantum error-correcting codes: Toy models for
  the bulk/boundary correspondence},\ }\href
  {https://doi.org/10.1007/JHEP06(2015)149} {\bibfield  {journal} {\bibinfo
  {journal} {Journal of High Energy Physics}\ }\textbf {\bibinfo {volume}
  {06}},\ \bibinfo {pages} {149} (\bibinfo {year} {2015})}\BibitemShut
  {NoStop}%
\bibitem [{\citenamefont {Dong}\ \emph {et~al.}(2016)\citenamefont {Dong},
  \citenamefont {Harlow},\ and\ \citenamefont {Wall}}]{dong2016}%
  \BibitemOpen
  \bibfield  {author} {\bibinfo {author} {\bibfnamefont {X.}~\bibnamefont
  {Dong}}, \bibinfo {author} {\bibfnamefont {D.}~\bibnamefont {Harlow}},\ and\
  \bibinfo {author} {\bibfnamefont {A.~C.}\ \bibnamefont {Wall}},\ }\bibfield
  {title} {\bibinfo {title} {Reconstruction of {{Bulk Operators}} within the
  {{Entanglement Wedge}} in {{Gauge-Gravity Duality}}},\ }\href
  {https://doi.org/10.1103/PhysRevLett.117.021601} {\bibfield  {journal}
  {\bibinfo  {journal} {Physical Review Letters}\ }\textbf {\bibinfo {volume}
  {117}},\ \bibinfo {pages} {021601} (\bibinfo {year} {2016})}\BibitemShut
  {NoStop}%
\bibitem [{\citenamefont {Harlow}(2017)}]{harlow2017}%
  \BibitemOpen
  \bibfield  {author} {\bibinfo {author} {\bibfnamefont {D.}~\bibnamefont
  {Harlow}},\ }\bibfield  {title} {\bibinfo {title} {The {{Ryu}}--{{Takayanagi
  Formula}} from {{Quantum Error Correction}}},\ }\href
  {https://doi.org/10.1007/s00220-017-2904-z} {\bibfield  {journal} {\bibinfo
  {journal} {Communications in Mathematical Physics}\ }\textbf {\bibinfo
  {volume} {354}},\ \bibinfo {pages} {865} (\bibinfo {year}
  {2017})}\BibitemShut {NoStop}%
\bibitem [{\citenamefont {Witten}(1998)}]{witten1998}%
  \BibitemOpen
  \bibfield  {author} {\bibinfo {author} {\bibfnamefont {E.}~\bibnamefont
  {Witten}},\ }\bibfield  {title} {\bibinfo {title} {Anti de {{Sitter}} space
  and holography},\ }\href {https://doi.org/10.4310/ATMP.1998.v2.n2.a2}
  {\bibfield  {journal} {\bibinfo  {journal} {Advances in Theoretical and
  Mathematical Physics}\ }\textbf {\bibinfo {volume} {2}},\ \bibinfo {pages}
  {253} (\bibinfo {year} {1998})}\BibitemShut {NoStop}%
\bibitem [{\citenamefont {Maldacena}(1999)}]{maldacena1999}%
  \BibitemOpen
  \bibfield  {author} {\bibinfo {author} {\bibfnamefont {J.}~\bibnamefont
  {Maldacena}},\ }\bibfield  {title} {\bibinfo {title} {The {{Large-N Limit}}
  of {{Superconformal Field Theories}} and {{Supergravity}}},\ }\href
  {https://doi.org/10.1023/A:1026654312961} {\bibfield  {journal} {\bibinfo
  {journal} {International Journal of Theoretical Physics}\ }\textbf {\bibinfo
  {volume} {38}},\ \bibinfo {pages} {1113} (\bibinfo {year}
  {1999})}\BibitemShut {NoStop}%
\bibitem [{\citenamefont {Donnelly}\ \emph {et~al.}(2017)\citenamefont
  {Donnelly}, \citenamefont {Marolf}, \citenamefont {Michel},\ and\
  \citenamefont {Wien}}]{donnelly2017}%
  \BibitemOpen
  \bibfield  {author} {\bibinfo {author} {\bibfnamefont {W.}~\bibnamefont
  {Donnelly}}, \bibinfo {author} {\bibfnamefont {D.}~\bibnamefont {Marolf}},
  \bibinfo {author} {\bibfnamefont {B.}~\bibnamefont {Michel}},\ and\ \bibinfo
  {author} {\bibfnamefont {J.}~\bibnamefont {Wien}},\ }\bibfield  {title}
  {\bibinfo {title} {Living on the edge: A toy model for holographic
  reconstruction of algebras with centers},\ }\href
  {https://doi.org/10.1007/JHEP04(2017)093} {\bibfield  {journal} {\bibinfo
  {journal} {Journal of High Energy Physics}\ }\textbf {\bibinfo {volume}
  {2017}},\ \bibinfo {pages} {93} (\bibinfo {year} {2017})}\BibitemShut
  {NoStop}%
\bibitem [{\citenamefont {Pastawski}\ and\ \citenamefont
  {Preskill}(2017)}]{pastawski2017}%
  \BibitemOpen
  \bibfield  {author} {\bibinfo {author} {\bibfnamefont {F.}~\bibnamefont
  {Pastawski}}\ and\ \bibinfo {author} {\bibfnamefont {J.}~\bibnamefont
  {Preskill}},\ }\bibfield  {title} {\bibinfo {title} {Code {{Properties}} from
  {{Holographic Geometries}}},\ }\href
  {https://doi.org/10.1103/PhysRevX.7.021022} {\bibfield  {journal} {\bibinfo
  {journal} {Physical Review X}\ }\textbf {\bibinfo {volume} {7}},\ \bibinfo
  {pages} {021022} (\bibinfo {year} {2017})}\BibitemShut {NoStop}%
\bibitem [{\citenamefont {Almheiri}(2018)}]{almheiri2018}%
  \BibitemOpen
  \bibfield  {author} {\bibinfo {author} {\bibfnamefont {A.}~\bibnamefont
  {Almheiri}},\ }\href@noop {} {\bibinfo {title} {Holographic {{Quantum Error
  Correction}} and the {{Projected Black Hole Interior}}}} (\bibinfo {year}
  {2018}),\ \Eprint {https://arxiv.org/abs/1810.02055} {arXiv:1810.02055
  [gr-qc, physics:hep-th]} \BibitemShut {NoStop}%
\bibitem [{\citenamefont {Akers}\ and\ \citenamefont {Rath}(2019)}]{akers2019}%
  \BibitemOpen
  \bibfield  {author} {\bibinfo {author} {\bibfnamefont {C.}~\bibnamefont
  {Akers}}\ and\ \bibinfo {author} {\bibfnamefont {P.}~\bibnamefont {Rath}},\
  }\bibfield  {title} {\bibinfo {title} {Holographic {{Renyi}} entropy from
  quantum error correction},\ }\href {https://doi.org/10.1007/JHEP05(2019)052}
  {\bibfield  {journal} {\bibinfo  {journal} {Journal of High Energy Physics}\
  }\textbf {\bibinfo {volume} {2019}},\ \bibinfo {pages} {52} (\bibinfo {year}
  {2019})}\BibitemShut {NoStop}%
\bibitem [{\citenamefont {Cotler}\ \emph {et~al.}(2019)\citenamefont {Cotler},
  \citenamefont {Hayden}, \citenamefont {Penington}, \citenamefont {Salton},
  \citenamefont {Swingle},\ and\ \citenamefont {Walter}}]{cotler2019}%
  \BibitemOpen
  \bibfield  {author} {\bibinfo {author} {\bibfnamefont {J.}~\bibnamefont
  {Cotler}}, \bibinfo {author} {\bibfnamefont {P.}~\bibnamefont {Hayden}},
  \bibinfo {author} {\bibfnamefont {G.}~\bibnamefont {Penington}}, \bibinfo
  {author} {\bibfnamefont {G.}~\bibnamefont {Salton}}, \bibinfo {author}
  {\bibfnamefont {B.}~\bibnamefont {Swingle}},\ and\ \bibinfo {author}
  {\bibfnamefont {M.}~\bibnamefont {Walter}},\ }\bibfield  {title} {\bibinfo
  {title} {Entanglement {{Wedge Reconstruction}} via {{Universal Recovery
  Channels}}},\ }\href {https://doi.org/10.1103/PhysRevX.9.031011} {\bibfield
  {journal} {\bibinfo  {journal} {Physical Review X}\ }\textbf {\bibinfo
  {volume} {9}},\ \bibinfo {pages} {031011} (\bibinfo {year}
  {2019})}\BibitemShut {NoStop}%
\bibitem [{\citenamefont {Hayden}\ and\ \citenamefont
  {Penington}(2020)}]{hayden2020}%
  \BibitemOpen
  \bibfield  {author} {\bibinfo {author} {\bibfnamefont {P.}~\bibnamefont
  {Hayden}}\ and\ \bibinfo {author} {\bibfnamefont {G.}~\bibnamefont
  {Penington}},\ }\bibfield  {title} {\bibinfo {title} {Approximate {{Quantum
  Error Correction Revisited}}: {{Introducing}} the {{Alpha-Bit}}},\ }\href
  {https://doi.org/10.1007/s00220-020-03689-1} {\bibfield  {journal} {\bibinfo
  {journal} {Communications in Mathematical Physics}\ }\textbf {\bibinfo
  {volume} {374}},\ \bibinfo {pages} {369} (\bibinfo {year}
  {2020})}\BibitemShut {NoStop}%
\bibitem [{\citenamefont {Osborne}\ and\ \citenamefont
  {Stiegemann}(2020)}]{osborne2020}%
  \BibitemOpen
  \bibfield  {author} {\bibinfo {author} {\bibfnamefont {T.~J.}\ \bibnamefont
  {Osborne}}\ and\ \bibinfo {author} {\bibfnamefont {D.~E.}\ \bibnamefont
  {Stiegemann}},\ }\bibfield  {title} {\bibinfo {title} {Dynamics for
  holographic codes},\ }\href {https://doi.org/10.1007/JHEP04(2020)154}
  {\bibfield  {journal} {\bibinfo  {journal} {Journal of High Energy Physics}\
  }\textbf {\bibinfo {volume} {2020}},\ \bibinfo {pages} {154} (\bibinfo {year}
  {2020})}\BibitemShut {NoStop}%
\bibitem [{\citenamefont {Penington}(2020)}]{penington2020}%
  \BibitemOpen
  \bibfield  {author} {\bibinfo {author} {\bibfnamefont {G.}~\bibnamefont
  {Penington}},\ }\bibfield  {title} {\bibinfo {title} {Entanglement wedge
  reconstruction and the information paradox},\ }\href
  {https://doi.org/10.1007/JHEP09(2020)002} {\bibfield  {journal} {\bibinfo
  {journal} {Journal of High Energy Physics}\ }\textbf {\bibinfo {volume}
  {2020}},\ \bibinfo {pages} {2} (\bibinfo {year} {2020})}\BibitemShut
  {NoStop}%
\bibitem [{\citenamefont {Cao}\ and\ \citenamefont {Lackey}(2021)}]{cao2021}%
  \BibitemOpen
  \bibfield  {author} {\bibinfo {author} {\bibfnamefont {C.}~\bibnamefont
  {Cao}}\ and\ \bibinfo {author} {\bibfnamefont {B.}~\bibnamefont {Lackey}},\
  }\bibfield  {title} {\bibinfo {title} {Approximate {{Bacon-Shor}} code and
  holography},\ }\href {https://doi.org/10.1007/JHEP05(2021)127} {\bibfield
  {journal} {\bibinfo  {journal} {Journal of High Energy Physics}\ }\textbf
  {\bibinfo {volume} {2021}},\ \bibinfo {pages} {127} (\bibinfo {year}
  {2021})}\BibitemShut {NoStop}%
\bibitem [{\citenamefont {Harlow}\ and\ \citenamefont
  {Ooguri}(2021)}]{harlow2021}%
  \BibitemOpen
  \bibfield  {author} {\bibinfo {author} {\bibfnamefont {D.}~\bibnamefont
  {Harlow}}\ and\ \bibinfo {author} {\bibfnamefont {H.}~\bibnamefont
  {Ooguri}},\ }\bibfield  {title} {\bibinfo {title} {Symmetries in {{Quantum
  Field Theory}} and {{Quantum Gravity}}},\ }\href
  {https://doi.org/10.1007/s00220-021-04040-y} {\bibfield  {journal} {\bibinfo
  {journal} {Communications in Mathematical Physics}\ }\textbf {\bibinfo
  {volume} {383}},\ \bibinfo {pages} {1669} (\bibinfo {year}
  {2021})}\BibitemShut {NoStop}%
\bibitem [{\citenamefont {Akers}\ and\ \citenamefont
  {Penington}(2022)}]{akers2022a}%
  \BibitemOpen
  \bibfield  {author} {\bibinfo {author} {\bibfnamefont {C.}~\bibnamefont
  {Akers}}\ and\ \bibinfo {author} {\bibfnamefont {G.}~\bibnamefont
  {Penington}},\ }\bibfield  {title} {\bibinfo {title} {Quantum minimal
  surfaces from quantum error correction},\ }\href
  {https://doi.org/10.21468/SciPostPhys.12.5.157} {\bibfield  {journal}
  {\bibinfo  {journal} {SciPost Physics}\ }\textbf {\bibinfo {volume} {12}},\
  \bibinfo {pages} {157} (\bibinfo {year} {2022})}\BibitemShut {NoStop}%
\bibitem [{\citenamefont {Akers}\ \emph {et~al.}(2022)\citenamefont {Akers},
  \citenamefont {Engelhardt}, \citenamefont {Harlow}, \citenamefont
  {Penington},\ and\ \citenamefont {Vardhan}}]{akers2022}%
  \BibitemOpen
  \bibfield  {author} {\bibinfo {author} {\bibfnamefont {C.}~\bibnamefont
  {Akers}}, \bibinfo {author} {\bibfnamefont {N.}~\bibnamefont {Engelhardt}},
  \bibinfo {author} {\bibfnamefont {D.}~\bibnamefont {Harlow}}, \bibinfo
  {author} {\bibfnamefont {G.}~\bibnamefont {Penington}},\ and\ \bibinfo
  {author} {\bibfnamefont {S.}~\bibnamefont {Vardhan}},\ }\href@noop {}
  {\bibinfo {title} {The black hole interior from non-isometric codes and
  complexity}} (\bibinfo {year} {2022}),\ \Eprint
  {https://arxiv.org/abs/2207.06536} {arXiv:2207.06536 [gr-qc, physics:hep-th,
  physics:quant-ph]} \BibitemShut {NoStop}%
\bibitem [{\citenamefont {Balasubramanian}\ \emph {et~al.}(2023)\citenamefont
  {Balasubramanian}, \citenamefont {Kar}, \citenamefont {Li},\ and\
  \citenamefont {Parrikar}}]{balasubramanian2023}%
  \BibitemOpen
  \bibfield  {author} {\bibinfo {author} {\bibfnamefont {V.}~\bibnamefont
  {Balasubramanian}}, \bibinfo {author} {\bibfnamefont {A.}~\bibnamefont
  {Kar}}, \bibinfo {author} {\bibfnamefont {C.}~\bibnamefont {Li}},\ and\
  \bibinfo {author} {\bibfnamefont {O.}~\bibnamefont {Parrikar}},\ }\bibfield
  {title} {\bibinfo {title} {Quantum error correction in the black hole
  interior},\ }\href {https://doi.org/10.1007/JHEP07(2023)189} {\bibfield
  {journal} {\bibinfo  {journal} {Journal of High Energy Physics}\ }\textbf
  {\bibinfo {volume} {2023}},\ \bibinfo {pages} {189} (\bibinfo {year}
  {2023})}\BibitemShut {NoStop}%
\bibitem [{\citenamefont {Gubser}\ \emph {et~al.}(2017)\citenamefont {Gubser},
  \citenamefont {Knaute}, \citenamefont {Parikh}, \citenamefont {Samberg},\
  and\ \citenamefont {Witaszczyk}}]{gubser2017}%
  \BibitemOpen
  \bibfield  {author} {\bibinfo {author} {\bibfnamefont {S.~S.}\ \bibnamefont
  {Gubser}}, \bibinfo {author} {\bibfnamefont {J.}~\bibnamefont {Knaute}},
  \bibinfo {author} {\bibfnamefont {S.}~\bibnamefont {Parikh}}, \bibinfo
  {author} {\bibfnamefont {A.}~\bibnamefont {Samberg}},\ and\ \bibinfo {author}
  {\bibfnamefont {P.}~\bibnamefont {Witaszczyk}},\ }\bibfield  {title}
  {\bibinfo {title} {P-{{Adic AdS}}/{{CFT}}},\ }\href
  {https://doi.org/10.1007/s00220-016-2813-6} {\bibfield  {journal} {\bibinfo
  {journal} {Communications in Mathematical Physics}\ }\textbf {\bibinfo
  {volume} {352}},\ \bibinfo {pages} {1019} (\bibinfo {year}
  {2017})}\BibitemShut {NoStop}%
\bibitem [{\citenamefont {Heydeman}\ \emph {et~al.}(2018)\citenamefont
  {Heydeman}, \citenamefont {Marcolli}, \citenamefont {Saberi},\ and\
  \citenamefont {Stoica}}]{heydeman2018a}%
  \BibitemOpen
  \bibfield  {author} {\bibinfo {author} {\bibfnamefont {M.}~\bibnamefont
  {Heydeman}}, \bibinfo {author} {\bibfnamefont {M.}~\bibnamefont {Marcolli}},
  \bibinfo {author} {\bibfnamefont {I.~A.}\ \bibnamefont {Saberi}},\ and\
  \bibinfo {author} {\bibfnamefont {B.}~\bibnamefont {Stoica}},\ }\bibfield
  {title} {\bibinfo {title} {Tensor networks, p-adic fields, and algebraic
  curves: Arithmetic and the {{AdS3}}/{{CFT2}} correspondence},\ }\href
  {https://doi.org/https://dx.doi.org/10.4310/ATMP.2018.v22.n1.a4} {\bibfield
  {journal} {\bibinfo  {journal} {Advances in Theoretical and Mathematical
  Physics}\ }\textbf {\bibinfo {volume} {22}},\ \bibinfo {pages} {93} (\bibinfo
  {year} {2018})}\BibitemShut {NoStop}%
\bibitem [{\citenamefont {Bhattacharyya}\ \emph {et~al.}(2018)\citenamefont
  {Bhattacharyya}, \citenamefont {Hung}, \citenamefont {Lei},\ and\
  \citenamefont {Li}}]{bhattacharyya2018}%
  \BibitemOpen
  \bibfield  {author} {\bibinfo {author} {\bibfnamefont {A.}~\bibnamefont
  {Bhattacharyya}}, \bibinfo {author} {\bibfnamefont {L.-y.}\ \bibnamefont
  {Hung}}, \bibinfo {author} {\bibfnamefont {Y.}~\bibnamefont {Lei}},\ and\
  \bibinfo {author} {\bibfnamefont {W.}~\bibnamefont {Li}},\ }\bibfield
  {title} {\bibinfo {title} {Tensor network and (p-adic) {{AdS}}/{{CFT}}},\
  }\href {https://doi.org/10.1007/JHEP01(2018)139} {\bibfield  {journal}
  {\bibinfo  {journal} {Journal of High Energy Physics}\ }\textbf {\bibinfo
  {volume} {2018}},\ \bibinfo {pages} {139} (\bibinfo {year}
  {2018})}\BibitemShut {NoStop}%
\bibitem [{\citenamefont {Hung}\ \emph {et~al.}(2019)\citenamefont {Hung},
  \citenamefont {Li},\ and\ \citenamefont {{Melby-Thompson}}}]{hung2019}%
  \BibitemOpen
  \bibfield  {author} {\bibinfo {author} {\bibfnamefont {L.~Y.}\ \bibnamefont
  {Hung}}, \bibinfo {author} {\bibfnamefont {W.}~\bibnamefont {Li}},\ and\
  \bibinfo {author} {\bibfnamefont {C.~M.}\ \bibnamefont {{Melby-Thompson}}},\
  }\bibfield  {title} {\bibinfo {title} {P-adic {{CFT}} is a holographic tensor
  network},\ }\href {https://doi.org/10.1007/JHEP04(2019)170} {\bibfield
  {journal} {\bibinfo  {journal} {Journal of High Energy Physics}\ }\textbf
  {\bibinfo {volume} {04}},\ \bibinfo {pages} {170} (\bibinfo {year}
  {2019})}\BibitemShut {NoStop}%
\bibitem [{\citenamefont {Chen}\ \emph
  {et~al.}(2021{\natexlab{a}})\citenamefont {Chen}, \citenamefont {Liu},\ and\
  \citenamefont {Hung}}]{chen2021a}%
  \BibitemOpen
  \bibfield  {author} {\bibinfo {author} {\bibfnamefont {L.}~\bibnamefont
  {Chen}}, \bibinfo {author} {\bibfnamefont {X.}~\bibnamefont {Liu}},\ and\
  \bibinfo {author} {\bibfnamefont {L.-Y.}\ \bibnamefont {Hung}},\ }\bibfield
  {title} {\bibinfo {title} {Emergent {{Einstein Equation}} in p -adic
  {{Conformal Field Theory Tensor Networks}}},\ }\href
  {https://doi.org/10.1103/PhysRevLett.127.221602} {\bibfield  {journal}
  {\bibinfo  {journal} {Physical Review Letters}\ }\textbf {\bibinfo {volume}
  {127}},\ \bibinfo {pages} {221602} (\bibinfo {year}
  {2021}{\natexlab{a}})}\BibitemShut {NoStop}%
\bibitem [{\citenamefont {Chen}\ \emph
  {et~al.}(2021{\natexlab{b}})\citenamefont {Chen}, \citenamefont {Liu},\ and\
  \citenamefont {Hung}}]{chen2021}%
  \BibitemOpen
  \bibfield  {author} {\bibinfo {author} {\bibfnamefont {L.}~\bibnamefont
  {Chen}}, \bibinfo {author} {\bibfnamefont {X.}~\bibnamefont {Liu}},\ and\
  \bibinfo {author} {\bibfnamefont {L.-Y.}\ \bibnamefont {Hung}},\ }\bibfield
  {title} {\bibinfo {title} {Bending the {{Bruhat-Tits}} tree. {{Part I}}.
  {{Tensor}} network and emergent {{Einstein}} equations},\ }\href
  {https://doi.org/10.1007/JHEP06(2021)094} {\bibfield  {journal} {\bibinfo
  {journal} {Journal of High Energy Physics}\ }\textbf {\bibinfo {volume}
  {2021}},\ \bibinfo {pages} {94} (\bibinfo {year}
  {2021}{\natexlab{b}})}\BibitemShut {NoStop}%
\bibitem [{\citenamefont {Yan}\ \emph {et~al.}(2023)\citenamefont {Yan},
  \citenamefont {Jepsen},\ and\ \citenamefont {Oz}}]{yan2023}%
  \BibitemOpen
  \bibfield  {author} {\bibinfo {author} {\bibfnamefont {H.}~\bibnamefont
  {Yan}}, \bibinfo {author} {\bibfnamefont {C.~B.}\ \bibnamefont {Jepsen}},\
  and\ \bibinfo {author} {\bibfnamefont {Y.}~\bibnamefont {Oz}},\ }\href@noop
  {} {\bibinfo {title} {\$p\$-adic {{Holography}} from the {{Hyperbolic Fracton
  Model}}}} (\bibinfo {year} {2023}),\ \Eprint
  {https://arxiv.org/abs/2306.07203} {arXiv:2306.07203 [cond-mat,
  physics:hep-th]} \BibitemShut {NoStop}%
\bibitem [{\citenamefont {Ebert}\ \emph {et~al.}(2023)\citenamefont {Ebert},
  \citenamefont {Sun},\ and\ \citenamefont {Zhang}}]{ebert2023}%
  \BibitemOpen
  \bibfield  {author} {\bibinfo {author} {\bibfnamefont {S.}~\bibnamefont
  {Ebert}}, \bibinfo {author} {\bibfnamefont {H.-Y.}\ \bibnamefont {Sun}},\
  and\ \bibinfo {author} {\bibfnamefont {M.-Y.}\ \bibnamefont {Zhang}},\
  }\bibfield  {title} {\bibinfo {title} {Probing holography in p -adic
  {{CFT}}},\ }\href {https://doi.org/10.1103/PhysRevD.107.126011} {\bibfield
  {journal} {\bibinfo  {journal} {Physical Review D}\ }\textbf {\bibinfo
  {volume} {107}},\ \bibinfo {pages} {126011} (\bibinfo {year}
  {2023})}\BibitemShut {NoStop}%
\bibitem [{\citenamefont {Faist}\ \emph {et~al.}(2020)\citenamefont {Faist},
  \citenamefont {Nezami}, \citenamefont {Albert}, \citenamefont {Salton},
  \citenamefont {Pastawski}, \citenamefont {Hayden},\ and\ \citenamefont
  {Preskill}}]{faist2020}%
  \BibitemOpen
  \bibfield  {author} {\bibinfo {author} {\bibfnamefont {P.}~\bibnamefont
  {Faist}}, \bibinfo {author} {\bibfnamefont {S.}~\bibnamefont {Nezami}},
  \bibinfo {author} {\bibfnamefont {V.~V.}\ \bibnamefont {Albert}}, \bibinfo
  {author} {\bibfnamefont {G.}~\bibnamefont {Salton}}, \bibinfo {author}
  {\bibfnamefont {F.}~\bibnamefont {Pastawski}}, \bibinfo {author}
  {\bibfnamefont {P.}~\bibnamefont {Hayden}},\ and\ \bibinfo {author}
  {\bibfnamefont {J.}~\bibnamefont {Preskill}},\ }\bibfield  {title} {\bibinfo
  {title} {Continuous {{Symmetries}} and {{Approximate Quantum Error
  Correction}}},\ }\href {https://doi.org/10.1103/PhysRevX.10.041018}
  {\bibfield  {journal} {\bibinfo  {journal} {Physical Review X}\ }\textbf
  {\bibinfo {volume} {10}},\ \bibinfo {pages} {041018} (\bibinfo {year}
  {2020})}\BibitemShut {NoStop}%
\bibitem [{\citenamefont {Harris}\ \emph {et~al.}(2020)\citenamefont {Harris},
  \citenamefont {Coupe}, \citenamefont {McMahon}, \citenamefont {Brennen},\
  and\ \citenamefont {Stace}}]{harris2020}%
  \BibitemOpen
  \bibfield  {author} {\bibinfo {author} {\bibfnamefont {R.~J.}\ \bibnamefont
  {Harris}}, \bibinfo {author} {\bibfnamefont {E.}~\bibnamefont {Coupe}},
  \bibinfo {author} {\bibfnamefont {N.~A.}\ \bibnamefont {McMahon}}, \bibinfo
  {author} {\bibfnamefont {G.~K.}\ \bibnamefont {Brennen}},\ and\ \bibinfo
  {author} {\bibfnamefont {T.~M.}\ \bibnamefont {Stace}},\ }\bibfield  {title}
  {\bibinfo {title} {Decoding holographic codes with an integer optimization
  decoder},\ }\href {https://doi.org/10.1103/PhysRevA.102.062417} {\bibfield
  {journal} {\bibinfo  {journal} {Physical Review A}\ }\textbf {\bibinfo
  {volume} {102}},\ \bibinfo {pages} {062417} (\bibinfo {year}
  {2020})}\BibitemShut {NoStop}%
\bibitem [{\citenamefont {Cree}\ \emph {et~al.}(2021)\citenamefont {Cree},
  \citenamefont {Dolev}, \citenamefont {Calvera},\ and\ \citenamefont
  {Williamson}}]{cree2021}%
  \BibitemOpen
  \bibfield  {author} {\bibinfo {author} {\bibfnamefont {S.}~\bibnamefont
  {Cree}}, \bibinfo {author} {\bibfnamefont {K.}~\bibnamefont {Dolev}},
  \bibinfo {author} {\bibfnamefont {V.}~\bibnamefont {Calvera}},\ and\ \bibinfo
  {author} {\bibfnamefont {D.~J.}\ \bibnamefont {Williamson}},\ }\bibfield
  {title} {\bibinfo {title} {Fault-{{Tolerant Logical Gates}} in {{Holographic
  Stabilizer Codes Are Severely Restricted}}},\ }\href
  {https://doi.org/10.1103/PRXQuantum.2.030337} {\bibfield  {journal} {\bibinfo
   {journal} {PRX Quantum}\ }\textbf {\bibinfo {volume} {2}},\ \bibinfo {pages}
  {030337} (\bibinfo {year} {2021})}\BibitemShut {NoStop}%
\bibitem [{\citenamefont {Cao}\ and\ \citenamefont {Lackey}(2022)}]{cao2022}%
  \BibitemOpen
  \bibfield  {author} {\bibinfo {author} {\bibfnamefont {C.}~\bibnamefont
  {Cao}}\ and\ \bibinfo {author} {\bibfnamefont {B.}~\bibnamefont {Lackey}},\
  }\bibfield  {title} {\bibinfo {title} {Quantum {{Lego}}: {{Building Quantum
  Error Correction Codes}} from {{Tensor Networks}}},\ }\href
  {https://doi.org/10.1103/PRXQuantum.3.020332} {\bibfield  {journal} {\bibinfo
   {journal} {PRX Quantum}\ }\textbf {\bibinfo {volume} {3}},\ \bibinfo {pages}
  {020332} (\bibinfo {year} {2022})}\BibitemShut {NoStop}%
\bibitem [{\citenamefont {Farrelly}\ \emph {et~al.}(2022)\citenamefont
  {Farrelly}, \citenamefont {Milicevic}, \citenamefont {Harris}, \citenamefont
  {McMahon},\ and\ \citenamefont {Stace}}]{farrelly2022}%
  \BibitemOpen
  \bibfield  {author} {\bibinfo {author} {\bibfnamefont {T.}~\bibnamefont
  {Farrelly}}, \bibinfo {author} {\bibfnamefont {N.}~\bibnamefont {Milicevic}},
  \bibinfo {author} {\bibfnamefont {R.~J.}\ \bibnamefont {Harris}}, \bibinfo
  {author} {\bibfnamefont {N.~A.}\ \bibnamefont {McMahon}},\ and\ \bibinfo
  {author} {\bibfnamefont {T.~M.}\ \bibnamefont {Stace}},\ }\bibfield  {title}
  {\bibinfo {title} {Parallel decoding of multiple logical qubits in
  tensor-network codes},\ }\href {https://doi.org/10.1103/PhysRevA.105.052446}
  {\bibfield  {journal} {\bibinfo  {journal} {Physical Review A}\ }\textbf
  {\bibinfo {volume} {105}},\ \bibinfo {pages} {052446} (\bibinfo {year}
  {2022})}\BibitemShut {NoStop}%
\bibitem [{\citenamefont {Bao}\ \emph {et~al.}(2022)\citenamefont {Bao},
  \citenamefont {Cao},\ and\ \citenamefont {Zhu}}]{bao2022a}%
  \BibitemOpen
  \bibfield  {author} {\bibinfo {author} {\bibfnamefont {N.}~\bibnamefont
  {Bao}}, \bibinfo {author} {\bibfnamefont {C.}~\bibnamefont {Cao}},\ and\
  \bibinfo {author} {\bibfnamefont {G.}~\bibnamefont {Zhu}},\ }\bibfield
  {title} {\bibinfo {title} {Deconfinement and error thresholds in
  holography},\ }\href {https://doi.org/10.1103/PhysRevD.106.046009} {\bibfield
   {journal} {\bibinfo  {journal} {Physical Review D}\ }\textbf {\bibinfo
  {volume} {106}},\ \bibinfo {pages} {046009} (\bibinfo {year}
  {2022})}\BibitemShut {NoStop}%
\bibitem [{\citenamefont {Periwal}\ \emph {et~al.}(2021)\citenamefont
  {Periwal}, \citenamefont {Cooper}, \citenamefont {Kunkel}, \citenamefont
  {Wienand}, \citenamefont {Davis},\ and\ \citenamefont
  {{Schleier-Smith}}}]{periwal2021}%
  \BibitemOpen
  \bibfield  {author} {\bibinfo {author} {\bibfnamefont {A.}~\bibnamefont
  {Periwal}}, \bibinfo {author} {\bibfnamefont {E.~S.}\ \bibnamefont {Cooper}},
  \bibinfo {author} {\bibfnamefont {P.}~\bibnamefont {Kunkel}}, \bibinfo
  {author} {\bibfnamefont {J.~F.}\ \bibnamefont {Wienand}}, \bibinfo {author}
  {\bibfnamefont {E.~J.}\ \bibnamefont {Davis}},\ and\ \bibinfo {author}
  {\bibfnamefont {M.}~\bibnamefont {{Schleier-Smith}}},\ }\bibfield  {title}
  {\bibinfo {title} {Programmable interactions and emergent geometry in an
  array of atom clouds},\ }\href {https://doi.org/10.1038/s41586-021-04156-0}
  {\bibfield  {journal} {\bibinfo  {journal} {Nature}\ }\textbf {\bibinfo
  {volume} {600}},\ \bibinfo {pages} {630} (\bibinfo {year}
  {2021})}\BibitemShut {NoStop}%
\bibitem [{\citenamefont {Daley}\ \emph {et~al.}(2022)\citenamefont {Daley},
  \citenamefont {Bloch}, \citenamefont {Kokail}, \citenamefont {Flannigan},
  \citenamefont {Pearson}, \citenamefont {Troyer},\ and\ \citenamefont
  {Zoller}}]{daley2022}%
  \BibitemOpen
  \bibfield  {author} {\bibinfo {author} {\bibfnamefont {A.~J.}\ \bibnamefont
  {Daley}}, \bibinfo {author} {\bibfnamefont {I.}~\bibnamefont {Bloch}},
  \bibinfo {author} {\bibfnamefont {C.}~\bibnamefont {Kokail}}, \bibinfo
  {author} {\bibfnamefont {S.}~\bibnamefont {Flannigan}}, \bibinfo {author}
  {\bibfnamefont {N.}~\bibnamefont {Pearson}}, \bibinfo {author} {\bibfnamefont
  {M.}~\bibnamefont {Troyer}},\ and\ \bibinfo {author} {\bibfnamefont
  {P.}~\bibnamefont {Zoller}},\ }\bibfield  {title} {\bibinfo {title}
  {Practical quantum advantage in quantum simulation},\ }\href
  {https://doi.org/10.1038/s41586-022-04940-6} {\bibfield  {journal} {\bibinfo
  {journal} {Nature}\ }\textbf {\bibinfo {volume} {607}},\ \bibinfo {pages}
  {667} (\bibinfo {year} {2022})}\BibitemShut {NoStop}%
\bibitem [{\citenamefont {Angl{\`e}s~Munn{\'e}}\ \emph
  {et~al.}(2024)\citenamefont {Angl{\`e}s~Munn{\'e}}, \citenamefont {Kasper},\
  and\ \citenamefont {Huber}}]{anglesmunne2024}%
  \BibitemOpen
  \bibfield  {author} {\bibinfo {author} {\bibfnamefont {G.}~\bibnamefont
  {Angl{\`e}s~Munn{\'e}}}, \bibinfo {author} {\bibfnamefont {V.}~\bibnamefont
  {Kasper}},\ and\ \bibinfo {author} {\bibfnamefont {F.}~\bibnamefont
  {Huber}},\ }\bibfield  {title} {\bibinfo {title} {Engineering holography with
  stabilizer graph codes},\ }\href {https://doi.org/10.1038/s41534-024-00822-z}
  {\bibfield  {journal} {\bibinfo  {journal} {npj Quantum Information}\
  }\textbf {\bibinfo {volume} {10}},\ \bibinfo {pages} {48} (\bibinfo {year}
  {2024})}\BibitemShut {NoStop}%
\bibitem [{\citenamefont {Bluvstein}\ \emph {et~al.}(2022)\citenamefont
  {Bluvstein}, \citenamefont {Levine}, \citenamefont {Semeghini}, \citenamefont
  {Wang}, \citenamefont {Ebadi}, \citenamefont {Kalinowski}, \citenamefont
  {Keesling}, \citenamefont {Maskara}, \citenamefont {Pichler}, \citenamefont
  {Greiner}, \citenamefont {Vuleti{\'c}},\ and\ \citenamefont
  {Lukin}}]{bluvstein2022}%
  \BibitemOpen
  \bibfield  {author} {\bibinfo {author} {\bibfnamefont {D.}~\bibnamefont
  {Bluvstein}}, \bibinfo {author} {\bibfnamefont {H.}~\bibnamefont {Levine}},
  \bibinfo {author} {\bibfnamefont {G.}~\bibnamefont {Semeghini}}, \bibinfo
  {author} {\bibfnamefont {T.~T.}\ \bibnamefont {Wang}}, \bibinfo {author}
  {\bibfnamefont {S.}~\bibnamefont {Ebadi}}, \bibinfo {author} {\bibfnamefont
  {M.}~\bibnamefont {Kalinowski}}, \bibinfo {author} {\bibfnamefont
  {A.}~\bibnamefont {Keesling}}, \bibinfo {author} {\bibfnamefont
  {N.}~\bibnamefont {Maskara}}, \bibinfo {author} {\bibfnamefont
  {H.}~\bibnamefont {Pichler}}, \bibinfo {author} {\bibfnamefont
  {M.}~\bibnamefont {Greiner}}, \bibinfo {author} {\bibfnamefont
  {V.}~\bibnamefont {Vuleti{\'c}}},\ and\ \bibinfo {author} {\bibfnamefont
  {M.~D.}\ \bibnamefont {Lukin}},\ }\bibfield  {title} {\bibinfo {title} {A
  quantum processor based on coherent transport of entangled atom arrays},\
  }\href {https://doi.org/10.1038/s41586-022-04592-6} {\bibfield  {journal}
  {\bibinfo  {journal} {Nature}\ }\textbf {\bibinfo {volume} {604}},\ \bibinfo
  {pages} {451} (\bibinfo {year} {2022})}\BibitemShut {NoStop}%
\bibitem [{\citenamefont {Jafferis}\ \emph {et~al.}(2022)\citenamefont
  {Jafferis}, \citenamefont {Zlokapa}, \citenamefont {Lykken}, \citenamefont
  {Kolchmeyer}, \citenamefont {Davis}, \citenamefont {Lauk}, \citenamefont
  {Neven},\ and\ \citenamefont {Spiropulu}}]{jafferis2022a}%
  \BibitemOpen
  \bibfield  {author} {\bibinfo {author} {\bibfnamefont {D.}~\bibnamefont
  {Jafferis}}, \bibinfo {author} {\bibfnamefont {A.}~\bibnamefont {Zlokapa}},
  \bibinfo {author} {\bibfnamefont {J.~D.}\ \bibnamefont {Lykken}}, \bibinfo
  {author} {\bibfnamefont {D.~K.}\ \bibnamefont {Kolchmeyer}}, \bibinfo
  {author} {\bibfnamefont {S.~I.}\ \bibnamefont {Davis}}, \bibinfo {author}
  {\bibfnamefont {N.}~\bibnamefont {Lauk}}, \bibinfo {author} {\bibfnamefont
  {H.}~\bibnamefont {Neven}},\ and\ \bibinfo {author} {\bibfnamefont
  {M.}~\bibnamefont {Spiropulu}},\ }\bibfield  {title} {\bibinfo {title}
  {Traversable wormhole dynamics on a quantum processor},\ }\href
  {https://doi.org/10.1038/s41586-022-05424-3} {\bibfield  {journal} {\bibinfo
  {journal} {Nature}\ }\textbf {\bibinfo {volume} {612}},\ \bibinfo {pages}
  {51} (\bibinfo {year} {2022})}\BibitemShut {NoStop}%
\bibitem [{\citenamefont {Shapoval}\ \emph {et~al.}(2022)\citenamefont
  {Shapoval}, \citenamefont {Su}, \citenamefont {{de Jong}}, \citenamefont
  {Urbanek},\ and\ \citenamefont {Swingle}}]{shapoval2022}%
  \BibitemOpen
  \bibfield  {author} {\bibinfo {author} {\bibfnamefont {I.}~\bibnamefont
  {Shapoval}}, \bibinfo {author} {\bibfnamefont {V.~P.}\ \bibnamefont {Su}},
  \bibinfo {author} {\bibfnamefont {W.}~\bibnamefont {{de Jong}}}, \bibinfo
  {author} {\bibfnamefont {M.}~\bibnamefont {Urbanek}},\ and\ \bibinfo {author}
  {\bibfnamefont {B.}~\bibnamefont {Swingle}},\ }\href@noop {} {\bibinfo
  {title} {Towards {{Quantum Gravity}} in the {{Lab}} on {{Quantum
  Processors}}}} (\bibinfo {year} {2022}),\ \Eprint
  {https://arxiv.org/abs/2205.14081} {arXiv:2205.14081 [gr-qc, physics:hep-th,
  physics:quant-ph]} \BibitemShut {NoStop}%
\bibitem [{\citenamefont {Bhattacharyya}\ \emph {et~al.}(2022)\citenamefont
  {Bhattacharyya}, \citenamefont {Joshi},\ and\ \citenamefont
  {Sundar}}]{bhattacharyya2022}%
  \BibitemOpen
  \bibfield  {author} {\bibinfo {author} {\bibfnamefont {A.}~\bibnamefont
  {Bhattacharyya}}, \bibinfo {author} {\bibfnamefont {L.~K.}\ \bibnamefont
  {Joshi}},\ and\ \bibinfo {author} {\bibfnamefont {B.}~\bibnamefont
  {Sundar}},\ }\bibfield  {title} {\bibinfo {title} {Quantum information
  scrambling: From holography to quantum simulators},\ }\href
  {https://doi.org/10.1140/epjc/s10052-022-10377-y} {\bibfield  {journal}
  {\bibinfo  {journal} {The European Physical Journal C}\ }\textbf {\bibinfo
  {volume} {82}},\ \bibinfo {pages} {458} (\bibinfo {year} {2022})}\BibitemShut
  {NoStop}%
\bibitem [{\citenamefont {Brown}\ \emph {et~al.}(2023)\citenamefont {Brown},
  \citenamefont {Gharibyan}, \citenamefont {Leichenauer}, \citenamefont {Lin},
  \citenamefont {Nezami}, \citenamefont {Salton}, \citenamefont {Susskind},
  \citenamefont {Swingle},\ and\ \citenamefont {Walter}}]{brown2023}%
  \BibitemOpen
  \bibfield  {author} {\bibinfo {author} {\bibfnamefont {A.~R.}\ \bibnamefont
  {Brown}}, \bibinfo {author} {\bibfnamefont {H.}~\bibnamefont {Gharibyan}},
  \bibinfo {author} {\bibfnamefont {S.}~\bibnamefont {Leichenauer}}, \bibinfo
  {author} {\bibfnamefont {H.~W.}\ \bibnamefont {Lin}}, \bibinfo {author}
  {\bibfnamefont {S.}~\bibnamefont {Nezami}}, \bibinfo {author} {\bibfnamefont
  {G.}~\bibnamefont {Salton}}, \bibinfo {author} {\bibfnamefont
  {L.}~\bibnamefont {Susskind}}, \bibinfo {author} {\bibfnamefont
  {B.}~\bibnamefont {Swingle}},\ and\ \bibinfo {author} {\bibfnamefont
  {M.}~\bibnamefont {Walter}},\ }\bibfield  {title} {\bibinfo {title} {Quantum
  {{Gravity}} in the {{Lab}}. {{I}}. {{Teleportation}} by {{Size}} and
  {{Traversable Wormholes}}},\ }\href
  {https://doi.org/10.1103/PRXQuantum.4.010320} {\bibfield  {journal} {\bibinfo
   {journal} {PRX Quantum}\ }\textbf {\bibinfo {volume} {4}},\ \bibinfo {pages}
  {010320} (\bibinfo {year} {2023})}\BibitemShut {NoStop}%
\bibitem [{\citenamefont {Nezami}\ \emph {et~al.}(2023)\citenamefont {Nezami},
  \citenamefont {Lin}, \citenamefont {Brown}, \citenamefont {Gharibyan},
  \citenamefont {Leichenauer}, \citenamefont {Salton}, \citenamefont
  {Susskind}, \citenamefont {Swingle},\ and\ \citenamefont
  {Walter}}]{nezami2023}%
  \BibitemOpen
  \bibfield  {author} {\bibinfo {author} {\bibfnamefont {S.}~\bibnamefont
  {Nezami}}, \bibinfo {author} {\bibfnamefont {H.~W.}\ \bibnamefont {Lin}},
  \bibinfo {author} {\bibfnamefont {A.~R.}\ \bibnamefont {Brown}}, \bibinfo
  {author} {\bibfnamefont {H.}~\bibnamefont {Gharibyan}}, \bibinfo {author}
  {\bibfnamefont {S.}~\bibnamefont {Leichenauer}}, \bibinfo {author}
  {\bibfnamefont {G.}~\bibnamefont {Salton}}, \bibinfo {author} {\bibfnamefont
  {L.}~\bibnamefont {Susskind}}, \bibinfo {author} {\bibfnamefont
  {B.}~\bibnamefont {Swingle}},\ and\ \bibinfo {author} {\bibfnamefont
  {M.}~\bibnamefont {Walter}},\ }\bibfield  {title} {\bibinfo {title} {Quantum
  {{Gravity}} in the {{Lab}}. {{II}}. {{Teleportation}} by {{Size}} and
  {{Traversable Wormholes}}},\ }\href
  {https://doi.org/10.1103/PRXQuantum.4.010321} {\bibfield  {journal} {\bibinfo
   {journal} {PRX Quantum}\ }\textbf {\bibinfo {volume} {4}},\ \bibinfo {pages}
  {010321} (\bibinfo {year} {2023})}\BibitemShut {NoStop}%
\bibitem [{\citenamefont {Bluvstein}\ \emph {et~al.}(2024)\citenamefont
  {Bluvstein}, \citenamefont {Evered}, \citenamefont {Geim}, \citenamefont
  {Li}, \citenamefont {Zhou}, \citenamefont {Manovitz}, \citenamefont {Ebadi},
  \citenamefont {Cain}, \citenamefont {Kalinowski}, \citenamefont {Hangleiter},
  \citenamefont {Bonilla~Ataides}, \citenamefont {Maskara}, \citenamefont
  {Cong}, \citenamefont {Gao}, \citenamefont {Sales~Rodriguez}, \citenamefont
  {Karolyshyn}, \citenamefont {Semeghini}, \citenamefont {Gullans},
  \citenamefont {Greiner}, \citenamefont {Vuleti{\'c}},\ and\ \citenamefont
  {Lukin}}]{bluvstein2024}%
  \BibitemOpen
  \bibfield  {author} {\bibinfo {author} {\bibfnamefont {D.}~\bibnamefont
  {Bluvstein}}, \bibinfo {author} {\bibfnamefont {S.~J.}\ \bibnamefont
  {Evered}}, \bibinfo {author} {\bibfnamefont {A.~A.}\ \bibnamefont {Geim}},
  \bibinfo {author} {\bibfnamefont {S.~H.}\ \bibnamefont {Li}}, \bibinfo
  {author} {\bibfnamefont {H.}~\bibnamefont {Zhou}}, \bibinfo {author}
  {\bibfnamefont {T.}~\bibnamefont {Manovitz}}, \bibinfo {author}
  {\bibfnamefont {S.}~\bibnamefont {Ebadi}}, \bibinfo {author} {\bibfnamefont
  {M.}~\bibnamefont {Cain}}, \bibinfo {author} {\bibfnamefont {M.}~\bibnamefont
  {Kalinowski}}, \bibinfo {author} {\bibfnamefont {D.}~\bibnamefont
  {Hangleiter}}, \bibinfo {author} {\bibfnamefont {J.~P.}\ \bibnamefont
  {Bonilla~Ataides}}, \bibinfo {author} {\bibfnamefont {N.}~\bibnamefont
  {Maskara}}, \bibinfo {author} {\bibfnamefont {I.}~\bibnamefont {Cong}},
  \bibinfo {author} {\bibfnamefont {X.}~\bibnamefont {Gao}}, \bibinfo {author}
  {\bibfnamefont {P.}~\bibnamefont {Sales~Rodriguez}}, \bibinfo {author}
  {\bibfnamefont {T.}~\bibnamefont {Karolyshyn}}, \bibinfo {author}
  {\bibfnamefont {G.}~\bibnamefont {Semeghini}}, \bibinfo {author}
  {\bibfnamefont {M.~J.}\ \bibnamefont {Gullans}}, \bibinfo {author}
  {\bibfnamefont {M.}~\bibnamefont {Greiner}}, \bibinfo {author} {\bibfnamefont
  {V.}~\bibnamefont {Vuleti{\'c}}},\ and\ \bibinfo {author} {\bibfnamefont
  {M.~D.}\ \bibnamefont {Lukin}},\ }\bibfield  {title} {\bibinfo {title}
  {Logical quantum processor based on reconfigurable atom arrays},\ }\href
  {https://doi.org/10.1038/s41586-023-06927-3} {\bibfield  {journal} {\bibinfo
  {journal} {Nature}\ }\textbf {\bibinfo {volume} {626}},\ \bibinfo {pages}
  {58} (\bibinfo {year} {2024})}\BibitemShut {NoStop}%
\bibitem [{\citenamefont {Xu}\ \emph {et~al.}(2024)\citenamefont {Xu},
  \citenamefont {Bonilla~Ataides}, \citenamefont {Pattison}, \citenamefont
  {Raveendran}, \citenamefont {Bluvstein}, \citenamefont {Wurtz}, \citenamefont
  {Vasi{\'c}}, \citenamefont {Lukin}, \citenamefont {Jiang},\ and\
  \citenamefont {Zhou}}]{xu2024}%
  \BibitemOpen
  \bibfield  {author} {\bibinfo {author} {\bibfnamefont {Q.}~\bibnamefont
  {Xu}}, \bibinfo {author} {\bibfnamefont {J.~P.}\ \bibnamefont
  {Bonilla~Ataides}}, \bibinfo {author} {\bibfnamefont {C.~A.}\ \bibnamefont
  {Pattison}}, \bibinfo {author} {\bibfnamefont {N.}~\bibnamefont
  {Raveendran}}, \bibinfo {author} {\bibfnamefont {D.}~\bibnamefont
  {Bluvstein}}, \bibinfo {author} {\bibfnamefont {J.}~\bibnamefont {Wurtz}},
  \bibinfo {author} {\bibfnamefont {B.}~\bibnamefont {Vasi{\'c}}}, \bibinfo
  {author} {\bibfnamefont {M.~D.}\ \bibnamefont {Lukin}}, \bibinfo {author}
  {\bibfnamefont {L.}~\bibnamefont {Jiang}},\ and\ \bibinfo {author}
  {\bibfnamefont {H.}~\bibnamefont {Zhou}},\ }\bibfield  {title} {\bibinfo
  {title} {Constant-overhead fault-tolerant quantum computation with
  reconfigurable atom arrays},\ }\bibfield  {journal} {\bibinfo  {journal}
  {Nature Physics}\ }\href {https://doi.org/10.1038/s41567-024-02479-z}
  {10.1038/s41567-024-02479-z} (\bibinfo {year} {2024})\BibitemShut {NoStop}%
\bibitem [{\citenamefont {Ryu}\ and\ \citenamefont
  {Takayanagi}(2006)}]{ryu2006}%
  \BibitemOpen
  \bibfield  {author} {\bibinfo {author} {\bibfnamefont {S.}~\bibnamefont
  {Ryu}}\ and\ \bibinfo {author} {\bibfnamefont {T.}~\bibnamefont
  {Takayanagi}},\ }\bibfield  {title} {\bibinfo {title} {Holographic derivation
  of entanglement entropy from the anti-de sitter space/conformal field theory
  correspondence},\ }\href {https://doi.org/10.1103/PhysRevLett.96.181602}
  {\bibfield  {journal} {\bibinfo  {journal} {Physical Review Letters}\
  }\textbf {\bibinfo {volume} {96}},\ \bibinfo {pages} {181602} (\bibinfo
  {year} {2006})}\BibitemShut {NoStop}%
\bibitem [{Note1()}]{Note1}%
  \BibitemOpen
  \bibinfo {note} {We use ``HQEC'' to refer to the hypothetical
  quantum-information structure for AdS/CFT, and use ``holographic code'' for
  models that are expected to realize the structure.}\BibitemShut {Stop}%
\bibitem [{\citenamefont {Jahn}\ and\ \citenamefont {Eisert}(2021)}]{jahn2021}%
  \BibitemOpen
  \bibfield  {author} {\bibinfo {author} {\bibfnamefont {A.}~\bibnamefont
  {Jahn}}\ and\ \bibinfo {author} {\bibfnamefont {J.}~\bibnamefont {Eisert}},\
  }\bibfield  {title} {\bibinfo {title} {Holographic tensor network models and
  quantum error correction: A topical review},\ }\href
  {https://doi.org/10.1088/2058-9565/ac0293} {\bibfield  {journal} {\bibinfo
  {journal} {Quantum Science and Technology}\ }\textbf {\bibinfo {volume}
  {6}},\ \bibinfo {pages} {033002} (\bibinfo {year} {2021})}\BibitemShut
  {NoStop}%
\bibitem [{\citenamefont {Hayden}\ \emph {et~al.}(2016)\citenamefont {Hayden},
  \citenamefont {Nezami}, \citenamefont {Qi}, \citenamefont {Thomas},
  \citenamefont {Walter},\ and\ \citenamefont {Yang}}]{hayden2016}%
  \BibitemOpen
  \bibfield  {author} {\bibinfo {author} {\bibfnamefont {P.}~\bibnamefont
  {Hayden}}, \bibinfo {author} {\bibfnamefont {S.}~\bibnamefont {Nezami}},
  \bibinfo {author} {\bibfnamefont {X.-L.}\ \bibnamefont {Qi}}, \bibinfo
  {author} {\bibfnamefont {N.}~\bibnamefont {Thomas}}, \bibinfo {author}
  {\bibfnamefont {M.}~\bibnamefont {Walter}},\ and\ \bibinfo {author}
  {\bibfnamefont {Z.}~\bibnamefont {Yang}},\ }\bibfield  {title} {\bibinfo
  {title} {Holographic duality from random tensor networks},\ }\href
  {https://doi.org/10.1007/JHEP11(2016)009} {\bibfield  {journal} {\bibinfo
  {journal} {Journal of High Energy Physics}\ }\textbf {\bibinfo {volume}
  {2016}},\ \bibinfo {pages} {9} (\bibinfo {year} {2016})}\BibitemShut
  {NoStop}%
\bibitem [{\citenamefont {Evenbly}(2017)}]{evenbly2017}%
  \BibitemOpen
  \bibfield  {author} {\bibinfo {author} {\bibfnamefont {G.}~\bibnamefont
  {Evenbly}},\ }\bibfield  {title} {\bibinfo {title} {Hyperinvariant {{Tensor
  Networks}} and {{Holography}}},\ }\href
  {https://doi.org/10.1103/PhysRevLett.119.141602} {\bibfield  {journal}
  {\bibinfo  {journal} {Physical Review Letters}\ }\textbf {\bibinfo {volume}
  {119}},\ \bibinfo {pages} {141602} (\bibinfo {year} {2017})}\BibitemShut
  {NoStop}%
\bibitem [{\citenamefont {Kim}\ and\ \citenamefont
  {Kastoryano}(2017)}]{kim2017a}%
  \BibitemOpen
  \bibfield  {author} {\bibinfo {author} {\bibfnamefont {I.~H.}\ \bibnamefont
  {Kim}}\ and\ \bibinfo {author} {\bibfnamefont {M.~J.}\ \bibnamefont
  {Kastoryano}},\ }\bibfield  {title} {\bibinfo {title} {Entanglement
  renormalization, quantum error correction, and bulk causality},\ }\href
  {https://doi.org/10.1007/JHEP04(2017)040} {\bibfield  {journal} {\bibinfo
  {journal} {Journal of High Energy Physics}\ }\textbf {\bibinfo {volume}
  {2017}},\ \bibinfo {pages} {40} (\bibinfo {year} {2017})}\BibitemShut
  {NoStop}%
\bibitem [{\citenamefont {Harris}\ \emph {et~al.}(2018)\citenamefont {Harris},
  \citenamefont {McMahon}, \citenamefont {Brennen},\ and\ \citenamefont
  {Stace}}]{harris2018}%
  \BibitemOpen
  \bibfield  {author} {\bibinfo {author} {\bibfnamefont {R.~J.}\ \bibnamefont
  {Harris}}, \bibinfo {author} {\bibfnamefont {N.~A.}\ \bibnamefont {McMahon}},
  \bibinfo {author} {\bibfnamefont {G.~K.}\ \bibnamefont {Brennen}},\ and\
  \bibinfo {author} {\bibfnamefont {T.~M.}\ \bibnamefont {Stace}},\ }\bibfield
  {title} {\bibinfo {title} {Calderbank-{{Shor-Steane}} holographic quantum
  error-correcting codes},\ }\href {https://doi.org/10.1103/PhysRevA.98.052301}
  {\bibfield  {journal} {\bibinfo  {journal} {Physical Review A}\ }\textbf
  {\bibinfo {volume} {98}},\ \bibinfo {pages} {052301} (\bibinfo {year}
  {2018})}\BibitemShut {NoStop}%
\bibitem [{\citenamefont {Jahn}\ \emph
  {et~al.}(2019{\natexlab{a}})\citenamefont {Jahn}, \citenamefont {Gluza},
  \citenamefont {Pastawski},\ and\ \citenamefont {Eisert}}]{jahn2019}%
  \BibitemOpen
  \bibfield  {author} {\bibinfo {author} {\bibfnamefont {A.}~\bibnamefont
  {Jahn}}, \bibinfo {author} {\bibfnamefont {M.}~\bibnamefont {Gluza}},
  \bibinfo {author} {\bibfnamefont {F.}~\bibnamefont {Pastawski}},\ and\
  \bibinfo {author} {\bibfnamefont {J.}~\bibnamefont {Eisert}},\ }\bibfield
  {title} {\bibinfo {title} {Holography and criticality in matchgate tensor
  networks},\ }\bibfield  {journal} {\bibinfo  {journal} {Science Advances}\
  }\href {https://doi.org/10.1126/sciadv.aaw0092} {10.1126/sciadv.aaw0092}
  (\bibinfo {year} {2019}{\natexlab{a}})\BibitemShut {NoStop}%
\bibitem [{\citenamefont {Jahn}\ \emph
  {et~al.}(2019{\natexlab{b}})\citenamefont {Jahn}, \citenamefont {Gluza},
  \citenamefont {Pastawski},\ and\ \citenamefont {Eisert}}]{jahn2019a}%
  \BibitemOpen
  \bibfield  {author} {\bibinfo {author} {\bibfnamefont {A.}~\bibnamefont
  {Jahn}}, \bibinfo {author} {\bibfnamefont {M.}~\bibnamefont {Gluza}},
  \bibinfo {author} {\bibfnamefont {F.}~\bibnamefont {Pastawski}},\ and\
  \bibinfo {author} {\bibfnamefont {J.}~\bibnamefont {Eisert}},\ }\bibfield
  {title} {\bibinfo {title} {Majorana dimers and holographic quantum
  error-correcting codes},\ }\href
  {https://doi.org/10.1103/PhysRevResearch.1.033079} {\bibfield  {journal}
  {\bibinfo  {journal} {Physical Review Research}\ }\textbf {\bibinfo {volume}
  {1}},\ \bibinfo {pages} {033079} (\bibinfo {year}
  {2019}{\natexlab{b}})}\BibitemShut {NoStop}%
\bibitem [{\citenamefont {McMahon}\ \emph {et~al.}(2020)\citenamefont
  {McMahon}, \citenamefont {Singh},\ and\ \citenamefont
  {Brennen}}]{mcmahon2020}%
  \BibitemOpen
  \bibfield  {author} {\bibinfo {author} {\bibfnamefont {N.~A.}\ \bibnamefont
  {McMahon}}, \bibinfo {author} {\bibfnamefont {S.}~\bibnamefont {Singh}},\
  and\ \bibinfo {author} {\bibfnamefont {G.~K.}\ \bibnamefont {Brennen}},\
  }\bibfield  {title} {\bibinfo {title} {A holographic duality from lifted
  tensor networks},\ }\href {https://doi.org/10.1038/s41534-020-0255-7}
  {\bibfield  {journal} {\bibinfo  {journal} {npj Quantum Information}\
  }\textbf {\bibinfo {volume} {6}},\ \bibinfo {pages} {36} (\bibinfo {year}
  {2020})}\BibitemShut {NoStop}%
\bibitem [{\citenamefont {Pollack}\ \emph {et~al.}(2022)\citenamefont
  {Pollack}, \citenamefont {Rall},\ and\ \citenamefont
  {Rocchetto}}]{pollack2022}%
  \BibitemOpen
  \bibfield  {author} {\bibinfo {author} {\bibfnamefont {J.}~\bibnamefont
  {Pollack}}, \bibinfo {author} {\bibfnamefont {P.}~\bibnamefont {Rall}},\ and\
  \bibinfo {author} {\bibfnamefont {A.}~\bibnamefont {Rocchetto}},\ }\bibfield
  {title} {\bibinfo {title} {Understanding holographic error correction via
  unique algebras and atomic examples},\ }\href
  {https://doi.org/10.1007/JHEP06(2022)056} {\bibfield  {journal} {\bibinfo
  {journal} {Journal of High Energy Physics}\ }\textbf {\bibinfo {volume}
  {2022}},\ \bibinfo {pages} {56} (\bibinfo {year} {2022})}\BibitemShut
  {NoStop}%
\bibitem [{\citenamefont {Steinberg}\ \emph {et~al.}(2023)\citenamefont
  {Steinberg}, \citenamefont {Feld},\ and\ \citenamefont
  {Jahn}}]{steinberg2023}%
  \BibitemOpen
  \bibfield  {author} {\bibinfo {author} {\bibfnamefont {M.}~\bibnamefont
  {Steinberg}}, \bibinfo {author} {\bibfnamefont {S.}~\bibnamefont {Feld}},\
  and\ \bibinfo {author} {\bibfnamefont {A.}~\bibnamefont {Jahn}},\ }\bibfield
  {title} {\bibinfo {title} {Holographic codes from hyperinvariant tensor
  networks},\ }\href {https://doi.org/10.1038/s41467-023-42743-z} {\bibfield
  {journal} {\bibinfo  {journal} {Nature Communications}\ }\textbf {\bibinfo
  {volume} {14}},\ \bibinfo {pages} {7314} (\bibinfo {year}
  {2023})}\BibitemShut {NoStop}%
\bibitem [{\citenamefont {B{\'e}ny}\ \emph
  {et~al.}(2007{\natexlab{a}})\citenamefont {B{\'e}ny}, \citenamefont {Kempf},\
  and\ \citenamefont {Kribs}}]{beny2007}%
  \BibitemOpen
  \bibfield  {author} {\bibinfo {author} {\bibfnamefont {C.}~\bibnamefont
  {B{\'e}ny}}, \bibinfo {author} {\bibfnamefont {A.}~\bibnamefont {Kempf}},\
  and\ \bibinfo {author} {\bibfnamefont {D.~W.}\ \bibnamefont {Kribs}},\
  }\bibfield  {title} {\bibinfo {title} {Quantum error correction of
  observables},\ }\href {https://doi.org/10.1103/PhysRevA.76.042303} {\bibfield
   {journal} {\bibinfo  {journal} {Physical Review A}\ }\textbf {\bibinfo
  {volume} {76}},\ \bibinfo {pages} {042303} (\bibinfo {year}
  {2007}{\natexlab{a}})}\BibitemShut {NoStop}%
\bibitem [{\citenamefont {B{\'e}ny}\ \emph
  {et~al.}(2007{\natexlab{b}})\citenamefont {B{\'e}ny}, \citenamefont {Kempf},\
  and\ \citenamefont {Kribs}}]{beny2007a}%
  \BibitemOpen
  \bibfield  {author} {\bibinfo {author} {\bibfnamefont {C.}~\bibnamefont
  {B{\'e}ny}}, \bibinfo {author} {\bibfnamefont {A.}~\bibnamefont {Kempf}},\
  and\ \bibinfo {author} {\bibfnamefont {D.~W.}\ \bibnamefont {Kribs}},\
  }\bibfield  {title} {\bibinfo {title} {Generalization of {{Quantum Error
  Correction}} via the {{Heisenberg Picture}}},\ }\href
  {https://doi.org/10.1103/PhysRevLett.98.100502} {\bibfield  {journal}
  {\bibinfo  {journal} {Physical Review Letters}\ }\textbf {\bibinfo {volume}
  {98}},\ \bibinfo {pages} {100502} (\bibinfo {year}
  {2007}{\natexlab{b}})}\BibitemShut {NoStop}%
\bibitem [{\citenamefont {Faulkner}\ \emph {et~al.}(2013)\citenamefont
  {Faulkner}, \citenamefont {Lewkowycz},\ and\ \citenamefont
  {Maldacena}}]{faulkner2013}%
  \BibitemOpen
  \bibfield  {author} {\bibinfo {author} {\bibfnamefont {T.}~\bibnamefont
  {Faulkner}}, \bibinfo {author} {\bibfnamefont {A.}~\bibnamefont
  {Lewkowycz}},\ and\ \bibinfo {author} {\bibfnamefont {J.}~\bibnamefont
  {Maldacena}},\ }\bibfield  {title} {\bibinfo {title} {Quantum corrections to
  holographic entanglement entropy},\ }\href
  {https://doi.org/10.1007/JHEP11(2013)074} {\bibfield  {journal} {\bibinfo
  {journal} {Journal of High Energy Physics}\ }\textbf {\bibinfo {volume}
  {2013}},\ \bibinfo {pages} {74} (\bibinfo {year} {2013})}\BibitemShut
  {NoStop}%
\bibitem [{\citenamefont {Dauphinais}\ \emph {et~al.}(2024)\citenamefont
  {Dauphinais}, \citenamefont {Kribs},\ and\ \citenamefont
  {Vasmer}}]{dauphinais2024}%
  \BibitemOpen
  \bibfield  {author} {\bibinfo {author} {\bibfnamefont {G.}~\bibnamefont
  {Dauphinais}}, \bibinfo {author} {\bibfnamefont {D.~W.}\ \bibnamefont
  {Kribs}},\ and\ \bibinfo {author} {\bibfnamefont {M.}~\bibnamefont
  {Vasmer}},\ }\bibfield  {title} {\bibinfo {title} {Stabilizer {{Formalism}}
  for {{Operator Algebra Quantum Error Correction}}},\ }\href
  {https://doi.org/10.22331/q-2024-02-21-1261} {\bibfield  {journal} {\bibinfo
  {journal} {Quantum}\ }\textbf {\bibinfo {volume} {8}},\ \bibinfo {pages}
  {1261} (\bibinfo {year} {2024})},\ \Eprint {https://arxiv.org/abs/2304.11442}
  {arXiv:2304.11442 [quant-ph]} \BibitemShut {NoStop}%
\bibitem [{\citenamefont {Kitaev}(2003)}]{kitaev2003}%
  \BibitemOpen
  \bibfield  {author} {\bibinfo {author} {\bibfnamefont {A.~Y.}\ \bibnamefont
  {Kitaev}},\ }\bibfield  {title} {\bibinfo {title} {Fault-tolerant quantum
  computation by anyons},\ }\href
  {https://doi.org/10.1016/S0003-4916(02)00018-0} {\bibfield  {journal}
  {\bibinfo  {journal} {Ann. Phys. (N. Y.)}\ }\textbf {\bibinfo {volume}
  {303}},\ \bibinfo {pages} {2} (\bibinfo {year} {2003})}\BibitemShut {NoStop}%
\bibitem [{\citenamefont {Andersen}\ \emph {et~al.}(2020)\citenamefont
  {Andersen}, \citenamefont {Remm}, \citenamefont {Lazar}, \citenamefont
  {Krinner}, \citenamefont {Lacroix}, \citenamefont {Norris}, \citenamefont
  {Gabureac}, \citenamefont {Eichler},\ and\ \citenamefont
  {Wallraff}}]{andersen2020}%
  \BibitemOpen
  \bibfield  {author} {\bibinfo {author} {\bibfnamefont {C.~K.}\ \bibnamefont
  {Andersen}}, \bibinfo {author} {\bibfnamefont {A.}~\bibnamefont {Remm}},
  \bibinfo {author} {\bibfnamefont {S.}~\bibnamefont {Lazar}}, \bibinfo
  {author} {\bibfnamefont {S.}~\bibnamefont {Krinner}}, \bibinfo {author}
  {\bibfnamefont {N.}~\bibnamefont {Lacroix}}, \bibinfo {author} {\bibfnamefont
  {G.~J.}\ \bibnamefont {Norris}}, \bibinfo {author} {\bibfnamefont
  {M.}~\bibnamefont {Gabureac}}, \bibinfo {author} {\bibfnamefont
  {C.}~\bibnamefont {Eichler}},\ and\ \bibinfo {author} {\bibfnamefont
  {A.}~\bibnamefont {Wallraff}},\ }\bibfield  {title} {\bibinfo {title}
  {Repeated quantum error detection in a surface code},\ }\href
  {https://doi.org/10.1038/s41567-020-0920-y} {\bibfield  {journal} {\bibinfo
  {journal} {Nature Physics}\ }\textbf {\bibinfo {volume} {16}},\ \bibinfo
  {pages} {875} (\bibinfo {year} {2020})}\BibitemShut {NoStop}%
\bibitem [{\citenamefont {Satzinger}\ \emph {et~al.}(2021)\citenamefont
  {Satzinger}, \citenamefont {Liu}, \citenamefont {Smith}, \citenamefont
  {Knapp}, \citenamefont {Newman}, \citenamefont {Jones}, \citenamefont {Chen},
  \citenamefont {Quintana}, \citenamefont {Mi}, \citenamefont {Dunsworth},
  \citenamefont {Gidney}, \citenamefont {Aleiner}, \citenamefont {Arute},
  \citenamefont {Arya}, \citenamefont {Atalaya}, \citenamefont {Babbush},
  \citenamefont {Bardin}, \citenamefont {Barends}, \citenamefont {Basso},
  \citenamefont {Bengtsson}, \citenamefont {Bilmes}, \citenamefont {Broughton},
  \citenamefont {Buckley}, \citenamefont {Buell}, \citenamefont {Burkett},
  \citenamefont {Bushnell}, \citenamefont {Chiaro}, \citenamefont {Collins},
  \citenamefont {Courtney}, \citenamefont {Demura}, \citenamefont {Derk},
  \citenamefont {Eppens}, \citenamefont {Erickson}, \citenamefont {Faoro},
  \citenamefont {Farhi}, \citenamefont {Fowler}, \citenamefont {Foxen},
  \citenamefont {Giustina}, \citenamefont {Greene}, \citenamefont {Gross},
  \citenamefont {Harrigan}, \citenamefont {Harrington}, \citenamefont {Hilton},
  \citenamefont {Hong}, \citenamefont {Huang}, \citenamefont {Huggins},
  \citenamefont {Ioffe}, \citenamefont {Isakov}, \citenamefont {Jeffrey},
  \citenamefont {Jiang}, \citenamefont {Kafri}, \citenamefont {Kechedzhi},
  \citenamefont {Khattar}, \citenamefont {Kim}, \citenamefont {Klimov},
  \citenamefont {Korotkov}, \citenamefont {Kostritsa}, \citenamefont
  {Landhuis}, \citenamefont {Laptev}, \citenamefont {Locharla}, \citenamefont
  {Lucero}, \citenamefont {Martin}, \citenamefont {McClean}, \citenamefont
  {McEwen}, \citenamefont {Miao}, \citenamefont {Mohseni}, \citenamefont
  {Montazeri}, \citenamefont {Mruczkiewicz}, \citenamefont {Mutus},
  \citenamefont {Naaman}, \citenamefont {Neeley}, \citenamefont {Neill},
  \citenamefont {Niu}, \citenamefont {O'Brien}, \citenamefont {Opremcak},
  \citenamefont {Pat{\'o}}, \citenamefont {Petukhov}, \citenamefont {Rubin},
  \citenamefont {Sank}, \citenamefont {Shvarts}, \citenamefont {Strain},
  \citenamefont {Szalay}, \citenamefont {Villalonga}, \citenamefont {White},
  \citenamefont {Yao}, \citenamefont {Yeh}, \citenamefont {Yoo}, \citenamefont
  {Zalcman}, \citenamefont {Neven}, \citenamefont {Boixo}, \citenamefont
  {Megrant}, \citenamefont {Chen}, \citenamefont {Kelly}, \citenamefont
  {Smelyanskiy}, \citenamefont {Kitaev}, \citenamefont {Knap}, \citenamefont
  {Pollmann},\ and\ \citenamefont {Roushan}}]{satzinger2021}%
  \BibitemOpen
  \bibfield  {author} {\bibinfo {author} {\bibfnamefont {K.~J.}\ \bibnamefont
  {Satzinger}}, \bibinfo {author} {\bibfnamefont {Y.~J.}\ \bibnamefont {Liu}},
  \bibinfo {author} {\bibfnamefont {A.}~\bibnamefont {Smith}}, \bibinfo
  {author} {\bibfnamefont {C.}~\bibnamefont {Knapp}}, \bibinfo {author}
  {\bibfnamefont {M.}~\bibnamefont {Newman}}, \bibinfo {author} {\bibfnamefont
  {C.}~\bibnamefont {Jones}}, \bibinfo {author} {\bibfnamefont
  {Z.}~\bibnamefont {Chen}}, \bibinfo {author} {\bibfnamefont {C.}~\bibnamefont
  {Quintana}}, \bibinfo {author} {\bibfnamefont {X.}~\bibnamefont {Mi}},
  \bibinfo {author} {\bibfnamefont {A.}~\bibnamefont {Dunsworth}}, \bibinfo
  {author} {\bibfnamefont {C.}~\bibnamefont {Gidney}}, \bibinfo {author}
  {\bibfnamefont {I.}~\bibnamefont {Aleiner}}, \bibinfo {author} {\bibfnamefont
  {F.}~\bibnamefont {Arute}}, \bibinfo {author} {\bibfnamefont
  {K.}~\bibnamefont {Arya}}, \bibinfo {author} {\bibfnamefont {J.}~\bibnamefont
  {Atalaya}}, \bibinfo {author} {\bibfnamefont {R.}~\bibnamefont {Babbush}},
  \bibinfo {author} {\bibfnamefont {J.~C.}\ \bibnamefont {Bardin}}, \bibinfo
  {author} {\bibfnamefont {R.}~\bibnamefont {Barends}}, \bibinfo {author}
  {\bibfnamefont {J.}~\bibnamefont {Basso}}, \bibinfo {author} {\bibfnamefont
  {A.}~\bibnamefont {Bengtsson}}, \bibinfo {author} {\bibfnamefont
  {A.}~\bibnamefont {Bilmes}}, \bibinfo {author} {\bibfnamefont
  {M.}~\bibnamefont {Broughton}}, \bibinfo {author} {\bibfnamefont {B.~B.}\
  \bibnamefont {Buckley}}, \bibinfo {author} {\bibfnamefont {D.~A.}\
  \bibnamefont {Buell}}, \bibinfo {author} {\bibfnamefont {B.}~\bibnamefont
  {Burkett}}, \bibinfo {author} {\bibfnamefont {N.}~\bibnamefont {Bushnell}},
  \bibinfo {author} {\bibfnamefont {B.}~\bibnamefont {Chiaro}}, \bibinfo
  {author} {\bibfnamefont {R.}~\bibnamefont {Collins}}, \bibinfo {author}
  {\bibfnamefont {W.}~\bibnamefont {Courtney}}, \bibinfo {author}
  {\bibfnamefont {S.}~\bibnamefont {Demura}}, \bibinfo {author} {\bibfnamefont
  {A.~R.}\ \bibnamefont {Derk}}, \bibinfo {author} {\bibfnamefont
  {D.}~\bibnamefont {Eppens}}, \bibinfo {author} {\bibfnamefont
  {C.}~\bibnamefont {Erickson}}, \bibinfo {author} {\bibfnamefont
  {L.}~\bibnamefont {Faoro}}, \bibinfo {author} {\bibfnamefont
  {E.}~\bibnamefont {Farhi}}, \bibinfo {author} {\bibfnamefont {A.~G.}\
  \bibnamefont {Fowler}}, \bibinfo {author} {\bibfnamefont {B.}~\bibnamefont
  {Foxen}}, \bibinfo {author} {\bibfnamefont {M.}~\bibnamefont {Giustina}},
  \bibinfo {author} {\bibfnamefont {A.}~\bibnamefont {Greene}}, \bibinfo
  {author} {\bibfnamefont {J.~A.}\ \bibnamefont {Gross}}, \bibinfo {author}
  {\bibfnamefont {M.~P.}\ \bibnamefont {Harrigan}}, \bibinfo {author}
  {\bibfnamefont {S.~D.}\ \bibnamefont {Harrington}}, \bibinfo {author}
  {\bibfnamefont {J.}~\bibnamefont {Hilton}}, \bibinfo {author} {\bibfnamefont
  {S.}~\bibnamefont {Hong}}, \bibinfo {author} {\bibfnamefont {T.}~\bibnamefont
  {Huang}}, \bibinfo {author} {\bibfnamefont {W.~J.}\ \bibnamefont {Huggins}},
  \bibinfo {author} {\bibfnamefont {L.~B.}\ \bibnamefont {Ioffe}}, \bibinfo
  {author} {\bibfnamefont {S.~V.}\ \bibnamefont {Isakov}}, \bibinfo {author}
  {\bibfnamefont {E.}~\bibnamefont {Jeffrey}}, \bibinfo {author} {\bibfnamefont
  {Z.}~\bibnamefont {Jiang}}, \bibinfo {author} {\bibfnamefont
  {D.}~\bibnamefont {Kafri}}, \bibinfo {author} {\bibfnamefont
  {K.}~\bibnamefont {Kechedzhi}}, \bibinfo {author} {\bibfnamefont
  {T.}~\bibnamefont {Khattar}}, \bibinfo {author} {\bibfnamefont
  {S.}~\bibnamefont {Kim}}, \bibinfo {author} {\bibfnamefont {P.~V.}\
  \bibnamefont {Klimov}}, \bibinfo {author} {\bibfnamefont {A.~N.}\
  \bibnamefont {Korotkov}}, \bibinfo {author} {\bibfnamefont {F.}~\bibnamefont
  {Kostritsa}}, \bibinfo {author} {\bibfnamefont {D.}~\bibnamefont {Landhuis}},
  \bibinfo {author} {\bibfnamefont {P.}~\bibnamefont {Laptev}}, \bibinfo
  {author} {\bibfnamefont {A.}~\bibnamefont {Locharla}}, \bibinfo {author}
  {\bibfnamefont {E.}~\bibnamefont {Lucero}}, \bibinfo {author} {\bibfnamefont
  {O.}~\bibnamefont {Martin}}, \bibinfo {author} {\bibfnamefont {J.~R.}\
  \bibnamefont {McClean}}, \bibinfo {author} {\bibfnamefont {M.}~\bibnamefont
  {McEwen}}, \bibinfo {author} {\bibfnamefont {K.~C.}\ \bibnamefont {Miao}},
  \bibinfo {author} {\bibfnamefont {M.}~\bibnamefont {Mohseni}}, \bibinfo
  {author} {\bibfnamefont {S.}~\bibnamefont {Montazeri}}, \bibinfo {author}
  {\bibfnamefont {W.}~\bibnamefont {Mruczkiewicz}}, \bibinfo {author}
  {\bibfnamefont {J.}~\bibnamefont {Mutus}}, \bibinfo {author} {\bibfnamefont
  {O.}~\bibnamefont {Naaman}}, \bibinfo {author} {\bibfnamefont
  {M.}~\bibnamefont {Neeley}}, \bibinfo {author} {\bibfnamefont
  {C.}~\bibnamefont {Neill}}, \bibinfo {author} {\bibfnamefont {M.~Y.}\
  \bibnamefont {Niu}}, \bibinfo {author} {\bibfnamefont {T.~E.}\ \bibnamefont
  {O'Brien}}, \bibinfo {author} {\bibfnamefont {A.}~\bibnamefont {Opremcak}},
  \bibinfo {author} {\bibfnamefont {B.}~\bibnamefont {Pat{\'o}}}, \bibinfo
  {author} {\bibfnamefont {A.}~\bibnamefont {Petukhov}}, \bibinfo {author}
  {\bibfnamefont {N.~C.}\ \bibnamefont {Rubin}}, \bibinfo {author}
  {\bibfnamefont {D.}~\bibnamefont {Sank}}, \bibinfo {author} {\bibfnamefont
  {V.}~\bibnamefont {Shvarts}}, \bibinfo {author} {\bibfnamefont
  {D.}~\bibnamefont {Strain}}, \bibinfo {author} {\bibfnamefont
  {M.}~\bibnamefont {Szalay}}, \bibinfo {author} {\bibfnamefont
  {B.}~\bibnamefont {Villalonga}}, \bibinfo {author} {\bibfnamefont {T.~C.}\
  \bibnamefont {White}}, \bibinfo {author} {\bibfnamefont {Z.}~\bibnamefont
  {Yao}}, \bibinfo {author} {\bibfnamefont {P.}~\bibnamefont {Yeh}}, \bibinfo
  {author} {\bibfnamefont {J.}~\bibnamefont {Yoo}}, \bibinfo {author}
  {\bibfnamefont {A.}~\bibnamefont {Zalcman}}, \bibinfo {author} {\bibfnamefont
  {H.}~\bibnamefont {Neven}}, \bibinfo {author} {\bibfnamefont
  {S.}~\bibnamefont {Boixo}}, \bibinfo {author} {\bibfnamefont
  {A.}~\bibnamefont {Megrant}}, \bibinfo {author} {\bibfnamefont
  {Y.}~\bibnamefont {Chen}}, \bibinfo {author} {\bibfnamefont {J.}~\bibnamefont
  {Kelly}}, \bibinfo {author} {\bibfnamefont {V.}~\bibnamefont {Smelyanskiy}},
  \bibinfo {author} {\bibfnamefont {A.}~\bibnamefont {Kitaev}}, \bibinfo
  {author} {\bibfnamefont {M.}~\bibnamefont {Knap}}, \bibinfo {author}
  {\bibfnamefont {F.}~\bibnamefont {Pollmann}},\ and\ \bibinfo {author}
  {\bibfnamefont {P.}~\bibnamefont {Roushan}},\ }\bibfield  {title} {\bibinfo
  {title} {Realizing topologically ordered states on a quantum processor},\
  }\href {https://doi.org/10.1126/science.abi8378} {\bibfield  {journal}
  {\bibinfo  {journal} {Science}\ }\textbf {\bibinfo {volume} {374}},\ \bibinfo
  {pages} {1237} (\bibinfo {year} {2021})}\BibitemShut {NoStop}%
\bibitem [{\citenamefont {Marques}\ \emph {et~al.}(2022)\citenamefont
  {Marques}, \citenamefont {Varbanov}, \citenamefont {Moreira}, \citenamefont
  {Ali}, \citenamefont {Muthusubramanian}, \citenamefont {Zachariadis},
  \citenamefont {Battistel}, \citenamefont {Beekman}, \citenamefont {Haider},
  \citenamefont {Vlothuizen}, \citenamefont {Bruno}, \citenamefont {Terhal},\
  and\ \citenamefont {DiCarlo}}]{marques2022}%
  \BibitemOpen
  \bibfield  {author} {\bibinfo {author} {\bibfnamefont {J.~F.}\ \bibnamefont
  {Marques}}, \bibinfo {author} {\bibfnamefont {B.~M.}\ \bibnamefont
  {Varbanov}}, \bibinfo {author} {\bibfnamefont {M.~S.}\ \bibnamefont
  {Moreira}}, \bibinfo {author} {\bibfnamefont {H.}~\bibnamefont {Ali}},
  \bibinfo {author} {\bibfnamefont {N.}~\bibnamefont {Muthusubramanian}},
  \bibinfo {author} {\bibfnamefont {C.}~\bibnamefont {Zachariadis}}, \bibinfo
  {author} {\bibfnamefont {F.}~\bibnamefont {Battistel}}, \bibinfo {author}
  {\bibfnamefont {M.}~\bibnamefont {Beekman}}, \bibinfo {author} {\bibfnamefont
  {N.}~\bibnamefont {Haider}}, \bibinfo {author} {\bibfnamefont
  {W.}~\bibnamefont {Vlothuizen}}, \bibinfo {author} {\bibfnamefont
  {A.}~\bibnamefont {Bruno}}, \bibinfo {author} {\bibfnamefont {B.~M.}\
  \bibnamefont {Terhal}},\ and\ \bibinfo {author} {\bibfnamefont
  {L.}~\bibnamefont {DiCarlo}},\ }\bibfield  {title} {\bibinfo {title}
  {Logical-qubit operations in an error-detecting surface code},\ }\href
  {https://doi.org/10.1038/s41567-021-01423-9} {\bibfield  {journal} {\bibinfo
  {journal} {Nature Physics}\ }\textbf {\bibinfo {volume} {18}},\ \bibinfo
  {pages} {80} (\bibinfo {year} {2022})}\BibitemShut {NoStop}%
\bibitem [{\citenamefont {Krinner}\ \emph {et~al.}(2022)\citenamefont
  {Krinner}, \citenamefont {Lacroix}, \citenamefont {Remm}, \citenamefont
  {Di~Paolo}, \citenamefont {Genois}, \citenamefont {Leroux}, \citenamefont
  {Hellings}, \citenamefont {Lazar}, \citenamefont {Swiadek}, \citenamefont
  {Herrmann}, \citenamefont {Norris}, \citenamefont {Andersen}, \citenamefont
  {M{\"u}ller}, \citenamefont {Blais}, \citenamefont {Eichler},\ and\
  \citenamefont {Wallraff}}]{krinner2022}%
  \BibitemOpen
  \bibfield  {author} {\bibinfo {author} {\bibfnamefont {S.}~\bibnamefont
  {Krinner}}, \bibinfo {author} {\bibfnamefont {N.}~\bibnamefont {Lacroix}},
  \bibinfo {author} {\bibfnamefont {A.}~\bibnamefont {Remm}}, \bibinfo {author}
  {\bibfnamefont {A.}~\bibnamefont {Di~Paolo}}, \bibinfo {author}
  {\bibfnamefont {E.}~\bibnamefont {Genois}}, \bibinfo {author} {\bibfnamefont
  {C.}~\bibnamefont {Leroux}}, \bibinfo {author} {\bibfnamefont
  {C.}~\bibnamefont {Hellings}}, \bibinfo {author} {\bibfnamefont
  {S.}~\bibnamefont {Lazar}}, \bibinfo {author} {\bibfnamefont
  {F.}~\bibnamefont {Swiadek}}, \bibinfo {author} {\bibfnamefont
  {J.}~\bibnamefont {Herrmann}}, \bibinfo {author} {\bibfnamefont {G.~J.}\
  \bibnamefont {Norris}}, \bibinfo {author} {\bibfnamefont {C.~K.}\
  \bibnamefont {Andersen}}, \bibinfo {author} {\bibfnamefont {M.}~\bibnamefont
  {M{\"u}ller}}, \bibinfo {author} {\bibfnamefont {A.}~\bibnamefont {Blais}},
  \bibinfo {author} {\bibfnamefont {C.}~\bibnamefont {Eichler}},\ and\ \bibinfo
  {author} {\bibfnamefont {A.}~\bibnamefont {Wallraff}},\ }\bibfield  {title}
  {\bibinfo {title} {Realizing repeated quantum error correction in a
  distance-three surface code},\ }\href
  {https://doi.org/10.1038/s41586-022-04566-8} {\bibfield  {journal} {\bibinfo
  {journal} {Nature}\ }\textbf {\bibinfo {volume} {605}},\ \bibinfo {pages}
  {669} (\bibinfo {year} {2022})}\BibitemShut {NoStop}%
\bibitem [{\citenamefont {Zhao}\ \emph {et~al.}(2022)\citenamefont {Zhao},
  \citenamefont {Ye}, \citenamefont {Huang}, \citenamefont {Zhang},
  \citenamefont {Wu}, \citenamefont {Guan}, \citenamefont {Zhu}, \citenamefont
  {Wei}, \citenamefont {He}, \citenamefont {Cao}, \citenamefont {Chen},
  \citenamefont {Chung}, \citenamefont {Deng}, \citenamefont {Fan},
  \citenamefont {Gong}, \citenamefont {Guo}, \citenamefont {Guo}, \citenamefont
  {Han}, \citenamefont {Li}, \citenamefont {Li}, \citenamefont {Li},
  \citenamefont {Liang}, \citenamefont {Lin}, \citenamefont {Qian},
  \citenamefont {Rong}, \citenamefont {Su}, \citenamefont {Sun}, \citenamefont
  {Wang}, \citenamefont {Wu}, \citenamefont {Xu}, \citenamefont {Ying},
  \citenamefont {Yu}, \citenamefont {Zha}, \citenamefont {Zhang}, \citenamefont
  {Huo}, \citenamefont {Lu}, \citenamefont {Peng}, \citenamefont {Zhu},\ and\
  \citenamefont {Pan}}]{zhao2022}%
  \BibitemOpen
  \bibfield  {author} {\bibinfo {author} {\bibfnamefont {Y.}~\bibnamefont
  {Zhao}}, \bibinfo {author} {\bibfnamefont {Y.}~\bibnamefont {Ye}}, \bibinfo
  {author} {\bibfnamefont {H.-L.}\ \bibnamefont {Huang}}, \bibinfo {author}
  {\bibfnamefont {Y.}~\bibnamefont {Zhang}}, \bibinfo {author} {\bibfnamefont
  {D.}~\bibnamefont {Wu}}, \bibinfo {author} {\bibfnamefont {H.}~\bibnamefont
  {Guan}}, \bibinfo {author} {\bibfnamefont {Q.}~\bibnamefont {Zhu}}, \bibinfo
  {author} {\bibfnamefont {Z.}~\bibnamefont {Wei}}, \bibinfo {author}
  {\bibfnamefont {T.}~\bibnamefont {He}}, \bibinfo {author} {\bibfnamefont
  {S.}~\bibnamefont {Cao}}, \bibinfo {author} {\bibfnamefont {F.}~\bibnamefont
  {Chen}}, \bibinfo {author} {\bibfnamefont {T.-H.}\ \bibnamefont {Chung}},
  \bibinfo {author} {\bibfnamefont {H.}~\bibnamefont {Deng}}, \bibinfo {author}
  {\bibfnamefont {D.}~\bibnamefont {Fan}}, \bibinfo {author} {\bibfnamefont
  {M.}~\bibnamefont {Gong}}, \bibinfo {author} {\bibfnamefont {C.}~\bibnamefont
  {Guo}}, \bibinfo {author} {\bibfnamefont {S.}~\bibnamefont {Guo}}, \bibinfo
  {author} {\bibfnamefont {L.}~\bibnamefont {Han}}, \bibinfo {author}
  {\bibfnamefont {N.}~\bibnamefont {Li}}, \bibinfo {author} {\bibfnamefont
  {S.}~\bibnamefont {Li}}, \bibinfo {author} {\bibfnamefont {Y.}~\bibnamefont
  {Li}}, \bibinfo {author} {\bibfnamefont {F.}~\bibnamefont {Liang}}, \bibinfo
  {author} {\bibfnamefont {J.}~\bibnamefont {Lin}}, \bibinfo {author}
  {\bibfnamefont {H.}~\bibnamefont {Qian}}, \bibinfo {author} {\bibfnamefont
  {H.}~\bibnamefont {Rong}}, \bibinfo {author} {\bibfnamefont {H.}~\bibnamefont
  {Su}}, \bibinfo {author} {\bibfnamefont {L.}~\bibnamefont {Sun}}, \bibinfo
  {author} {\bibfnamefont {S.}~\bibnamefont {Wang}}, \bibinfo {author}
  {\bibfnamefont {Y.}~\bibnamefont {Wu}}, \bibinfo {author} {\bibfnamefont
  {Y.}~\bibnamefont {Xu}}, \bibinfo {author} {\bibfnamefont {C.}~\bibnamefont
  {Ying}}, \bibinfo {author} {\bibfnamefont {J.}~\bibnamefont {Yu}}, \bibinfo
  {author} {\bibfnamefont {C.}~\bibnamefont {Zha}}, \bibinfo {author}
  {\bibfnamefont {K.}~\bibnamefont {Zhang}}, \bibinfo {author} {\bibfnamefont
  {Y.-H.}\ \bibnamefont {Huo}}, \bibinfo {author} {\bibfnamefont {C.-Y.}\
  \bibnamefont {Lu}}, \bibinfo {author} {\bibfnamefont {C.-Z.}\ \bibnamefont
  {Peng}}, \bibinfo {author} {\bibfnamefont {X.}~\bibnamefont {Zhu}},\ and\
  \bibinfo {author} {\bibfnamefont {J.-W.}\ \bibnamefont {Pan}},\ }\bibfield
  {title} {\bibinfo {title} {Realization of an {{Error-Correcting Surface
  Code}} with {{Superconducting Qubits}}},\ }\href
  {https://doi.org/10.1103/PhysRevLett.129.030501} {\bibfield  {journal}
  {\bibinfo  {journal} {Physical Review Letters}\ }\textbf {\bibinfo {volume}
  {129}},\ \bibinfo {pages} {030501} (\bibinfo {year} {2022})}\BibitemShut
  {NoStop}%
\bibitem [{\citenamefont {Liu}\ \emph {et~al.}(2022)\citenamefont {Liu},
  \citenamefont {Shtengel}, \citenamefont {Smith},\ and\ \citenamefont
  {Pollmann}}]{liu2022}%
  \BibitemOpen
  \bibfield  {author} {\bibinfo {author} {\bibfnamefont {Y.-J.}\ \bibnamefont
  {Liu}}, \bibinfo {author} {\bibfnamefont {K.}~\bibnamefont {Shtengel}},
  \bibinfo {author} {\bibfnamefont {A.}~\bibnamefont {Smith}},\ and\ \bibinfo
  {author} {\bibfnamefont {F.}~\bibnamefont {Pollmann}},\ }\bibfield  {title}
  {\bibinfo {title} {Methods for {{Simulating String-Net States}} and
  {{Anyons}} on a {{Digital Quantum Computer}}},\ }\href
  {https://doi.org/10.1103/PRXQuantum.3.040315} {\bibfield  {journal} {\bibinfo
   {journal} {PRX Quantum}\ }\textbf {\bibinfo {volume} {3}},\ \bibinfo {pages}
  {040315} (\bibinfo {year} {2022})}\BibitemShut {NoStop}%
\bibitem [{\citenamefont {{Google Quantum AI}}\ \emph
  {et~al.}(2023)\citenamefont {{Google Quantum AI}}, \citenamefont {Acharya},
  \citenamefont {Aleiner}, \citenamefont {Allen}, \citenamefont {Andersen},
  \citenamefont {Ansmann}, \citenamefont {Arute}, \citenamefont {Arya},
  \citenamefont {Asfaw}, \citenamefont {Atalaya}, \citenamefont {Babbush},
  \citenamefont {Bacon}, \citenamefont {Bardin}, \citenamefont {Basso},
  \citenamefont {Bengtsson}, \citenamefont {Boixo}, \citenamefont {Bortoli},
  \citenamefont {Bourassa}, \citenamefont {Bovaird}, \citenamefont {Brill},
  \citenamefont {Broughton}, \citenamefont {Buckley}, \citenamefont {Buell},
  \citenamefont {Burger}, \citenamefont {Burkett}, \citenamefont {Bushnell},
  \citenamefont {Chen}, \citenamefont {Chen}, \citenamefont {Chiaro},
  \citenamefont {Cogan}, \citenamefont {Collins}, \citenamefont {Conner},
  \citenamefont {Courtney}, \citenamefont {Crook}, \citenamefont {Curtin},
  \citenamefont {Debroy}, \citenamefont {Del Toro~Barba}, \citenamefont
  {Demura}, \citenamefont {Dunsworth}, \citenamefont {Eppens}, \citenamefont
  {Erickson}, \citenamefont {Faoro}, \citenamefont {Farhi}, \citenamefont
  {Fatemi}, \citenamefont {Flores~Burgos}, \citenamefont {Forati},
  \citenamefont {Fowler}, \citenamefont {Foxen}, \citenamefont {Giang},
  \citenamefont {Gidney}, \citenamefont {Gilboa}, \citenamefont {Giustina},
  \citenamefont {Grajales~Dau}, \citenamefont {Gross}, \citenamefont
  {Habegger}, \citenamefont {Hamilton}, \citenamefont {Harrigan}, \citenamefont
  {Harrington}, \citenamefont {Higgott}, \citenamefont {Hilton}, \citenamefont
  {Hoffmann}, \citenamefont {Hong}, \citenamefont {Huang}, \citenamefont
  {Huff}, \citenamefont {Huggins}, \citenamefont {Ioffe}, \citenamefont
  {Isakov}, \citenamefont {Iveland}, \citenamefont {Jeffrey}, \citenamefont
  {Jiang}, \citenamefont {Jones}, \citenamefont {Juhas}, \citenamefont {Kafri},
  \citenamefont {Kechedzhi}, \citenamefont {Kelly}, \citenamefont {Khattar},
  \citenamefont {Khezri}, \citenamefont {Kieferov{\'a}}, \citenamefont {Kim},
  \citenamefont {Kitaev}, \citenamefont {Klimov}, \citenamefont {Klots},
  \citenamefont {Korotkov}, \citenamefont {Kostritsa}, \citenamefont
  {Kreikebaum}, \citenamefont {Landhuis}, \citenamefont {Laptev}, \citenamefont
  {Lau}, \citenamefont {Laws}, \citenamefont {Lee}, \citenamefont {Lee},
  \citenamefont {Lester}, \citenamefont {Lill}, \citenamefont {Liu},
  \citenamefont {Locharla}, \citenamefont {Lucero}, \citenamefont {Malone},
  \citenamefont {Marshall}, \citenamefont {Martin}, \citenamefont {McClean},
  \citenamefont {McCourt}, \citenamefont {McEwen}, \citenamefont {Megrant},
  \citenamefont {Meurer~Costa}, \citenamefont {Mi}, \citenamefont {Miao},
  \citenamefont {Mohseni}, \citenamefont {Montazeri}, \citenamefont {Morvan},
  \citenamefont {Mount}, \citenamefont {Mruczkiewicz}, \citenamefont {Naaman},
  \citenamefont {Neeley}, \citenamefont {Neill}, \citenamefont {Nersisyan},
  \citenamefont {Neven}, \citenamefont {Newman}, \citenamefont {Ng},
  \citenamefont {Nguyen}, \citenamefont {Nguyen}, \citenamefont {Niu},
  \citenamefont {O'Brien}, \citenamefont {Opremcak}, \citenamefont {Platt},
  \citenamefont {Petukhov}, \citenamefont {Potter}, \citenamefont {Pryadko},
  \citenamefont {Quintana}, \citenamefont {Roushan}, \citenamefont {Rubin},
  \citenamefont {Saei}, \citenamefont {Sank}, \citenamefont {Sankaragomathi},
  \citenamefont {Satzinger}, \citenamefont {Schurkus}, \citenamefont
  {Schuster}, \citenamefont {Shearn}, \citenamefont {Shorter}, \citenamefont
  {Shvarts}, \citenamefont {Skruzny}, \citenamefont {Smelyanskiy},
  \citenamefont {Smith}, \citenamefont {Sterling}, \citenamefont {Strain},
  \citenamefont {Szalay}, \citenamefont {Torres}, \citenamefont {Vidal},
  \citenamefont {Villalonga}, \citenamefont {Vollgraff~Heidweiller},
  \citenamefont {White}, \citenamefont {Xing}, \citenamefont {Yao},
  \citenamefont {Yeh}, \citenamefont {Yoo}, \citenamefont {Young},
  \citenamefont {Zalcman}, \citenamefont {Zhang},\ and\ \citenamefont
  {Zhu}}]{googlequantumai2023}%
  \BibitemOpen
  \bibfield  {author} {\bibinfo {author} {\bibnamefont {{Google Quantum AI}}},
  \bibinfo {author} {\bibfnamefont {R.}~\bibnamefont {Acharya}}, \bibinfo
  {author} {\bibfnamefont {I.}~\bibnamefont {Aleiner}}, \bibinfo {author}
  {\bibfnamefont {R.}~\bibnamefont {Allen}}, \bibinfo {author} {\bibfnamefont
  {T.~I.}\ \bibnamefont {Andersen}}, \bibinfo {author} {\bibfnamefont
  {M.}~\bibnamefont {Ansmann}}, \bibinfo {author} {\bibfnamefont
  {F.}~\bibnamefont {Arute}}, \bibinfo {author} {\bibfnamefont
  {K.}~\bibnamefont {Arya}}, \bibinfo {author} {\bibfnamefont {A.}~\bibnamefont
  {Asfaw}}, \bibinfo {author} {\bibfnamefont {J.}~\bibnamefont {Atalaya}},
  \bibinfo {author} {\bibfnamefont {R.}~\bibnamefont {Babbush}}, \bibinfo
  {author} {\bibfnamefont {D.}~\bibnamefont {Bacon}}, \bibinfo {author}
  {\bibfnamefont {J.~C.}\ \bibnamefont {Bardin}}, \bibinfo {author}
  {\bibfnamefont {J.}~\bibnamefont {Basso}}, \bibinfo {author} {\bibfnamefont
  {A.}~\bibnamefont {Bengtsson}}, \bibinfo {author} {\bibfnamefont
  {S.}~\bibnamefont {Boixo}}, \bibinfo {author} {\bibfnamefont
  {G.}~\bibnamefont {Bortoli}}, \bibinfo {author} {\bibfnamefont
  {A.}~\bibnamefont {Bourassa}}, \bibinfo {author} {\bibfnamefont
  {J.}~\bibnamefont {Bovaird}}, \bibinfo {author} {\bibfnamefont
  {L.}~\bibnamefont {Brill}}, \bibinfo {author} {\bibfnamefont
  {M.}~\bibnamefont {Broughton}}, \bibinfo {author} {\bibfnamefont {B.~B.}\
  \bibnamefont {Buckley}}, \bibinfo {author} {\bibfnamefont {D.~A.}\
  \bibnamefont {Buell}}, \bibinfo {author} {\bibfnamefont {T.}~\bibnamefont
  {Burger}}, \bibinfo {author} {\bibfnamefont {B.}~\bibnamefont {Burkett}},
  \bibinfo {author} {\bibfnamefont {N.}~\bibnamefont {Bushnell}}, \bibinfo
  {author} {\bibfnamefont {Y.}~\bibnamefont {Chen}}, \bibinfo {author}
  {\bibfnamefont {Z.}~\bibnamefont {Chen}}, \bibinfo {author} {\bibfnamefont
  {B.}~\bibnamefont {Chiaro}}, \bibinfo {author} {\bibfnamefont
  {J.}~\bibnamefont {Cogan}}, \bibinfo {author} {\bibfnamefont
  {R.}~\bibnamefont {Collins}}, \bibinfo {author} {\bibfnamefont
  {P.}~\bibnamefont {Conner}}, \bibinfo {author} {\bibfnamefont
  {W.}~\bibnamefont {Courtney}}, \bibinfo {author} {\bibfnamefont {A.~L.}\
  \bibnamefont {Crook}}, \bibinfo {author} {\bibfnamefont {B.}~\bibnamefont
  {Curtin}}, \bibinfo {author} {\bibfnamefont {D.~M.}\ \bibnamefont {Debroy}},
  \bibinfo {author} {\bibfnamefont {A.}~\bibnamefont {Del Toro~Barba}},
  \bibinfo {author} {\bibfnamefont {S.}~\bibnamefont {Demura}}, \bibinfo
  {author} {\bibfnamefont {A.}~\bibnamefont {Dunsworth}}, \bibinfo {author}
  {\bibfnamefont {D.}~\bibnamefont {Eppens}}, \bibinfo {author} {\bibfnamefont
  {C.}~\bibnamefont {Erickson}}, \bibinfo {author} {\bibfnamefont
  {L.}~\bibnamefont {Faoro}}, \bibinfo {author} {\bibfnamefont
  {E.}~\bibnamefont {Farhi}}, \bibinfo {author} {\bibfnamefont
  {R.}~\bibnamefont {Fatemi}}, \bibinfo {author} {\bibfnamefont
  {L.}~\bibnamefont {Flores~Burgos}}, \bibinfo {author} {\bibfnamefont
  {E.}~\bibnamefont {Forati}}, \bibinfo {author} {\bibfnamefont {A.~G.}\
  \bibnamefont {Fowler}}, \bibinfo {author} {\bibfnamefont {B.}~\bibnamefont
  {Foxen}}, \bibinfo {author} {\bibfnamefont {W.}~\bibnamefont {Giang}},
  \bibinfo {author} {\bibfnamefont {C.}~\bibnamefont {Gidney}}, \bibinfo
  {author} {\bibfnamefont {D.}~\bibnamefont {Gilboa}}, \bibinfo {author}
  {\bibfnamefont {M.}~\bibnamefont {Giustina}}, \bibinfo {author}
  {\bibfnamefont {A.}~\bibnamefont {Grajales~Dau}}, \bibinfo {author}
  {\bibfnamefont {J.~A.}\ \bibnamefont {Gross}}, \bibinfo {author}
  {\bibfnamefont {S.}~\bibnamefont {Habegger}}, \bibinfo {author}
  {\bibfnamefont {M.~C.}\ \bibnamefont {Hamilton}}, \bibinfo {author}
  {\bibfnamefont {M.~P.}\ \bibnamefont {Harrigan}}, \bibinfo {author}
  {\bibfnamefont {S.~D.}\ \bibnamefont {Harrington}}, \bibinfo {author}
  {\bibfnamefont {O.}~\bibnamefont {Higgott}}, \bibinfo {author} {\bibfnamefont
  {J.}~\bibnamefont {Hilton}}, \bibinfo {author} {\bibfnamefont
  {M.}~\bibnamefont {Hoffmann}}, \bibinfo {author} {\bibfnamefont
  {S.}~\bibnamefont {Hong}}, \bibinfo {author} {\bibfnamefont {T.}~\bibnamefont
  {Huang}}, \bibinfo {author} {\bibfnamefont {A.}~\bibnamefont {Huff}},
  \bibinfo {author} {\bibfnamefont {W.~J.}\ \bibnamefont {Huggins}}, \bibinfo
  {author} {\bibfnamefont {L.~B.}\ \bibnamefont {Ioffe}}, \bibinfo {author}
  {\bibfnamefont {S.~V.}\ \bibnamefont {Isakov}}, \bibinfo {author}
  {\bibfnamefont {J.}~\bibnamefont {Iveland}}, \bibinfo {author} {\bibfnamefont
  {E.}~\bibnamefont {Jeffrey}}, \bibinfo {author} {\bibfnamefont
  {Z.}~\bibnamefont {Jiang}}, \bibinfo {author} {\bibfnamefont
  {C.}~\bibnamefont {Jones}}, \bibinfo {author} {\bibfnamefont
  {P.}~\bibnamefont {Juhas}}, \bibinfo {author} {\bibfnamefont
  {D.}~\bibnamefont {Kafri}}, \bibinfo {author} {\bibfnamefont
  {K.}~\bibnamefont {Kechedzhi}}, \bibinfo {author} {\bibfnamefont
  {J.}~\bibnamefont {Kelly}}, \bibinfo {author} {\bibfnamefont
  {T.}~\bibnamefont {Khattar}}, \bibinfo {author} {\bibfnamefont
  {M.}~\bibnamefont {Khezri}}, \bibinfo {author} {\bibfnamefont
  {M.}~\bibnamefont {Kieferov{\'a}}}, \bibinfo {author} {\bibfnamefont
  {S.}~\bibnamefont {Kim}}, \bibinfo {author} {\bibfnamefont {A.}~\bibnamefont
  {Kitaev}}, \bibinfo {author} {\bibfnamefont {P.~V.}\ \bibnamefont {Klimov}},
  \bibinfo {author} {\bibfnamefont {A.~R.}\ \bibnamefont {Klots}}, \bibinfo
  {author} {\bibfnamefont {A.~N.}\ \bibnamefont {Korotkov}}, \bibinfo {author}
  {\bibfnamefont {F.}~\bibnamefont {Kostritsa}}, \bibinfo {author}
  {\bibfnamefont {J.~M.}\ \bibnamefont {Kreikebaum}}, \bibinfo {author}
  {\bibfnamefont {D.}~\bibnamefont {Landhuis}}, \bibinfo {author}
  {\bibfnamefont {P.}~\bibnamefont {Laptev}}, \bibinfo {author} {\bibfnamefont
  {K.-M.}\ \bibnamefont {Lau}}, \bibinfo {author} {\bibfnamefont
  {L.}~\bibnamefont {Laws}}, \bibinfo {author} {\bibfnamefont {J.}~\bibnamefont
  {Lee}}, \bibinfo {author} {\bibfnamefont {K.}~\bibnamefont {Lee}}, \bibinfo
  {author} {\bibfnamefont {B.~J.}\ \bibnamefont {Lester}}, \bibinfo {author}
  {\bibfnamefont {A.}~\bibnamefont {Lill}}, \bibinfo {author} {\bibfnamefont
  {W.}~\bibnamefont {Liu}}, \bibinfo {author} {\bibfnamefont {A.}~\bibnamefont
  {Locharla}}, \bibinfo {author} {\bibfnamefont {E.}~\bibnamefont {Lucero}},
  \bibinfo {author} {\bibfnamefont {F.~D.}\ \bibnamefont {Malone}}, \bibinfo
  {author} {\bibfnamefont {J.}~\bibnamefont {Marshall}}, \bibinfo {author}
  {\bibfnamefont {O.}~\bibnamefont {Martin}}, \bibinfo {author} {\bibfnamefont
  {J.~R.}\ \bibnamefont {McClean}}, \bibinfo {author} {\bibfnamefont
  {T.}~\bibnamefont {McCourt}}, \bibinfo {author} {\bibfnamefont
  {M.}~\bibnamefont {McEwen}}, \bibinfo {author} {\bibfnamefont
  {A.}~\bibnamefont {Megrant}}, \bibinfo {author} {\bibfnamefont
  {B.}~\bibnamefont {Meurer~Costa}}, \bibinfo {author} {\bibfnamefont
  {X.}~\bibnamefont {Mi}}, \bibinfo {author} {\bibfnamefont {K.~C.}\
  \bibnamefont {Miao}}, \bibinfo {author} {\bibfnamefont {M.}~\bibnamefont
  {Mohseni}}, \bibinfo {author} {\bibfnamefont {S.}~\bibnamefont {Montazeri}},
  \bibinfo {author} {\bibfnamefont {A.}~\bibnamefont {Morvan}}, \bibinfo
  {author} {\bibfnamefont {E.}~\bibnamefont {Mount}}, \bibinfo {author}
  {\bibfnamefont {W.}~\bibnamefont {Mruczkiewicz}}, \bibinfo {author}
  {\bibfnamefont {O.}~\bibnamefont {Naaman}}, \bibinfo {author} {\bibfnamefont
  {M.}~\bibnamefont {Neeley}}, \bibinfo {author} {\bibfnamefont
  {C.}~\bibnamefont {Neill}}, \bibinfo {author} {\bibfnamefont
  {A.}~\bibnamefont {Nersisyan}}, \bibinfo {author} {\bibfnamefont
  {H.}~\bibnamefont {Neven}}, \bibinfo {author} {\bibfnamefont
  {M.}~\bibnamefont {Newman}}, \bibinfo {author} {\bibfnamefont {J.~H.}\
  \bibnamefont {Ng}}, \bibinfo {author} {\bibfnamefont {A.}~\bibnamefont
  {Nguyen}}, \bibinfo {author} {\bibfnamefont {M.}~\bibnamefont {Nguyen}},
  \bibinfo {author} {\bibfnamefont {M.~Y.}\ \bibnamefont {Niu}}, \bibinfo
  {author} {\bibfnamefont {T.~E.}\ \bibnamefont {O'Brien}}, \bibinfo {author}
  {\bibfnamefont {A.}~\bibnamefont {Opremcak}}, \bibinfo {author}
  {\bibfnamefont {J.}~\bibnamefont {Platt}}, \bibinfo {author} {\bibfnamefont
  {A.}~\bibnamefont {Petukhov}}, \bibinfo {author} {\bibfnamefont
  {R.}~\bibnamefont {Potter}}, \bibinfo {author} {\bibfnamefont {L.~P.}\
  \bibnamefont {Pryadko}}, \bibinfo {author} {\bibfnamefont {C.}~\bibnamefont
  {Quintana}}, \bibinfo {author} {\bibfnamefont {P.}~\bibnamefont {Roushan}},
  \bibinfo {author} {\bibfnamefont {N.~C.}\ \bibnamefont {Rubin}}, \bibinfo
  {author} {\bibfnamefont {N.}~\bibnamefont {Saei}}, \bibinfo {author}
  {\bibfnamefont {D.}~\bibnamefont {Sank}}, \bibinfo {author} {\bibfnamefont
  {K.}~\bibnamefont {Sankaragomathi}}, \bibinfo {author} {\bibfnamefont
  {K.~J.}\ \bibnamefont {Satzinger}}, \bibinfo {author} {\bibfnamefont {H.~F.}\
  \bibnamefont {Schurkus}}, \bibinfo {author} {\bibfnamefont {C.}~\bibnamefont
  {Schuster}}, \bibinfo {author} {\bibfnamefont {M.~J.}\ \bibnamefont
  {Shearn}}, \bibinfo {author} {\bibfnamefont {A.}~\bibnamefont {Shorter}},
  \bibinfo {author} {\bibfnamefont {V.}~\bibnamefont {Shvarts}}, \bibinfo
  {author} {\bibfnamefont {J.}~\bibnamefont {Skruzny}}, \bibinfo {author}
  {\bibfnamefont {V.}~\bibnamefont {Smelyanskiy}}, \bibinfo {author}
  {\bibfnamefont {W.~C.}\ \bibnamefont {Smith}}, \bibinfo {author}
  {\bibfnamefont {G.}~\bibnamefont {Sterling}}, \bibinfo {author}
  {\bibfnamefont {D.}~\bibnamefont {Strain}}, \bibinfo {author} {\bibfnamefont
  {M.}~\bibnamefont {Szalay}}, \bibinfo {author} {\bibfnamefont
  {A.}~\bibnamefont {Torres}}, \bibinfo {author} {\bibfnamefont
  {G.}~\bibnamefont {Vidal}}, \bibinfo {author} {\bibfnamefont
  {B.}~\bibnamefont {Villalonga}}, \bibinfo {author} {\bibfnamefont
  {C.}~\bibnamefont {Vollgraff~Heidweiller}}, \bibinfo {author} {\bibfnamefont
  {T.}~\bibnamefont {White}}, \bibinfo {author} {\bibfnamefont
  {C.}~\bibnamefont {Xing}}, \bibinfo {author} {\bibfnamefont {Z.~J.}\
  \bibnamefont {Yao}}, \bibinfo {author} {\bibfnamefont {P.}~\bibnamefont
  {Yeh}}, \bibinfo {author} {\bibfnamefont {J.}~\bibnamefont {Yoo}}, \bibinfo
  {author} {\bibfnamefont {G.}~\bibnamefont {Young}}, \bibinfo {author}
  {\bibfnamefont {A.}~\bibnamefont {Zalcman}}, \bibinfo {author} {\bibfnamefont
  {Y.}~\bibnamefont {Zhang}},\ and\ \bibinfo {author} {\bibfnamefont
  {N.}~\bibnamefont {Zhu}},\ }\bibfield  {title} {\bibinfo {title} {Suppressing
  quantum errors by scaling a surface code logical qubit},\ }\href
  {https://doi.org/10.1038/s41586-022-05434-1} {\bibfield  {journal} {\bibinfo
  {journal} {Nature}\ }\textbf {\bibinfo {volume} {614}},\ \bibinfo {pages}
  {676} (\bibinfo {year} {2023})}\BibitemShut {NoStop}%
\bibitem [{\citenamefont {Zeng}\ \emph {et~al.}(2019)\citenamefont {Zeng},
  \citenamefont {Chen}, \citenamefont {Zhou},\ and\ \citenamefont
  {Wen}}]{zeng2019}%
  \BibitemOpen
  \bibfield  {author} {\bibinfo {author} {\bibfnamefont {B.}~\bibnamefont
  {Zeng}}, \bibinfo {author} {\bibfnamefont {X.}~\bibnamefont {Chen}}, \bibinfo
  {author} {\bibfnamefont {D.-L.}\ \bibnamefont {Zhou}},\ and\ \bibinfo
  {author} {\bibfnamefont {X.-G.}\ \bibnamefont {Wen}},\ }\href@noop {} {\emph
  {\bibinfo {title} {Quantum {{Information Meets Quantum Matter}} -- {{From
  Quantum Entanglement}} to {{Topological Phase}} in {{Many-Body Systems}}}}}\
  (\bibinfo  {publisher} {Springer},\ \bibinfo {year} {2019})\BibitemShut
  {NoStop}%
\bibitem [{\citenamefont {Low}\ \emph {et~al.}(2020)\citenamefont {Low},
  \citenamefont {White}, \citenamefont {Cox}, \citenamefont {Day},\ and\
  \citenamefont {Senko}}]{low2020}%
  \BibitemOpen
  \bibfield  {author} {\bibinfo {author} {\bibfnamefont {P.~J.}\ \bibnamefont
  {Low}}, \bibinfo {author} {\bibfnamefont {B.~M.}\ \bibnamefont {White}},
  \bibinfo {author} {\bibfnamefont {A.~A.}\ \bibnamefont {Cox}}, \bibinfo
  {author} {\bibfnamefont {M.~L.}\ \bibnamefont {Day}},\ and\ \bibinfo {author}
  {\bibfnamefont {C.}~\bibnamefont {Senko}},\ }\bibfield  {title} {\bibinfo
  {title} {Practical trapped-ion protocols for universal qudit-based quantum
  computing},\ }\href {https://doi.org/10.1103/PhysRevResearch.2.033128}
  {\bibfield  {journal} {\bibinfo  {journal} {Physical Review Research}\
  }\textbf {\bibinfo {volume} {2}},\ \bibinfo {pages} {033128} (\bibinfo {year}
  {2020})}\BibitemShut {NoStop}%
\bibitem [{\citenamefont {Chi}\ \emph {et~al.}(2022)\citenamefont {Chi},
  \citenamefont {Huang}, \citenamefont {Zhang}, \citenamefont {Mao},
  \citenamefont {Zhou}, \citenamefont {Chen}, \citenamefont {Zhai},
  \citenamefont {Bao}, \citenamefont {Dai}, \citenamefont {Yuan}, \citenamefont
  {Zhang}, \citenamefont {Dai}, \citenamefont {Tang}, \citenamefont {Yang},
  \citenamefont {Li}, \citenamefont {Ding}, \citenamefont {Oxenl{\o}we},
  \citenamefont {Thompson}, \citenamefont {O'Brien}, \citenamefont {Li},
  \citenamefont {Gong},\ and\ \citenamefont {Wang}}]{chi2022}%
  \BibitemOpen
  \bibfield  {author} {\bibinfo {author} {\bibfnamefont {Y.}~\bibnamefont
  {Chi}}, \bibinfo {author} {\bibfnamefont {J.}~\bibnamefont {Huang}}, \bibinfo
  {author} {\bibfnamefont {Z.}~\bibnamefont {Zhang}}, \bibinfo {author}
  {\bibfnamefont {J.}~\bibnamefont {Mao}}, \bibinfo {author} {\bibfnamefont
  {Z.}~\bibnamefont {Zhou}}, \bibinfo {author} {\bibfnamefont {X.}~\bibnamefont
  {Chen}}, \bibinfo {author} {\bibfnamefont {C.}~\bibnamefont {Zhai}}, \bibinfo
  {author} {\bibfnamefont {J.}~\bibnamefont {Bao}}, \bibinfo {author}
  {\bibfnamefont {T.}~\bibnamefont {Dai}}, \bibinfo {author} {\bibfnamefont
  {H.}~\bibnamefont {Yuan}}, \bibinfo {author} {\bibfnamefont {M.}~\bibnamefont
  {Zhang}}, \bibinfo {author} {\bibfnamefont {D.}~\bibnamefont {Dai}}, \bibinfo
  {author} {\bibfnamefont {B.}~\bibnamefont {Tang}}, \bibinfo {author}
  {\bibfnamefont {Y.}~\bibnamefont {Yang}}, \bibinfo {author} {\bibfnamefont
  {Z.}~\bibnamefont {Li}}, \bibinfo {author} {\bibfnamefont {Y.}~\bibnamefont
  {Ding}}, \bibinfo {author} {\bibfnamefont {L.~K.}\ \bibnamefont
  {Oxenl{\o}we}}, \bibinfo {author} {\bibfnamefont {M.~G.}\ \bibnamefont
  {Thompson}}, \bibinfo {author} {\bibfnamefont {J.~L.}\ \bibnamefont
  {O'Brien}}, \bibinfo {author} {\bibfnamefont {Y.}~\bibnamefont {Li}},
  \bibinfo {author} {\bibfnamefont {Q.}~\bibnamefont {Gong}},\ and\ \bibinfo
  {author} {\bibfnamefont {J.}~\bibnamefont {Wang}},\ }\bibfield  {title}
  {\bibinfo {title} {A programmable qudit-based quantum processor},\ }\href
  {https://doi.org/10.1038/s41467-022-28767-x} {\bibfield  {journal} {\bibinfo
  {journal} {Nature Communications}\ }\textbf {\bibinfo {volume} {13}},\
  \bibinfo {pages} {1166} (\bibinfo {year} {2022})}\BibitemShut {NoStop}%
\bibitem [{\citenamefont {Ringbauer}\ \emph {et~al.}(2022)\citenamefont
  {Ringbauer}, \citenamefont {Meth}, \citenamefont {Postler}, \citenamefont
  {Stricker}, \citenamefont {Blatt}, \citenamefont {Schindler},\ and\
  \citenamefont {Monz}}]{ringbauer2022}%
  \BibitemOpen
  \bibfield  {author} {\bibinfo {author} {\bibfnamefont {M.}~\bibnamefont
  {Ringbauer}}, \bibinfo {author} {\bibfnamefont {M.}~\bibnamefont {Meth}},
  \bibinfo {author} {\bibfnamefont {L.}~\bibnamefont {Postler}}, \bibinfo
  {author} {\bibfnamefont {R.}~\bibnamefont {Stricker}}, \bibinfo {author}
  {\bibfnamefont {R.}~\bibnamefont {Blatt}}, \bibinfo {author} {\bibfnamefont
  {P.}~\bibnamefont {Schindler}},\ and\ \bibinfo {author} {\bibfnamefont
  {T.}~\bibnamefont {Monz}},\ }\bibfield  {title} {\bibinfo {title} {A
  universal qudit quantum processor with trapped ions},\ }\href
  {https://doi.org/10.1038/s41567-022-01658-0} {\bibfield  {journal} {\bibinfo
  {journal} {Nature Physics}\ }\textbf {\bibinfo {volume} {18}},\ \bibinfo
  {pages} {1053} (\bibinfo {year} {2022})}\BibitemShut {NoStop}%
\bibitem [{\citenamefont {Seifert}\ \emph {et~al.}(2023)\citenamefont
  {Seifert}, \citenamefont {Li}, \citenamefont {Roy}, \citenamefont {Schuster},
  \citenamefont {Chong},\ and\ \citenamefont {Baker}}]{seifert2023a}%
  \BibitemOpen
  \bibfield  {author} {\bibinfo {author} {\bibfnamefont {L.~M.}\ \bibnamefont
  {Seifert}}, \bibinfo {author} {\bibfnamefont {Z.}~\bibnamefont {Li}},
  \bibinfo {author} {\bibfnamefont {T.}~\bibnamefont {Roy}}, \bibinfo {author}
  {\bibfnamefont {D.~I.}\ \bibnamefont {Schuster}}, \bibinfo {author}
  {\bibfnamefont {F.~T.}\ \bibnamefont {Chong}},\ and\ \bibinfo {author}
  {\bibfnamefont {J.~M.}\ \bibnamefont {Baker}},\ }\bibfield  {title} {\bibinfo
  {title} {Exploring ququart computation on a transmon using optimal control},\
  }\href {https://doi.org/10.1103/PhysRevA.108.062609} {\bibfield  {journal}
  {\bibinfo  {journal} {Physical Review A}\ }\textbf {\bibinfo {volume}
  {108}},\ \bibinfo {pages} {062609} (\bibinfo {year} {2023})}\BibitemShut
  {NoStop}%
\bibitem [{\citenamefont {Liu}\ \emph {et~al.}(2023)\citenamefont {Liu},
  \citenamefont {Wang}, \citenamefont {Zhang}, \citenamefont {Zhang},
  \citenamefont {Cai}, \citenamefont {Xu}, \citenamefont {Li}, \citenamefont
  {Han}, \citenamefont {Li}, \citenamefont {Xue}, \citenamefont {Liu},
  \citenamefont {You}, \citenamefont {Jin},\ and\ \citenamefont
  {Yu}}]{liu2023}%
  \BibitemOpen
  \bibfield  {author} {\bibinfo {author} {\bibfnamefont {P.}~\bibnamefont
  {Liu}}, \bibinfo {author} {\bibfnamefont {R.}~\bibnamefont {Wang}}, \bibinfo
  {author} {\bibfnamefont {J.-N.}\ \bibnamefont {Zhang}}, \bibinfo {author}
  {\bibfnamefont {Y.}~\bibnamefont {Zhang}}, \bibinfo {author} {\bibfnamefont
  {X.}~\bibnamefont {Cai}}, \bibinfo {author} {\bibfnamefont {H.}~\bibnamefont
  {Xu}}, \bibinfo {author} {\bibfnamefont {Z.}~\bibnamefont {Li}}, \bibinfo
  {author} {\bibfnamefont {J.}~\bibnamefont {Han}}, \bibinfo {author}
  {\bibfnamefont {X.}~\bibnamefont {Li}}, \bibinfo {author} {\bibfnamefont
  {G.}~\bibnamefont {Xue}}, \bibinfo {author} {\bibfnamefont {W.}~\bibnamefont
  {Liu}}, \bibinfo {author} {\bibfnamefont {L.}~\bibnamefont {You}}, \bibinfo
  {author} {\bibfnamefont {Y.}~\bibnamefont {Jin}},\ and\ \bibinfo {author}
  {\bibfnamefont {H.}~\bibnamefont {Yu}},\ }\bibfield  {title} {\bibinfo
  {title} {Performing {{SU}} ( d ) {{Operations}} and {{Rudimentary
  Algorithms}} in a {{Superconducting Transmon Qudit}} for d = 3 and d = 4},\
  }\href {https://doi.org/10.1103/PhysRevX.13.021028} {\bibfield  {journal}
  {\bibinfo  {journal} {Physical Review X}\ }\textbf {\bibinfo {volume} {13}},\
  \bibinfo {pages} {021028} (\bibinfo {year} {2023})}\BibitemShut {NoStop}%
\bibitem [{\citenamefont {Fischer}\ \emph {et~al.}(2023)\citenamefont
  {Fischer}, \citenamefont {Chiesa}, \citenamefont {Tacchino}, \citenamefont
  {Egger}, \citenamefont {Carretta},\ and\ \citenamefont
  {Tavernelli}}]{fischer2023}%
  \BibitemOpen
  \bibfield  {author} {\bibinfo {author} {\bibfnamefont {L.~E.}\ \bibnamefont
  {Fischer}}, \bibinfo {author} {\bibfnamefont {A.}~\bibnamefont {Chiesa}},
  \bibinfo {author} {\bibfnamefont {F.}~\bibnamefont {Tacchino}}, \bibinfo
  {author} {\bibfnamefont {D.~J.}\ \bibnamefont {Egger}}, \bibinfo {author}
  {\bibfnamefont {S.}~\bibnamefont {Carretta}},\ and\ \bibinfo {author}
  {\bibfnamefont {I.}~\bibnamefont {Tavernelli}},\ }\bibfield  {title}
  {\bibinfo {title} {Universal {{Qudit Gate Synthesis}} for {{Transmons}}},\
  }\href {https://doi.org/10.1103/PRXQuantum.4.030327} {\bibfield  {journal}
  {\bibinfo  {journal} {PRX Quantum}\ }\textbf {\bibinfo {volume} {4}},\
  \bibinfo {pages} {030327} (\bibinfo {year} {2023})}\BibitemShut {NoStop}%
\bibitem [{Note2()}]{Note2}%
  \BibitemOpen
  \bibinfo {note} {With the operator summation and the multiplication by a
  scalar, $\protect \mathbf {L}(\protect \EuScript {H})$ can be viewed as a
  vector space; with the operator summation and the operator multiplication,
  $\protect \mathbf {L}(\protect \EuScript {H})$ can be viewed as a ring in
  which the identity operator plays the role of the identity element; the
  coexistence of the vector-space and ring structures defines $\protect \mathbf
  {L}(\protect \EuScript {H})$ as an algebra.}\BibitemShut {Stop}%
\bibitem [{Note3()}]{Note3}%
  \BibitemOpen
  \bibinfo {note} {As a basic property of operators, a subspace $\protect
  \EuScript {H}_{\protect \mathrm {code}}$ is invariant under the action of
  both $O_A$ and ${O_A}^+$ if and only if $[O_A,P_{\protect \mathrm
  {code}}]=0$.}\BibitemShut {Stop}%
\bibitem [{Note4()}]{Note4}%
  \BibitemOpen
  \bibinfo {note} {$\protect \mathbb {Q}_p$ is simply an alternative geometric
  completion of the rationals $\protect \mathbb {Q}$ with respect to the
  alternatively defined norm associated with the prime $p$.}\BibitemShut
  {Stop}%
\bibitem [{Note5()}]{Note5}%
  \BibitemOpen
  \bibinfo {note} {In this work we only consider the state-independent or exact
  complementary recovery~\cite {cao2021}.}\BibitemShut {Stop}%
\bibitem [{\citenamefont {Cleve}\ \emph {et~al.}(1999)\citenamefont {Cleve},
  \citenamefont {Gottesman},\ and\ \citenamefont {Lo}}]{cleve1999}%
  \BibitemOpen
  \bibfield  {author} {\bibinfo {author} {\bibfnamefont {R.}~\bibnamefont
  {Cleve}}, \bibinfo {author} {\bibfnamefont {D.}~\bibnamefont {Gottesman}},\
  and\ \bibinfo {author} {\bibfnamefont {H.-K.}\ \bibnamefont {Lo}},\
  }\bibfield  {title} {\bibinfo {title} {How to {{Share}} a {{Quantum
  Secret}}},\ }\href {https://doi.org/10.1103/PhysRevLett.83.648} {\bibfield
  {journal} {\bibinfo  {journal} {Physical Review Letters}\ }\textbf {\bibinfo
  {volume} {83}},\ \bibinfo {pages} {648} (\bibinfo {year} {1999})}\BibitemShut
  {NoStop}%
\bibitem [{Note6()}]{Note6}%
  \BibitemOpen
  \bibinfo {note} {By definition and according to similar arguments in
  App.~\ref {possc}, $\otimes _{\protect \bm {x}\in \protect \mathrm
  {W}[A]}\protect \mathbf {L}(\protect \mathfrak {e}_{\protect \bm {x}})\subset
  R^+\protect \EuScript {M}_A R$, hence $\protect \mathrm {W}[A]$ must lie
  within the support of $R^+\protect \EuScript {M}_A R$. It is also easy to
  show that $\protect \mathrm {E}[A\protect \overline {A}]$ lies within the
  support. Indeed, suppose that there exists a bulk qudit $\protect \bm {x}\in
  \protect \mathrm {E}[A\protect \overline {A}]$ which is excluded from the
  support, then all operators in $R^+\protect \EuScript {M}_A R$ must act
  trivially on this qudit. Or equivalently, $\protect \EuScript {M}(\protect
  \bm {x})\subset \protect \EuScript {M}'_A=\protect \EuScript {M}_{\protect
  \overline {A}}$, i.e. $\protect \bm {x}\in \protect \mathrm {W}[\protect
  \overline {A}]$, which contradicts the fact that $\protect \bm {x}\in
  \protect \mathrm {E}[A\protect \overline {A}]$. On the other hand, it is easy
  to see that $\protect \mathrm {W}[\protect \overline {A}]$ is excluded from
  the support of $R^+\protect \EuScript {M}_A R$, because it lies within the
  support of $R^+\protect \EuScript {M}_{\protect \overline {A}} R$ which
  commute with $R^+\protect \EuScript {M}_A R$. Hence, we can conclude that the
  support of $R^+\protect \EuScript {M}_A R$ is $\protect \mathrm {W}[A]\cup
  \protect \mathrm {E}[A\protect \overline {A}]$. Similar arguments apply to
  the support of $R^+\protect \EuScript {M}_{\protect \overline {A}}
  R$.}\BibitemShut {Stop}%
\bibitem [{\citenamefont {Bao}\ and\ \citenamefont {Naskar}(2022)}]{bao2022}%
  \BibitemOpen
  \bibfield  {author} {\bibinfo {author} {\bibfnamefont {N.}~\bibnamefont
  {Bao}}\ and\ \bibinfo {author} {\bibfnamefont {J.}~\bibnamefont {Naskar}},\
  }\bibfield  {title} {\bibinfo {title} {Code properties of the holographic
  {{Sierpinski}} triangle},\ }\href
  {https://doi.org/10.1103/PhysRevD.106.126006} {\bibfield  {journal} {\bibinfo
   {journal} {Physical Review D}\ }\textbf {\bibinfo {volume} {106}},\ \bibinfo
  {pages} {126006} (\bibinfo {year} {2022})}\BibitemShut {NoStop}%
\bibitem [{\citenamefont {Yoshida}(2013)}]{yoshida2013}%
  \BibitemOpen
  \bibfield  {author} {\bibinfo {author} {\bibfnamefont {B.}~\bibnamefont
  {Yoshida}},\ }\bibfield  {title} {\bibinfo {title} {Exotic topological order
  in fractal spin liquids},\ }\href
  {https://doi.org/10.1103/PhysRevB.88.125122} {\bibfield  {journal} {\bibinfo
  {journal} {Physical Review B}\ }\textbf {\bibinfo {volume} {88}},\ \bibinfo
  {pages} {125122} (\bibinfo {year} {2013})}\BibitemShut {NoStop}%
\bibitem [{\citenamefont {Kempkes}\ \emph {et~al.}(2019)\citenamefont
  {Kempkes}, \citenamefont {Slot}, \citenamefont {Freeney}, \citenamefont
  {Zevenhuizen}, \citenamefont {Vanmaekelbergh}, \citenamefont {Swart},\ and\
  \citenamefont {Smith}}]{kempkes2019}%
  \BibitemOpen
  \bibfield  {author} {\bibinfo {author} {\bibfnamefont {S.~N.}\ \bibnamefont
  {Kempkes}}, \bibinfo {author} {\bibfnamefont {M.~R.}\ \bibnamefont {Slot}},
  \bibinfo {author} {\bibfnamefont {S.~E.}\ \bibnamefont {Freeney}}, \bibinfo
  {author} {\bibfnamefont {S.~J.}\ \bibnamefont {Zevenhuizen}}, \bibinfo
  {author} {\bibfnamefont {D.}~\bibnamefont {Vanmaekelbergh}}, \bibinfo
  {author} {\bibfnamefont {I.}~\bibnamefont {Swart}},\ and\ \bibinfo {author}
  {\bibfnamefont {C.~M.}\ \bibnamefont {Smith}},\ }\bibfield  {title} {\bibinfo
  {title} {Design and characterization of electrons in a fractal geometry},\
  }\href {https://doi.org/10.1038/s41567-018-0328-0} {\bibfield  {journal}
  {\bibinfo  {journal} {Nature Physics}\ }\textbf {\bibinfo {volume} {15}},\
  \bibinfo {pages} {127} (\bibinfo {year} {2019})}\BibitemShut {NoStop}%
\bibitem [{\citenamefont {Manna}\ \emph {et~al.}(2020)\citenamefont {Manna},
  \citenamefont {Pal}, \citenamefont {Wang},\ and\ \citenamefont
  {Nielsen}}]{manna2020}%
  \BibitemOpen
  \bibfield  {author} {\bibinfo {author} {\bibfnamefont {S.}~\bibnamefont
  {Manna}}, \bibinfo {author} {\bibfnamefont {B.}~\bibnamefont {Pal}}, \bibinfo
  {author} {\bibfnamefont {W.}~\bibnamefont {Wang}},\ and\ \bibinfo {author}
  {\bibfnamefont {A.~E.~B.}\ \bibnamefont {Nielsen}},\ }\bibfield  {title}
  {\bibinfo {title} {Anyons and fractional quantum {{Hall}} effect in fractal
  dimensions},\ }\href {https://doi.org/10.1103/PhysRevResearch.2.023401}
  {\bibfield  {journal} {\bibinfo  {journal} {Physical Review Research}\
  }\textbf {\bibinfo {volume} {2}},\ \bibinfo {pages} {023401} (\bibinfo {year}
  {2020})}\BibitemShut {NoStop}%
\bibitem [{\citenamefont {Xu}\ \emph {et~al.}(2021)\citenamefont {Xu},
  \citenamefont {Wang}, \citenamefont {Chen}, \citenamefont {Smith},\ and\
  \citenamefont {Jin}}]{xu2021}%
  \BibitemOpen
  \bibfield  {author} {\bibinfo {author} {\bibfnamefont {X.~Y.}\ \bibnamefont
  {Xu}}, \bibinfo {author} {\bibfnamefont {X.~W.}\ \bibnamefont {Wang}},
  \bibinfo {author} {\bibfnamefont {D.~Y.}\ \bibnamefont {Chen}}, \bibinfo
  {author} {\bibfnamefont {C.~M.}\ \bibnamefont {Smith}},\ and\ \bibinfo
  {author} {\bibfnamefont {X.~M.}\ \bibnamefont {Jin}},\ }\bibfield  {title}
  {\bibinfo {title} {Quantum transport in fractal networks},\ }\href
  {https://doi.org/10.1038/s41566-021-00845-4} {\bibfield  {journal} {\bibinfo
  {journal} {Nature Photonics}\ }\textbf {\bibinfo {volume} {15}},\ \bibinfo
  {pages} {703} (\bibinfo {year} {2021})}\BibitemShut {NoStop}%
\bibitem [{\citenamefont {Biesenthal}\ \emph {et~al.}(2022)\citenamefont
  {Biesenthal}, \citenamefont {Maczewsky}, \citenamefont {Yang}, \citenamefont
  {Kremer}, \citenamefont {Segev}, \citenamefont {Szameit},\ and\ \citenamefont
  {Heinrich}}]{biesenthal2022}%
  \BibitemOpen
  \bibfield  {author} {\bibinfo {author} {\bibfnamefont {T.}~\bibnamefont
  {Biesenthal}}, \bibinfo {author} {\bibfnamefont {L.~J.}\ \bibnamefont
  {Maczewsky}}, \bibinfo {author} {\bibfnamefont {Z.}~\bibnamefont {Yang}},
  \bibinfo {author} {\bibfnamefont {M.}~\bibnamefont {Kremer}}, \bibinfo
  {author} {\bibfnamefont {M.}~\bibnamefont {Segev}}, \bibinfo {author}
  {\bibfnamefont {A.}~\bibnamefont {Szameit}},\ and\ \bibinfo {author}
  {\bibfnamefont {M.}~\bibnamefont {Heinrich}},\ }\bibfield  {title} {\bibinfo
  {title} {Fractal photonic topological insulators},\ }\href
  {https://doi.org/10.1126/science.abm2842} {\bibfield  {journal} {\bibinfo
  {journal} {Science}\ }\textbf {\bibinfo {volume} {376}},\ \bibinfo {pages}
  {1114} (\bibinfo {year} {2022})}\BibitemShut {NoStop}%
\bibitem [{\citenamefont {Zhu}\ \emph {et~al.}(2022)\citenamefont {Zhu},
  \citenamefont {{Jochym-O'Connor}},\ and\ \citenamefont {Dua}}]{zhu2022}%
  \BibitemOpen
  \bibfield  {author} {\bibinfo {author} {\bibfnamefont {G.}~\bibnamefont
  {Zhu}}, \bibinfo {author} {\bibfnamefont {T.}~\bibnamefont
  {{Jochym-O'Connor}}},\ and\ \bibinfo {author} {\bibfnamefont
  {A.}~\bibnamefont {Dua}},\ }\bibfield  {title} {\bibinfo {title} {Topological
  {{Order}}, {{Quantum Codes}}, and {{Quantum Computation}} on {{Fractal
  Geometries}}},\ }\href {https://doi.org/10.1103/PRXQuantum.3.030338}
  {\bibfield  {journal} {\bibinfo  {journal} {PRX Quantum}\ }\textbf {\bibinfo
  {volume} {3}},\ \bibinfo {pages} {030338} (\bibinfo {year}
  {2022})}\BibitemShut {NoStop}%
\bibitem [{Note7()}]{Note7}%
  \BibitemOpen
  \bibinfo {note} {Here, by a almost connected subregion, we mean a subregion
  which can be viewed as the union of a small number of connected
  parts.}\BibitemShut {Stop}%
\bibitem [{Note8()}]{Note8}%
  \BibitemOpen
  \bibinfo {note} {Ideal polygons on the Poincar\'e disk have their endpoints
  on the asymptotic boundary, and have their sides the geodesics of the
  hyperbolic geometry. An ideal polygons can be simply specified by the
  geodesics, and an ideal hyperbolic tessellation can be specified in a similar
  way with a family of geodesics.}\BibitemShut {Stop}%
\bibitem [{\citenamefont {Levin}\ and\ \citenamefont {Wen}(2005)}]{levin2005}%
  \BibitemOpen
  \bibfield  {author} {\bibinfo {author} {\bibfnamefont {M.~A.}\ \bibnamefont
  {Levin}}\ and\ \bibinfo {author} {\bibfnamefont {X.-G.}\ \bibnamefont
  {Wen}},\ }\bibfield  {title} {\bibinfo {title} {String-net condensation:
  {{A}} physical mechanism for topological phases},\ }\href
  {https://doi.org/10.1103/PhysRevB.71.045110} {\bibfield  {journal} {\bibinfo
  {journal} {Physical Review B}\ }\textbf {\bibinfo {volume} {71}},\ \bibinfo
  {pages} {045110} (\bibinfo {year} {2005})}\BibitemShut {NoStop}%
\bibitem [{\citenamefont {Chen}\ \emph {et~al.}(2010)\citenamefont {Chen},
  \citenamefont {Gu},\ and\ \citenamefont {Wen}}]{chen2010}%
  \BibitemOpen
  \bibfield  {author} {\bibinfo {author} {\bibfnamefont {X.}~\bibnamefont
  {Chen}}, \bibinfo {author} {\bibfnamefont {Z.-C.}\ \bibnamefont {Gu}},\ and\
  \bibinfo {author} {\bibfnamefont {X.-G.}\ \bibnamefont {Wen}},\ }\bibfield
  {title} {\bibinfo {title} {Local unitary transformation, long-range quantum
  entanglement, wave function renormalization, and topological order},\ }\href
  {https://doi.org/10.1103/PhysRevB.82.155138} {\bibfield  {journal} {\bibinfo
  {journal} {Physical Review B}\ }\textbf {\bibinfo {volume} {82}},\ \bibinfo
  {pages} {155138} (\bibinfo {year} {2010})}\BibitemShut {NoStop}%
\bibitem [{\citenamefont {Wang}\ and\ \citenamefont
  {{Capogrosso-Sansone}}(2017)}]{wang2017}%
  \BibitemOpen
  \bibfield  {author} {\bibinfo {author} {\bibfnamefont {W.}~\bibnamefont
  {Wang}}\ and\ \bibinfo {author} {\bibfnamefont {B.}~\bibnamefont
  {{Capogrosso-Sansone}}},\ }\bibfield  {title} {\bibinfo {title} {The
  \${\textbackslash}mathbb\{\vphantom\}{{Z}}\vphantom\{\}\_2\$ toric-code and
  the double-semion topological order of hardcore {{Bose-Hubbard-type}} models
  in the strong-interaction limit},\ }\href@noop {} {\bibfield  {journal}
  {\bibinfo  {journal} {Scientific Reports}\ }\textbf {\bibinfo {volume} {7}},\
  \bibinfo {pages} {11071} (\bibinfo {year} {2017})}\BibitemShut {NoStop}%
\bibitem [{Note9()}]{Note9}%
  \BibitemOpen
  \bibinfo {note} {For the proof, we notice that there is some operator $Q_A$
  commuting with $P_{\protect \mathrm {code}}$ and satisfying $P_{\protect
  \mathrm {code}}Q_AP_{\protect \mathrm {code}}=R\dyad {\protect \bm {\beta
  }'_{\protect \bm {x}}}{\protect \bm {\beta }_{\protect \bm {x}}}R^+$. Then,
  we have $\ket {\protect \tilde {\varphi }_{n'}}=P_{\protect \mathrm
  {code}}Q_AP_{\protect \mathrm {code}}\ket *{\protect \tilde {\varphi }_n}$.
  With this equality, we can easily show that $\ev {O_A}{\varphi _n}=1$ while
  $\ev {O_A}{\varphi _{n'}}=\ev {O_A}{P_{\protect \mathrm {code}}Q_AP_{\protect
  \mathrm {code}}\varphi _n}=0$. In a similar way, using the commutativity
  between operators recoverable on $A$ and operators supported on $\protect
  \overline {A}$, we can also easily show that $\ev {O_{\protect \overline
  {A}}}{\varphi _n}=\ev {O_{\protect \overline {A}}}{\varphi _{n'}}$ and $\mel
  {\varphi _n}{O'_{\protect \overline {A}}}{\varphi _{n'}}=0$ for any
  $O_{\protect \overline {A}}$ and $O'_{\protect \overline {A}}$.}\BibitemShut
  {Stop}%
\bibitem [{Note10()}]{Note10}%
  \BibitemOpen
  \bibinfo {note} {The argument of sub-configuration $\alpha _{j_1}\alpha
  _{j_2}\protect \cdots $ in the complement $\protect \overline {A}$ applies to
  any $\ket *{\protect \tilde {\varphi }_n}=R\ket {\protect \bm {\beta
  }_{\protect \mathbf {1}}\protect \cdots \protect \bm {\beta }_{\protect \bm
  {x}}\protect \cdots \protect \bm {\beta }_{K}}$ and $\ket {\protect \tilde
  {\varphi }_{n'}}=R\ket {\protect \bm {\beta }_{\protect \mathbf {1}}\protect
  \cdots \protect \bm {\beta }'_{\protect \bm {x}}\protect \cdots \protect \bm
  {\beta }_{K}}$. And it simply means that the numbers of qudit-product-states
  in their expansion are the same. Then, since any two basis states are indeed
  the two ends of a sequence of such states appearing pairwise, we conclude
  that all such basis states are expanded by the same number of
  qudit-product-states.}\BibitemShut {Stop}%
\bibitem [{Note11()}]{Note11}%
  \BibitemOpen
  \bibinfo {note} {The minimal boundary subregion for recovering $\protect
  \EuScript {M}(\protect \bm {x})$ cannot be unique, since Characteristic 1
  guarantees that any small enough part of a minimal recovery is correctable,
  i.e., there exists another minimal recovery that excludes this part. However,
  two minimal recovery cannot be disjoint, since if so then one of the two
  should be correctable, and we get a contradiction to the nature of a minimal
  recovery.}\BibitemShut {Stop}%
\bibitem [{\citenamefont {Wang}(2023)}]{wang2023}%
  \BibitemOpen
  \bibfield  {author} {\bibinfo {author} {\bibfnamefont {W.}~\bibnamefont
  {Wang}},\ }\href@noop {} {\bibinfo {title} {A new realization of the
  long-range entanglement: Fractality replacing the topological order}}
  (\bibinfo {year} {2023}),\ \Eprint {https://arxiv.org/abs/2201.13041}
  {arXiv:2201.13041 [quant-ph]} \BibitemShut {NoStop}%
\bibitem [{Note12()}]{Note12}%
  \BibitemOpen
  \bibinfo {note} {For states orthogonal to $\protect \EuScript {H}_0$, they
  are assigned to $0$.}\BibitemShut {Stop}%
\bibitem [{Note13()}]{Note13}%
  \BibitemOpen
  \bibinfo {note} {It should be noted that by fixing different $\ket {\psi _m}$
  and $\ket {\psi _{m'}}$ lying within the same sub-collection, by applying the
  $T(\protect \bm {x})$ gates we can traverse the same sub-collection of
  states. There is no contradiction. Indeed, because of the equality
  $T(\protect \bm {x})T(\protect \bm {x})=\protect \mathds {1}$, if we have
  $\ket {\psi _{m'}}=(T(\protect \bm {x})T(\protect \bm {x}')\protect \cdots
  )\ket {\psi _m}$ then we also have $(T(\protect \bm {x})T(\protect \bm
  {x}')\protect \cdots )\ket {\psi _{m'}}=(T(\protect \bm {x})T(\protect \bm
  {x}')\protect \cdots )(T(\protect \bm {x})T(\protect \bm {x}')\protect \cdots
  )\ket {\psi _m}=\ket {\psi _m}$. That is, $\ket {\psi _m}$ and $\ket {\psi
  _{m'}}$ can be expressed in terms of each other with the same product
  $(T(\protect \bm {x})T(\protect \bm {x}')\protect \cdots )$.}\BibitemShut
  {Stop}%
\bibitem [{Note14()}]{Note14}%
  \BibitemOpen
  \bibinfo {note} {Accroding to $[T(\protect \bm {x}),T(\protect \bm {x}')]=0$,
  $T(\protect \bm {x})T(\protect \bm {x})=\protect \mathds {1}$ and $(\protect
  \frac {\protect \mathds {1}+T(\protect \bm {x})}{2})(\protect \frac {\protect
  \mathds {1}+T(\protect \bm {x})}{2})=(\protect \frac {\protect \mathds
  {1}+T(\protect \bm {x})}{2})$, it can be easily checked that this operator is
  a projection.}\BibitemShut {Stop}%
\bibitem [{Note15()}]{Note15}%
  \BibitemOpen
  \bibinfo {note} {The commutativity guarantees that $P_{\protect \mathrm
  {code}}$ is a projection operator. According to Eq.~\ref {cbs1}, we have
  $P_{\protect \mathrm {code}}\ket {\psi _m}=\protect \frac {1}{\protect \sqrt
  {2^{K}}}\ket *{\protect \tilde {\varphi }_n}\in \protect \EuScript
  {H}_{\protect \mathrm {code}}$ and $P_{\protect \mathrm {code}}\ket
  *{\protect \tilde {\varphi }_n}=\protect \sqrt {2^{K}}P_{\protect \mathrm
  {code}}P_{\protect \mathrm {code}}\ket {\psi _m}=\protect \sqrt
  {2^{K}}P_{\protect \mathrm {code}}\ket {\psi _m}=\ket *{\protect \tilde
  {\varphi }_n}$.}\BibitemShut {Stop}%
\bibitem [{\citenamefont {Wang}\ \emph {et~al.}(2020)\citenamefont {Wang},
  \citenamefont {Hu}, \citenamefont {Sanders},\ and\ \citenamefont
  {Kais}}]{wang2020a}%
  \BibitemOpen
  \bibfield  {author} {\bibinfo {author} {\bibfnamefont {Y.}~\bibnamefont
  {Wang}}, \bibinfo {author} {\bibfnamefont {Z.}~\bibnamefont {Hu}}, \bibinfo
  {author} {\bibfnamefont {B.~C.}\ \bibnamefont {Sanders}},\ and\ \bibinfo
  {author} {\bibfnamefont {S.}~\bibnamefont {Kais}},\ }\bibfield  {title}
  {\bibinfo {title} {Qudits and {{High-Dimensional Quantum Computing}}},\
  }\href {https://doi.org/10.3389/fphy.2020.589504} {\bibfield  {journal}
  {\bibinfo  {journal} {Frontiers in Physics}\ }\textbf {\bibinfo {volume}
  {8}},\ \bibinfo {pages} {589504} (\bibinfo {year} {2020})}\BibitemShut
  {NoStop}%
\bibitem [{Note16()}]{Note16}%
  \BibitemOpen
  \bibinfo {note} {We write $\otimes _i\dyad {\alpha _i}$ as $\dyad {\alpha
  _{i_1}\alpha _{i_2}\protect \cdots }$. Then the state $\ket {\alpha
  _{i_1}\alpha _{i_2}\protect \cdots }$ on $A$ either matches some $\ket {\psi
  _m}$ on which $\otimes _i\dyad {\alpha _i}$ acts as the identity operator, or
  it matches no states in $\protect \EuScript {H}_0$ and hence acts as $0$. It
  means that the action of $\otimes _i\dyad {\alpha _i}$ (and also $(\otimes
  _i\dyad {\alpha _i})^+=\otimes _i\dyad {\alpha _i}$) leaves $\protect
  \EuScript {H}_0$ invariant. Hence $\otimes _i\dyad {\alpha _i}$ commutes with
  $P_0$.}\BibitemShut {Stop}%
\bibitem [{Note17()}]{Note17}%
  \BibitemOpen
  \bibinfo {note} {First, in these models, each of the component tensor itself
  does not satisfy the genuine operator-algebra correctability for arbitrary
  bipartition of its legs. Second, in the process of thinning out the bulk
  degrees of freedom, tensors with a dangling leg for the bulk qubit are
  replaced by perfect tensors without bulk legs. For these tensors which lie at
  the boundary, they are usually contracted with the whole tensor network
  through only quite few number of legs, and most legs are uncontracted.
  Consequently, through the operator pushing, some of these uncontracted
  boundary legs can be avoided and are hence correctable with respect to the
  whole of the bulk. This simply means that the genuine operator-algebra
  correctability does not hold for boundary bipartition that separating these
  uncontracted legs from others.}\BibitemShut {Stop}%
\bibitem [{Note18()}]{Note18}%
  \BibitemOpen
  \bibinfo {note} {That is to consider the features on the paths whether it
  exactly implies the parity on the loops surrounding the central hole in
  specifying the emergent bulk degrees of freedom. For $\protect \mathfrak
  {a}$, the correspondence between the parity is: (1) (even, even, even) or
  (odd, odd, odd) on the three paths corresponds to (even, even, even) on the
  loop; (2) (even, odd, odd) or (odd, even, even) on the paths corresponds to
  (odd, even, odd) on the loop... It is important to note that for each $\ket
  {\psi _m}$ the total number of dark sides on each of the three laterals of
  the alternative geometry is even.}\BibitemShut {Stop}%
\bibitem [{Note19()}]{Note19}%
  \BibitemOpen
  \bibinfo {note} {By definition, we have $\protect \bm {\protect \EuScript
  {M}}_{\protect \overline {a}}(\protect \bm {x})\subset \protect \bm {\protect
  \EuScript {M}}'_a(\protect \bm {x})$. Then, $\protect \overline {\protect \bm
  {\protect \EuScript {N}}}_{\protect \bm {x}}=\protect \bm {\protect \EuScript
  {N}}_{\protect \bm {x}}'\subset \protect \bm {\protect \EuScript
  {M}}_{\protect \overline {a}}(\protect \bm {x})\subset \protect \bm {\protect
  \EuScript {M}}'_a(\protect \bm {x})$ means that $\protect \bm {\protect
  \EuScript {M}}_a(\protect \bm {x})\subset \protect \bm {\protect \EuScript
  {N}}_{\protect \bm {x}}$. Combining $\protect \bm {\protect \EuScript
  {N}}_{\protect \bm {x}}\subset \protect \bm {\protect \EuScript
  {M}}_a(\protect \bm {x})$ in the main context, we have $\protect \bm
  {\protect \EuScript {N}}_{\protect \bm {x}}=\protect \bm {\protect \EuScript
  {M}}_a(\protect \bm {x})$. Similarly, we can show $\protect \overline
  {\protect \bm {\protect \EuScript {N}}}_{\protect \bm {x}}=\protect \bm
  {\protect \EuScript {M}}_{\protect \overline {a}}(\protect \bm {x})$. And
  according to the condition $\protect \overline {\protect \bm {\protect
  \EuScript {N}}}_{\protect \bm {x}}=\protect \bm {\protect \EuScript
  {N}}_{\protect \bm {x}}'$ in the main context, we have $\protect \bm
  {\protect \EuScript {M}}_{\protect \overline {a}}(\protect \bm {x})=\protect
  \bm {\protect \EuScript {M}}'_a(\protect \bm {x})$.}\BibitemShut {Stop}%
\bibitem [{Note20()}]{Note20}%
  \BibitemOpen
  \bibinfo {note} {It is shown in Ref.~\cite {harlow2017} that any given von
  Neumann algebra ensures such a decomposition. Based on the orthogonality of
  the subspaces in a direct sum of Hilbert spaces, it is also easy to prove
  that any such decomposition uniquely determines a von Neumann
  algebra.}\BibitemShut {Stop}%
\bibitem [{Note21()}]{Note21}%
  \BibitemOpen
  \bibinfo {note} {Note that in Ref.~\cite {harlow2017}, the part for
  decompositions of $\protect \EuScript {H}_A$ and $\protect \EuScript
  {H}_{\protect \overline {A}}$ are described by decomposition followed by
  unitary operators. Indeed, since the decomposition is in nature a family of
  ``orthogonal'' isometry statifying properties as above for the $\protect \bm
  {I}^{\protect \bm {\mu }}$s, decompositions (``orthogonal'' isometries) of
  $\protect \EuScript {H}_A$ and $\protect \EuScript {H}_{\protect \overline
  {A}}$ followed by unitary operators on $\protect \EuScript {H}_A$ and
  $\protect \EuScript {H}_{\protect \overline {A}}$ are again decompositions
  (``orthogonal'' isometries). Hence, we directly, and also equivalently, use
  $\protect \bm {U}^{\protect \bm {\mu }}_A$ and $\protect \bm {U}^{\protect
  \bm {\mu }}_{\protect \overline {A}}$. In terms of this perspective, in the
  decompositions of $\protect \EuScript {H}_A$ and $\protect \EuScript
  {H}_{\protect \overline {A}}$, we equivalently and directly use $\protect
  \EuScript {E}^{\protect \bm {\mu }}_a$ and $\protect \EuScript {E}^{\protect
  \bm {\mu }}_{\protect \overline {a}}$ that decompose $\protect \EuScript
  {E}$, instead of $\protect \EuScript {H}^{\protect \bm {\mu }}_{A_1}$ and
  $\protect \EuScript {H}^{\protect \bm {\mu }}_{\protect \overline {A}_1}$
  (see Ref.~\cite {harlow2017} where the superscript was $\alpha $) which were
  used as general copies of $\protect \EuScript {E}^{\protect \bm {\mu }}_a$
  and $\protect \EuScript {E}^{\protect \bm {\mu }}_{\protect \overline {a}}$
  respectively. And we specify the necessary isometry $\protect \bm
  {J}^{\protect \bm {\mu }}:\ket {\protect \bm {B}^{\protect \bm {\mu
  }}_a}\otimes \ket {\protect \bm {B}^{\protect \bm {\mu }}_{\protect \overline
  {a}}}\DOTSB \mapstochar \rightarrow \ket {\protect \bm {B}^{\protect \bm {\mu
  }}_a}\otimes \ket {\chi ^{\protect \bm {\mu }}}\otimes \ket {\protect \bm
  {B}^{\protect \bm {\mu }}_{\protect \overline {a}}}$ by directly using $\ket
  {\protect \bm {B}^{\protect \bm {\mu }}_a}$ and $\ket {\protect \bm
  {B}^{\protect \bm {\mu }}_{\protect \overline {a}}}$ on the right-hand-side,
  instead of using the copies $\ket {\protect \bm {B}^{\protect \bm {\mu
  }}_a}_{A_1}$ and $\ket {\protect \bm {B}^{\protect \bm {\mu }}_{\protect
  \overline {a}}}_{\protect \overline {A}_1}$. In our notation, there is no
  loss of generality, since the generality is contained the ``orthogonal''
  isometries that underlie the decomposition.}\BibitemShut {Stop}%
\bibitem [{\citenamefont {Cao}(2023)}]{cao2023}%
  \BibitemOpen
  \bibfield  {author} {\bibinfo {author} {\bibfnamefont {C.}~\bibnamefont
  {Cao}},\ }\href@noop {} {\bibinfo {title} {Stabilizer {{Codes Have Trivial
  Area Operators}}}} (\bibinfo {year} {2023}),\ \Eprint
  {https://arxiv.org/abs/2306.14996} {arXiv:2306.14996 [hep-th,
  physics:quant-ph]} \BibitemShut {NoStop}%
\bibitem [{\citenamefont {Hartnoll}\ \emph {et~al.}(2021)\citenamefont
  {Hartnoll}, \citenamefont {Sachdev}, \citenamefont {Takayanagi},
  \citenamefont {Chen}, \citenamefont {Silverstein},\ and\ \citenamefont
  {Sonner}}]{hartnoll2021}%
  \BibitemOpen
  \bibfield  {author} {\bibinfo {author} {\bibfnamefont {S.}~\bibnamefont
  {Hartnoll}}, \bibinfo {author} {\bibfnamefont {S.}~\bibnamefont {Sachdev}},
  \bibinfo {author} {\bibfnamefont {T.}~\bibnamefont {Takayanagi}}, \bibinfo
  {author} {\bibfnamefont {X.}~\bibnamefont {Chen}}, \bibinfo {author}
  {\bibfnamefont {E.}~\bibnamefont {Silverstein}},\ and\ \bibinfo {author}
  {\bibfnamefont {J.}~\bibnamefont {Sonner}},\ }\bibfield  {title} {\bibinfo
  {title} {Quantum connections},\ }\href
  {https://doi.org/10.1038/s42254-021-00319-0} {\bibfield  {journal} {\bibinfo
  {journal} {Nature Reviews Physics}\ }\textbf {\bibinfo {volume} {3}},\
  \bibinfo {pages} {391} (\bibinfo {year} {2021})}\BibitemShut {NoStop}%
\bibitem [{\citenamefont {Str{\v a}til{\v a}}\ and\ \citenamefont
  {Zsid{\'o}}(2019)}]{stratila2019}%
  \BibitemOpen
  \bibfield  {author} {\bibinfo {author} {\bibfnamefont {{\c S}.}~\bibnamefont
  {Str{\v a}til{\v a}}}\ and\ \bibinfo {author} {\bibfnamefont
  {L.}~\bibnamefont {Zsid{\'o}}},\ }\href@noop {} {\emph {\bibinfo {title}
  {Lectures on von {{Neumann}} Algebras}}},\ \bibinfo {edition} {second
  edition}\ ed.,\ Cambridge - {{IISc}} Series\ (\bibinfo  {publisher}
  {Cambridge University Press},\ \bibinfo {address} {Cambridge New York, NY},\
  \bibinfo {year} {2019})\BibitemShut {NoStop}%
\bibitem [{Note22()}]{Note22}%
  \BibitemOpen
  \bibinfo {note} {A basic property in operator algebra on finite-dimension
  Hilbert space is that the subspace $\protect \EuScript {H}_{\protect \mathrm
  {code}}$ is invariant under the action of both operator $O$ and its adjoint
  $O^+$ if and only if $[P_{\protect \mathrm {code}},O]=0$.}\BibitemShut
  {Stop}%
\bibitem [{Note23()}]{Note23}%
  \BibitemOpen
  \bibinfo {note} {There are two important properties of the pictorial
  representations of the $\ket {\psi _m}$s as discussed in Par.~\ref {picrep}:
  (1) The number of all dark sides on either $A$ or $\protect \overline {A}$ is
  even in total. (2) By sequentially applying the $T_{ii'}$ operators, one can
  arbitrarily change the parities in the torn halves of any hole or on the
  lateral indicated by the solid segments, without changing the parities in
  other torn halves except on the torn halves on the other lateral, i.e., the
  one indicated by the dashed segment. Here, by ``arbitrarily changed'', we
  mean that before and after the change, $\ket *{\psi _{A,m_A}}$ and $\ket
  *{\psi _{\protect \overline {A},m_{\protect \overline {A}}}}$ satisfying the
  configurations of parities can all be viewed as torn from some $\ket {\psi
  _m}$ state.}\BibitemShut {Stop}%
\end{thebibliography}%

\end{document}